\documentclass[UTF8,oneside,12pt]{book}
\usepackage[margin=1in]{geometry}
\usepackage{subfiles}

\usepackage{amsmath}
\usepackage{amsthm}

\usepackage{verbatim}
\usepackage{enumerate}
\usepackage{amssymb}
\usepackage{verbatim}
\usepackage{graphicx}
\usepackage{color}
\usepackage{float}
\usepackage{amscd}
\usepackage[hyperindex,breaklinks]{hyperref}
\usepackage[bbgreekl]{mathbbol}
\usepackage{cancel}  % to be able to strike through a formula
 \usepackage[sep=3pt, offset=1.2em]{simpler-wick}
\usepackage[all]{xy}  %comm diagrams with diag arrows
\usepackage{array}   %can format array columns with displaystyle

\newcommand{\HH}{\mathcal{H}}
\newcommand{\dd}{\partial}
\newcommand{\mr}{\mathrm}
\newcommand{\ra}{\rightarrow}

\newcommand{\xra}{\xrightarrow}
\newcommand{\xla}{\xleftarrow}
\newcommand{\bt}{\bullet}
\newcommand{\ul}{\underline}
\newcommand{\ii}{\mathrm{in}}
\newcommand{\oo}{\mathrm{out}}
\newcommand{\CC}{\mathbb{C}}
\newcommand{\til}{\widetilde}
\newcommand{\wh}{\widehat}
\newcommand{\E}{\mr{Eucl}}
\newcommand{\pt}{\mathrm{pt}}
\newcommand{\End}{\mathrm{End}}
\newcommand{\Hom}{\mathrm{Hom}}
\newcommand{\RR}{\mathbb{R}}
\newcommand{\mc}{\mathcal}
\newcommand{\bb}{\mathbb}
\newcommand{\ZZ}{\mathbb{Z}}
\newcommand{\tr}{\mathrm{tr}}
\newcommand{\F}{\mathcal{F}}
\newcommand{\GGamma}{\mathbb{\Gamma}}
\newcommand{\hra}{\hookrightarrow}
\newcommand{\MM}{\mathcal{M}}
\newcommand{\LL}{\mathcal{L}}
\newcommand{\CP}{\mathbb{CP}}
\newcommand{\RP}{\mathbb{RP}}
\newcommand{\conf}{\mathrm{conf}}
\newcommand{\Conf}{\mathrm{Conf}}
\newcommand{\W}{\mathcal{W}}
\newcommand{\ol}{\overline}
\newcommand{\X}{\mathfrak{X}}
\newcommand{\calt}{\rotatebox[origin=c]{90}{$\curvearrowleft$}}
\newcommand{\car}{\rotatebox[origin=c]{270}{$\curvearrowright$}}
\newcommand{\Diff}{\mathrm{Diff}}
\newcommand{\MCG}{\mathrm{MCG}}
\newcommand{\pMCG}{\mathrm{pMCG}}
\newcommand{\dvol}{d\mathrm{vol}}
\newcommand{\loc}{\mathrm{loc}}
\newcommand{\EL}{\mathcal{EL}}
\newcommand{\g}{\mathfrak{g}}
\newcommand{\simEL}{\underset{EL}{\sim}}
\newcommand{\Hosc}{\mathcal{H}^\mathrm{osc}}
\newcommand{\vac}{\mathrm{vac}}
\newcommand{\fs}{\mathsf{s}}
\newcommand{\e}{\mathsf{e}}
\newcommand{\m}{\mathsf{m}}
\newcommand{\ppsi}{\boldsymbol{\psi}}
\newcommand{\bdd}{\boldsymbol{\partial}}
\newcommand{\mb}{\mathbf}
\newcommand{\qq}{\mathrm{q}}
\newcommand{\frh}{\mathfrak{h}}
\newcommand{\bJ}{\mathbf{J}}
\newcommand{\CS}{\mathrm{CS}}
\newcommand{\WZW}{\mathrm{WZW}}
\newcommand{\KZ}{\mathrm{KZ}}
\newcommand{\GW}{\mathrm{GW}}

\newcommand{\sj}{\mathsf{j}}
\newcommand{\sn}{\mathsf{n}}

\newcommand{\paz}{\partial_z}
\newcommand{\pabz}{\partial_{\overline z}}
\newcommand{\vp}{\varphi}
\newcommand{\il}{\iota}
\newcommand{\pa}{\partial}
\newcommand{\cL}{\mathcal{L}}

\newcommand{\be}{\begin{equation}}
\newcommand{\ee}{\end{equation}}
\def\beqa{\begin{eqnarray}}
\def\eeqa{\end{eqnarray}}
\def\beqan{\begin{eqnarray*}}
\def\eeqan{\end{eqnarray*}}
\newcommand\myeq{\mathrel{\stackrel{\makebox[0pt]{\mbox{\normalfont\tiny def}}}{=}}}
\newtheorem{theorem}{Theorem}[section]

\theoremstyle{remark}
\newtheorem{remark}{Remark}[section]
\theoremstyle{plain}
\newtheorem{lemma}[remark]{Lemma}

\newtheorem{thm}[remark]{Theorem}
\newtheorem{corollary}[remark]{Corollary}
\theoremstyle{definition}
\newtheorem{definition}[remark]{Definition}
\newtheorem{example}[remark]{Example}
\newtheorem{assumption}[remark]{Assumption}

\newtheorem{exercise}[remark]{Exercise}

\allowdisplaybreaks
\usepackage[T1,T2A]{fontenc}
\usepackage[utf8]{inputenc}
\usepackage[russian,english]{babel}
\usepackage{epigraph}
\AtBeginDocument{}
  {%\clearpage           % we want a new page          %% I commented this
  }
\begin{document}

\title{Lecture notes on conformal field theory}
\author{Pavel Mnev
%and Nicolai Reshetikhin
}
\date{}
\maketitle

\setcounter{secnumdepth}{3}

\setcounter{tocdepth}{3}

\frontmatter

%\chapter{Dedication}
%\begin{dedication}
    %To Mom
%    To my mom, in loving memory.
%\end{dedication}

\chapter[Dedication]{}
\begin{center}
\vspace*{\stretch{1}}
    \emph{To my mom, Tatjana Mneva, in loving memory.}\\
 \vspace{\stretch{2}}
\begin{comment}
Замкнувши вечные уста,\\
Стоят строители событий.\\
Уже изъят из вещества \\
Их труд. Уже пора забыть их.\\ 
\vspace{0.3cm}
Уж близится час торжества, \\
Уж пахнет пиршеством вселенским. \\
И мы на цыпочках по-детски \\
Крадемся к пище божества.\\
\vspace{0.3cm}
{\hspace{3cm} Татьяна Мнёва}
\end{comment}
% \epigraphhead[25]{
  \epigraph{
%\centering
  Замкнувши вечные уста,\\
Стоят строители событий.\\
Уже изъят из вещества \\
Их труд. Уже пора забыть их.\\ 
\vspace{0.3cm}
Уж близится час торжества, \\
Уж пахнет пиршеством вселенским. \\
И мы на цыпочках по-детски \\
Крадемся к пище божества.}{Татьяна Мнёва}
%}
 \vspace{\stretch{1}}
\end{center}

\chapter{Abstract} 
These are the notes on two-dimensional conformal field theory, based on a lecture course for graduate math students, given by P.M. in fall 2022 at the University of Notre Dame.
%and in fall and spring 2011 at the University of Zurich
These notes are intended to be substantially reworked and expanded in coauthorship with Nicolai Reshetikhin.

\vspace{\stretch{1}}

{\let\clearpage\relax \chapter{Acknowledgements} }
I want to thank Andrey Losev and Nicolai Reshetikhin %for years of inspiring conversations on CFT.
for many inspiring conversations on CFT over the years.
%QFT and 2d CFT in particular. 
These notes are based on a course that had several iterations -- spring and fall 2011 at the University of Zurich, spring 2019 and fall 2022 at the University of Notre Dame. I thank the participants of those courses for their comments and interest. When preparing an early iteration of the CFT course (in 2010--2011) I also benefitted from conversations with Vladimir Fock and %Bob 
Robert C. Penner.
For my early exposure to conformal field theory I am indebted %owe much 
to mini-courses by Alexander Belavin and Boris Feigin and %especially 
-- for the related complex geometry %/Teichm\"uller theory
-- to a course by Peter Zograf at St. Petersburg Fizmatclub in 2005.

\vspace{\stretch{1}}

\tableofcontents

\mainmatter

\chapter{A long introduction: functorial 
%view on classical and quantum local field theories
picture of 2d conformal field theory
}
\chaptermark{%Segal's picture of 2d CFT
Functorial picture of 2d CFT
}

\begin{comment}
\epigraphhead[25]{
  \epigraph{Замкнувши вечные уста,\\
Стоят строители событий.\\
Уже изъят из вещества \\
Их труд. Уже пора забыть их.\\ 
\vspace{0.3cm}
Уж близится час торжества, \\
Уж пахнет пиршеством вселенским. \\
И мы на цыпочках по-детски \\
Крадемся к пище божества.}{Татьяна Мнёва}}
\end{comment}

\section{%Atiyah-
%Segal's axioms of %$D$-dimensional 
Functorial framework for
quantum field theory} \marginpar{Lecture 1,\\ 8/24/2022}

%Main reference: \cite{Segal88}, also \cite{Atiyah}.

\subsection{The definition of QFT}

Segal \cite{Segal88} suggested the following geometrical definition of a quantum field theory. Segal focused %concentrated 
mainly on the case of 2D conformal theories; Atiyah in \cite{Atiyah} described the case of topological theories.\footnote{
Another reference for the functorial viewpoint on QFT, with motivation from quantization of classical field theories, 
is \cite{Reshetikhin}. In the exposition here I was inspired by Losev's lectures \cite{Losev2008}.
}

\textbf{Data.} A $D$-dimensional QFT is the following assignment:
\begin{itemize}
\item A closed oriented $(D-1)$-manifold $\gamma$ is assigned a vector space $\HH_\gamma$ over $\CC$ (the ``space of states'').
\item An oriented $D$-manifold $\Sigma$ with boundary split into disjoint in- and out-components such that $\dd \Sigma = -\gamma_\mr{in}\sqcup \gamma_\mr{out}$ (minus means orientation reversal),\footnote{We will say that $\Sigma$ is a cobordism from $\gamma_\mr{in}$ to $\gamma_\mr{out}$ and write $\gamma_\mr{in}\xra{\Sigma}\gamma_\mr{out}$ and think of $\Sigma$ as an arrow in a cobordism category, where objects are oriented closed $(D-1)$-manifolds. See also Remark \ref{l1 rem: def of cob with inclusions} below for a more careful definition of a cobordism.} is assigned a linear map $Z_\Sigma\colon \HH_{\gamma_\mr{in}}\ra \HH_{\gamma_\mr{out}}$ (the ``evolution operator'' or ``partition function'').%\marginpar{define a cobordism more carefully, with inclusions of boundaries as a part of data}
\begin{figure}[H]
$$\vcenter{\hbox{ \includegraphics[scale=0.5]{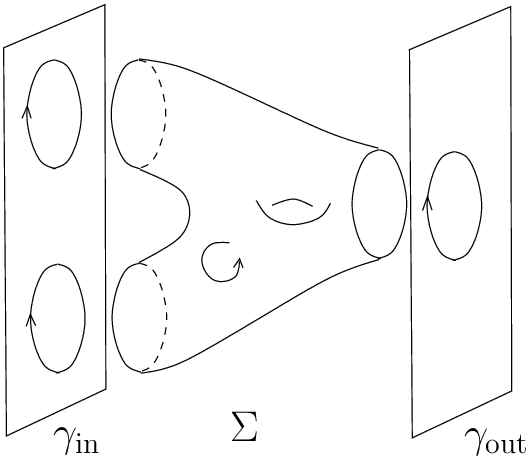} }} $$
\caption{Cobordism.}
\end{figure}
\end{itemize}

\textbf{Axioms.}

$\bullet$ \ul{Multiplicativity}: ``$\sqcup \ra \otimes$'' (disjoint unions are mapped to tensor products).
\begin{enumerate}[(a)]
\item Given two closed $(D-1)$-manifolds $\gamma_1,\gamma_2$, one has 
$$\HH_{\gamma_1\sqcup \gamma_2}=\HH_{\gamma_1}\otimes \HH_{\gamma_2}.$$
\item Given two $D$-cobordisms $\gamma_1^\mr{in}\xra{\Sigma_1}\gamma_1^\mr{out}$, $\gamma_2^\mr{in}\xra{\Sigma_2}\gamma_2^\mr{out}$, one has
$$ Z_{\Sigma_1\sqcup \Sigma_2}=Z_{\Sigma_1}\otimes Z_{\Sigma_2}
%\quad \colon \HH_{\gamma_1^\mr{in}}\otimes \HH_{\gamma_2^\mr{in}}\ra \HH_{\gamma_1^\mr{out}}\otimes \HH_{\gamma_2^\mr{out}} 
$$
where both sides are linear maps $\HH_{\gamma_1^\mr{in}}\otimes \HH_{\gamma_2^\mr{in}}\ra \HH_{\gamma_1^\mr{out}}\otimes \HH_{\gamma_2^\mr{out}} $.
\begin{figure}[H]
$$\vcenter{\hbox{ \includegraphics[scale=0.5]{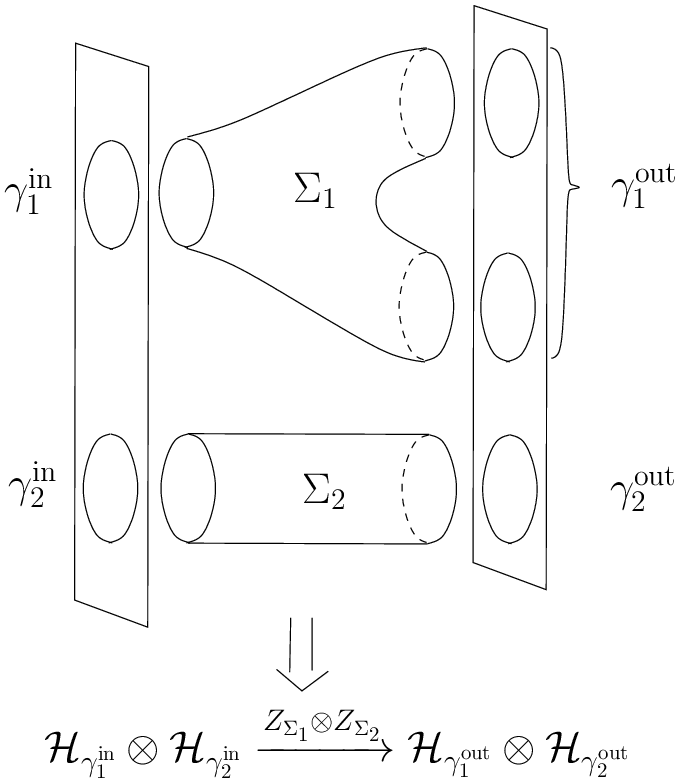} }} $$
\caption{Multiplicativity with respect to disjoint unions.}
\end{figure}
\end{enumerate}

$\bt$ \ul{Sewing axiom}: ``$\cup\ra \circ$'' (sewing of cobordisms is mapped to composition of linear maps).
Given two cobordisms $\gamma_1\xra{\Sigma'} \gamma_2$ and $\gamma_2 \xra{\Sigma''} \gamma_3$  one can sew\footnote{
%comment on gluing subtleties
%If one is just gluing smooth manifolds along boundary
In the case of a topological theory (cobordisms are smooth oriented manifolds with boundary, with no extra geometric structure), one can consider cobordisms modulo diffeomorphisms relative to the boundary, and then sewing is a well-defined operation. In 2d conformal theory, cobordisms are Riemann surfaces with parametrized boundary and the sewing operation, identifying two circles along the parametrization, is also well-defined.
}
 the out-boundary of the first one to the in-boundary of the second one, obtaining a sewn cobordism $\Sigma=\Sigma'\cup_{\gamma_2} \Sigma''$.
\begin{figure}[H]
$$\vcenter{\hbox{ \includegraphics[scale=0.3]{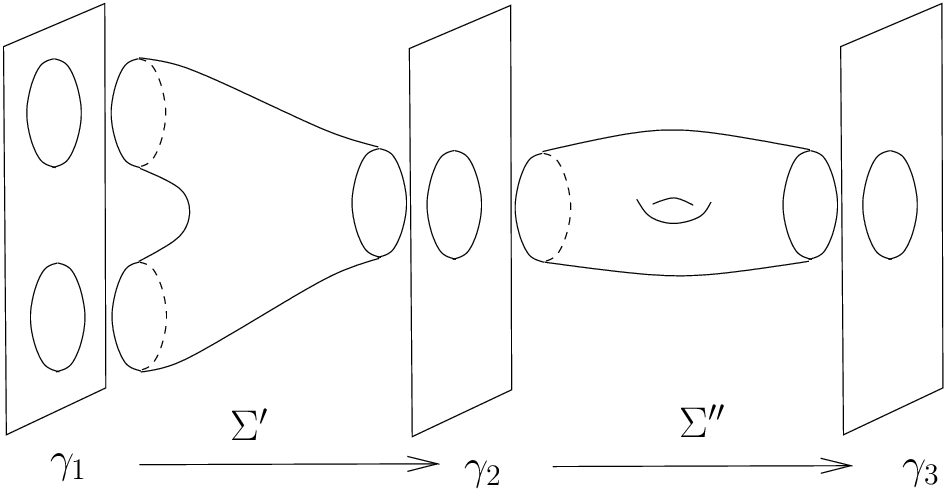} }} 
\quad %\underset{\mr{sewing\; along\;}\gamma_2}{\leftrightarrows} 
\begin{array}{c}
{\tiny \mbox{cutting along $\gamma_2$} }\\
\longleftarrow \\
\longrightarrow \\
{\tiny \mbox{sewing along $\gamma_2$ }}
\end{array}
\quad
\vcenter{\hbox{ \includegraphics[scale=0.3]{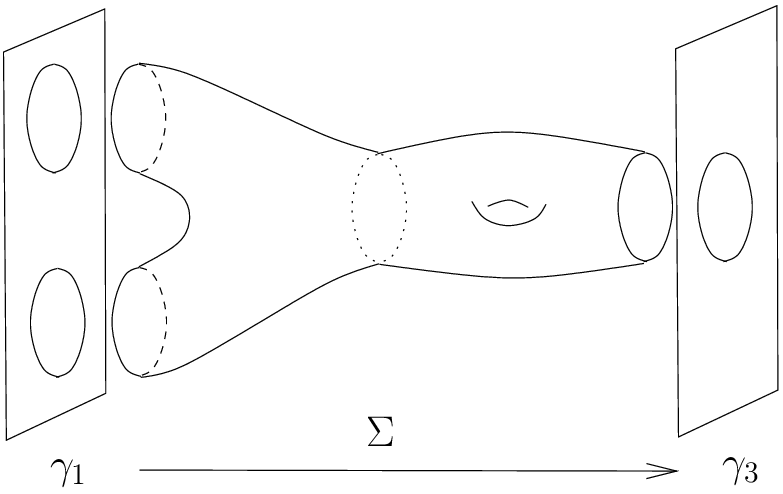} }} 
$$
\caption{Sewing.}
\end{figure}

Then one has
\begin{equation}\label{sewing}
 Z_\Sigma = Z_{\Sigma''}\circ Z_{\Sigma'} 
\end{equation}
or, making domains and codomains explicit, 
$$\HH_{\gamma_3}\xla{Z_\Sigma} \HH_{\gamma_1}=\HH_{\gamma_3}\xla{Z_{\Sigma''}}\HH_{\gamma_2} \xla{Z_{\Sigma'}}\HH_{\gamma_1}.$$

$\bt$ \ul{Normalization}.
\begin{enumerate}[(a)]
\item For the empty $(D-1)$-manifold, one has
$$ \HH_\varnothing= \CC. $$
\item For any closed oriented $(D-1)$ -manifold $\gamma$, the partition function for a ``very short'' cylinder\footnote{The precise meaning of ``very short'' depends on the type of geometric data we put on cobordisms.} 
$\gamma\times [0,\epsilon]$ tends to the identity map on the space of states:
$$ \lim_{\epsilon\ra 0} Z_{\gamma\times [0,\epsilon]}=\mr{id}\colon \HH_\gamma \ra \HH_\gamma $$
\end{enumerate}

\textbf{Additional data.}\\
%\ul{Diffeomorphisms act on spaces of states}.  
\ul{Action of diffeomorphisms}.
For $\phi\colon \gamma\ra\til\gamma$ a diffeomorphism, we have a map 
\begin{equation}\label{action of diffeo on H}
\rho(\phi)\colon \HH_\gamma\ra \HH_{\til\gamma}
\end{equation}
which is linear if $\phi$ is orientation-preserving and is antilinear if $\phi$ is orientation-reversing. Moreover, this is an action, i.e., $\rho(\phi_2\circ \phi_1)=\rho(\phi_2)\circ \rho(\phi_1)$.

\ul{Geometric data}. Cobordisms $\Sigma$ are equipped with \emph{local geometric data} $\xi_\Sigma\in \mr{Geom}_\Sigma$ of type which depends on the particular QFT.\footnote{In our notations, $\mr{Geom}_\Sigma$ is the space of all geometric data of given type on $\Sigma$, and $\xi_\Sigma$ is a particular choice.} Examples of $\xi_\Sigma$:
\begin{enumerate}
\item Riemannian (or pseudo-Riemannian) metric on $\Sigma$. This is the case for many physically relevant QFTs, like, e.g., Yang-Mills theory or electrodynamics.
\item Conformal structure on $\Sigma$ (metric up to rescaling by a positive function). This is the case relevant to us (especially for $D=2$).
\item Nothing. Despite its apparent triviality, actually a very interesting case corresponding to \emph{topological} quantum field theories, in the sense of Atiyah \cite{Atiyah}.
\end{enumerate}
Boundaries $\gamma$ should also be equipped with geometric data, $\xi_\gamma\in \mr{Geom}_\gamma$. E.g., in the cases above the corresponding boundary data is:
\begin{enumerate}
\item A germ of Riemannian bi-collars on $\gamma$ (a germ of Riemannian metrics on $\gamma\times (-\epsilon,\epsilon)$).
\item A parametrization of a boundary circle $\gamma$.
\item Nothing.
\end{enumerate}
The relation between geometric data for cobordisms and for boundaries is that one wants that for a sewn cobordism $\Sigma$, $\mr{Geom}_\Sigma$ is the fiber product $\mr{Geom}_\Sigma=\mr{Geom}_{\Sigma'}\times_{\mr{Geom}_\gamma} \mr{Geom}_{\Sigma''}$. I.e., when we sew cobordisms in the sewing axiom, we also sew the geometric data.

\textbf{Axioms continued.}

\ul{Naturality} (equivariance under diffeomorphisms).

Given a diffeomorphism between cobordisms, $\phi\colon \Sigma\ra \til\Sigma$, %(or, more explicitly, $(\gamma_\ii\xra{\Sigma}\gamma_\oo)\xra{\phi}  (\til\gamma_\ii\xra{\til\Sigma}\til\gamma_\oo)$)
one has a commutative diagram
\begin{equation}\label{naturality}
\begin{CD}
\HH_{\gamma_\ii} @>{Z_{\Sigma,\xi}}>> \HH_{\gamma_\oo} \\
@V{\rho(\phi|_\ii)}VV @VV{\rho(\phi|_\oo)}V  \\
\HH_{\til\gamma_\ii} @>>{Z_{\til\Sigma,\til\xi=\phi_*\xi}}> \HH_{\til\gamma_\oo}
\end{CD}
\end{equation}

\subsection{Remarks.}

\begin{remark}
For a closed $D$-manifold $\varnothing \xra{\Sigma}\varnothing$, the partition function is $Z_\Sigma\colon \CC \xra{\cdot \zeta}\CC$ -- multiplication by some complex number $\zeta$. By abuse of notations, this number $\zeta$ is also called the partition function (and also denoted $Z_\Sigma$).
\end{remark}

\begin{remark}
One may summarize %part of 
the axioms above by saying that a QFT is a functor of symmetric monoidal categories
\begin{equation} \label{Cob -> Vect}
%\mr{Spacetimes}
\mr{Cob}
 \xra{(\HH,Z)} \mr{Vect}_\CC 
\end{equation}
where on the left one has the category of spacetimes  (a.k.a. geometric cobordism category), where:
\begin{itemize}
\item The objects $(\gamma,\xi_\gamma)$ are closed oriented $(D-1)$-manifolds $\gamma$ equipped with geometric structure $\xi_\gamma\in \mr{Geom}_\gamma$.
\item The morphisms  $(\Sigma,\xi_\Sigma)$ are $D$-dimensional oriented cobordisms with geometric structure $\xi_\Sigma\in \mr{Geom}_\Sigma$.
\item Composition is sewing of cobordisms (accompanied by sewing the geometric data).
\item Monoidal tensor product is given by disjoint unions. Monoidal unit is the empty $(D-1)$-manifold.
\item $\mr{Cob}$ is a non-unital category: %(so strictly speaking, not a category): 
it does not have identity morphisms.
 Instead, it has ``almost identity'' morphisms -- short cylinders.\footnote{An exception is the topological case $\mr{Geom}=\varnothing$ where finite cylinders $\gamma\times [0,1]$ play the role of identity morphisms on the nose, without having to approximate identity by a family.}
\end{itemize}
The right hand side of (\ref{Cob -> Vect}) is the category of complex vector spaces and linear maps with obvious monoidal structure given by tensor product.

Naturality axiom says that diffeomorphisms act on the functor $(\HH,Z)$ by natural transformations.

Another way to understand diffeomorphisms categorically is as an enhancement of $\mr{Cob}$ to a bicategory, where the second type of $1$-morphisms is diffeomorphisms of $(D-1)$-manifolds and 2-morphisms are diffeomorphisms of cobordisms. Then naturality says that (\ref{Cob -> Vect}) extends to a functor of bicategories.
%$$  \left\{ 
%\begin{array}{c|c}
%\mbox{spacetime category} & \mr{Ob}\colon (\gamma,\xi_\gamma) \\
%& \mr{Hom}\colon 
%\end{array}
%\right\}  
%$$
\end{remark}

\begin{remark}
It is very interesting to restrict the naturality axiom (\ref{naturality}) to the subgroup $\mr{Sym}_{\Sigma,\xi}\subset \mr{Diff}_\Sigma$ of diffeomorphisms $\phi\colon\Sigma\ra\Sigma$ preserving the chosen geometric data $\xi$, i.e., satisfying $\phi_*\xi=\xi$. Then, (\ref{naturality}) yields \emph{symmetries} of $Z_{\Sigma,\xi}$ (the ``Ward identities''):
\begin{equation} \label{l1 Ward identity}
Z_{\Sigma,\xi}=\rho(\phi|_\oo)\circ  Z_{\Sigma,\xi} \circ \rho(\phi|_\ii)^{-1}. 
\end{equation}
\end{remark}

%\marginpar{? Add remark: $\rho$ is in fact a part of cobordism data (short mapping cylinders) - is it always so or only in TQFT case?}

\begin{remark}\label{l1 rem: def of cob with inclusions}
A careful definition of a $D$-cobordism is as a quintuple $(\Sigma,\gamma_\ii,\gamma_\oo,i_\ii,i_\oo)$ consisting of the following: 
\begin{itemize}
\item $\gamma_\ii$, $\gamma_\oo$ two closed oriented $(D-1)$-manifolds, 
\item $\Sigma$ an oriented $D$-manifold with boundary,
\item two embeddings $i_\ii\colon \gamma_\ii\hookrightarrow \dd\Sigma$,  $i_\oo\colon \gamma_\oo\ra \dd\Sigma$ with disjoint images, such that
\begin{itemize}
\item $\dd\Sigma=i_\ii(\gamma_\ii)\sqcup i_\oo(\gamma_\oo)$,
\item $i_\ii$ is orientation-reversing and $i_\oo$ is orientation-preserving.
\end{itemize}
%an orientation-reversing embedding and $i_\oo\colon \gamma_\oo\ra \dd\Sigma$ an orientation-preserving embedding, such that $\dd\Sigma=i_\ii(\gamma_\ii)\sqcup i_\oo(\gamma_\oo)$.
\end{itemize}

With this definition, one can say that the data of the action of a diffeomorphism $\phi$ on the spaces of states (\ref{action of diffeo on H}) is redundant, as it is already contained in the data of partition functions assigned to cobordisms, as $Z$ for an infinitesimally short mapping cylinder 
\begin{equation}\label{l1 mapping cylinder}
M_\phi=\left(\gamma\times [0,\epsilon],\gamma,\gamma,
\begin{array}{cccc}
i_\ii\colon &\gamma&\hra& \gamma\times [0,\epsilon] \\
& x & \mapsto & (x,0)
\end{array}
, 
\begin{array}{cccc}
i_\oo\colon &\gamma&\hra& \gamma\times [0,\epsilon] \\
& x & \mapsto & (\phi(x),\epsilon)
\end{array}
\right).
\end{equation}

From this viewpoint, the naturality axiom (\ref{naturality}) is a special case of the sewing axiom (when one is attaching two short mapping cylinders to the in-/out-ends of a cobordism).
%\marginpar{Is this right? One needs then to adjoin conjugation $\rho(r)$ as a separate piece of data?}
\end{remark}

\begin{remark} \label{rem: H_-gamma is dual to H_gamma}
One has a natural identification between $\HH_{-\gamma}$ and the linear dual of $\HH_\gamma$, since the partition function of a short cylinder, seen as a cobordism $\gamma\sqcup (-\gamma) \xra{\gamma\times [0,\epsilon]} \varnothing$ yields (in $\epsilon\ra 0$ limit) a bilinear pairing 
\begin{equation}\label{pairing between H_gamma and H_-gamma}
(,)\colon \HH_\gamma\otimes \HH_{-\gamma}\ra \CC,
\end{equation}
which is nondegenerate.\footnote{Nondegeneracy is shown by the following argument. One can consider a second short cylinder $\varnothing \xra{\gamma\times [0,\epsilon]} (-\gamma)\sqcup \gamma$. Attaching $-\gamma$ from the in-boundary of the first cylinder to the $-\gamma$ from the out-boundary of the second cylinder, we obtain a cylinder $\gamma \xra{\gamma\times [0,2\epsilon]} \gamma$ whose partition function converges to identity. That implies that the pairing (\ref{pairing between H_gamma and H_-gamma}) cannot have any kernel vectors.}
\end{remark}

\begin{remark}
Given a cobordism, one can always reassign a connected component of the in-boundary as a component of the out-boundary with reversed orientation. The corresponding partition functions are equal:
$$ Z\Big(\gamma_1\sqcup \gamma_2 \xra{\Sigma} \gamma_3\Big) = Z\Big(\gamma_1 \xra{\Sigma} \gamma_3\sqcup (-\gamma_2)\Big), $$
%(the equality means that the operators are conjuhate to one another, 
using the identification $\HH_{-\gamma_2}=\HH_{\gamma_2}^*$ from Remark \ref{rem: H_-gamma is dual to H_gamma}. In \cite{Segal88} this property is called the ``crossing axiom.''
\end{remark}

\subsection{Unitarity (and its Euclidean counterpart)\label{ss: Segal unitarity}
}
For any $\gamma$ (we are suppressing the geometric data in notation) one has the tautological orientation-reversing mapping $r\colon \gamma\ra -\gamma$ mapping each point to itself. By (\ref{action of diffeo on H}), one has a corresponding antilinear map $\rho(r)\colon \HH_\gamma\ra \HH_{-\gamma}$. Combining it with pairing (\ref{pairing between H_gamma and H_-gamma}), 
one has a sesquilinear form
\begin{equation}\label{<,>}
\langle ,\rangle\colon \HH_\gamma \otimes \HH_\gamma \xra{\rho(r)\otimes \mr{id}} \HH_{-\gamma}\otimes \HH_\gamma \xra{(,)} \CC.
\end{equation}

Unitarity is an \emph{optional} collection of assumptions on a QFT which it might %(or might not) 
satisfy (or not):
\begin{enumerate}[(a)]
\item \label{unitarity (a)} $(\HH_\gamma,\langle,\rangle)$ is a Hilbert space for each $\gamma$.
%\marginpar{mention that it might be good to drop the completeness assumtion and refer to Remark \ref{l23 rem: H^big, H^small}?}
%\footnote{
%In fact, while it is customary to talk about the space of states as a Hilbert space, the completeness assumption 
%} 
In particular, the sesquilinear form $\langle,\rangle$ is positive definite.
%(I.e., the assumption is positivity of (\ref{<,>}) and completeness.)
%\item \label{unitarity (b)}  The partition function of a cobordism $\gamma_1\xra{\Sigma}\gamma_2$, and of its orientation-reversed copy $\gamma_2 \xra{-\Sigma}\gamma_1$ are related by
%$$ Z_{-\Sigma} = \bar{Z}_\Sigma^* $$
%where bar stands for complex conjugation and star is the dual (transpose) map.\footnote{
%In Osterwalder-Schrader axioms, this property is called ``reflection positivity.''  Segal \cite{Segal88} calls it ``*-functor'' property.
%}
\item \label{unitarity (c)} For a cylinder $\gamma\times [0,t]$, the partition function $Z_{\gamma\times [0,t]}$ is a \emph{unitary} operator $\HH_\gamma\ra \HH_\gamma$. 
%\marginpar{this is a special case of the previous?}
%\end{equation}
\item \label{unitarity (d)} The representation of diffeomorphisms on spaces of states (\ref{action of diffeo on H}) is unitary.
\end{enumerate}

We will be studying 2d CFTs in Euclidean signature; they are not unitary theories in the sense above. 
In fact, properties (\ref{unitarity (a)}) and %, (\ref{unitarity (b)}), 
(\ref{unitarity (d)}) may hold for them (in which case one talks about a ``unitary CFT''), but (\ref{unitarity (c)}) fails. Instead, (\ref{unitarity (c)}) gets replaced by its Euclidean counterpart:
\begin{enumerate}[(a)]
\item[(\ref{unitarity (c)}')]   The partition function of a cobordism $\gamma_1\xra{\Sigma}\gamma_2$, and of its orientation-reversed copy $\gamma_2 \xra{-\Sigma}\gamma_1$ are related by
$$ Z_{-\Sigma} = \bar{Z}_\Sigma^*, $$
where bar stands for complex conjugation and star is the dual (transpose) map.\footnote{
In Osterwalder-Schrader axioms, this property is called ``reflection positivity.''  Segal \cite{Segal88} calls it ``*-functor'' property.
}
\end{enumerate}

Note also that if $\dim\HH=+\infty$, (\ref{unitarity (c)}) is incompatible with the trace-class property that one wants to have in a CFT.

\section{Examples of functorial QFTs} 
\marginpar{Lecture 2,\\ 8/26/2022}
\subsection{%Examples of Segal's QFT's: topological case 
TQFTs 
and a silly example}
A functorial QFT with no geometric data on cobordisms and boundaries is a topological quantum field theory in the sense of Atiyah \cite{Atiyah}.
A TQFT assigns to a closed oriented $D$-manifold a complex number $Z_\Sigma\in \CC$ -- invariant of a $D$-manifold up to diffeomorphism, behaving nicely with respect to cutting/gluing.

There are very interesting examples like e.g. $D=3$ Chern-Simons theory. % and semi-interesting examples like Dijkgraaf-Witten model.

\textbf{A silly example.}
%Here is however a silly example: 
For any $D$ we can construct a TQFT with $\HH_\gamma=\CC$ for any $\gamma$ and 
$$Z(\Sigma)= e^{\chi(\Sigma)-\chi(\gamma_\ii)}$$ 
for any cobordism $\gamma_\ii\xra{\Sigma}\gamma_\oo$.\footnote{
Slightly more generally, we can set $Z(\Sigma)= e^{\chi(\Sigma)-\alpha\chi(\gamma_\ii)-\beta\chi(\gamma_\oo)}$ where $\alpha,\beta$ are fixed numbers such that $\alpha+\beta=1$. E.g. one can make a symmetric choice $\alpha=\beta=\frac12$.
} Here $\chi$ is the Euler characteristic. It follows from the additivity of Euler characteristic that axioms of QFT are satisfied (in particular, multiplicativity and sewing).

\subsection{$D=1$ 
Riemannian QFT -- quantum mechanics}
\label{sss: QM}
%Here $D=1$. 
Here objects of the spacetime category ($0$-manifolds) are collections of points with orientation $\pm$.  Fix a vector space $\HH$ and let the space of states for $\mr{pt}^+$ be $\HH_{\mr{pt}^+}\colon=\HH$. Then $\HH_{\mr{pt}^-}=\HH^*$.

Morphisms of the spacetime category ($1$-cobordisms) are collections of oriented intervals and circles equipped with Riemannian metric. Note that naturality axiom implies that the partition function for a cobordism depends only on metric modulo diffeomorphisms, %relative to the boundary. 
i.e., only on lengths of connected components.
%In particular, for an interval it depends only on the length of the interval.
Denote the partition function for an interval  of length $t$ (thought of as a cobordism $\mr{pt}^+\xra{[0,t]} \mr{pt}^+$) by $Z_t\colon \HH\ra \HH$.

Sewing intervals of lengths $t_1$ and $t_2$, we get an interval of length $t_1+t_2$. Thus, the sewing axiom implies the semi-group law
\begin{equation}\label{l2 semi-group law}
Z_{t_1+t_2}=Z_{t_2}\circ Z_{t_1} .
\end{equation}
Assume that we have an improved normalization property:
\begin{equation}\label{l2 improved normalization}
Z_\epsilon \underset{\epsilon\ra 0}{\sim}\mr{id} +A %-\frac{i}{\hbar} \wh{H} 
\epsilon + O(\epsilon^2)
\end{equation}
with $A\in \mr{End}(\HH)$ some linear operator. In physical normalization, one writes $A=-\frac{i}{\hbar}\wh{H}$, then the operator $\wh{H}\in \mr{End}(\HH)$ is called the ``quantum Hamiltonian'' (or ``Schr\"odinger operator''). Together, (\ref{l2 semi-group law}) and (\ref{l2 improved normalization}) imply
\begin{equation}\label{l2 QM evol op}
Z_t=(Z_{\frac{t}{N}})^N = \lim_{N\ra \infty} (\mr{id}+A\frac{t}{N}+O(\frac{1}{N^2}))^N=e^{At} = e^{-\frac{i}{\hbar}\wh{H} t}.
\end{equation}
(I.e., the idea is that we cut a finite interval into $N$ tiny intervals where $Z$ is well-approximated by (\ref{l2 improved normalization}), and then reassemble them using the sewing axiom.)

Formula (\ref{l2 QM evol op}), which we recovered from the axioms of QFT, is the standard expression for the evolution operator in time $t$ in quantum mechanics with quantum Hamiltonian $\wh{H}$. In quantum mechanics, one recovers (\ref{l2 QM evol op}) from Shr\"odinger equation 
\begin{equation}\label{l2 Schroedinger eq}
(i\hbar\,\dd_t+\wh{H})\psi_t=0
\end{equation}
for a $t$-dependent state $\psi_t\in \HH$. Equation (\ref{l2 Schroedinger eq}) implies $\psi_t= Z_t (\psi_0)$, with $Z_t$ given by (\ref{l2 QM evol op}). One may also say that the Schr\"odinger equation itself (\ref{l2 Schroedinger eq}), seen from the standpoint of functorial QFT, expresses the sewing axiom for sewing an infinitesimal interval of length $dt$ to a finite interval of length $t$.

\begin{remark}
Recall that $\HH$ is automatically equipped with a sesquilinear form (\ref{<,>}).
The 1D QFT above is unitary if additionally $\HH, \langle,\rangle$ is a Hilbert space and if $\wh{H}$ is a \emph{self-adjoint} operator, which implies that the evolution operator (\ref{l2 QM evol op}) is unitary. 
\end{remark}

\begin{remark} If we ask $\wh{H}$ to be self-adjoint, but consider evolution in imaginary time $t=-i T_\E$ with $T_\E>0$ (the ``Euclidean time''), then (\ref{l2 QM evol op}) becomes a self-adjoint operator 
\begin{equation}\label{l2 QM Euclidean evol op}
Z=e^{-\frac{T_\E}{\hbar}\wh{H}}
\end{equation}
(instead of unitary) and the theory satisfies (\ref{unitarity (c)}') of Section \ref{ss: Segal unitarity} instead of (\ref{unitarity (c)}). 
%\marginpar{remark on trace-class, compactness,...}
%\marginpar{Comment more on Wick rotation?}
\end{remark}

\begin{remark}
It follows from the sewing axiom that the partition function for a circle of length $t$ is given by the trace of the partition function for the interval of length $t$
\begin{equation}
Z(S^1_t)=\tr_\HH Z_t = \tr_\HH e^{-\frac{i}{\hbar}\wh{H}t}
\end{equation}
\end{remark}

\begin{example}[Quantum mechanics of a free particle on a circle]
%\marginpar{This is a shortened write-up of problem 1 from exercise sheet 1}
\label{l2 example: QM on S^1}
Let $X$ be a circle of length $L$. %Quantum mechanics of a 
Free particle on $X$ is described by the quantum Hamiltonian 
\begin{equation}\label{l2 QM H}
\wh{H}=-\frac{1}{4\pi}\frac{\dd^2}{\dd x^2} 
\end{equation}
acting on the Hilbert space $\HH=L^2(X)$; $x\in \RR/L\cdot\ZZ$ is the coordinate on the circle $X$. Here for simplicity we adopted the units where $\hbar=1$ and the mass of the particle is $2\pi$ (this normalization of the Hamiltonian is chosen in order to have simpler formulae below). 

The partition function for an interval of length $t$ is a unitary integral operator $Z_t=e^{-i\wh{H} t}$ 
with integral kernel
\begin{equation}\label{l2 circle QM K}
K_t(x_1,x_0)=\sum_{n=-\infty}^\infty (i t)^{-\frac12}  e^{\pi i\frac{(x_1-x_0+nL)^2}{t}}.
\end{equation}
The partition function for $\Sigma$ a circle of length $t$ is then
\begin{equation}\label{l2 circle QM via coord trace}
Z(S^1_t)=\tr_\HH Z_t = \int_X dx\, K_t(x,x) = L(i t)^{-\frac12} \sum_{n=-\infty}^\infty e^{\pi i\frac{L^2}{ t} n^2}
\end{equation}
We note that another way to obtain $\tr_\HH Z_t$ is via the eigenvalue spectrum of the Hamiltonian (\ref{l2 QM H}). The eigenfunctions of $\wh{H}$ are $\psi_k=e^{\frac{2\pi ikx}{L}}$ and the corresponding eigenvalues are $E_k=\pi \left(\frac{k}{L}\right)^2$. Thus,
one has 
\begin{equation}\label{l2 circle QM via eval}
Z(S^1_t)=\tr_\HH e^{-i \wh{H} t}=\sum_{k=-\infty}^\infty e^{-i E_k t} = \sum_{k=-\infty}^\infty e^{
-\pi i \frac{t}{L^2} k^2
%\frac{-it}{2}\left(\frac{2\pi}{L}\right)^2 k^2
}
\end{equation}
One can show directly by Poisson summation formula\footnote{\label{l2 footnote: Poisson summation}
Recall that Poisson summation formula says that for a function $f(x)$ on $\RR$ decaying sufficiently fast at $x\ra \pm\infty$,
with $\til{f}(p)=\int_\RR f(x) e^{2\pi i p x} dx$  its Fourier transform,
 one has ${\sum_{n\in \ZZ} f(n)=\sum_{k\in \ZZ} \til{f}(k)}$. 
 One can see this as the equality of distributions $\sum_{n\in \ZZ} \delta(x-n)=\sum_{k\in \ZZ} e^{2\pi i k x}$, integrated against a test function $f$.
} that the right hand sides of (\ref{l2 circle QM via coord trace}) and (\ref{l2 circle QM via eval}) agree; in Poisson summation, the sum over ``winding numbers'' $n$ is transformed into a sum over the dual summation index -- the ``momentum'' $k$.

We note that one can consider the evolution in Euclidean time $t=-i T_\E$ with $T_\E>0$. Then the operator $Z_t$ becomes trace-class and sums (\ref{l2 circle QM via coord trace}), (\ref{l2 circle QM via eval}) become absolutely convergent.

Denoting for convenience $\lambda\colon=\frac{L^2}{T_\E}$ and denoting the partition function for a circle (\ref{l2 circle QM via coord trace}), (\ref{l2 circle QM via eval}) by $\zeta(\lambda)$, we have an interesting transformation property under $\lambda\ra \lambda^{-1}$:
\begin{equation}\label{l2 circle QM T-duality}
\zeta(\lambda)=\lambda^{-\frac12}\zeta(\lambda^{-1})
\end{equation}
This property can be regarded as a very simple instance of the so-called 
$T$-duality (behavior under inversion of the radius of the target circle). Alternatively, if one fixes $L=1$, (\ref{l2 circle QM T-duality}) becomes a toy 1d model of modular invariance in 2d conformal field theory, see (\ref{l2 modular invariance}) below.

\end{example}

\section{Quantum observables in the language of functorial QFT (the idea)}
Fix a functorial QFT. For $\gamma_\ii\xra{\Sigma}\gamma_\oo$ a cobordism, let $\Gamma\subset \Sigma$ be a CW subcomplex disjoint from $\dd \Sigma$.\footnote{It is very interesting to allow $\Gamma$ to go to intersect the boundary of $\Sigma$, but that would lead us into QFTs with corners (known in the topological case, as \emph{extended} TQFTs in the sense of Baez-Dolan-Lurie).} Let consider a family of $\epsilon$-thickenings  $U_\epsilon(\Gamma)$ of $\Gamma$ in $\Sigma$, with $\epsilon\in (0,\epsilon_0)$.\footnote{E.g. we can equip $\Sigma$ with a metric and define $U_\epsilon(\Gamma)$ as the set of points of distance $\leq \epsilon$ from $\Gamma$.}  %a solid tube around $\Gamma$ of radius $\epsilon$.

\begin{figure}[H]
%$$\vcenter{\hbox{ \includegraphics[scale=0.5]{cob1.eps} }} $$
\begin{center}
\includegraphics[scale=0.7]{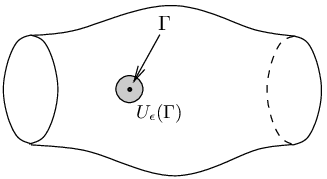}  \qquad
\includegraphics[scale=0.7]{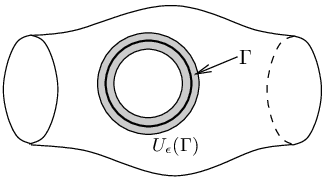}  
\end{center}
\caption{$\epsilon$-thickenings. %$U_\epsilon(\Gamma)$.
}
\end{figure}

%\textcolor{red}{PICTURE: $\epsilon$-thickenings .}

A quantum observable supported on $\Gamma$  is a family (parametrized by $\epsilon\in (0,\epsilon_0)$) of elements 
\begin{equation}\label{l2 obs}
\wh{O}_{\Gamma,\epsilon}\in \HH_{\dd U_\epsilon(\Gamma)}
\end{equation}
I.e. for each $\epsilon$ we have a state  on the boundary of the $\epsilon$-tube around $\Gamma$.

The correlator (or VEV -- ``vacuum expectation value'') of the observable is defined as 
\begin{equation}\label{l2 VEV}
\langle \wh{O}_\Gamma \rangle_\Sigma\colon= \lim_{\epsilon\ra 0} Z_{\Sigma-U_\epsilon(\Gamma)}\circ \wh{O}_{\Gamma,\epsilon} \qquad \in \mr{Hom}(\HH_{\gamma_\ii},\HH_{\gamma_\oo})
\end{equation}
The idea here is that $\Sigma$ with the tube around $\Gamma$ cut out has as its boundaries $\gamma_\ii$, $\gamma_\oo$ and a new piece of boundary -- the boundary of the tube, where we plug the state given by the observable. An important case is when $\Sigma$ is closed (i.e., $\gamma_\ii=\gamma_\oo=\varnothing$). Then the correlator (\ref{l2 VEV}) is a complex number.

The $\epsilon$-dependence  in the family (\ref{l2 obs}) is supposed to be such that the limit in the r.h.s. of (\ref{l2 VEV}) exists. One way to arrange it is to require 
 that elements (\ref{l2 obs}) for different $\epsilon$ are related by 
\begin{equation}\label{l2 rel between epsilons}
\wh{O}_{\Gamma,\epsilon'} = Z_{U_{\epsilon'}(\Gamma)-U_\epsilon(\Gamma)}\circ \wh{O}_{\Gamma,\epsilon}
\end{equation} 
for $0<\epsilon<\epsilon'<\epsilon_0$. In this case the expression  under the limit in (\ref{l2 VEV})
 does not depend on $\epsilon\in (0,\epsilon_0)$ (as follows from the sewing axiom).

%\begin{remark} 
%\begin{itemize}
%\item 
For us, the most important case would be when $\Gamma$ is a collection of points (correlators of point observables). However, in topological and gauge theories it is natural to consider different $\Gamma$s, e.g., Wilson loop observable in Chern-Simons and Yang-Mills theories corresponds to $\Gamma$ an embedded circle in $\Sigma$; its generalization -- Wilson graph -- corresponds to $\Gamma$ an embedded graph in $\Sigma$.
%\item Wilson loop observables, with $\Gamma$ a circle embedded in $\Sigma$ are interesting in Chern-Simons theory (a $D=3$ TQFT) and in Yang-Mills theory (the original setting for Wilson loop observables).
%\item Wilson graph observables, with $\Gamma$ an embedded graph in $\Sigma$, are interesting in Chern-Simons theory.
%\end{itemize}
%\end{remark}

\subsection{Example: point observables in quantum mechanics}\label{sss point observables in QM}
In the setting of Section \ref{sss: QM} -- quantum mechanics as 1d QFT -- consider the cobordism $\Sigma=[t_\ii,t_\oo]$ and consider an observable supported at a single point $\Gamma=\{t\}$ inside $\Sigma$. As the thickening we can take small intervals 
$$U_\epsilon(\Gamma)=[t-\epsilon,t+\epsilon].$$ 
The boundary of the thickening is a pair of points of opposite orientation 
$$\dd U_\epsilon(\Gamma)=\mr{pt}^-\sqcup \mr{pt^+}.$$ 
Thus, a quantum observable is an element 
\begin{equation}
\wh{O}\in \HH_{\dd U_\epsilon(\Gamma)}=\HH_{\pt^-\sqcup\pt^+}= \HH^*\otimes \HH\cong \mr{End}(\HH)
\end{equation}
-- an operator on the space of states $\HH$.
\begin{figure}[H]
%$$\vcenter{\hbox{ \includegraphics[scale=0.5]{cob1.eps} }} $$
\begin{center}
\includegraphics[scale=0.8]{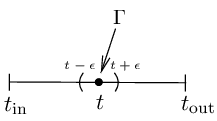}  
\end{center}
\caption{Point observable in quantum mechanics. %$U_\epsilon(\Gamma)$.
}
\end{figure}

We can similarly consider several point observables on $\Sigma$, supported at $\Gamma=\{t_1,\ldots,t_n\}$ (we assume that $t_\ii<t_1<t_2<\cdots<t_n<t_\oo$). The picking a state on the boundary of $\epsilon$-thickening of each point amounts to choosing a collection of operators $\wh{O}_1,\ldots,\wh{O}_n\in \End(\HH)$. The correlator (\ref{l2 VEV}) then is
$$ \langle \wh{O}_1(t_1)\cdots \wh{O}_n(t_n) \rangle_\Sigma = e^{-\frac{i}{\hbar}\wh{H}(t_\oo-t_n)}\wh{O}_n\cdots e^{-\frac{i}{\hbar}\wh{H}(t_2-t_1)} \wh{O}_1 e^{-\frac{i}{\hbar}\wh{H}(t_1-t_\ii)} $$

\begin{figure}[H]
%$$\vcenter{\hbox{ \includegraphics[scale=0.5]{cob1.eps} }} $$
\begin{center}
\includegraphics[scale=0.7]{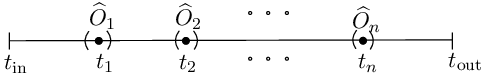}  
\end{center}
\caption{Correlator of several point observables in quantum mechanics. %$U_\epsilon(\Gamma)$.
}
\label{l2_fig6}
\end{figure}

\section{2d conformal field theory as a functorial QFT}
In the main case of interest for us -- two-dimensional conformal field theory -- the geometric structure on cobordisms is conformal structure (Riemannian metric up to rescaling by a positive function), plus orientation; in two dimensions this data is equivalent\footnote{
See Section \ref{ss 2d conf str = cx str} below.
%We will come to this later.
} to complex structure. Thus,
cobordisms are (possibly disconnected) Riemann surfaces  with parametrized boundary circles (when sewing in- and out-circles, one should respect the parametrization -- points with the same angle parameter are identified).\footnote{Parametrization of boundary circles can be seen in terms of Remark \ref{l1 rem: def of cob with inclusions} as the embeddings $i_\ii,i_\oo$ of unions of standard circles into $\dd\Sigma$.} 
Parametrization of boundaries is needed for the sewn surface to have a well-defined complex structure.\footnote{E.g. sewing the two boundary circles of a cylinder with a twist by angle $\theta$, one obtains non-equivalent complex tori for different $\theta$.}

Such a Riemann surface with $n$ in-circles and $m$ out-circles, $\sqcup_{i=1}^n S^1 \xra{\Sigma} \sqcup_{j=1}^m S^1$, is assigned a linear map $Z(\Sigma)\colon \HH_{S^1}^{\otimes n}\ra \HH_{S^1}^{\otimes m}$.

\begin{figure}[H]
%$$\vcenter{\hbox{ \includegraphics[scale=0.5]{cob1.eps} }} $$
$$\mr{(a)}\vcenter{\hbox{ 
\includegraphics[scale=0.4]{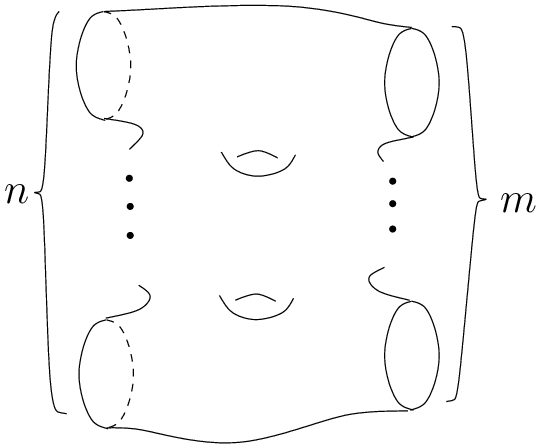}
 }}  
 \quad  \mr{(b)}
 \vcenter{\hbox{ 
 \includegraphics[scale=0.5]{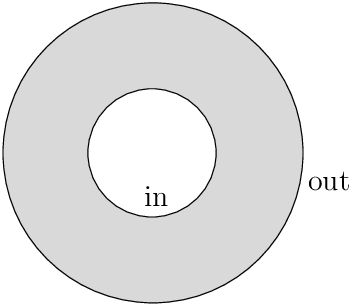}
}}
 \quad  \mr{(c)}
 \vcenter{\hbox{ 
 \includegraphics[scale=0.5]{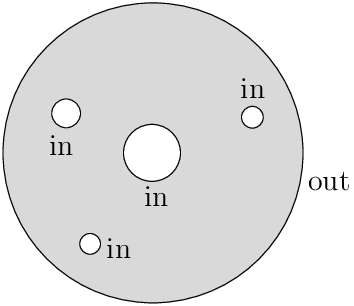}
}}
 $$
\caption{(a) a generic cobordism in 2d CFT and some relevand cobordisms embedded in $\CC$ -- (b) annulus  (coformally equivalent to a cylinder) and (c) a 2d equivalent of Figure \ref{l2_fig6} (corresponding to several point observables). %$U_\epsilon(\Gamma)$.
}
\label{l2_fig7}
\end{figure}

The space of states for a circle $\HH_{S^1}$ is a Hilbert space carrying a representation of the group of diffeomorphisms $\mr{Diff}(S^1)$,
\begin{equation}\label{l2 rho}
\rho\colon \mr{Diff}(S^1) \ra \mr{End}(\HH_{S^1}).
\end{equation} 

\ul{Vacuum vector}. The space $\HH_{S^1}$ contains a distinguished vector 
\begin{equation}\label{l2 vac}
|\mr{vac}\rangle \in \HH_{S^1}
\end{equation} 
-- ``vacuum vector'' -- the partition function of the %standard unit 
disk.\footnote{This vector is not invariant under $\mr{Diff}(S^1)$. However, as a consequence of naturality, it is invariant under the $3$-dimensional subgroup (isomorphic, via identifying the  disk with upper half-plane, to $PSL_2(\RR)$ -- real M\"obius transformations) consisting of diffeomorphisms of $S^1$ which can be extended as conformal transformations over the whole disk.
} %(with standard parametrization of the boundary). 
In (b), (c) of Figure \ref{l2_fig7}, pairing with $|\mr{vac}\rangle$ for any of the in-boundaries corresponds to removing (or filling in with the disk) the corresponding hole.

\ul{Self-sewing}. If the surface $\til\Sigma$ is obtained from $\Sigma$ by gluing $i$-th in-circle to $j$-th out-circle, one has
\begin{equation} \label{l2 self-sewing}
Z(\til\Sigma)=\mr{tr}_\HH Z(\Sigma)
\end{equation} 
Here on the right hand side we mean a partial trace -- the trace taken in the first factor of 
$$Z(\Sigma)\in \Hom\left(\HH_{S^1_{\ii,i}},\HH_{S^1_{\oo,j}}\right) \bigotimes \Hom\left(\bigotimes_{1\leq k \leq n, k\neq i}\HH_{S^1_{\ii,k}},\bigotimes_{1\leq l \leq m, l\neq j}\HH_{S^1_{\oo,l}}\right). $$

Self-sewing formula (\ref{l2 self-sewing}) is not an extra axiom -- it follows from the usual sewing axiom by attaching an infinitesimally short cylinder to $S^1_{\ii,i}$ and $S^1_{\oo,j}$.

%\ul{On the trace-class property.}
In particular, 
%in order to have a CFT defined on surfaces of all genera, 
traces (\ref{l2 self-sewing}) must exist if we have %is order to have 
a full CFT.\footnote{One may %entertain the idea of 
consider
a partial CFT where partition functions are only defined on genus zero cobordisms. In that case one can make do with partition function for which traces do not exist. An example of such a model is massless scalar field with values in $\RR$; the variant with values in $S^1$ (a.k.a. ``compactified free boson'') is a full CFT existing in all genera.}
Segal in \cite{Segal88} imposes a slightly stronger condition that traces exist in the sense of absolute convergence, i.e., that partition functions are \emph{trace-class} operators.

\subsection{Genus one partition function, modular invariance}\label{sss: modular invariance}
Given a complex number $\tau\in \CC$ with $\mr{Im}\,\tau>0$, one can consider the Riemann surface 
\begin{equation}\label{l2 torus}
\bb{T}_\tau\colon = %\frac{\CC}{\ZZ\oplus \tau \ZZ}
\CC/(\ZZ\oplus \tau \ZZ)
\end{equation}
-- the quotient of $\CC$ equipped with standard complex structure by a lattice;
%(where $\CC$ is equipped with standard complex structure) -- 
(\ref{l2 torus}) is the complex torus with modular parameter $\tau$.

One can evaluate the CFT on $\bb{T}_\tau$. Denote 
$$Z(\tau)\colon= Z(\bb{T}_\tau)\in \CC$$
Then since tori $\bb{T}_\tau$ and $\bb{T}_{-1/\tau}$ are equivalent as complex manifolds (via the holomorphic map $z\mapsto z/\tau$), $Z_\tau$ as a function of $\tau$ possesses modular invariance property
\begin{equation}\label{l2 modular invariance}
 Z(\tau)=Z(-\frac{1}{\tau}) 
\end{equation}
Also, tori $\bb{T}_\tau$ and $\bb{T}_{\tau+n}$ are equivalent for any $n\in \ZZ$ (via the tautological map $z\mapsto z$), hence one also has $Z(\tau+n)=Z(\tau)$.

In particular one can consider the torus (\ref{l2 torus}) with $\tau=i T$, $T>0$, as obtained from a cylider $\Sigma=S^1\times [0,T]$ (we think of $S^1$ as having length $1$) by sewing the out-end to in-end. CFT restricted to cylinders can be regarded as quantum mechanics with partition functions 
\begin{equation}\label{l2 Z cylinder}
Z(S^1\times [0,T])=e^{-2\pi T \wh{H}}
\end{equation}
for some self-adjoint operator $\wh{H}\in \End(\HH_{S^1})$ -- the Hamiltonian, cf. Section \ref{sss: QM}.\footnote{Here we are considering evolution in ``Euclidean time'' $T$, cf. (\ref{l2 QM Euclidean evol op}). We also set $\hbar=1$. The factor $2\pi$ in the exponential is a choice of normalization of the Hamiltonian and is put there 
%so that (\ref{l2 Z(theta+iT)}) below is compatible 
for compatibility with standard CFT conventions.
}
Then by (\ref{l2 self-sewing}) we have
\begin{equation}\label{l2 Z(iT)}
Z(iT) =  \tr_{\HH_{S^1}} e^{-2\pi T\wh{H}}
\end{equation}
As a function of $T$, (\ref{l2 Z(iT)}) has to be invariant under inversion $T\leftrightarrow \frac{1}{T}$, as a special case of (\ref{l2 modular invariance}).

The general torus (\ref{l2 torus}), with $\tau = \frac{\theta}{2\pi}+iT$ can be obtained from (\ref{l2 Z cylinder}) by identifying boundary circles with a twist by  the angle $\theta$:
\begin{equation*}
\bb{T}_\tau= \frac{S^1\times [0,T]}{(\sigma,0)\sim (\sigma+\frac{\theta}{2\pi},T), \;  \sigma\in S^1}
\end{equation*}
By sewing and naturality axioms, the corresponding partition function is
\begin{equation}\label{l2 Z(theta+iT)}
Z(\tau) = \tr_{\HH_{S^1}} e^{-2\pi T \wh{H} +i\theta\wh{P}}
\end{equation} 
where $\wh{P}\in \End(\HH)$ is the infinitesimal generator of the action of the subgroup of rigid  rotations $S^1\subset \mr{Diff}(S^1)$ on $\HH_{S^1}$ (in particular, $\wh{P}$ is a self-adjoint operator with integer eigenvalues).

\subsection{Correction to the picture: conformal anomaly} \label{sss: conf anomaly}
Conformal field theories one constructs in reality satisfy the axioms of QFT in a weakened -- ``\emph{projective}'' -- sense: 
\begin{itemize}
\item The representation of $\mr{Diff}(S^1)$ on $\HH_{S^1}$ is projective. Put another way, there is an honest representation of a central extension $\wh{\mr{Diff}(S^1)}$ of the group of diffeomorphisms on $\HH_{S^1}$. This central extension is known as the Virasoro group.%\marginpar{Q: Naturality -- strict or projective? \textcolor{blue}{A: strict.}}
%\footnote{
%Remark: naturality axiom when restricted to conformal holds strictly, not projectively: the central element acts trivially.
%}
%\item Naturality axiom holds projectively, i.e. (\ref{naturality}) commutes up to a factor in $\CC^*$.
\item Sewing  axiom (\ref{sewing}) holds up to a factor in $\CC^*$. -- One says that (\ref{Cob -> Vect}) is a \emph{projective functor}. Equivalently, one can say that partition functions are operators in $\Hom(\HH_\ii,\HH_\oo)$ defined up to scaling by a factor in $\CC^*$.
%: cobordisms $\Sigma$ are equipped with a complex line $L_\Sigma$ (which depends on the geometric structure on $\Sigma$) 
\end{itemize}

As another viewpoint, one can understand a projective functor as a strict functor out of a central extension of the cobordism category (see \cite{Segal88}). This is equivalent to saying that $Z_\Sigma$ is not a function on $\mr{Geom}_\Sigma$ but rather a section of a line bundle on it.\footnote{\label{l2 footnote: conf anomaly bundle}
More explicitly, in CFT this line bundle is $\mc{L}^{\otimes c}\otimes \bar{\mc{L}}^{\otimes \bar{c}}$ as a bundle over the moduli space of complex structures on $\Sigma$. Thus, $Z_\Sigma\in \Hom(\HH_\ii,\HH_\oo)\otimes \mc{L}^{\otimes c}\otimes \bar{\mc{L}}^{\otimes \bar{c}}$. Here $\mc{L}=\mr{Det}(\bar\dd)$ is the Quillen line bundle -- the determinant line bundle of the Dolbeault operator; $\bar{\mc{L}}$ is the complex conjugate one; $c$ and $\bar{c}$ are numbers -- holomorphic and anti-holomorphic central charges. Usually one has $c=\bar{c}$ (this is the case assumed in (\ref{l2 change of Z under Weyl}) below).
}

Yet another viewpoint on the projectivity phenomenon is that CFT partition functions are well-defined as operators for a given Riemannian metric $g$ on the surface $\Sigma$, but if one changes the metric within its conformal class, $g\mapsto \Omega\cdot g$, with Weyl factor $\Omega=e^{2\sigma}$, then the partition function scales by a complex factor:
\begin{equation}\label{l2 change of Z under Weyl}
Z(\Sigma,e^{2\sigma} g) = e^{ic S_\mr{Liouville}(\sigma,g)} \cdot Z(\Sigma,g).
\end{equation}
Here $c$ is a number (the ``strength'' of the projectivity effect), known as the \emph{central charge} of the CFT; 
$$ S_\mr{Liouville}(\sigma,g)=\int_\Sigma \frac12 (d\sigma\wedge *d\sigma + 4\sigma R_g\, dvol_g)  $$
is the ``Liouville action,''  $R_g$ is the scalar curvature.

\section{Heuristic motivation for the axioms of QFT from path integral quantization} \label{ss Segal from path integral}
\marginpar{Lecture 3,\\ 8/29/2022}

A classical (Lagrangian) field theory on a cobordism $\gamma_\ii\xra{\Sigma}\gamma_\oo$ is determined by the following data:
\begin{enumerate}[(a)]
\item The space of fields on $\Sigma$,
$$ \F_\Sigma = \Gamma(\Sigma, E) $$
-- the space of smooth sections of a fiber bundle $E$ over $\Sigma$ -- the bundle of fields. (For instance, fields could be maps from $\Sigma$ to some target manifold $X$, or fields could be differential forms on $E$.)
\item The action functional -- a real-valued function on the space of fields of the form
\begin{equation} \label{l3 S}
S_\Sigma(\phi)=\int_\Sigma L(\phi, \dd \phi,\cdots) \qquad \in \RR 
\end{equation}
where $\phi\in \F_\Sigma$ is a field.
Here $L$ (the Lagrangian) is a $D$-form (or density) on $\Sigma$, depending on the field $\phi$ in a local way: the value of $L$ at a point $x\in\Sigma$ can depend only on the value of $\phi$ at $x$ and its derivatives up to a finite order at  $x$.\footnote{
It is convenient (see \cite{Anderson} for details) to consider the ``variational bicomplex''  $\Omega^{p,q}_\mr{loc}(\Sigma\times \F_\Sigma)$ of $(p,q)$-forms on $\Sigma\times \F_\Sigma$ local in the same sense. In this terms, the Lagrangian $L$ is in $\Omega^{D,0}_\mr{loc}(\Sigma\times \F_\Sigma)$.
}
\end{enumerate}

Given the data above, at the classical level one is interested in the solutions $\phi\in \F_\Sigma$ of the ``equations of motion'' -- the critical point equation
\begin{equation}\label{l3 delta S=0}
\delta S=0
\end{equation}
(with $\delta$ the de Rham operator on $\F_\Sigma$).
One considers the equation (\ref{l3 delta S=0}) (which is the Euler-Lagrange PDE) with boundary conditions on the field at $\gamma_\ii,\gamma_\oo$. In a range of cases (Lagangians of second order in derivatives), one can consider the boundary conditions of the form
\begin{equation}\label{l3 bc}
\phi|_{\gamma_\ii} = \phi_\ii,\quad  \phi|_{\gamma_\oo} = \phi_\oo
\end{equation}
where $\phi_\ii=\F_{\gamma_\ii}$ and $\phi_\oo\in \F_{\gamma_\oo}$ -- fixed sections of the bundle $E$ over the boundaries $\gamma_\ii$ and $\gamma_\oo$, respectively.

\begin{remark}\label{l3 rem general bc}
One can consider more general boundary conditions on $\gamma$ of the form 
\begin{equation}\label{l3 general bc}
\pi(\mr{Jet}(\phi)|_\gamma)=b_\gamma
\end{equation} 
where $\mr{Jet}(\phi)|_\gamma$ is the normal $\infty$-jet of $\phi$ at $\gamma$; 
%(storing the information on the pullback of $\phi$ to $\gamma$ and on normal derivatives of all orders of $\phi$ at $\gamma$).
$$\pi\colon %\Gamma(\gamma,i^*\mr{Jet}_\infty E) 
\{\mr{normal\;jets\;of\;fields\;at\;}\gamma\}
\ra B_\gamma$$ 
is some fibration and $b_\gamma\in B_\gamma$ a point in the base. The desired scenario is when the solution of (\ref{l3 delta S=0}) with boundary condition  (\ref{l3 general bc}) exists and is locally-unique (non-deformable).
%here $i\colon \gamma\hra \Sigma$ is the inclusion of the boundary. 
\end{remark}

\begin{example}[Classical mechanics of a particle on %$\RR^N$
a Riemannian manifold
]  \label{l3 exa: particle}
Let $D=1$. Fix a Riemannian manifold $(M,g)$ (target), 
a positive number $m$ (mass) and a function $V\in C^\infty(M)$ (the force potential). Consider as the cobordism the interval $\Sigma=[0,t]$ and set
\begin{equation}
\F_\Sigma=\mr{Map}([0,t],M)
\end{equation}
-- the space of paths in $M$ %(the target) 
parametrized by the interval $[0,t]$.
We set the action $S\colon \F
%\mr{Map}(\Sigma,M)
\ra \RR$ to be defined by
\begin{equation}\label{l3 particle action}
S_\Sigma(\phi)=\int_0^t d\tau \left(\frac{m}{2}g_{\phi(\tau)}(\dot\phi(\tau),\dot\phi(\tau))-V(\phi(\tau))\right)
\end{equation}
for $\phi\colon [0,t]\ra M$ a field (a path).\footnote{Note that the Riemannian metric on the source cobordism $\Sigma$ is implicitly used in (\ref{l3 particle action}): the action (\ref{l3 particle action}) is not invariant under reparametrization of a path.
One can also write a $\mr{Diff}(\Sigma)$-invariant version of the action (\ref{l3 particle action}):
$$ S_{\Sigma,\xi}(\phi)=\int d\tau  \sqrt{\xi(\tau)} \left(\xi(\tau)^{-1}\frac{m}{2}g_{\phi(\tau)}(\dot\phi(\tau),\dot\phi(\tau))-V(\phi(\tau))\right) $$
Here $\xi(\tau) d\tau^2$ is the metric on $\Sigma$. Then for $\psi\colon \Sigma\ra \Sigma$ a diffeomorphism, one has $S_{\Sigma,\xi}(\phi)=S_{\Sigma,\psi^*\xi}(\psi^*\phi)$.
}

Setting for simplicity $(M,g)=\RR^N$ with standard Euclidean metric, the critical point equation $\delta S=0$ is equivalent to the ODE
\begin{equation}\label{l3 Newton eq}
m\ddot\phi(\tau)+\mr{grad}V (\phi(\tau)) =0
\end{equation}
-- the Newtonian equation of motion of a particle in $\RR^N$ in the force field with potential $V$. One can consider this equation with Dirichlet boundary conditions $\phi(0)=\phi_\ii$, $\phi(t)=\phi_\oo$ where $\phi_\ii$, $\phi_\oo$ -- two given points in $\RR^N$. Thus, we are considering parametrized paths in $\RR^N$ satisfying the equation (\ref{l3 Newton eq}) with \emph{fixed endpoints}. E.g. if $V=0$, there is a unique solution -- the straight interval connecting $\phi_\ii$ and $\phi_\oo$ with constant-velocity parametrization by $[0,t]$.

If we take a general Riemannian manifold $(M,g)$ and set $V=0$, then $\delta S=0$ is equivalent to the geodesic equation. So, solutions of the boundary problem (\ref{l3 delta S=0}), (\ref{l3 bc}) are the geodesics in $M$ connecting the two given points.

For a general Riemannian manifold $(M,g)$ and general potential $V$, the Euler-Lagrange equation for the action (\ref{l3 particle action}) (the critical point equation $\delta S=0$) written in local coordinates on %the target 
$M$ takes the form
\begin{equation}
m\left(\frac{d^2 \phi^i(\tau)}{d\tau^2}+\Gamma^i_{jk}(\phi(\tau)) \frac{d\phi^j(\tau)}{d\tau} \frac{d\phi^k(\tau)}{d\tau}\right) + g^{ij}(\phi(\tau)) \dd_i V (\phi(\tau)) =0
\end{equation}
where $\Gamma^i_{jk}$ are the Christoffel symbols.
\end{example}

\begin{example}[Scalar field]\label{l3 exa: scalar}
Let $D\geq 1$ be any, fix $m\geq 0$ (the mass) and fix some polynomial function $V$ on $\RR$ (interaction potential). Consider a cobordism $\Sigma$ equipped with Riemannian metric and set 
\begin{equation}
\F_\Sigma=\mr{Map}(\Sigma,\RR)
\end{equation}
and 
\begin{equation}
S_\Sigma=\int_\Sigma \frac12 d\phi\wedge *d\phi +\frac{m^2}{2}\phi^2 d\mr{vol} +V(\phi)d\mr{vol}
\end{equation}
with $\phi\colon \Sigma\ra \RR$ the scalar field on $\Sigma$.
The corresponding equation of motion $\delta S=0$ is equivalent to the PDE
\begin{equation}\label{l3 scal field eom}
\Delta \phi+ m^2 \phi + V'(\phi)=0
\end{equation}
with $\Delta$ the Laplacian on functions on $\Sigma$.  Equation (\ref{l3 scal field eom}) is the Laplace equation if $m=0$, $V=0$, Helmholtz equation if $m\neq 0$, $V=0$; for general $V$, it is a nonlinear PDE. Equation (\ref{l3 scal field eom}) can be considered with Dirichlet boundary conditions (\ref{l3 bc}) where $\phi_{\ii,\oo}$ are fixed functions on $\gamma_{\ii,\oo}$.

\end{example}

\subsection{Path integral quantization} 
Given a classical field theory on $\Sigma$, we want to define a corresponding QFT. Consider the following expression depending on $\phi_\ii$, $\phi_\oo$ -- sections of $E$ over $\gamma_{\ii,\oo}$:
\begin{equation}\label{l3 PI}
K_\Sigma(\phi_\oo,\phi_\ii)\colon= \underset{\scriptsize
\begin{array}{c}
\phi\in \F_\Sigma\mr{\; s.t.} \\
\phi|_{\gamma_\ii}=\phi_\ii,\\
\phi|_{\gamma_\oo}=\phi_\oo
\end{array}
}{\int}
\mc{D} \phi\; e^{-S_\Sigma(\phi)}
\end{equation}
The right hand side is a formal expression -- the integral over the (infinite-dimensional) space of fields on $\Sigma$ subject to boundary conditions; the ``measure'' $\mc{D}\phi$ on fields is a formal symbol. 
%\marginpar{Comment on minus vs $i/\hbar$ conventions (Euclidean/stat. physics vs unitary Lorentzian QFT)}

\begin{remark}\label{l3 rem hbar}
 Depending on the context, there are different normalizations of the exponential in (\ref{l3 PI}):
\begin{itemize}
\item In unitary (or ``relativistic'') quantum field theory on a Lorentzian spacetime manifold, one writes the integrand of (\ref{l3 PI}) as $e^{\frac{i}{\hbar}S(\phi)}$.
\item In statistical mechanics one writes the integrand as $e^{-\beta E(\phi)}$ (the Gibbs measure on states of the statistical system), with $\phi $ a state of the system on $\Sigma$, $E(\phi)$ the energy of the state and $\beta=\frac{1}{T}$ the inverse temperature.
%\footnote{Statistical mechanics terminology here contradicts QFT terminology where states live on the boundary. So, (states in stat. mechanics) correspond to (fields in QFT), not (states in QFT).}
Summarizing the comparison between QFT and statistical mechanics, we have the following.

\vspace{0.2cm}
\hspace{-2cm}\begin{tabular}{c|c}
QFT & statistical mechanics \\ \hline \hline
field $\phi$ on $\Sigma$ & state $\phi$ of the system on $\Sigma$ \\
action functional $S$ & energy functional $E$ \\  \hline
path integral $\int \mc{D}\phi\; e^{\frac{i}{\hbar} S(\phi)}$ & sum over states $\int\mc{D}\phi\; e^{-\beta E(\phi)}$ \\
at $\hbar\ra 0$: fast oscillating integrand &  at temperature$\ra 0$: integrand is supported near \\
stationary phase point = classical solution & the state with minimal energy
\end{tabular}
\vspace{0.3cm}

\item In Euclidean field theory (which will be our setting for 2d CFT), on a Riemannian (as opposed to pseudo-Riemannian) spacetime manifold $\Sigma$, one considers the path integral with the integrand $e^{-\beta S(\phi)}$ where $\beta=\frac{1}{\hbar}$ and -- unless we want to do perturbation theory yielding a power series in $\hbar$ -- we can choose to set $\beta=\hbar=1$.
\end{itemize}
One can transition a unitary QFT on a cobordisms of cylinder type $\gamma\times [0,t]$ to a Euclidean field theory on $\gamma\times [0,T_\E]$ by ``Wick rotation'' -- analytical continuation in $t$ to $t=-i T_\E$.
\end{remark}

\begin{remark}
There are ways to make mathematical sense of the path integral (a.k.a. functional integral or Feynman integral) (\ref{l3 PI}), like e.g. 
\begin{enumerate}[(a)]
\item perturbative approach  -- expansion in Feynman diagrams (replacing the path integral by its stationary phase or Laplace approximation), or
\item  lattice approach -- replacing $\Sigma$ with a lattice with the field defined at the nodes -- then (\ref{l3 PI}) is replaced by a finite-dimensional integral; after that one needs to take the limit of the lattice spacing going to zero (one should think of this procedure as an analog of a Riemannian sum for an ordinary integral).
\end{enumerate}
\end{remark}

We define the space of states of the QFT on $\gamma$ as 
\begin{equation}\label{l3 H}
\HH_\gamma=\mr{Fun}_\CC(\F_\gamma)
\end{equation}
the space of complex-valued functions  on $\F_\gamma$ (the space parametrizing the possible boundary conditions in (\ref{l3 bc})).\footnote{For more general boundary conditions of the type (\ref{l3 general bc}), instead of (\ref{l3 H}) we should write $\mr{Fun}_\CC(B_\gamma)$. Occurrences on $\F_\gamma$ as integration space throughout this subsection (such as e.g. in (\ref{l3 Z via K})) should then also be swapped for $B_\gamma$.}

For instance, in Example \ref{l3 exa: particle}, one would set $\HH_\mr{pt}=\mr{Fun}_\CC(M)$. If we want to have unitarity, then we should be more specific about regularity of allowed functions  and ask that it is of $L^2$ class:
$$\HH_\mr{pt}=L^2(M)$$
-- the standard Hilbert space in the quantum mechanical system consisting of a particle moving on $M$. By extension, it is tempting to write (\ref{l3 H}) as $\HH_\gamma=L^2(\F_\gamma)$.

We define the partition function of the cobordism $\Sigma$ using the path integral (\ref{l3 PI}) as follows: for $\Psi_\ii\in \mr{Fun}_\CC(\F_{\gamma_\ii})$, we set
\begin{equation}\label{l3 Z via K}
(Z_\Sigma \Psi_\ii)(\phi_\oo)\colon = \int_{\F_{\gamma_\ii}} \mc{D}\phi_\ii \; K_\Sigma(\phi_\oo,\phi_\ii) \Psi_\ii(\phi_\ii)
\end{equation}
In other words, $Z_\Sigma$ is an integral operator, determined by the \emph{integral kernel} $K_\Sigma$ defined by the path integral (\ref{l3 PI}).

\begin{remark}[Dirac's bra- and ket-notations]
One can consider a basis in $\HH_\gamma$ consisting of vectors $\{|\phi_0\rangle\}_{\phi_0\in \F_\gamma}$. The vector $|\phi_0\rangle$ is understood as corresponding to the delta-function on the space $\F_\gamma$ centered at $\phi=\phi_0$.  In particular, a vector $\Psi\in \HH_\gamma$ can be written tautologically as 
$$\Psi=\int_{\F_\gamma}\mc{D}\phi_0\, \Psi(\phi_0) |\phi_0\rangle $$
 Likewise, one has a dual basis in $\HH^*$ consisting of covectors $\{\langle \phi_0 |\}_{\phi_0\in\F_\gamma}$. In terms of these notations, it is natural to denote the integral kernel (\ref{l3 PI}) by
\begin{equation}
\langle \phi_\oo | Z_\Sigma|\phi_\ii\rangle \colon= K_\Sigma(\phi_\oo,\phi_\ii)
\end{equation}
One also calls this expression the ``matrix element'' of $Z_\Sigma$ (corresponding to ``row'' $\phi_\oo$ and ``column'' $\phi_\ii$).
\end{remark}

\subsection{Sewing as Fubini theorem for path integrals}
Let $\gamma_1\xra{\Sigma'}\gamma_2$ and $\gamma_2\xra{\Sigma''}\gamma_3$ be two cobordisms and $\gamma_1\xra{\Sigma}\gamma_3$ the corresponding sewn cobordism. Then we have
\begin{multline}\label{l3 Fubini computation}
\langle \phi_3 | Z_\Sigma |\phi_1 \rangle = \underset{
\scriptsize
\begin{array}{c}
\phi\in \F_\Sigma\mr{\; s.t.} \\
\phi|_{\gamma_1}=\phi_1,\\
\phi|_{\gamma_3}=\phi_3
\end{array}
}{\int} \mc{D}\phi\;e^{-S_\Sigma(\phi)} 
\\\underset{\mr{Fubini}}{=}
\underset{\scriptsize \phi_2\in\F_{\gamma_2}}{\int} \mc{D}\phi_2 
\underset{\scriptsize
\begin{array}{c}
\phi'\in \F_{\Sigma'}\mr{\; s.t.} \\
\phi|_{\gamma_1}=\phi_1,\\
\phi|_{\gamma_2}=\phi_2
\end{array}}{\int} \mc{D}\phi' 
\underset{\scriptsize
\begin{array}{c}
\phi''\in \F_{\Sigma''}\mr{\; s.t.} \\
\phi|_{\gamma_2}=\phi_2,\\
\phi|_{\gamma_3}=\phi_3
\end{array}}{\int} \mc{D}\phi'' \;\underbrace{e^{-S_\Sigma(\phi)}}_{e^{-S_{\Sigma'}(\phi')}e^{-S_{\Sigma''}(\phi'')}} \\
= 
\underset{\scriptsize \phi_2\in\F_{\gamma_2}}{\int} \mc{D}\phi_2 
\underset{\scriptsize
\begin{array}{c}
\phi'\in \F_{\Sigma'}\mr{\; s.t.} \\
\phi|_{\gamma_1}=\phi_1,\\
\phi|_{\gamma_2}=\phi_2
\end{array}}{\int} \mc{D}\phi'\;  e^{-S_{\Sigma'}(\phi')}
\underset{\scriptsize
\begin{array}{c}
\phi''\in \F_{\Sigma''}\mr{\; s.t.} \\
\phi|_{\gamma_2}=\phi_2,\\
\phi|_{\gamma_3}=\phi_3
\end{array}}{\int} \mc{D}\phi''\;  e^{-S_{\Sigma''}(\phi'')}\\
=\underset{\scriptsize \phi_2\in\F_{\gamma_2}}{\int} \mc{D}\phi_2 \;\;
\langle \phi_3 | Z_{\Sigma''} |\phi_2\rangle \; 
\langle \phi_2 |Z_{\Sigma''}|\phi_1 \rangle
\end{multline}
This is the convolution property of integral kernels equivalent to the relation 
$$ Z_\Sigma=Z_{\Sigma''}\circ Z_{\Sigma'} $$
between the corresponding integral operators, i.e. the sewing property.

The idea in (\ref{l3 Fubini computation}) is to treat the integration over fields on $\Sigma$ in the following way:
\begin{enumerate}[(i)]
\item Fix the value $\phi_2$ of the field on the sewing interface $\gamma_2$.
\item Integrate over fields on the two sub-cobordisms $\Sigma',\Sigma''$ with $\phi_2$ becoming a boundary condition -- this gives the matrix elements of partition functions for the sub-cobordisms.
\item Integrate out the field $\phi_2$ on the interface.
\end{enumerate}
%\textcolor{red}{PICTURE}

%\marginpar{adjust the font size in the picture}
\begin{figure}[H]
%$$\vcenter{\hbox{ \includegraphics[scale=0.5]{cob1.eps} }} $$
\begin{center}
\includegraphics[scale=0.7]{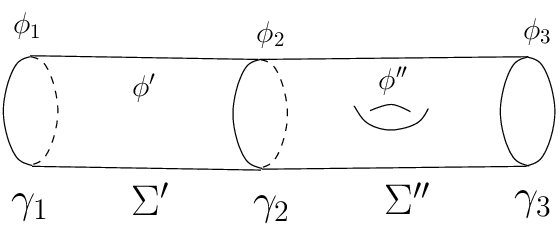}
\end{center}
\caption{Sewing: integrating over the field $\phi$ everywhere is equivalent to integrating over $\phi',\phi''$ and then over $\phi_2$.}
\end{figure}

In particular, we think of the space of fields on $\Sigma$ (with boundary conditions on $\gamma_{1,3}$) as fibered over fields on $\gamma_2$, and we write this integral using ``Fubini theorem for path integrals'' as an intergral over the fiber followed by integral over the base.\footnote{
This Fubini theorem is heuristically clear if the path integral measure is thought of as a continuum product of measures $d\phi(x)$ over points $x$ of $M$. However, when one defines path integrals mathematically, e.g., as perturbative integrals (via Feynman diagrams), this statement requires an independent proof. For special cases studied in detail, see e.g. \cite{TJF} (quantum mechanics), \cite{2dscalar} (2d scalar theory with polynomial potential), \cite{CMRpert} (topological field theories of AKSZ type), \cite{IM} (2d Yang-Mills). %\textcolor{red}{REF: Johnson-Freyd,...}.
}

In the computation (\ref{l3 Fubini computation}) we also used additivity of action (which is automatic from the local ansatz (\ref{l3 S})): $S_\Sigma(\phi)=S_{\Sigma'}(\phi')+S_{\Sigma''}(\phi'')$ if $\phi',\phi''$ are the restrictions of the field $\phi$  on $\Sigma$ to $\Sigma'$, $\Sigma''$.

\subsection{Observables in path integral formalism}
Suppose we are given a classical field theory on a cobordism $\Sigma$ and also given $i\colon\GGamma\hookrightarrow \Sigma$ a CW complex embedded into $\Sigma$ (with the image disjoint from the boundary). We define a classical observable $O_\GGamma$ supported on $\GGamma$ as some function on 
$\Gamma(\GGamma,i^*\mr{Jet}_\infty E)$, i.e., a function of jets of fields on $\GGamma$. 

For instance, if $\GGamma=\{x\}$ is a single point, then a classical observable at $x$ is just a function of the jet of 
the field at $x$, $O_x=f(\phi(x),\dd\phi(x),\ldots)$.

The expectation value of $O_\GGamma$ is formally defined in the path integral formalism as
\begin{equation}\label{l3 PI corr on Sigma closed}
\langle O_\GGamma \rangle = \underset{\phi\in \F_\Sigma}{\int} \mc{D}\phi\; e^{-S_\Sigma(\phi)} O_\GGamma(\mr{Jet}_\infty(\phi)|_\GGamma)\qquad \in\CC
\end{equation}
Here we assumed for simplicity that $\Sigma$ is closed.

If $\Sigma$ has a boundary, then we should include boundary conditions in the r.h.s., as in (\ref{l3 PI}), thus obtaining the ``matrix element,'' between states $|\phi_\ii\rangle$ and $\langle\phi_\oo|$, of the theory on $\Sigma$ enriched by the observable $O_\GGamma$:
\begin{equation} \label{l3 PI corr Sigma cob}
\langle \phi_\oo| Z_{\Sigma,O_\GGamma} |\phi_\ii\rangle =  
\underset{\scriptsize
\begin{array}{c}
\phi\in \F_\Sigma\mr{\; s.t.} \\
\phi|_{\gamma_\ii}=\phi_\ii,\\
\phi|_{\gamma_\oo}=\phi_\oo
\end{array}
}{\int} \mc{D}\phi\; e^{-S_\Sigma(\phi)} O_\GGamma(\mr{Jet}_\infty(\phi)|_\GGamma)\qquad \in \CC
\end{equation}

In quantization, a classical observable $O_\GGamma$ is mapped to a quantum observable $\wh{O}_\GGamma$ such that the expectation value (\ref{l2 VEV}) of $\wh{O}_\GGamma$ defined within the functorial QFT 
language agrees with the path integral expression (\ref{l3 PI corr on Sigma closed}), (\ref{l3 PI corr Sigma cob}). This can be arranged by defining $\wh{O}_\Gamma$ to be the state on the boundary of a thickening $U_\epsilon(\GGamma)$ given by the expression (\ref{l3 PI corr Sigma cob}) where instead of the cobordism $\Sigma$ we take the ``tube'' $U_\epsilon(\GGamma)$ (seen as a cobordism from $\varnothing$ to $\dd U_\epsilon(\GGamma)$).

\section{Why care about CFT?} \marginpar{Lecture 4\\ 8/31/2022}
Here we list some of the points of motivation, why one might be interested in 2d conformal field theory.

\subsection{CFT description of 2d Ising model}
This is the historical point of motivation, and it was the point of the seminal paper on CFT by Belavin-Polyakov-Zamolodchikov \cite{BPZ}.

One considers the Ising model -- a lattice model of statistical physics. On a graph $\Xi$, a state of the system is an assignment of spins $\pm 1$ (or ``spin up/spin down'') to vertices of $\Xi$. In particular, there are $2^{\# V(\Xi)}$ states in total where $V(\Xi)$ is the set of vertices of $\Xi$. The energy of a state is defined as
\begin{equation}\label{l4 Ising E}
E(s)=-\sum_{\mr{edges}\; (u,v)} s_u s_v - h \sum_{\mr{vertices}\; v} s_v
\end{equation} 
where $h\in \RR$ is a parameter (``external magnetic field''). Then one has the Gibbs probability measure on the set of states
\begin{equation}\label{l4 Gibbs measure}
\mr{Probability}(s) = \frac{1}{Z(T,h)} e^{-\frac{1}{T} E(s)}
\end{equation}
where $T>0$ is the temperature and 
\begin{equation}
Z(T,h)= \sum_{\mr{states}\; s} e^{-\frac{1}{T} E(s)}
\end{equation}
is the partition function (the normalization factor in the Gibbs measure (\ref{l4 Gibbs measure}), needed to normalize it to total mass $1$).

Then one considers the  continuum (or ``thermodynamical'') limit, taking $\Xi$ to be a very fine square lattice on a large square on $\RR^2$ and sending the spacing of the lattice to zero (while appropriately rescaling the energy function (\ref{l4 Ising E})). 

In the continuum limit, the system has a phase transition: the partition function $Z(T,h)$ and the $n$-point correlation functions of spins become real-analytic functions of $(T,h)$ on $\RR_{>0}\times \RR$ except on the interval $(0,T_\mr{crit}]$ with some positive critical temperature $T_\mr{crit}$.  The partition function and correlation functions are discontinuous (have a finite jump) across the interval $(0,T_\mr{crit})$, when going from small negative $h$ to small positive $h$. Points $(0<T<T_\mr{crit},h=0)$ are points of first-order phase transition of the system and $(T=T_\mr{crit},h=0)$ is the point of second order phase transition.

From now on, set $h=0$. If $T>T_\mr{crit}$, the two-point correlation function behaves as 
\begin{equation}
\langle s(x) s(y) \rangle  \sim e^{- \frac{||x-y||}{r_\mr{corr}}}
\end{equation}
where $r_\mr{corr}$ is the ``correlation radius,'' depending on $T$. In the limit $T\ra T_\mr{crit}$, the correlation radius goes to $+\infty$ and the system loses the ``characteristic scale'' -- becomes scaling invariant. In particular, the two-point function (\ref{l4 Ising 2-point}) at $T=T_\mr{crit}$ becomes a power law
\begin{equation}\label{l4 Ising 2-point}
\langle s(x) s(y) \rangle \sim \frac{1}{||x-y||^{\frac14}}
\end{equation}
The power $\frac14$ here is a result from the explicit solution of 2d Ising model (at any $T$) by Onsager \cite{Onsager}.

Thus, at the point $(T_\mr{crit},h=0)$ of second-order phase transition, the system becomes scaling invariant. Put another way, its symmetry gets enhanced from Euclidean motions (translations+rotations) to include scaling. 
%In fact, it turns out that the symmetry becomes even bigger -- the system becomes invariant under all conformal maps. 
At this point it is natural to conjecture the system on $\RR^2$, at the point of second-order phase transition, can be described by some model of conformal field theory (which would also mean that the symmetry is further enhanced from rotations+translations+scaling to all conformal transformations). This was proven -- and the corresponding CFT was identified as the \emph{free Weyl fermion} -- in \cite{BPZ}.

It turns out that a much wider class of statistical systems exhibiting phase transitions at the point of phase transition can be described by a CFT, which eventually leads to explanation of the interesting rational exponents (``critical exponents'') one encounters in these systems -- such as the power $\frac14$ in (\ref{l4 Ising 2-point}).\footnote{Ultimately, $\frac14$ comes from the fact that Ising spin field can be identified with a primary field of conformal weight $(\frac{1}{16},\frac{1}{16})$ in the free fermion field.}

\subsection{Bosonic string theory}
Classically, bosonic string theory can be though of as a classical field theory of maps  from a surface $\Sigma$ (``worldsheet'') to the target $\RR^N$ (``spacetime'' in string theory terminology), with action 
\begin{equation}\label{l4 bosonic string}
S(\Phi;b,c,\bar{b},\bar{c})=\int_\Sigma\sum_{i=1}^N \frac12 d\Phi_i \wedge * d\Phi_i + b \bar\dd c + \bar{b}\dd\bar{c}.
\end{equation}
Here $\Phi\colon \Sigma\ra \RR^N$ is the bosonic field describing the string in $\RR^N$, $\Phi_i$ are components corresponding to coordinates on $\RR^N$, so that each $\Phi_i$ can be seen as a scalar field on $\Sigma$. The last two terms in (\ref{l4 bosonic string}) (the ``reparametrization ghost system'') are auxiliary anticommuting fields (``Faddeev-Popov ghosts'') that appear in the action through Faddeev-Popov mechanism, because one wants to consider the path integral over $\mr{Map}(\Sigma,\RR^N)/\mr{Diff}(\Sigma)$ -- they appear in essence from homological resolution of this quotient. The fields $c,\bar{c}$ are sections of $T^{1,0}\Sigma$, $T^{0,1}\Sigma$ -- holomorphic/antiholomorphic tangent bundle; fields $b,\bar{b}$ are quadratic differentials -- sections of $((T^{(1,0)})^*\Sigma)^{\otimes 2}$, $((T^{(0,1)})^*\Sigma)^{\otimes 2}$,  respectively.

Upon quantization, (\ref{l4 bosonic string}) becomes a particular CFT on $\Sigma$ -- the ``sum'' of several mutually non-interacting theories -- $N$ free massless scalar fields and the ghost system. The central charge of this CFT (measuring the ``strength'' of projectivity effect/conformal anomaly, see Section \ref{sss: conf anomaly}) turns out to be 
\begin{equation}
c= N-26
\end{equation}
-- each free scalar contributes $1$ to the central charge and the ghost system contributes $-26$. In particular, the central charge (and thus the conformal anomaly) vanishes iff $N=26$.  Which is the reason why dimension $26$ of the target is distinguished in bosonic string theory.

\subsection{Invariants of 3-manifolds}
There are interesting connections between 3d topological quantum field theories and 2d conformal field theories on the boundary of a 3-manifold.

Notably, there is a relation between 3d Chern-Simons theory (which is topological) and 2d Wess-Zumino-Witten theory (which is a CFT). This relation was very fruitfully exploited in \cite{Witten89} to construct invariants of knots and 3-manifolds. 

One relation is that Chern-Simons correlator of a tangle in a 3-ball can be interpreted as a correlator of point observables in WZW theory on the boundary 2-sphere. This fact was explained and used in \cite{Witten89} to explain why the correlators of Wilson loop observables in 3d theory have to satisfy certain skein relation (which is ultimately a move performed on the portion of a knot contained in a small ball).

Put differently, the relation between Chern-Simons theory on a 3-manifold $M$ and 2d WZW on the boundary $\Sigma=\dd M$ is that the space of states that Chern-Simons assigns to $\Sigma$ is the ``space of conformal blocks'' (holomorphic building blocks of correlators) that WZW assigns to $\Sigma$, see e.g.  \cite{GK}.

\subsection{A zoo of computable QFTs}
Part of motivation to study CFTs is that they give examples quantum field theories with explicit and nontrivial answers.

For instance in a typical CFT situation,
\begin{itemize} 
\item two-point functions are often given by power laws with interesting rational exponents, 
\item four-point functions can be expressed in terms of the hypergeometric function,
\item genus 1 partition function can be expressed in terms of such objects as Jacobi theta functions and Dedekind eta function.
\end{itemize}

The zoo of well-known examples of CFTs includes among others:
\begin{itemize}
\item Free theories: 
\begin{itemize}
\item free massless scalar field (or ``free boson''),
\item free massless scalar with values in $S^1$,
\item free fermion,
\item $bc$-system (and a very similar $\beta\gamma$-system).
\end{itemize}
\item Minimal models $M(p,q)$ of CFT.
\item Wess-Zumino-Witten model.
\end{itemize}

\subsection{Motivation from representation theory}\leavevmode 
\\ \indent
\ul{Representations of loop groups/Lie algebras}. CFT is naturally linked to representation theory of loop groups and loop Lie algebras (or rather their central extensions). E.g., the space of states $\HH_{S^1}$ always carries a representation of the Virasoro algebra. In the case of WZW models, $\HH_{S^1}$ also carries a representation of a Kac-Moody algebra $\wh{\mathfrak{g}}$ (which gives in a sense a ``refinement''\footnote{In the sense that Virasoro generators act as quadratic expressions in Kac-Moody generators, via the so-called Sugawara construction.} of the Virasoro representation). 

\ul{Representations of the mapping class group}.
Additionally, a part of the data of CFT (the space of conformal blocks) naturally carries a representation of the mapping class group of the surface.

\subsection{Motivation from topology of $\MM_{g,n}$ and enumerative geometry}
In topological conformal field theories (such as Witten's A-model), special correlators define closed differential forms on the moduli space of algebraic curves $\overline{\MM}_{g,n}$ (with Deligne-Mumford compactification) yielding interesting elements in de Rham cohomology of the moduli space. Periods of these forms over compactification cycles satisfy certain quadratic relations (equivalently, the corresponding generating functions satisfy the so-called Witten-Dijkgraaf-Verlinde-Verlinde equation).

In the A-model, such periods are the Gromov-Witten invariants -- counts of holomorphic curves in the target K\"ahler manifold $X$ intersecting a given collection of cycles.

\section{CFT as a system of correlators}\label{ss: CFT as a system of correlators}

CFT is often studied in a simplified setting (as compared to the functorial QFT picture): instead of surfaces with boundary, one considers surfaces with punctures (marked points). 

%\textcolor{red}{PICTURE}
\begin{figure}[H]
%$$\vcenter{\hbox{ \includegraphics[scale=0.5]{cob1.eps} }} $$
\begin{center}
\includegraphics[scale=1]{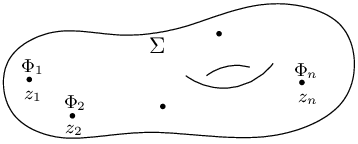}
\end{center}
\caption{Surface with punctures decorated by fields.}
\end{figure}

One can think of punctures as ``infinitesimally small circles.'' Instead of partition function on surfaces with boundary, one studies $n$-point correlation functions 
\begin{equation}\label{l4 corr}
\langle \Phi_1(z_1)\cdots \Phi_n(z_n) \rangle  \quad \in \CC
\end{equation}
depending on a configuration of $n$ distinct ordered points on the Riemann surface $\Sigma$ and on a choice of vectors $\Phi_1,\ldots,\Phi_n$ in the vector space $V$ (the space of states $\HH_{S^1}$ in the language of functorial QFT). There are different possible names for elements of $V$:
\begin{itemize}
\item Fields (or ``composite fields'') at a point $z$.\footnote{Not to be confused with the fields of the Lagrangian formulation of the underlying classical field theory. %Rather, one can say that in examples coming from quantization of a Lagrangian field theory, composite fields are (quantization of) differential polynomials of the classical field at the point $z$.
}
\item Point observables.
\item Operators.
\end{itemize}

In the path integral language, (\ref{l4 corr}) corresponds to the expression 
\begin{equation}\label{l4 corr path integral}
 \int \mc{D}\phi\;\; e^{-S(\phi)} \Phi_1(z_1)\cdots \Phi_n(z_n) 
\end{equation}
where expressions $\Phi_i$ under the path integral are point classical observables  -- functions of the jet of the classical field $\phi$ at $z_i$ (in the notations we are blurring the distinction between classical observables and corresponding quantum observables).

\ul{Subtlety}: to make sense of a correlator (\ref{l4 corr}) as a number, one needs to fix a complex coordinate chart around each point $z_1,\ldots,z_k$.\footnote{Or at least one needs to fix an $\infty$-jet of complex coordinate charts centered at each $z_i$ -- a ``formal'' complex coordinate chart at $z_i$.}

For particularly nice elements of $V$ -- so-called ``primary'' fields (see below), one doesn't need the full data of coordinate charts -- it is sufficient to have a trivialization of tangent spaces $T_{z_i} \Sigma$, thus the correlators of primary fields can be regarded as a section
\begin{equation}\label{l4 corr as in Gamma(Conf,L)}
\langle \Phi_1\cdots \Phi_n \rangle \in \Gamma(\mr{Conf}_n(\Sigma),\LL)
\end{equation}
over the open configuration space 
$\mr{Conf}_n(\Sigma)=\{(z_1,\ldots,z_n)\in\CC^n | z_i\neq z_j \;\mr{if}\; i\neq j\}$
of $n$ ordered points on $\Sigma$, of a certain complex line bundle $\LL$ depending of the fields $\Phi_i$. In (\ref{l4 corr as in Gamma(Conf,L)}) we allow points $z_1,\ldots,z_n$ to move around on a fixed Riemann surface $\Sigma$ (i.e. the complex structure is fixed). 

We can also allow the complex structure to change (then movement of points is absorbed into changes of complex structure). Then the correlator of primary fields becomes a section of certain complex line bundle $\til\LL$ over the moduli space of complex structures on $\Sigma$ with $n$ punctures:
\begin{equation}\label{l4 corr as in Gamma(M_g,n,L)}
\langle \Phi_1\cdots \Phi_n \rangle \in \Gamma(\MM_{\Sigma,n},\til\LL)
\end{equation}

\begin{remark}
%\marginpar{On $\MM^\mr{coor}$ the correlator is just a function -- the bundle can be taken to be trivial}
For general (possibly non-primary) $\Phi_i$, one needs to replace $\MM_{\Sigma,n}$ in (\ref{l4 corr as in Gamma(M_g,n,L)}) with an enhanced version $\MM^\mr{coor}_{\Sigma,n}$ of the moduli space where each puncture carries a formal coordinate system. Put another way, when defining $\MM^\mr{coor}_{\Sigma,n}$ as complex surfaces modulo diffeomorphisms, one should only quotient by diffeomorphisms which have the $\infty$-jet of identity at each $z_i$. In this setup the line bundle over $\MM^\mr{coor}_{\Sigma,n}$ is trivial and the general $n$-point correlator is a function on $\MM^\mr{coor}_{\Sigma,n}$ with values in $\mr{Hom}(V^{\otimes n},\CC)$, invariant under formal conformal vector fields at the punctures $z_i$ (acting both on $V$ at $z_i$ and on the formal coordinate system):
\begin{equation}
\langle \cdots \rangle\in C^\infty(\MM^\mr{coor}_{\Sigma,n},\mr{Hom}(V^{\otimes n},\CC))^\mr{formal\;c.v.f.\;at\;punctures}.
\end{equation}
\end{remark}

\subsection{The action of conformal vector fields on $V$}
\label{sss 1.7.1 action of cvf on V}
The space $V$ comes equipped with a projective representation of the Lie algebra $\mc{A}^\mr{loc}$ of conformal vector fields on $\CC^*$ (real parts of meromorphic vector fields with only pole at zero allowed),%\footnote{
%More precisely, we should take conformal/meromorphic vector fields on a punctured formal disk.
%}
\begin{equation}\label{l4 rho}
\rho\colon \mc{A}^\mr{loc}\ra \mr{End}(V)
\end{equation}
This representation can be thought of as the complexified (in a certain sense)
infinitesimal version of the representation (\ref{l2 rho}) in the picture of functorial QFT,  see Section \ref{sss: double complexification} below.\footnote{Remark: representation (\ref{l4 rho}) contains strictly more information (morally, ``twice more'') than the action of diffeomorphisms (\ref{l2 rho}). 
For instance, the difference of conformal weights $h-\bar{h}$ of a field (see Section \ref{sss: conf weights} below) corresponds to the action of rotation around the origin and is a part of the data of (\ref{l2 rho}), while $h+\bar{h}$ corresponds to the action of dilation, which infinitesimally is a vector field on  $S^1$ not tangential to $S^1$, and it is not a part of the data (\ref{l2 rho}) but is a part of the data (\ref{l4 rho}).
}

In the common nomenclature, the standard generators of  $\mc{A}^\mr{loc}_\CC$ -- the 
complexified 
Lie algebra of conformal vector fields on $\CC^*$ --
 are denoted
\begin{equation}\label{l4 l_n}
l_n\colon= -z^{n+1}\frac{\dd}{\dd z} , \quad \bar{l}_n\colon= -\bar{z}^{n+1}\frac{\dd}{\dd \bar{z}}, \qquad n\in \ZZ 
\end{equation}
%vector field $-z^{n+1}\frac{\dd}{\dd z}$ on $\CC^*$ is denoted $l_n$ and 
The corresponding operators acting on $V$ are denoted
\begin{equation}\label{l4 L_n}
L_n\colon = \rho(l_n),\quad \bar{L}_n\colon= \rho(\bar{l}_n)
\end{equation}
%the operator $\rho(l_n)$ representing it on $V$ is denoted $L_n$.

%\begin{remark} [``Double complexification'']
\subsection{The ``double complexification''} \label{sss: double complexification}
%\marginpar{Did not mention in the lecture. Should mention later on.}
The Lie algebra $\mc{A}^\mr{loc}_\CC= \mc{A}^\mr{loc}\otimes \CC$ conveniently splits into holomorphic and antiholomorphic copies of complex Witt\footnote{Witt algebra is the Lie algebra of %holomorphic vector fields on $\CC^*$ (or, equivalently,
 meromorphic vector fields on $\CC$ with only pole at $0$ allowed, see Section \ref{sss: Witt algebra}. In terms of (\ref{l4 l_n}), it is $\mr{Span}_\CC(\{l_n\}_{n\in \ZZ})$. } algebra and its central extension splits similarly into two copies of complex Virasoro algebras. The Lie algebra $\mc{A}^\mr{loc}_\CC$ can be seen, in a sense, as  ``double complexification'' of the Lie algebra of diffeomorphisms of a circle:
\begin{equation}\label{l4 double complexification}
\begin{CD}
\mathfrak{X}(S^1) @>\mr{complexification}>>  \mc{A}^\mr{loc} @>\mr{complexification}>{\otimes\CC}> \underset{\simeq \mr{Witt}\oplus \overline{\mr{Witt}}}{\mc{A}^\mr{loc}_\CC} \\
@AAA @AAA \\
\mr{Diff}(S^1) @>>\mr{``complexification''}> \mr{Ann}
\end{CD}
\end{equation}
Here $\mathfrak{X}(S^1)$ is the Lie algebra of real vector fields on a circle, $\mr{Ann}$ is the Segal's semi-group of annuli \cite{Segal88} -- the full subcategory of Segal's cobordism category consisting of cobordisms $S^1\xra{\Sigma} S^1$ (with conformal structure on $\Sigma$ and parametrization of boundary circles). The vertical arrows are the transitions from a Lie group or semi-group to its Lie algebra. The first complexification in the top row  of (\ref{l4 double complexification}) allows vector fields on $S^1$ that are not necessarily tangential to $S^1$ and then extends them to real conformal vector  fields (which are special sections of the \emph{non-complexified} tangent bundle $T\CC^*$ of $\CC^*$ seen as a smooth 2-manifold) on $\CC^*$. The second complexification allows complex-valued conformal vector fields on $\CC^*$ -- special sections of the complexified tangent bundle $T_\CC\CC^*$. Explicitly, one has 
{
\begin{equation}
\begin{aligned}
\mathfrak{X}(S^1)&=\mr{Span}_\RR(\{\cos n\theta\, \dd_\theta\}_{n\geq 0}, \{\sin n\theta\, \dd_\theta\}_{n\geq 1})\\
&=
\mr{Span}_\RR\left(\left\{-\frac{i}{2}(l_n+l_{-n}-\bar{l}_{n}-\bar{l}_{-n})\right\}_{n\geq 0}, \left\{-\frac12 (l_n-l_{-n}+\bar{l}_n-\bar{l}_{-n})\right\}_{n\geq 1}\right),\\
\mc{A}^\mr{loc}
%&=
%\mr{Span}_\RR(\{\cos n\theta\, \dd_\theta, \cos n\theta\, \dd_r\}_{n\geq 0},\{\sin n\theta\, \dd_\theta,\sin n\theta\, \dd_r\}_{n\geq 1})\\
&=
\mr{Span}_\RR\left(\left\{\frac{l_n+\bar{l}_n}{2}\right\}_{n\in\ZZ},\left\{\frac{l_n-\bar{l}_n}{2i}\right\}_{n\in\ZZ}\right),
\\
\mc{A}^\mr{loc}_\CC &= 
\mr{Span}_\CC\Big( \{l_n,\bar{l}_n\}_{n\in\ZZ}\Big) .
\end{aligned}
\end{equation}
}
The bottom horizontal arrow in (\ref{l4 double complexification}) is explained in \cite{Segal88}.
%\end{remark}

 %end color

\subsection{Grading on $V$ by conformal weights}\label{sss: conf weights}
The complexified Lie algebra $\mc{A}^\mr{loc}_\CC$ is naturally graded by elements of $\ZZ\oplus \ZZ$. In particular, the meromorphic vector field $z^{n+1}\frac{\dd}{\dd z}$ on $\CC^*$ has degree $(n,0)$ and the antimeromorphic vector field $\bar{z}^{n+1}\frac{\dd}{\dd \bar{z}}$ has degree $(0,n)$.  Accordingly, $V$ carries a grading by ``conformal weight'' $(h,\bar{h})\in \RR\oplus \RR$. A field $\Phi\in V$ is said to have conformal weight $(h,\bar{h})$ if 
\begin{equation}
\rho\left(-z\frac{\dd}{\dd z}\right)\circ \Phi = h \Phi,\quad \rho\left(-\bar{z}\frac{\dd}{\dd\bar{z}}\right)\circ\Phi= \bar{h} \Phi.
\end{equation} 
The grading on the Lie algebra is compatible with the grading on the module: acting by an element of $\mc{A}^\mr{loc}_\CC$ of degree $(n,\bar{n})$ shifts the conformal weight of a vector in $V$ as $(h,\bar{h})\ra (h-n,\bar{h}-\bar{n})$.\footnote{
We emphasize that in $\bar{h}$, $\bar{n}$, the bar does not mean complex conjugation.
}
%The conformal weight is the pair of eigenvalues of $\rho(z\frac{\dd}{\dd z}),\rho(\bar{z}\frac{\dd}{\dd\bar{z}})$ acting on $\Phi\in V$. 
One can split $V$ into graded components:
$$V=\bigoplus_{(h,\bar{h})\in \Lambda} V^{(h,\bar{h})}.$$
%-- splitting of $V$ into graded components. 
Here $\Lambda \subset \RR\oplus \RR$ is the set of admissible conformal weights (dependent on a particular CFT model); $\Lambda$ is necessarily a $\ZZ\oplus \ZZ$-module. 

\begin{remark} \label{l4 rem: h-hbar in Z}
The condition that the representation $\rho$ of $\mc{A}^\mr{loc}_\CC$ comes from a representation of the group $\mr{Diff}(S^1)$ implies in particular that rotation by the angle $2\pi$ should act on a field as identity (or, in the notations (\ref{l4 L_n}), one should have $e^{2\pi i (L_0-\bar{L}_0)}=\mr{id}$). That implies  
\begin{equation}\label{l4 h-barh in Z}
h-\bar{h}\in \ZZ
\end{equation}
for any element of $V$.\footnote{One can consider models where (\ref{l4 h-barh in Z}) is violated, but in this case correlators are multivalued. In other words, correlators are functions (or sections of a line bundle) not on the configuration space of $n$-points but rather on its covering space.}
\end{remark}

\subsection{Conformal Ward identity} 
Conformal Ward identity is the following symmetry property of correlators. Fix a Riemann surface $\Sigma$ with punctures $z_1,\ldots,z_n$ and fix fields $\Phi_1,\ldots,\Phi_n\in V$. Let  $v$ be a conformal vector field on $\Sigma$ with singularities allowed at $\{z_i\}$ -- the real part of a meromorphic vector field with poles allowed at $z_1,\ldots,z_n$ (we will denote the Lie algebra of such vector fields $\mc{A}_{\Sigma,\{z_i\}}$). Then we
have the Ward identity
\begin{equation}\label{l4 conformal Ward identity}
\underbrace{\sum_{i=1}^n \langle \Phi_1(z_1)\cdots \rho(\mr{Laurent}_{z_i}(v))\circ \Phi_i(z_i)\cdots \Phi_n(z_n)  \rangle}_{\LL_v \langle \Phi_1(z_1)\cdots \Phi_n(z_n) \rangle}  = 0 .
\end{equation}
Here the left hand side can be thought of %seen 
as the ``Lie derivative of the correlator along $v$;'' 
$$\mr{Laurent}_{z_i}\colon \mc{A}_{\Sigma,\{z_i\}}\ra \mc{A}^\mr{loc}$$ 
is the Laurent expansion of a (real part of the) meromorphic vector field at the point $z_i$.

One can think of (\ref{l4 conformal Ward identity}) as a version of naturality (\ref{l1 Ward identity}) in the functorial QFT setting.\footnote{In this version, one passes (a) from finite boundaries to infinitesimal ones (punctures), (b) from Lie group action to the associated Lie algebra action, (c) one complexifies the Lie algebra, which corresponds to allowing vector fields not tangential to the boundary.}

\subsection{The ``$L_{-1}$ axiom''}
%\marginpar{Did not explain this in the class.}
Representation (\ref{l4 rho}) is supposed to satisfy the following natural property:%\marginpar{signs?}
{\small 
\begin{eqnarray}
\langle \Phi_1(z_1)\cdots \rho\left(\frac{\dd}{\dd w}\right) \circ \Phi_i(z_i) \cdots \Phi_n(z_n) \rangle
&= &
\frac{\dd}{\dd z_i} \langle \Phi_1(z_1)\cdots \Phi_i(z_i)\cdots \Phi_n(z_n) \rangle \label{l4 L_-1 axiom} \\
\langle \Phi_1(z_1)\cdots \rho\left(\frac{\dd}{\dd \bar{w}}\right) \circ \Phi_i(z_i) \cdots \Phi_n(z_n) \rangle
&= &
\frac{\dd}{\dd \bar{z}_i} \langle \Phi_1(z_1)\cdots \Phi_i(z_i)\cdots \Phi_n(z_n) \rangle \label{l4 bar L_-1 axiom}
\end{eqnarray}
}
for any surface with any collection of punctures and fields; $w$ is a local complex coordinate centered at $z_i$.

 Thus, (\ref{l4 L_-1 axiom}) says that acting by 
% $l_{-1}=-\frac{\dd}{\dd z}$ 
$L_{-1}$
 on a field under the correlator is tantamount to taking  the holomorphic derivative in the position of the corresponding puncture (up to a sign). Similarly, (\ref{l4 bar L_-1 axiom}) says that acting by %the conjugate vector field $\bar{l}_{-1}=-\frac{\dd}{\dd\bar{z}}$
 $\bar{L}_{-1}$ is tantamount to taking the antiholomorphic derivative in the position.

\subsection{Some special fields} \leavevmode \\ \indent
\ul{Identity field.} The identity field $\mathbb{1}\in V^{(0,0)}$ corresponds in the functorial QFT picture to the vacuum vector  $|\mr{vac}\rangle\in \HH_{S^1}$ -- the partition function of a disk.  The field $\mathbb{1}$ is characterized by the property that for any fields $\Phi_1,\ldots,\Phi_n$ and any points $z_0,z_1,\ldots,z_n$ on $\Sigma$, one has
\begin{equation}
\langle \mathbb{1}(z_0) \Phi_1(z_1)\cdots \Phi_n(z_n) \rangle  = \langle  \Phi_1(z_1)\cdots \Phi_n(z_n) \rangle 
\end{equation}
Put another way, putting the field $\mathbb{1}$ at a puncture effectively forgets that puncture.

\ul{Stress-energy tensor.} The stress-energy tensor $T\in V^{(2,0)}\oplus V^{(0,2)}$ is defined as
\begin{equation}
T\colon=\rho\left(\mr{Re}\left(\frac{-2}{z}\dd_z
\right)\right)\circ \mathbb{1} 
\end{equation}
Or in terms of standard notations (\ref{l4 L_n}) introduced above,
\begin{equation}
T=(L_{-2}+\bar{L}_{-2})\circ \mathbb{1}
\end{equation}

\marginpar{Lecture 5\\
9/2/2022}
\ul{Primary fields.}
A field $\Phi\in V^{(h,\bar{h})}$ is said to be primary if it is a highest weight vector under the action of $\mc{A}^\mr{loc}_\CC$, i.e., if
\begin{equation}
L_n \Phi=0,\; \bar{L}_n\Phi=0 \quad \mr{for\; any\;}n>0.
\end{equation}
Equivalently, field $\Phi$ is primary if it is annihilated by conformal vector fields which vanish to second order at the origin (the point of insertion of $\Phi$).

It is natural to assign to a primary field of conformal weight $(h,\bar{h})$ a complex line bundle 
\begin{equation} \label{l5 line bun primary}
\LL^{h,\bar{h}}=K^{\otimes h}\otimes \bar{K}^{\otimes \bar{h}}
\end{equation} over $\Sigma$ where 
$$ K=(T^{1,0})^*\Sigma, \quad \bar{K}=(T^{0,1})^*\Sigma $$
are the holomorphic and antiholomorphic cotangent bundles of $\Sigma$, respectively.

Then the correlator (\ref{l4 corr as in Gamma(Conf,L)}) of primary fields $\Phi_i\in V^{h_i,\bar{h}_i}$ is a section over $\mr{Conf}_n(\Sigma)$ of the line bundle
\begin{equation}
\LL=\iota^*\boxtimes_{i=1}^n \LL^{h_i,\bar{h}_i}
\end{equation}
where  $\iota\colon \mr{Conf}_n(\Sigma)\ra \Sigma^{\times n}$ is the natural inclusion. 

From the standpoint of the moduli space of complex structures, the correlator of primary fields (\ref{l4 corr as in Gamma(M_g,n,L)}) is a section of the line bundle
\begin{equation}
\til\LL=\left(\bigotimes_{i=1}^n \LL^{h_i,\bar{h}_i}_i\right)\otimes \LL_\mr{anomaly}
\end{equation}
over the moduli space $\MM_{\Sigma,n}$. Here  $\LL^{h_i,\bar{h}_i}_i$ is the line bundle (\ref{l5 line bun primary}) assoicated to $i$-th puncture on $\Sigma$; 
%$\pi_i\colon \MM_{Sigma,n}\ra \Sigma$ 
\begin{equation}
\LL_\mr{anomaly}=(\mr{Det}\,\bar\dd)^{\otimes c} \otimes (\mr{Det}\,\dd)^{\otimes \bar{c}}
\end{equation}
is the effect of conformal anomaly, with $(c,\bar{c})$ the central charge (see Section \ref{sss: conf anomaly} and footnote \ref{l2 footnote: conf anomaly bundle}).

\subsection{Operator product expansions}\label{sss: OPEs intro}
When studying CFT as a system of correlators,
instead of sewing along boundaries, one studies OPEs (``operator product expansions'') governing the singularities of correlators of fields (\ref{l4 corr as in Gamma(Conf,L)}) as the point of insertion of one field approaches another, $z_i\ra z_j$.

%\textcolor{red}{PICTURE}
\begin{figure}[H]
%$$\vcenter{\hbox{ \includegraphics[scale=0.5]{cob1.eps} }} $$
\begin{center}
\includegraphics[scale=1]{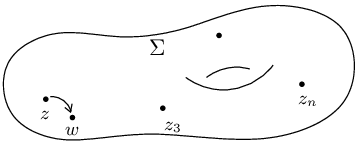}
\end{center}
\caption{One puncture approaching another.}
\end{figure}

An OPE is an expression of the form
\begin{equation}\label{l5 OPE}
\Phi_1(z) \Phi_2(w) \underset{z\ra w}{\sim} \sum_{\til\Phi} f^{\til\Phi}_{\Phi_1\Phi_2}(z,w) \til\Phi(w)+\mr{reg}
\end{equation}
Here on the right hand side: 
\begin{itemize}
\item The sum is over a basis $\{\til\Phi\}$ in $V$. 
\item Coefficient functions $f^{\til\Phi}_{\Phi_1\Phi_2}(z,w)$
are some real-analytic functions on a neighborhood of $\mr{Diag}\subset \Sigma\times \Sigma$, singular on $\Sigma$.
\item $\mr{reg}$ stands for terms that are continuous (in special cases, even holomorphic) on the diagonal $z=w$.
\end{itemize}

In (\ref{l5 OPE}) %in the right hand side 
we could have chosen instead to express the operator product in terms of fields $\til\Phi$ at $z$ rather than $w$ (or even, say, at some point between $z$ and $w$); this choice affects the coefficients in the OPE.

The expression (\ref{l5 OPE}) is understood as a substitution that one can perform under the correlator of $\Phi_1(z)$, $\Phi_2(w)$, and any collection of other fields away from $z$ and $w$, in the asymptotics $z\ra w$:
\begin{equation}\label{l5 OPE under corr}
\langle \Phi_1(z) \Phi_2(w) \underbrace{\Phi_3(z_3)\cdots \Phi_n(z_n)}_{\mr{away\;from\;}z,w} \rangle \underset{z\ra w}{\sim} 
\sum_{\til\Phi} f^{\til\Phi}_{\Phi_1\Phi_2}(z,w) \langle \til\Phi(w) \Phi_3(z_3)\cdots \Phi_n(z_n) \rangle + \mr{reg}
\end{equation}

Thus, singularities of $n$-point correlators are governed by $(n-1)$-point correlators. 

Note: the OPE (\ref{l5 OPE}) does not depend on the collection of ``test fields'' $\Phi_3,\ldots, \Phi_n$ in the correlator (\ref{l5 OPE under corr}).

\ul{Idea.} %In favorable circumstances
One wants to recover $n$-point correlators functions from $(n-1)$-point correlators using the OPEs (\ref{l5 OPE under corr}), ultimately reducing everything to 3-point correlators.  The idea is similar to recovering a meromorphic function from knowing the principal part of its Laurent expansion at each pole. 

The idea that all correlators can be derived from 3-point correlators is close to the idea in the functorial QFT approach, that one can cut any surface into  ``pairs of pants'' (spheres with three holes).\footnote{Cf. Section \ref{sec 7.5} below.}

%\marginpar{Edit/remove?}
Another form of that thought: an $n$-point correlator on a plane can be seen as a sewing of a collection of annuli with one hole.\footnote{Cf. the discussion of radial quantization in Section \ref{ss 4.3.4 radial quantization}.}

\begin{remark} The asymptotic of two punctures on $\Sigma$ approaching one another from the standpoint of the moduli space of  curves $\MM_{\Sigma,n}$ corresponds to approaching a nodal curve, where punctures $z,w$ are in one component, connected by a ``neck'' to the other component, where the remaining punctures $z_3,\ldots,z_n$ are (where we put the ``test fields'').
%\textcolor{red}{PICTURE}
\begin{figure}[H]
%$$\vcenter{\hbox{ \includegraphics[scale=0.5]{cob1.eps} }} $$
\begin{center}
\includegraphics[scale=1]{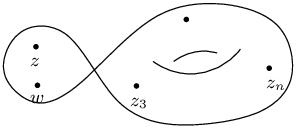}
\end{center}
\caption{Nodal curve.}
\end{figure}
\end{remark}

\chapter{Elements of conformal geometry}

\section{Conformal maps}
Reference: \cite{Schottenloher}.

Let $(M,g)$ be a Riemannian (or pseudo-Riemannian) manifold.

\begin{definition} A \emph{Weyl transformation} is a change of metric on a (pseudo-)Riemannian manifold $(M,g)\ra (M,g'=\Omega\cdot g)$ consisting in multiplying the metric by an everywhere positive function $\Omega\in C^\infty_{>0}(M)$ (the ``Weyl factor'').

Two metrics on $M$ differing by a Weyl transformation are said to be \emph{conformally equivalent}. A metric on $M$ modulo conformal equivalence is called a \emph{conformal structure} on $M$.
\end{definition}

\begin{definition}
A smooth map of (pseudo-)Riemannian manifolds $\phi\colon (M,g)\ra (M',g')$ is a \emph{conformal map} if 
\begin{equation} 
\phi^*g'=\Omega\cdot g
\end{equation}
for some positive function $\Omega\in C^\infty_{>0}(M)$ (the \emph{conformal factor} associated to $\phi$).

One says that two (pseudo-)Riemannian manifolds $(M,g)$ and $(M',g')$ are conformally equivalent if there exists a conformal diffeomorphism
\begin{equation}
\phi\colon (M,g)\ra (M',g').
\end{equation}
\end{definition}

Some immediate properties of conformal maps:
\begin{enumerate}[(a)]
\item If $\phi_1\colon (M,g)\ra (M',g')$ and $\phi_2\colon (M',g')\ra (M'',g'')$ are two conformal maps with conformal factors $\Omega',\Omega''$, then $\phi_2\circ\phi_1\colon (M,g)\ra (M'',g'')$ is a conformal map with $\Omega=\phi_1^*\Omega_2\cdot \Omega_1$.
\item If $\phi\colon (M,g)\ra (M',g')$ is a conformal diffeomorphism with conformal factor $\Omega$, then $\phi^{-1}\colon (M',g')\ra (M,g)$ is also a conformal diffeomorphism with conformal factor $(\phi^{-1})^*\Omega^{-1}$.
\item \label{l5 (c)} If $\phi\colon (M,g)\ra (M',g')$ is a conformal map with conformal factor $\Omega$ and $\Lambda\in C^\infty_{>0}(M)$, $\Lambda'\in C^\infty_{>0}(M')$ are positive functions, then
$\phi\colon (M,\Lambda\cdot g)\ra (M',\Lambda'\cdot g')$ is also a conformal map, with conformal factor $ \frac{\phi^*\Lambda'}{\Lambda}\cdot\Omega$. 

In particular, the notion of a conformal map between manifolds equipped with just conformal structure (rather than metric) is well-defined, but the conformal factor of such a map is not well-defined.
\end{enumerate}

\begin{definition} 
Conformal automorphisms $\phi\colon(M,g)\ra (M,g)$ form a group -- the \emph{conformal group} $\mr{Conf}(M,g)$. By (\ref{l5 (c)}) above, this group depends only on the conformal class of $g$. 
\end{definition}

\section{Examples of conformal maps}
\begin{example}\label{l5 example: isometries}
Isometries of $(M,g)$ form a subgroup of $\mr{Conf}(M,g)$ (characterized by the property $\Omega=1$).
\end{example}

\begin{example}
%This is a special case of Example \ref{l5 example: isometries}.
 Translations and $O(n)$-rotations of Euclidean space $\RR^n$ (with the standard metric $g=(dx^1)^2+\cdots + (dx^n)^2$) are conformal automorphisms:
\begin{equation}
ISO(n)= O(n)\ltimes \RR^n \subset \mr{Conf}(\RR^n).
\end{equation}
(This is a special case of Example \ref{l5 example: isometries}.)

More generally, one can consider the space $\RR^{p,q}$ with metric $g=(dx^1)^2+\cdots+(dx^p)^2-(dx^{p+1})^2-\cdots-(dx^{p+q})^2$. Then one has translations and $O(p,q)$-rotations as isometries (and in particular, conformal automorphisms) of $\RR^{p,q}$.
\end{example}

\begin{example}[Dilations]
Fix a nonzero real number $\lambda$.
The dilation (or %homothety, or 
scaling) map
\begin{equation}
\begin{array}{ccc}
\RR^n&\ra & \RR^n\\
\vec{x}&\mapsto & \lambda\vec{x}
\end{array}
\end{equation}
%with $\lambda\neq 0$ a fixed real number 
is a conformal map with $\Omega=\lambda^2$. (One can replace $\RR^n$ with $\RR^{p,q}$ in this example.)
\end{example}

\marginpar{Lecture 6,\\
9/5/2022}

\begin{example}[Stereographic projection]
Let 
\begin{equation}
S^n=\{ (x^0,\ldots,x^n)\in \RR^n\;|\;  \sum_{i=0}^n (x^i)^2=1\}
\end{equation}
be the unit sphere in $\RR^{n+1}$ with $N=(1,0,\ldots,0)$ the North pole. Consider the stereographic projection
\begin{equation}
\begin{array}{cccc}
\phi\colon & S^n\backslash\{N\} & \ra & \RR^n \\
& (x^0,x^1,\ldots, x^n) & \mapsto & \frac{1}{1-x^0}(x^1,\ldots,x^n)
\end{array}
\end{equation}
The map $\phi$ is a conformal diffeomorphism (w.r.t. the round metric on $S^n$ -- induced from the standard flat metric on the ambient $\RR^{n+1}$ -- and w.r.t. the standard metric on $\RR^n$). The conformal factor is $\Omega=\frac{1}{(1-x^0)^2}$.
\end{example}

\begin{example} \label{l6 ex: Conf(R^1)}
Any diffeomorphism $\phi\colon \RR\ra \RR$ is a a conformal map (w.r.t. the metric $g=(dx)^2$ on the source and the target), with $\Omega=\left(\frac{d\phi}{dx}\right)^2$. 
\end{example}

\begin{example}[Inversion]
The map 
\begin{equation}
\begin{array}{cccc}
\phi\colon & \RR^n\backslash\{0\} & \ra & \RR^n\backslash\{0\} \\
& \vec{x} &\mapsto  & \frac{\vec{x}}{|| \vec{x} ||^2}
\end{array}
\end{equation}
is an orientation-reversing diffeomorphism. It is a conformal map (w.r.t. the metric induced from the standard one on $\RR^n$), with $\Omega= \frac{1}{||\vec{x}||^4}$.
\end{example}

The following lemma gives a full classification of local holomorphic maps on $\RR^2$.
\begin{lemma} \label{l6 lm}
Let $D\subset \RR^2$ be an open set. For a smooth map $\phi\colon D\ra \RR^2$ the following statements are equivalent:
\begin{enumerate}[(i)]
\item \label{l6 lm (i)} $\phi$ is conformal (w.r.t. the standard metric on source and target) 
\item  \label{l6 lm (ii)} $\phi$ is either holomorphic or antiholomorphic (we are identifying $\RR^2$ with $\CC$) and has no critical points in $D$.
\end{enumerate}
\end{lemma}
\begin{proof}
Let $x,y$ be the real coordinates on $D$ and let $u,v$ be the coordinates on the target $\RR^2$. Let $z=x+iy$ be the complex coordinate on $D$
%, with $\bar{z}=x-iy$ its conjugate, 
and let $w=u+iv$ be the complex coordinate on the target $\RR^2=\CC$.
The pullback of the target metric $g=du^2+dv^2=dw \,d\bar{w}$ is then
%$$ \phi^*g_{\RR^2} =\phi^*(du^2+dv^2)=(u_x^2+v_x^2) dx^2+(u_y^2+v_y^2) dy^2 + (u_x v_y+u_y v_x) dx dy
%$$
\begin{equation}\label{l6 lm phi^* g on C}
\phi^* g = \phi^* (dw \,d\bar{w}) = %\frac{\dd w}{\dd z} \frac{\dd }{}
\dd_z \phi\, \dd_z \bar\phi (dz)^2 + \dd_{\bar{z}} \phi\, \dd_{\bar{z}} \bar\phi (d\bar{z})^2 + (\dd_z \phi\, \dd_{\bar{z}}\bar\phi+\dd_{\bar{z}}\phi\, \dd_z \bar\phi) dz d\bar{z}
\end{equation}
We are using the standard notations for holmorphic/antiholomorphic derivatives: 
\begin{equation}
\dd_z=\frac{\dd}{\dd z}=\frac12(\dd_x-i \dd_y),\quad \dd_{\bar{z}}=\frac{\dd}{\dd \bar{z}}=\frac12(\dd_x+i \dd_y).
\end{equation}

(\ref{l6 lm (i)})$\Rightarrow$(\ref{l6 lm (ii)}):  If we know that $\phi$ is conformal, then 
\begin{equation}\label{l6 lm phi^* g=Omega g}
\phi^*g=\Omega g_D=\Omega dzd\bar{z}
\end{equation}
for some positive function $\Omega$, thus coefficients of  $(dz)^2$ and $(d\bar{z})^2$ must vanish. For this there are two possibilites:
\begin{enumerate}[(a)]
\item $\dd_{\bar z} \bar\phi=0$ (and thus also $\dd_{z}\bar\phi=0$), i.e., $\phi$ is holomorphic. In this case, comparing the $dz d\bar{z}$ term in (\ref{l6 lm phi^* g on C}) with (\ref{l6 lm phi^* g=Omega g}), we have 
\begin{equation} \label{l6 lm Omega1}
\Omega=|\dd_z \phi|^2.
\end{equation}
\item $\dd_z \phi=0$ (and thus also $\dd_{\bar{z}}\bar\phi=0$), i.e., $\phi$ is antiholomorphic. In this case we have
\begin{equation}\label{l6 lm Omega2}
\Omega=|\dd_{\bar{z}} \phi|^2.
\end{equation}
\end{enumerate}
Note that in both cases $\phi$ cannot have critical points, since  there $\Omega$ would vanish (by (\ref{l6 lm Omega1}), (\ref{l6 lm Omega2})).

(\ref{l6 lm (ii)})$\Rightarrow$ (\ref{l6 lm (i)}): Assume $\phi$ is holomorphic with no critical points. Then $\dd_{\bar z}\phi=\dd_{z}\bar\phi=0$, thus by (\ref{l6 lm phi^* g on C}), $\phi^*g = |\dd_z \phi|^2 dzd\bar{z}$. Hence, $\phi$ is conformal with $\Omega=|\dd_z \phi|^2$ which is positive, since $\phi$ has no critical points. The antiholomorphic case is similar.
\end{proof}

\begin{example}[M\"obius transformations]
The Lie group 
\begin{equation}
PSL_2(\CC)= \left\{\left.\left( 
\begin{array}{cc}
a&b\\c&d
\end{array}
\right) \; \right|\; a,b,c,d\in \CC,\; ad-bc=1 \right\} / \ZZ_2,
\end{equation}
where quotient by $\ZZ_2$ identifies a matrix and its negative, acts on the Riemann sphere $\bar\CC=\CP^1$ by fractional-linear transformations (or ``M\"obius transformations'')
\begin{equation}\label{l6 Mobius transf}
\left( 
\begin{array}{cc}
a&b\\c&d
\end{array}
\right)\colon \quad z\mapsto z'=\frac{az+b}{cz+d}
\end{equation}
For any element of $PSL_2(\CC)$, (\ref{l6 Mobius transf}) is a conformal map with conformal factor (w.r.t. the standard metric on $\RR^2$)\footnote{
Note that if $c\neq 0$, then (\ref{l6 Mobius Omega}) vanishes at the point $\{\infty\} \in\bar{\CC}$ (and also explodes at $z=-\frac{d}{c}$)  which seems to contradict that $\Omega$ should be a positive (and everywhere defined) function. This is to do with the fact that we chose a metric on $\CC$ which does not extend to the point  $\{\infty\}$. One can choose another metric in the same conformal class which extends to $\{\infty\}$ (e.g. the round metric on $\bar\CC$ seen as $S^2$), then $\Omega$ relative to that metric will be truly everywhere positive and everywhere defined.
}
\begin{equation}\label{l6 Mobius Omega}
\Omega=\left|\frac{dz'}{dz} \right|^2 = \frac{1}{|cz+d|^{4}}
\end{equation}
%Note: if $c\neq 0$, then $\Omega$ vanishes

For instance, one has the following interesting classes of M\"obius transformations:
\begin{enumerate}[(a)]
\item \label{l6 Mobius (a)} Element $\left( 
\begin{array}{cc}
1&b\\0&1
\end{array}
\right)$, with $b\in\CC$, acts by translation $z\mapsto z+b$.
\item \label{l6 Mobius (b)} $\left( 
\begin{array}{cc}
e^{i\phi/2}&0\\0& e^{-i\phi/2}
\end{array}
\right)$ acts by rotation by angle $\phi$,  $z\mapsto e^{i\phi} z$.
\item \label{l6 Mobius (c)} $\left( 
\begin{array}{cc}
\lambda^{1/2}& 0\\0& \lambda^{-1/2}
\end{array}
\right)$ with $\lambda>0$ acts by dilation $z\mapsto \lambda z$.
\item  \label{l6 Mobius (d)} 
$\left( 
\begin{array}{cc}
1&0\\c&1
\end{array}
\right)$ with $c\in \CC$ yields a \emph{special conformal transformation} (SCT), ${z\mapsto \frac{z}{cz+1}=\frac{1}{c+z^{-1}} }$. In particular, it maps  $-c^{-1}\mapsto \infty$ and $\infty\mapsto c^{-1}$.
\end{enumerate}
Note that translations, rotations and dilations are conformal automorphisms of $\CC\subset \bar{\CC}$, but SCTs are not -- they have a pole on $\CC$. 
\end{example}

\begin{example}
Consider the exponential map
\begin{equation}  
\CC/2\pi i \ZZ \xra{\exp} \CC\backslash\{0\}  
\end{equation}
from the cylinder to the punctured plane. By Lemma \ref{l6 lm}, it is a conformal diffeomorphism, with $\Omega=e^{z+\bar{z}}$ (w.r.t. to the standard Euclidean metric on the source and target).
\end{example}

\section{Conformal vector fields}
One can think of conformal vector fields as ``infinitesimal conformal maps.''
\begin{definition}
A conformal vector field on a (pseudo-)Riemannian manifold $(M,g)$ is a vector field $v\in \mathfrak{X}(M)$ satisfying 
\begin{equation}\label{l6 cvf defining eq}
\LL_v g =\omega\,g
\end{equation}
for some function $\omega\in C^\infty(M)$ (the inifitesimal conformal factor); $\LL_v$ stands for the Lie derivative along $v$.\footnote{Note that there is no positivity constraint on $\omega$.}
We denote the set of all conformal vector fields on $(M,g)$  by $\mr{conf}(M,g)$.
%The left hand side of (\ref{l6 cvf defining eq}) is the Lie derivative along $v$.
\end{definition}

Conformal vector fields form a Lie subalgebra in the Lie algebra of all vector fields w.r.t. the standard Lie bracket of vector fields:
\begin{equation}
\mr{conf}(M,g) \subset \mathfrak{X}(M).
\end{equation}

One has a natural inclusion
\begin{equation}
\iota\colon \mr{Lie}\left(\mr{Conf}(M,g)\right) \hra \mr{conf}(M,g)
\end{equation}
of the Lie algebra of the group of conformal automorphisms into the Lie algebra of conformal vector fields (by taking derivative the at $t=0$ of a 1-parametric subgroup).  If $M$ is compact, $\iota$ is an isomorphism (one can construct the flow of a conformal vector field $v\mapsto \mr{Flow}_t (v)$ yielding a 1-parametric subgroup of $\mr{Conf}(M,g)$). However, for $M$ noncompact, conformal vector fields can fail to be complete, so only a part of elements of $\mr{conf}(M,g)$ can be exponentiated.

%If $M$ is \emph{compact}, then

\section{Conformal symmetry of $\RR^{p,q}$ with $p+q>2$}
\subsection{Conformal vector fields on $\RR^{p,q}$}
Consider the space $\RR^{p,q}$ with its standard metric $g=\eta_{ij}dx^i dx^j$ with the matrix $\eta_{ij}$ being
\begin{equation}
 \eta_{ij}=\mr{diag}(\underbrace{+1,\ldots,+1}_p,\underbrace{-1,\ldots,-1}_q). 
\end{equation}
We denote $n=p+q$.

We are looking for conformal vector fields $v=v^k(x)\dd_k$ on $\RR^{p,q}$. (Summation over repeated indices is implied everywhere in this section.) The defining equation (\ref{l6 cvf defining eq}) for them takes the form
\begin{equation}\label{l6 eq1}
\dd_i v_j+\dd_j v_i=\omega \eta_{ij}
\end{equation}
with $v_i\colon= \eta_{ij}v^j$. (\ref{l6 eq1}) is a system of $n^2$ (dependent) differential equations on $n+1$ unknown functions -- components $v_i$ of the conformal vector field and $\omega$ -- the infinitesimal conformal factor. Solving (\ref{l6 eq1}) is a well-known exercise \cite{Schottenloher,DMS,Ginsparg};
%\marginpar{locate the right reference} 
for reader's convenience, we reproduce the argument.\footnote{Part of the value of the explicit argument here is that it gives an explanation (albeit a technical one) of why the cases $n=1,2$ and $n>2$ are so vastly different.}
\begin{enumerate}[(i)]
\item Contracting (\ref{l6 eq1}) with $\eta^{ij}$, we get
\begin{equation}\label{l6 eq2}
2\underbrace{\dd_i v^i}_{\mr{div}\,v}=n \omega.
\end{equation}
\item Applying $\dd^j$ to (\ref{l6 eq1}), we get
\begin{equation}
\dd_i(\mr{div}\, v)+\Delta v_i=\dd_i\omega,
\end{equation}
where $\Delta=\dd_j\dd^j$. By (\ref{l6 eq2}), this implies
\begin{equation}\label{l6 eq3}
\Delta v_i = \left(1-\frac{n}{2}\right) \dd_i\omega.
\end{equation}
\item Applying $\dd_j$ to (\ref{l6 eq3}), symmetrizing in $i\leftrightarrow j$ and using (\ref{l6 eq1}), we get
\begin{equation}\label{l6 eq4}
\frac12 \eta_{ij} \Delta \omega = \left(1-\frac{n}{2}\right) \dd_i\dd_j\omega.
\end{equation}
\item Applying $\dd^i$ to (\ref{l6 eq3}), we get 
\begin{equation}
\Delta(\underbrace{\mr{div}\,v}_{\underset{(\ref{l6 eq2})}{=}\frac{n}{2}\omega})=\left(1-\frac{n}{2}\right) \Delta\omega,
\end{equation}
which implies
\begin{equation}\label{l6 eq5}
(n-1)\Delta\omega =0
\end{equation}
\item Equations (\ref{l6 eq4}) and (\ref{l6 eq5}) imply that for $n\neq 1,2$ one has 
\begin{equation}\label{l6 eq6}
\dd_i\dd_j\omega=0.
\end{equation}
I.e., $\omega$ is at most linear in coordinates.
\item Taking a derivative of (\ref{l6 eq1}), we have
\begin{equation}\label{l6 eq6.5}
\dd_i\dd_jv_k+\dd_i\dd_k v_j = \dd_i\omega\,\eta_{jk}
\end{equation}
The equation $(\ref{l6 eq6.5})+(\ref{l6 eq6.5})_{(ijk)\ra (jik)}-(\ref{l6 eq6.5})_{(ijk)\ra (kij)}$ then reads 
\begin{equation}\label{l6 eq7}
2\dd_i\dd_j v_k=\dd_i\omega\,\eta_{jk} + \dd_j\omega\, \eta_{ik}-\dd_k\omega\, \eta_{ij}
\end{equation}
\item Equation (\ref{l6 eq6}) and (\ref{l6 eq7}) together imply, for $n\neq 1,2$, that
\begin{equation}
\dd_i\dd_j\dd_k v_l =0.
\end{equation}
I.e., $v$ is at most quadratic in coordinates.
\end{enumerate}

Now, specializing to the case $n>2$, we have an ansatz
\begin{equation}\label{l6 ansatz}
v_i(x)=a_i+b_{ij}x^j + c_{ijk} x^j x^k,\quad \omega(x)= 2\mu+4\nu_i x^i
\end{equation}
with $a_i,b_{ij},c_{ijk},\mu,\nu_i$ some coefficients. Substituting this ansatz into (\ref{l6 eq1}), we find that (\ref{l6 ansatz}) is a conformal vector field and its conformal factor if the coefficients satisfy the following:
\begin{enumerate}[(a)]
\item No restriction on $a_i$.
\item $b_{ij}+b_{ji}=2\mu\eta_{ij}$ which implies 
\begin{equation}
b_{ij}=\mu \eta_{ij}+\beta_{ij}
\end{equation}
with some anti-symmetric tensor $\beta_{ij}=-\beta_{ji}$.
\item $c_{ijk}+c_{jik}=2\nu_k\eta_{ij}$ which implies, similarly to the derivation of (\ref{l6 eq7}) above,
\begin{equation}
 c_{ijk}=\nu_j \eta_{ik}+\nu_k \eta_{ij}-\nu_i \eta_{jk}.
\end{equation}
\end{enumerate}

This proves the following.
\begin{thm}[Liouville] \label{l6 thm Liouville}
For $n=p+q>2$, the Lie algebra of conformal vector fields on $\RR^{p,q}$ splits into the following subspaces:
\begin{equation}\label{l6 Liouville thm eq}
\mr{conf}(\RR^{p,q}) = \underset{\simeq \RR^n}{\{\mr{translations}\}}\oplus \underset{\simeq \mathfrak{so}(p,q)}{\{\mr{rotations}\}}\oplus \underset{\simeq \RR}{\{\mr{dilations}\}}\oplus \underset{\simeq \RR^n}{\{\mr{SCTs}\}}
\end{equation}
where SCTs stands for ``special linear transformations.'' Explicitly, these conformal vector fields are as follows.

\begin{center}
\begin{tabular}{c|c|c}
& conf. vector field & $\omega$  \\ \hline
translation & $v^i(x)=a^i$ & $0$  \\
rotation & $v^i(x)=\beta^i_j x^j$ with $\beta_{ij}=-\beta_{ji}$ & $0$  \\
dilation & $v^i(x)=\mu x^i$ & $2\mu$  \\
SCT& $v^i(x) = 2(\vec{x},\vec{\nu})x^i-\nu^i ||\vec{x}||^2$ & $4(\vec{\nu},\vec{x})$
\end{tabular}
\end{center} 

%\begin{tabular}{c|c|c|| c|c}
%& conf. vector field & $\omega$ & conf. map & $\Omega$ \\ \hline
%translation & $v^i(x)=a^i$ & $0$ & $x^i\mapsto x^i+a^i$, $a^i\in \RR^n$& $1$ \\
%rotation & $v^i(x)=\beta^i_j x^j$ with $\beta_{ij}=-\beta_{ji}$ & $0$ & $x^i\mapsto O^i_j x^j$, $O^i_j\in SO(p,q)$ & $1$ \\
%dilation & $v^i(x)=\mu x^i$ & $2\mu$ & $x^i\mapsto \lambda x^i$, $\lambda>0$ & $\lambda^2$ \\
%SCT& $v^i(x) = 2(\vec{x},\vec{\nu})x^i-\nu^i ||\vec{x}||^2$ & $4(\vec{\nu},\vec{x})$ &
%$x^i\mapsto \frac{x^i-||\vec{x}||^2 b^i}{1-2(\vec{b},\vec{x})+||\vec{b}||^2 ||\vec{x} ||^2}$, $\vec{b}\in \RR^n$ & $(1-2(\vec{b},\vec{x})+||\vec{b}||^2 || \vec{x} ||^2)^{-2}$
%\end{tabular}  
\end{thm}

\subsection{Finite conformal automorphisms of $\RR^{p,q}$ with $p+q>2$. }
Here are the finite\footnote{``Finite'' conformal maps are just conformal maps. We use the adjective ``finite'' to emphasize the difference from ``infinitesimal conformal maps,'' i.e., conformal vector fields.} conformal maps 
%(given by a flow generated by the conformal vector fields above)
exponentiating (via %$v\mapsto \mr{Flow}_{t=1}(v)$
constructing the flow in time $1$) the conformal vector fields of Theorem \ref{l6 thm Liouville}.\footnote{Under the flow-in-time-one map, the parameters of the finite conformal maps are related to the parameters of the conformal vector fields by  $\vec{a}=\vec{a}$, $O=\exp(\beta)$, $\lambda=e^\mu$, $\vec{b}=\vec{\nu}$.}

\begin{center}
\begin{tabular}{c|c|c}
&  conf. map & $\Omega$ \\ \hline
translation &  $x^i\mapsto x^i+a^i$, $\vec{a}\in \RR^n$& $1$ \\
rotation & $x^i\mapsto O^i_j x^j$, $O^i_j\in SO(p,q)$ & $1$ \\
dilation & $x^i\mapsto \lambda x^i$, $\lambda>0$ & $\lambda^2$ \\
SCT&
$x^i\mapsto \frac{x^i-||\vec{x}||^2 b^i}{1-2(\vec{b},\vec{x})+||\vec{b}||^2 ||\vec{x} ||^2}$, $\vec{b}\in \RR^n$ & $(1-2(\vec{b},\vec{x})+||\vec{b}||^2 || \vec{x} ||^2)^{-2}$
\end{tabular} 
\end{center}

\begin{definition}\label{l6 def conformal compactification}
Given a manifold $M$  equipped with a conformal structure $\gamma_N$ (a choice of metric modulo Weyl transformations), 
we say that a \emph{compact} manifold $N$  equipped with conformal structure, is a \emph{conformal compactification of $M$}, if the following holds:
\begin{itemize}
\item One has an embedding $M\hra N$ with open dense image.
\item All conformal vector fields on $M$ extend to $N$. (And $N$ they can automatically be integrated to conformal automorphisms.)
\end{itemize}
\end{definition}

\begin{remark}[On finite SCTs]
\begin{enumerate}[(a)]
\item
A finite SCT can be written as 
\begin{equation}
(\mr{inversion})\circ (\mr{translation\;by\;}-\vec{b})\circ(\mr{inversion}).
\end{equation}
I.e., it maps $\vec{x}\mapsto \vec{x}'$ with image and preimage related by
\begin{equation}
 \frac{\vec{x}'}{||\vec{x}'||^2} =\frac{\vec{x}}{||\vec{x}||^2}-\vec{b}.
\end{equation}
\item A finite SCT is not everywhere defined as a map $\RR^{p,q}\ra \RR^{p,q}$ (the denominator in the formula for SCT may vanish). This corresponds to the quadratic vector field describing the infinitesimal SCT not being complete on $\RR^{p,q}$.
\item In Section \ref{sss: SO action on R^p,q, conf compactification} we will construct a conformal compactification $N^{p,q}$ of $\RR^{p,q}$,  
%-- a \emph{compact} manifold $N^{p,q}$ containing $\RR^{p,q}$ as an open dense subset, 
such that SCTs are everywhere well-defined on $N^{p,q}$.
\end{enumerate}
\end{remark}

%\begin{definition}\label{l6 def conformal compactification}
%Given a manifold $M$  equipped with a conformal structure $\gamma_N$ (a choice of metric modulo Weyl transformations), 
%we say that a \emph{compact} manifold $N$  equipped with conformal structure, is a \emph{conformal compactification of $M$}, if the following holds:
%\begin{itemize}
%\item One has an embedding $M\hra N$ with open dense image.
%\item All conformal vector fields on $M$ extend to $N$. (And $N$ they can automatically be integrated to conformal automorphisms.)
%\end{itemize}
%\end{definition}

We also remark that in the exceptional dimensions $n=1,2$, the r.h.s. of (\ref{l6 Liouville thm eq}) is a (small) subspace of the  l.h.s., while the l.h.s is an $\infty$-dimensional Lie algebra.

\marginpar{Lecture 7,\\ 9/7/2022}
\begin{thm}\label{l7 thm}
Assume $p+q>2$.  
\begin{enumerate}[(i)]
\item \label{l7 thm (i)} One has an isomorphism of Lie algebras 
\begin{equation}\label{l7 conf = so}
\mr{conf}(\RR^{p,q})\cong \mathfrak{so}(p+1,q+1).
\end{equation}
\item \label{l7 thm (ii)} For the group $\mr{Conf}^\mr{sing}$ of almost everywhere defined conformal automorphisms of $\RR^{p,q}$, one has:
\begin{itemize}
\item If $-1$ and $1$ are in different connected components of $SO(p+1,q+1)$, then
\begin{equation}
\Conf^\mr{sing}_0(\RR^{p,q})\cong SO_0(p+1,q+1) 
\end{equation}
Subscript $0$ on both sides stands for ``connected component of $1$.''
\item Otherwise,
\begin{equation}
\Conf^\mr{sing}_0(\RR^{p,q})\cong SO_0(p+1,q+1)/\ZZ_2
\end{equation}
\end{itemize}
%
%{\small
%\begin{equation}
%\Conf^\mr{sing}_0(\RR^{p,q})\cong 
%\left\{ \begin{array}{lc}
%SO_0(p+1,q+1)/\ZZ_2 & \mbox{if $-1$ and $1$ are in the same connected component of $SO$} \\
%SO_0(p+1,q+1) & \mbox{otherwise}
%\end{array}\right.
%\end{equation}
%}
\item \label{l7 thm (iii)} The conformal manifold $\RR^{p,q}$ possesses a conformal compactification $N^{p,q}$ in the sense of Definition \ref{l6 def conformal compactification}.
\end{enumerate}
\end{thm}

For the proof, see \cite{Schottenloher}.

As a sanity check of (\ref{l7 conf = so}), let us check that the dimensions of both sides match:
\begin{multline}\label{l7 dim count}
\dim\conf(\RR^{p,q})\underset{(\ref{l6 Liouville thm eq})}{=}\dim\{\mr{translations}\}+\dim\{\mr{rotations}\}+\dim\{\mr{dilations}\}+\dim\{\mr{SCTs}\} \\
= n+\frac{n(n-1)}{2}+1+n = \frac{(n+1)(n+2)}{2} = \dim \mathfrak{so}(p+1,q+1)
\end{multline}

\subsection{Sketch of proof of Theorem \ref{l7 thm}: action of $SO(p+1,q+1)$ on $\RR^{p,q}$ and the conformal compactification of $\RR^{p,q}$}
\label{sss: SO action on R^p,q, conf compactification}
For the following construction, 
we also follow \cite{Schottenloher}.

%The logical scheme of the proof of Theorem\ref{l7 thm} is as follows:
%\begin{enumerate}[1.]
%\item We construct a compact manifold $N^{p,q}$ equipped with an inclusion $\RR^{p,q}\hra N^{p,q}$ (compatible with conformal structures) as an open dense subset.
%\item We construct an action of $SO(p+1,q+1)$ on $N^{p,q}$ by conformal diffeomorphisms. The only elements acting trivially are multiples of identity, i.e., $1$ and $-1$ (in the case when $-1$ belongs to $SO(p+1,q+1)$).
%\item The differential of the action of $SO(p+1,q+1)$ gives an injective Lie algebra map $\mathfrak{so}(p+1,q+1)\hra \conf(N^{p,q})$ (and by restriction to $\RR^{p,q}$, a map to $\mathfrak{so}(p+1,q+1)\hra\conf{\RR^{p,q}}$). By the dimension count (\ref{l7 dim count}), these inclusions are in fact isomorphisms. This proves (\ref{l7 thm (i)}) and (\ref{l7 thm (iii)}) of Theorem \ref{l7 thm}.
%\end{enumerate}

\subsubsection{Case of $\RR^n$}. Consider first the case  $(p,q)=(n,0)$. 
\begin{itemize}
\item The group $SO(n+1,1)$ acts on $\RR^{n+1,1}$ by linear isometries and preserves the light cone
\begin{equation}
LC=\{(x^0,\ldots,x^n,y)\in \RR^{n+1,1}\;|\; (x^0)^2+\cdots+(x^n)^2-y^2=0\} \quad \subset \RR^{n+1,1}
\end{equation}
\item We have two commuting actions 
\begin{equation}
SO(n+1,1)\quad \rotatebox[origin=c]{90}{$\curvearrowleft$}\quad LC \underset{\mr{dilations}}{\rotatebox[origin=c]{270}{$\curvearrowright$}} \RR^*
\end{equation}
\item In particular, $SO(n+1,1)$ acts on $LC-\{0\}/\RR^*$.
\item $LC-\{0\}$ inherits a \emph{degenerate} metric from $\RR^{n+1,1}$. Its kernel is the fundamental vector field of the $\RR^*$-action and thus is killed by quotienting over $\RR^*$. 
\item By the previous,
$LC-\{0\}/\RR^*$ inherits a conformal structure and $SO(n+1,1)$ acts on $LC-\{0\}/\RR^*$ by conformal maps.
\item Note: $LC-\{0\}/\RR^*$ can be identified with the unit sphere $S^n\subset \RR^{n+1}$: intersecting $LC$ with the hyperplane $y=1$ in $\RR^{n+1,1}$, we are selecting a single point from each $\RR^*$-orbit. 
%\textcolor{red}{PICTURE.}
\begin{figure}[H]
%$$\vcenter{\hbox{ \includegraphics[scale=0.5]{cob1.eps} }} $$
\begin{center}
\includegraphics[scale=1]{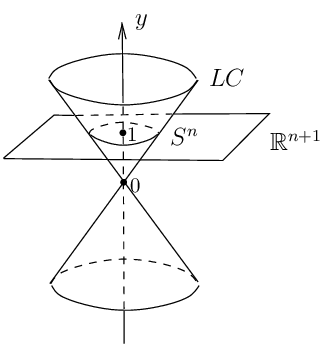}
\end{center}
\caption{Light cone and its section by $y=1$ hyperplane.}
\end{figure}
\item One has a stereographic projection 
$$S^n-\{\underbrace{(1,0,\ldots,0)}_{\mr{North\;pole}}
%\mr{North\; pole}
\}\ra \RR^n$$ 
(which is a conformal diffeomorphism). Thus we identify $S^n$ as a conformal compactification of $\RR^n$: conformal vector fields on $\RR^n$ extend to $S^n$ and finite conformal maps are everywhere defined on $S^n$.
\end{itemize}

\subsubsection{Case of general $\RR^{p,q}$}.
\begin{itemize}
\item We have the light cone 
\begin{equation}
LC=\left\{(x^0,\ldots,x^p,y^0,\ldots,y^q)\;\Big| \; \sum_{i=0}^p(x^i)^2-\sum_{j=0}^q (y^j)^2=0\right\} \quad \subset \RR^{p+1,q+1}.
\end{equation}
\item We have two commuting actions 
\begin{equation}
SO(p+1,q+1)\underset{\mr{lin.\;isometries}}{\rotatebox[origin=c]{90}{$\curvearrowleft$}} LC-\{0\} \underset{\mr{dilations}}{\rotatebox[origin=c]{270}{$\curvearrowright$}} \RR^*.
\end{equation}
\item We have a projection 
\begin{equation}\label{l7 pi}
\pi\colon LC-\{0\}\;\;\ra\;\; \RP^{n+1}.
\end{equation}
Denote its image 
\begin{equation}\label{l7 N^p,q def}
N^{p,q}\colon= \mr{im}(\pi)\simeq (LC-\{0\})/\RR^*
\end{equation}
Being a submanifold of a compact manifold $\RP^{n+1}$, $N^{p,q}$ is compact.
\item Consider the map $\iota\colon \RR^{p,q}\ra N^{p,q}$ defined by
{\small
\begin{multline}\label{l7 iota}
 \iota(x^1,\ldots,x^p,y^1,\ldots,y^q)=\\
 =
\left( 
\frac12 \Big(1-\sum_{i=1}^p(x^i)^2+\sum_{j=1}^q (y^j)^2\Big):x^1:\cdots: x^p:\frac12\Big(1+\sum_{i=1}^p(x^i)^2-\sum_{j=1}^q (y^j)^2\Big):y^1:\cdots: y^q
\right)
\end{multline}
}
where $(-:-\cdots :-)$ stands for the homogeneous coordinates on the projective space. The map $\iota$ is injective and has open dense image.
\end{itemize}

%This discussion proves the following.
%\begin{lemma}
%The compact manifold $N^{p,q}$ defined by (\ref{l7 N^p,q def}) is a conformal compactification of $\RR^{p,q}$
%in the sense of Definition \ref{l6 def conformal compactification}.
%\end{lemma}

%To summarize, the logical scheme of the proof of Theorem \ref{l7 thm} is as follows:
\begin{proof}[Sketch of proof of Theorem \ref{l7 thm}]
\leavevmode
\begin{enumerate}[1.]
\item We have constructed a compact manifold $N^{p,q}$ equipped with an inclusion $\RR^{p,q}\hra N^{p,q}$ (compatible with conformal structures) as an open dense subset.
\item We have constructed an action of $SO(p+1,q+1)$ on $N^{p,q}$ by conformal diffeomorphisms. The only elements acting trivially are multiples of identity, i.e., $1$ and $-1$ (in the case when $-1$ belongs to $SO(p+1,q+1)$).
\item \label{l7 proof 3} The differential of the action of $SO(p+1,q+1)$ gives an injective Lie algebra map $\mathfrak{so}(p+1,q+1)\hra \conf(N^{p,q})$ (and by restriction to $\RR^{p,q}$, an inclusion $\mathfrak{so}(p+1,q+1)\hra\conf(\RR^{p,q})$). By the dimension count (\ref{l7 dim count}), these inclusions are in fact isomorphisms. This proves (\ref{l7 thm (i)}) and (\ref{l7 thm (iii)}) of Theorem \ref{l7 thm}, identifying (\ref{l7 N^p,q def}) as the desired conformal compactification.
\item The previous two points imply that the Lie group $\mr{Conf}(N^{p,q})$ contains $SO(p,q)/\ZZ_2$ and both groups have the same Lie algebra. That implies that the connected components of $1$ in both groups coincide. That proves (\ref{l7 thm (ii)}) of Theorem \ref{l7 thm}.
\end{enumerate}
\end{proof}

\begin{remark}
The product of unit spheres
\begin{equation}
S^p\times S^q = \{(x^0,\ldots,x^p,y^0,\ldots,y^q)\;|\; \sum_{i=0}^p (x^i)^2=1,\; \sum_{j=0}^q (y^j)^2=1\} %\quad \ra N^{p,q}
\end{equation}
is a submanifold of $LC-\{0\}$ and intersect each $\RR^*$-orbit twice ($(x,y)$ and $(-x,-y)$ are in the same $\RR^*$-orbit). Thus, one has a twofold covering map 
\begin{equation}
S^p\times S^q \ra N^{p,q}
\end{equation} 
given by the projection (\ref{l7 pi}) restricted to $S^p\times S^q$.
In particular, we can identify $N^{p,q}$ with the quotient
\begin{equation}
N^{p,q}\simeq S^p\times S^q/\ZZ_2
\end{equation}
where $\ZZ_2$ acts by the diagonal antipodal map, $(x,y)\mapsto (-x,-y)$.
\end{remark}

\section{Conformal symmetry of $\RR^2$}
A vector field $v=v_i(x,y)\dd_i$ (with $x=x^1$, $y=x^2$) on $\RR^2$ equipped with the standard Euclidean metric is conformal if  
the equation (\ref{l6 cvf defining eq}) holds:
\begin{equation}
\dd_i v_j+\dd_j v_i=\omega \delta_{ij} \quad \Leftrightarrow
\left\{ 
\begin{array}{l}
\dd_x v_x=\dd_y v_y =\frac12 \omega \\
\dd_x v_y=-\dd_y v_x
\end{array}
\right.
\end{equation}
for some function (conformal factor) $\omega$. On the right side we can recognize the Cauchy-Riemann equations. Thus, the vector field $v=v_i\dd_i$ is conformal if and only if the function 
\begin{equation}
u\colon=v_x+iv_y
\end{equation} 
is holomorphic.
Note that the vector field $v$ can be written in terms of the holomorphic function $u$ and its complex conjugate $\bar{u}$ as 
\begin{equation}\label{l7 cvf via holom vf}
v=u(z)\dd_z+\bar{u}(\bar{z})\dd_{\bar{z}} = 2\,\mr{Re}(u(z)\dd_z)
\end{equation}
The corresponding conformal factor is $\omega=\dd_z u+\dd_{\bar{z}}\bar{u}$.

In (\ref{l7 cvf via holom vf}) we use the complex coordinate $z=x+iy$, its conjugate $\bar{z}=x-iy$ and the corresponding derivatives $\dd_z=\frac12(\dd_x-i\dd_y)$, $\dd_{\bar{z}}=\frac12(\dd_x+i\dd_y)$.

To summarize, we have the following.
\begin{lemma}\label{l7 lemma cvf to hol vf}
One has an isomorphism of Lie algebras
\begin{equation}\label{l7 cvf-holom vf iso}
\psi\colon \conf(\RR^2) \xra{\sim} \{\mr{holomorphic\;vector\;fields\;on\;}\CC\}.
\end{equation}
It maps
a conformal vector field $v_x\dd_x+v_y\dd_y$ to the holomorphic vector field $u(z)\dd_z$ where $u(z)=v_x+iv_y$. The inverse map $\psi^{-1}$
assigns to a holomorphic vector field $u(z)\dd_z$ a conformal vector field 
$ 2\,\mr{Re}(u(z)\dd_z)= 
u(z)\dd_z+\bar{u}(\bar{z})\dd_{\bar{z}}$. %(\mr{Re}\,u)\dd_x+(\mr{Im}\,u)\dd_y.
\end{lemma}
The fact that $\psi$ intertwines  the Lie brackets on the two sides of (\ref{l7 cvf-holom vf iso}) is a straightforward check.

\begin{remark}
In the isomorphism (\ref{l7 cvf-holom vf iso}), we are thinking of both sides as Lie algebras over $\RR$. However, the right hand side is also a Lie algebra over $\CC$. Multiplication by $i$ on the right side translates in the left side to acting on a conformal vector field by pointwise rotation by $\pi/2$ (in the tangent space at each point of $\RR^2$). 
%I.e. it maps a conformal vector field $v\in\Gamma(\RR^2,T\RR^2)$ to $J\circ v$ where $J\in \mr{End}(T\RR^2)$
\end{remark}

\marginpar{Lecture 8,\\
9/9/2022}

Lemma \ref{l7 lemma cvf to hol vf} classifies infinitesimal conformal maps; its counterpart for finite conformal maps is Lemma  \ref{l6 lm} above, or its rephrasing:

\begin{lemma}\label{l8 lm 2.21}
Let $D,D'$ be two open sets in $\CC$. A map $\phi\colon D\ra D'$ is a conformal diffeomorphism if and only if $\phi$ is either biholomorphic or biantiholomorphic (i.e., the complex conjugate map $\bar\phi\colon D\ra \bar{D'}$ is bihomolomorphic).
\end{lemma}

\subsection{%$\conf(\CC^*)$
Conformal vector fields on $\CC^*$, Witt algebra} 
\label{sss: Witt algebra}
\begin{definition}
We define the Witt algebra $\mc{W}$ as the Lie algebra of meromorphic vector fields on $\CC$ with a pole (of finite order) allowed only at $0$.
%$\CP^1=\CC\cup\{\infty\}$ with  poles (of finite order) allowed only at points $0\in \CP^1$ and $\infty\in \CP^1$. 
%In terms of the 
The Lie algebra $\mc{W}$ has a standard
basis of meromorphic vector fields 
\begin{equation}
l_n=-z^{n+1}\frac{\dd}{\dd z}, \quad n\in \ZZ.
\end{equation}
Thus,
the Witt algebra is
\begin{equation}\label{l8 W}
\mc{W}=\{\sum_{n=-n_0}^{\infty} c_n l_n \;|\; c_n\in \CC, \;\mr{the\;sum\;converges\;on\;}\CC^* \}.
\end{equation}
\end{definition}

The generators $l_n$ of $\mc{W}$ satisfy the commutation relations
\begin{equation}\label{l8 [l_n,l_m]}
[l_n, l_m]=(n-m)l_{n+m}.
\end{equation}
Indeed: 
\begin{multline}
[-z^{n+1}\dd_z,-z^{m+1}\dd_z]=z^{n+1}[\dd_z,z^{m+1}\dd_z]-z^{m+1}[\dd_z,z^{n+1}\dd_z] =\\=((m+1)z^{n+m+1}-(n+1)z^{n+m+1})\dd_z= (m-n)z^{n+m+1}\dd_z=(n-m)l_{n+m}.
\end{multline}

There are several relevant variants of the Lie algebra $\mc{W}$, all with the same collection of generators $\{l_n\}$ but with different asymptotic conditions on the coefficients $c_n$ as $n\ra \pm\infty$:
\begin{enumerate}[(i)]
%\item \marginpar{This one is not a Lie algebra} All formal Laurent series of the form 
%\begin{equation}
%\mc{W}^\mr{formal}=\CC[[z,z^{-1}]]\dd_z=\{\sum_{n=-\infty}^\infty c_n l_n\;|\; c_n \in \CC\}
%\end{equation}
%(no conditions on $c_n$).
\item Holomorphic vector fields on the punctured formal disk:
\begin{equation}
\CC[[z,z^{-1}]\dd_z=\{\sum_{n=-n_0}^\infty c_n l_n\;|\; c_n \in \CC\}.
\end{equation}
-- This is a good model for the local conformal algebra $\mc{A}^\mr{loc}$ of Section \ref{sss 1.7.1 action of cvf on V}.
\item Meromorphic vector fields on $\CP^1$ with  finite-order poles allowed only at $0$ and $\infty$: 
\begin{equation}
\{\sum_{n=-n_0}^{n_1}c_n l_n\;|\; c_n\in \CC\}.
\end{equation}
-- This model has the benefit that it is symmetric under the involution $z\ra 1/z$ on $\CP^1$.
\end{enumerate}

We remark that the space of vector fields with coefficients in all formal Laurent power series $\displaystyle\{\sum_{n=-\infty}^\infty c_n l_n\}$ does not form a Lie algebra, since coefficients of the Lie bracket of two elements involves infinite sums that do not have to converge.

By abuse of notations and terminology, we will call all 
complex Lie algebras spanned by $\{l_n\}_{n\in\ZZ}$
%variants of $\mc{W}$, 
with  different decay conditions on coefficients, the Witt algebra and denote them $\mc{W}$.

By Lemma \ref{l7 lemma cvf to hol vf}, conformal vector fields on $\CC^*=\CC\backslash \{0\}$ are the real parts of meromorphic vector fields on $\CC^*$:
\begin{equation}\label{l8 conf(C^*)}
\conf(\CC^*)\simeq \mc{W}=\mr{span}_{\CC}\{l_n\}_{n\in\ZZ}
\end{equation}
(When we write ``span,'' we are being noncommittal about the decay  conditions on coefficients.)
Thus, one may also write
\begin{equation}
\conf(\CC^*)=\mr{span}_{\RR}\{l_n+\bar{l}_n,i(l_n-\bar{l}_n)\}_{n\in \ZZ}.
\end{equation}
Thus, $\conf(\CC^*)$ embeds as a real slice into its complexification 
\begin{equation}
\conf(\CC^*)\otimes_\RR \CC=\underbrace{\mc{W}}_{\mr{span}_\CC\{l_n\}}\oplus \underbrace{\ol{\mc{W}}}_{\mr{span}_\CC\{\bar{l}_n\}} . 
\end{equation}
Here  
\begin{equation}
\bar{l}_n=-\bar{z}^{n+1}\dd_{\bar{z}}
\end{equation} 
are the antimeromorphic vector fields on $\CC^*$ complex-conjugate to $l_n$. They satisfy the commutation relation similar to (\ref{l8 [l_n,l_m]}), 
\begin{equation}
[\bar{l}_n,\bar{l}_m]=(n-m)\bar{l}_{n+m}.
\end{equation}
Also, one has 
\begin{equation}
[l_n,\bar{l}_m]=0.
\end{equation}

\subsubsection{Some interesting Lie subalgebras of $\conf(\CC^*)$}
Here are some relevant Lie subalgebras of $\conf(\CC^*)$:
\begin{enumerate}[(a)]
\item 
Conformal vector fields on $\CC$:
\begin{equation}\label{l8 conf(C)}
%\conf(\CC)=
\mr{span}_\RR\{l_n+\bar{l}_n,i(l_n-\bar{l}_n)\}_{n\geq -1} %=\psi^{-1}\mr{span}_\CC \{l_n\}_{n\geq -1} .
\end{equation}
Indeed, vector fields $l_n,\bar{l}_n$ are holomorphic at $0$ iff $n\geq -1$.
\item Conformal vector fields on $\CC$  vanishing at $0$:
\begin{equation}
\mr{span}_\RR\{l_n+\bar{l}_n,i(l_n-\bar{l}_n)\}_{n\geq 0}
\end{equation}
Indeed, $l_n,\bar{l}_n$ vanish at $0$ iff $n\geq 0$.
\item Conformal vector fields on $\CP^1\backslash\{0\}$:
\begin{equation}\label{l8 conf(CP^1-0)}
\mr{span}_\RR\{l_n+\bar{l}_n,i(l_n-\bar{l}_n)\}_{n\leq 1}
\end{equation}
Indeed in the local coordinate $w=z^{-1}$ on $\CP^1\backslash\{0\}$ one has $l_n=w^{-n+1}\frac{\dd}{\dd w}$. Thus, $l_n$ is regular at the point $z=\infty$ (or $w=0$) iff $-n+1\geq 0$. (And similarly for $\bar{l}_n$.)
\end{enumerate}

\begin{remark}
Naively, the punctured plane $\CC^*$, the punctured unit disk $\{z\in\CC|0<|z|<1\}$ and annulus $\mr{Ann}_r^R=\{z\in\CC| r<|z|<R\}$ all have the same Lie algebra $\conf(-)\simeq \W=\mr{span}_\CC\{l_n\}_{n\in\ZZ}$. But in fact, for all these domains, the decay conditions on the coefficients $c_n$ in (\ref{l8 W}) are different. In the case of the annulus, the decay conditions depend on the inner and outer radii,\footnote{
Explicitly, the decay conditions for the annulus $\mr{Ann}_r^R$ are: $c_n\rho^n\underset{n\ra +\infty}{=}O(n^{-\infty})$ for any $0<\rho<R$ and $c_n\rho^n\underset{n\ra-\infty}{=}O(|n|^{-\infty})$ for any $\rho>r$.
} so that e.g. if one has $r'<r<R<R'$, then one has a proper inclusion $\mr{conf}(\mr{Ann}_{r'}^{R'})\hra \mr{conf}(\mr{Ann}_{r}^{R})$ (so that the thinner annulus has a bigger Lie algebra of conformal vector fields).
\end{remark}

\subsection{Conformal symmetry of $\CP^1$}
Conformal vector fields on $\CP^1$ are:
\begin{equation}\label{l8 conf(CP^1)}
\conf(\CP^1)=\mr{span}_\RR\{l_n+\bar{l}_n,i(l_n-\bar{l}_n)\}_{n\in\{-1,0,1\}}
\end{equation}
This is the subalgebra of $\conf(\CC^*)$ comprised of vector fields which are regular at $0$ and at $\infty$, i.e, it is  the intersection of (\ref{l8 conf(C)}) and (\ref{l8 conf(CP^1-0)}). The Lie algebra $\conf(\CP^1)$ is also isomorphic to $\mathfrak{sl}_2(\CC)$ and to $\mathfrak{so}(3,1)$.\footnote{
One has an action of $\mathfrak{so}(3,1)$ on $\CP^1$ by conformal vector fields by the construction of Section \ref{sss: SO action on R^p,q, conf compactification}. Also, in the last isomorphism in (\ref{l8 Conf(CP^1)}) we are referring to the finite version of that action.
} We can identify the generators of $\conf(\CP^1)$ explicitly as infinitesimal translations, rotation, dilation, and special canonical transformations:

\vspace{0.2cm}
\begin{center}
\begin{tabular}{ccc}
$-(l_{-1}+\bar{l}_{-1})$ & $=\dd_x $ & translation \\
$-i(l_{-1}-\bar{l}_{-1})$ & $=\dd_y$ & translation \\
$-(l_0+\bar{l}_0)$ & $=x\dd_x+y\dd_y$ & dilation \\
$-i(l_0-\bar{l}_0$) & $ = -y\dd_x+x\dd_y$ & rotation \\
$-(l_1+\bar{l}_1)$ & $=(x^2-y^2)\dd_x+2xy\dd_y$ & SCT\\
$-i(l_1-\bar{l}_1)$ & $=-2xy\dd_x + (x^2-y^2)\dd_y$ & SCT
\end{tabular}
\end{center}
\vspace{0.2cm}

The orientation-preserving part of the group of conformal automorphisms of $\CP^1$ is given by M\"obius transformations (\ref{l6 Mobius transf}):
\begin{equation}\label{l8 Conf(CP^1)}
\mr{Conf}_+(\CP^1)=PSL_2(\CC)\simeq SO_+(3,1)
\end{equation}
Where $SO_+(3,1)$ is the \emph{orthochronous} component of $SO(3,1)$, consisting of the elements preserving the positive ($y>0$) half of the light-cone.

\begin{remark}
Note that while $\conf(\CC)$ is an infinite-dimensional Lie algebra,
passing to the one-point compactification $\CC\ra \CP^1=\CC\cup\{\infty\}$ reduces this algebra to a finite-dimensional one (\ref{l8 conf(CP^1)}). In fact, $\CC$ does not have a conformal compactification (see Definition \ref{l6 def conformal compactification}), unlike $\RR^{p,q}$ with $p+q>2$.
\end{remark}

\subsection{The group of conformal automorphisms of a simply-connected domain in $\CC$}
\begin{lemma}
\begin{enumerate}[1.]
\item The group of conformal automorphisms of the upper half-plane  $\mathbb{H}=\{z\in \CC \;|\; \mr{Im}(z)>0 \}$ is
\begin{equation}
\Conf(\mathbb{H})=PSL_2(\RR)
\end{equation}
where the elements of $PSL_2(\RR)$ are acting by M\"obius transformations (\ref{l6 Mobius transf}) with $a,b,c,d\in \RR$.
\item The group of conformal automorphisms of the unit disk $D=\{z\in\CC\;|\; |z|<1\}$ is
\begin{equation}
\Conf(D) = PSU(1,1)
\end{equation}
-- the group of M\"obius transformations of the form
\begin{equation}
z\mapsto e^{i\phi} \frac{z-a}{\bar{a}z-1}
\end{equation}
where $\phi\in \RR/2\pi\ZZ$, $a\in \CC$ with $|a|$ are parameters.
\end{enumerate}
\end{lemma}
This is proven straightforwardly, by finding the part of the $PSL_2(\CC)$ which preserves the boundary of the domain (the real line or the unit circle) and does not swap the domain with its complement in $\CP^1$.

\begin{remark}
The groups $PSL_2(\RR)$ and $PSU(1,1)$ are conjugate subgroups  $PSL_2(\CC)$, with conjugating element corresponding to the map $z\mapsto \frac{z-i}{z+i}$ -- a conformal diffeomorphism $\mathbb{H}\ra D$.
\end{remark}

Recall the key result of complex analysis:
\begin{thm}[Riemann mapping theorem]
For any simply-connected open set $U\subset \CC$, there exists a %orientation-preser conformal diffeomorphism 
biholomorphic map
$\phi\colon U\ra D$ with $D$ the open unit disk.
\end{thm}

\begin{corollary}
For any simply-connected open set $U$, the group of conformal automorphisms is 
\begin{equation}
\Conf(U)=\phi^* PSL_2(\RR)
\end{equation}
where $\phi\colon U\ra D$ is the map from the Riemann mapping theorem. %for $U'$ the unit disk.
\end{corollary}

\subsection{Vector fields on $S^1$ vs. Witt algebra}
A real vector field tangent to the unit circle $S^1\subset \CC$ can be written as
\begin{equation}
v=f(\theta)\dd_\theta = \sum_{n\in\ZZ}a_n e^{in\theta} \dd_\theta
\end{equation}
with the Fourier coefficients $a_n$ satisfying the reality condition
\begin{equation}
a_{-n}=\bar{a}_n.
\end{equation}
Here $\theta\in \RR/2\pi \ZZ$ is the angle coordinate on $S^1$. We denote the Lie algebra of such vector fields $\mathfrak{X}(S^1)$.

One can express the basis tangent vector fields on $S^1$ in terms of Witt generators restricted to $S^1$:\footnote{
A related point: consider the inversion map $\mathbb{I}\colon \CC^*\ra \CC^*$, mapping  $z\mapsto \frac{1}{\bar{z}}$. The pushforward of $l_n$ by the inversion is $\mathbb{I}_*l_n=-\bar{l}_{-n}$. Vector fields tangent to $S^1$ appearing in the r.h.s. of (\ref{l8 tang vf to S^1 via l_n}) are invariant under $\mathbb{I}_*$.
}
\begin{equation}\label{l8 tang vf to S^1 via l_n}
e^{in\theta}\dd_\theta=-i(l_n-\bar{l}_{-n})\Big|_{S^1}
\end{equation}
Likewise, one has a  basis of normal vector fields to $S^1$:
\begin{equation}
e^{in\theta}\dd_r=-(l_n+\bar{l}_{-n}) \Big|_{S^1}
\end{equation}

We have a map%\marginpar{Re vs 2Re?}
\begin{equation}\label{l8 W to vector fields on S^1}
\begin{array}{ccc}
\W&\ra &\Gamma(S^1,T\CC|_{S^1}) \\
\displaystyle\sum_{n=-\infty}^\infty c_nl_n &\mapsto & \displaystyle2\mr{Re}\sum_{n}c_n l_n\Big|_{S^1}
\end{array}
\end{equation}
In fact, it is an isomorphism, under appropriate decay assumptions on $c_n$. The r.h.s. of (\ref{l8 W to vector fields on S^1}) consist of vector fields on $S^1$ that are allowed to have both tangent and normal component. The part of $\W$ that maps to vector fields \emph{tangent} to $S^1$ is the real Lie subalgebra
\begin{equation}\label{l8 vect(S^1) real slice in W}
\underbrace{\{\sum_n c_n l_n \;|\; c_{-n}=-\bar{c}_n\}}_{\simeq \mathfrak{X}(S^1)} \subset \W
\end{equation}

Thus, one has the following. %can say that 
\begin{lemma}
The Witt algebra
$\W$ (with decay conditions on coefficients as above) is a complexification of $\mathfrak{X}(S^1)$.
\end{lemma}

One might ask: which vector fields on $S^1$ extend into the unit disk $D$ (cobounding $S^1$) as conformal vector fields? The answer depends drastically on whether the vector fields are required to be tangent to $S^1$ or are allowed to have a normal component on $S^1$.
\begin{lemma}
\begin{enumerate}[(i)]
\item The subalgebra %$\{v\in \mathfrak{X}(S^1)\;|\; \}$
 of $\mathfrak{X}(S^1)$ given by vector fields extending as conformal vector fields into the unit disk $D$ is 
\begin{equation}
\{\mr{Re}\sum_{n=-1}^1 c_n l_n\;|\; c_n=-\bar{c}_n\}\simeq \mathfrak{sl}_2(\RR)
\end{equation}
\item The subalgebra of $\Gamma(S^1,T\CC|_{S^1})$ given by vector fields on $S^1$ (with normal component allowed) extending as conformal vector fields into the unit disk $D$ is
\begin{equation}
\{\mr{Re}\sum_{n\geq -1} c_n l_n\;|\; c_n=-\bar{c}_n\}
\end{equation}
\end{enumerate}
\end{lemma}
In particular, we have a finite-dimensional Lie algebra in one case and an infinite-dimensional one in the other case.
\begin{proof}
Immediate consequence of (\ref{l8 W to vector fields on S^1}), (\ref{l8 vect(S^1) real slice in W}) and the fact that $l_n$ is regular at $0$ iff $n\geq -1$.
\end{proof}

\section{Conformal symmetry of $\RR^{1}$ (trivial case)}

Recall from Example \ref{l6 ex: Conf(R^1)} that on $\RR^1$ any diffeomorphism is conformal, $\mr{Conf}(\RR^1)=\mr{Diff}(\RR^1)$. Likewise, any vector field on $\RR^1$ is conformal, $\conf(\RR^1)=\mathfrak{X}(\RR^1)$. 

Also, one can replace $\RR^1$ with $S^1$ (thought of as a one-point compactification of $\RR^1$). % $\Conf(S^1)=\mr{Diff}(S^1)$.
 Here one has as a distinguished subgroup the  M\"obius transformations of $S^1$:
\begin{equation}\label{l8 Conf(S^1)}
\Conf(S^1)=\mr{Diff}(S^1)\supset \underbrace{PSL_2(\RR) \simeq SO_+(2,1)}_{\mr{``restricted\;conformal\;group''}}
\end{equation}
The action of $SO(2,1)$ on $S^1$ by conformal automorphisms is by the construction of Section \ref{sss: SO action on R^p,q, conf compactification}.

\section{Conformal symmetry of $\RR^{1,1}$}
Consider Minkowski plane $\RR^{1,1}$ with coordinates $x,y$ and  metric $g=(dx)^2-(dy)^2$. Introduce the ``light-cone coordinates''
\begin{equation}
x^+=x+y,\quad x^- = x-y
\end{equation}
(they are Minkowski analogs of the complex coordinates $z,\bar{z}$ in the Euclidean case $\RR^2$). 
\begin{figure}[H]
%$$\vcenter{\hbox{ \includegraphics[scale=0.5]{cob1.eps} }} $$
\begin{center}
\includegraphics[scale=0.7]{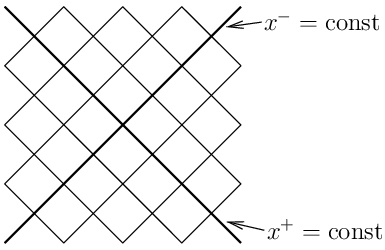}
\end{center}
\caption{Light cone coordinates on $\RR^{1,1}$.}
\end{figure}
In terms of the light-cone coordinates, the metric is: $g=dx^+ dx^-$. Let us write a  vector field on $\RR^{1,1}$ as 
$$ v=v^+(x^+,x^-)\dd_+ + v^-(x^+,x^-)\dd_- $$
with $v^\pm$ some functions on $\RR^{1,1}$; we denoted  $\dd_\pm = \frac12(\dd_x\pm \dd_y)$. The condition that $v$ is conformal (\ref{l6 cvf defining eq}) becomes
\begin{equation}
\dd_- v^+=0,\quad \dd_+ v^-=0,\quad \dd_+v^++\dd_- v^-=\omega
\end{equation}
Thus, a general conformal vector field on $\RR^{1,1}$ is of the form
\begin{equation}
v=v^+(x^+)\dd_+ + v^-(x^-)\dd_-
\end{equation}
Note that coefficient functions now depend on a single light-cone variable; this is an analog of holomorphic/antiholomorphic coefficient functions in the $\RR^2$ case. The conformal factor of  $v$ is:
\begin{equation}
\omega=\dd_+v_++\dd_-v_-
\end{equation}

Thus we have the following.
\begin{lemma} The Lie algebra of conformal vector fields on $\RR^{1,1}$ splits into two copies of the Lie algebra of vector fields on the line:
$$\conf(\RR^{1,1})=\underbrace{\X(\RR^1)}_{v_+\dd_+}\oplus \underbrace{\X(\RR^1)}_{v_-\dd_-}$$
\end{lemma}

One can similarly classify (finite) conformal automorphisms of $\RR^{1,1}$ -- one has the following analog of Lemma \ref{l8 lm 2.21}:
\begin{lemma}
A map $\phi\colon \RR^{1,1}\ra \RR^{1,1}$ with components $\phi^+(x^+,x^-)$, $\phi^-(x^+,x^-)$ is a conformal automorphism of $\RR^{1,1}$ if and only if one of the two following options holds:
\begin{enumerate}[1.]
\item $\phi^+=\phi^+(x^+)$, $\phi^-=\phi^-(x^-)$. \\I.e., $\phi\in \mr{Diff}(\RR)\times \mr{Diff}(\RR)$ -- a reparametrization of $x^+$ and of $x^-$. The conformal factor in this case is $\Omega=(\dd_+\phi^+)(\dd_-\phi^-)$.
\item $\phi^+=\phi^+(x^-)$, $\phi^-=\phi^-(x^+)$. \\I.e., $\phi$ is a composition of a reparametrization of $x^+$ and $x^-$ with a reflection $(x,y)\mapsto (x,-y)$. The conformal factor in this case is $\Omega=(\dd_-\phi^+)(\dd_+\phi^-)$.
\end{enumerate}
\end{lemma}

In particular, we have
\begin{equation}
\Conf_0(\RR^{1,1}) = \mr{Diff}_+(\RR)\times \mr{Diff}_+(\RR)
\end{equation}
Subscript in $\mr{Diff}_+$ stands for orientation-preserving diffeomorphisms. Note that the whole group $\mr{Conf}(\RR^{1,1})$ has $8=2\times 2\times 2$ connected components: one can choose to preserve or reverse the orientation along $x_+$ and $x_-$ and whether or not to compose with the reflection $x_+\leftrightarrow x_-$.

\begin{remark}
One can consider $\ol{\RR^{1,1}}\colon=S^1\times S^1$ as a (partial) conformal compactification of $\RR^{1,1}$, with respect to a (large) subalgebra of $\conf(\RR^{1,1})$ consisting of pairs of vector fields on $\RR$ which extend to  $S^1=\RR\cup\{\infty\}$. Then, in analogy with (\ref{l8 Conf(S^1)}), one has
\begin{equation}
\Conf_0(\ol{\RR^{1,1}})=\mr{Diff}_+(S^1)\times \mr{Diff}_+(S^1)\supset \underbrace{PSL_2(\RR)}_{\mr{M\ddot{o}bius}_+} \times \underbrace{PSL_2(\RR)}_{\mr{M\ddot{o}bius}_-} \simeq \underbrace{SO(2,2)}_{\mr{restricted\;conformal\; group}}
\end{equation}
\end{remark}

\marginpar{Lecture 9,\\
9/12/2022}

\section{Moduli space of conformal structures}
\begin{definition}\label{l9 def: conf flatness}
A (pseudo-)Riemannian manifold $(M,g)$ with metric of signature $(p,q)$ is said to be \emph{conformally flat} if 
%in a neighborhood of any point one can find a coordinate chart on $M$ in which the metric has the form
one can find an atlas of coordinate neighborhoods $U_\alpha\subset M$ with local coordinates $\{x^i_\alpha\}$,
%coordinate charts $\{U_\alpha\subset M,(x^i_\alpha)\}_\alpha$ on $M$ 
such that in each chart the metric has the form
\begin{equation}\label{l9 g=Omega g_stand}
g|_{U_\alpha}=\Omega_\alpha(x)\cdot ((dx^1_\alpha)^2+\cdots + (dx^p_\alpha)^2- (dx^{p+1}_\alpha)^2-\cdots - (dx^{p+q}_\alpha)^2)
\end{equation}
with some positive functions $\Omega_\alpha$. %(generally, different functions in different charts). 
Coordinate charts in which the metric satisfies the ansatz (\ref{l9 g=Omega g_stand}) are called ``isothermal coordinates'' on $(M,g)$.
\end{definition}

Note that being conformally flat is a local property.

The situation with conformal flatness of manifolds depends on the dimension.
\begin{itemize}
\item If $\dim M=1$ any Riemannian manifold admits local coordinates in which $g=(dx)^2$. I.e. any 1-dimensional Riemannian manifold is flat and, a fortiori, is conformally flat.
\item If $\dim M=2$ (case of main interest for us), any (pseudo-)Riemannian manifold is conformally flat.\footnote{This is not a trivial fact. It can be proven from existence of a solution of the Beltrami equation for the change of coordinates 
from generic starting coordinates
 to isothermal coordinates. Originally this statement was proven by Gauss.
} %\marginpar{Explain why?}
\item If $\dim M=3$ a (pseudo-)Riemannian manifold is conformally flat if and only if its Cotton tensor vanishes at every point -- this is a certain tensor $C\in \Omega^2(M,TM)$ constructed in terms of derivatives of the Ricci tensor of the metric.
\item If $\dim M\geq 4$, a (pseudo-)Riemannian manifold is conformally flat if and only if the Weyl curvature tensor vanishes at every point -- this is a certain tensor $W\in \Omega^2(M,\wedge^2 T^*M)$ expressed in terms of the Riemann curvature tensor of $g$.
\end{itemize}

In particular, (pseudo-)Rimeannian manifolds of dimension $\geq 2$ are conformally flat, while in dimension $\geq 3$ there are local obstructions for conformal flatness.

Given a smooth manifold $M$, one has an action of the Lie group of diffeomorphisms of $M$ on the space of conformal structures:
\begin{equation}\label{l9 Diff acts on conf str}
\mr{Diff}(M)\;\;\calt \;\;\{\mr{conformal\;structures\;on\;}M\}
\end{equation}
\begin{definition}\label{l9 def: mod space of conf structures}
We call the orbit space $\MM_M$ of the action (\ref{l9 Diff acts on conf str}) the \emph{moduli space of conformal structures}.
\end{definition}

Note that the action (\ref{l9 Diff acts on conf str}) is not free: for $\xi$ a conformal structure on $M$ there can be a nontrivial stabilizer subgroup
\begin{equation}
\mr{Stab}_\xi =\{\phi\colon M\ra M\;|\; \phi^* \xi=\xi\}=\Conf(M,\xi)\quad \subset \mr{Diff}(M)
\end{equation}
-- the group of conformal automorphisms of $(M,\xi)$. Also, if $\psi\colon M\ra M$ is a diffeomorphism, then $\mr{Stab}_{\xi}$ and $\mr{Stab}_{\psi^*\xi}$ are conjugate subgroups of $\mr{Diff}(M)$.

\begin{remark} In which sense $\MM_M$ is a  ``space''? There are several ways to understand this object:
\begin{enumerate}[(i)]
\item As a topological space, with quotient topology.
\item As an  orbifold -- a manifold with ``nice'' singularities (of the local form $\RR^N/\Gamma$, with $\Gamma$ a finite group acting on $\RR^N$ properly).
\item As a ``stack.'' This is the correct way to talk about $\MM_M$, but we will be a bit 
simple-minded about it and just remember a part of the ``stacky data'' -- that points $[\xi]\in\MM_M$ come equipped with stabilizers -- subgroups $\mr{Stab}_\xi\subset \mr{Diff}(M)$.
\end{enumerate}
\end{remark}

\begin{remark}
The discussion below Definition \ref{l9 def: conf flatness} suggests that the moduli space of conformal structures on a manifold of dimension $\geq 3$ is infinite-dimensional, due to the presence of local moduli (Cotton and Weyl tensors). In dimension $2$, there are no local moduli: all metric are locally conformally equivalent to the standard flat metric, and only global moduli remain. So, one would expect the $\MM_M$ to be ``small'' (finite-dimensional) in this case. This indeed turns out to be the case, as we discuss below.
\end{remark}

\subsection{Reminder: almost complex structures and complex structures}
For details on complex and almost complex manifolds we refer the reader e.g. to \cite[Section 15]{CdS}.

\begin{definition}
An \emph{almost complex structure} on a smooth manifold $M$ is smooth family over $M$ of endomorphisms of (real) tangent spaces that square to $-\mr{id}$:
\begin{equation}
J\in \Gamma(M,\mr{End}(TM)),\quad \mr{s.t.} \;J_x^2=-\mr{id}\quad \mr{for\;all}\; x\in M.
\end{equation}
\end{definition}
Consider the matrix of $J_x$ with respect to some basis in $T_xM$. Note that the eigenvalues of a real matrix with square $-\mr{id}$ must be $+i$ and $-i$, moreover $+i$ and $-i$ must have the same multiplicity. In particular, if $M$ has an almost complex structure, $\dim M=2m$ must be even.

Also note that an almost complex structure induces an orientation on $M$: for $(v_1,\ldots,v_m)$ an $m$-tuple of generic vectors in $T_xM$, we say that the $(2m)$-tuple $(v_1,Jv_1,v_2,Jv_2,\ldots, v_m, Jv_m)$ is positively oriented in $T_xM$ (it is a straightforward check that this orientation is independent of the choice of the initial $m$-tuple).

Given an almost complex structure, we have a splitting of the complexified tangent bundle into ``holomorphic'' and ``antiholomorphic'' parts:
\begin{equation}\label{l9 T_C M splitting}
\underbrace{T_{\CC}M}_{\CC\otimes TM} = T^{1,0}M\oplus T^{0,1}M.
\end{equation}
On the right, for each $x\in M$, the complex vector spaces $T^{1,0}_x M$, $T^{0,1}_x M$ are defined as $+i$- and $-i$-eigenspaces of $J_x$, respectively. The splitting (\ref{l9 T_C M splitting}) induces a dual splitting of the complexified cotangent bundle
\begin{equation}\label{l9 T^*_C M splitting}
T^*_\CC M = \underbrace{(T^{1,0})^*M}_K\oplus \underbrace{(T^{0,1})^*M}_{\bar{K}}
\end{equation}
we will denote the holomorphic/antiholomorphic cotangent bundles on the right by $K$, $\bar{K}$. Furthermore, the splitting (\ref{l9 T^*_C M splitting}) of $k$-forms on $M$ (with complex coefficients) as
\begin{equation}
\Omega^k_\CC(M)=\bigoplus_{p\geq 0,q\geq 0, p+q=k} \underbrace{\Omega^{p,q}(M)}_{\Gamma(M,\wedge^p K\otimes \wedge^q \bar{K})}
\end{equation}
We refer to elements of $\Omega^{p,q}$ as $(p,q)$-forms on $M$.

Note that if the (real) dimension of the manifold $M$ is $2m$, then $T^{1,0}_xM$, $T^{0,1}_xM$ have complex dimension $m$ -- then we say that $M$ has \emph{complex dimension} 
$$\dim_\CC M=m=\frac12 \dim M.$$
 In particular, one has $\Omega^{p,q}(M)=0$ if either $p>m$ or $q>m$.

Consider the differential operators
$\dd,\bar\dd\colon \Omega^\bt(M)\ra \Omega^{\bt+1}(M)$ 
defined by
\begin{equation}
\dd\alpha\colon= \pi_{p+1,q}(d\alpha),\;\bar\dd\alpha \colon= \pi_{p,q+1} (d\alpha)
\end{equation}
for $\alpha\in \Omega^{p,q}$. Here $d$ is de Rham operator and $\pi_{p,q}$ is the projection of $\Omega(M)$ onto its component $\Omega^{p,q}(M)$.
One calls $\dd,\bar{\dd}$ the holomorphic/antiholomorphic Dolbeault operators. By default, just ``Dolbeault operator'' is $\bar\dd$.

\begin{definition}\label{l9 def integrable acs}
An almost complex structure $J$ on a manifold $M$ is \emph{integrable} if one can find an atlas of complex coordinates $(z^j_\alpha,\bar{z}^{\bar{j}}_\alpha)$ on coordinate neighborhoods $U_\alpha$ such that
\begin{itemize}
\item $J\dd_{z^j}=i\dd_{z^j},\quad J\dd_{\bar{z}^{\bar{j}}}=-i\dd_{\bar{z}^{\bar{j}}}$,
\item the transition functions between charts are holomorphic: $\frac{\dd z_\beta^j}{\dd \bar{z}_\alpha^{\bar{j}}}=0$, $\frac{\dd \bar{z}_\beta^{\bar{j}}}{\dd z_\alpha^j}=0$ for any $j,\bar{j}$ and any two overlapping neighborhoods $U_\alpha$, $U_\beta$ from the atlas.
\end{itemize}
An integrable almost complex structure $J$ is called a complex structure (not ``almost''). A manifold $M$ with a complex structure $J$ is called a complex manifold.
\end{definition}

Equivalent characterizations of integrability of $J$ are:
\begin{enumerate}[(i)]
\item An almost complex structure $J$ is integrable if and only if its Nijenhuis tensor $N_J\in \Omega^2(M,TM)$
vanishes: 
\begin{equation} \label{l9 N_J}
N_J(X,Y)\colon= -J^2[X,Y]+J[JX,Y]+J[X,JY]-[JX,JY]  =0
\end{equation}
for $X,Y\in\X(M)$. 
An equivalent restatement of (\ref{l9 N_J}) is: for $X^{1,0},Y^{1,0}\in \Gamma(M,T^{1,0}M)$ two sections of the holomorphic tangent bundle, 
their Lie bracket is also a section of the holomorphic tangent bundle (the antiholomorphic component vanishes):
\begin{equation}\label{l9 involutivity}
[X^{1,0},Y^{1,0}]^{0,1} =0.
\end{equation}
%where the superscript in $[,]^{0,1}$ stands for taking the $T^{0,1}$-component of the Lie bracket.
\item An almost complex structure $J$ is integrable if and only if one has 
\begin{equation}\label{l9 dbar^2=0}
\bar\dd^2=0.
\end{equation}
%Equivalently $\dd^2=0$ and equivalently $[\dd,\bar\dd]=0$.
\item The de Rham operator splits as\footnote{For a non-integrable almost complex structure, $d$ additionally has components of bi-degree $(2,-1)$ and $(-1,2)$ w.r.t. the $(p,q)$-grading on forms.}
\begin{equation}
d=\dd+\bar\dd.
\end{equation}
\end{enumerate}

Equivalence of Definition \ref{l9 def integrable acs} with the characterizations above is known as the Newlander-Nirenberg theorem.\footnote{In particular, the equivalence of Definition \ref{l9 def integrable acs} and vanishing property of the Nijenhuis tensor (\ref{l9 N_J}) can be viewed as a complex analog of Frobenius theorem. Recall that Frobenius theorem says that a tangential distribution is involutive if and only if it integrates locally to a foliation. In the case of Newlander-Nirenberg theorem, the distribution in question is complex, $T^{1,0}M\subset T_\CC M$. In this analogy, a foliation corresponds to local complex coordinates and involutivity is the property (\ref{l9 involutivity}).}

On a complex manifold $(M,J)$, the Dolbeault operators written locally in terms of complex coordinates are
\begin{equation}
\dd=\sum_j dz^j \frac{\dd}{\dd z^j},\quad \bar\dd=\sum_{\bar{j}} d\bar{z}^{\bar{j}}\frac{\dd}{\dd \bar{z}^{\bar{j}}}
\end{equation}

\begin{lemma}\label{l9 lemma: 2d acs is integrable}
  Any almost complex structure $J$ on a manifold $M$ of dimension $\dim M=2$ is integrable.
\end{lemma}
\begin{proof}
This follows e.g. from (\ref{l9 dbar^2=0}): $\bar\dd^2$ maps $(p,q)$-forms to $(p,q+2)$-forms. But there are no forms of degree $(*,\geq 2)$ on a $2$-manifold. 
\end{proof}

\subsection{2d conformal structures (of Riemannian signature) = complex structures}
\label{ss 2d conf str = cx str}
We will reserve the letter $\Sigma$ for 2-dimensional surfaces, while manifolds of general dimension we denote by $M$. 
\begin{lemma} \label{l9 lemma 2d conf str = cx str}
Fix an oriented 2-dimensional surface $\Sigma$.
One has a natural bijection between the following two sets:
\begin{enumerate}[(i)]
\item the set of conformal structures on $\Sigma$ of signature $(2,0)$ (i.e. \emph{Riemannian} metrics modulo Weyl transformations),
\item the set of complex structures $J$ on $\Sigma$, compatible with orientation.
\end{enumerate}
%Moreover, two conformal structures on $\Sigma$ are equivalent (related by a conformal diffeomorphism)
\end{lemma}

\begin{proof}
Given a conformal structure $\xi=g/\sim$ on $\Sigma$, we assign to it the complex structure $J\colon T_x\Sigma\ra T_x\Sigma$ which maps a tangent vector $u\in T_x \Sigma$ to the vector $v\in T_x \Sigma$ uniquely characterized by the following properties:
\begin{itemize}
\item $v$ is orthogonal to $u$ (according to any metric $g$ representing $\xi$),
\item $v$ and $u$ have the same length (according to any metric $g$ representing $\xi$),
\item $(u,v)$ is a positively oriented pair in $T_x \Sigma$.
\end{itemize}
%\textcolor{red}{PICTURE}
\begin{figure}[H]
%$$\vcenter{\hbox{ \includegraphics[scale=0.5]{cob1.eps} }} $$
\begin{center}
\includegraphics[scale=0.7]{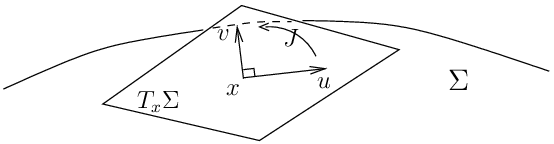}
\end{center}
\caption{Complex structure on a surface.}
\end{figure}

Here is the inverse construction. Given a complex structure $J$ on $\Sigma$, we assign to it a conformal structure $\xi$ on $\Sigma$, defined as follows:
%\begin{itemize}
%\item 
Choose some volume form $\sigma\in \Omega^2(\Sigma)$ compatible with the orientation.
%\item 
Set $g_x(u,v)\colon= \sigma (u,Jv)$.
It is a straightforward check that $g_x$ is positive symmetric bilinear form on $T_x\Sigma$, i.e., a metric. The conformal class of $g$ does not depend on a choice of the volume form $\sigma$ (changing $\sigma\ra \Omega\sigma$ with $\Omega$ a positive function, induces a change of $g$ by a Weyl transformation). This construction $J\ra \xi$ inverts the construction $\xi\ra J$ above.
%\end{itemize}
\end{proof}

\begin{remark} \label{l9 rem: 2d conf equivalences = cx equivalences}
Under the correspondence between conformal and complex structure of Lemma \ref{l9 lemma 2d conf str = cx str}, equivalences of conformal and complex surfaces also go into one another: $\phi\colon (\Sigma,\xi)\ra (\Sigma',\xi')$ is a conformal diffeomorphism of surfaces equipped with conformal structures if and only if $\phi$ is a biholomorphic map of the corresponding complex surfaces $\phi\colon (\Sigma,J)\ra (\Sigma',J')$. 

In particular, the correspondence of Lemma \ref{l9 lemma 2d conf str = cx str} gives an equivalence of categories, between 
\begin{enumerate}[(a)]
\item the category of surfaces equipped with conformal structure, with morphisms being conformal diffeomorphisms on one side and
\item the category of complex surfaces and biholomorphic maps on the other side.
\end{enumerate}
%For a 2-dimensional oriented surface, equivalences 
\end{remark}

\begin{remark}
As a consequence of Lemma \ref{l9 lemma 2d conf str = cx str}, in the case of 2d surfaces, the moduli space of conformal structures (Definition \ref{l9 def: mod space of conf structures}) and the moduli space of complex structures (\ref{l10 mod space of cx str on M}) are the same.
\end{remark}

\begin{definition} A smooth manifold $\Sigma$ of dimension $2$ equipped with a complex structure is called a Riemann surface.
Equivalently, a Riemann surface is a smooth 2-manifold equipped with orientation and conformal structure.\footnote{Note that a Riemann surface is not a Riemannian manifold: it does not come with a preferred metric.} 
\end{definition}
%We typically denote surfaces by $\Sigma$ while manifolds of general dimension we denote by $M$. \marginpar{Put this someplace earlier (in this section?); maybe as a footnote.}

\begin{definition}%\marginpar{Can move this definition to some later point, when stability becomes an issue (Teichm\"uller theory, uniformization,...)}
We will call a Riemann surface \emph{stable} if it does not admit nonzero conformal vector fields. In the case of a Riemann surface with marked points $p_1,\ldots,p_n$, we call it stable if there are no nonzero conformal vector fields which vanish at the points $p_i$.
\end{definition}

\subsection{Deformations of a complex structure. Parametrization of deformations by Beltrami differentials}
\label{sss Beltrami differentials}
Let $(M,J)$ be a complex manifold. A deformation of a complex structure in the class of almost complex structures can be described as a change of the Dolbeault operator $\bar\dd$:
\begin{equation}\label{l9 dbar deformation by mu}
\bar\dd \ra \underbrace{\bar\dd-\mu}_{\bar\dd_\mu} %,\quad \dd\ra \dd-\bar\mu
\end{equation}
where the parameter of the deformation
\begin{equation}
\mu\in \Omega^{0,1}(M,T^{1,0}M)
\end{equation}
is called the \emph{Beltrami differential}; $\bar\mu \in \Omega^{1,0}(M,T^{0,1}M)$ is the complex conjugate object. In local complex coordinates, $\mu$ has the form
%\marginpar{using $z$,$\bar{z}$ notation here for a smooth function of $z$ -- uniformize with Sec. 4.3?}
\begin{equation}\label{l9 mu in coords}
\mu=\mu^j_{\bar{i}}(z)d\bar{z}^{\bar{i}}\frac{\dd}{\dd{z^j}}
\end{equation}
where the coefficient functions $\mu^j_{\bar{i}}(z)$ are arbitrary smooth complex-valued functions on $M$. In (\ref{l9 dbar deformation by mu}), we understand $\mu$ as a first-order differential operator $\Omega^{p,q}\ra \Omega^{p,q+1}$. The deformed Dolbeault operator written locally thus has the form
\begin{equation}
\bar\dd_\mu= d\bar{z}^{\bar{i}}\left(\frac{\dd}{\dd \bar{z}^{\bar{i}}}-\mu^j_{\bar{i}}(z)\frac{\dd}{\dd z^j}\right)
\end{equation}

The deformation (\ref{l9 dbar deformation by mu}) is accompanied by the deformation of the holomorphic Dolbeault operator 
\begin{equation}
\dd\ra \dd-\bar\mu
\end{equation}
where $\bar\mu$ is the complex conjugate of the Beltrami differential $\mu$.

Expressed as a deformation of $J$, 
%in the first-order in the deformation parameter, 
(\ref{l9 dbar deformation by mu}) corresponds to the change
%the following change of the almost complex structure $J$:
%\marginpar{EDIT!! this f-la is wrong!}
%improve this: write a finite def. f-la
\begin{equation}
J_x\ra %(1-\mu\bar\mu)^{-\frac12}
J_x+2i(\mu_x-\bar\mu_x)
\end{equation}
for any $x\in M$ (in the first order in $\mu,\bar\mu$).

In order for the deformation (\ref{l9 dbar deformation by mu}) to be a complex structure (rather than almost complex), it must satisfy the integrability condition 
\begin{equation}\label{l9 KS}
(\bar\dd_\mu)^2=0 \quad \Leftrightarrow \quad \bar\dd\mu-\frac12 [\mu,\mu]=0
\end{equation}
The equation on the right is called the Kodaira-Spencer equation.

\begin{remark}
In other words, deformations of a complex structure on a given complex manifold are governed by Maurer-Cartan elements of the differential graded Lie algebra %$\mc{D}^\bt$ where
\begin{equation}\label{l9 deformation complex}
\Omega^{0,*}(M,T^{1,0}M),\;\;\bar\dd,\;\;[,]
\end{equation}
of $(0,q)$-forms with coefficients in the holomorphic tangent bundle,
with differential $\bar\dd$ and Lie bracket $[,]$ coming as the wedge product of forms tensored with the Lie bracket of $(1,0)$-vector fields.\footnote{We should mention that there is a natural and very deep %interesting 
generalization of deformations of complex structures due to Barannikov-Kontsevich \cite{Barannikov-Kontsevich}. Here one replaces the dg Lie algebra (\ref{l9 deformation complex}) by a bigger one: $\Omega^{0,p}(M,\wedge^qT^{1,0}M)$, with total grading by $p+q-1$, and considers Maurer-Cartan elements there.}
\end{remark}

We emphasize that the formula (\ref{l9 dbar deformation by mu}), with $\mu$ satisfying the Kodaira-Spencer equation (\ref{l9 KS}), describes \emph{finite} deformations of a complex structure, not just infinitesimal (first-order) deformations.

We also remark that if $\dim M=2$, then the Kodaira-Spencer equation (\ref{l9 KS}) holds trivially (as there are no $(0,2)$-forms on $M$), cf. Lemma \ref{l9 lemma: 2d acs is integrable}.

\begin{remark} For $\dim M=2$, a metric on a surface compatible with the complex structure deformed by a Beltrami differential $\mu$ locally has the form
\begin{equation}\label{l9 g via Beltrami}
g=\rho(z)^2 |dz + \mu^z_{\bar{z}}(z) d\bar{z}|^2
\end{equation}
with $\rho$ some positive function %(local Weyl factor) 
and $z,\bar{z}$ the non-deformed local complex coordinates (associated to the reference complex structure). For the metric (\ref{l9 g via Beltrami}) to be nondegenerate one needs the Beltrami differential to be sufficiently small: $|\mu^z_{\bar{z}}(z)|<1$ everywhere on $M$.

The deformed complex coordinate $z'$ is a solution of the Beltrami equation:
\begin{equation}
(\bar\dd - \mu) z'=0.
\end{equation}
\end{remark}

\marginpar{Lecture 10,\\ 9/14/2022}

%Restricting to the first-order deformations of a given complex structure $J$ (in the class of complex structures), 
\subsubsection{Tangent space to the space of complex structures.}
The discussion above implies that the tangent space to the space of complex structures on a manifold $M$ at a complex structure $J$ is the space of $\bar\dd$-closed Beltrami differentials (with $\bar\dd$-closed condition being the first-order approximation of the Kodaira-Spencer equation (\ref{l9 KS})):
\begin{equation}\label{l10 T_J (cx str)}
T_J\Big(\mbox{space of complex structures on }M\Big) \simeq \Omega^{0,1}_{\bar\dd\mr{-closed}}(M,T^{1,0}M).
\end{equation}
For the moduli space of of complex structures,\footnote{
In this subsection we use $\MM_M$ for the moduli space of complex (not conformal) structures on $M$. Later, when we specialize to surfaces, there will be no difference between moduli of complex and conformal structures, due to Lemma \ref{l9 lemma 2d conf str = cx str}.
}
\begin{equation}\label{l10 mod space of cx str on M}
\MM_M=\{\mbox{complex structures on }M\}/\mr{Diff}(M),
\end{equation} 
the tangent space at the class of $J$ is given by the quotient of (\ref{l10 T_J (cx str)}) modulo the action of (infinitesimal) diffeomorphisms on Beltrami differentials, 
\begin{equation}
\mu\sim \mu+\bar\dd v^{1,0}\label{l10 mu -> mu + dbar v}
\end{equation} 
with $v^{1,0}$ the projection to $T^{1,0}$ of any vector field on $M$. I.e., one has
\begin{equation}
T_J \MM_M = H^{0,1}(M,T^{1,0} M)
\end{equation}
-- the cohomology of the complex (\ref{l9 deformation complex}) in degree one.

\subsubsection{Cotangent space to the space of complex structures (case of surfaces).}
In the case of a 2-dimensional surface, the $\bar\dd$-closed condition in (\ref{l10 T_J (cx str)}) is automatic. In this case, one can describe the \emph{cotangent} space to the space of complex structures as % the dual space of (\ref{l10 T_J (cx str)}):
\begin{equation}\label{l10 T^*_J (cx str)}
T^*_J\Big(\mbox{space of complex structures on }\Sigma\Big) = \Omega^{1,0}(\Sigma,K %(T^{1,0})^*\Sigma
) \simeq \Gamma(\Sigma, K^{\otimes 2})
%= \\= \{\mr{quadratic\;differentials\;on\;}\Sigma\}
\end{equation}
where $K=(T^{1,0})^*\Sigma$ is the holomorphic cotangent bundle. Elements of (\ref{l10 T^*_J (cx str)}) are quadratic differentials $\tau$ on $\Sigma$ -- tensors written in a local complex coordinate chart as $\tau=f(z)(dz)^2$. The pairing between an element $\mu$ of (\ref{l10 T_J (cx str)}) (a Beltrami differential) and an element $\tau$ of (\ref{l10 T^*_J (cx str)}) is
\begin{equation}
\int_\Sigma \langle \mu,\tau \rangle
\end{equation}
where $\langle,\rangle$ is a pairing between vectors $T^{1,0}_x\Sigma$ and covectors $(T^{1,0}_x\Sigma)^*$; thus, $\langle \mu,\tau\rangle$ is a $(1,1)$-form on $\Sigma$, i.e., a $2$-form, which can be integrated. 

For the cotangent space of the moduli space of complex structures $\MM_\Sigma$, (\ref{l10 T^*_J (cx str)}) implies
\begin{equation}\label{l10 T^*J (moduli space)}
T^*_J \MM_\Sigma \simeq \Omega^{1,0}_{\bar\dd\mr{-closed}}(\Sigma,K) = \{\mr{holomorphic\;quadratic \;differentials\;on\;}\Sigma \}
\end{equation}
-- the space of \emph{holomorphic} quadratic differentials, locally of the form $\tau=f(z)(dz)^2$ with a \emph{holomorphic} coefficient function.

The holomorphicity condition in (\ref{l10 T^*J (moduli space)}) arises because we are looking for the elements of (\ref{l10 T^*_J (cx str)}) annihilating all vectors of the form 
$$\bar\dd v^{1,0}\in T_J(\mr{space\; of\; complex\;structures}),$$ 
cf. (\ref{l10 mu -> mu + dbar v}).

\begin{remark}
In 2d conformal field theory, the stress-energy tensor $T$ is a holomorphic quadratic differential, so it can be seen via (\ref{l10 T^*J (moduli space)}) as a cotangent vector to the moduli space of complex structures.
\end{remark}

\subsection{Uniformization theorem}
The following statement is a key result on Riemann surfaces, known as the Uniformization Theorem.
\begin{thm}[%``Uniformization theorem,'' 
Klein-Koebe-Poincar\'e] \label{l10 thm: Uniformization}
Any simply-connected Riemann surface $(\Sigma,\xi)$ is conformally equivalent to exactly one the following three model surfaces:
\begin{enumerate}[(i)]
\item $\CP^1$,
\item $\CC$,
\item Open disk $D=\{z\in \CC\;|\; |z|<1 \}$ (``Poincar\'e disk'') or, equivalently (a conformally equivalent model), upper half-plane $\Pi_+=\{z\in \CC\;|\; \mr{Im}(z)>0\}$.
\end{enumerate}
\end{thm}
\begin{remark} For each of the model surfaces from Theorem \ref{l10 thm: Uniformization}, there is a metric of constant scalar curvature $R=+1,0,-1$ representing its conformal class:
\begin{enumerate}[(i)]
\item $\CP^1$ has a unique metric in its conformal class of scalar curvature $R=+1$ -- the Fubini-Study metric $g=\frac{4dzd\bar{z}}{(1+z\bar{z})^2}$.
\item $\CC$ has a unique up to scaling flat (i.e. $R=0$) metric in its conformal class, $g=C dzd\bar{z}$, for any $C>0$.
\item $D$ has a unique metric of scalar curvature $R=-1$ in its conformal class, $g=\frac{4dz d\bar{z}}{(1-z\bar{z})^2}$. Equivalently, $\Pi_+$ has a unique $R=-1$ metric $g=\frac{dz d\bar{z}}{(\mr{Im}(z))^2}$.
\end{enumerate}
We also remark that for these distinguished metrics, in cases (i) and (iii) the groups of isometries and all conformal automorphisms coincide (put another way, each conformal automorphism is an isometry).
\end{remark}

For a general Riemann surface $\Sigma$ (not necessarily simply-connected), its universal cover $\til\Sigma$ inherits a conformal structure from $\Sigma$, is simply-connected and corresponds to one of the model surfaces from Theorem \ref{l10 thm: Uniformization}. The group of covering transformations acts on $\til\Sigma$ by conformal automorphisms. Thus, any Riemann surface $\Sigma$ is conformally equivalent to a surface of the form
\begin{equation}\label{l10 Sigma^model/Gamma}
\Sigma^\mr{model}/\Gamma
\end{equation} 
where $\Gamma$ is the image of a group homomorphism 
\begin{equation}\label{l10 rho}
\rho\colon \pi_1(\Sigma)\ra \Conf(\Sigma^\mr{model})
\end{equation} In particular, we need $\Gamma$ to be a discrete subgroup of $\Conf(\Sigma^\mr{model})$, acting freely on $\Sigma^\mr{model}$ (so that the quotient (\ref{l10 Sigma^model/Gamma}) is a smooth manifold).

\begin{remark}\label{l10 rem: conjugation}
If we change in (\ref{l10 Sigma^model/Gamma}) the subgroup $\Gamma$ to a conjugate subgroup $\chi \Gamma\chi^{-1}$ 
%change the homomorphism $\rho$ by a conjugate one, $\gamma\rho\gamma^{-1}$, 
with $\chi\in \Conf(\Sigma^\mr{model})$ a fixed element (or, put another way, we change the homomorphism (\ref{l10 rho}) to a conjugate one, $\rho\mapsto \chi\rho\chi^{-1}$), then the quotient (\ref{l10 Sigma^model/Gamma}) changes to a conformally equivalent surface.
\end{remark}

This leads to the following classification of connected Riemann surfaces:
\begin{enumerate}[(i)]
\item $\CP^1$
\item  
\begin{enumerate}[(a)]
\item $\CC$
\item $\CC\backslash\{0\}$ or, equivalently, infinite cylinder $\CC/\ZZ$.
\item 2-torus $\CC/\Lambda$ where $\Lambda=u\ZZ\oplus v \ZZ\in \CC$ is a lattice spanned by vectors $u,v\in \CC$ with $u/v \not\in \RR$. Using rotation and scaling,\footnote{
In this example, $\rho$ maps $\pi_1(S^1\times S^1)$ to a lattice $\Lambda$ seen as a subgroup of $\{\mr{translations}\}\subset \Conf(\CC)$. The change of the generators of $\Lambda$ by translation and scaling corresponds to the conjugation of $\rho$, as in Remark \ref{l10 rem: conjugation}, by rotation and scaling.
} one can convert the pair $(u,v)$ to $(1,\tau)$ with $\tau\in \Pi_+$.
\end{enumerate}
\item $\Pi_+/\Gamma$ for some $\Gamma\subset PSL_2(\RR)$ a ``Fuchsian group'' -- a discrete subgroup of $PSL_2(\RR)$ isomorphic to $\pi_1(\Sigma)$. This case includes all surfaces of genus $g\geq 0$ with $n\geq 0$ boundary circles  (the surfaces are considered as open -- the boundary circles are not a part of $\Sigma$), with $\chi(\Sigma)=2-2g-n<0$, and also includes annulus (or finite cylinder) and punctured disk (or semi-infinite cylinder). %\marginpar{example: closed surface of genus g}
\end{enumerate}

Surfaces of types (i), (ii), (iii) above are called, respectively, elliptic, parabolic and hyperbolic. 
%Each surface admits a unique metric of curvature 
Elliptic surfaces admit (in their conformal class) a unique metric of scalar curvature $+1$, parabolic surfaces -- a unique-up-to-scaling flat metric, hyperbolic surfaces -- a unique metric of scalar curvature $-1$.

\begin{example} A closed Riemann surface of genus $g\geq 2$ falls into the type (iii) (hyperbolic). Using the standard presentation of the fundamental group of a surface as 
$$\pi_1(\Sigma)=\langle \alpha_1,\ldots,\alpha_g,\beta_1,\ldots,\beta_g \;|\; \prod_{i=1}^g  \alpha_i \beta_i\alpha_i^{-1}\beta_i^{-1}=1\rangle$$
we see that its image in $PSL_2(\RR)$ under $\rho$ is a $2g$-tuple of elements 
$$a_1,\ldots,a_g,b_1,\ldots,b_g\quad \in PSL_2(\RR)$$ 
subject to a relation 
$$ \prod_{i=1}^g a_i b_i a_i^{-1} b_i^{-1}=1. $$
\end{example}
Moreover, by Remark \ref{l10 rem: conjugation}, two $2g$-tuples should be considered equivalent if they are related by conjugation by an element $h\in PSL_2(\RR)$:
\begin{equation}
(a_1,\ldots,a_g,b_1,\ldots,b_g)\sim (ha_1 h^{-1},\ldots, ha_g h^{-1}, hb_1 h^{-1},\ldots, hb_g h^{-1}).
\end{equation}

\subsection{Moduli space $\MM_{g,n}$ of %conformal/
complex structures on a surface with $n$ marked points}

\begin{definition}
Fix a smooth closed oriented surface $\Sigma$ of genus $g$. Let $p_1,\ldots,p_n\in \Sigma$ be a collection of pairwise distinct points on $\Sigma$.
The moduli space of complex structures on $\Sigma$ with $n$ marked points is the quotient space\footnote{Again, there are different ways to understand the quotient here: as a topological space with quotient topology (``coarse'' moduli space), as an orbifold, as a stack.}
\begin{equation}\label{l10 M_g,n def}
\MM_{g,n}\colon= \{\mr{complex\;structures\; on\;}\Sigma\}/\mr{Diff}_+(\Sigma,\{p_i\}) ,
\end{equation}
where $\mr{Diff}_+(\Sigma,\{p_i\})$ stands for the orientation-preserving diffeomorphisms of $\Sigma$ that do not move each of the marked points $p_i$.\footnote{
Other names used for $\MM_{g,n}$ include: ``moduli space of conformal structures'' (since in 2d, conformal and complex structures correspond to one another), ``moduli space of Riemann surfaces'' and (in the context of algebraic geometry) ``moduli space of (algebraic) curves.'' 
}
\end{definition}
There is another version of the moduli space where we quotient by orientation-preserving diffeomorphisms which are allowed to move a marked point to another marked point: 
$$
\Diff_+^\mr{unordered}(\Sigma,\{p_i\})\colon=
\{\phi\in \mr{Diff}_+(\Sigma)\;|\; \phi(p_i)=p_{\sigma(i)}\;\mr{for\;some\;}\sigma\in S_n \}.$$
We denote the quotient of the space of complex structures on $\Sigma$ by such diffeomorphisms $\MM^\mr{unordered}_{g,n}$ (unordered marked points), whereas (\ref{l10 M_g,n def}) is the moduli space of complex structures with $n$ \emph{ordered} marked points,  $\MM_{g,n}=: \MM_{g,n}^\mr{ordered}$.\footnote{
One has an action of the symmetric group $S_n$ on $\MM_{g,n}^\mr{ordered}$ by relabeling the marked points. The unordered moduli space is naturally identified with the orbit space of this action:  $\MM_{g,n}^\mr{unordered}=\MM_{g,n}^\mr{ordered}/S_n$.
}

\begin{definition}
%\marginpar{is this def/terminology ok?}
\label{l10 def: univ family}
We call the \emph{universal family} (of Riemann surfaces) the fiber bundle $\mc{E}_{g,n}$ over $\MM_{g,n}$ where the fiber over the point corresponding to a Riemann surface $\Sigma$ with marked points $\{p_i\}$ is that same surface with same marked points.
\end{definition}

The idea of Teichm\"uller theory is to do the quotient (\ref{l10 M_g,n def}) in two steps:
\begin{enumerate}[1.]
\item Take the quotient
\begin{equation}\label{l10 cx str/Diff_0}
\{\mbox{complex structures on }\Sigma\}/\mr{Diff}_0(\Sigma,\{p_i\})=: \mc{T}_{g,n}
\end{equation}
with respect to the \emph{connected component of identity}  in the  group of diffeomorphisms preserving the marked points, $\mr{Diff}_0\subset \Diff_+$. The quotient (\ref{l10 cx str/Diff_0}) is called the Teichm\"uller space $\mc{T}_{g,n}$.\footnote{The points of $\mc{T}_{g,n}$ correspond to equivalence classes of complex structures on $\Sigma$ (modulo diffeomorphisms fixing the marked points), equipped with a ``marking'' --  a diffeomorphism $\phi\colon \Sigma_{g,n}^\mr{stand}\ra \Sigma$ from a ``standard'' surface to $\Sigma$ (taking marked points to marked points), where $\phi$ is considered up to isotopy.}
In the case $\chi=2-2g-n<0$ (the ``stable'' case), the Teichm\"uller space is diffeomorphic to $\RR^{6g-6+2n}$. It carries a natural complex structure and several natural metrics, see e.g. \cite{Penner}.%\marginpar{reference OK?}
\item Take the quotient of (\ref{l10 cx str/Diff_0}) by the discrete group of connected components of the diffeomorphism group appearing in (\ref{l10 M_g,n def}),
\begin{equation}
\pi_0 \Diff(\Sigma,\{p_i\})=: \pMCG_{g,n} .
\end{equation}
This group is known as the ``pure mapping class group'' of a surface of genus $g$ with $n$ marked points. One has a natural action of $\pMCG_{g,n}$ on the Teichm\"uller space inherited from the action of diffeomorphisms on complex structures. Thus, we consider the quotient
\begin{equation}
\MM_{g,n}=\mc{T}_{g,n}/\pMCG_{g,n} .
\end{equation}
\end{enumerate}
%\marginpar{Remark: orbifold singularities of $\MM_{g,n}$}

\begin{remark}
If one wants to construct the moduli space with unordered punctures, one extra step is needed: a quotient by the symmetric group $S_n$ (which acts by permuting the marked points):
\begin{equation}
\MM^\mr{unordered}_{g,n}=\MM_{g,n}/S_n .
\end{equation}
Another way to write it is directly as a quotient of the Teichm\"uller space
\begin{equation}
\MM^\mr{unordered}_{g,n}=\mc{T}_{g,n}/\MCG_{g,n}
\end{equation}
by the full (not ``pure'') mapping class group 
\begin{equation}
\MCG_{g,n}\colon= \pi_0\Diff_+^\mr{unordered}(\Sigma,\{p_i\}).
\end{equation}
\end{remark}

\begin{remark}%\marginpar{factcheck}
The action of the mapping class group %$\MCG$ and $\pMCG$ 
on the Teichm\"uller space $\mc{T}_{g,n}$ is free almost everywhere, except for a discrete set of points where it has a discrete (in fact, finite, for $g,n$ sufficiently large) stabilizer. These points correspond to orbifold singularities of the quotient $\MM_{g,n}$.
\end{remark}

%\marginpar{Add a remark: $\MM_{g,n}$ as moduli space of $PSL_2(\RR)$-flat bundles/local systems}

\begin{remark} The following remark is from \cite{Penner}.
Given a closed surface $\Sigma$ of genus $g\geq 2$, by the Uniformization Theorem (see (\ref{l10 Sigma^model/Gamma}) and Remark \ref{l10 rem: conjugation}) one has a map
\begin{equation}
\{\mbox{conformal structures on }\Sigma\}\ra \{\mr{subgroups\;}\Gamma\subset PSL_2(\RR)\;\mr{s.t.}\; \Gamma\simeq \pi_1(\Sigma)\}/PSL_2(\RR)
\end{equation}
More specifically, one has a map 
\begin{equation}
\mc{T}_{g,0}\xra{p} \mr{Hom}(\pi_1(\Sigma),PSL_2(\RR))/PSL_2(\RR)
\end{equation}
In fact, $p$ is injective and its image is 
\begin{equation}
\mr{im}(p)= \mr{Hom}^{df}(\pi_1(\Sigma),PSL_2(\RR))/PSL_2(\RR)
\end{equation}
where superscript $d$ stands for ``discrete'' (so that $1$ is not an accumulation point of the image of $\pi_1$), $f$ is for ``faithful'' (injective). One can also allow marked points -- then one gets bijection
\begin{equation}\label{l10 T as moduli space of PSL_2(R)-local systems}
\mc{T}_{g,n}\xra{\sim} \mr{Hom}^{dfp}(\pi_1(\Sigma_{g,n}),PSL_2(\RR))/PSL_2(\RR)
\end{equation}
where superscripts $d,f$ are as above and $p$ means ``periferal cycles map to parabolic elements of $PSL_2(\RR)$'' (i.e. elements with trace $\pm 2$). On the right hand side, $\Sigma_{g,n}$ is understood as a surface of genus $g$ with $n$ points removed. Thus, one has an identification of the Teichm\"uller space with a (part of) the moduli space of $PSL_2(\RR)$-local systems on $\Sigma$. For instance, the formula for the dimension of the Teichm\"uller space
\begin{equation}
\dim \mc{T}_{g,n}=6g-6+2n
\end{equation}
follows from (\ref{l10 T as moduli space of PSL_2(R)-local systems}) immediately.
\end{remark}

\subsection{Aside: cross-ratio}
\begin{definition} Given four pairwise distinct points $z_1,z_2,z_3,z_4$ in $\CP^1$, their \emph{cross-ratio} is the number
\begin{equation}\label{l10 cross-ratio}
[z_1,z_2:z_3,z_4]\colon= \frac{(z_1-z_3)(z_2-z_4)}{(z_1-z_4)(z_2-z_3)} = \frac{z_1-z_3}{z_1-z_4}:\frac{z_2-z_3}{z_2-z_4}\quad \in
\CC\backslash\{0,1\} .
 %\CC^*
\end{equation}
\end{definition}
\begin{lemma} The cross-ratio is invariant under M\"obius transformations:
\begin{equation}
[Az_1,Az_2:Az_3,Az_4]=[z_1,z_2:z_3:z_4]
\end{equation}
for any $A\in PSL_2(\CC)$. Put another way, the cross-ratio is
 a function on the open configuration space $C_4(\CP^1)$ of 4 points on $\CP^1$ invariant under the diagonal action of $PSL_2(\CC)$.
\end{lemma}
\begin{proof}
The M\"obius group is generated by translations $z\ra z+a$ with $a\in \CC$, rotations plus dilations $z\ra \lambda z$ with $\lambda\in \CC^*$, and the transformation $z\ra 1/z$. The expression (\ref{l10 cross-ratio}) depends only on differences of $z$'s, so it is invariant under translations. It is a rational function of total homogeneity degree $0$, so it is invariant under $z\ra \lambda z$. The only thing left to check is that the cross-ratio is invariant under $z\ra 1/z$. We have
$$[z_1^{-1},z_2^{-1}:z_3^{-1},z_4^{-1}]=\frac{(z_1^{-1}-z_3^{-1})(z_2^{-1}-z_4^{-1})}{(z_1^{-1}-z_4^{-1})(z_2^{-1}-z_3^{-1})} =
\frac{(z_3-z_1)(z_4-z_2)}{(z_4-z_1)(z_3-z_2)} = [z_1,z_2:z_3,z_4].
$$
\end{proof}

\begin{definition}
A an action of a group on a manifold $\rho\colon G\ra \Diff(M)$ is said to be $k$-transitive, for some $k\geq 1$, if any $k$-tuple of distinct points  in $M$ can be mapped to any other $k$-tuple of distinct points by acting with some element $g\in G$. Put another way, the action $\rho$ is $k$-transitive if the corresponding diagonal action on the open configuration space of $k$ points,
$\rho\colon G\ra \Diff(C_k(M))$ is transitive.
\end{definition}

\begin{lemma}\label{l10 lm 3-transitivity} 
\begin{enumerate}[(a)]
\item \label{l10 lm 3-transitivity (a)} 
The action of $PSL_2(\CC)$ on $\CP^1$ by M\"obius transformations is 3-transitive.
\item \label{l10 lm 3-transitivity (b)}  The M\"obius transformation sending any one given triple of distinct points in $\CP^1$ to any other triple is unique.
\end{enumerate}
\end{lemma}
\begin{proof} For (\ref{l10 lm 3-transitivity (a)}),
it suffices to check that for any triple of distinct points $(z_1,z_2,z_3)$ in $\CP^1$ there exists a M\"obius transformation that moves it to the triple $(\infty,0,1)$. We can find it as the following composition of simple M\"obius transformation:
\begin{multline}
(z_1,z_2,z_3)\xra{z\ra z^{-1}} (z_1^{-1},z_2^{-1},z_3^{-1})\xra{z\ra z-z_1^{-1}} (0,z_2^{-1}-z_1^{-1},z_3^{-1}-z_1^{-1}) \xra{}\\
\xra{z\ra z^{-1}} (\infty, \frac{z_1z_2}{z_{12}},\frac{z_1z_3}{z_{13}}) \xra{z\ra z-\frac{z_1z_2}{z_{12}}} 
(\infty,0,% \frac{z_1z_3}{z_{13}}-\frac{z_1z_2}{z_{12}}
\frac{z_1^2z_{32}}{z_{12}z_{13}}
) 
\xra{z\ra z\cdot \frac{z_{12}z_{13}}{z_1^2 z_{32}}} 
(\infty,0,1).
\end{multline}
Here we used a shorthand notation $z_{ij}\colon=z_i-z_j$.

For (\ref{l10 lm 3-transitivity (b)}) it suffices to show that the only M\"obius transformation mapping $(0,\infty,1)$ to $(0,\infty,1)$ is the identity map $z\ra z$. Indeed, for a general M\"obius transformation (\ref{l6 Mobius transf}), we have
$$ (0,\infty,1)\mapsto (\frac{b}{d},\frac{a}{c},\frac{a+b}{c+d}). $$
For the right hand side to be $(0,\infty,1)$, one needs $b=c=0$ and $a=d$, thus the transformation (\ref{l6 Mobius transf}) is the identity.
\end{proof}

\begin{lemma}
The cross-ratio (\ref{l10 cross-ratio}) has the following meaning: start with a quadruple of distinct points $(z_1,z_2,z_3,z_4)$ 
in $\CP^1$. Find the (unique) M\"obius transformation that transforms the quadruple to one of the form $(\varkappa,1,0,\infty)$ with some $\varkappa\in \CP^1\backslash\{0,1,\infty\}$.
Then one has %$\varkappa$ is the cross-ratio of the original quadruple
\begin{equation}
[z_1,z_2:z_3,z_4]=\varkappa.
\end{equation}
%and move the last three points to $(1,0,\infty)$ by a (unique) M\"obius transformation. Then this transformation moves $z_1$
\end{lemma}
\begin{proof}
By 3-transitivity of the M\"obius transformations, it suffices to check that the cross-ratio $[\varkappa,1:0,\infty]$ is $\varkappa$, and this is obvious from the definition (\ref{l10 cross-ratio}).
\end{proof}

\begin{remark} The group $S_4$ of permutations of $z_1,z_2,z_3,z_4$ acts on the cross-ratio. Its orbits consists of sextuples of the form
\begin{equation}
\varkappa\sim \frac{1}{\varkappa}\sim 1-\varkappa\sim \frac{\varkappa}{\varkappa-1}\sim\frac{1}{1-\varkappa}\sim \frac{\varkappa-1}{\varkappa}.
\end{equation}
More precisely, one has a short exact sequence of groups
$$  \ZZ_2\times \ZZ_2\ra S_4 \ra S_3 , $$
where $\ZZ_2\times \ZZ_2$ (the ``Klein four-group'') is the symmetries of the cross-ratio -- permutations of the four points that don't change it. Explicitly, these symmetries are:
$$ [z_1,z_2:z_3,z_4]=[z_2,z_1:z_4,z_3]= [z_3,z_4:z_1,z_2]=[z_4,z_3:z_2,z_1] .$$
\end{remark}

\marginpar{Lecture 11,\\
9/16/2022}

\subsection{Moduli space $\MM_{0,n}$}

A sphere $\Sigma=S^2$ equipped with some conformal structure and $n$ (distinct) marked points is conformally equivalent to the standard $\CP^1$, by the Uniformization Theorem. Under this conformal equivalence, the points are mapped to  the $n$-tuple of distinct points $z_1,\ldots,z_n \in \CP^1$. Note that the surfaces $(\CP^1,\{z_i\})$ and $(\CP^1,\{z'_i\})$ are conformally equivalent if and only if one can find a conformal automorphism  $\alpha\in \Conf(\CP^1)=PSL_2(\CC)$ such that $z'_i=\alpha( z_i)$, for $i=1,\ldots,n$.

Thus, we have the following:
\begin{itemize}
\item For $n=3$, any three points can be mapped to $0,1,\infty\in \CP^1$ by a M\"obius transformation (in a unique way). Thus, all surfaces $(\CP^1,\{z_1,z_2,z_3\})$ are conformally equivalent to the standard one $(\CP^1,\{0,1,\infty\})$. Hence the moduli space $\MM_{0,3}$ is a single point.
\item For $n=4$,  a quadruple of points can be mapped by unique M\"obius transformation to the quadruple of the form $(\varkappa,1,0,\infty)$ where $\varkappa=[z_1,z_2:z_3,z_4]$ -- the cross-ratio. Thus, the surface $(\CP^1,\{z_1,z_2,z_3,z_4\})$ is conformally equivalent the surface of the form $(\CP^1,\{\varkappa,1,0,\infty\})$. So, genus $0$ Riemann surfaces with $4$ marked points up to conformal equivalence are parametrized by a single complex parameter $\varkappa\in \CP^1\backslash\{0,1,\infty\}$. Hence, we have 
\begin{equation}\label{l11 M_0,4}
\MM_{0,4}\simeq \CP^1\backslash \{0,1,\infty\}
\end{equation}
and the coordinate on the moduli space is provided by the cross-ratio of the four marked points on $\Sigma=\CP^1$.
\item For $n=5$, one can map the last $3$ out of $5$ marked points to $1,0,\infty$ by a unique M\"obius transformation; this transformation moves the first two points to some $\varkappa_1\neq \varkappa_2 \in \CP^1\backslash\{0,1,\infty\}$, with 
%$\varkappa_{1,2}$ 
$$ \varkappa_{1,2}=[z_{1,2},z_3:z_4,z_5] $$
the cross-ratios. Thus, one has
\begin{equation}
\MM_{0,5}\simeq C_2(\CP^1\backslash\{0,1,\infty\})
\end{equation}
-- the open configuration space of two distinct points $\varkappa_1,\varkappa_2$ in $\CP^1\backslash\{0,1,\infty\}$.
\item Similarly, for any $n\geq 3$, one has 
\begin{equation}\label{l11 M_0,n}
\MM_{0,n}\simeq C_{n-3}(\CP^1\backslash\{0,1,\infty\})
\end{equation}
where the surface $(\CP^1,\{z_1,\ldots,z_n\})$ corresponds to the point $(\varkappa_i={[z_,z_{n-2}:z_{n-1},z_n])_{i=1}^{n-3}}$ in the configuration space in the r.h.s. of (\ref{l11 M_0,n}).
\item (``Unstable case.'') For $n<3$, one can fix $n$ marked points to standard positions, but by a non-unique M\"obius transformation. So, the corresponding moduli can be thought of as a the quotient of a point (the standard $\CP^1$ with $n$  marked points in standard positions) by the subgroup $G_n\subset PSL_2(\CC)$ fixing the marked points:
\begin{equation}
\MM_{0,n}\simeq \mr{pt}/G_n
\end{equation}
-- thought of as category with a single object and $G_n$ worth of morphisms, or as a stack. %\marginpar{write $G$ explicitly for $n=0,1,2$}
Explicitly, the groups $G_n$ are:

\vspace{0.3cm}
\noindent
%\hspace{-1.5cm}
\begin{tabular}{c|c}
$n$ & $G_n$ \\ \hline
$0$ & $PSL_2(\CC)$ \\
$1$ & $\mr{Stab}_\infty(PSL_2(\CC)\;\calt\; \CP^1)= \{\mr{dilations}\}\oplus \{\mr{rotations}\}\oplus \{\mr{translations}\}\simeq \CC^*\ltimes \CC$ \\
$2$ & $\mr{Stab}_\infty\cap \mr{Stab}_0(PSL_2(\CC)\;\calt\; \CP^1)=
\{\mr{dilations}\}\oplus \{\mr{rotations}\}\simeq \CC^*
$
\end{tabular}
\end{itemize}

\vspace{0.5cm}

\subsubsection{Deligne-Mumford compactification}.  The moduli space $\MM_{0,n}$ with $n\geq 3$ is a smooth noncompact manifold. It admits the so-called Deligne-Mumford compactification $\ol\MM_{0,n}$ -- a stratified complex manifold. The main stratum (of codimension $0$) is $\MM_{0,n}$. A stratum $D_{S_1,S_2}$ of complex  codimension $1$ corresponds to a partitioning of the set of marked points $z_1,\ldots,z_n$ into two subsets $S_1,S_2$, each containing $\geq 2$ points; the stratum $D_{S_1,S_2}$ corresponds to ``nodal curves/surfaces''\footnote{There are competing terminologies for complex manifolds of complex dimension $1$ -- ``curves'' (mainly, in algebraic geometry literature) and ``surfaces'' (differential geometry literature). 
%While we are  mostly sticking with ``surfaces,'' we nevertheless have to say ``nodal curve'' as that is the standard term.
We will try to be consistent, sticking with ``surfaces.'' In particular, instead of ``nodal curve'' (a standard term in algebraic geometry), we say ``nodal surface.''
}
$$(\CP^1,\{S_1,p\})\cup_p (\CP^1,\{S_2,p\})$$ with ``neck'' at a point $p$. The moduli space of such nodal surfaces is 
\begin{equation}\label{l11 DM stratum codim=1}
D_{S_1,S_2}\simeq \MM_{0,|S_1|+1}\times \MM_{0,|S_2|+1}
\end{equation}
One adds higher-codimension strata by induction, compactifying the r.h.s. of (\ref{l11 DM stratum codim=1}).
%The stratum of complex codimension $k$ corresponds to partitioning the set of marked points 

We refer to all the strata of $\ol\MM_{0,n}$ except for the main one ($\MM_{0,n}$) as \emph{compactification strata}.

%\marginpar{straighten up terminology: nodal curves vs nodal surfaces}
\begin{example}[$\ol\MM_{0,4}$] The Deligne-Mumford compactification of the moduli space $\MM_{0,4}$ (\ref{l11 M_0,4}) glues back in the points $\varkappa=0,1,\infty$ (as compactification strata of complex codimension $1$), thus 
\begin{equation}
\ol\MM_{0,4}=\underbrace{\CP^1\backslash\{0,1,\infty\}}_{\MM_{0,4}}\cup\{0,1,\infty\}=\CP^1.
\end{equation} 
E.g., the point $\varkappa=0$ corresponds to the asymptotic situation for a surface $\CP^1,\{z_1,z_2,z_3,z_4\}$ where $z_1$ approaches $z_3$. Note that such configuration can be mapped by a M\"obius transformation to one where $z_1,z_3$ stay at finite distance from each other but $z_2$ and $z_4$ approach one another. The limiting configuration is described by a nodal surface -- two $\CP^1$'s, one containing $z_1,z_3$ and $p$ (the ``neck'') and the other containing $z_2,z_4$ and $p$. This singular surface is acted on by $PSL_2(\CC)\times PSL_2(\CC)$ -- independent M\"obius transformations of both $\CP^1$'s. Thus, on both components of the singular surface, there are no moduli (3 marked points can be brought into standard position), so the stratum is $\MM_{0,3}\times \MM_{0,3}=\mr{pt}$.

%\textcolor{red}{PICTURE}
\begin{figure}[H]
%$$\vcenter{\hbox{ \includegraphics[scale=0.5]{cob1.eps} }} $$
\begin{center}
\includegraphics[scale=0.6]{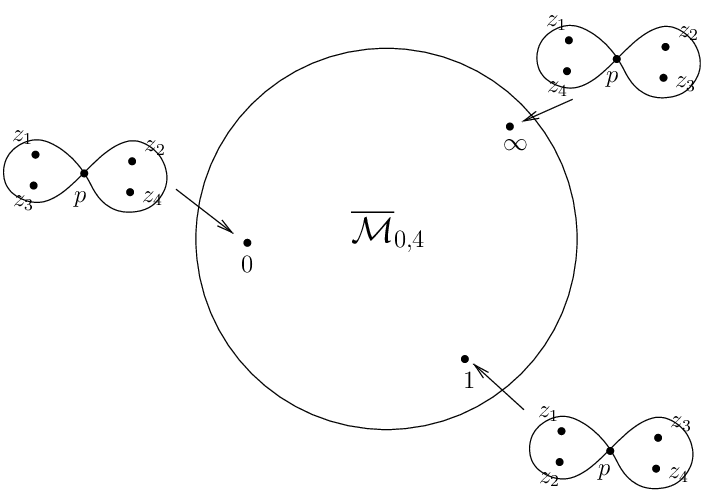}
\end{center}
\caption{Deligne-Mumford compactification of $\MM_{0,4}$. We are drawing the nodal surfaces corresponding to the compactification strata $\varkappa=0$, $\varkappa=1$,$\varkappa=\infty$. Put another way, the universal family (Definition \ref{l10 def: univ family}) degenerates at these three points and we draw the degenerate fibers over them.}
\label{l11 fig Mbar_0,4}
\end{figure}
\end{example}

\begin{example}[Higher-codimension strata]
In the Deligne-Mumford compactification of $\MM_{0,5}$, one can consider the $\mr{codim}_\CC=1$ compactification stratum of the form 
\begin{equation}\label{l11 stratum in M_0,5}
\MM_{0,3}\times \MM_{0,4},
\end{equation} 
corresponding to partitioning the marked points as $\{z_1,z_2\}\cup \{z_3,z_4,z_5\}$, i.e., nodal surfaces of the form 
\begin{equation}\label{l11 stratum in M_0,5 surface}
(\CP^1,\{z_1,z_2,p\})\cup_p (\CP^1,\{p,z_3,z_4,z_5\}) 
\end{equation}
(corresponding to either $z_1$ approaching $z_2$ or, as alternative viewpoint, corresponding to $z_3,z_4,z_5$ colliding together). The right factor in (\ref{l11 stratum in M_0,5}) also should be further compactified, by adjoining to the product, e.g., the stratum $\MM_{0,3}\times \MM_{0,3}\times \MM_{0,3}$ corresponding to surfaces with two necks, of the form
\begin{equation}
(\CP^1,\{z_1,z_2,p\})\cup_p (\CP^1,\{p,z_3,q\}) \cup_q (\CP^1,\{q,z_4,z_5\}) 
\end{equation}
of complex codimension $2$ (as a stratum in $\ol\MM_{0,5}$; it corresponds to a stratum of complex codimension 1 in the right factor of (\ref{l11 stratum in M_0,5})).
%\textcolor{red}{PICTURE}
\begin{figure}[H]
%$$\vcenter{\hbox{ \includegraphics[scale=0.5]{cob1.eps} }} $$
\begin{center}
\includegraphics[scale=1]{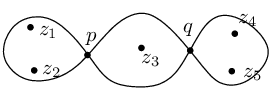}
\end{center}
\caption{Nodal surface with two ``necks,'' corresponding to a  stratum in $\ol\MM_{0,5}$ of complex codimension two.
}
\end{figure}
\end{example}

%\marginpar{remark: DM compactification for higher genus}
\begin{remark} \label{l11 rem: DM for higher genus}
The construction of Deligne-Mumford compactification extends to $\MM_{g,n}$ with nonvanishing genus $g$. Then one has compactification strata (of complex codimension $1$) of two types:
\begin{enumerate}[1.]
\item Strata isomorphic to $\MM_{g_1,n_1+1}\times \MM_{g_2,n_2+1}$ with $g_1+g_2=g$, $n_1+n_2=n$ -- this is essentially the same construction as above, where not only marked points but also genus is distributed between the two components of the nodal surface.
\item Strata isomorphic to $\MM_{g-1,n+2}$ -- this corresponds to introducing a neck on handle, thus trading one handle for two extra marked points.
\end{enumerate}
\end{remark}

\subsection{Moduli space $\MM_{1,0}$}
A 2-torus with conformal structure, is, by Uniformization Theorem, conformally equivalent to 
\begin{equation}\label{l11 C/Lambda}
\Sigma_\Lambda\colon=\CC/\Lambda
\end{equation}
with 
\begin{equation}\label{l11 Lambda}
\Lambda=\mr{span}_\ZZ(u,v)
\end{equation} 
a lattice in $\CC$ spanned by two non-collinear\footnote{Otherwise, the quotient is not diffeomorphic to the 2-torus.} vectors $u,v\in \CC$. Since the order of $(u,v)$ does not matter, we may assume that $\mr{Im}(v/u)>0$. Surfaces (\ref{l11 C/Lambda}) are conformally equivalent for lattices $\Lambda$, $\Lambda'$ if and only if the lattices are related by rotation and scaling. There is a unique rotation+scaling that transforms $v$ to $1$. Thus, the surface (\ref{l11 C/Lambda}) is equivalent to a surface of the form
\begin{equation}
T_\tau\colon= \CC/\Lambda_\tau 
%\qquad \mr{where}\;\; \Lambda_\tau\colon=\mr{span}_\ZZ(\tau,1),\;\;\mr{with}\;\; \tau\in \Pi_+
\end{equation}
where $\Lambda_\tau\colon=\mr{span}_\ZZ(\tau,1)$ with $\tau=\frac{u}{v}\in \Pi_+$.  

Choosing a different basis in $\Lambda$,
$$ (u,v)\mapsto (u'=au+bv,v'=cu+dv) 
\quad \mbox{ with }
\left(
\begin{array}{cc}
a&b \\ c&d
\end{array}
\right) \in SL_2(\ZZ),
$$
one obtains that tori $T_\tau$ and $T_{\tau'}$ are equivalent if and only if 
\begin{equation}
\tau'=\frac{a\tau+b}{c\tau+d} \quad \mbox{ with }
\left(
\begin{array}{cc}
a&b \\ c&d
\end{array}
\right) \in PSL_2(\ZZ).
\end{equation}

Thus, we have the following.
\begin{thm} The moduli space of complex structures on a 2-torus with no marked points is
\begin{equation}
\MM_{1,0} = \Pi_+/PSL_2(\ZZ).
\end{equation}
I.e. any complex torus is conformally equivalent to a torus of the form $T_\tau=\CC/(\ZZ\oplus \tau \ZZ)$ where the \emph{modular parameter} $\tau\in \Pi_+/PSL_2(\CC)$ provides a complex coordinate on $\MM_{1,0}$.
\end{thm}

\begin{remark} The standard way to choose a \emph{fundamental domain}\footnote{
I.e. a subset of $\Pi_+$ such that each $PSL_2(\ZZ)$-orbit intersects $\mc{D}$ and if two points in $\mc{D}$ are in the same orbit, then they are boundary points of $\mc{D}$.
}
 $\mc{D}\subset \Pi_+$ for the action of $PSL_2(\ZZ)$ on $\Pi_+$ is the following:
\begin{equation}
\mc{D}=\{z\in \CC\;|\; \mr{Re}(z)\in [-\frac12,\frac12],\; |z|\geq 1\}
\end{equation}
The action of $PSL_2(\ZZ)$ identifies points on the boundary of $\mc{D}$ as follows:
\begin{equation}
\MM_{1,0}\simeq \frac{\mc{D}}{-\frac12+iy \sim \frac12+iy\;\;\mr{for}\; y\geq \frac{\sqrt{3}}{2},\quad e^{i\theta}\sim e^{i(\pi-\theta)}\;\;\mr{for}\; \theta\in [\frac{\pi}{3},\frac{2\pi}{3}]}
\end{equation}
Points $\tau=i$ and $\tau=e^{\pi i/3}\sim e^{2\pi i/3}$ in $\mc{D}$ have nontrivial stabilizers ($\ZZ_2$ and $\ZZ_3$, respectively) under the action of $PSL_2(\ZZ)$ and correspond to orbifold singularities in $\MM_{1,0}$.

%\textcolor{red}{PICTURE}
\begin{figure}[H]
%$$\vcenter{\hbox{ \includegraphics[scale=0.5]{cob1.eps} }} $$
\begin{center}
\includegraphics[scale=1]{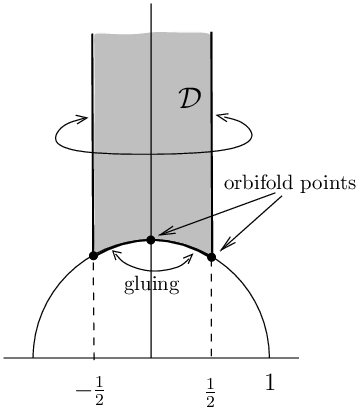}
\end{center}
\caption{$\MM_{1,0}$}
\end{figure}
\end{remark}

\begin{remark} Each complex torus $T_\tau$ has a nontrivial group of conformal automorphisms -- translations by vectors in $T_\tau$,
\begin{equation}\label{l11 Conf(T_tau)}
\Conf(T_\tau)=T_\tau.
\end{equation}
\end{remark}

\begin{remark}
The moduli space $\MM_{1,1}$ of complex tori with a single marked point can be identified with $\MM_{1,0}$: one can convert the underlying complex torus to a standard one $T_\tau$ and then move the marked point to the standard position (say, $z=0$) by a translation from (\ref{l11 Conf(T_tau)}).
\end{remark}

\subsection{The mapping class group of a surface}
We refer to \cite{Farb_Margalit} as an excellent detailed introduction to the subject of
%for all the details on
 %the lore of 
 mapping class groups of surfaces. Here we just want to give some simple examples.

\begin{example}
The mapping class group of a 2-torus (seen as a smooth manifold $\RR^2/\ZZ^2$ with no marked points) is
\begin{equation}
\MCG_{1,0}=SL_2(\ZZ)
\end{equation}
-- elements of the mapping class group can be represented by linear automorphisms $\RR^2\ra \RR^2$ preserving the lattice $\ZZ^2$. 
\end{example}

\begin{example}
The mapping class group of the sphere $S^2$ with $n$ marked points is the ``spherical braid group on $n$ strands,'' i.e.,
\begin{equation}
\MCG_{0,n} = \pi_1 C_n^\mr{non-ordered}(S^2)
\end{equation}
-- the fundamental group of the open configuration space of $n$ non-ordered points on $S^2$. 

The version for the pure mapping class group (respectively, pure spherical braid group on $n$ strands) is:
\begin{equation}
\pMCG_{0,n} = \pi_1 C_n^\mr{ordered}(S^2).
\end{equation}
\end{example}

\begin{example}
The mapping class group of the annulus relative to the boundary (i.e. $\pi_0$ of diffeomorphisms of the annulus not moving the boundary points) is 
\begin{equation}
\MCG(\mr{Ann},\dd \mr{Ann}) \simeq \ZZ
\end{equation}
This group is generated by the \emph{Dehn twist}. Thinking of $\mr{Ann}$ as the domain ${\{z\in \CC\;|\; r\leq |z|\leq R\}}$, the Dehn twist can be represented a diffeomorphism\footnote{
Equivalently, thinking of $\mr{Ann}$ as a cylinder $[0,1]\times S^1$, one can represent the Dehn twist by the diffeomorphism $(t,\theta)\mapsto (t,\theta+2\pi t)$.
}
\begin{equation}
\begin{array}{ccc}
\mr{Ann} & \ra & \mr{Ann} \\
z& \mapsto & e^{2\pi i\frac{|z|-r}{R-r}}\cdot z
\end{array}
\end{equation}
\begin{figure}[H]
%$$\vcenter{\hbox{ \includegraphics[scale=0.5]{cob1.eps} }} $$
\begin{center}
\includegraphics[scale=1]{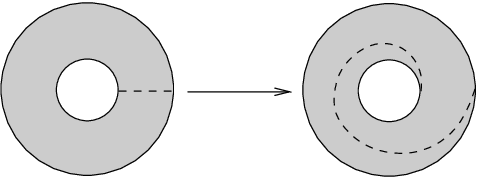}
\end{center}
\caption{Dehn twist (illustrated by the image of the dashed curve). }
\label{l11 fig: Dehn twist}
\end{figure}
\end{example}

For general genus $g$ and number $n$ of marked points, one can write a presentation of the mapping class group $\MCG_{g,n}$ with two types of generators:
\begin{itemize}
\item Dehn twists along a finite collection of  nonseparating closed  simple curves  on the surface.\footnote{The Dehn twist along a closed simple curve $\gamma$ on a surface $\Sigma$ is the diffeomorphism that is identity everywhere except in in a small tubular neighborhood $U_\gamma\subset \Sigma$ of $\gamma$; in $U_\gamma$ (which is diffeomorphic to an annulus or, equivalently, a cylinder), one performs the standard Dehn twist (Figure \ref{l11 fig: Dehn twist}).}
\item ``Dehn half-twists'' which permute pairs of marked points.
\begin{figure}[H]
%$$\vcenter{\hbox{ \includegraphics[scale=0.5]{cob1.eps} }} $$
\begin{center}
\includegraphics[scale=1]{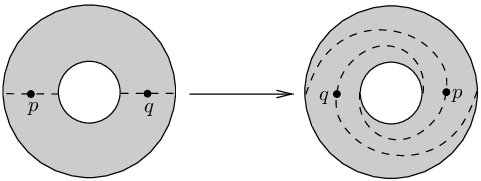}
\end{center}
\caption{Dehn half-twist permuting the marked points $p$ and $q$.}
\end{figure}
\end{itemize}
These generators are subject to a set of relations. We refer to \cite{Farb_Margalit} for the details.

For the \emph{pure} mapping class group $\pMCG_{g,n}$, one can make do with just Dehn twists (without half-twists). E.g., $\pMCG_{0,n}$ can be generated by Dehn twists along curves encircling pairs of marked points.

Let us also mention the following result, helpful in computing mapping class groups. Let $\MCG^\pm(\Sigma)=\pi_0\Diff(\Sigma)$ -- the group of isotopy classes of diffeomorphisms that either preserve or reverse the orientation of $\Sigma$. Note that there is a natural action of $\MCG^\pm$ on the fundamental group $\pi_1(\Sigma)$ (by pushing loops along the diffeomorphism).
\begin{thm}[Dehn-Nielsen-Baer]
\begin{equation}\label{l11_DNB thm}
\MCG^\pm(\Sigma)\simeq \mr{Out}(\pi_1(\Sigma))
\end{equation}
where for $G$ a group, $\mr{Out}(G)\colon=\mr{Aut}(G)/\mr{Inn}(G)$ is the group of ``outer automorphism'' -- the quotient of all automorphisms by inner ones.
\end{thm}
Then, the usual mapping class group (classes of orientation preserving diffeomorphisms) is an index two subgroup of (\ref{l11_DNB thm}).

\chapter{Symmetries in classical field theory, stress-energy tensor}
\chaptermark{Symmetries in classical field theory}

\marginpar{Lecture
12, 9/19/2022}

\section{Classical field theory, Euler-Lagrange equations}

\subsection{Basic setup} An outline of classical field theory was given in Section \ref{ss Segal from path integral}).
Here we give more details.

Let $M$ be a smooth $n$-dimensional manifold, possibly with a geometric structure such as a metric,
an orientation etc.. A classical local field theory on $M$ is determined by the following data. 

\begin{enumerate}[(a)]
\item The space of fields associated to $M$ is
%is a 
%%vector space that is attached 
%Fr\'echet manifold associated
%to $M$. It can be the space of sections
%of a fiber bundle, the space of connections, etc. 
%%Let us focus on the example of a {\it sigma-model} with the target space $X$. Fix a fiber bundle $E$ over $M$ with fiber $X$. 
%Fix a fiber bundle $E$ over $M$.
%The corresponding space of
%fields is the space of smooth sections of $E$:
\begin{equation}
\F_{M}=\Gamma(M,E).
\end{equation}
the Fr\'echet manifold of smooth sections of a fiber bundle $E$ over  $M$ (the ``field bundle'').
 %\footnote{For example, for the \emph{sigma-model} with target space $X$ one has $E=M\times X$ -- the trivial fiber bundle over $M$ with fiber $X$. In this case fields are maps $\phi\colon M\ra X$.}
We denote fields, i.e. sections of this bundle,  by $\phi$. When $E=M\times X$ is a trivial bundle with fiber $X$, the space $X$ is called the target space and fields are mappings $\phi: M\to X$.

For example, fields can be differential forms, spinors, connections in a principal bundle, etc. All these examples correspond to particular choices of the field bundle $E$.

\item The local action functional
\begin{equation}\label{l12 S}
S_{M,g}(\phi) = \int_{M}   L(\phi,\dd\phi,\ldots; g). 
%d^nx.
\end{equation}
\end{enumerate}

Here  $L
%(\phi,\dd\phi,\ldots; g)d^nx
$ %should be regarded as one symbol and 
is a density on $M$. When $M$ is oriented, it is an $n$-form on $M$ depending locally on fields and on possible geometric data on $M$ (which, typically, include a metric).  If the geometric data on $M$ provide a volume form $d^nx$, we have 
$L=\mathsf{L}d^n x$ where 
the Langrangian
function $\mathsf{L}(\phi,\dd\phi,\ldots; g)$ depending locally on fields and possible geometric structures. 

\subsection{Variational bicomplex} Recall that one can introduce the ``variational bicomplex'' (see \cite{Anderson} for details) 
\begin{equation}
\Omega^{p,q}_\mr{loc}(M\times \F_M)
\end{equation}
of ``local'' $q$-forms on $\F_M$ valued in $p$-forms on $M$; locality means that for $\omega\in \Omega^{p,q}_\mr{loc}$, its value at $x\in M$ depends on $x$ and the jet of fields at $x$ (but not on the values of the fields away from $x$). The bicomplex $\Omega^{\bt,\bt}_\loc(M,\F_M)$ comes with 
\begin{itemize}
\item the ``vertical'' differential $\delta\colon \Omega^{p,q}_\loc\ra \Omega^{p,q+1}_\loc$ -- the de Rham operator on the space of fields. Locally, in a local trivialization of $E$, one has $\displaystyle \delta = \sum_{r\geq 0}\delta \phi^a_{i_1\cdots i_r} \frac{\dd}{\dd \phi^a_{i_1\cdots i_r}}$
where $a$ labels the field components (coordinates in the fiber of the field bundle $E$); $\phi^a_{i_1\cdots i_r}=\dd_{i_1}\cdots \dd_{i_r}\phi^a$ are components of the $r$-th jet of the field.

\item the ``horizontal'' differential $d\colon \Omega^{p,q}_\loc\ra \Omega^{p+1,q}_\loc$ -- the de Rham operator on $M$. It is understood that $d$ also acts on fields. Locally, one has 
\begin{equation}
d=dx^i \left(\frac{\dd}{\dd x^i}+ \sum_{r\geq 0} \phi^a_{ii_1\cdots i_r} \frac{\dd}{\dd \phi^a_{i_1\cdots i_r}}
+ \sum_{r\geq 0}(\delta \phi^a_{ii_1\cdots i_r}) \frac{\dd}{\dd (\delta\phi^a_{i_1\cdots i_r})}
\right) .
\end{equation}

\end{itemize}
The two differentials $d,\delta$ both square to zero and anticommute with each other.

%\newpage

Another viewpoint on the bicomplex $\Omega^{\bt,\bt}_\loc$ is as follows. Assume that the space of fields is as above $\F_M=\Gamma(M,E)$. Consider the composition of maps \\ ${M\times \Gamma(M,E)\xra{\mr{id}\times j_\infty} M\times \Gamma(M,\mr{Jet}_\infty E) \xra{\mr{ev}}\mr{Jet}_\infty E}$ where $j_\infty$ takes the jet
%\footnote{Recall that a jet of a section is....} 
of a section of $E$ at each point of $M$; $\mr{ev}$ is the evaluation of a section at a point of $M$.
Consider the complex of forms $\Omega(\mr{Jet}_\infty E)$ on the total space of the jet bundle. Then $\Omega_\loc(M\times F_M)$ is the image of $\Omega(\mr{Jet}_\infty E)$ under the pullback $(\mr{ev}\circ (\mr{id}\times j_\infty))^*$. %{\bf edit}

%\marginpar{Image where?}

In terms of the variational bicomplex, the Lagrangian density in (\ref{l12 S}) is an element 
\begin{equation}
L\in \Omega^{n,0}_\loc (M\times \F_M).
\end{equation}

\subsection{Covariance} 
We are assuming ``covariance'' of the given classical field theory (as assignment of action functionals to (pseudo-)Riemannian $n$-manifolds):
for a diffeomorphism $f\colon M\xra{\sim} M'$ of smooth manifolds, one has
\begin{equation}\label{l12 covariance}
S_{M,g}(\phi)=S_{M',(f^{-1})^*g}((f^{-1})^*\phi).
\end{equation}

\subsection{Euler-Lagrange equations}\label{sss EL}

The Euler-Lagrange equation (or ``equation of motion'') is the condition on a field $\phi\in \F_M$ 
is a critical point of the action functional $S(\phi)$ (ignoring boundary terms).

Let  $\gamma=\{\phi_t\}_{t\in [0,\epsilon)}$ be a path in $\F_M$ such that $\phi_0=\phi$.

 and such that $\phi_t$ coincides with $\phi$ in a neighborhood of the boundary $\dd M\subset M$, one has
\begin{equation}\label{l12 EL as Frechet derivative}
\left.\frac{d}{dt}\right|_{t=0} S(\phi_t)=0.
\end{equation}
(Put another way, Fr\'echet derivatives of $S$ at $\phi$ in the directions given by fluctuations of $\phi$ supported away from $\dd M$ vanish.) This condition is a equivalent to a PDE on $\phi$. 
We will sometimes shorten ``Euler-Lagrange equation'' to just ``EL.''

More explicitly, applying the de Rham derivative in fields to $S$ (or in other words, considering the variation of $S$ with respect to the variation of the field), one can write the result in the form
%\marginpar{Mention ``source forms'' and unique splitting (source)+d(...)?}
\begin{equation}\label{l12 delta S}
\delta S=\int_M \EL(\phi)_a\delta\phi^a + \int_{\dd M}\ul\alpha,
\end{equation}
where $\EL(\phi)_a\delta\phi^a \in \Omega^{n,1}_\loc(M\times \F_M)$ is an expression containing variations of the field not hit by derivatives along $M$ (labels $a$ refer to the local trivialization of the field bundle $E$). To obtain an answer in this form, one has to integrate by parts  (to move the geometric derivatives from $\delta\phi$), which results in the appearance of the boundary term in (\ref{l12 delta S}). The integrand $\ul\alpha$ in the boundary term of (\ref{l12 delta S}) is an element of $\Omega^{n-1,1}_\loc(M\times\F_M)$; the whole boundary term $\int_{\dd M}\ul\alpha$ is a 1-form on $\F_M$, known as the Noether 1-form (thus, $\ul\alpha$ is the density of the Noether 1-form).

Expressed in terms of densities, the equation (\ref{l12 delta S}) is:
\begin{equation}\label{l12 delta L= EL+d alpha}
\delta L=(-1)^n \Big(\EL_a(\phi)\delta\phi^a+ d\ul\alpha \Big).
\end{equation}

The Euler-Lagrange equation then is: 
\begin{equation}
\EL_a(\phi)=0,
\end{equation}
with $\EL_a(\phi)$ an expression in the jet of the field appearing in (\ref{l12 delta S}).

\subsubsection{Aside on source forms.}
%\begin{remark} 
In the variational bicomplex one can consider the subspace of ``source forms'' 
\begin{equation}
\Omega^{n,1\;\mr{source}}_\loc \subset \Omega^{n,1}_\loc 
%(M\times \F_M)
\end{equation}
-- the set of elements of the form $\omega=\omega_a(x,\phi,\dd\phi,\ldots) \delta\phi^a$, i.e., not depending on variations of derivatives of fields $\delta \phi^a_{i_1\cdots i_r}$ for $r\geq 1$. Then one has a straightforward lemma.
\begin{lemma}
\begin{equation}\label{l12 source form splitting}
\Omega^{n,1}_\loc = \Omega^{n,1\;\mr{source}}_\loc \oplus d\left(\Omega^{n-1,1}_\loc \right).
\end{equation}
I.e., any $(n,1)$-form $\beta$ can be written in a unique way as $\beta=\omega+d\eta$ with $\omega$ a source form. 
\end{lemma}
\begin{proof}
It is proven straightforwardly, by moving the derivatives from $\delta\phi$ to its prefactor in $\beta$, at the cost of adding a $d$-exact term:
\begin{multline}\label{l12 lemma3.1 computation}
\beta_a^{i_1\cdots i_r}(x,\phi,\dd\phi,\cdots) \delta \phi^a_{i_1\cdots i_r} =\\
= -(\dd_{i_r} \beta_a^{i_1\cdots i_r}) \delta \phi^a_{i_1\cdots i_{r-1}}+ \underbrace{\dd_{i_r} (\beta_a^{i_1\cdots i_r} \delta\phi^a_{i_1\cdots i_{r-1}})}_{d\iota_{\dd_{i_r}}(\beta^{i_1\cdots i_r}_a \delta \phi^a_{i_1\cdots i_{r-1}})}
\\
=\cdots = (-1)^r (\dd_{i_1}\cdots \dd_{i_r}\beta_a^{i_1\cdots i_r}) \delta\phi^a + d \left(\sum_{k=0}^{r-1} (-1)^{k}\iota_{\dd_{i_{r-k}}}((\dd_{i_r}\cdots \dd_{i_{r-k+1}} \beta_a^{i_1\cdots i_r}) \delta \phi^a_{i_1\cdots i_{r-k-1}}) \right).
\end{multline}
One extends this computation by $\RR$-linearity %additivity 
to general $\beta$'s.
%General $\beta$ is a sum of terms of the form  as in the l.h.s. (\ref{l12 lemma3.1 computation}), th
This gives a splitting $\beta=\omega+d\eta$, with $\omega$ a source form. The fact that the splitting is unique follows from the observation that $d$ of a field-dependent $(*,1)$-form will necessarily contain a term depending on $\delta\phi^a_{i_1\cdots i_r}$ with $r\geq 1$. Thus, $\Omega^{n,1\;\mr{source}}_\loc\cap d(\Omega^{n-1,1}_\loc)=0$.
\end{proof}

Equation (\ref{l12 delta L= EL+d alpha}) is an application of this lemma to $\beta=\delta L\in \Omega^{n,1}$; Euler-Lagrange equation    says that 
\begin{equation}
[\delta L]_\mr{source}=0,
\end{equation} 
where $[\cdots]_\mr{source}$ is the projection onto the 
%first summand in (\ref{l12 source form splitting}). %
source forms.
%\end{remark}

\section{Examples of classical field theories}

%\subsection{Classical mechanics}

\subsection{Free massive scalar field} \label{l12 ex free massive scalar}  Let $(M,g)$ be a (pseudo-)Riemannian $n$-manifold. The fields are smooth functions on $M$, $\phi\in C^\infty(M)$ (i.e., the field bundle is $E=M\times \RR\ra M$ -- the trivial rank one bundle over $M$). The action is
\begin{equation}\label{l12 S scalar}
S(\phi)=\int_{M} \underbrace{\frac12 d\phi\wedge *d\phi +\frac{m^2}{2} \phi^2 \dvol_g}_L
\end{equation}
with $*$ the Hodge star associated with the metric $g$ and $\dvol_g=*1$ the metric volume element; $m\geq 0$ is a fixed number (``mass''). It is obvious that the assignment (\ref{l12 S scalar}) satisfies the covariance property (\ref{l12 covariance}) -- essentially because the action is written in terms of natural geometric operations on forms.

We have 
\begin{multline}\label{l12 scal field EL computation}
\delta S =\int_M (-1)^{n+1}\underbrace{d\delta \phi \wedge *d\phi}_{d(\delta\phi \wedge *d\phi)+\delta\phi\wedge d*d\phi} + (-1)^n m^2 \delta\phi\,\phi\,\dvol_g =\\
= \int_M \underbrace{\dvol_g (\Delta+m^2)\phi \wedge \delta \phi}_{\EL(\phi)\delta\phi}+ d(\underbrace{*d\phi\wedge \delta\phi}_{\ul\alpha}) = 
\int_M \dvol_g (\Delta+m^2)\phi \wedge \delta \phi + \underbrace{\int_{\dd M}*d\phi\wedge \delta\phi}_{\mr{Noether\;1-form}}.
\end{multline}
Thus, the Euler-Lagrange equation  is the linear PDE
\begin{equation}
(\Delta+m^2)\phi =0
\end{equation}
with $\Delta=-*d*d$ the Laplace-Beltrami operator.

\subsection{Scalar field with self-interaction} 
Note that if we start instead with the action 
\begin{equation}
 S(\phi)=\int_M \frac12 d\phi\wedge *d\phi + V(\phi)\dvol_g
\end{equation}
with $V$ some smooth function on $\RR$ (the ``potential''), then by repeating the computation we see that the Noether 1-form $\alpha$ (and its density $\ul\alpha$) does not change from (\ref{l12 scal field EL computation}) but the Euler-Lagrange equation becomes a nonlinear PDE
\begin{equation}
\Delta\phi+V'(\phi)=0.
\end{equation}

\subsection{Yang-Mills theory}\label{l12 YM}
\label{l12 ex YM}
Fix a Lie group $G$ with a nondegenerate ad-invariant quadratic form $\langle,\rangle$ on its Lie algebra $\g$. Let  $(M,g)$ be a  (pseudo-)Riemannian $n$-manifold. The fields of the theory are pairs $(\mc{P},A)$ consisting of a principal $G$-bundle $\mc{P}$ over $M$ and a connection $A$ in $\mc{P}$. The action of Yang-Mills theory is
\begin{equation}
S(A)=\int_M \frac12 \langle  F_A \stackrel{\wedge}{,}  *F_A\rangle,
\end{equation}
where $F_A\in \Omega^2(M,\mr{ad}(\mc{P}))$ is the curvature $2$-form of the connection $A$; $*$ is again the Hodge star.
In a local trivialization of $\mc{P}$, $A$ is represented by a $\g$-valued $1$-form on $M$ (or rather on the trivializing neighborhood $U\subset M$) and $F_A$ is represented by the $\g$-valued $2$-form $dA+\frac12 [A,A]$.

The corresponding Euler-Lagrange equation is:
\begin{equation}\label{l12 YM eq}
d_A*F_A=0
\end{equation}
with $d_A\colon \Omega^\bt(M,\mr{ad}(\mc{P}))\ra \Omega^{\bt+1}(M,\mr{ad}(\mc{P}))$ the covariant derivative operator associated with $A$.  The equation (\ref{l12 YM eq}) is a nonlinear PDE (for nonabelian $G$) -- the ``Yang-Mills equation.'' 

In the special case $G=\RR$, the Yang-Mills theory drastically simplifies (in this case, it is called electrodynamics or Maxwell theory): fields are global $1$-forms $A\in\Omega^1(M)$, the action is $S(A)=\frac12\int_M dA\wedge *dA$ and the Euler-Lagrange equation becomes the Maxwell equation
\begin{equation}
 d*d A=0 
\end{equation}
-- a linear PDE.

\subsection{Chern-Simons theory}\label{l12 ex CS}
Fix a simply connected Lie group $G$ with a nondegenerate ad-invariant bilinear form $\langle,\rangle$ on its Lie algebra $\g$. Let $M$ be an oriented $3$-manifold. The fields of the theory are connections $A$ in the trivial principal bundle $\mc{P}=M\times G$ over $M$.\footnote{
In fact, one should allow connections in all principal $G$-bundles over $M$. However, for $G$ simply connected and $M$ 3-dimensional there are no nontrivial $G$-bundles over $M$ (since $BG$ is 3-connected and hence there is a unique homotopy class of classifying maps $M\ra BG$). This is why we asked $G$ to be simply connected -- to have this simplification. Case of non-simply connected $G$ can be treated but requires more care.
} Since $\mc{P}$ is trivial, connections can be identified with $\g$-valued $1$-forms, $A\in \Omega^1(M,\g)$. The action is defined as
\begin{equation}\label{l12 S Chern-Simons}
S(A) = \int_M \frac12\langle A\stackrel{\wedge}{,}dA \rangle+\frac16 \langle A\stackrel{\wedge}{,}[A,A] \rangle.
\end{equation}

We have
\begin{multline}
\delta S= \int_M -\frac12\langle \delta A , dA\rangle - \frac12 \underbrace{\langle A,d\delta A \rangle}_{-d\langle A,\delta A \rangle+\langle dA,\delta A \rangle} -\frac12 \langle \delta A, [A,A] \rangle=\\
=\int_M -\langle\delta A, dA+\frac12 [A,A] \rangle +\int_{\dd M} \frac12 \langle A,\delta A  \rangle = 
\int_M -\langle\delta A, F_A \rangle +\int_{\dd M} \underbrace{\frac12 \langle A,\delta A  \rangle}_{\ul\alpha},
\end{multline}
where $F_A=dA+\frac12 [A,A]\in \Omega^2(M,\g)$ is the curvature $2$-form. Thus, the Euler-Lagrange equation is
\begin{equation}
F_A=0
\end{equation}
-- zero-curvature (or ``flatness'') condition for the connection field $A$.

Note that the action (\ref{l12 S Chern-Simons}) does not depend on a metric on $M$ (it only depends on the orientation, no other geometric structure is used). -- It is an example of a topological field theory (in the sense of Schwarz).

\subsection{Nonlinear sigma model}\label{l12 ex: sigma model}
Fix a (pseudo-)Riemannian manifold  $(X,h)$ (the target) and let $(M,g)$ be a (pseudo-)Riemannian $n$-manifold (the source). The fields are smooth maps $\Phi\colon M\ra X$ (one can think of them as section of the field bundle $E=M\times X\ra X$ -- the trivial fiber bundle with fiber $X$) and the action is
\begin{equation}\label{l12 sigma model}
S(\Phi)=\int_M \frac12 \langle d\Phi \stackrel{\wedge}{,} *d\Phi\rangle_{\Phi^*h},
\end{equation}
where $*=*_g\colon \Omega^\bt(M)\ra \Omega^{n-\bt}$ is the Hodge star associated with the source metric $g$; $d\Phi\in \Omega^1(M,\Phi^* TX)$ is the differential of the map $\Phi$, $\langle,\rangle_{\Phi^*h}$  is the fiberwise metric on the vector bundle $\Phi^*TX\ra M$ coming from the pullback of the target metric. Using local coordinates $u^a$ on the target and local coordinates $x^i$ on the source, the action (\ref{l12 sigma model}) can be written as
\begin{equation}
S(\Phi)=\int_M \frac12 h_{ab}(\Phi)\,d\Phi^a \wedge*d\Phi^b  =
\int_M \frac12 g^{ij}(x)\, h_{ab}(\Phi(x))\,\dd_i\Phi^a(x)\, \dd_j\Phi^b(x) \,\dvol_g.
\end{equation} %\marginpar{finish the example: write EL}

The variation of the action is:
\begin{multline}\label{l12 sigma model delta S computation}
\delta S=\int_M (-1)^n \Big(
\frac12 \dd_c h_{ab}(\Phi) d\Phi^a\wedge *d\Phi^b -\underbrace{h_{ab}(\Phi) d\delta\Phi^a\wedge *d\Phi^b}_{d(h_{ab}(\Phi)\delta \Phi^a\wedge *d\Phi^b)+(-1)^n d(h_{ab}(\Phi)*d\Phi^b)\wedge \delta\Phi^a}
\Big)
\\=
\int_M \Big(\frac12 \dd_a h_{bc}(\Phi)d\Phi^b\wedge *d\Phi^c -d\Big(h_{ab}(\Phi)*d\Phi^b\Big)\Big) \wedge \delta\Phi^a +
\int_{\dd M} \Big(h_{ab}(\Phi)*d\Phi^b\Big) \wedge\delta\Phi^a\\
=\int_M \Big(
\frac12 \dd_a h_{bc}(\Phi)d\Phi^b\wedge *d\Phi^c-\dd_b h_{ac}(\Phi)d\Phi^b\wedge *d\Phi^c - h_{ab}(\Phi) d*d\Phi^b
\Big)\wedge \delta \Phi^a
+\int_{\dd M} \Big(h_{ab}(\Phi)*d\Phi^b\Big) \wedge\delta\Phi^a\\
=\int_M\Big( 
h_{ab}(\Phi)*\Delta \Phi^b-\Gamma_{abc}(\Phi)d\Phi^b\wedge *d\Phi^c
\Big)\wedge \delta \Phi^a
+\int_{\dd M} \Big(h_{ab}(\Phi)*d\Phi^b\Big)\wedge \delta\Phi^a\\
=\int_M \dvol_g h_{ab}(\Phi)(\Delta\Phi^b-\Gamma^b_{cd}(\Phi)\langle d\Phi^c, d \Phi^d\rangle_{g^{-1}})\wedge \delta\Phi^a
+\int_{\dd M} \underbrace{\Big(h_{ab}(\Phi)*d\Phi^b\Big)\wedge \delta\Phi^a}_{\ul\alpha}.
\end{multline}
Here $\Gamma^\bt_{\bt\bt}$ are the Christoffel symbols of the target metric; $\Delta=-*d*d$ is the Laplacian on $(M,g)$. Thus, the Euler-Lagrange equation is
\begin{equation}\label{l12 sigma model EL}
\Delta\Phi^a-\Gamma^a_{bc}(\Phi)\langle d\Phi^b, d \Phi^c\rangle_{g^{-1}}=0.
\end{equation}

Note that in the special case $\dim M=1$, (\ref{l12 sigma model EL}) becomes the equation of geodesic motion on $X$. 

In the other extreme case, $X=\RR$ the model becomes the massless free scalar field.

One can also consider a modification of the sigma model action functional (\ref{l12 sigma model}) by a potential:
\begin{equation}\label{l12 sigma model + potential}
S(\Phi)=\int_M \frac12 \langle d\Phi \stackrel{\wedge}{,} *d\Phi \rangle_{\Phi^*h} - V(\Phi)\,\dvol_g
\end{equation}
with $V\in C^\infty(X)$ a fixed function on the target manifold (the ``potential''). 
This modification does not change the density 
\begin{equation}
\ul\alpha=\langle *d\Phi\stackrel{\wedge}{,}\delta\Phi\rangle_{h^*\Phi}
\end{equation}
of the Noether 1-form,
as one can see by repeating the computation (\ref{l12 sigma model delta S computation})). However it changes the Euler-Lagrange equation (\ref{l12 sigma model EL}) to
\begin{equation}
\Delta\Phi^a-\Gamma^a_{bc}(\Phi)\langle d\Phi^b, d \Phi^c\rangle_{g^{-1}}- h^{ab}(\Phi) \dd_b V(\Phi)=0.
\end{equation}

\section{Symmetries and Noether currents}

%Assume that we have 
Consider an infinitesimal transformation of fields of the field theory given as  %local symmetry of the classical field theory
\begin{equation}\label{l12 inf tranformation of fields}
\phi(x) \mapsto \phi(x)+ \epsilon v(\phi(x),\dd\phi(x),\ldots) 
\end{equation}
given by a local vector field 
\begin{equation}\label{l12 v}
v\in \X_\loc(\F_M),
\end{equation}
where $\loc$ means that the infinitesimal change of the field given by $v$ at the point $x\in M$ depends only on the jet of the field at $x$.

\begin{definition}\label{l12 def: inf symmetry}
We say that $v$ is an \emph{infinitesimal symmetry} of the given classical field theory if one has 
\begin{equation}\label{l12 L_v L = d Lambda}
\LL_v L= d\Lambda
\end{equation}
for some element $\Lambda\in \Omega^{n-1,0}_\loc(M,\F_M)$; here $\LL_v$ stands for the Lie derivative in the direction of $v$.\footnote{The vector field (\ref{l12 v})  naturally induces an ``evolutionary'' (i.e. commuting with derivatives along $M$) vertical  vector field $v^\mr{evo}$ on the jet bundle $\mr{Jet}_\infty E\ra M$ (see \cite{Anderson}). 
%on jets of fields at $x\in M$ (i.e. on the fiber of the jet bundle $\mr{Jet}_\infty E\ra M$ over $x$). 
It is that latter vector field that we act with in (\ref{l12 L_v L = d Lambda}); by an abuse of notation, we still denote it $v$. Cf. Example \ref{l12 example: cl mech energy} and footnote \ref{l12 footnote: v^evo} below.
}
\end{definition}
%Equivalently,\footnote{
%When we say ``equivalently,'' we have in mind a sheaf-like picture -- that (\ref{l12 L_v S}) should hold for different $M$'s, in particular, for small open subsets of the original $M$.
%} 
%As a consequence of 
An equivalent formulation of 
(\ref{l12 L_v L = d Lambda}) is:
one has 
\begin{equation}\label{l12 L_v S}
\LL_v S_N=\int_{\dd N} \Lambda.
\end{equation}
for any   $n$-dimensional submanifold with boundary $N\subset M$.

\begin{lemma}
%\marginpar{Strengthen to if and only if? Also: talk about families of finite symmetries?}
If $v$ is a symmetry in the sense of  (\ref{l12 L_v L = d Lambda}), then the corresponding infinitesimal transformation of fields (\ref{l12 L_v L = d Lambda}) takes solutions of Euler-Lagrange equation to solutions of Euler-Lagrange equations.
\end{lemma}

\begin{proof}%\marginpar{Make it more clear? Write a second proof, within var bicomplex?}
%Assume that $v$ integrates to a flow $\Phi_\epsilon$ on $\F_M$.
 Consider a path $\phi_t$ in $\F_M$ with $\frac{d}{dt}\phi_t$ supported away from $\dd M$, as in the beginning of Section \ref{sss EL} and assume that $\phi_0=\phi$ is a solution of the Euler-Lagrange equation (\ref{l12 EL as Frechet derivative}).
Then 
%$$
%\frac{d}{dt}\Big|_{t=0}\frac{d}{d\epsilon}\Big|_{\epsilon=0}S(\Phi_\epsilon(\phi_t)) = \frac{d}{dt}\Big|_{t=0}\LL_v S(\phi_t) = \frac{d}{dt}\Big|_{t=0} \int_{\dd M} \Lambda (\phi_t) =0
%$$
\begin{equation} 
\frac{d}{dt}\Big|_{t=0}S(\phi_t + \epsilon v(\phi_t)) = \epsilon\frac{d}{dt}\Big|_{t=0}\LL_v S(\phi_t) = \epsilon\frac{d}{dt}\Big|_{t=0} \int_{\dd M} \Lambda (\phi_t) =0 \quad \bmod \epsilon^2.
\end{equation}
In the last step we used that $\dot{\phi_t}$ (and its jet) vanishes on the boundary.
Thus, any fluctuation away from the boundary of the transformed $\phi$ (the r.h.s. of (\ref{l12 inf tranformation of fields})) preserves the value of the action in the first order in fluctuation (i.e., first order in $t$).
%
%
%Then we have
%\begin{equation}
%\LL_v S=\frac{d}{dt}\Big|_{t=0} \Phi_t^*S .....
%\end{equation}
\end{proof}

\begin{definition}\label{l12 def: J}
Given an infinitesimal symmetry $v\in\X_\loc(\F_M)$ of a given classical field theory, defines an element $J\in \Omega^{n-1,0}_\loc(M,\F_M)$ by the formula
\begin{equation}\label{l12 J}
J\colon= (-1)^n \iota_v \ul{\alpha} + \Lambda.
\end{equation}
$J$ is called the ``Noether current'' associated with the symmetry $v$.
\end{definition}

\begin{thm}[Noether theorem]\label{l12 thm Noether}
The field-dependent $(n-1)$-form $J$ defined by (\ref{l12 J}) is closed on $M$ when restricted to the solutions of the Euler-Lagrange equation.
\end{thm}

%\begin{notation}
\textbf{Notation.}
For two expressions $A,B$ depending on the field  we will write 
\begin{equation}
A=B\bmod \mr{EL} \qquad \mbox{or}\qquad A\underset{\mr{EL}}{\sim} B
\end{equation}
to indicate that an equality holds ``modulo the Euler-Lagrange equation,'' i.e., when both sides are evaluated on a field configuration $\phi$ satisfying the Euler-Lagrange equation. 
%\end{notation}

Thus, Noether theorem reads
\begin{equation}\label{l12 dJ=0 mod EL}
dJ = 0 \bmod \mr{EL}.
\end{equation}

\begin{proof}[Proof of Theorem \ref{l12 thm Noether}]
Applying $d$ to the definition (\ref{l12 J}), we have
\begin{multline}
dJ = (-1)^{n+1} \iota_v \underbrace{d\ul\alpha}_{\underset{(\ref{l12 delta L= EL+d alpha})}{=}(-1)^n\delta L-\EL_a\delta\phi^a}  + \underbrace{d\Lambda}_{\underset{(\ref{l12 L_v L = d Lambda})}{= } \LL_v L}  
= \\
=-\LL_v L+\underbrace{(-1)^n \EL_a v^a}_{= 0\bmod\; \mr{EL}} + \LL_v L  =0 \bmod \mr{EL}.
\end{multline}
\end{proof}

\begin{corollary}\label{l12 corr: Noether charge conservation}
Assume we have a symmetry $v$ and $J$ is the associated Noether current. Then for $\gamma_1,\gamma_2$ two \emph{cobordant} submanifolds of $M$ of codimension $1$, one has
\begin{equation}\label{l12 Noether charge conservation}
\int_{\gamma_1}J = \int_{\gamma_2} J \quad \bmod \;\mr{EL}.
\end{equation}
\end{corollary}

\begin{proof}
Let $N\subset M$ be the cobordism between $\gamma_1$ and $\gamma_2$, i.e., $\dd N=\gamma_2-\gamma_1$. Then we have
\begin{equation}
\int_{\gamma_2}J-\int_{\gamma_1}J\underset{\mr{Stokes'}}{=}\int_N dJ \underset{\mr{Noether\;thm}}{=}0\quad \bmod \;\mr{EL}.
\end{equation}
\end{proof}

The expression $\int_\gamma J$, with $J$ a Noether current associated with a symmetry and with $\gamma\subset M$ a codimension $1$ submanifold, is called a ``Noether charge.''

Equation (\ref{l12 Noether charge conservation})  expresses the conservation property of the Noether charge (for a fixed field configuration satisfying EL), as one slides $\gamma$ along $M$.

%\textcolor{red}{PICTURE}
\begin{figure}[H]
%$$\vcenter{\hbox{ \includegraphics[scale=0.5]{cob1.eps} }} $$
\begin{center}
\includegraphics[scale=1]{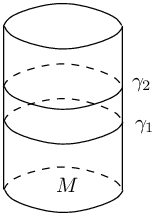}
\end{center}
\caption{Noether charge is conserved (modulo EL) when changing the hypersurface in its cobordism class $\gamma_1\ra \gamma_2$.
%evaluated on cobordant hypersurfaces $\gamma_1,\gamma_2$ is the same (modulo EL).
}
\end{figure}

%By a slight abuse of language, one calls the property (\ref{l12 dJ=0 mod EL}) the \emph{conservation property} of the current $J$ (with ``current'' meaning an element of ${\Omega^{n-1,0}_\loc(M\times \F_M)}$), in the sense that it gives rise to conserved quantities via Corollary \ref{l12 corr: Noether charge conservation}. Noether theorem gives a mechanism producing conserved currents out of infinitesimal symmetries.

\begin{definition}
One calls an element of $J\in\Omega^{n-1,0}_\loc(M\times \F_M)$ a ``conserved current'' if it satisfies
\begin{equation}
dJ\simEL 0.
\end{equation}
\end{definition}
A conserved current gives rise to a charge $\int_\gamma J$, with $\gamma\subset M$ a hypersurface, which is ``conserved'' -- independent under deformations of $\gamma$, cf. (\ref{l12 Noether charge conservation}).

Noether theorem gives a mechanism producing conserved currents out of infinitesimal symmetries.

\begin{definition}\label{l12 def: equivalent currents}
We call two conserved currents $J,J'$ ``equivalent'' if one has 
\begin{equation}
J'\simEL J+ d K
\end{equation}
for some element $K\in \Omega^{n-2,0}_\loc(M\times \F_M)$. 
\end{definition}
In particular, if $J$ is conserved ($dJ\simEL 0$), then the equivalent current $J'$ is automatically conserved. Also, if $J$ and $J'$ are equivalent, they yield the same (modulo EL) charges, \begin{equation}
\int_\gamma J \simEL \int_{\gamma} J',
\end{equation}
for $\gamma$ closed.

\begin{example}\label{l12 example: cl mech energy}
Consider the classical mechanics of a particle moving on a target  manifold $\RR$. The spacetime (source cobordism) is $M=[t_0,t_1]$, fields are maps $x\colon [t_0,t_1]\ra \RR$ and the action is 
\begin{equation}
S[x(\tau)]=\int_{t_0}^{t_1}\underbrace{ d\tau\left(m\frac{\dot{x}(\tau)^2}{2}-U(x(\tau))\right) }_L .
\end{equation}
The derivative in fields yields
\begin{equation}
\delta S = \int_{t_0}^{t_1} d\tau (m\dot{x}\delta \dot{x}-U'(x)\delta x) = \int_{t_0}^{t_1}\underbrace{d\tau(-m\ddot{x}-U'(x))\delta x}_{\EL \delta x} +\Big|\underbrace{m\dot{x}\delta x}_{\ul\alpha}\Big|^{t_1}_{t_0} .
\end{equation}
The Euler-Lagrange equation is: 
\begin{equation}\label{l12 particle EL}
m\ddot{x}+U'(x)=0.
\end{equation}
Consider the infinitesimal transformation of fields 
\begin{equation}
x(\tau)\mapsto x(\tau)-\epsilon \dot{x}(\tau).
\end{equation}
The corresponding vector field is $v=-\int_{t_0}^{t_1}\dot{x}\frac{\delta}{\delta x}$.\footnote{\label{l12 footnote: v^evo}
When acting on jets of fields at $\tau$, $v$ acts as $v^\mr{evo}=-(\dot{x}\frac{\dd}{\dd x}+\ddot{x} \frac{\dd}{\dd \dot{x}}+\dddot{x}\frac{\dd}{\dd \ddot{x}}+\cdots)$ where the superscript ``evo'' stands for ``evolutionary'' (i.e. commuting with $d_\tau$) prolongation of $v$.}
%+\ddot{x} \frac{\delta}{\delta \dot{x}}+\cdots$. \marginpar{$v$ vs. jet prolongation of $v$?}
Acting with it on $L$ yields
\begin{equation}
\LL_v L =d\tau (-m\dot{x}\ddot{x}+U'(x)\dot{x}) = d\Big(\underbrace{-m\frac{\dot{x}^2}{2}+U(x)}_\Lambda\Big) ,
\end{equation}
where $d=d\tau\frac{d}{d\tau}$ the ``horizontal'' differential. Thus, the Noether current is 
\begin{equation}\label{l12 particle H}
J=-\iota_v \ul\alpha+\Lambda = m\dot{x}^2  -m\frac{\dot{x}^2}{2}+U(x)= m\frac{\dot{x}(\tau)^2}{2}+U(x(\tau)) \quad \in \Omega^{0,0}([t_0,t_1]\times \F_{[t_0,t_1]})
\end{equation}
-- this is the ``energy'' of the particle. 

The conservation law (\ref{l12 Noether charge conservation}) says that if $\gamma_1=\{\tau_1\}$, $\gamma_2=\{\tau_2\}$ are two points on the time interval $M=[t_0,t_1]$, then, assuming the trajectory  $x(\tau)$ satisfies the Euler-Lagrange equation, the expression (\ref{l12 particle H}) yields the same number at time $\tau_1$ and at time $\tau_2$. In other words, we get that the energy (\ref{l12 particle H}) is constant in $\tau\in [t_0,t_1]$, if $[x(\tau)]$ is a solution of EL.  Of course, we can verify this statement directly: applying  $\frac{d}{d\tau}$ to the r.h.s. of (\ref{l12 particle H}), we obtain minus the l.h.s. of (\ref{l12 particle EL}).
\end{example}

\marginpar{Lecture 13,\\
9/21/22}

\begin{remark}[Noether current as a vector field] In Definition \ref{l12 def: J} we introduced the Noether current as a field-dependent $(n-1)$-form on $M$, $J\in \Omega^{n-1}(M)$, closed modulo $M$. One can consider the associated field-dependent vector field $J^\# \in \X(M)$ uniquely determined by $\iota_{J^\#}\dvol_g = J$.  Then:
\begin{itemize}
\item The conservation property $dJ\underset{EL}{\sim} 0$ corresponds in terms of the vector-field Noether current to the property
\begin{equation}
\mr{div}_{\dvol_g} J^\# \underset{EL}{\sim} 0
\end{equation}
-- the divergence of the vector field $J^\#$ w.r.t. the metric volume form on $M$ vanishes.\footnote{
Recall that to define the divergence of a vector field $u$ on a manifold $M$, one needs to specify a volume form $\mu$ on $M$. Then the divergence is defined via $ \int_M \mu u(f)= -\int_M \mu f \mr{div}_\mu(u)$ for any compactly supported test function $f$. Equivalent definition: $\mr{div}_\mu(u)=\frac{\LL_u \mu}{\mu}$.
} Equivalently, in the index notation,
\begin{equation}
\nabla_i (J^\#)^i  \underset{EL}{\sim} 0.
\end{equation}
\item The Noether charge $\int_\gamma J$ is the flux of the vector field $J^\#$ through the hypersurface $\gamma$.
\end{itemize}
\end{remark}

\begin{example}\label{l13 example massless scalar shift symmetry}
Consider the free massless scalar on a Riemannian $n$-manifold $(M,g)$, defined by the action
\begin{equation}\label{l13 massless boson S}
S(\phi)=\int_M \underbrace{\frac12 d\phi\wedge *d\phi}_L .
\end{equation}
Consider the infinitesimal transformation of fields
\begin{equation}\label{l13 phi shift}
\phi \ra \phi+\epsilon
\end{equation}
-- a shift of the value of the field $\phi$ by a constant function. The corresponding vector field on the space of fields is 
\begin{equation}
v=\int_M \frac{\delta}{\delta \phi(x)}\in \X(\F_M) .
\end{equation}
The transformation (\ref{l13 phi shift}) clearly does not change the action $S$ and the Lagrangian $L$ and clearly takes a solution of the Euler-Lagrange equation
\begin{equation}\label{l13 massless boson EL}
\Delta\phi=0
\end{equation}
to another solution. In particular, (\ref{l13 phi shift}) is a symmetry with $\Lambda=0$ (cf. Definition \ref{l12 def: inf symmetry}). Thus, the Noether current (\ref{l12 J}) corresponding to (\ref{l13 phi shift}) is
\begin{equation}
J=(-1)^n\iota_v \ul\alpha = (-1)^n \iota_v \left((-1)^{n+1}\delta\phi\wedge *d\phi\right) = \boxed{-*d\phi} \quad\in \Omega^{n-1,0}_\loc (M\times\F_M),
\end{equation}
where we used $\ul\alpha$ we obtained before, in the computation (\ref{l12 scal field EL computation}).
Noether theorem the tells that $J=-*d\phi$ is conserved (closed) modulo EL. One can check it independently:
\begin{equation}
dJ=-d*d\phi=*(\Delta \phi) \underset{EL}{\sim} 0,
\end{equation}
cf. the Euler-Lagrange equation (\ref{l13 massless boson EL}).
\end{example}

\subsection{Canonical stress-energy tensor}
\begin{example}[Canonical stress-energy tensor for the free massive scalar field]\label{l13 ex T_can}
Consider the free massive scalar field on $M=\RR^n$ (equipped with standard Euclidean metric), with the action
\begin{equation}
S(\phi)=\int_{\RR^n} \underbrace{ \frac12 d\phi\wedge *d\phi + \frac{m^2}{2}\phi^2 \dvol}_L .
\end{equation}

Consider the symmetry given by a translation on $\RR^n$:
\begin{equation}\label{l13 translations}
\begin{array}{cccc}
R_\epsilon\colon  &\RR^n &\ra &\RR^n \\
 &\vec{x}&\mapsto &\vec{x}'=\vec{x}+\epsilon \vec{a}
\end{array}
\end{equation}
with $\vec{a}\in \RR^n$ a fixed vector. This symmetry acts on fields as
\begin{equation}\label{l13 ex massive boson phi ->...}
\phi \ra (R_\epsilon^{-1})^*\phi = \phi-\epsilon a^i \dd_i\phi + O(\epsilon^2).
\end{equation}
This transformation is described by the vector field
\begin{equation}
v=-\int_{\RR^n} a^i \dd_i\phi\frac{\dd}{\dd \phi}.
\end{equation}
We have 
\begin{equation}\label{l13 ex massive boson L_v L}
\LL_v L = -a^i \frac{\dd}{\dd x^i} L = -\LL_{\vec{a}} L = d(\underbrace{-\iota_{\vec{a}}L}_\Lambda).
\end{equation}
Here the derivatives $\frac{\dd}{\dd x^i}$ act on fields, $\vec{a}$ is understood as a constant vector field on $\RR^n$. Computation (\ref{l13 ex massive boson L_v L}) in particular shows that (\ref{l13 ex massive boson phi ->...}) is indeed a symmetry, in the sense of Definition \ref{l12 def: inf symmetry}. The corresponding Noether current is
\begin{equation}
J_{\vec{a}}=(-1)^n\iota_{v}\ul\alpha +\Lambda = 
* d\phi \langle \vec{a},d\phi \rangle - \iota_{\vec{a}} \underbrace{\left(\frac12 d\phi\wedge *d\phi + \frac{m^2}{2}\phi^2 \dvol\right)}_{L} .
\end{equation}
So, one gets a family of conserved charges parametrized by $\vec{a}\in \RR^n$. This family is linear in $\vec{a}$, so it can be written as
\begin{equation}\label{l13 ex massive boson J_a}
J_{\vec{a}}=\langle \vec{a},T_\mr{can} \rangle ,
\end{equation}
where the generating object of the family%\marginpar{Write it more explicitly?}
\begin{equation}\label{l13 ex massive boson Tcan}
T_\mr{can} \in \Omega^{n-1}(M)\otimes_{C^\infty(M)}\Omega^1(M)= \Gamma(M,\wedge^{n-1}T^*M\otimes T^* M)
\end{equation}
(depending on a field) is called the ``canonical stress-energy tensor.'' In (\ref{l13 ex massive boson J_a}), the second factor in (\ref{l13 ex massive boson Tcan}) (covectors) is contracted with the constant vector field $\vec{a}$.

By Noether theorem, one has
\begin{equation}
(d\otimes \mr{id}) T_\mr{can} \underset{EL}{\sim} 0.
\end{equation}

If we switch in (\ref{l13 ex massive boson Tcan}) from $(n-1)$-forms on $M$ to vector fields on $M$ by contacting with metric volume form, we obtain the tensor 
\begin{equation}
(T_\mr{can})^\bt_{\;\;\bt} =\left(\dd^i\phi\, \dd_j \phi -\delta^i_j \left(\frac12 \dd_k\phi\, \dd^k\phi+\frac{m^2}{2}\phi^2\right)\right) \dd_i\otimes dx^j\quad \in \Gamma(M,TM\otimes T^*M).
\end{equation}
Here the bullets $(\cdots)^\bt_{\;\;\bt}$ indicate the location of indices -- the type of tensor -- once covariant and once contravariant.

\end{example}%\marginpar{Add a remark: canonical stress-energy tensor for a general field theory on flat $\RR^{p,q}$}

\begin{remark}
More generally, one can repeat the computation of Example \ref{l13 ex T_can} for any classical field theory on $M=\RR^{p,q}$ (or any full-dimensional submanifold $M$ of $\RR^{p,q}$), defined by some Lagrangian density 
\begin{equation}
L=d^nx\; \mathsf{L}(\phi,\dd \phi) .
\end{equation}
Then one obtains as the generating object for Noether currents associated with translations (\ref{l13 translations}) on $\RR^{p,q}$, the ``canonical stress-energy tensor'' 
\begin{equation}\label{l13 Tcan general}
\begin{gathered}
(T_\mr{can})^\bt_{\;\;\bt}= T^i_{\;\; j}
\dd_i\otimes dx^j
\quad \in \Gamma(M,TM\otimes T^*M)
\\ 
\mr{with}\quad
T^i_{\;\;j}=\frac{\dd\mathsf{L}(\phi,\dd\phi)}{\dd (\dd_i\phi^A)}\dd_j\phi^A - \delta^i_j \mathsf{L}(\phi,\dd\phi).
\end{gathered}
\end{equation}
By Noether theorem, it satisfies the conservation property
\begin{equation}
(\mr{div}\otimes \mr{id})(T_\mr{can})^\bt_{\;\;\bt}\simEL 0\quad \mr{or} \quad  \dd_i T^i_{\;\; j}\simEL 0.
\end{equation}
\end{remark}

\begin{remark} One can trade tangent and cotangent coefficient bundles in (\ref{l13 Tcan general}) (i.e. raise/lower indices), using the standard metric on $M=\RR^{p,q}$. In particular, one has the versions 
\begin{eqnarray}
(T_\mr{can})^{\bt\bt}=(T_\mr{can})^{ij}\;\dd_i\otimes \dd_j &\in & \Gamma(M,TM\otimes TM) ,\\
(T_\mr{can})_{\bt\bt}=(T_\mr{can})_{ij}\;dx^i\otimes dx^j &\in & \Gamma(M,T^*M\otimes T^*M) .
\end{eqnarray}
In the example of the free massive scalar field (Example \ref{l13 ex T_can}), these two versions of the canonical stress-energy tensor happen to be symmetric. However, in a general (not necessarily scalar) field theory on $\RR^{p,q}$ this fails: the canonical stress energy tensor is generally not symmetric.
\end{remark}

\section{Hilbert stress-energy tensor}
The notion of canonical stress-energy tensor above has deficiencies. Most importantly, it is only defined on a flat manifold. There is different version of the stress-energy tensor, which is defined on any manifold, is always symmetric and conserved.

\begin{definition}
Given a covariant (see (\ref{l12 covariance})) classical field theory, one defines the \emph{Hilbert stress-energy tensor} as a field-dependent tensor
\begin{equation}
T=T^{ij}\dd_i\cdot \dd_j \quad \in \Gamma(M,\mr{Sym}^2 TM)
\end{equation}
characterized by 
\begin{equation}\label{l13 delta_g S}
\delta_g S_{M,g}(\phi) = -\frac12\int_M \dvol_g T^{ij}\delta g_{ij},
\end{equation}
where $\delta_g S$ means ``variation of $S$ w.r.t. the variation of the metric $g\ra g+\delta g$.'' In other words, (\ref{l13 delta_g S}) means
\begin{equation}
\frac{d}{d\epsilon}\Big|_{\epsilon=0} S_{M,g+\epsilon h} = -\frac12\int_M \dvol_g T^{ij}h_{ij}
\end{equation}
for any fluctuation of the metric $h\in \Gamma(M,\mr{Sym}^2TM)$.

Equivalently, one defines $T$ as the variational derivative of $S_{M,g}$ w.r.t. the metric at a given point:
\begin{equation}\label{l13 T as derivative in g}
T^{ij}(x)\colon= -\frac{2}{\sqrt{\det(g)}}\, \frac{\delta S_{M,g}}{\delta g_{ij}(x)}.
\end{equation}
\end{definition}

From now on when we say ``stress-energy tensor'' we will mean ``Hilbert stress-energy tensor,'' unless stated otherwise.

Hilbert stress-energy tensor satisfies the following properties.
\begin{lemma}
\begin{enumerate}[(i)]
\item \label{l13 lm (i)}  $T$ is a symmetric tensor: $T^{ij}=T^{ji}$.
\item \label{l13 lm (ii)} $T$ is conserved, in the sense that
\begin{equation}
\nabla_i T^{ij} \underset{EL}{\sim} 0
\end{equation}
(or, in coordinate-free language, $(\mr{div}_{\dvol}\otimes \mr{id})T \underset{EL}{\sim} 0$).
\end{enumerate}
\end{lemma}

\begin{proof}
(\ref{l13 lm (i)}) is obvious by construction.

Proof of (\ref{l13 lm (ii)}): Let $r\in \X(M)$ be aby vector field vanishing in a neighborhood of the boundary; let $R_\epsilon\in \Diff(M)$ be the flow along $r$ in time $\epsilon$.  Covariance (\ref{l12 covariance}) implies 
\begin{equation}\label{l13 T conservation computation}
S_{M,g}(\phi)= S_{M,(R^{-1}_\epsilon)^* g} ((R^{-1}_\epsilon)^* \phi).
\end{equation}
Taking the derivative of both sides in $\epsilon$ at $\epsilon=0$, we get
\begin{multline}
0= \frac{d}{d\epsilon}\Big|_{\epsilon=0}S_{M,(R^{-1}_\epsilon)^* g} (\phi)+ \frac{d}{d\epsilon}\Big|_{\epsilon=0}S_{M,g} ((R^{-1}_\epsilon)^* \phi) =\\
= -\frac12 \int_M \dvol_g T^{ij} (\nabla_i r_j+\nabla_j r_i) + \underbrace{(\cdots)}_{\simEL 0} \quad \simEL\quad  2 \int_M \dvol_g (\nabla_i T^{ij}) r_j.
\end{multline}
In the last step we integrated by parts and used that $r$ vanished near $\dd M$. Since the computation (\ref{l13 T conservation computation}) holds for any $r$ supported away from $\dd M$, we get $\nabla_i T^{ij}\simEL 0$.

\end{proof}

\begin{definition}
Given a covariant classical field theory on a Riemannian manifold $(M,g)$, we say that a vector field $r\in \X(M)$ is a ``source symmetry'' (or ``spacetime symmetry,'' or ``horizontal symmetry'') if for any $n$-dimensional submanifold $N\subset M$, possibly with boundary, one has 
\begin{equation}\label{l13 source sym property}
\frac{d}{d\epsilon}\Big|_{\epsilon=0} S_{R_\epsilon(N),g}((R_\epsilon^{-1})^*\phi) = 0,
\end{equation}
where $R_\epsilon$ is the flow of $r$ in time $\epsilon$, in a neighborhood of $N$.\footnote{Note that in (\ref{l13 source sym property}) the metric is not pushed forward by  the flow in the r.h.s. If it were, the property would hold automatically for any vector field $r$ by covariance (\ref{l12 covariance}).}

Equivalently, $r$ is a source symmetry if $v=\LL_r \in \X_\loc(\F)$ is a symmetry in the sense of Definition \ref{l12 def: inf symmetry},  with $\Lambda=\iota_r L$. 
\end{definition}

For instance, in Example \ref{l13 ex T_can}, \emph{constant} vector fields on $\RR^n$ are source symmetries for the massive scalar field theory.

\begin{definition} We call an infinitesimal symmetry (\ref{l12 inf tranformation of fields}) a ``target symmetry'' (or ``vertical symmetry'') if it has the form 
\begin{equation}\label{l13 target symmetry}
\phi(x)\mapsto \phi(x)+\epsilon v(\phi(x)),
\end{equation}
with 
$v$ 
not depending on derivatives of the field, and if it is a symmetry in the sense of (\ref{l12 L_v L = d Lambda}).
\end{definition}
For instance, constant shift of the field in Example \ref{l13 example massless scalar shift symmetry} is a target symmetry. More generally, in the sigma model (\ref{l12 sigma model}), a Killing vector field on the target (infinitesimal isometry) gives rise to a target symmetry (\ref{l13 target symmetry}).

\begin{lemma}\label{l13 lm <T,r> conserved}
%\label{l13 lm (iii)}  
 Let $r%\in \X(M)
 $ be a vector field on $M$. Then $r$
 source symmetry of the theory if and only if
 %, in the sense that $v=\LL_r \in \X_\loc(\F)$ is a symmetry in the sense of Definition \ref{l12 def: inf symmetry},  with $\Lambda=\iota_r L$. 
the expression 
\begin{equation}\label{l13 J_r}
J^i_r\colon= T^{ij}r_j
\end{equation}
is a conserved current, i.e., 
\begin{equation}\label{l13 J_r conserved}
\nabla_i J^i_r\simEL 0
\end{equation} 
(or, in coordinate-free language, $\mr{div}_{\dvol} J_r\simEL 0$).
%The converse also holds: if  (\ref{l13 J_r}) 
\end{lemma}

\begin{proof}
%Let $N\subset M$ be any $n$-dimensional submanifold of $M$, possibly with boundary. Let $R_\epsilon$ be the flow of $r$ in time $\epsilon$. (We only need the flow in a neighborhood of $N$.) Then one has \marginpar{explain this point more}
%\begin{equation}
%\frac{d}{d\epsilon}\Big|_{\epsilon=0} S_{R_\epsilon(N),g}((R_\epsilon^{-1})^*\phi) = 0
%\end{equation}
Assume that $r$ is a source symmetry.
Applying covariance relation (\ref{l12 covariance}) with $m=R_\epsilon^{-1}\colon R_\epsilon(N)\ra N$ to (\ref{l13 source sym property}), we have
\begin{equation}\label{l13 lm3.23 eq1}
\frac{d}{d\epsilon}\Big|_{\epsilon=0}  S_{N, R_\epsilon^* g}(\phi) =0.
\end{equation}
By (\ref{l13 delta_g S}), this means
\begin{equation}\label{l13 (286)}
-\int_N \dvol_g \underbrace{T^{ij} \nabla_i r_j}_{\nabla_i (T^{ij} r_j) -(\nabla_i T^{ij}) r_j } =0.
\end{equation}
Since $\nabla_i T^{ij}\simEL 0$ and since (\ref{l13 (286)}) holds in particular for any small disk $N$ in $M$, we infer that 
\begin{equation}\label{l13 nabla J_r = 0}
\nabla_i(T^{ij}r_j)\simEL 0
\end{equation}
everywhere on $M$.

The converse is proven by reversing the argument: conservation of $J^i_r$ implies (\ref{l13 lm3.23 eq1}), which implies (by covariance) the source symmetry property of $r$.
\end{proof}

In particular, (\ref{l13 J_r conserved}) can be interpreted as follows: $T$ converts source symmetries of the theory into conserved currents:
\begin{equation}
r  \ra J_r= \langle T,r \rangle. 
\end{equation}

We remark that the conserved current $\langle T, r\rangle$ does not generally coincide with the conserved current associated with the source symmetry $r$ by the Noether theorem, see Example \ref{l14 ex J for source symmetry in CS} below.
%\marginpar{I think they are always equivalent though?}

\marginpar{Lecture 14,\\
9/23/2022}
\begin{example}[$T$ %Hilbert stress-energy tensor
 for the free massive scalar field] \label{l14 ex T for massive scalar}
Consider the free massive scalar field (Example \ref{l12 ex free massive scalar}). The variation of the action w.r.t. metric is
\begin{multline}\label{l14 T for massive scalar - computation}
\delta_g S_{M,g}(\phi)= S_{M,g+\delta g}(\phi)-S_{M,g}(\phi)\; \bmod\; (\delta g)^2 =\\ =
\delta_g \int_M \underbrace{\Big(\frac12 (g^{-1})^{ij} \dd_i \phi \,\dd_j \phi + \frac{m^2}{2}\phi^2\Big)}_{\mathsf{L}} \underbrace{\sqrt{\det(g)}d^n x}_{\dvol_g} \\ =
\int_M \frac12 \Big(-(g^{-1})^{ik}\delta g_{kl} (g^{-1})^{lj}\Big) \dd_i\phi\, \dd_j\phi \sqrt{\det(g)}d^n x+\\
+ %\left(\frac12 (g^{-1})^{ij} \dd_i \phi \,\dd_j \phi+\frac{m^2}{2}\phi^2\right) 
\mathsf{L}
\frac12 (g^{-1})^{kl}\delta g_{kl} \sqrt{\det (g)} d^n x\\
=-\frac12 \int_M  \sqrt{\det (g)} d^n x\,\delta g_{ij}\; \underbrace{\Big( \dd^i\phi\, \dd^j\phi-(g^{-1})^{ij}
%\Big(\frac12 (g^{-1})^{kl}\dd_k\phi\, \dd_l \phi +\frac{m^2}{2}\phi^2\Big) 
\mathsf{L}
\Big)}_{T^{ij}} .
\end{multline}
Thus, the Hilbert stress-energy tensor is:
\begin{equation}\label{l14 THilb for massive scalar}
T=T^{ij}\,\dd_i\cdot \dd_j = \left(\dd^i\phi\,\dd^j\phi-(g^{-1})^{ij}\Big(\frac12 \dd_k \phi\, \dd^k \phi+\frac{m^2}{2}\phi^2\Big)\right)\;\dd_i\cdot \dd_j,
\end{equation}
where the indices are raised using the metric, e.g., $\dd^i\phi\colon= (g^{-1})^{il}\dd_l$. In a coordinate-free language, one has
\begin{equation}
T=(d\phi)^\#\cdot (d\phi)^\# -g^{-1}\left(\frac12 \langle d\phi,d\phi \rangle)_{g^{-1}}+\frac{m^2}{2}\phi^2\right),
\end{equation}
where $(\cdots)^\#$ is the bundle map $T^*M\ra TM$ provided by the metric $g$ (``index-raising'').

We remark that the Hilbert stress-energy tensor we computed coincides (for $M=\RR^n$) with the canonical stress-energy tensor we found in Example \ref{l13 ex T_can}: (\ref{l14 THilb for massive scalar}) coincides with (\ref{l13 Tcan general}) (upon raising an index). However, in more general classical field theories it does not happen.
\end{example}

\begin{example} \label{l14 ex T for sigma model with potential}
For the sigma model with target potential (Example \ref{l12 ex: sigma model}, (\ref{l12 sigma model + potential})), the Hilbert stress-energy tensor is
\begin{equation}\label{l14 T for sigma model with potential}
T=T^{ij}\dd_i\cdot \dd_j=
\left(
h_{ab}(\Phi) \dd^i\Phi^a \dd^j \Phi^b-(g^{-1})^{ij}\left(
\frac12 h_{ab}(\Phi)\langle d\Phi^a,d\Phi^b \rangle_{g^{-1}}-V(\Phi)
\right)
\right)\dd_i\cdot \dd_j,
\end{equation}
by a computation similar to (\ref{l14 T for massive scalar - computation}).
\end{example}

\begin{example}[$T$ in the Yang-Mills theory] \label{l14 ex: T for YM}
Consider the Yang-Mills theory (Example \ref{l12 ex YM}). The variation of the action
\begin{equation}
S_{M,g}(A)=\frac12\int_M \langle F_A \stackrel{\wedge}{,} *F_A \rangle = \frac14 \int_M \sqrt{\det (g)}d^nx\; (g^{-1})^{ik}(g^{-1})^{jl}\langle (F_A)_{ij},(F_A)_{kl}\rangle
\end{equation}
with respect to the metric. The computation is similar to (\ref{l14 T for massive scalar - computation}) and yields
\begin{equation}\label{l14 T YM}
T=T^{ij}\dd_i\cdot \dd_j = \left(
\langle F^{ik}, F^{j}_{\;\; k}\rangle-\frac14 (g^{-1})^{ij} \langle F_{kl}, F^{kl}  \rangle
\right) \dd_i\cdot \dd_j ,
\end{equation}
where $F\colon=F_A$ the curvature of the connection.

\end{example}

\begin{example}[$T$ in Chern-Simons theory]
Consider Chern-Simons theory (Example \ref{l12 ex CS}). Since the action (\ref{l12 S Chern-Simons}) does not depend on the metric, Hilbert stress-energy tensor (\ref{l13 T as derivative in g}) automatically vanishes:
\begin{equation}
T=0.
\end{equation}

%\marginpar{improve the result}
We remark that the canonical (rather than Hilbert) stress-energy tensor (\ref{l13 Tcan general}) for Chern-Simons theory on $\RR^3$, is nonzero. Seen as an element of ${\Omega^2(M)\otimes_{C^\infty(M)} \Omega^1(M)}$ and then projected  to $3$-forms (i.e. skew-symmetrized), it is
\begin{equation}
[T_\mr{can}]_{\Omega^3}= -\langle A,F_A\rangle.
\end{equation}
This expression is not zero on the nose, but vanishes modulo EL.
\end{example}

\begin{example}[Noether current for source symmetries in Chern-Simons theory]\label{l14 ex J for source symmetry in CS}
In Chern-Simons theory (and generally, in any metric-independent, i.e. topological field theory) any vector field $r\in \X(M)$ is a source symmetry. The corresponding Noether current is
\begin{multline}
J_r^\mr{Noether}=-\iota_{\LL_r}\ul\alpha+\underbrace{\Lambda}_{\iota_r L}=\\
=-\iota_{\LL_r} (\frac12 \langle \delta A,A \rangle)+
\iota_r\left(\frac12\langle A,dA \rangle+\frac16 \langle A,[A,A]\rangle\right)\\
=-\frac12 \langle \LL_r A,A \rangle +
\iota_r\left(\frac12\langle A,dA \rangle+\frac16 \langle A,[A,A]\rangle\right)\\
=-\frac12 \langle d\iota_r A + \cancel{\iota_r d A}, A \rangle + \underbrace{\frac12 \langle \iota_r A,dA\rangle}_{-\frac12 \langle \iota_r A,dA\rangle+\langle \iota_r A,dA\rangle} - \cancel{\frac12 \langle A,\iota_r dA \rangle} +\frac12 \langle \iota_r A,[A,A] \rangle\\
=d(-\frac12 \langle \iota_r A,A \rangle) + \langle \iota_r A,\underbrace{dA+\frac12 [A,A]}_{\simEL 0} \rangle .
\end{multline}
So, it is a $d$-exact term plus a term vanishing modulo EL. On the other hand, the conserved current associated to $r$ by Lemma \ref{l13 lm <T,r> conserved} is identically zero, since the stress-energy tensor vanishes:
\begin{equation}
J_r=0.
\end{equation}
Note that although the currents $J_r^\mr{Noether}$ and $J_r$ are different on the nose, they are equivalent in the sense of Definition \ref{l12 def: equivalent currents}.
\end{example}

\section{Conformally invariant classical field theories}

\begin{definition}
We say that a classical field theory is ``conformally invariant'' (or just ``conformal'') if its action is invariant under Weyl transformations of the metric:
\begin{equation}\label{l14 Weyl invariance}
S_{M,g}(\phi)= S_{M,\Omega\cdot g}(\phi)
\end{equation}
for any positive function $\Omega\in C^\infty_{>0}(M)$.
\end{definition}

\begin{thm}
A classical field theory is conformally invariant if and only if its Hilbert stress-energy tensor is traceless 
\begin{equation}
\tr\, T=0.
\end{equation}
Here the trace of the stress-energy tensor is understood as 
\begin{equation}
\tr\, T= \tr\, T^\bt_{\;\;\bt}=T^i_{\,\;i}=g_{ij}T^{ij} \quad\in \Omega^{0,0}_\loc(M\times \F_M) .
\end{equation}
\end{thm}

\begin{proof}
Given a Weyl-invariant classical field theory, we have
\begin{equation}
0=\frac{d}{d\epsilon}\Big|_{\epsilon=0} S_{M,(1+\epsilon\omega)g}(\phi)=-\frac12\int_M \dvol_g \underbrace{T^{ij}(x)g_{ij}(x)}_{\tr\,T(x)} \omega(x)
\end{equation}
for any function $\omega\in C^\infty(M)$. Hence, $\tr\,T=0$.

For the ``only if'' part: given that $\tr T=0$ we have by the same computation (read right-to-left) that $S$ is invariant under infinitesimal Weyl transformations. Since Weyl orbits are path connected, this implies the full Weyl-invariance property (\ref{l14 Weyl invariance}).
\end{proof}

Weyl-invariance (\ref{l14 Weyl invariance}) implies (via covariance)  that every \emph{conformal} vector field $r\in\conf(M)$  is a source symmetry. In particular, by Lemma \ref{l13 lm <T,r> conserved}, 
\begin{equation}\label{l14 J_r=<T,r>}
J_r\colon=\langle T,r\rangle
\end{equation}
is a conserved current for any conformal vector field $M$.

Given a conformally invariant classical field theory, the stress-energy tensor depends on the metric (in addition to its dependence on fields) and is generally not Weyl-invariant. However, it is Weyl-equivariant. More precisely, we have
\begin{eqnarray}
\label{l14 T Weyl equivariance} (T_{\Omega g})^{\bt\bt}&=&\Omega^{-1-\frac{n}{2}}(T_g)^{\bt\bt},\\
\label{l14 T Weyl equivariance 2}
(T_{\Omega g})_{\bt\bt}&=&\Omega^{1-\frac{n}{2}}(T_g)_{\bt\bt} ,
\end{eqnarray}
where we are indicating the background metric as a subscript; $n$ is the dimension of the spacetime manifold $M$.

Indeed, to see (\ref{l14 T Weyl equivariance}), 
we compute
\begin{equation}
\begin{CD}
S_{\Omega(g+\delta g)} @= \displaystyle -\frac12\int \underbrace{\sqrt{\det (\Omega\, g)}}_{\Omega^{\frac{n}{2}}\sqrt{\det(g)}}\, d^n x\,T^{ij}_{\Omega g} \Omega\delta g_{ij}\\
@| \\
S_{g+\delta g} @= \displaystyle -\frac12\int \sqrt{\det (g)}\, d^n x\,T^{ij}_g \delta g_{ij}
\end{CD}
\end{equation}
which immediately implies (\ref{l14 T Weyl equivariance}). We get (\ref{l14 T Weyl equivariance 2}) by contracting (\ref{l14 T Weyl equivariance}) with two copies of $\Omega g$.

In particular, (\ref{l14 T Weyl equivariance 2}) implies that for a \emph{2-dimensional} conformally invariant classical field theory the stress-energy tensor $T_{\bt\bt}$ is Weyl-invariant -- depends only on the conformal class of the metric $g$.

\begin{example}
Consider again the massive scalar field (Examples \ref{l12 ex free massive scalar}, \ref{l14 ex T for massive scalar}), $S=\int\frac12 d\phi\wedge*d\phi+\frac{m^2}{2}\phi^2\dvol_g$. Using (\ref{l14 THilb for massive scalar}), the trace of the stress-energy tensor is
\begin{equation}
\tr\,T= \frac{2-n}{2} \dd^i\phi \,\dd_i\phi-n\frac{m^2}{2}\phi^2.
\end{equation}
The only way this expression can be identically zero is if $n=2$ and $m=0$. I.e., \emph{only the 2d massless scalar field theory is conformally invariant} (among all scalar field theories in different dimensions and with different masses).

Another way to see this is to look directly at the action where one performs a Weyl transformation with the metric:
\begin{equation}\label{l14 massive scalar Wey transf of metric}
S_{M,\Omega g}(\phi)=\int_M \Omega^{\frac{n}{2}-1}\frac12 d\phi\wedge *_g d\phi + \Omega^{\frac{n}{2}}\frac{m^2}{2}\dvol_g .
% =S_{M,g}(\phi) \quad \mr{iff}\quad \left\{ 
%\begin{array}{l}
%n=2,\\
%m=0
%\end{array}
%\right.
\end{equation}
It is independent of $\Omega$ (and coincides with $S_{M,g}(\phi)$) if and only if $n=2$ (which makes the first term $\Omega$-independent) and $m=0$ (which makes the second term $\Omega$-independent). Here we made use of the fact that  the Hodge star behaves under Weyl transformations as 
\begin{equation}
*_{\Omega g}\alpha=\Omega^{\frac{n}{2}-p}*_g\alpha,
\end{equation} 
where $\alpha\in \Omega^p(M)$ is any $p$-form.
\end{example}

\begin{example}
Similarly to the previous example, the sigma model (Example \ref{l12 ex: sigma model}, \ref{l14 ex T for sigma model with potential}) is conformally invariant if and only if $n=2$ and the potential $V(\Phi)$ is zero.
\end{example}

\begin{example}
Consider again the Yang-Mills theory (Examples \ref{l12 YM}, \ref{l14 ex: T for YM}). The trace of the stress-energy tensor (\ref{l14 T YM}) is
\begin{equation}
\tr\,T= \frac{n-4}{4} \langle F_{ij},F^{ij} \rangle .
\end{equation}
Thus, $n$-dimensional Yang-Mills theory is conformally invariant if and only if  $n=4$.

Another way to see this is via a computation similar to (\ref{l14 massive scalar Wey transf of metric}):
\begin{equation}
S_{M,\Omega g}(A)=\frac12 \int_M  \Omega^{\frac{n}{2}-2} \langle F_A\stackrel{\wedge}{,}*_g F_A\rangle.
\end{equation}
This expression is independent of $\Omega$ (and coincides with $S_{M,g}(A)$) if and only if the power of $\Omega$ in the integrand vanishes, i.e., if $n=4$.
\end{example}

\begin{example} 3d Chern-Simons theory (Example \ref{l12 ex CS}) is conformally invariant a fortiori: stress-energy tensor vanishes identically, in particular its trace vanishes. Put another way, the model does not depend on metric, hence it is invariant under Weyl transformations of metric.
\end{example}

\begin{comment}
\subsection{Noether currents for conformal space time symmetries}

{\bf in coordinates}

The Noether current  for a vector field $r$  representing infinitesimally a family of  diffeomorphisms of $M$
Noether current is defined as 
\[
J_r=\sum_{i}J^i_r\pa_i=\sum_{ij}J^i_jr^j\pa_i
\]
For a closed $(n-1)$-dimensional submanifold $C\subset M$, the flux of $J_r$  through $C$ is  
\[
J_r(C)=\int_C J^i_r  
{\iota}_{\pa_i} (\sqrt g )
\]
 Assume that vector field $r$ is conformal,  i.e.  
 \[
 \delta_r g_{ij}= \nabla_i r_j+\nabla_j r_i=\omega g_{ij}
 \]
 and that  our theory is conformally  invariant.
 
 Let $C^\prime$ be a continuous deformation of $C$ and $\Sigma (C, C^\prime)$ contractible
 cobordism between $C$ and $C^\prime$
\begin{center}
\includegraphics[width=60mm]{fig3.jpg}
\end{center}
 The difference between fluxes through $C$ and $C^\prime$ is
 
 \beqan
J_r(C)-J_r(C^\prime )&=&\int_{\Sigma(C,C^\prime)} \nabla _i J^i_r \sqrt g d^nx\\
		&=&\int_{\Sigma(C,C^\prime)} \nabla _i (T^{ij}r_j) \sqrt g \\
		&=&\int_{\Sigma(C,C^\prime)}  (\nabla _i T^{ij}r_j +\frac{1}{2} T^{ij}(\nabla_i r_j +\nabla_j r_i))\sqrt g\\
		&=&\int_{\Sigma(C,C^\prime)} T^{ij}_w g_{ij} \sqrt g =0.
\eeqan
We used  $\nabla_i T^{ij}=0$. and $T^{ij}g_{ij}=0$ for conformally invariant theories. 

Thus in conformally invariant theories 
%\begin{framed}
\[
\nabla_i J^i_r=0.
\]
%\end{framed}
\end{comment}

\section{2d classical conformal field theory
%on $\RR^2\simeq \CC$
}
Consider a conformal classical field theory on a Riemann surface $\Sigma$.
%with $M$ being $\RR^2$ (with standard Euclidean metric), or a 2-dimensional submanifold of it. In fact, since any Riemann surface $\Sigma$ is locally conformally equivalent to an open set in $\RR^2\simeq \CC$, we are talking about the local structure of a general 2d classical conformal field theory.

\subsection{Stress-energy tensor in local complex coordinates}
Let us use local coordinates $x^1=x, x^2=y$ in which the conformal structure is represented by the metric $(dx)^2+(dy)^2$.
Symmetry and tracelessness of the stress-energy tensor implies the ansatz
\begin{equation}
T_{ij}=\left(
\begin{array}{cc}
T_{11} & T_{12} \\ T_{12} & -T_{11}
\end{array}
\right)
\end{equation}
%in the standard real coordinates $x^1=x$, $x^2=y$.
Thus, there are two independent components $T_{11}$, $T_{12}$
The conservation property $\dd^i T_{ij}\simEL 0$ is tantamount to
\begin{eqnarray}
\dd_1 T_{11}+\dd_2 T_{12} & \simEL & 0, \label{l14 T conservation eq1}\\
\dd_1 T_{12}-\dd_2 T_{11} & \simEL & 0. \label{l14 T conservation eq2}
\end{eqnarray}

Switching to local complex coordinates $z,\bar{z}$, we have
\begin{multline}\label{l14 T in cx coords}
T_{\bt\bt}=T_{ij}dx^i dx^j = T_{11}\underbrace{(dx^2-dy^2)}_{\frac12(dz^2+d\bar{z}^2)}+T_{12}\underbrace{2dx\,dy}_{\frac{1}{2i}(dz^2-d\bar{z}^2)} =\\ =
\underbrace{\frac{T_{11}-iT_{12}}{2}}_{T_{zz}}(dz)^2+
\underbrace{\frac{T_{11}+iT_{12}}{2}}_{T_{\bar{z}\bar{z}}}(d\bar{z})^2 = T_{zz}(dz)^2+T_{\bar{z}\bar{z}}(d\bar{z})^2.
\end{multline}
Thus, $T$ is a sum of a quadratic differential and its complex conjugate.
Note that the the mixed term $T_{z\bar{z}}dz d\bar{z}$ does not appear. In fact, this property is equivalent to conformal invariance, since\footnote{
Since the metric is $g=dz \cdot d\bar{z}$, the inverse metric is $g^{-1}=4 \dd_z\cdot \dd_{\bar{z}}$; the matrix  of the latter in $z,\bar{z}$-coordinates is 
$(g^{-1})^{ij}=\left( 
\begin{array}{cc}
0&2\\2&0
\end{array}
\right)$. Hence, $\tr\,T=(g^{-1})^{ij}T_{ij}=(g^{-1})^{z\bar{z}}T_{z\bar{z}}+(g^{-1})^{\bar{z}z}T_{\bar{z}z}=2T_{z\bar{z}}+2T_{\bar{z}z}=4T_{z\bar{z}}$. 
%FINISH
}
\begin{equation}
T_{z\bar{z}}=\frac{1}{4}\tr\,T \underset{\mr{conformal\;invariance}}{=} 0.
\end{equation}

Conservation property (\ref{l14 T conservation eq1}), (\ref{l14 T conservation eq2}) in the complex coordinates reads
\begin{equation}\label{l14 T conservation cx coord}
\dd_{\bar{z}}T_{zz}\simEL 0,\quad \dd_z T_{\bar{z}\bar{z}}\simEL 0.
\end{equation}
So, modulo EL, $T_{zz}$ is a holomorphic function and $T_{\bar{z}\bar{z}}$ is antiholomorphic. Thus, modulo EL, the stress-energy tensor (\ref{l14 T in cx coords}) is a sum of a holomorphic quadratic differential 
\begin{equation}
T_{zz}(z)(dz)^2
\end{equation}
and its complex conjugate. %\marginpar{remark: it is no coincidence that holomorphic quadratic differentials parametrize the cotangent space of $\MM_{g,n}$}

\begin{remark}
%It is no coincidence that we see that the stress-energy tensor in a 2d conformal field theory is a quadratic differential which becomes holomorphic modulo EL (plus complex conjugate)
Note that holomorphic quadratic differentials arise
\begin{itemize}
\item as a parametrization of the cotangent space to the moduli space of complex structures on a surface (cf. (\ref{l10 T^*J (moduli space)})),
\item as a component of the stress-energy in a 2d conformal classical field theory.
\end{itemize}
%Recall that previously we encountered holomorphic quadratic differentials as parametrizing the cotangent space to the moduli space of complex structures on a surface. 
These two occurrences are related:
the variation $\delta_g S_{\Sigma,g}(\phi)$ is (for a fixed field $\phi$) a cotangent vector to the space of metrics, but due to Weyl-invariance it descends to a cotangent vector to the space of conformal (or complex) structures on the surface. 

Next, if the field $\phi$ satisfies the Euler-Lagrange equation, then for $\psi_t\in \mr{Diff}(\Sigma)$ the flow of some (not necessarily conformal) vector field $u$ on $\Sigma$, one has 
\begin{equation}
\frac{d}{dt}\Big|_{t=0} S_{\Sigma,\psi_t^*g}(\phi)\underset{\mr{covariance}}{=}\frac{d}{dt}\Big|_{t=0}S_{\Sigma,g}((\psi^{-1}_t)^*\phi) %\underset{EL}{=}
\simEL
0.
\end{equation}
\end{remark}
Thus, for $\phi$ satisfying EL, $\delta_g S_{\Sigma,g}(\phi)$ actually gives a cotangent vector to the Teichm\"uller space 
\begin{equation}
\mc{T}_\Sigma=\{\mr{conformal\;structures}\}/\{\mr{action\;by\;vector\;fields}\}
\end{equation} 
and hence to the moduli space of complex structures $\MM_\Sigma$.

More explicitly, consider an infinitesimal deformation of the metric 
\begin{equation}\label{l14 variation of g}
g=dz d\bar{z} \mapsto g+\delta g = g(1+\mu+\bar\mu)(1+\omega) = (1+\omega)dz d\bar{z}+
\bar\mu^{\bar{z}}_z (dz)^2+ \mu^z_{\bar{z}}(d\bar{z})^2 
\end{equation}
with $\mu,\bar\mu$ the infinitesimal Beltrami differentials and $\omega$ the infinitesimal Weyl factor (this is an infinitesimal version of (\ref{l9 g via Beltrami})). Then one has
\begin{equation}\label{l14 delta_g S via Beltrami}
\delta_g S(\phi) = -2\int d^2z (T_{zz}\mu^z_{\bar{z}}+T_{\bar{z}\bar{z}} \bar\mu^{\bar{z}}_z),
\end{equation}
as a consequence of (\ref{l13 delta_g S}). The r.h.s. of (\ref{l14 delta_g S via Beltrami}) is invariant (modulo EL) under shifts (\ref{l10 mu -> mu + dbar v}) of the Beltrami differential, due to the conservation property of the stress-energy tensor (\ref{l14 T conservation cx coord}).

\marginpar{Lecture 15,\\ 9/26/2022}

\subsection{Conserved currents and charges associated to conformal symmetry}
Given a conformal vector field on $\Sigma$
\begin{equation}
r=u(z)\dd_z+\bar{u}(\bar{z})\dd_{\bar z} \in \conf(\Sigma)
\end{equation}
(which is automatically a source symmetry for a conformal field theory), the associated conserved current (\ref{l14 J_r=<T,r>}) is
\begin{equation}\label{l15 J_r}
J_r=\langle T,r \rangle = u T_{zz} dz + \bar{u} T_{\bar{z}\bar{z}} d\bar{z}
\end{equation}
-- a field-dependent $1$-form on $\Sigma$, which is closed modulo EL. Indeed, we see from (\ref{l14 T conservation cx coord}) that
\begin{equation}
d(uT_{zz}dz)= \bar\dd (u T_{zz} dz )= \dd_{\bar{z}}(u(z) T_{zz}) d\bar{z}\wedge dz =u(z)(\dd_{\bar{z}}T_{zz}) d\bar{z}\wedge dz \simEL 0
\end{equation}
and similarly for the second term in (\ref{l15 J_r}).

Given a closed loop $\gamma\subset \Sigma$, one has the corresponding conserved charge is 
\begin{equation}
C_{r,\gamma}\colon=\oint_\gamma  J_r
\end{equation}
and it is invariant under deformations of $\gamma$ modulo EL -- by Stokes' theorem (as an integral of a closed 1-form), or equivalently by Cauchy theorem (as an integral of a holomorphic 1-form, plus complex conjugate).
%(since the integrand is closed mod EL).

%\begin{example}[Massless scalar field on a Riemann surface]
\subsection{Example: massless scalar field on a Riemann surface}
%Fix a Riemann surface $\Sigma$.
Fields are smooth real-valued functions $\phi\in C^\infty(M)$. 
The action written in real  local coordinates on $\Sigma$ reads
\begin{equation}\label{l15 massless scalar real coords}
S(\phi)=\int_\Sigma \sqrt{\det(g)}  dx\wedge dy \; \frac12(g^{-1})^{ij}\dd_i\phi\, \dd_j\phi.
% = \int_\Sigma \frac{i}{2} dz\wedge d\bar{z}\;2\dd_z\phi \,\dd_{\bar{z}}\phi.
\end{equation}
Here $g$ can be any metric withing the given conformal class (the combination $\sqrt{\det(g)}(g^{-1})^{ij}$ is Weyl-invariant). 

Written in local complex coordinates $z,\bar{z}$ on $\Sigma$, the action reads
\begin{equation}\label{l15 massless scalar cx coords}
S(\phi)=\int_\Sigma \frac{i}{2} dz\wedge d\bar{z}\;2\dd_z\phi \,\dd_{\bar{z}}\phi.
\end{equation}
To see this, it is sufficient to consider the Lagrangian density in (\ref{l15 massless scalar real coords}) in the standard real/complex coordinates on the standard $\RR^2\simeq \CC$, since this locally describes the general surface. In the standard metric one has $dz\wedge d\bar{z}=(dx+idy)\wedge (dx-idy)=-2i dx\wedge dy$, thus $dx\wedge dy=\frac{i}{2} dz\wedge d\bar{z}$. Also, one has
$\frac12((\dd_x \phi)^2+(\dd_y \phi)^2)=2 \dd_z \phi \dd_{\bar{z}}\phi$. This proves (\ref{l15 massless scalar cx coords}).

The Euler-Lagrange equation reads 
\begin{equation}\label{l15 EL real coord}
\Delta\phi =0 
\end{equation}
or equivalently, in complex coordinates,
\begin{equation}
\dd_z\dd_{\bar{z}}\phi =0.
\end{equation}
I.e., $\phi$ satisfies EL if and only if it is a harmonic function on $\Sigma$. We remark that although the Laplacian 
\begin{equation}
\Delta=\frac{1}{\sqrt{\det(g)}} \dd_i \sqrt{\det(g)} (g^{-1})^{ij}\dd_j
\end{equation}
itself is not a Weyl-invariant operator on a surface (it changes under Weyl transformations as $\Delta_{\Omega g}=\Omega^{-1}\Delta_g$ on a 2d manifold), the equation (\ref{l15 EL real coord}) is Weyl-invariant.

The components of the stress-energy tensor in complex coordinates read
\begin{equation}
T_{zz}=(\dd_z\phi)^2,\qquad T_{\bar{z}\bar{z}}=(\dd_{\bar{z}}\phi)^2.
\end{equation}

\chapter{2d quantum free massless scalar field}

%\marginpar{To edit}
In this section our goal is to study the 2d free massless scalar field (a.k.a. free boson) as a quantum conformal field theory (in Euclidean signature): the space of states for the circle $\HH$, correlation functions  on the plane
%, partition function of a torus, Virasoro algebra action on $\HH$
 and the operator product expansions. 

The logic of the approach is as follows:
\begin{enumerate}[(i)]
\item We start by constructing the quantum theory on a Minkowski cylinder (via canonical quantization of the classical theory) -- along the way we identify the space of states for the circle. 
%To start ourselves on this way, 
As a warm-up, we start with the quantization of a simple 1d system -- the harmonic oscillator; as we will see, the free scalar field on a cylinder can be represented (via Fourier transform on $S^1$) as a tensor product of a family of harmonic oscillators.
\item  
We switch from Minkowski to Euclidean metric on the cylinder by Wick's rotation. Then we identify -- via the exponential map -- the Euclidean cylinder with the punctured complex plane $\CC^*$. At this point we are ready to calculate correlation functions of several point observables on $\CC$.
\end{enumerate}

\section{A warm-up: harmonic oscillator}

\subsection{Harmonic oscillator as a classical mechanical system}

In classical mechanics, in Hamiltonian formalism, the harmonic oscillator is a system with the phase space 
\begin{equation}\label{l15 Phi}
\Phi=T^*\RR
\end{equation}
seen as a symplectic vector space, with symplectic form
\begin{equation}
\omega_\mr{symp}=dp\wedge dx
\end{equation}
where $x$ is the coordinate on $\RR$ and $p$ -- the coordinate on the cotangent fiber. The symplectic form equips the algebra of smooth functions $C^\infty(\Phi)$ with the Poisson bracket
\begin{equation}
\{-,-\}\colon \Phi\times \Phi \ra \Phi
\end{equation}
-- a skew-symmetric bilinear (over $\RR$) operation which is a derivation in either slot and satisfies the generating relation
\begin{equation}
\{p,x\}=1.
\end{equation}

A more geometric definition of the Poisson bracket (valid for any symplectic manifold $(\Phi,\omega_\mr{symp})$) is:
\begin{itemize}
\item For each function $f\in C^\infty(\Phi)$, there is the corresponding Hamiltonian vector field $X_f\in \X(\Phi)$ uniquely characterized by the property
\begin{equation}
\iota_{X_f}\omega_\mr{symp}=-df.
\end{equation}
\item The Poisson bracket is defined by
\begin{equation}
\{f,g\}\colon= X_f(g)
\end{equation}
for any $f,g\in C^\infty(\Phi)$.
\end{itemize}

Back to the harmonic oscillator:
the phase space $\Phi$ is equipped with the function
\begin{equation}\label{l15 H}
H=\frac{p^2}{2}+\omega^2 \frac{x^2}{2}\quad \in C^\infty(\Phi)
\end{equation}
-- the classical Hamiltonian; here $\omega>0$ is a parameter of the system (``frequency''). The function $H$ generates the Hamiltonian vector field 
\begin{equation}
X_H=\{H,-\}=p\frac{\dd}{\dd x}-\omega^2 x\frac{\dd}{\dd p}.
\end{equation}

Hamilton's equations of motion of the system is the equation of an integral curve of the vector field $X_H$ on $\Phi$. In the case of the oscillator, they are:
%begin{equation}
\begin{eqnarray}
\dot{x}=\{H,x\}&=&p, \label{l15 Ham eqs for oscillator eq1}\\
\dot{p}=\{H,p\} &=& -\omega^2 x. \label{l15 Ham eqs for oscillator eq2}
\end{eqnarray}
%\end{equation}
Solving this system is straightforward: one combines this system to the single equation on $x$
\begin{equation} \label{l15 ddot(x)+omega^2 x=0}
\ddot{x}+\omega^2 x =0
\end{equation}
which has general solution $x(t)=A \cos(\omega t)+B \sin(\omega t)$ -- oscillatory motion with frequency $\omega$ and $A,B$ arbitrary parameters. Then one uses (\ref{l15 Ham eqs for oscillator eq1}) to find $p(t)$.

In Lagrangian mechanics, the same system is described by space of fields 
\begin{equation}
\F_{[t_0,t_1]}=\mr{Map}([t_0,t_1],\RR)
\end{equation}
-- maps from the source (or ``worldline'') interval $[t_0,t_1]$ to the target $\RR$ (the base of the cotangent bundle (\ref{l15 Phi})). The action for a function $x(\tau)$ is
\begin{equation}\label{l15 S}
S[x(\tau)]=\int_{t_0}^{t_1} d\tau \left( \frac{\dot{x}^2}{2}-\frac{\omega^2}{2}x^2 \right)
\end{equation}
The corresponding Euler-Lagrange equation is exactly the equation (\ref{l15 ddot(x)+omega^2 x=0}). Thus, indeed, the Euler-Lagrange equations for the action (\ref{l15 S}) are equivalent to the Hamilton's equations corresponding to the Hamiltonian (\ref{l15 H}). %-- so these are two equivalent descriptions of the same system.
%Thus, indeed, the motion described by the action functional (\ref{l15 S}) 

\subsection{Correspondence between Lagrangian and Hamiltonian descriptions of classical mechanics.}
%\begin{remark}
%Generally, 
Stepping aside from the harmonic oscillator for the moment,
consider the general classical mechanical system in Lagrangian formalism, with fields 
\begin{equation}
\F_{[t_0,t_1]}=\mr{Map}([t_0,t_1],X),
\end{equation}
with $X$ some target manifold, and with action functional
\begin{equation}\label{l15 S cl mech general}
S[x(\tau)]=\int_{t_0}^{t_1} d\tau\; \mathsf{L}(x(\tau),\dot{x}(\tau)),
\end{equation}
where 
\begin{equation}\label{l15 L as fun on TX}
\mathsf{L}(x,v)\in C^\infty(TX)
\end{equation}
is some function on the tangent bundle of the target $X$; here $v\in T_x X$ is a tangent vector. Then the Euler-Lagrange equation is
\begin{equation}\label{l15 EL for general L}
\frac{\dd\mathsf{L}(x,\dot{x})}{\dd x^i}-\frac{d}{dt}\left(\frac{\dd\mathsf{L}(x,v)}{\dd v^i}\Big|_{v=\dot{x}}\right)=0
\end{equation}
-- an ODE on the map $x\colon [t_0,t_1]\ra X$; here we use local coordinates $x^i$ on $X$.

The same system can be described as a Hamiltonian system with the phase space
\begin{equation}
\Phi=T^* X
\end{equation}
--  the cotangent bundle of the target $X$ equipped with the canonical symplectic form of the cotangent bundle, $\omega_\mr{symp}=dp_i\wedge dx^i$. The Hamiltonian function $H\in C^\infty(\Phi)$ is obtained as the \emph{Legendre transform} of the Lagrangian $\mathsf{L}$, trading velocity $v$ for momentum $p$:
\begin{equation}\label{l15 Legendre 1}
H(x,p)\colon = v^ip_i-\mathsf{L}(x,v),
\end{equation}
where $v=v(x,p)$ determined implicitly by the equation
\begin{equation}\label{l15 Legendre 2}
p_i=\frac{\dd\mathsf{L}(x,v)}{\dd v^i}.
\end{equation}
For the Legendre transform to exist and be invertible, one needs $\mathsf{L}(x,v)$ to be a convex function in $v$ (for any $x$).

The key observation is that the Hamiltonian equations generated by $H$ and Euler-Lagrange equations determined by the action (\ref{l15 S cl mech general}) are equivalent, provided that the Lagrnagian $\mathsf{L}$ and the Hamiltonian $H$ are linked by the Legendre transform (\ref{l15 Legendre 1}), (\ref{l15 Legendre 2}). Indeed, the Hamiltonian equations read
\begin{multline}
\dot{x}^i=\frac{\dd H}{\dd p_i}\underset{(\ref{l15 Legendre 1})}{=} v^i+\cancel{p_j \frac{\dd v^j}{\dd p_i}}-\cancel{\frac{\dd v^j}{\dd p_i}\underbrace{ \frac{\dd \mathsf{L}}{\dd v^j}}_{p_j}}=v^i,
\\
\dot{p}_i=-\frac{\dd H}{\dd x^i} =
-p_j \frac{\dd v^j}{\dd x^i} + \frac{\dd \mathsf{L}}{\dd x^i}\Big|_{p=\mr{const}}
=\cancel{-p_j \frac{\dd v^j}{\dd x^i}}+\Big(\frac{\dd \mathsf{L}}{\dd x^i}\Big|_{v=\mr{const}}+\cancel{\underbrace{\frac{\dd\mathsf{L}}{\dd v^j}}_{p_j}\frac{\dd v^j}{\dd x^i}}\Big) = \frac{\dd \mathsf{L}}{\dd x^i}.
\end{multline}
Substituting (\ref{l15 Legendre 2}) in the second equation above, we get the Euler-Lagrange equation (\ref{l15 EL for general L}).
%\end{remark}

\begin{remark} Legendre transform admits the following geometric description. If $l(v)$ is convex\footnote{Convexity implies that the Lagrangian submanifold (\ref{l15 Legendre graph}) is projectable onto both $V$ and $V^*$.
} smooth function on a vector space $V\ni v$  then its Legendre transform $h(p)$ is a smooth convex function on $V^*\ni p$ with the property that the Lagrangian submanifold   $V\oplus V^*$ that is the graph of $dl$ (here we think of $V\oplus V^*$ as the cotangent bundle $T^* V$ with the standard symplectic form $dp_i\wedge dv^i$) is also described as the graph of $dh$ (where we think of $V\oplus V^*$ as $T^*(V^*)$ with symplectic structure $dv^i\wedge dp_i$): 
%where the bar means that one changes the sign of the canonical symplectic structure of the cotangent bundle):
\begin{equation}\label{l15 Legendre graph}
\mr{graph}(dl) = \left\{(v,p)\;|\; p_i= \frac{\dd l}{\dd v^i}\right\} \quad  = \quad \mr{graph}(dh) =\left\{(v,p)\;|\; v_i=\frac{\dd h}{\dd p_i}\right\}\qquad \subset V\oplus V^*
\end{equation}
Put another way, the Lagrangian submanifold (\ref{l15 Legendre graph}) has $l$ as its generating function on $V$ and $h$ as its generating function on $V^*$. If $l$ is given, the property (\ref{l15 Legendre graph}) determines $h$ uniquely up to a possible shift by a constant function. %\marginpar{FINISH}

In (\ref{l15 Legendre 1}), (\ref{l15 Legendre 2}), the Legendre transform is done pointwise on $X$, with $V=T_x X$, $V^*=T^*_x X$, $l(v)=\mathsf{L}(x,v)$ and $h(p)=H(x,p)$ for any point $x\in X$.
\end{remark}

\subsection{Preparing for canonical quantization: Weyl algebra and Heisenberg Lie algebra}\label{sss Weyl and Heisenberg}
\begin{definition}\label{l15 def: Heis}
Let $(V,\omega_\mr{symp})$ be a (real) symplectic vector space and let $V_\CC=\CC\otimes V$ be its complexification. One defines  the 
 %Weyl algebra as the associative algebra
Heisenberg Lie algebra associated to $(V,\omega_\mr{symp})$ as the Lie $*$-algebra
%the  Lie algebra  
\begin{equation}
\mr{Heis}(V,\omega_\mr{symp})=  V_\CC\oplus \CC\cdot \mathbb{K}
\end{equation}
where $\mathbb{K}$ is a central element and one has the  commutators
\begin{equation}\label{l15 Heis comm rel}
[\wh{u},\wh{v}]=%-i\hbar 
-i\omega_\mr{symp}(u,v)\cdot\mathbb{K}
\end{equation}
for $u,v\in V$. We put a hat on an element of $V$ when we think of it as an element of $\mr{Heis}$. 
Elements $\wh{v}$ and $\mathbb{K}$ are understood as self-adjoint. 
%Here $\hbar$ is a parameter (one can think of it as a fixed positive number, or consider a family over $\hbar\ra 0$ and think of it as a formal parameter.)
\end{definition}

Thus, Heisenberg Lie algebra is a central extension of $V_\CC$ seen as an abelian Lie algebra, 
\begin{equation}
\CC \ra \mr{Heis}(V,\omega_\mr{symp})\ra V_\CC,
\end{equation}
with the Lie 2-cocycle of $V$ defining the central extension being $\omega_\mr{symp}$.

\begin{thm}[Stone-von Neumann] Assume that $V$ is finite-dimensional. Then there exists a unique (up to isomorphism) irreducible unitary representation of $\mr{Heis}(V,\omega_\mr{symp})$. 
\end{thm}
``Unitary'' here means that the representation is on a Hilbert space $\mc{H}$ and for each $v\in V$, $\wh{v}$ is represented by a hermitian operator. 
%and also $\mathbb{K}$ is represented by $\hbar\cdot\mr{Id}$ for some $\hbar>0$.

%\marginpar{$\ra$ star-algebra}
\begin{definition}\label{l15 def Weyl}
Weyl algebra of the symplectic vector space $(V,\omega)$ is defined as the following associative $*$-algebra over the ring of formal power series $\CC[[\hbar]]$:
\begin{equation}
\mr{Weyl}(V,\omega_\mr{symp})\colon= \CC[[\hbar]]\otimes U \mr{Heis}(V,\omega_\mr{symp})/(\mathbb{K}=\hbar)
\end{equation}
-- the universal enveloping of the Heiseberg Lie algebra (with scalars extended to formal power series), with the central element $\mathbb{K}$ identified with the scalar $\hbar$.
The involution (hermitian conjugation) maps the Heisenberg generators $\wh{v}$ to themselves.
\end{definition}

%Here $\hbar$ is a parameter (one can either think of it as a fixed positive number, or consider a family over $\hbar\ra 0$ and think of it as a formal parameter.)
Here we think of the Planck constant $\hbar$ as an infinitesimal formal parameter.

%\begin{thm}[Stone-von Neumann] Assume that $V$ is finite-dimensional. Then there exists a unique (up to isomorphism) irreducible unitary representation of $\mr{Heis}(V,\omega_\mr{symp})$. 
%\end{thm}
%``Unitary'' here means that the representation is on a Hilbert space $\mc{H}$ and for each $v\in V$, $\wh{v}$ is represented by a hermitian operator.

\begin{example}[Main example]
Consider  $V=T^*\RR^n$ with coordinates $x_1,\ldots, x_n$ on the base $\RR^n$ and dual fiber coordinates $p^1,\ldots, p^n$, with standard symplectic form
\begin{equation}
\omega_\mr{symp}=\sum_i dp_i \wedge dx^i.
\end{equation}
The corresponding Weyl algebra is generated by elements $\wh{x}^i, \wh{p}_i$, $i=1,\ldots,n$, subject to relations
\begin{equation}\label{l15 canonical commutation relations}
[\wh{x}^i,\wh{x}^j]=0,\quad [\wh{p}_i,\wh{p}_j]=0,\quad [\wh{p}_i,\wh{x}^j]=-i\hbar\,\delta^j_i,\qquad \forall\;\; 0\leq i,j \leq n
\end{equation}
-- the ``canonical commutation relations.''

The standard representation of this algebra -- the Schr\"odinger representation -- is on the Hilbert space $\HH=L^2_\CC(\RR^n)$ of complex-valued square-integrable function on $\RR^n$, with hermitian structure 
$$\langle \psi_1,\psi_2\rangle\colon=\int_{\RR^n} d^n x\; \ol{\psi_1(x)}\psi_2(x)$$
for $\psi_1,\psi_2$ two square-integrable functions on $\RR^n$.
 The generators $\wh{x}^i$, $\wh{p}_i$  of the Weyl algebra act on $\HH$ as the following hermitian operators:
\begin{equation}
\wh{x}^i\colon \psi(x) \mapsto x^i\psi(x),\quad
\wh{p}_j\colon \psi(x)\mapsto -i\hbar \frac{\dd}{\dd x^i}\psi(x)
\end{equation}
I.e. $\wh{x}^i$ acts as a multiplication operator (by a coordinate function) and $\wh{p}^i$ acts as a derivation.\footnote{These operators are unbounded on $L^2_\CC(\RR^n)$.}

In particular, using this representation, one can identify the Weyl algebra of $T^*\RR^n$ with the algebra of polynomial differential operators in $n$ variables.
\end{example}

\subsection{Canonical quantization of the harmonic oscillator}
The idea of canonical quantization is to start with a classical system  in Hamiltonian formalism with phase space $\Phi=T^*\RR^n$ and lift (or ``quantize'') %\marginpar{explain, in which sense we are ``lifting''}
the Hamiltonian function $H(x,p)$ to an element $\wh{H}=H(\wh{x},\wh{p})$ of the corresponding Weyl algebra -- the quantum Hamiltonian.

By quantizing/lifting a polynomial function $f$ on $\Phi $ we mean 
choosing a preimage of $f$ under the ``dequantization map''
\begin{equation}\label{l15 pi}
\pi\colon \mr{Weyl}(\Phi)\xra{\bmod\hbar} C^\infty_\mr{poly}(\Phi).
\end{equation} 
where $C^\infty_\mr{poly}(\Phi)=\mr{Sym}^\bt\Phi^*$ is the algebra of polynomial functions on $\Phi$. 
Put another way, we take
a polynomial function $f(x,p)\in C^\infty_\mr{poly}(\Phi)$ and replace $x^i$, $p_j$  with corresponding generators of the Weyl algebra $\wh{x}^i$, $\wh{p}_j$, where we are allowed to add any terms proportional to $\hbar^k$ for $k>0$. The possibility to add such terms reflects the ordering ambiguity. E.g., $xp=px$ as functions on $\Phi=T^*\RR$, but $\wh{x}\wh{p}=\wh{p}\wh{x}+i\hbar$ in the Weyl algebra; so both $\wh{x}\wh{p}$ and $\wh{p}\wh{x}$ should be considered as legitimate quantizations of the monomial $xp$, and these quantizations are different. 

A systematic approach to lifting is to choose a ``quantization map'' (or ``operator ordering'').
\begin{definition} We call a ``quantization map'' a $\CC$-linear map
\begin{equation}
q\colon C^\infty_\mr{poly}(\Phi)\ra \mr{Weyl}(\Phi)
\end{equation} 
which satisfies $\pi\circ q=\mr{id}$, where $\pi$ is the map (\ref{l15 pi}). 
%sends $x^i\mapsto \wh{x}^i$, $p_j\mapsto \wh{p}_j$. 
\end{definition}
Note that $q$ is not required to be an algebra morphism; in fact, it cannot be one.

\begin{example}[Weyl quantization map]
Consider the map $q\colon C^\infty_\mr{poly}(\Phi)\ra \mr{Weyl}(\Phi)$ which sends a monomial in $x^i, p_j$ to the corresponding monomial in $\wh{x}^i$, $\wh{p}_j$, where one averages over all possible orderings of the factors (i.e. for a monomial of degree $d$, one averages over the symmetric group $S_d$). Then one extends $q$ to general polynomials by linearity.  One calls this map $q$ the Weyl (or ``symmetric'') quantization map.
%This is the so-called Weyl ordering/Weyl quantization map. 
\end{example}
%An example of a quantization map $q$ 
%maps monomials in $x^i, p_j$ to the corresponding monomials in $\wh{x}^i$, $\wh{p}_j$, where one averages over all possible orderings of the factors (i.e. for a monomial of degree $d$, one averages over the symmetric group $S_d$); then one extends $q$ to general polynomials by linearity. This is the so-called Weyl ordering/Weyl quantization map. 
%It is impossible to find a quantization map that is an algebra morphism (and which sends $x^i\mapsto \wh{x}^i$, $p_j\mapsto \wh{p}_j$) -- the reason is obvious: the generators of the algebra of polynomial functions commute while Weyl algebra generators do not commute.

In the case of harmonic oscillator, we 
lift the coordinate function $x,p$ on the phase space $\Phi=T^*\RR$ to the generators of the Weyl algebra $\wh{x}, \wh{p}$ satisfying the relation
\begin{equation}\label{l15 canonical commutation relation harm osc}
[\wh{p},\wh{x}]=-i\hbar.
\end{equation}
We lift the Hamiltonian function to the element
\begin{equation}\label{l15 H hat}
\wh{H}=\frac{\wh{p}^2}{2}+\omega^2 \frac{\wh{x}^2}{2}
\end{equation}
of the Weyl algebra.

\textbf{Disclaimer.}
In the discussion below, we will be thinking of $\hbar$ as a small positive real number (rather than a formal parameter), and formulae involving $\hbar$ should be thought of as a family over $\hbar\in \RR_{>0}$.

In the Schr\"odinger representation, the Weyl algebra is acting on the Hilbert space 
\begin{equation}
\HH=L^2_\CC(\RR),
\end{equation}
with 
\begin{equation}
\wh{x}=x\cdot,\quad \wh{p}=-i\hbar \frac{\dd}{\dd x}.
\end{equation}
The quantum Hamiltonian (\ref{l15 H hat}) is then represented as the differential operator
\begin{equation}\label{l15 Hhat harmonic oscillator}
\wh{H}=-\frac{\hbar^2}{2}\frac{\dd^2}{\dd x^2}+\frac{\omega^2}{2}x^2.
\end{equation}

To construct the evolution operator of the quantum system 
\begin{equation}\label{l15 U harm osc}
U(t)=e^{-\frac{it \wh{H}}{\hbar}}\quad \in U(\HH),
\end{equation}
where $U(\HH)$ is the unitary group,
 one needs to find the eigenvalues and eigenvectors 
 (as square-integrable functions) 
 of $\wh{H}$. I.e., one is looking for all pairs $\psi\neq 0\in L^2_\CC(\RR)$, $E\in \RR$ such that
\begin{equation}
 \left(-\frac{\hbar^2}{2}\frac{\dd^2}{\dd x^2}+\frac{\omega^2}{2}x^2\right) \psi(x) = E\psi(x).
\end{equation}
 This is a well-known instance of a singular Sturm-Liouville problem. The answer is: 
\begin{thm}\label{l15 thm harm osc spectrum}
The operator (\ref{l15 Hhat harmonic oscillator}) admits a complete orthonormal system of eigenvectors $\{\psi_n\}_{n\geq 0}$ in $L^2(\HH)$ of the form
\begin{equation}\label{l15 psi_n}
\psi_n= C_n e^{-\frac{\omega x^2 }{2\hbar}} H_n\left(\sqrt{\frac{\omega}{\hbar}}\,x\right)
\end{equation}
where 
\begin{equation}
H_n(x)\colon = (-1)^n e^{x^2}\frac{d^n}{dx^n} e^{-x^2}
\end{equation}
are Hermite polynomials; %\marginpar{write the f-la for $C_n$}
 $C_n=
 \left(\frac{\omega}{\pi \hbar}\right)^{\frac14} (2^n n!)^{-\frac12}
 $ is a normalization constant.
The eigenvalue of $\wh{H}$ corresponding to $\psi_n$ is
\begin{equation}
E_n=\hbar\omega (n+\frac12).
\end{equation}
\end{thm} 

The first few Hermite polynomials are:
\begin{center}
\begin{tabular}{c|c}
$n$ & $H_n(x)$ \\ \hline
$0$ & $1$ \\
$1$ & $2x$ \\
$2$ & $4x^2-2$ \\
$3$ & $8x^3-12 x$ \\
$4$ & $16x^4-48 x^2+12$\\
$\vdots$ & $\vdots$
\end{tabular}
\end{center}

The evolution operator (\ref{l15 U harm osc}) is then
\begin{equation}\label{l15 U(t)}
\begin{array}{cccc}
U(t)\colon& \HH & \ra &\HH \\
 &\psi & \mapsto& \displaystyle \sum_{n\geq 0} e^{-\frac{iE_n t}{\hbar}}\langle \psi_n,\psi\rangle\, \psi_n = 
  \sum_{n\geq 0} e^{-i(n+\frac12)\omega t}\langle \psi_n,\psi\rangle\, \psi_n
\end{array}
\end{equation}
 
\subsection{Creation/annihilation operators}
Instead of directly looking for eigenvectors and eigenvalues of the operator (\ref{l15 Hhat harmonic oscillator}), one can obtain the result of Theorem \ref{l15 thm harm osc spectrum} by exploiting the hidden algebraic structure of the operator $\wh{H}$ (specific to the harmonic oscillator case).

Let us introduce two new elements of the Weyl algebra\footnote{
More pedantically, here we extend the ring of scalars in the Weyl algebra by tensoring it with $\CC[\hbar^{1/2},\hbar^{-1/2}]$.
} -- special complex linear combinations of $\wh{x}$, $\wh{p}$:
\begin{eqnarray}\label{l15 cr/annih op eq1}
\wh{a}&=& \sqrt{\frac{\omega}{2\hbar}} (\wh{x}+\frac{i}{\omega}\, \wh{p}),\\ \label{l15 cr/annih op eq2}
\wh{a}^+&=& \sqrt{\frac{\omega}{2\hbar}} (\wh{x}-\frac{i}{\omega}\, \wh{p}).
\end{eqnarray}
Operators $\wh{a},\wh{a}^+$ are called the ``annihilation operator'' and  ``creation operator,'' respectively. The are hermitian conjugates of one another and satisfy the commutation relation
\begin{equation}\label{l15 comm rel btw a, a^plus}
[\wh{a},\wh{a}^+]=1
\end{equation}
as a consequence of the canonical commutation relation (\ref{l15 canonical commutation relation harm osc}). The inverse formulae to (\ref{l15 cr/annih op eq1}), (\ref{l15 cr/annih op eq2}) are:
\begin{eqnarray}
\wh{x}&=&\sqrt{\frac{\hbar}{2\omega}} (\wh{a}^++\wh{a}), \\
\wh{p}&=&i\sqrt{\frac{\hbar\omega}{2}} (\wh{a}^+-\wh{a}).
\end{eqnarray}

The quantum Hamiltonian (\ref{l15 Hhat harmonic oscillator}) expressed in terms of $\wh{a},\wh{a}^+$ is
\begin{equation}\label{l15 Hhat via cr/annih}
\wh{H}=\hbar\omega\frac12(\wh{a}^+\wh{a}+\wh{a}\wh{a}^+) =\hbar\omega\left(\wh{a}^+\wh{a}+\frac12\right)
\end{equation}

The relation (\ref{l15 comm rel btw a, a^plus}) implies the commutation relations between $\wh{H}$ and $\wh{a}$, $\wh{a}^+$:
\begin{gather}
[\wh{H},\wh{a}] = -\hbar\omega \wh{a} \\
[\wh{H},\wh{a}^+] = \hbar\omega \wh{a}^+
\end{gather}

This implies that if in a representation of the Weyl algebra on a Hilbert space $\HH$, a vector
 $\psi\in \HH$ is an eigenvector of $\wh{H}$ with eigenvalue $E$, then
\begin{eqnarray}
\wh{H}(\wh{a}\psi) &=& (E-\hbar\omega) (\wh{a}\psi),  \label{l15 a lowers E} \\
\wh{H}(\wh{a}^+\psi) &=& (E+\hbar\omega) (\wh{a}^+\psi). \label{l15 a^+ raises E}
\end{eqnarray}
Thus, $\wh{a}$, $\wh{a}^+$ take eigenvectors of $\wh{H}$ to eigenvectors; applying $\wh{a}^+$ \emph{raises} the eigenvalue of  by  $\hbar\omega$, while $\wh{a}$ \emph{lowers} the eigenvalue by $\hbar\omega$.

\marginpar{Lecture 16,\\
9/28/2022}

We can construct an irreducible unitary representation $\Hosc$ of the Weyl algebra as follows: let $|0\rangle \in \Hosc$ be the ``vacuum vector'' -- a vector with the property 
\begin{equation}
\wh{a} |0\rangle =0.
\end{equation}
We will assume that $|0\rangle$ has norm $1$ in $\Hosc$.
From (\ref{l15 Hhat via cr/annih}) we infer that
\begin{equation}
\wh{H}|0\rangle = \frac{\hbar\omega}{2} |0\rangle.
\end{equation}
We then introduce the vectors $|n\rangle\in \Hosc$ with $n=1,2,3,\ldots$ as
\begin{equation}\label{l16 |n>}
|n\rangle \colon =  \alpha_n(\wh{a}^+)^n |0\rangle
\end{equation}
where $\alpha_n$ is a normalization factor, chosen in such a way that the vectors $|n\rangle$ are of norm $1$.
From (\ref{l15 a^+ raises E}) we infer that
\begin{equation}
\wh{H} |n\rangle = \Big(n+\frac12\Big) \hbar\omega |n\rangle
\end{equation}

The representation space $\Hosc$ of the Weyl algebra is then
\begin{equation}\label{l16 H osc}
\Hosc=%\mr{Span}_\CC\{|0\rangle, |1\rangle,|2\rangle,\ldots\}
\Big\{\sum_{n\geq 0} c_n |n\rangle \;\;\Big|\;\; c_n\in \CC,\; \sum_{n\geq 0}|c_n|^2<\infty\Big\}.
\end{equation}

One can calculate the norms/inner products of vectors in $\Hosc$ from the fact that $\wh{a}$, $\wh{a}^+$ are Hermitian conjugate, using the commutation relation (\ref{l15 comm rel btw a, a^plus}). E.g., one has
\begin{equation}
\Big\langle \wh{a}^+|0\rangle, \wh{a}^+|0\rangle  \Big\rangle = 
\Big\langle |0\rangle, \wh{a}\wh{a}^+|0\rangle  \Big\rangle  = \langle 0| \underbrace{\wh{a}\wh{a}^+}_{\wh{a}^+\wh{a}+1} |0\rangle
=\langle 0| \wh{a}^+ \underbrace{\wh{a} |0\rangle}_0 + \underbrace{\langle 0|0\rangle}_{||\;|0\rangle \;||^2=1} =1
\end{equation}
Here we used the Dirac's notation: a covector in $(\Hosc)^*$ dual to the vector $|\psi\rangle\in \Hosc$ is denoted $\langle \psi|$; the inner product $\Big\langle |\psi_1\rangle, |\psi_2\rangle \Big\rangle_{\Hosc}$ of two vectors in $\Hosc$ is also denoted $\langle \psi_1 | \psi_2 \rangle$.

More generally, using the same strategy -- commuting $\wh{a}$ to the right of the word of creation/annihilition operators --  one can show the following.

\begin{lemma} For $n,m=0,1,2,\ldots$, one has
\begin{equation}\label{l16 lm eq}
\langle 0 | \wh{a}^m (\wh{a}^+)^n|0 \rangle = n!\,\delta_{nm}.
\end{equation}
\end{lemma}
%\marginpar{add proof}
\begin{proof}
First note that we have the commutation relation
\begin{equation}
[\wh{a},(\wh{a}^+)^n]=\sum_{k=1}^n (\wh{a}^+)^{k-1} \underbrace{[\wh{a},\wh{a}^+]}_1(\wh{a}^+)^{n-k} = n(\wh{a}^+)^{n-1}.
\end{equation}
Using it, we find
\begin{equation}\label{l16 eq1}
\wh{a}(\wh{a}^+)^n|0\rangle 
= [\wh{a},(\wh{a}^+)^n]|0\rangle+ (\wh{a}^+)^n \underbrace{\wh{a} |0\rangle}_{0} = n(\wh{a}^+)^{n-1}|0\rangle.
\end{equation}
Thus, for $m\leq n$, we have
\begin{multline}
(\wh{a})^m (\wh{a}^+)^n |0\rangle = 
(\wh{a})^{m-1} \wh{a} (\wh{a}^+)^n |0\rangle 
= (\wh{a})^{m-1} n (\wh{a}^+)^{n-1} |0\rangle 
=\\= (\wh{a})^{m-2} n \wh{a}(\wh{a}^+)^{n-1} |0\rangle
=(\wh{a})^{m-2} n(n-1) (\wh{a}^+)^{n-2} |0\rangle
\\=\cdots =
n(n-1)\cdots (n-m+1) (\wh{a}^+)^{n-m}|0\rangle
\end{multline}
In particular, for $m=n$ we have
\begin{equation}
(\wh{a})^n (\wh{a}^+)^n |0\rangle = n! |0\rangle,
\end{equation}
which implies (\ref{l16 lm eq}) for $m=n$.

If $m<n$, we have
\begin{equation}
\langle 0 | \wh{a}^m (\wh{a}^+)^n|0 \rangle = \frac{n!}{(n-m)!}\langle 0| (\wh{a}^+)^{n-m} |0 \rangle =0,
\end{equation}
where we use the fact that $\langle 0| \wh{a}^+$ is the covector dual to $\wh{a}|0\rangle$ and thus vanishes.

Likewise, if $m>n$, we have
\begin{equation}
\langle 0 | \wh{a}^m (\wh{a}^+)^n|0 \rangle = n! \langle 0| \underbrace{\wh{a}^{m-n} |0\rangle}_0 = 0.
\end{equation}

%(proven using (\ref{l16 eq1}) by induction in $m$).
\end{proof}

In particular, vectors (\ref{l16 |n>}) with $n=0,1,2,\ldots$ form an orthonormal basis for $\HH$ if we set the normalization factors to be
\begin{equation}
\alpha_n=\frac{1}{\sqrt{n!}}.
\end{equation}
In this basis, the operators $\wh{a},\wh{a}^+$ act as
\begin{equation}
\wh{a} |n\rangle = \frac{1}{\sqrt{n!}}\underbrace{[\wh{a},(\wh{a}^+)^n]}_{n(\wh{a}^+)^{n-1}} |0\rangle = \sqrt{n}\; |n-1\rangle
\end{equation}
and
\begin{equation}
\wh{a}^+ |n\rangle = \underbrace{\frac{1}{\sqrt{n!}}}_{\frac{\sqrt{n+1}}{\sqrt{(n+1)!}}} (\wh{a}^+)^{n+1}|0\rangle = \sqrt{n+1}\; |n+1\rangle .
\end{equation}

By Stone-von Neumann theorem, there is an isomorphism of representations of the Weyl algebra 
\begin{equation} 
\Hosc\simeq L^2_\CC(\RR)
\end{equation} 
-- the ``oscillator representation'' and Schr\"odinger representation. Under this isomorphism vectors $|n\rangle\in \Hosc$ correspond to vectors (\ref{l15 psi_n}). In fact, one can obtain the formula (\ref{l15 psi_n}) from (\ref{l16 |n>}). Indeed: in Schr\"odinger representation,  the operators $\wh{a}$, $\wh{a}^+$ are
\begin{gather}
\wh{a}= \frac{1}{\sqrt 2} \left(y+\frac{\dd}{\dd y}\right) = \frac{1}{\sqrt{2}} e^{-\frac{y^2}{2}}\frac{\dd}{\dd y}e^{\frac{y^2}{2}} ,\\ 
\wh{a}^+=\frac{1}{\sqrt 2} \left(y-\frac{\dd}{\dd y}\right)  = \frac{-1}{\sqrt 2}e^{\frac{y^2}{2}}\frac{\dd}{\dd y}e^{-\frac{y^2}{2}},
\end{gather}
where we denoted $y=\sqrt{\frac{\omega}{\hbar}}\,x$.
Thus, the vacuum vector $|0\rangle$ in Schr\"odinger representation is a function $\psi_0$ satisfying the first-order ODE
\begin{equation}
\wh{a}\psi_0=0 \quad \Leftrightarrow \quad \frac{\dd}{\dd y}\left(e^{\frac{y^2}{2}} \psi_0(y)\right)=0\quad  \Leftrightarrow\quad  \psi_0(y) = C_0 e^{-\frac{y^2}{2}}
\end{equation}
with $C_0$ a constant (which can be chosen to normalize $\psi_0$ to unit norm).  Vectors $|n\rangle$ in Schr\"odinger representation are then
\begin{multline}
\psi_n(y) = \alpha_n (\wh{a}^+)^n |0\rangle = \alpha_n (-1)^n 2^{-\frac{n}{2}} e^{\frac{y^2}{2}} \frac{\dd^n}{\dd y^n} \left( e^{-\frac{y^2}{2}} \psi_0(y) \right)=\\
 =   2^{-\frac{n}{2}}  C_0 \alpha_n e^{-\frac{y^2}{2}}\underbrace{ \left( (-1)^n e^{y^2} \frac{\dd^n}{\dd y^n} e^{-y^2}\right) }_{H_n(y)} .
\end{multline}
This is exactly the formula (\ref{l15 psi_n}).

In terms of the basis $\{|n\rangle\}$ in the Hilbert space $\Hosc$, the evolution operator (\ref{l15 U(t)})  acts as
\begin{equation}
U(t)=e^{-\frac{i\wh{H}t}{\hbar}}\colon \quad\sum_{n\geq 0} c_n |n\rangle \mapsto  \sum_{n\geq 0} c_n e^{-i(n+\frac12) \omega t} |n\rangle.
\end{equation}

\begin{remark}
The partition function of the harmonic oscillator on the circle of length $t$ (cf. Example \ref{l2 example: QM on S^1}) is
\begin{equation}\label{l15 harm osc Z}
Z(S^1_t)=\mr{tr}_{\Hosc}U(t)= \sum_{n\geq 0} e^{-i(n+\frac12)\omega t}=\frac{e^{-\frac{i\omega t}{2}}}{1-e^{-i\omega t}}=\frac{1}{e^{\frac{i\omega t}{2}}-e^{-\frac{i\omega t}{2}}}=\frac{1}{2i\sin \frac{\omega t}{2}}.
\end{equation}
The Euclidean version of the partition function is obtained by the Wick rotation $t=-iT_\E$ with $T_\E>0$. In this version, the sum over eigenvalues in (\ref{l15 harm osc Z}) becomes absolutely convergent and one has
\begin{equation}
Z_\E(S^{1}_{T_\E})\colon =Z(S^1_{t=-iT_\E})=
\sum_{n\geq 0} e^{-(n+\frac12)\omega T_\E}=
\frac{1}{2\sinh \frac{\omega T_\E}{2}}.
\end{equation}
\end{remark}

\begin{remark}
The algebra of creation/annihilation operators (\ref{l15 comm rel btw a, a^plus}) admits another useful representation (unitarily isomorphic to $\Hosc$ and to the Schr\"odinger representation), on the Segal-Bargmann space, constructed as follows.
Consider the following hermitian inner product on the space $\mr{Hol}(\CC)$ of holomorphic functions on $\CC$:
\begin{equation}\label{l15 SB inner product}
\langle f,g \rangle_\mr{SB} = \frac{1}{\pi}\int \frac{i}{2}dz\wedge d\bar{z}\; e^{-|z|^2} \ol{f(z)}g(z).
\end{equation}
Then the Segal-Bargmann space is defined as
\begin{equation}
\HH_\mr{SB}\colon= \{f\in \mr{Hol}(\CC)\;|\; \langle f,f \rangle_\mr{SB}<\infty\} .
\end{equation}
% -- the space of holomorphic functions on $\CC$, completed w.r.t. the $L^2$ norm associated with the 
In this representation, creation and annihilation operators act as
\begin{equation}
\wh{a}=\frac{\dd}{\dd z},\quad \wh{a}^+=z\cdot
\end{equation}
-- holomorphic derivative and multiplication operator by the holomorphic coordinate, respectively; these operators are hermitian conjugate of one another w.r.t. the inner product (\ref{l15 SB inner product}). The vacuum vector $|0\rangle$ can be identified with the function $1\in \HH_\mr{SB}$; then the vectors $|n\rangle$ are identified with $\frac{1}{\sqrt{n!}} z^n\in\HH_\mr{SB}$. The Hamiltonian $\wh{H}=\hbar\omega (z\frac{\dd}{\dd z}+\frac12)$, up to normalization and a shift, is the Euler vector field and thus counts the monomial degree of a function in $z$.
\end{remark}

\subsubsection{Normal ordering}. Normal ordering is an operation acting on linear combination words in the creation-annihilation operators $\wh{a},\wh{a}^+$ which reshuffles the letters in each word, putting annihilation operators $\wh{a}$ to the right and creation operators $\wh{a}^+$ to the left. Normal ordering applied to a word $W$ is denoted $:W:$. For instance, one has
\begin{equation}
: \wh{a}\wh{a}^+\wh{a}\wh{a}^+: = \wh{a}^+\wh{a}^+\wh{a}\wh{a}.
\end{equation}
We stress that normal ordering is an operation on words -- it does not descend to the Weyl algebra.

An important property of normal ordering is that if $O$ is a sum of words, each containing at least one creation or annihilation operator (i.e. no constant summand in $O$), then one has
\begin{equation}\label{l16 <0| :O: |0>}
\langle 0| :O: |0 \rangle = 0.
\end{equation}
This property is obvious: for each normally ordered word, the expression (\ref{l16 <0| :O: |0>})  will contain a term $\wh{a}|0\rangle$ and/or a term $\langle 0| \wh{a}^+$, both of which vanish.

In particular, if we represent the Hamiltonian of the harmonic oscillator by the combination of words $\wh{H}=\hbar\omega\frac12 (\wh{a}\wh{a}^++\wh{a}^+\wh{a})$, then we have
\begin{equation}
:\wh{H}: = \hbar\omega \wh{a}^+\wh{a}.
\end{equation}
-- which differs from (\ref{l15 Hhat via cr/annih}) by $\frac{\hbar\omega}{2}$. In particular, one has
\begin{equation}
:\wh{H}:|n\rangle = \hbar\omega n |n \rangle.
\end{equation}
In particular, the vacuum vector $|0\rangle$ is  has zero eigenvalue w.r.t. the normally ordered Hamiltonian $: \wh{H}:$,
\begin{equation}
:\wh{H}:|0\rangle =0.
\end{equation}

\section{Free massless scalar field on Minkowski cylinder}
\label{ss: free scalar on Minkowski cylinder}

\subsection{Lagrangian formalism}

Consider the massless scalar field on the cylinder $\Sigma=\RR\times S^1$ with Minkowski metric $g=(dt)^2-(d\sigma)^2$. Here $t$ (time) is the coordinate on $\RR$ and $\sigma\in S^1=\RR/2\pi\ZZ$ is the ``spatial coordinate.''

%\textcolor{red}{PICTURE.}
\begin{figure}[H]
%$$\vcenter{\hbox{ \includegraphics[scale=0.5]{cob1.eps} }} $$
\begin{center}
\includegraphics[scale=1]{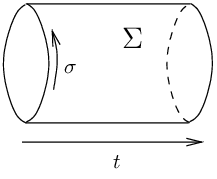}
\end{center}
\caption{Cylinder.
}
\end{figure}

Fields of the theory are smooth real functions $\phi(t,\sigma)$ on $\Sigma$ and the action functional is
\begin{equation}\label{l16 S}
S(\phi)=\frac{\kappa}{2}\int_{\Sigma} dt\,d\sigma (\dot{\phi}^2-(\dd_\sigma \phi)^2)
\end{equation}
where dot means the derivative in $t$. We put a normalization factor $\kappa$ in the definition of the action -- we will fix it later.

The space of fields of the theory $\F=\mr{Map}(\RR\times S^1,\RR)$ can be thought of as $\mr{Map}(\RR,\mr{Map}(S^1,\RR))$. Thus, one can think of the field theory on the cylinder $\Sigma$ as classical mechanics on the worldline $\RR$ with target 
\begin{equation}\label{l16 X}
X=\mr{Map}(S^1,\RR)=C^\infty(S^1) \quad \ni \phi(\sigma)
\end{equation} 
and Lagrangian 
\begin{equation}\label{l16 L}
\mathsf{L}=\frac{\kappa}{2}\oint_{S^1} d\sigma (\dot{\phi}^2-(\dd_\sigma \phi)^2)
\end{equation}
-- a function on $TX$ (cf. (\ref{l15 L as fun on TX})). We understand $\phi(\sigma)$\footnote{Here we mean the function on $S^1$, not its value at some particular $\sigma$.} as a point in the base of $TX$ and $\dot\phi(\sigma)$ as a tangent vector to $X$ at $\phi(\sigma)$. 

The Euler-Lagrange equation of the theory is the wave equation
\begin{equation}\label{l16 wave eq}
\ddot\phi-\dd_\sigma^2\phi=0.
\end{equation}
Its solution can be though of as a path in $X$ parametrized by $t\in\RR$.

Let us expand $\phi(\sigma)$ in the Fourier series
\begin{equation}
\phi(\sigma)= \sum_{n\in\ZZ} \phi_n e^{in\sigma}.
\end{equation}
Since $\phi$ is real-valued, the Fourier coefficients (or ``modes'') $\phi_n\in \CC$ must satisfy the reality condition $\phi_{-n}=\ol{\phi}_n$.
A path in $X$ is then specified by a collection of Fourier modes $\phi_n(t)$ as functions of $t\in \RR$.

In terms of Fourier modes, the Lagrangian (\ref{l16 L}) is
\begin{equation}
\mathsf{L}=\frac{\kappa}{2}2\pi \sum_{n\in\ZZ}\left(\dot\phi_n\dot\phi_{-n}-n^2 \phi_n\phi_{-n}\right)
\end{equation}

\subsection{Hamiltonian formalism}

In Hamiltonian formalism, the phase space of the system is 
\begin{equation}
\Phi=T^*X,
\end{equation}
with $X$ as in (\ref{l16 X}). Since $X$ is a linear space, we identify $T^*X$ with $X\times T^*X$ -- pairs of a function $\phi(\sigma)$ on $S^1$ and a distribution $\pi(\sigma)$ on $S^1$ (the ``momentum''). The canonical symplectic form on $\Phi$ is 
\begin{equation}
\omega_\mr{symp}= \oint_{S^1} dt\;\delta\pi(\sigma) \wedge \delta\phi(\sigma).
\end{equation}
The corresponding Poisson brackets between $\phi(\sigma)$, $\pi(\sigma')$ (thought of as coordinate functions on $\Phi$) are
\begin{equation}\label{l16 Poisson brackets, coord rep}
\{\phi(\sigma),\pi(\sigma')\}=-\delta_\mr{per}(\sigma-\sigma'),\quad
\{\phi(\sigma),\phi(\sigma')\}=0,\quad \{\pi(\sigma),\pi(\sigma')\}=0,
\end{equation}
where $\delta_\mr{per}$ is the periodic Dirac delta-distibution on $S^1$, $\displaystyle \delta_\mr{per}(\sigma)=\sum_{n\in\ZZ} \delta(\sigma+2\pi n)$ (where on the right $\delta$ are the usual Dirac delta-distributions on $\RR$).

To find the Legendre transform of the Lagrangian (\ref{l16 L}), we first find the relation between momenta and velocities:
\begin{equation}\label{l16 pi via phi dot}
\pi(\sigma)=\frac{\delta \mathsf{L}}{\delta\dot\phi(\sigma)}=\kappa \dot\phi(\sigma),
\end{equation}
cf. (\ref{l15 Legendre 2}). Then we find the Hamiltonian (cf. (\ref{l15 Legendre 1})) as
\begin{equation}\label{l16 H coord rep}
H=\oint_{S^1} d\sigma \pi(\sigma) \dot\phi(\sigma) - \mathsf{L} = \oint_{S^1}d\sigma \left(
\frac{\pi(\sigma)^2}{2\kappa}+\frac{\kappa}{2}(\dd_\sigma \phi)^2
\right),
\end{equation}
where in the second step we expresses velocities in terms of momenta using (\ref{l16 pi via phi dot}).

The Hamiltonian equations generated by the Hamiltonian $H$ are
\begin{equation}\label{l16 Ham eq}
\dot\phi= \frac{1}{\kappa}\pi, \quad \dot\pi = \kappa\dd^2_\sigma \phi.
\end{equation}
In particular, these equations imply the wave equation (\ref{l16 wave eq}) for $\phi$.

\begin{remark}\label{l16 Rem H and P via T}
The components of the stress-energy tensor of the theory are
\begin{eqnarray}
T_{00}=T_{11}&=& \frac{\kappa}{2}(\dot\phi^2+(\dd_\sigma\phi)^2), \\
T_{01}=T_{10}&=& \kappa\dot\phi \dd_\sigma \phi
\end{eqnarray}
We note that integrating $T_{00}$ over $\{t\}\times S^1$ one gets
\begin{equation}
H=\oint_{S^1} d\sigma\, T_{00}
\end{equation}
-- the Hamiltonian (or ``total energy''). Integrating $T_{01}$ over $\{t\}\times S^1$ one gets
\begin{equation}\label{l16 P}
P\colon= \oint_{S^1} d\sigma\, T_{01}
\end{equation}
-- the ``total momentum'' of the system. 

Modulo equations of motion, $H$ and $P$ do not depend on $t$ -- the position of the spatial slice. One can infer this from Lemma \ref{l13 lm <T,r> conserved}: translations along $\RR$ and rotations along $S^1$ are source symmetries and yield conserved currents, $T_{i0}$ and $T_{i1}$, hence the corresponding charges (fluxes through a spatial slice $\{t\}\times S^1$) are conserved  -- independent of $t$ modulo equations of motion.
\end{remark}

Expanding the field $\phi(\sigma)$ and the momentum $\pi(\sigma)$ in Fourier modes on $S^1$, we have
\begin{equation}
\phi(\sigma)=\sum_{n\in \ZZ} \phi_n e^{in\sigma},\quad \pi(\sigma)=\frac{1}{2\pi}\sum_{n\in \ZZ} \pi_n e^{in\sigma},
\end{equation}
with reality conditions $\phi_{-n}=\ol{\phi_n}$ and $\pi_{-n}=\ol{\pi_n}$. Poisson brackets (\ref{l16 Poisson brackets, coord rep}) correspond to the following brackets between the modes:
\begin{equation}\label{l16 Poisson brackets between modes}
\{\phi_n,\pi_m\}=-\delta_{n,-m},\quad \{\phi_n,\phi_m\}=0,\quad \{\pi_n,\pi_m\}=0.
\end{equation}

The Hamiltonian (\ref{l16 H coord rep}) written in terms of the parametrization of the phase space by Fourier modes $\phi_n,\pi_n$ is:
\begin{equation}
H=\sum_{n\in \ZZ} \frac12 \frac{1}{2\pi\kappa} \pi_n \pi_{-n}+\frac12 2\pi\kappa n^2 \phi_n\phi_{-n}.
\end{equation}
At this point we want to fix the normalization factor $\kappa$ to the value
\begin{equation}\label{l16 kappa value}
\kappa=\frac{1}{4\pi}.
\end{equation}
Then we have
\begin{equation}\label{l16 H via modes}
H=\sum_{n\in\ZZ} \pi_n\pi_{-n}+\frac14 n^2 \phi_n\phi_{-n} =(\pi_0)^2 +2\sum_{n>0 } \left(
|\pi_n|^2+\frac14 n^2 |\phi_n|^2
\right).
\end{equation}

\marginpar{Lecture 17,\\ 9/30/2022}
Similarly, the total momentum (\ref{l16 P}) is:
\begin{equation}\label{l17 P via modes}
P=\sum_{n\in\ZZ} in\pi_{-n}\phi_n.
\end{equation}

The Hamiltonian equations (\ref{l16 Ham eq}) spelled in terms of coordinates $\phi_n,\pi_n$ on the phase space read
\begin{equation}
\dot\phi_n=2\pi_n,\quad \dot\pi_n =-\frac{n^2}{2} \phi_n.
\end{equation}
As a consequence, $\phi_n$ satisfies the second-order ODE $\ddot\phi_n+n^2\phi_n=0$ (cf. (\ref{l15 ddot(x)+omega^2 x=0})).

Thus the system is a superposition of a collection of non-interacting subsystems: variables $(\phi_0,\pi_0)$ describe a free particle of mass $\mu=\frac12$ while variables $(\phi_n,\pi_{n})$ for $n\neq 0$ describe a complex harmonic oscillator with frequency $\omega_n=|n|$.
%\marginpar{Explain more/better? Also: $(\phi_n,\pi_n)$ or $(\phi_n,\pi_{-n})$?}
%\footnote{
%More precisely: the phase space $\Phi$ can be seen as real slice (given by equations $\phi_{-n}=\ol{\phi_n}$, $\pi_{-n}=\ol{\pi_n}$) in the complex symplectic space $\prod_{n\in\ZZ} \Phi_n$ where $\Phi_n=T^*\CC$ is equipped the with canonical symplectic structure and the Hamiltonian $H_n=|\pi_n|^2+\frac{n^2}{4} |\phi_n|^2$.
%}

\subsubsection{Real oscillators.}
To get a better understanding of how the system breaks up into a collection of harmonic oscillators (plus a free particle), it useful to rewrite it in the real parametrization. %-- real and imaginary parts  $\phi^{(1,2)}_n$ 
%of $\phi_n$ and real and imaginary parts (up to normalization) $\pi_n^{(1,2)}$ of $\pi_n$. 
Introduce the real variables $\phi_n^{(1,2)}$, $\pi_n^{(1,2)}$, with $n>0$, related to complex variables $\phi_n,\pi_n$ by
\begin{equation}\label{l17 real osc}
\phi_n=\phi_n^{(1)}+i\phi_n^{(2)},\quad \pi_n=\frac12(\pi_n^{(1)} +i\pi_n^{(2)}) \qquad \mr{for}\;\;n>0.
\end{equation}
I.e., $\phi^{(1,2)}_n$ are the real/imaginary parts of $\phi_n$, $n>0$, and similarly for $\pi_n$.
The real variables satisfy the Poisson brackets
\begin{equation}
\{\phi_n^{(\alpha)},\pi_m^{(\beta)}\}=\delta_{nm}\delta_{\alpha\beta},\quad \{\phi_n^{(\alpha)},\phi_m^{(\beta)}\}=0,\quad \{\pi_n^{(\alpha)},\pi_m^{(\beta)}\}=0
\end{equation}
for $n>0$ and $\alpha,\beta\in\{1,2\}$. The Hamiltonian (\ref{l16 H via modes}) in these variables reads
\begin{multline}\label{l17 H real osc}
H=\pi_0^2+\sum_{n\geq 1} \sum_{\alpha=1}^2 \left(
\frac{(\pi_n^{(\alpha)})^2}{2}+\frac{n^2}{2} (\phi_n^{(\alpha)})^2
\right)\\
=H_{\mr{free\;particle,\;}\mu=\frac12} + \sum_{n\geq 1}\sum_{\alpha=1}^2 H_{\mr{harmonic\; oscillator,\;}\omega_n=n}
\end{multline}

The general solution of the Hamiltonian equations (\ref{l16 Ham eq}) is
\begin{eqnarray}\label{l17 phi cl solution}
\phi(t,\sigma)&=&\sum_{n\neq 0} \Big(\underbrace{
A_n e^{in(t+\sigma)} + B_n e^{in(-t+\sigma)}}_{\phi_n(t)e^{in\sigma}}
\Big)+\underbrace{Ct+D}_{\phi_0(t)},\\
\pi(t,\sigma)&=& \frac{1}{2\pi}\Big(\sum_{n\neq 0}\underbrace{\frac{in}{2}\Big(  A_n e^{in(t+\sigma)}-B_n e^{in(-t+\sigma)}}_{\pi_n(t) e^{in\sigma}} \Big) + \underbrace{\frac{C}{2}}_{\pi_0(t)} \Big),
\end{eqnarray}
where $A_n,B_n,C,D$ are arbitrary constants subject to the reality constraints
\begin{equation}
A_{-n}=\ol{A_n},\; B_{-n}=\ol{B_n}\;\;\mr{for}\; n\neq 0,\qquad C,D\in\RR.
\end{equation}

\begin{remark} For the \emph{massive} scalar field (\ref{l12 S scalar}) on the Minkowski cylinder we can repeat all the computations above, introducing the same parametrization of the phase space by modes $\phi_n,\pi_n$. The Hamiltonian instead of (\ref{l16 H via modes}) will then be
\begin{equation}
H=\sum_{n\in \ZZ} \pi_n \pi_{-n}+\frac14 \omega_n^2 \phi_n\phi_{-n}
\end{equation}
with 
\begin{equation}\label{l17 omega_n}
\omega_n=\sqrt{n^2+m^2}
\end{equation} 
(with $m$ the mass of the scalar field). Thus, the system is a collection of non-interacting harmonic oscillators, one for each $n\in\ZZ$, with $n$-th oscillator having frequency  (\ref{l17 omega_n}). 
\begin{figure}[H]
%$$\vcenter{\hbox{ \includegraphics[scale=0.5]{cob1.eps} }} $$
\begin{center}
\includegraphics[scale=0.6]{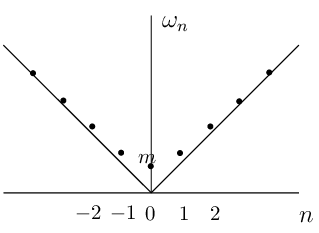}
\end{center}
\caption{Frequencies $\omega_n$ of oscillators comprising the free massive scalar field.
}
\end{figure}
In the massless limit $m\ra 0$, the frequencies become $\omega_n\ra |n|$. In particular, the $n=0$ oscillator in the limit becomes a free particle.
\end{remark}

\subsection{Aside: free particle}\label{sss free particle}
Since the free massless scalar field on a cylinder splits into a family of harmonic oscillators and a single free particle (cf. (\ref{l17 H real osc})), we stop for a moment to discuss the free particle, as a classical and as a quantum mechanical system.

The free particle moving on $\RR$ is the Lagrangian formalism is defined by the space of fields
$\F=\mr{Map}([t_0,t_1],\RR)$ with action functional
\begin{equation}
S[x(t)]=\int_{t_0}^{t_1} d\tau \underbrace{\frac{\mu \dot{x}^2}{2}}_{\mathsf{L}},
\end{equation}
where $\mu>0$ is a parameter -- ``mass'' of the particle.

In the Hamiltonian formalism, the system is described by the phase space $\Phi=T^*\RR$ and the Hamiltonian 
\begin{equation}
H=\frac{p^2}{2\mu}
\end{equation} 
(which is in particular the Legendre transform of the Lagrangian $\mathsf{L}=\frac{\mu v^2}{2}$).

In canonical quantization, we have the Weyl algebra generated by $\wh{x},\wh{p}$ subject to $[\wh{p},\wh{x}]=-i\hbar$, and the quantum Hamiltonian (using the symmetric Weyl quantization) is
\begin{equation}
\wh{H}=\frac{\wh{p}^2}{2\mu}.
\end{equation}
In Schr\"odinger representation of the Weyl algebra, the Hamiltonian acts as the differential operator 
\begin{equation}\label{l17 Hhat free particle}
\wh{H}= -\frac{\hbar^2}{2\mu} \dd_x^2
\end{equation}
on the Hilbert space $\HH=L^2_\CC(\RR)$.

The eigenvectors of $\wh{H}$ are the vectors
\begin{equation}\label{l17 |p>}
|p\rangle\colon= e^{\frac{i}{\hbar}px}
\end{equation}
with $p\in \RR$ a parameter (momentum). One then has
\begin{equation}
\wh{H} |p\rangle = \frac{p^2}{2\mu} |p\rangle.
\end{equation}
In particular, the operator $\wh{H}$ has a continuum eigenvalue spectrum $[0,\infty)$, where the eigenvalue $0$ is nondegenerate and all positive eigenvalues have multiplicity $2$. We also note that eigenvectors (\ref{l17 |p>}) are not points of $L^2_\CC(\RR)$ (not square-integrable), but rather are limit points of the space  (which is the usual case for a continuum spectrum).

\subsection{Canonical quantization% of the free massless scalar field on a Minkowski cylinder
}
\label{sss free scalar field can quant}
We now proceed to the canonical quantization of the free massless scalar field on the Minkowski cylinder.

We promote the modes $\phi_n,\pi_n$ to generators $\wh\phi_n,\wh\pi_n$ of the Weyl algebra, subject to the relations
\begin{equation}\label{l17 can comm rel}
[\wh\pi_n,\wh\phi_m]=-i \delta_{n,-m},\quad [\wh\phi_n,\wh\phi_m]=0,\quad [\wh\pi_n,\wh\pi_m]=0.
\end{equation}
For convenience we set $\hbar=1$.

Next, we introduce a system of creation/annihilation operators $\wh{a}_n,\wh{\ol{a}}_n$, $n\neq 0$, subject to hermitian conjugation properties
\begin{equation}
(\wh{a}_n)^+=\wh{a}_{-n}, \quad (\wh{\ol{a}}_n)^+=\wh{\ol{a}}_{-n}
\end{equation}
and related to the Weyl generators $\wh\phi_n,\wh\pi_n$, with $n\neq 0$ by\footnote{
One can also express the operators $\wh{a}_n,\wh{\ol{a}}_n$ in terms of the standard creation/annihilation operators (\ref{l15 cr/annih op eq1}), (\ref{l15 cr/annih op eq2}) for the real oscillators, as in (\ref{l17 real osc}): for $n>0$ one sets
$\wh{a}_n={\sqrt{\frac{n}{2}} (-i \wh{a}_n^{(1)}-\wh{a}_n^{(2)})}$, $\wh{a}_{-n}=\sqrt{\frac{n}{2}} (i \wh{a}_n^{(1)+}-\wh{a}_n^{(2)+})$,
$\wh{\ol{a}}_n=\sqrt{\frac{n}{2}} (-i \wh{a}_n^{(1)}+\wh{a}_n^{(2)})$,
$\wh{\ol{a}}_{-n}=\sqrt{\frac{n}{2}} (i \wh{a}_n^{(1)+}+\wh{a}_n^{(2)+})$.
}
\begin{equation}\label{l17 phi, pi via a, abar}
\begin{aligned}
\wh{\phi}_n
&=\frac{i}{n}(-\wh{a}_{-n}+\wh{\ol{a}}_n),\\
\wh{\pi}_n&=\frac{\wh{a}_{-n}+\wh{\ol{a}}_n}{2}.
\end{aligned}
\end{equation}
The commutation relations corresponding to (\ref{l17 can comm rel}) are
\begin{equation}\label{l17 a,abar comm rel}
[\wh{a}_n,\wh{a}_m]=n\delta_{n,-m},\quad [\wh{\ol{a}}_n,\wh{\ol{a}}_m]=n\delta_{n,-m},\quad
[\wh{a}_n,\wh{\ol{a}}_m]=0.
\end{equation}

In terms of these creation/annihilation operators (and the zero-mode operators $\wh{\phi}_0$, $\wh{\pi}_0$ which need to be treated separately), the quantum Hamiltonian (obtained by symmetric Weyl quantization) is:
\begin{equation}\label{l17 H via a,abar}
\wh{H}=\sum_{n\neq 0} \frac{\wh{a}_{-n}\wh{a}_{n}+\wh{\ol{a}}_{-n}\wh{\ol{a}}_n}{2} + (\wh{\pi}_0)^2
=\frac12\sum_{n\in\ZZ} \left( \wh{a}_{-n}\wh{a}_{n}+\wh{\ol{a}}_{-n}\wh{\ol{a}}_n\right).
\end{equation}
In the second equality we introduced the notation 
\begin{equation}\label{l17 a_0}
\wh{a}_0=\wh{\ol{a}}_0\colon=\wh{\pi}_0.
\end{equation}
%If we introduce the notation $\wh{a}_0=\wh{\ol{a}}_0\colon=\wh{\pi}_0$, then we can rewrite $\wh{H}$ as
%\begin{equation}
%\wh{H}=\frac12\sum_{n\in\ZZ} \left( \wh{a}_{-n}\wh{a}_{n}+\wh{\ol{a}}_{-n}\wh{\ol{a}}_n\right).
%\end{equation}

The canonical quantization of the total momentum operator (\ref{l17 P via modes}), written in terms of creation/annihilation operators, is
\begin{equation}\label{l17 P via a,abar}
\wh{P}=\frac12\sum_{n\in\ZZ} \left( \wh{a}_{-n}\wh{a}_{n}-\wh{\ol{a}}_{-n}\wh{\ol{a}}_n\right).
\end{equation}

\begin{remark}[Heisenberg Lie algebra]
One can consider the Lie $*$-algebra (the Heisenberg Lie algebra)
\begin{equation}\label{l17 Heisenberg Lie algebra}
\mr{Heis}\colon=\mr{Span}_\CC(\{\wh{a}_n\}_{n\in\ZZ},\mathbb{K})
\end{equation}
where $\mathbb{K}$ is the central element and the commutation relations are
\begin{equation}\label{l17 Heis comm rel}
[\wh{a}_n,\wh{a}_m]=n\delta_{n,-m}\mathbb{K}.
\end{equation}
with the involution (hermitian conjugation) acting as $\wh{a}_n^+=\wh{a}_{-n}$, $\mathbb{K}^+=\mathbb{K}$. It is the special case of the general Heisenberg Lie algebra (Definition \ref{l15 def: Heis}), for the symplectic vector space $V$ of Laurent series on $\CC^*$
\begin{equation}
V=\Big\{f(z)=\sum_{n\in \ZZ} f_n z^{-n}\Big\}%\qquad\mr{with}\;\;
%\omega_\mr{symp}(f,g)=\mr{res}_{z=0}(fdg)
\end{equation} 
with symplectic form
\begin{equation}\label{l17 Heis omega_symp}
\omega_\mr{symp}(f,g)=i\;\mr{res}_{z=0}(fdg)
\end{equation}
-- the residue at $z=0$ (i.e. the coefficient of $z^{-1}dz$) of the meromorphic $1$-form $fdg$.\footnote{The normalization factor $i$ in (\ref{l17 Heis omega_symp}) compensates the factor $-i$ in the general definition of Heisenberg Lie algebra (\ref{l15 Heis comm rel}), i.e., one has the commutation relation $[\wh{f},\wh{g}]=\mr{res}_{z=0}(fdg)\;\mathbb{K}$.}
The basis vectors $z^{-n}$ in $V$ correspond to the generators $\wh{a}_n$ of $\mr{Heis}$.

The full Lie algebra of mode operators of the free massless scalar fields  can then be described via two copies $\mr{Heis}$, $\ol{\mr{Heis}}$ of the algebra above:
\begin{equation}
\mr{Span}_\CC(\{\wh\phi_n,\wh\pi_n\}_{n\in\ZZ},\mathbb{K}) =
\frac{\mr{Heis}\oplus \ol{\mr{Heis}}}{\wh{a}_0=\wh{\ol{a}}_0, \mathbb{K}=\ol{\mathbb{K}}} \oplus \CC\cdot \wh{\phi}_0,
\end{equation}
where on the right the extra generator $\wh{\phi}_0$ interacts with the Heisenberg Lie algebras via
\begin{equation}
[\wh{a}_0,\wh{\phi}_0]=-i \mathbb{K}.
\end{equation}
\end{remark}

From (\ref{l17 H via a,abar}) and (\ref{l17 a,abar comm rel}) one easily finds the commutators between the Hamiltonian $\wh{H}$ and the operators $\wh{a}_n$, $\wh{\ol{a}}_n$:
\begin{equation}
[\wh{H}, \wh{a}_n] = -n \wh{a}_n,\quad 
[\wh{H}, \wh{\ol{a}}_n] = -n \wh{\ol{a}}_n,\qquad n\in\ZZ.
\end{equation}
In particular, for $n>0$ applying $\wh{a}_n$ or $\wh{\ol{a}}_n$ to an eigenvector of $\wh{H}$ \emph{decreases} the eigenvalue (total energy of the state) by $n$, while 
applying $\wh{a}_{-n}$ or $\wh{\ol{a}}_{-n}$ \emph{increases} the eigenvalue by $n$. Thus, it is natural to think of $\wh{a}_n,\wh{\ol{a}}_n$ as annihilation operators and of $\wh{a}_{-n},\wh{\ol{a}}_{-n}$ as creation operators.

Next, consider the commutators of $\wh{a}_n$, $\wh{\ol{a}}_n$ with the total momentum operator (\ref{l17 P via a,abar}):
\begin{equation}
[\wh{P}, \wh{a}_n] = -n \wh{a}_n,\quad 
[\wh{P}, \wh{\ol{a}}_n] = +n \wh{\ol{a}}_n,\qquad n\in\ZZ.
\end{equation}
Thus, for $n>0$, applying $\wh{a}_{-n}$ to a joint eigenvector of $\wh{H}$ and $\wh{P}$ increases the energy and the total momentum of the system (creates  -- or adjoins to the system -- a ``left-mover'' -- a quantum with positive momentum), while applying $\wh{\ol{a}}_{-n}$ increases the energy but decreases the total momentum (creates a ``right-mover'').

To summarize, we have the following table for each $n>0$.%\marginpar{I think I mixed up the right- and left-movers. (Left-movers should correspond to holomorphic modes of $\phi_\mr{Heisenberg}$ after Wick rotation.)}

\vspace{0.3cm}
\begin{center}
\begin{tabular}{c|cc}
& annihilation operator & creation operator \\ \hline
left-mover &  $\wh{a}_n$ & $\wh{a}_{-n}$ \\
right-mover & $\wh{\ol{a}}_n$ & $\wh{\ol{a}}_{-n}$
\end{tabular}
\end{center}
\vspace{0.3cm}

\subsubsection{The space of states}. The space of states of the full system (the massless free scalar theory) can be described as the tensor product of the spaces of states for the constituent subsystems:
\begin{equation}\label{l17 H factorized}
\HH=\HH_\mr{free\;particle}\otimes \bigotimes_{n\neq 0} \HH_{\mr{harmonic\;oscillator\;}\omega_n=|n|}.
\end{equation}
One can choose to represent each factor in (\ref{l17 H factorized}) by the Schr\"odinger representation, thereby obtaining a tensor product of countably many copies of $L^2(\RR)$.

A different (better) description of $\HH$ is as a ``Fock space'' -- in the vein of the description (\ref{l16 H osc}) of the space of states of harmonic oscillator as spanned by excitations of a vacuum state given by repeatedly applying creation operators (Verma module description). 
In the case of the free massless scalar field, we pick from the first factor of (\ref{l17 H factorized}) any vector $|\pi_0\rangle$ (cf. (\ref{l17 |p>})), with $\pi_0\in \RR$ the zero-mode momentum, tensored with vacua $|0\rangle$ in each oscillator factor -- we denote the result by abuse of notations again $|\pi_0\rangle$ (this vector is referred to as ``psedovacuum'').\footnote{Note that by construction we have $\wh{a}_{n}|\pi_0\rangle = \wh{\ol{a}}_{n}|\pi_0\rangle =0$ for any $n>0$.} Then we act on $|\pi_0\rangle$ by the creation operators corresponding to different oscillators, creating an excited state; this gives a basis for $\HH$:
\begin{equation}\label{l17 Fock space}
\HH=\bigoplus_{r\geq 0,s\geq 0}\mr{Span}_\CC \Big\{
 \prod_{i=1}^r \wh{a}_{-n_i} \prod_{j=1}^s \wh{\ol{a}}_{-\ol{n}_j} |\pi_0\rangle \;\; \Big| \; 
\begin{array}{c}
1\leq n_1\leq n_2\leq\cdots \leq n_r,\\
1\leq \ol{n}_1\leq \ol{n}_2\leq\cdots \leq \ol{n}_s, \\
\pi_0\in \RR
\end{array}
\Big\}.
\end{equation}
Let us denote the basis vectors spanning $\HH$ by
\begin{equation}\label{l17 basis vector in H}
| \pi_0; \{n_i\},\{\ol{n}_j\}\rangle\colon=\prod_{i=1}^r \wh{a}_{-n_i} \prod_{j=1}^s \wh{\ol{a}}_{-\ol{n}_j} |\pi_0\rangle.
\end{equation}

We think of the basis vector $| \pi_0; \{n_i\},\{\ol{n}_j\}\rangle$ as a multiparticle state, consisting of 
\begin{itemize}
\item $r$ left-moving quanta carrying energy-momentum 2-vectors $(n_i,n_i)$, $i=1,\ldots,r$ and  
\item $s$ right-moving quanta carrying energy-momentum $(\bar{n}_j,-\bar{n}_j)$, $j=1,\ldots, s$.
\end{itemize}
We motivate this interpretation more below, after (\ref{l18 :P: action on vectors}).

\begin{remark}
%\marginpar{This is a bit murky, I should explain better}
Thinking of the system a string moving in the target $\RR$ (for each time $t$, we have a map $\phi\colon \{t\}\times S^1\ra \RR$), the zero-mode momentum $\pi_0$ can be understood as the (target) momentum of the center-of-mass of the string, and has nothing to do with the (source) total momentum $P$.
\end{remark}
%The state $|\pi_0\rangle$ (to which the creation operators )

An equivalent description of $\HH$ as a Fock space (a different way to enumerate the basis vectors) is as follows:
\begin{equation}\label{l17 H occupation numbers}
\HH=\mr{Span}_\CC \left\{  \prod_{n\geq 1}(\wh{a}_{-n})^{k_n} \prod_{\ol{n}\geq 1} (\wh{\ol{a}}_{-\ol{n}})^{\ol{k}_{\ol{n}}}|\pi_0\rangle \;\; \Big|\;\; 
\begin{array}{l}
k_n\geq 0,\ol{k}_{\ol{n}}\geq 0,\\\mr{finitely\;many\;of\;} k_n, \ol{k}_{\ol{n}}\;\mr{are\;nonzero}
\end{array}
 \right\}.
\end{equation}
The numbers $k_n$, $\ol{k}_{\ol{n}}$ are the ``occupation numbers'' for the excitations with energy-momentum $(n,n)$ and $(\ol{n},-\ol{n})$, respectively (i.e. $k_n$, $\ol{k}_{\ol{n}}$ are the numbers of quanta of these types).

\begin{figure}[H]
%$$\vcenter{\hbox{ \includegraphics[scale=0.5]{cob1.eps} }} $$
\begin{center}
\includegraphics[scale=0.6]{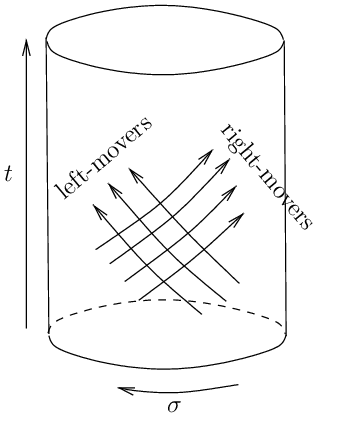}
\end{center}
\caption{Left- and right-movers on a cylinder.
}
\end{figure}

\marginpar{Lecture 18,\\
10/3/2022}

Note that applying the Hamiltonian (\ref{l17 H via a,abar}) to the pseudovacuum $|\pi_0\rangle$ we obtain
\begin{multline}\label{l18 divergency}
\wh{H} |\pi_0\rangle = \wh{\pi}_0^2|\pi_0\rangle+\frac12\sum_{n>0}\Big(\wh{a}_{-n}\underbrace{\wh{a}_n |\pi_0\rangle}_0 + \wh{\ol{a}}_{-n}\underbrace{\wh{\ol{a}}_n |\pi_0\rangle}_0\Big) 
+\\
+ 
\frac12\sum_{n<0}\Big(\underbrace{\wh{a}_{-n}\wh{a}_n}_{-n+\wh{a}_{n}\wh{a}_{-n}} |\pi_0\rangle + \underbrace{\wh{\ol{a}}_{-n}\wh{\ol{a}}_n}_{-n+\wh{\ol{a}}_{n}\wh{\ol{a}}_{-n}} |\pi_0\rangle\Big)
=\Big(\pi_0^2 +\underbrace{\sum_{n<0}(-n)}_{\mr{divergence!}}\Big)|\pi_0\rangle 
\end{multline}
-- $|\pi_0\rangle$ multiplied by a divergent factor. By a similar reason, each basis vector $|\pi_0;\{n_i\},\{\ol{n}_j\}\rangle$ is an eigenvector of $\wh{H}$ with a divergent eigenvalue.
To deal with this problem, one uses normal ordering.

%\marginpar{Uniformize the conventions (FreeAss to Weyl) with harmonic oscillator case.}
\subsubsection{Normal ordering} Normal ordering (in the context of the free massless scalar field) is defined as a $\CC$-linear map  $:\cdots:$ from the free associative algebra generated by the operators $\{\wh{a}_n,\wh{\ol{a}}_n\}_{n\in \ZZ}$ to the Weyl algebra (i.e., to the quotient of the free associative algebra by relations (\ref{l17 a,abar comm rel})). Acting on a word, it puts the annihilation operators $\wh{a}_{>0},\wh{\ol{a}}_{>0}$ to the right and creation operators $\wh{a}_{<0},\wh{\ol{a}}_{<0}$ to the left (and then projects to the Weyl algebra). %We also supplement this prescription by putting $\wh{\phi_0}$ to the 

For example, the normally ordered Hamiltonian (\ref{l17 H via a,abar}) and total momentum operators (\ref{l17 P via a,abar}) are
\begin{eqnarray}
:\wh{H}: &=&  \wh{\pi}_0^2+ \sum_{n>0} \left( \wh{a}_{-n}\wh{a}_n + \wh{\ol{a}}_{-n} \wh{\ol{a}}_n\right), \label{l18 :H:}\\
:\wh{P}: &=&  \sum_{n>0} \left( \wh{a}_{-n}\wh{a}_n - \wh{\ol{a}}_{-n} \wh{\ol{a}}_n \right). \label{l18 :P:}
\end{eqnarray}

Acting with these normally ordered operators on basis vectors (\ref{l17 basis vector in H}), we don't encounter any divergencies (unlike in (\ref{l18 divergency})), and we have
\begin{eqnarray}
:\wh{H}:|\pi_0;\{n_i\},\{\ol{n}_j\} \rangle &=& \left(\pi_0^2+\sum_i n_i +\sum_j \ol{n}_j\right)\;|\pi_0;\{n_i\},\{\ol{n}_j\} \rangle, \\ \label{l18 :P: action on vectors}
:\wh{P}:|\pi_0;\{n_i\},\{\ol{n}_j\} \rangle &=& \left(\sum_i n_i -\sum_j \ol{n}_j\right)\;|\pi_0;\{n_i\},\{\ol{n}_j\} \rangle.
\end{eqnarray}
In particular, all states $|\pi_0;\{n_i\},\{\ol{n}_j\} \rangle$ are eigenvectors of both $:\wh{H}:$ and $:\wh{P}:$. Interpreting the joint eigenvalue as the energy-momentum 2-vector, we see that:
\begin{itemize}
\item The pseudovacuum $|\pi_0\rangle$ has energy-momentum $(\pi_0^2,0)$.
\item  Applying $\wh{a}_{-n}$ with $n>0$ to a state, we increase the energy-momentum by $(n,n)$ (which we interpret as adjoining a left-moving quantum to the system).
\item Applying $\wh{\ol{a}}_{-\ol{n}}$ with $n>0$ to a state, we increase the energy-momentum by $(\ol{n},-\ol{n})$ (which we interpret as adjoining a right-moving quantum).
\end{itemize}

\begin{remark}
There is a single (up to normalization) null-vector of $:\wh{H}:$ in $\HH$ -- the vector 
\begin{equation}
|\vac\rangle \colon= |\pi_0=0 \rangle,
\end{equation}
i.e. the pseudovacuum with $\pi_0=0$. We call this vector the vacuum vector (or vacuum state). It is a null-vector for both $:\wh{H}:$ and $:\wh{P}:$, which is interpreted as invariance of $|0\rangle$ under time-translations and rotation along $S^1$.\footnote{
Time-translation by time $t$ is represented on the space of states by the evolution operator $U(t)=e^{-it \wh{H}}$. Rotation by angle $\theta$ along $S^1$ is similarly represented by $R(\theta)=e^{-i\theta \wh{P}}$.
}
\end{remark}

\begin{remark}
%\marginpar{Too long? Also, put links to later sections where central charge and $Z_{torus}$ are covered.}
Later -- after switching to Euclidean metric -- we will see that the partition function of a torus defined using the normally ordered operators $:\wh{H}:$ and $:\wh{P}:$ does not have the expected modular invariance property (see Section \ref{sss: modular invariance}). To restore it, one should replace $:\wh{H}:$ with the operator $:\wh{H}:-\frac{1}{12}$ (while $:\wh{P}:=\wh{P}$ does not have to be changed), which can be seen as the original operator $\wh{H}$ (\ref{l17 H via a,abar}) with the divergence regularized by Riemann zeta-function regularization:
\begin{multline}
\wh{H}=\frac12\sum_{n\in\ZZ} \left(\wh{a}_{-n}\wh{a}_n + \wh{\ol{a}}_{-n}\wh{\ol{a}}_n \right) = \\
=\frac12 \sum_{n>0}
\left(\wh{a}_{-n}\wh{a}_n + \wh{\ol{a}}_{-n}\wh{\ol{a}}_n \right) + \frac12\sum_{n<0} \left(-2n+\wh{a}_{-n}\wh{a}_n + \wh{\ol{a}}_{-n}\wh{\ol{a}}_n\right)\\ =
:\wh{H}:+\sum_{n>0}n \underbrace{=}_{\mr{zeta-regularization}} 
:\wh{H}: + \lim_{s\ra -1}\sum_{n>0} n^s= 
:\wh{H}: +\zeta(-1)= :\wh{H}:-\frac{1}{12}.
\end{multline}
At the moment this zeta-regularization prescription looks entirely ad hoc, and it is not clear why it should help with modular invariance. Note that with respect to this regularized $\wh{H}$, the vacuum state $|\vac\rangle$ has energy $-\frac{1}{12}$ instead of zero.

%\marginpar{EDIT}
The (somewhat surprising) take-home message for the moment is that the normal ordering breaks conformal invariance (in a mild way\footnote{
The change of the quantum Hamiltonian by a multiple of identity is a somewhat subtle effect: we usually need the commutators with $\wh{H}$, not $\wh{H}$ itself. E.g. 
time-dependence of observables in the Heisenberg picture (\ref{l18 Heisenberg eq}) only depends on commutators with $\wh{H}$.
%the discussion of Heisenberg time-dependent fields below -- Heisenberg equation does not )
}) -- in fact we will not see any problem with normal ordering in the genus zero theory (correlators of point observables on a cylinder/plane) -- they do not contradict conformal invariance, but in genus one we have a problem.

\end{remark}

\subsection{Aside: Schr\"odinger vs Heisenberg picture in quantum mechanics}

In the \ul{Schr\"odinger picture} of quantum mechanics, time-evolution acts on states. I.e., one has time-dependent families of states linked by the evolution operator:
\begin{equation}\label{l18 evol of states}
|\psi\rangle_t = U(t-t_0)|\psi\rangle_{t_0}
\end{equation}
where
\begin{equation}
U(t)=e^{-i\wh{H}t}
\end{equation}
is the unitary evolution operator. 
Put another way, one has a family of the spaces of states $\HH_t$ linked by isomorphisms $\HH_{t}\xra{U(t'-t)}\HH_{t'}$. Observables are operators $\wh{O}$ acting on $\HH_t$ at some particular $t$.

The infinitesimal version of (\ref{l18 evol of states}) is the Schr\"odinger equation
\begin{equation}
\frac{d}{dt}|\psi\rangle_t = -i \wh{H} |\psi\rangle_t
\end{equation}
(we mention it for comparison with the Heisenberg picture).

In the \ul{Heisenberg picture}, evolution acts on observables instead of states. All states are thought of as elements of $\HH_{t_0}$ for some fixed reference time $t_0$. But an observable is understood as a family $\wh{O}_t$ arising as a pullback of some fixed ($t$-independent) operator $\wh{O}$ acting on $\HH_t$, along the evolution $U(t-t_0)\colon \HH_{t_0}\ra \HH_t$: 
\begin{equation}
\wh{O}_t\;\;\calt \;\;\HH_{t_0} \xra{U(t-t_0)} \HH_t\;\; \car\;\; \wh{O}
\end{equation}
I.e., one has
\begin{equation}\label{l18 Heisenberg t-dep observables}
\wh{O}_t=U(t-t_0)^{-1} \,\wh{O}\, U(t-t_0).
\end{equation}
The infinitesimal version of this equation is the Heisenberg equation
\begin{equation}\label{l18 Heisenberg eq}
-i\frac{d}{dt}\wh{O}_t=[\wh{H},\wh{O}_t].
\end{equation}
Below we will use the notation $\wh{O}(t)\colon=\wh{O}_t$ for the time-dependent operators of the Heisenberg picture.

Consider a correlator 
in Schr\"odinger picture (cf. Section \ref{sss point observables in QM}) of quantum mechanics on the source interval (cobordism) $[t_\ii,t_\oo]$, with in/out states $|\psi_\ii\rangle$, $\langle\psi_\oo|$,\footnote{We remind that in Dirac's notation $|\cdots\rangle$ are vectors in $\HH$ and $\langle \cdots |$ are vectors in the linear dual $\HH^*$.}
of observables $\wh{O}_1,\ldots, \wh{O}_n$ inserted at times $t_\ii<t_1<\cdots t_n < t_\oo$. 
\begin{figure}[H]
%$$\vcenter{\hbox{ \includegraphics[scale=0.5]{cob1.eps} }} $$
\includegraphics[scale=0.8]{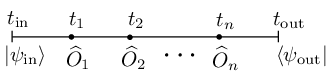}
\caption{Correlator in quantum mechanics.
}
\end{figure}
The correlator is given by
\begin{equation}
\langle \psi_\oo |\; U(t_\oo-t_n) \;\wh{O}_n \cdots  \wh{O}_2\; U(t_2-t_1)\; \wh{O}_1 \;U(t_1-t_\ii)\;| \psi_\ii \rangle
\end{equation}

The same quantity can be equivalently written in Heisenberg picture, as 
\begin{equation}\label{l18 corr Heisenberg}
\langle \til\psi_\oo |\;\wh{O}_n(t_n)\cdots \wh{O}_2(t_2)\;\wh{O}_1(t_1) \; |\til\psi_\ii \rangle
\end{equation}
where $\wh{O}_k(t_k)\colon= U(t_k-t_0)^{-1}\;\wh{O}\; U(t_k-t_0)$ are the time-dependent observables (\ref{l18 Heisenberg t-dep observables}) and 
\begin{equation}
|\til\psi_\ii\rangle=U(t_0-t_\ii) |\psi_\ii\rangle,\qquad
|\til\psi_\oo\rangle=U(t_0-t_\oo) |\psi_\oo\rangle
\end{equation}
are the in-out states expressed as elements of the reference Hilbert space $\HH_{t_0}$. Herethe reference time $t_0$ is chosen arbitrarily.

\begin{remark}
We remark that the product of time-dependent observables $\wh{O}_n(t_n)\cdots \wh{O}_1(t_1)$ in (\ref{l18 corr Heisenberg}) is time-ordered -- the times satisfy $t_n>\cdots > t_1$.

%\marginpar{clean it up a bit?}
When we later consider field theory in Euclidean signature, this will correspond to setting $t=-it_\E$ in the formulae above, with $t_\E>0$. Then the evolution operator $U(T_\E)=e^{-T_\E \wh{H}}$ is non-invertible and only defined for positive $T_\E$. In this situation, \emph{only} time-ordered products of operators are defined. In this setting we should use (\ref{l18 Heisenberg eq}) to define $T_\E$-dependent observables.
\end{remark}

\subsection{Back to free massless scalar field on a cylinder: time-dependent field operator}
Back to the quantum field theory on the cylinder, we think of it as a special model of quantum mechanics, where we understood the space of states (\ref{l17 Fock space}) and we have a family of special operators 
\begin{equation}
\wh{\phi}(\sigma) = \sum_{n\in\ZZ} \wh{\phi}_n e^{in\sigma} =
\wh{\phi}_0+ \sum_{n\neq 0}\frac{i}{n}(-\wh{a}_{-n}+\wh{\ol{a}}_n)e^{in\sigma}
\end{equation}
(one operator for each $\sigma\in S^1$)  acting on $\HH$ and independent of $t$. We can treat these as special examples of observables in the Schr\"odinger picture.
% so far we were thinking within the Schr\"odinger picture. In particular we had a time-independent

The corresponding time-dependent observables in the Heisenberg picture are obtained by solving the equation (\ref{l18 Heisenberg eq}), which yields
\begin{multline}\label{l18 phi time-dep}
\wh\phi(t,\sigma)= \underbrace{e^{i\wh{H}t}}_{U(t)^{-1}} \wh\phi(\sigma)\underbrace{ e^{-i\wh{H}t}}_{U(t)}=\\
=\wh{\phi}_0+2t\wh{\pi}_0+\sum_{n\neq 0}\frac{i}{n}\Big(
-\wh{a}_{-n} e^{in(t+\sigma)}+\wh{\ol{a}}_n e^{in(-t+\sigma)}
 \Big).
\end{multline}
Note the similarity of this formula with the formula for the general solution of the equations of motion in the classical theory (\ref{l17 phi cl solution}).

Then we can consider, e.g., correlators of the form
\begin{equation}
\vcenter{\hbox{ \includegraphics[scale=0.8]{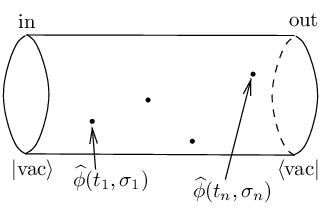} }}  \quad =\quad
\langle \vac | \wh{\phi}(t_n,\sigma_n)\cdots \wh{\phi}(t_1,\sigma_1)|\vac\rangle
\end{equation}
with $t_n>\cdots >t_1$ and with $\sigma_n,\ldots,\sigma_1\in S^1$. These correlators can be explicitly computed using (\ref{l18 phi time-dep}) and using the commutation relations (\ref{l17 a,abar comm rel}). We will discuss such correlators below, once we switch to Euclidean signature.

\section{Free massless scalar field on $\CC$}
\subsection{%Free massless scalar field on 
From Minkowski to
Euclidean cylinder (via Wick rotation), and then to $\CC^*$ (via exponential map)}
\label{sss: free scalar Minkowski cylinder, Eucl cylinder, C^*}
Now let us switch %from the 
the spacetime manifold of the
free massless scalar field from Minkowski cylinder to the cylinder $\Sigma=\RR\times S^1$ with Euclidean metric $g=(d\tau)^2+(d\sigma)^2$. Here we will be denoting the Euclidean time -- the coordinate on $\RR$ -- by $\tau$ (instead of $T_\E$); $\sigma$ is the coordinate on $S^1$ as before. 

Introducing a complex coordinate 
\begin{equation}
\zeta=\tau+i\sigma,\quad \ol\zeta=\tau-i\sigma,
\end{equation}
we can identify $\Sigma$ with $\CC/2\pi i \ZZ$ (where $\zeta$ is the standard coordinate on $\CC$).

Another useful model for $\Sigma$ for us is the punctured complex plane $\CC^*=\CC\backslash\{0\}$ with complex coordinate $z=e^\zeta$. This is in fact the model we will be using the most.

\begin{figure}[H]
%$$\vcenter{\hbox{ \includegraphics[scale=0.5]{cob1.eps} }} $$
\begin{center}
\includegraphics[scale=0.7]{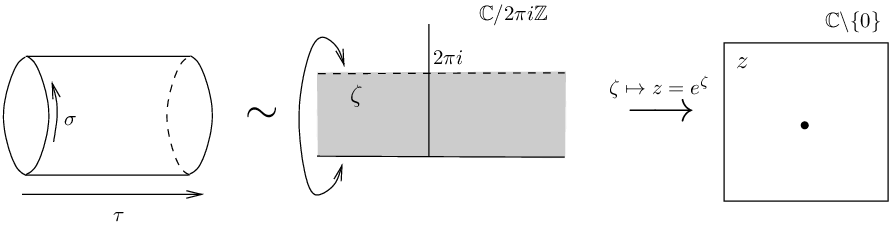}
\end{center}
\caption{Three models of Euclidean cylinder.
}\label{l18 fig: three models of Euclidean cylinder}
\end{figure}

The action functional of the classical theory
is
\begin{equation}\label{l18 S_Eucl}
\begin{aligned}
S_\E(\phi)&=\frac{\kappa}{2}\int_{\RR\times S^1} d\tau d\sigma\,((\dd_\tau\phi)^2+(\dd_\sigma \phi)^2)\\
&= 2\kappa \int_{\CC/2\pi i\ZZ} \frac{i}{2}d\zeta\wedge d\ol\zeta\; \dd_\zeta\phi\,\dd_{\ol\zeta}\phi\\
&= 2\kappa \int_{\CC\backslash\{0\}} \frac{i}{2}dz\wedge d\ol{z}\; \dd_z\phi\,\dd_{\ol{z}}\phi
\end{aligned}
\end{equation}
where $\kappa=\frac{1}{4\pi}$, as before (\ref{l16 kappa value}).

The stress-energy tensor written in the complex coordinates $\zeta,\ol\zeta$ or $z,\ol{z}$ reads
\begin{equation}\label{l18 T in cx coors}
\begin{aligned}
T&= \underbrace{\kappa (\dd_\zeta \phi)^2}_{T_{\zeta\zeta}} (d\zeta)^2 + \underbrace{\kappa (\dd_{\ol\zeta}^2 \phi)}_{T_{\ol\zeta\,\ol\zeta}} (d\ol\zeta)^2\\
&=\underbrace{\kappa (\dd_z \phi)^2}_{T_{zz}} (dz)^2 + \underbrace{\kappa (\dd_{\ol{z}}^2 \phi)}_{T_{\ol{z}\,\ol{z}}} (d\ol{z})^2.
\end{aligned}
\end{equation}

The switch from Minkowski cylinder to Euclidean cylinder is achieved via ``Wick rotation'' -- by substituting 
\begin{equation}
t=-i\tau
\end{equation}
in the formulae for the Minkowski cylinder with $\tau>0$ the Euclidean time. In particular, the evolution operator changes as
\begin{equation}
e^{i\wh{H}t} \rightsquigarrow  e^{-\wh{H}\tau}.
\end{equation}
%from $e^{i\wh{H}t}$ to $e^{-\wh{H}\tau}$. 
The space of states $\HH$ and the quantum Hamiltonian $\wh{H}$ are the same in Minkowski and in Euclidean setting.\footnote{
If we were to retrace our steps and start from the Euclidean action functional, reinterpret it as Lagrangian mechanics, do the Legendre transform to obtain a Hamiltonian description and then canonically quantize, we would have obtained a different quantum Hamiltonian. This has to do with the fact that the rule of canonical quantization (\ref{l15 canonical commutation relations}) is attuned to the unitary evolution; in Euclidean theory the canonical commutation relations have to be changed accordingly. 
}

The time-dependent (Heisenberg) field operator (\ref{l18 phi time-dep}) in Euclidean setting becomes
\begin{equation}
\begin{aligned}
\wh\phi(\zeta) &= \wh\phi_0-i\wh\pi_0 (\zeta+\ol\zeta)+\sum_{n\neq 0}\frac{i}{n} \left(\wh{a}_n e^{-n\zeta}+\wh{\ol{a}}_ne^{-n\ol\zeta}\right)\\
&= 
\wh\phi_0-i\wh\pi_0 \log(z\ol{z})+\sum_{n\neq 0} \frac{i}{n}\left( \wh{a}_n z^{-n}+\wh{\ol{a}}_n\ol{z}^{-n}\right).
\end{aligned}
\end{equation}

\marginpar{Lecture 19,\\ 10/5/2022}

\subsection{Aside: Wick's lemma (in the operator formalism)} \label{sss: Wick's lemma}
Let 
\begin{equation}\label{l19 Heis}
\mc{A}=\mr{Span}_\CC\left(\{\wh{b}_k,\wh{b}_k^+\}_{k\in I},\mathbb{K} \right)
\end{equation} 
be the Heisenberg Lie algebra spanned by pairs of creation/annihilation operators indexed by some set $I$, and the central element $\mathbb{K}$, subject to commutation relations\footnote{We call the creation/annihilation operators here $\wh{b},\wh{b}^+$ to avoid confusion with the operators $\wh{a},\wh{\ol{a}}$ in the scalar field theory -- which are also creation/annihilation operators, just with a different normalization convention.}
\begin{equation}\label{l19 Heisenberg comm rel}
[\wh{b}_i,\wh{b}^+_j]=\delta_{ij}\mathbb{K},\quad [\wh{b}_i,\wh{b}_j]=0,\quad [\wh{b}_i^+,\wh{b}_j^+]=0.
\end{equation}

\begin{remark}
More abstractly, we think of a symplectic vector space $(V,\omega)$ equipped with a compatible complex structure $J\colon V\ra V$, $J^2=-\mr{id}$, with $g(x,y)=\omega(x,Jy)$ a positive-definite bilinear form. (Put another way, $(V,\omega,J,g)$ is a K\"ahler vector space.) Then one has a splitting $\CC\otimes V= U\oplus \ol{U}$ of the complexified space $V$ into the $\pm i$-eigenspaces of  $J$. Then the Lie algebra (\ref{l19 Heis}) is the Heisenberg Lie algebra of $(V,\omega)$ in the sense of Definition \ref{l15 def: Heis}, where we have chosen some basis $\{b_i\}$ in $\ol{U}$ and the dual basis $\{b_i^+\}$ in $U$, which corresponds to creation/annihilation operators $\{\wh{b}_i,\wh{b}_i^+\}$ in $\mc{A}$.
\end{remark}

Let $\{A_p\}_{p\in Y}$ be a collection of ``preferred'' elements of $\mc{A}$ which are some linear combinations of creation/annihilation operators,
$$A_p=\sum_{i\in I}c_{pi} \wh{b}_i+d_{pi} \wh{b}^+_i,$$
with $c_{pi}, d_{pi}$ complex coefficients. The indexing set $Y$ for the collection $\{A_p\}$ is arbitrary; it has no a priori relation to the set $I$ indexing the basis in $\mc{A}$.

We define the normal ordering $:\cdots:$ of an element of the free associative algebra generated by $\{\wh{a}_k,\wh{a}_k^+\}_{k\in I}$ as a $\CC$-linear operation which reorders each word, putting the annihiliation operators $\wh{a}_k$ to the right of the word and creation operators $\wh{a}^+_k$ to the left of the word, and then projects the reordered word to the Weyl algebra of $\mc{A}$,\footnote{
%By the Weyl algebra, we mean the enveloping algebra of $\mc{A}$ modulo the relation $\mathbb{K}=1$, 
Cf. Definition \ref{l15 def Weyl}. Unlike the setup of Section \ref{sss Weyl and Heisenberg}, here we are not thinking about $\hbar\ra 0$ asymptotics (we are in purely quantum theory where we set $\hbar=1$), so we don't consider coefficients in formal power series in $\hbar$.
}
\begin{equation}
\mr{Weyl}(\mc{A})=U\mc{A}/(\mathbb{K}=1).
\end{equation}

For any pair $p,q$ one has the equality 
\begin{equation}\label{l19 Wick propagator}
A_p A_q - :A_p A_q: = g_{pq}
\end{equation}
in the Weyl algebra, with $g_{pq}\in \CC$ some complex numbers; we will (suggestively) refer to the matrix $(g_{pq})_{p,q\in Y}$ as the ``propagator.'' 

The reason for equality (\ref{l19 Wick propagator}), with a multiple of identity on the right, is that it is clearly true if both $A_p$ and $A_q$ are creation or annihilation operators, due to the commutation relations (\ref{l19 Heisenberg comm rel}); by linearity this property extends to $A_p,A_q$ any linear combinations of creation/annihilation operators.

\begin{remark}\label{l19 commutativity under :...:}
Note that the normally ordered products satisfy the symmetry property 
\begin{equation}
:A_{p_1}\cdots A_{p_n}: = :A_{p_{\sigma(1)}}\cdots A_{p_{\sigma(n)}}:
\end{equation}
for $\sigma$ any permutation of the set $\{1,\ldots,n\}$. This property is obvious for $A_p$'s being just creation/annihilation operators, then one extends to general $A_p$'s by $\CC$-linearity.
\end{remark}

The following is a very useful combinatorial statement allowing one to express any element of the Weyl algebra (or the subalgebra generated by the elements $\{A_p\}_{p\in Y}$) in terms of normally ordered elements.
\begin{lemma}[Wick]\label{l19 lemma: Wick}
For $n>0$ and any sequence $p_1,\ldots, p_n\in Y$, one has the following equality in the Weyl algebra:
\begin{multline}\label{l19 Wick formula}
A_{p_1}A_{p_2}\cdots A_{p_n} =\\
= \sum_{
\begin{array}{c}
\{\alpha_1,\beta_1\}\sqcup \cdots \sqcup \{\alpha_s,\beta_s\} \subset \{1,\ldots,n\} \\
\mr{a\;matching\; on\;}\{1,\ldots,n\}
\end{array}
} g_{p_{\alpha_1}p_{\beta_1}}\cdots g_{p_{\alpha_s}p_{\beta_s}} \cdot : \prod_{i\in \{1,\ldots,n\}\backslash \cup_k \{\alpha_k,\beta_k\} } A_{p_i} :.
\end{multline}
The sum here goes over matchings on the set $\{1,\ldots,n\}$ -- collections of non-overlapping 2-element subsets considered up to permutation.
\end{lemma}

\textbf{Examples:}
\begin{itemize}
\item For $n=2$, there are two matchings on the set $\{1,2\}$: $\{1,2\}$ and $\wick{\{\c1{1},\c1{2}\}}$. We indicate by the bracket the matched elements, so in the first case, the set is completely unmatched, $s=0$. In the second case, both elements are matched, $s=1$. So, (\ref{l19 Wick formula}) yields
\begin{equation}
A_a A_b =g_{ab}+ :A_a A_b:
\end{equation}
(we are calling the indices $a,b$ instead of $p_1,p_2$ for convenience). In fact, this formula is just (\ref{l19 Wick propagator}).
\item For $n=3$, the possible matchings are $\wick{\{\c1{1},\c1{2},3\}}$, 
$\wick{\{\c1{1},{2},\c1{3}\}}$, $\wick{\{{1},\c1{2},\c1{3}\}}$,
$\{1,2,3\}$, thus the Wick's formula gives
\begin{equation}\label{l19 Wick n=3 ex}
A_a A_b A_c = g_{ab} A_c + g_{ac}A_b+g_{bc}A_a + :A_a A_b A_c:.
\end{equation}
Note that $:A_p:=A_p$ for any $p\in Y$, so we don't have to write the normal ordering symbol for linear expressions in $A_p$'s.
\item For $n=4$, we have the following possible matchings:
\begin{equation}
\begin{gathered}
\wick{\{\c1{1},\c1{2},\c2{3},\c2{4}\}},\;
\wick{\{\c1{1},\c2{2},\c1{3},\c2{4}\}},\;
\wick{\{\c1{1},\c2{2},\c2{3},\c1{4}\}},\\
\wick{\{\c1{1},\c1{2},{3},{4}\}},\;
\wick{\{\c1{1},{2},\c1{3},{4}\}},\;
\wick{\{\c1{1},{2},{3},\c1{4}\}},\;
\wick{\{{1},\c1{2},\c1{3},{4}\}},\;
\wick{\{{1},\c1{2},{3},\c1{4}\}},\;
\wick{\{{1},{2},\c1{3},\c1{4}\}},\\
\{1,2,3,4\}.
\end{gathered}
\end{equation}
In the first row here we have three \emph{perfect} matchings (i.e. all of the set is matched). Wick's formula in this case gives
\begin{equation}
\begin{gathered}
A_aA_bA_cA_d = \\
=
g_{ab} g_{cd}+g_{ac}g_{bd} + g_{ad}g_{bc}+\\
+g_{ab} :A_c A_d:+ g_{ac} :A_b A_d:+ g_{ad} :A_b A_c: +g_{bc} :A_a A_d:+ g_{bd}:A_a A_c:+ g_{cd}:A_aA_b:+\\
+:A_a A_b A_c A_d:.
\end{gathered}
\end{equation}
\end{itemize}

Wick's lemma is proven by considering $A_1\cdots A_n$ to be a word comprised of only the creation and annihilation operators -- in which case it is proven directly, by induction in $n$. Then the statement is extended to any $A_p$'s by $\CC$-linearity.

\subsection{Propagator for the free massless scalar field on $\CC^*$}
Going back to the free 2d massless scalar field on Euclidean cylinder (which we can parameterize by the complex coordinate $z\in \CC^*$), we are in the setting of Wick's lemma: we have the Weyl algebra generated by creation/annihilation operators $\{\wh{a}_n,\wh{\ol{a}}_n\}_{n\neq 0}\cup\{\wh\phi_0,\wh\pi_0\}$ (we are thinking of $\wh\pi_0$ as annihilation operator and of $\wh\phi_0$ as creation operator w.r.t. the normal ordering) and a family of preferred linear elements 
\begin{equation}\label{l19 phi}
\wh{\phi}(z) = \wh\phi_0-i\wh\pi_0\log(z\ol{z})+\sum_{n\neq 0} \frac{i}{n}\left(\wh{a}_n z^{-n}+\wh{\ol{a}}_n\, \ol{z}^{-n}\right)
\end{equation}
parametetrized by points $z\in \CC^*$. I.e., in the notations of Section \ref{sss: Wick's lemma}, we have $I=\ZZ$ (the indexing set for the basis of creation/annihilation operators) and $Y=\CC^*$ (the indexing set for preferred linear combinations).

\begin{lemma}\label{l19 lemma: boson propagator}
Assume $z,w\in \CC^*$ two points satisfying $|z|\geq|w|$, $z\neq w$. Then one has
\begin{equation}\label{l19 boson propagator}
\wh{\phi}(z)\wh\phi(w)-:\wh{\phi}(z)\wh\phi(w): = -2\log|z-w|.
\end{equation}
\end{lemma}
The right hand side of (\ref{l19 boson propagator})  is the propagator in the sense of (\ref{l19 Wick propagator}).
\begin{proof}
We compute
\begin{multline}
\wh{\phi}(z)\wh\phi(w)-:\wh{\phi}(z)\wh\phi(w):=\\
\hspace{-1cm}
=
\sum_{n,m\neq 0} \frac{i}{n}\cdot \frac{i}{m}\Big( \underbrace{(\wh{a}_n z^{-n}+\wh{\ol{a}}_n\, \ol{z}^{-n}) (\wh{a}_m w^{-m}+\wh{\ol{a}}_m\, \ol{w}^{-m})-
:(\wh{a}_n z^{-n}+\wh{\ol{a}}_n\, \ol{z}^{-n}) (\wh{a}_m w^{-m}+\wh{\ol{a}}_m\, \ol{w}^{-m}):}_I
 \Big)+\\
 +\Big((\wh\phi_0-i\wh\pi_0\log(z\ol{z}))(\wh\phi_0-i\wh\pi_0\log(w\ol{w})) - 
 :(\wh\phi_0-i\wh\pi_0\log(z\ol{z}))(\wh\phi_0-i\wh\pi_0\log(w\ol{w})) : \Big).
\end{multline}
We note that the expression $I$ vanishes if $n\neq m$, since in that case the elements $\wh{a}_n z^{-n}+\wh{\ol{a}}_n\, \ol{z}^{-n}$ and $\wh{a}_m w^{-m}+\wh{\ol{a}}_m\, \ol{w}^{-m}$ commute. Also, $I$ vanishes if $m>0$, because then product $(\wh{a}_n z^{-n}+\wh{\ol{a}}_n\, \ol{z}^{-n}) (\wh{a}_m w^{-m}+\wh{\ol{a}}_m\, \ol{w}^{-m})$ is already normally ordered. That leaves only the terms with $n=-m>0$. So, continuing the computation, we have
\begin{multline}\label{l19 boson propagator computation}
\wh{\phi}(z)\wh\phi(w)-:\wh{\phi}(z)\wh\phi(w):=\\
= \sum_{n>0} \frac{1}{n^2} \Big(\underbrace{[\wh{a}_n,\wh{a}_{-n}]}_nz^{-n}w^n + \underbrace{[\wh{\ol{a}}_n,\wh{\ol{a}}_{-n}]}_n\ol{z}^{-n}\ol{w}^n\Big) - i\underbrace{[\wh\pi_0,\wh\phi_0]}_{-i}\log(z\ol{z}) \\
=\sum_{n>0}\frac{1}{n}\left(\left(\frac{w}{z}\right)^n+\left(\frac{\ol{w}}{\ol{z}}\right)^n\right)-\log(z\ol{z})=
-\log\left(1-\frac{w}{z}\right)-\log\left(1-\frac{\ol{w}}{\ol{z}}\right)-\log(z\ol{z})\\=-2\log|z-w|.
\end{multline}
\end{proof}

Note that the propagator (\ref{l19 boson propagator}) extends to a function on the configuration space of two points $z,w$ on $\CC$ (allowing the point $0$) and this extension is invariant under translations on $\CC$, $(z,w)\mapsto (z+a,w+a)$.

Note also that the convergence behavior of the sum over $n$ in the computation (\ref{l19 boson propagator computation}) is as follows:
\begin{itemize}
\item it converges absolutely if $|z|>|w|$,
\item converges conditionaly if $|z|=|w|$ and $z\neq w$,
\item diverges if $|z|<|w|$ or if $z=w$.
\end{itemize}

\subsection{Correlators on the plane (in the radial quantization formalism)} \label{ss 4.3.4 radial quantization}
One calls the canonical quantization formalism\footnote{We say ``formalism'' where we should really say ``approach to quantization'' or ``method of constructing a quantum field theory out of a classical one.''} for the theory on the cylinder mapped to $\CC^*$ (see Figure \ref{l18 fig: three models of Euclidean cylinder}) the ``radial quantization'' formalism.

We define the \ul{radial ordering} of a product of local operators (observables) on $\CC^*$ inserted at $n$ \emph{distinct} points $z_1,\ldots,z_n\in \CC^*$ as follows:
\begin{equation}\label{l19 radial ordering}
\mc{R}\left( \wh{O}_1(z_1) \cdots  \wh{O}_n(z_n)  \right)\colon = \wh{O}_{\sigma(1)}(z_{\sigma(1)})\cdots \wh{O}_{\sigma(n)}(z_{\sigma(n)}) ,
\end{equation}
where $\sigma\in S_n$ is a permutation of indices such that $|z_{\sigma(1)}|\geq \cdots \geq |z_{\sigma(n)}| $. 

Examples of local operators $\wh{O}_k(z)$ are:
\begin{itemize}
\item The field operator $\wh\phi(z)$.
\item Any derivative of the field operator $\dd_{z}^r \dd_{\ol{z}}^s \wh\phi(z)$, with $r,s\geq 0$.
\item Any normally ordered differential polynomial in $\wh\phi(z)$, e.g., $:\dd_z \wh\phi(z)\, \dd_{\ol{z}}^2\wh\phi(z):$.
\end{itemize}

\begin{remark}
Local operators at the same radius commute:
\begin{equation}\label{l19 same-radius commutation}
[\wh{O}_1(z),\wh{O}_2(w)]=0\qquad\mr{if}\;\; |z|=|w|,\; z\neq w.
\end{equation}
This can be seen as the spacial locality property. In the example of free scalar field, for local operators as in the list above, (\ref{l19 same-radius commutation}) is a consequence of (\ref{l19 boson propagator}). This remark shows that the possible ambiguity of radial ordering arising when several of $z_i$'s have the same absolute value does not affect the right hand side of (\ref{l19 radial ordering}).
\end{remark}

\begin{example}
If $z_1$, $z_2$, $z_3$ are three points on $\CC^*$ with absolute values satisfying $|z_2|>|z_3|>|z_1|$ and $\wh{O}_{1,2,3}$ are some local operators, then one has
\begin{equation}
\mc{R} \Big(\wh{O}_1(z_1)\wh{O}_2(z_2) \wh{O}_3(z_3)\Big) = \wh{O}_2(z_2)\wh{O}_3(z_3) \wh{O}_1(z_1).
\end{equation}
In particular, one can consider the vacuum expectation value of this expression
\begin{equation}
\langle\vac|\mc{R} \Big(\wh{O}_1(z_1)\wh{O}_2(z_2) \wh{O}_3(z_3)\Big)|\vac\rangle = \langle\vac|\wh{O}_2(z_2)\wh{O}_3(z_3) \wh{O}_1(z_1)|\vac\rangle.
\end{equation}
%Note that 
Only with this ordering in the right-hand side this is guaranteed to be a well-defined expression.
\begin{figure}[H]
%$$\vcenter{\hbox{ \includegraphics[scale=0.5]{cob1.eps} }} $$
\begin{center}
\includegraphics[scale=0.7]{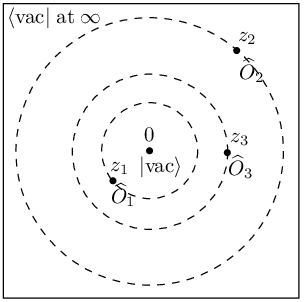}
\end{center}
\caption{Radial ordering.
}
\end{figure}
\end{example}

\begin{remark}
%\footnote{
One can see the necessity of radial ordering (for convergence of a product of local operators $\wh{O}_1(z_1)\cdots \wh{O}_k(z_k)$ -- more precisely, for the matrix element $\langle\vac|\cdots |\vac\rangle$ of such a product to exist) by converting back from Heisenberg to Schr\"odinger pricture. Then %we will see that 
the operators $\wh{O}_k^{\tiny \mbox{Schr\"odinger}}$ are joined by the evolution operators  $U(\log\frac{|z_{k}|}{|z_{k+1}|})$
%$U(\log\frac{|z_2|}{|z_3|})$, $U(\log\frac{|z_3|}{|z_1|})$ 
and only for a positive Euclidean time $\tau$ the evolution operator $U(\tau)=e^{-\tau\wh{H}}$ is well-defined. More precisely: for $\tau>0$, $U(\tau)$ is a smoothing operator.

The other way to see that radial ordering is necessary for convergence is to apply Wick's lemma to the product of local operators. Then we will have a computation similar to (\ref{l19 boson propagator computation}) where the infinite sum will converge if and only if the operators are radially ordered.

A related comment is that the vector $\prod_{i=1}^{n} \wh{O}_{i}(z_i)|\vac\rangle$ (assuming that it exists) is certain to be in the domain of a local operator $\wh{O}(z)$ if and only if $|z_i|\leq |z|$ and $z\neq z_i$ for $i=1,\ldots,n$. Using this argument inductively in $n$, one arrives to the necessity of radial ordering.
%}\marginpar{move the footnote into the main text? clean it up a bit?}
\end{remark}

\begin{definition}
In operator formalism, we will understand the correlator  of several local operators (point observables) $\wh{O}_1,\ldots,\wh{O}_n$ inserted at pairwise distinct points $z_1,\ldots,z_n\in \CC^*$ as the expression\footnote{
In the left hand side we think of $O_k$'s as elements of the abstract vector space $V$ of point observables in the sense of Section \ref{ss: CFT as a system of correlators},
%(of which we don't think as operators acting on any Hilbert space and which do not depend on points $z_k$), 
placed at points $z_1,\ldots,z_n$. In the r.h.s. these abstract elements are represented by operators acting on the space of states -- we denote these operators by hats. 
%From that viewpoint, we are defining the correlator by the operator expression in the r.h.s., so we don't decorate t

Also, in the path integral formalism, one can think of the l.h.s. 
as a product of classical observavles (functions of jets of classical fields at a point) averaged over the space of classical fields, cf. (\ref{l3 PI corr on Sigma closed}), (\ref{l4 corr path integral}).
%We are not putting hats on observables in the left hand side, because there we think of them as classical objects (functions of jets of classical fields at a point), and we think of the symbol $\langle \cdots \rangle$ in the l.h.s. as averaging over the space of classical fields, cf. (\ref{l3 PI corr on Sigma closed}), (\ref{l4 corr path integral}).

%Alternatively, we can say that $O_k$'s in the l.h.s. are elements of the abstract vector space $V$ of point observables in the sense of Section \ref{ss: CFT as a system of correlators} (of which we don't think as operators acting on any Hilbert space and which do not depend on points $z_k$), placed at points $z_1,\ldots,z_n$. From that viewpoint, we are defining the correlator by the operator expression in the r.h.s.
}
\begin{equation}
\langle O_1(z_1)\cdots O_n(z_n) \rangle\colon =\langle \vac| \mc{R}\left(  \wh{O}_1(z_1)\cdots  \wh{O}_n(z_n)  \right)   |\vac\rangle.
\end{equation}

\end{definition}

\begin{example}
Lemma \ref{l19 lemma: boson propagator} implies 
\begin{equation}\label{l19 R phi phi}
\mc{R}(\wh\phi(z) \wh\phi(w))= :\wh\phi(z) \wh\phi(w): -2\log|z-w|
\end{equation}
for any $z\neq w \in \CC^*$. Note that the normally ordered expression in the r.h.s. does is invariant under swapping $z$ and $w$ (cf. Remark \ref{l19 commutativity under :...:}).
\end{example}

\begin{example}{Two-point correlator of $\wh\phi$.}\label{l19 ex <phi phi>}
From (\ref{l19 R phi phi}) we find
\begin{equation}\label{l19 <phi phi> radial quantization}
\langle \phi(z)\phi(w) \rangle \colon= \langle \vac| \mc{R}\Big(\wh\phi(z) \wh\phi(w)\Big) |\vac\rangle = -2\log|z-w|+C
\end{equation}
where \marginpar{Lecture 20,\\ 10/7/2022}
\begin{equation}\label{l19 C}
C=\langle \vac|  :\wh\phi(z) \wh\phi(w): |\vac\rangle = \langle \vac| \wh\phi_0^2 |\vac\rangle
\end{equation}
Here we expand $ :\wh\phi(z) \wh\phi(w):$ using (\ref{l19 phi}). All terms in the expansion (except the term $\wh\phi_0^2$) contain $\wh{a}_{\geq 0}$ or $\wh{\ol{a}}_{\geq 0}$ on the right which yields zero when acting on $|\vac\rangle$, and/or contain $\wh{a}_{<0}$, $\wh{\ol{a}}_{<0}$ on the left, which vanishes when paired with $\langle\vac|$.

Note that (\ref{l19 C}) is an ill-defined expression formally independent of $z,w$ -- an ``infinite constant.'' This can be seen by examining the Schr\"odinger representation for the free particle (the zero-mode) where $|\vac\rangle = |\pi_0\rangle$ is represented by the Dirac delta-distribution $\delta(\pi_0)$ and $\wh\phi_0=i\frac{\dd}{\dd \pi_0}$. Thus, the expression $\langle \vac| \wh\phi_0^2 |\vac\rangle$ in Schr\"odinger representation reads ``the evaluation of distribution $\delta''(\pi_0)$ at $\pi_0=0$.'' This evaluation does not exist. 

Put differently, $\wh\phi_0$ is an unbounded operator on $\HH$ and the vector $|\vac\rangle$ is not in its domain.
\end{example}

\textbf{Notation.} In this section and onward we will be denoting the holomorphic derivative $\dd_z$ by $\dd$ and the antiholomorphic derivative $\dd_{\bar z}$ by $\bar\dd$. Thus, 
%in this section 
symbols $\dd$ and $\bar\dd$ no longer stand for the holomorphic/antiholomorphic Dolbeault operators $dz\,\dd_z$, $d\bar{z}\,\dd_{\bar{z}}$.

To summarize, correlators of the field $\phi$ are ill-defined due to the presence of the zero-mode $\wh\phi_0$. However, correlators of the fields $\dd\phi$, $\bar\dd\phi$ are well-defined!

Note that from (\ref{l19 phi}) one has the following nice expansions of the derivatives of the field in terms of creation/annihilation operators:
\begin{equation}\label{l20 dd phi, dbar phi via a abar}
i\dd\wh\phi(z)=\sum_{n\in\ZZ} \wh{a}_n z^{-n-1},\qquad i\bar\dd\wh\phi(z)=\sum_{n\in\ZZ} \wh{\ol{a}}_n \bar{z}^{-n-1}.
\end{equation}

\begin{example} For the two-point correlator of derivatives of the field  we have
\begin{multline}\label{l20 <dd phi  dd phi>}
\langle \dd\phi(z) \dd\phi(w) \rangle\colon =  \langle \vac| \mc{R}\Big(  \dd\wh\phi(z) \dd\wh\phi(w)\Big) |\vac\rangle =
\langle\vac|\dd_z\dd_w\underbrace{\mc{R}\Big(\wh\phi(z)\wh\phi(w)\Big)}_{-2\log|z-w|+:\wh\phi(z)\wh\phi(w):}   |\vac\rangle =
\\
= \langle \vac| 
\underbrace{-\frac{1}{(z-w)^2}}_{\dd_z \dd_w(-2\log|z-w|)}+:\dd\wh\phi(z)\dd\wh\phi(w):|\vac\rangle = -\frac{1}{(z-w)^2}.
\end{multline}
Here $z\neq w$ are any two distinct points in $\CC\backslash\{0\}$. %Here we first used the derivative $\dd_z\dd_w$ of 
We used the fact that ${:\dd\wh\phi(z)\dd\wh\phi(w):}$, when expanded using (\ref{l20 dd phi, dbar phi via a abar}), has only terms with $\wh{a}_{\geq 0}$ or $\wh{\ol{a}}_{\geq 0}$ on the right, and/or with $\wh{a}_{\leq 0}$, $\wh{\ol{a}}_{\leq 0}$ on the left. Hence the vacuum expectation value ${\langle\vac|:\dd\wh\phi(z)\dd\wh\phi(w):|\vac\rangle}$ is zero.

By similar reasoning one has
\begin{equation}
\langle \bar\dd\phi(z) \bar\dd\phi(w) \rangle =-\frac{1}{(\bar{z}-\bar{w})^2}
\end{equation}
and
\begin{equation}\label{l20 <dd phi dbar phi>}
\langle \dd\phi(z) \bar\dd\phi(w) \rangle =0.
\end{equation}
We stress again that points $z$ and $w$ are assumed be distinct.\footnote{It is possible to sense of the correlator (\ref{l20 <dd phi dbar phi>}) as a distibution on $\CC\times \CC$ rather than as a function on the open configuration space $C_2(\CC\backslash\{0\})$. Then the correlator becomes  
$\langle \dd\phi(z) \bar\dd\phi(w) \rangle={\pi\delta(z-w)}$ -- up to normalization, the Dirac delta-distribution supported on the diagonal $\mr{Diag}\subset \CC\times \CC$. This delta-distribution is an example of so-called ``contact term.''
}%\marginpar{put the correct normalization factor}
\end{example}

One can proceed to compute several-point correlators of observables $\dd\phi$, $\bar\dd\phi$ using Wick's lemma.
\begin{example}
For the four-point correlator, one finds
\begin{multline}\label{l20 4-point computation}
\langle \dd\phi(z_1)\dd\phi(z_2)\dd\phi(z_3)\dd\phi(z_4) \rangle\colon = 
\langle \vac| \mc{R}\Big( 
 \dd\wh\phi(z_1)\dd\wh\phi(z_2)\dd\wh\phi(z_3)\dd\wh\phi(z_4)
\Big) |\vac\rangle =\\
\langle \vac| \Big(
\wick{\c1{\dd\wh\phi(z_1)}\c1{\dd\wh\phi(z_2)}\c2{\dd\wh\phi(z_3)}\c2{\dd\wh\phi(z_4)} }
+\wick{\c1{\dd\wh\phi(z_1)}\c2{\dd\wh\phi(z_2)}\c1{\dd\wh\phi(z_3)}\c2{\dd\wh\phi(z_4)} }+
\wick{\c1{\dd\wh\phi(z_1)}\c2{\dd\wh\phi(z_2)}\c2{\dd\wh\phi(z_3)}\c1{\dd\wh\phi(z_4)} } +\\
+\wick{\c1{\dd\wh\phi(z_1)}\c1{\dd\wh\phi(z_2)}\cancel{:{\dd\wh\phi(z_3)}{\dd\wh\phi(z_4)}:}  }+\mbox{5 similar terms}+
\cancel{: \dd\wh\phi(z_1)\dd\wh\phi(z_2)\dd\wh\phi(z_3)\dd\wh\phi(z_4):} \Big)
|\vac\rangle\\
=\frac{1}{z_{12}^2z_{34}^2}+\frac{1}{z_{13}^2z_{24}^2}+\frac{1}{z_{14}^2z_{23}^2}.
\end{multline}
Here we denoted $z_{ij}\colon=z_i-z_j$. Note that in this computation only the three terms where all four operators are matched contribute.

By a similar computation one finds
\begin{equation}
\langle \dd\phi(z_1) \dd\phi(z_2)\bar\dd\phi(z_3)\bar\dd\phi(z_4) \rangle= 
\wick{\langle \c1{\dd\phi(z_1)} \c1{\dd\phi(z_2)}\c2{\bar\dd\phi(z_3)}\c2{\bar\dd\phi(z_4)} \rangle} = \frac{1}{z_{12}^2 \bar{z}_{34}^2}
\end{equation}
-- only a single matching contributes.

More generally, by the same logic, %as (\ref{l20 4-point computation}), one has that 
the correlator 
\begin{equation}\label{l20 general corr of dd phi}
\langle \dd\phi(z_1)\cdots \dd\phi(z_n) \bar\dd \phi(w_1)\cdots \bar\dd\phi(w_m) \rangle
\end{equation}
(all points $z_1,\ldots,z_n,w_1,\ldots,w_m$ are assumed to be distinct)
vanishes unless both $n$ and $m$ are even, $n=2\nu$, $m=2\mu$. If they are even, the correlator is given by a sum over pairs (perfect matching of  $z_i$'s, perfect matching of $w_j$'s) -- thus, in total there are $(2\nu-1)!!\cdot (2\mu-1)!!$ terms.  
E.g. in the case $m=0$, one obtains a meromorphic function on $\CC^n$ with second-order poles on principal diagonals. For instance,
\begin{equation}
\langle \phi(z_1)\cdots \phi(z_6) \rangle = \frac{-1}{z_{12}^2z_{34}^2 z_{56}^2}+ \mbox{14 similar terms},
\end{equation}
since one has $5!!=5\cdot 3\cdot 1$ perfect matchings on the set of $6$ elements.

Examining the terms contributing to the correlator (\ref{l20 general corr of dd phi}) for general $m,n$, we can notice that it factorizes into a meromorphic part and an antimeromorphic part:
\begin{equation}
\langle \dd\phi(z_1)\cdots \dd\phi(z_n) \bar\dd \phi(w_1)\cdots \bar\dd\phi(w_m) \rangle=
\langle \dd\phi(z_1)\cdots \dd\phi(z_n) \rangle \cdot
\langle \bar\dd \phi(w_1)\cdots \bar\dd\phi(w_m) \rangle
\end{equation}
\end{example}

\begin{example} We have
\begin{equation}
\langle \dd\bar\dd\phi(z) \dd\phi(w) \rangle = \frac{\dd}{\dd \bar z}\underbrace{ \langle \dd\phi(z) \dd\phi(w) \rangle }_{-\frac{1}{(z-w)^2}} = 0.
\end{equation}
\end{example}

For the next example we need a slightly enhanced version of Wick's lemma, rearranging a product of normally-ordered words in terms of fully normally-ordered expressions.
\begin{lemma} In the notations of Lemma \ref{l19 lemma: Wick}, for $n>0$, let $p_1,\ldots,p_n\in Y$ and  let %$S_1,\ldots, S_m $ be a collection of finite disjoint sets
\begin{equation}\label{l20 enh Wick partitioning}
\{1,\ldots, n\}=S_1\sqcup \cdots \sqcup S_m
\end{equation} 
be a  partitioning of the set $\{1,\ldots,n\}$ into nonempty disjoint subsets  $S_j$.
Then one has the following equality in the Weyl algebra:
\begin{multline}\label{l20 Wick enhanced}
:\prod_{p\in S_1} A_{p}: \cdots :\prod_{p\in S_m} A_{p}: =\\
=\sum_{
\begin{array}{c}
\{\alpha_1,\beta_1\}\sqcup\cdots\sqcup \{\alpha_s,\beta_s\}\subset \{1,\ldots,n\}\\
\mr{a\;matching\;on\;}\{1,\ldots,n\}\;\mr{s.t.} \\
\{\alpha_i,\beta_i\} \not\subset S_j\;\forall i,j
\end{array}
}
\prod_{i=1}^sg_{p_{\alpha_i}p_{\beta_i}} \cdot : \prod_{i\in\{1,\ldots,n\}\backslash \cup_k \{\alpha_k,\beta_k\}}A_{p_i} :.
\end{multline} 
\end{lemma}
In other words, the right hand side is the sum over matchings, as in (\ref{l19 Wick formula}), except that now elements of each subset of $S_j$ of labels corresponding to one of the normally-ordered words in the l.h.s. are not allowed to be matched.

\textbf{Example: }
\begin{equation}
:A_a A_b: A_c = g_{ac} A_b+g_{bc}A_a + :A_aA_bA_c:. 
\end{equation}
Here the partitioning (\ref{l20 enh Wick partitioning}) is $\{1,2,3\}=\{1,2\}\sqcup \{3\}$ and the labels are $p_1=a$, $p_2=b$, $p_3=c$.
Notice that in comparison with (\ref{l19 Wick n=3 ex}), the term $g_{ab}A_c$ corresponds to a prohibited contraction $\wick{\{\c1{1},\c1{2},3\}}$ and doesn't appear in the r.h.s.

\begin{example}
%Assume for simplicity that four points $z_1,\ldots,z_k\in\CC\backslash\{0\}$ are radially ordered, $|z_1|>|z_2|>|z_3|>|z_4|>0$. 
Consider the correlator
\begin{multline}\label{l20 4-point corr with :...:}
\langle \dd\phi(z_1)\big( :\dd\phi(z_2)\dd\phi(z_3):\big)\dd\phi(z_4) \rangle\colon=
\langle \vac | 
\mc{R}\Big(
\dd\wh\phi(z_1)\big( :\dd\wh\phi(z_2)\dd\wh\phi(z_3):\big)\dd\wh\phi(z_4)
\Big)
|\vac\rangle=\\
=\langle \vac | 
\wick{\c1{\dd\wh\phi(z_1)}\big( :\c1{\dd\wh\phi(z_2)}\c2{\dd\wh\phi(z_3)}:\big)\c2{\dd\wh\phi(z_4)}}
+
\wick{\c1{\dd\wh\phi(z_1)}\big( :\c2{\dd\wh\phi(z_2)}\c1{\dd\wh\phi(z_3)}:\big)\c2{\dd\wh\phi(z_4)}}
|\vac\rangle\\
=\frac{1}{z_{12}^2z_{34}^2}+\frac{1}{z_{13}^2z_{24}^2}.
\end{multline}
Here the 2-point observable $ :\dd\wh\phi(z_2)\dd\wh\phi(z_3):$ on the l.h.s. is a formal symbol defined by its correlators with other local fields, like in this example, where this observable is replaced in the operator language by a normally-ordered product of derivatives of field operators.
Notice in comparison with (\ref{l20 4-point computation}) the absence of the term $\frac{1}{z_{23}^2z_{14}^2}$ corresponding to a prohibited matching. In particular, (\ref{l20 4-point corr with :...:}) is a regular (in fact, holomorphic) function on the diagonal $z_2\ra z_3$ in $\CC^4$.
\end{example}

This example illustrates that one can define a new point observable 
\begin{equation}
:\dd\phi(z)\dd\phi(z):\;\;\colon= \lim_{w\ra z} :\dd\phi(w) \dd\phi(z): = \lim_{w\ra z} \left(\dd\phi(w)\dd\phi(z)+\frac{1}{(w-z)^2}\right).
\end{equation}
One sometimes calls point observables of this type -- constructed as normally-ordered  differential polynomials in the field -- ``composite fields.''
%expressions in terms of the (derivates of) the field
This definition is understood as an equality under a correlator with an arbitrary collection of other local observables (``test observables'') inserted at points  $\neq z$.
In the operator language, we should replace $\phi\ra\wh\phi$ everywhere.

This new observable has well-defined correlators. E.g., taking the limit $z_2\ra z_3$ in (\ref{l20 4-point corr with :...:}), we obtain
\begin{equation}
\langle \dd\phi(z_1)\big(:\dd\phi(z_3)\dd\phi(z_3):\big)\dd\phi(z_4) \rangle = \frac{2}{z_{13}^2 z_{34}^2}.
\end{equation}

\begin{definition}
We can define the (quantum) holomorphic/antiholomorphic 
 stress-energy tensor in the massless scalar field theory as the composite fields
\begin{equation}\label{l20 T components}
T_{zz}\colon= -\frac12 :\dd\phi(z)\dd\phi(z):\quad,\quad
T_{\bar{z}\bar{z}}\colon= - \frac12 :\bar\dd\phi(z)\bar\dd\phi(z):.
%T(z)=\underbrace{-\frac12 :\dd\phi(z)\dd\phi(z):}_{T_{zz}} (dz)^2 \underbrace{- \frac12 :\bar\dd\phi(z)\bar\dd\phi(z):}_{T_{\bar{z}\bar{z}}} (d\bar{z})^2
\end{equation}
\end{definition}
Note that the conventional normalization factor in (\ref{l20 T components})  is  different than what we had in the classical theory (\ref{l18 T in cx coors}). It is useful to also consider a local observable 
\begin{equation}\label{l20 T^total}
T^\mr{total}(z)=T_{zz}(dz)^2+T_{\bar{z}\bar{z}}(d\bar{z})^2
\end{equation}
valued in quadratic differentials -- the total (quantum) stress-energy tensor. For instance, its correlator with, e.g. a collection of fields $\dd\phi(z_i)$ will be a section of the pullback of the bundle of quadratic differentials $K^{\otimes 2}\oplus \ol{K}^{\otimes 2}\ra \Sigma$ to the space of configurations of points $(z,z_1,\ldots,z_n)\in\Sigma$, with $K=(T^{1,0})^*\Sigma$ the canonical line bundle. Here $\Sigma=\CC\backslash\{0\}$. %(but could in fact be a general Riemann surface).

\textbf{Notation.} 
From now on we will denote $T_{zz}$ by $T$ and $T_{\bar{z}\bar{z}}$ by $\bar{T}$. This is the standard convention in the literature on CFT.
%In the literature on CFTs, it is customary to denote $T_{zz}$ by $T$ and $T_{\bar{z}\bar{z}}$ by $\bar{T}$ (which conflicts with our old use of $T$ as a notation for the whole stress-energy tensor). We will use these standard notations from now on.

\marginpar{Lecture 21,\\ 10/10/2022
}

\section{Operator product expansions}
Recall from Section \ref{sss: OPEs intro} that the operator product expansions (OPEs) express the product of two local observables at points $z,w$ as a linear combination (with singular coefficients) of single local observables at $w$, in the asymptotics $z\ra w$. These expressions are to be substituted in a correlator with an arbitrary collection of ``test'' local observables at points $z_1,\ldots,z_n \neq z,w$ and control the asymptotics of the correlator as $z\ra w$.

\begin{example}
From Wick's lemma we have the equality
\begin{equation}\label{l21 eq1}
\mc{R}\, \dd\wh\phi(z) \dd\wh\phi(w) = -\frac{1}{(z-w)^2}\wh{\mathbb{1}} + :\dd\wh\phi(z) \dd\wh\phi(w):
\end{equation}
for any $z\neq w\in\CC\backslash \{0\}$, as equality of linear operators on $\HH$. Here for the moment we make the identity operator $\wh{\mathbb{1}}$ explicit in the notations. Note that the second term is regular\footnote{Generally, ``regular'' for us in the context of OPEs means just ``continuous.''} (in fact, holomorphic) as $z\ra w$. Thus, for any collection of point observables $O_1,\ldots,O_n$ 
%(``test'' observables) 
at points $z_1,\ldots,z_n$ (distinct among themselves and distinct from $w$), one has 
\begin{multline}
\langle \dd\phi(z)\dd\phi(w) O_1(z_1)\cdots O_n(z_n) \rangle = \langle \vac| \mc{R} \Big(  \dd\wh\phi(z)\dd\wh\phi(w) \wh{O}_1(z_1)\cdots \wh{O}_n(z_n) \Big)  |\vac\rangle \underset{z\ra w}{\sim} \\
\underset{z\ra w}{\sim}  
-\frac{1}{(z-w)^2} \langle \vac|\mc{R}\Big( \wh{\mathbb{1}}(w)\wh{O}_1(z_1)\cdots \wh{O}_n(z_n) \Big)| \vac\rangle + \mr{reg.}\\
=
-\frac{1}{(z-w)^2} \langle \mathbb{1}(w) {O}_1(z_1)\cdots {O}_n(z_n) \rangle + \mr{reg.}
\end{multline}
-- this is an asymptotic expression for the correlator as $z\ra w$ giving the principal part of its Laurent expansion in $z-w$; $\mr{reg.}$ stands for a term with regular behavior as $z\ra r$. The identity operator operator $\wh{\mathbb{1}}$ and identity field $\mathbb{1}$ do not affect the correlators in the r.h.s.

Thus, one has the operator product expansion
%\marginpar{$\sim$ vs. $=$ is inconsistent throughout this section}
\begin{equation}\label{l21 del phi del phi OPE}
\dd\phi(z) \dd\phi(w) \sim -\frac{\mathbb{1}}{(z-w)^2}+\mr{reg.}
\end{equation}
%\marginpar{Am I being completely redundant here?}
The symbol $\sim$ means that one can trade the l.h.s. with the r.h.s. under a correlator with test observables, yielding the asymptotics as $z\ra w$.
\end{example}

\begin{remark}
One can also be more explicit about the regular part: one can write the rightmost term in (\ref{l21 eq1}) as
\begin{equation}
:\dd\wh\phi(z)\dd\wh\phi(w):=\sum_{n\geq 0}\frac{1}{n!}(z-w)^n :\dd^{n+1}\wh\phi(w)\, \dd\wh\phi(w):.
\end{equation}
The refined version of the OPE (\ref{l21 del phi del phi OPE}) is then
\begin{equation}
\dd\phi(z) \dd\phi(w) \sim -\frac{\mathbb{1}}{(z-w)^2}+\underbrace{\sum_{n\geq 0 }\frac{1}{n!}(z-w)^n :\dd^{n+1}\phi(w)\, \dd\phi(w):}_{\mr{reg.}}.
\end{equation}
The r.h.s. is now a linear combination of local composite fields at the point $w$. Under a correlator with test observables, one has
%\footnote{
%``Fields'' now stands for ``composite fields'' or ``local observables'' -- elements of the vector space $V$ of Section \ref{ss: CFT as a system of correlators}, not fields of the classical field theory.
%}
\begin{multline}
\langle \dd\phi(z)\dd\phi(w) O_1(z_1)\cdots O_n(z_n) \rangle \underset{z\ra w}{\sim} \\
\underset{z\ra w}{\sim}
-\frac{1}{(z-w)^2} \langle \mathbb{1}(w) O_1(z_1)\cdots O_n(z_n) \rangle+\sum_{n\geq 0} \frac{1}{n!}(z-w)^n\langle  :\dd^{n+1}\phi(w)\dd\phi(w):\, O_1(z_1)\cdots O_n(z_n)\rangle. 
\end{multline}
The sum on the right converges absolutely if and only if $|z-w|< \min\{|z_i-w|\}_{i=1}^n$ and in this convergence radius is equal to the l.h.s. Thus, the $\sim$ symbol here is actually equality, for $z$ sufficiently close to $w$ (closer than any of the test observables).
\end{remark}

Similarly to (\ref{l21 eq1}) (or (\ref{l21 del phi del phi OPE})), one finds
\begin{equation}
\mc{R} \bar\dd\wh\phi(z) \bar\dd\wh\phi(w)\sim-\frac{\wh{\mathbb{1}}}{(\bar{z}-\bar{w})^2}+\mr{reg.},\qquad \mc{R} \dd\wh\phi(z) \bar\dd\wh\phi(w)\sim \mr{reg.}
\end{equation}
These are again equalities of operators on $\HH$; removing the hats and the radial ordering sign, we have the OPEs in the form similar to (\ref{l21 del phi del phi OPE}) -- in the language of abstract correlators of observables as elements of $V$ (of Section \ref{ss: CFT as a system of correlators}).

\begin{example}
As the next example, consider the OPE between the stress-energy tensor and $\dd\phi$. From Wick's lemma we find
\begin{multline}
\mc{R}\underbrace{\wh{T}(z)}_{:-\frac12 \dd\wh\phi(z)\dd\wh\phi(z):} \dd\wh\phi(w) =\\
=\wick{-\frac12 :\dd\wh\phi(z)\dd\c1{\wh\phi(z)}: \dd\c1{\wh\phi(w)}}+\wick{-\frac12 :\dd\c1{\wh\phi(z)}\dd\wh\phi(z): \dd\c1{\wh\phi(w)}}+\underbrace{:-\frac12 \dd\wh\phi(z)\dd\wh\phi(z)\dd\wh\phi(w):}_{\mr{reg.}}\\
\sim\frac{\dd\wh\phi(z)}{(z-w)^2}+\mr{reg.} 
\end{multline}
This is not quite the desired OPE yet, as the operator in the r.h.s is at $z$ whereas we want to express the operator productin terms of local operators at $w$. This is remedied by expanding $\dd\wh\phi(z)$ in Taylor series centered at $w$: $\dd\wh\phi(z)=\dd\wh\phi(w)+(z-w)\dd^2\wh\phi(w)+O((z-w)^2)$.\footnote{
Here we used the fact that $\dd\wh\phi(z)$ is holomorphic in $z$, see (\ref{l20 dd phi, dbar phi via a abar}), thus, e.g., one does not have a term $(\bar{z}-\bar{w})\dd\bar\dd\wh\phi(w)$ in the Taylor expansion.
} Thus, one has
\begin{equation}\label{l21 T dd phi OPE}
\mc{R} \wh{T}(z) \dd\wh\phi(w)\sim \frac{\dd \wh\phi(w)}{(z-w)^2}+\frac{\dd^2\wh\phi(w)}{z-w}+\mr{reg.}
\end{equation}

Similarly, one obtains
\begin{equation}\label{l21 T dbar phi OPE}
\mc{R} \wh{\ol{T}}(z) \bar\dd\wh\phi(w)\sim \frac{\bar\dd \wh\phi(w)}{(\bar{z}-\bar{w})^2}+\frac{\bar\dd^2\wh\phi(w)}{\bar{z}-\bar{w}}+\mr{reg.},\quad
\mc{R} \wh{T}(z) \bar\dd\wh\phi(w)\sim \mr{reg.},\quad
\mc{R} \wh{\ol{T}}(z) \dd\wh\phi(w)\sim \mr{reg.}
\end{equation}
\end{example}

\begin{example}[$TT$ OPE] \label{l21 example: TT OPE}
 Let us calculate the OPE of the holomorphic component of the stress-energy tensor $T$ with itself:
\begin{multline}\label{l21 TT OPE computation}
\mc{R} \wh{T}(z)\wh{T}(w)= \mc{R} :-\frac12 \dd\wh\phi(z)\dd\wh\phi(z) :\;:-\frac12 \dd\wh\phi(w)\dd\wh\phi(w) : \underset{\mr{Wick}}{=} \\
\underset{\mr{Wick}}{=} 
\frac14 \wick{:\dd\c1{\wh\phi(z)}\dd\c2{\wh\phi(z)} \dd\c1{\wh\phi(w)}\dd\c2{\wh\phi(w)} :} + 
\frac14 \wick{: \dd\c1{\wh\phi(z)}\dd\c2{\wh\phi(z)} \dd\c2{\wh\phi(w)}\dd\c1{\wh\phi(w)} :}+\\
+\frac14 \wick{:\dd\c1{\wh\phi(z)}\dd{\wh\phi(z)} \dd\c1{\wh\phi(w)}\dd{\wh\phi(w)} :}+
\frac14 \wick{:\dd\c1{\wh\phi(z)}\dd{\wh\phi(z)} \dd{\wh\phi(w)}\dd\c1{\wh\phi(w)} :}+\\
+
\frac14 \wick{:\dd{\wh\phi(z)}\c1\dd{\wh\phi(z)} \dd\c1{\wh\phi(w)}\dd{\wh\phi(w)} :}+
\frac14 \wick{:\dd{\wh\phi(z)}\c1\dd{\wh\phi(z)} \dd{\wh\phi(w)}\c1\dd{\wh\phi(w)} :}+
\frac14 :\dd{\wh\phi(z)}\dd{\wh\phi(z)} \dd{\wh\phi(w)}\dd{\wh\phi(w)} :\\
\sim\frac{\frac12\wh{\mathbb{1}}}{(z-w)^4} -\frac{:\dd\wh\phi(z)\dd\wh\phi(w):}{(z-w)^2} +\mr{reg.}\\
\underset{\mr{expand\;}\dd\wh\phi(z)\;\mr{at\;}w}{=}
\frac{\frac12\wh{\mathbb{1}}}{(z-w)^4} -\frac{:\dd\wh\phi(w)\dd\wh\phi(w):}{(z-w)^2}-\frac{:\dd^2\wh\phi(w)\dd\wh\phi(w):}{z-w} +\mr{reg.}
\\= \boxed{\frac{\frac12\wh{\mathbb{1}}}{(z-w)^4} + \frac{2\wh{T}(w)}{(z-w)^2}+\frac{\dd \wh{T}(w)}{z-w} +\mr{reg.}}
\end{multline}
Note the appearance of the fourth order pole here. As we will see later, it is linked to the phenomenon of central charge (and thus to projectivity of CFT as a Segal's functor).

By similar computations, one finds
\begin{equation}
\mc{R}\wh{\ol{T}}(z) \wh{\ol{T}}(w)\sim 
\frac{\frac12\wh{\mathbb{1}}}{(\bar{z}-\bar{w})^4} + \frac{2\wh{\ol{T}}(w)}{(\bar{z}-\bar{w})^2}+\frac{\dd \wh{\ol{T}}(w)}{\bar{z}-\bar{w}} +\mr{reg.},\qquad
\mc{R}\wh{T}(z) \wh{\ol{T}}(w)\sim \mr{reg.}
\end{equation}

\end{example}

\section{Digression: path integral formalism (in the example of free scalar field)}

\subsection{Finite-dimensional Gaussian integral}
Let $F$ be an $N$-dimensional real vector space equipped with Euclidean metric $h$ and with a positive-definite bilinear form $B\colon\mr{Sym}^2F\ra \RR$ and let $\ul{B}\in \mr{End}(F)$ be an endomorphism such that $B(u,v)=h(u,\ul{B} v)$. Then one has the following well-known Gaussian integral
\begin{equation}\label{l21 Gaussian int}
\int_F \mu_h e^{-\frac12 B(u,u)} = (2\pi)^{\frac{N}{2}}(\det \ul{B})^{-\frac12}.
\end{equation}
Here $\mu_h$ is the Lebesgue measure on $F$ associated with the metric $h$ and $B(u,u)$ is the quadratic function on $F$ -- the restriction of $B$ to the diagonal $\mr{Diag}\subset F\times F$.

\subsection{Wick's lemma for the moments of Gaussian measure}
For $f$ a polynomial function on $F$, consider its expectation value (average) with respect to the normalized Gaussian measure,
\begin{equation}\label{l21 expectation value}
\langle  f \rangle \colon = \frac{1}{(2\pi)^{\frac{N}{2}}(\det\ul{B})^{-\frac12}} \int_F \mu_h e^{-\frac12 B(u,u)} f(u).
\end{equation}
Note that the normalization factor in the r.h.s. is chosen such that one has 
\begin{equation}
\langle 1 \rangle =1.
\end{equation}

\begin{lemma}[Wick's lemma for the moments of Gaussian measure]\label{l21 lm: Wick}
Let $\theta_1,\ldots,\theta_n$ be some linear forms on $F$. Consider the Gaussian expectation value
\begin{equation}\label{l21 Gaussian expectation value}
\langle \theta_1 \cdots\theta_n  \rangle
\end{equation}
Then one has
\begin{enumerate}[(i)]
\item If $n$ is odd, the expectation value (\ref{l21 Gaussian expectation value}) is zero.
\item If $n=2m$ is even, one has
\begin{equation}\label{l21 Wick's lemma for Gaussian momenta: formula}
\langle \theta_1 \cdots\theta_n  \rangle = \sum_{
\begin{array}{c}
\mr{perfect\;matchings}\\
\{1,\ldots, n\}=\sqcup_{i=1}^m \{\alpha_i, \beta_i\}
\end{array}
} B^{-1}(\theta_{\alpha_1},\theta_{\beta_1})\cdots  B^{-1}(\theta_{\alpha_m},\theta_{\beta_m}).
\end{equation}
\end{enumerate}
\end{lemma}

For example, for $n=2$, one has
\begin{equation}\label{l21  <theta theta> Wick}
\langle \theta_1 \theta_2 \rangle = B^{-1}(\theta_1,\theta_2).
\end{equation}
Here on the r.h.s., $B^{-1}$ is understood as a map $B^{-1}\colon F^*\otimes F^* \ra \RR$ which is adjoint to the map $F^*\ra F$ -- the inverse of the map $B^\#\colon F\ra F^*$.

For $n=4$, one has 
\begin{equation}
\langle \theta_1\theta_2 \theta_3\theta_4 \rangle = B^{-1}(\theta_1,\theta_2) B^{-1}(\theta_3,\theta_4)+  B^{-1}(\theta_1,\theta_3) B^{-1}(\theta_2,\theta_4) +  B^{-1}(\theta_1,\theta_4) B^{-1}(\theta_2,\theta_3),
\end{equation}
where the terms correspond to the three perfect matchings on the set $\{1,2,3,4\}$.

Note that the r.h.s. of (\ref{l21 Wick's lemma for Gaussian momenta: formula}) looks similar to the r.h.s. of (\ref{l19 Wick formula}) if we were to retain only the contributions of perfect matchings (and identify the propagator $g_{pq}$ with $B^{-1}(\theta_p,\theta_q)$).

\begin{proof}[Sketch of proof of Lemma \ref{l21 lm: Wick}]
First note that part (i) of Lemma is obvious, since in this case the integrand in (\ref{l21 expectation value}) changes sign under $u\ra -u$.

For part (ii),
consider the ``generating functions for moments'' -- the following expectation value depending on the ``source'' parameter $J\in F^*$:
\begin{multline}\label{l21 source integral}
\langle e^{\langle J,u\rangle } \rangle = C \int_F
\mu_h\;e^{-\frac12 B(u,u)+ \langle J,u \rangle }=
C \int_F
\mu_h\;e^{-\frac12 B(u-B^{-1}J,u-B^{-1}J)+ \frac12 B^{-1}(J,J) }=\\
\underset{v\colon=u-B^{-1}J}{=} C\int_F \mu_h e^{-\frac12 B(v,v)+ \frac12 B^{-1}(J,J)} = e^{\frac12 B^{-1}(J,J)}
\end{multline}
where $C=(2\pi)^{-\frac{N}{2}}\det(\ul{B})^{\frac12}$. Then we can obtain correlators of monomials (\ref{l21 Wick's lemma for Gaussian momenta: formula}) by taking multiple partial derivatives of (\ref{l21 source integral}) in $J$ and then setting $J=0$.

More explicitly, consider an orthonormal basis in $F$ w.r.t. the metric $g$ and let $\{u^p\}$ be the corresponding coordinates on $F$. It suffices to prove (\ref{l21 Wick's lemma for Gaussian momenta: formula}) for $\theta_1=u^{p_1},\ldots, \theta_n=u^{p_n}$ a collection of coordinate functions; the general result then follows by linearity. We have
\begin{multline}\label{l21 Wick proof computation}
\langle u^{p_1}\cdots u^{p_n} \rangle = \left.\frac{\dd}{\dd J_{p_1}}\cdots \frac{\dd}{\dd J_{p_n}}\right|_{J=0}\langle e^{\langle J,u\rangle } \rangle =
\left.\frac{\dd}{\dd J_{p_1}}\cdots \frac{\dd}{\dd J_{p_n}}\right|_{J=0} e^{\frac12B^{-1}(J,J)}=\\=
\left.\frac{\dd}{\dd J_{p_1}}\cdots \frac{\dd}{\dd J_{p_n}}\right|_{J=0} \frac{1}{2^m m!} \left(B^{-1}(J,J)\right)^m
\end{multline}
where in the last step we selected the $m$-th term from the Taylor series of the exponential, since only it contributes to the $n=2m$-th derivative in $J$ at $J=0$ (note that in the last expression the restriction to $J=0$ is irrelevant -- the derivative is a constant). At this point we see that the answer is the sum over the ways to distribute the $2m$ derivatives in $J$ over $2m$ copies of $J$ in $\left(B^{-1}(J,J)\right)^m$.
%\marginpar{This is a bit rushed..} 
This results in the sum over perfect matchings in the r.h.s. of (\ref{l21 Wick's lemma for Gaussian momenta: formula}).\footnote{Note that the set of perfect matchings on the set of $2m$ elements can be seen as a coset of the symmetric group, $S_{2m}/(S_m\ltimes \ZZ_2^m)$.} E.g., for $m=1$ (i.e. $n=2$) we have
\begin{multline}\label{l21 <uu> Wick}
\langle u^{p_1} u^{p_2} \rangle = \frac12 \frac{\dd}{\dd J_{p_1}}\frac{\dd}{\dd J_{p_2}}(B^{-1})^{pq}J_p J_q=\\
=
\wick[offset=1.6em]{\frac12 \c1{\frac{\dd}{\dd J_{p_1}}}\c2{\frac{\dd}{\dd J_{p_2}}}(B^{-1})^{pq}\c1{J_p} \c2{J_q}}+
\wick[offset=1.6em]{\frac12 \c1{\frac{\dd}{\dd J_{p_1}}}\c2{\frac{\dd}{\dd J_{p_2}}}(B^{-1})^{pq}\c2{J_p} \c1{J_q}}
=
(B^{-1})^{p_1 p_2},
\end{multline}
which is (\ref{l21 Wick's lemma for Gaussian momenta: formula}) specialized to the coordinate monomial $\theta^1\theta^2$ with $\theta_1=u^{p_1}$, $\theta_2=u^{p_2}$.
Here brackets show which derivatives hit which instances of $J$. %For $m=2$ (i.e. $n=4$) we have
%\begin{multline}
%\langle u^{p_1} u^{p_2} u^{p_3} u^{p_4} \rangle = \frac{1}{2^22!} \frac{\dd}{\dd J_{p_1}}\frac{\dd}{\dd J_{p_2}}\frac{\dd}{\dd J_{p_3}}\frac{\dd}{\dd J_{p_4}}(B^{-1})^{pq} J_p J_q (B^{-1})^{p'q'}J_{p'} J_{q'}=
%\end{multline}
\end{proof}

\marginpar{Lecture 22,\\
10/12/2022}
\subsection{Scalar field theory in the path integral formalism}
\label{sss free scalar in PI formalism}
Let $\Sigma$ be a surface equipped with Riemannian metric $g$.
In the path integral (or more appropriately, ``functional integral'') approach, the partition function of the scalar field on $\Sigma$ is  given by a formal Gaussian integral
\begin{equation}\label{l22 Z}
Z(\Sigma)= ``\int_{\F_\Sigma} \mc{D}\phi\; e^{-\frac{1}{4\pi}S(\phi)}"
\end{equation}
over the (infinite-dimesnional) space of functions $\F_\Sigma=C^\infty(\Sigma)$. Here
\begin{equation}\label{l22 S}
S(\phi)=\int_\Sigma \frac12 (d\phi\wedge *d\phi+\frac{m^2}{2}\phi^2 \dvol_g) = \int_\Sigma \frac12 \phi (\Delta+m^2) \phi\; \dvol_g
\end{equation}
Here for the moment we are considering scalar field with mass $m\geq 0$; later we will want to set $m=0$ to have a conformal theory. In (\ref{l22 S}) we assume that either $\Sigma$ is closed or else an appropriate boundary condition is imposed on fields $\phi$, so that the boundary term $\int_{\dd \Sigma}\frac12 \phi\,*d\phi$ vanishes -- then the right equality in (\ref{l22 S}) is valid.

The expression (\ref{l22 Z}) is similar to the l.h.s. of (\ref{l21 Gaussian int}) if we make the identifications 
\begin{equation}
\begin{gathered}
F=\F_\Sigma,\quad u=\phi,\quad h(\phi_1,\phi_2)=\int_\Sigma \phi_1\phi_2\, \dvol_g,\\ B(\phi_1,\phi_2)=\frac{1}{4\pi}\int_\Sigma \phi_1(\Delta+m^2) \phi_2\,\dvol_g,\quad
\ul{B}=\frac{1}{4\pi}(\Delta+m^2).
\end{gathered}
\end{equation}

Understanding the infinite-dimensional integral (\ref{l22 Z}) as a measure-theoretic integral is problematic and we think of it as defined by the r.h.s. of (\ref{l21 Gaussian int}):
\begin{equation}\label{l22 Z as det}
Z(\Sigma)\colon= \det(c(\Delta+m^2))^{-\frac12},
\end{equation}
where $c=\frac{1}{8\pi^2}$.

\begin{remark}
Determinants of differential operators are also nontrivial to make sense of, but there are viable solutions. One method is ``zeta-regularization'': for $D$ a differential operator with a discrete eigenvalue spectrum, one constructs the zeta-function of $D$ -- a function of a complex variable $s$ defined as
\begin{equation}\label{l22 zeta}
\zeta_D(s)\colon= \sum_{\lambda} \lambda^{-s}.
\end{equation}
The sum is over the eigenvalues of $D$ (in the case of continuum spectrum, the sum should be replaced by an integral). 
%w.r.t. the Plancherel measure).
The sum converges to a holomorphic function for $\mr{Re}(s)>A$ for some $A$ and admits a unique meromorphic continuation to $\CC$ with $s=0$ a regular point. Then the zeta-regularized determinant is defined in terms of the derivative of the meromorphically continues zeta-function at $s=0$ as
\begin{equation}
{\det}_{\zeta\mr{-reg}}(D)\colon= e^{-\zeta_D'(0)}.
\end{equation}
\end{remark}

Note that in (\ref{l21 Gaussian int}) we wanted the quadratic form $B$ (and thus the operator $\ul{B}$) to be strictly positive. For the scalar field on a closed surface $\Sigma$ that forces $m>0$; in the massless case the operator $\ul{B}=\Delta$ has a 1-dimensional kernel given by constant functions on $\Sigma$. Correspondingly, the determinant $\det \Delta$ is not well-defined even with zeta-regularization due to appearance of the eigenvalue $\lambda=0$, which means that the zeta-function (\ref{l22 zeta}) is not defined. For $m>0$, the partition function (\ref{l22 Z as det}) is well-defined via zeta-regularization.

\subsubsection{Moments of Gaussian measure.}
Correlators in the path integral formalism are given as Gaussian averages of products of fields and so are given by the Wick's lemma (\ref{l21 Wick's lemma for Gaussian momenta: formula}). For instance for $p_1\neq p_2\in\Sigma$ two points, one has
\begin{equation}\label{l22 <phi phi>}
\langle \phi(p_1) \phi(p_2) \rangle = ``\frac{1}{Z(\Sigma)}\int_{\F_\Sigma} \mc{D}\phi\;  e^{-\frac{1}{4\pi}S(\phi)} \phi(p_1) \phi(p_2) "\colon= G(p_1,p_2)
\end{equation} 
-- the Green's function of the operator $\frac{1}{4\pi}(\Delta+m^2)$. Here the Green's function -- the integral kernel of the operator $(\Delta+m^2)^{-1}=\ul{B}^{-1}$ is analogous to the matrix element of $\ul{B}^{-1}$ appearing in (\ref{l21  <theta theta> Wick}), (\ref{l21 <uu> Wick}). One should think of the r.h.s. of (\ref{l22 <phi phi>}) as the mathematical definition of the l.h.s., motivated by Wick's lemma in the finite-dimensional case. Put another way, in the context of infinite-dimensional Gaussian integrals, Wick's lemma becomes not a lemma (equality between two well-defined objects), but a \emph{definition} of the moments of the infinite-dimensional Gaussian measure.

Likewise, for the four-point correlator one has
\begin{multline}\label{l22 4-point}
\langle \phi(p_1) \phi(p_2) \phi(p_3) \phi(p_4) \rangle = ``\frac{1}{Z(\Sigma)}\int_{\F_\Sigma} \mc{D}\phi\;  e^{-\frac{1}{4\pi}S(\phi)} \phi(p_1) \phi(p_2) \phi(p_3) \phi(p_4) "\colon= \\
\colon=
G(p_1,p_2) G(p_3,p_4) + G(p_1,p_3) G(p_2,p_4) + G(p_1,p_4) G(p_2,p_3)
\end{multline}

We note that formulae (\ref{l22 <phi phi>}), (\ref{l22 4-point}) make sense for any closed surface, for $m>0$  (for $m=0$, the operator $\Delta$ is non-invertible and hence the Green's function does not exist).

\subsubsection{Case $\Sigma=\CC$.} Let us restrict to the case $\Sigma=\CC$ -- the complex plane. The Green's function $G(z,w)$ can be explicitly found in terms of Bessel's function $K_0$,\footnote{
Bessel's function $K_0(r)$ is a solution of the ODE $\left(\frac{d^2}{dr^2}-\frac{1}{r}\frac{d}{dr}+1\right)y=0$; it has logarithmic asymptotics $K_0(r)\sim -\log r + (\log 2-\gamma)+o(r)$ as $r\ra 0$ (where $\gamma= 0.5772\ldots$ is the Euler's constant). At $r\ra +\infty$ the function $K_0$ is exponentially decaying, $K_0(r)\sim \sqrt{\frac{\pi}{2r}}e^{-r}$.
}
\begin{equation}
G(z,w)=2 K_0(m\cdot |z-w|),
\end{equation}
In particular for $m\ra 0$ and $z\neq w$ fixed one has the asymptotic behavior
\begin{equation}
G(z,w)\underset{m\ra 0}{\sim} -2\log|z-w|+C(m)
\end{equation}
where $C(m)= -2\log m+c$ is a constant (in $z,w$) which diverges as $m\ra 0$; here $c=2(\log 2-\gamma)$.
Thus, we find that the two-point correlator 
\begin{equation}
\langle \phi(z) \phi(w) \rangle =G(z,w)
\end{equation}
computed in the path integral formalism does not exist in the conformal limit $m\ra 0$. Recall that its counterpart in the radial quantization picture (\ref{l19 <phi phi> radial quantization}) is also problematic due to the appearance of an ``infinite constant''  $\langle\vac| \wh\phi_0^2 |\vac\rangle$.

Next, if we consider the two-point correlator of derivatives of the field
\begin{equation}
\langle \dd\phi(z)\dd\phi(w) \rangle = \dd_z\dd_w G(z,w) \underset{m\ra 0}{\ra} -\frac{1}{|z-w|^2},
\end{equation}
we see that it has a well-defined limit $m\ra 0$, which also agrees with our earlier result obtained in the radial quantization picture (\ref{l20 <dd phi  dd phi>}).

One can apply this method to construct similar correlators of derivatives of fields on any surface -- the Green's function itself does not exist in the limit $m\ra 0$ but its derivatives do have a limit.\footnote{
Of course, on a general surface we don't have the radial quantization picture to compare to -- that one is specific to $\Sigma=\CC$. So on general $\Sigma$ it makes sense to take the path integral prescription as the definition of CFT correlators.}

As an example of a more complicated local observable, we can consider the following quadratic polynomial on $\F_\Sigma$:
\begin{equation}\label{l22 :dd phi dd phi:}
:\dd\phi(z)\dd\phi(z):\;\;\colon= \lim_{w\ra z}\left(\dd\phi(w)\dd\phi(z)+\frac{1}{(w-z)^2}\right)
\end{equation}
When computing the correlator of this observable with a collection of other other observables by Wick's lemma, the correction $\frac{1}{(z-w)^2}$ cancels the contribution of Wick contraction $\wick{\dd\c1{\phi(w)}\c1{\dd\phi(z)}}$ -- so effectively one can say this contraction is prohibited when computing correlators involving $:\dd\phi(z)\dd\phi(z):$.\footnote{In the path integral formalism we cannot talk about normal ordering of operators -- since we don't have operators -- so the limiting process in the r.h.s. of (\ref{l22 :dd phi dd phi:}) becomes the definition of the ``normally-ordered'' differential polynomial in the l.h.s.}

%For instance, we can compute the correlator
As an illustration, let us compute the correlator of the stress-energy tensor with itself (in the path integral formalism):
\begin{multline}
\langle T(z) T(w) \rangle = \langle :-\frac12 \dd\phi(z)\dd\phi(z): \;:-\frac12 \dd\phi(w)\dd\phi(w): \rangle =\\
= \wick{ \langle :-\frac12 \dd\c1{\phi(z)}\dd\c2{\phi(z)}: \;:-\frac12 \dd\c1{\phi(w)}\dd\c2{\phi(w)}: \rangle }+
 \wick{ \langle :-\frac12 \dd\c1{\phi(z)}\dd\c2{\phi(z)}: \;:-\frac12 \dd\c2{\phi(w)}\dd\c1{\phi(w)}: \rangle } \\
 =\frac{2}{4}\frac{-1}{(z-w)^2}\,\frac{-1}{(z-w)^2}=\frac{1}{2} \frac{1}{(z-w)^4}.
\end{multline}
Note that contractions inside  $:\cdots :$ are prohibited.

\subsubsection{OPEs.} We remark that one can also find OPEs within the path integral formalism (from Wick's lemma). For example, consider the correlator
\begin{equation}\label{l22 <dd phi dd phi OO>}
\langle \dd \phi(z) \dd\phi(w) \underbrace{O_1(z_1)\cdots O_n(z_n)}_{\mr{test\;observables}} \rangle
\end{equation}
in the asymptotics $z\ra w$. The correlator is given by a sum over perfect matchings of constituent fields, where we should distinguish two subclasses of matchings:
\begin{enumerate}[(i)]
\item Matchings where $\dd\phi(z)$ and $\dd\phi(w)$ are paired (Wick-contracted) -- these terms sum up to 
$-\frac{1}{(z-w)^2}\langle  O_1(z_1)\cdots O_n(z_n)\rangle $.
\item Matchings where $\dd\phi(z)$ and $\dd\phi(w)$ are not paired (rather, each is paired with one of $O_i$'s.) These terms are regular as $z\ra w$.
\end{enumerate}
Thus, one obtains 
\begin{equation}
\langle \dd \phi(z) \dd\phi(w)O_1(z_1)\cdots O_n(z_n) \rangle \underset{z\ra w}{\sim} -\frac{1}{(z-w)^2}\langle  O_1(z_1)\cdots O_n(z_n)\rangle +\mr{reg.}
\end{equation}
This corresponds to the OPE (\ref{l21 del phi del phi OPE}) which we previously obtained from the radial quantization picture.

\begin{figure}[H]
%$$\vcenter{\hbox{ \includegraphics[scale=0.5]{l22_OPE.eps} }} $$
\begin{center}
\includegraphics[scale=0.7]{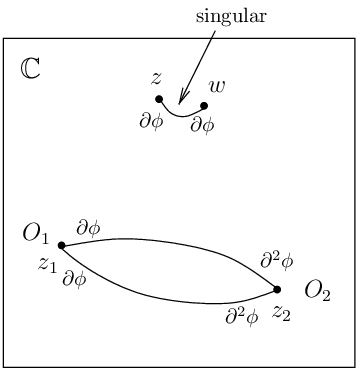}\hspace{2cm}
\includegraphics[scale=0.7]{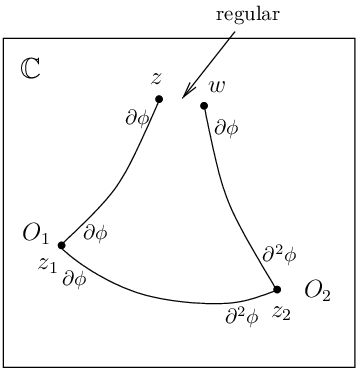}
\end{center}
\caption{An example of a singular and a regular (as $z\ra w$) contribution to the correlator (\ref{l22 <dd phi dd phi OO>}). In this example, the two test observables are $O_1=:\dd\phi\dd\phi:$ and $O_2=:\dd^2\phi \dd^2\phi:$; we depict observables as corollas with the number of prongs being the degree of the differential monomial in $\phi$; edges correspond to Wick contractions. Thus each picture is one summand in the computation of the correlator via Wick's lemma. 
%(matched pairs). 
}
\end{figure}

\begin{remark} When studying the theory on $\CC$ we introduced a small positive mass $m$ in order to have well-defined Green's function (and then we let $m\ra 0$ in correlators). Another possibility, instead of introducing a mass, is to have a massless theory, but replace $\CC$ with a disk $D_R=\{z\in \CC\;|\; |z|\leq R\}$ of large radius $R$, where one imposes Dirichlet boundary condition $\phi|_{\dd D_R}=0$. Then one can write an explicit Green's function
\begin{equation}
G(z,w)=-2\log\frac{|z-w|}{|R-\frac{z\bar{w}}{R}|}\;\; \underset{R\ra\infty}{\sim} -2\log |z-w|+C
\end{equation} 
with $C=2\log R$.
\end{remark}

\subsubsection{Summary: path integral vs. radial quantization.} The path integral formalism allows one another way to compute the same quantities as the radial quantization (or ``operator formalism'') does -- correlators and OPEs. The two formalisms should be seen as complementing each other: path integral formalism has the benefit that it can be applied to general surfaces, not just $\CC$. The benefit of the operator formalism is that it also recovers the space of states (and extra structure it might have, e.g., in the case of scalar field, the action of the Heisenberg Lie algebra). So, ultimately, the path integral formalism is better suited for handling global geometry (nontrivial surfaces) while the operator formalism gives a good handle of the local picture of CFT near a puncture (where $\Sigma$ can be approximated by $\CC^*$).

\chapter{%Some generalities of CFTs on $\CC$
Conformal Field Theory on $\CC$ Belavin-Polyakov-Zamolodchikov axiomatic picture
}
\chaptermark{CFT on $\CC$: BPZ axiomatic picture}
%\chaptermark{Belavin-Polyakov-Zamolodchikov axiomatic picture}

In this chapter we will %talking about 
present
Belavin-Polyakov-Zamolodchikov \cite{BPZ} picture of a general CFT on $\CC$, sometimes using the scalar field as an illustration. 

\section{Virasoro algebra}

\begin{definition} Virasoro algebra is the %complex Lie algebra arising as the 
central extension 
$\CC\ra \mr{Vir}\ra \mc{W} $
of the Witt algebra $\mc{W}$ 
(the Lie algebra of meromorphic vector fields on $\CC$ with only pole at $0$ allowed, see Section \ref{sss: Witt algebra}),  %$\mr{Vir}=\mc{W}\oplus \CC\cdot \mathbb{K}$
defined by the Lie brackets
\begin{equation}\label{l22 Virasoro}
\left[f(z)\frac{\dd}{\dd z}, g(z)\frac{\dd}{\dd z}\right]^\mr{Vir} = (fg'-gf')\frac{\dd}{\dd z}+\frac{c}{12} \mathbb{K} \oint_{\gamma} \frac{dz}{2\pi i} f'''(z) g(z),
\end{equation}
where $\mathbb{K}$ is the central element, $c\in\CC$ is a complex number (a parameter of the central extension) -- the ``central charge,'' $\gamma$ is a closed simple curve going around $0$ counterclockwise.\footnote{The conventional normalization factor $\frac{1}{12}$ in (\ref{l22 Virasoro}) is chosen in such a way that the central charge of the free massless scalar field is $c=1$.}
\end{definition}

Virasoro algebra has the standard set of generators $\{L_n\}_{n\in \ZZ},\mathbb{K}$ subject to commutation relations
\begin{equation}\label{l22 Virasoro via L}
[L_n,L_m]=(n-m)L_{n+m} + \delta_{n,-m}\frac{c}{12}(n^3-n)\mathbb{K},\quad n,m\in\ZZ
\end{equation}
and $[\mathbb{K},\cdots]=0$; $L_n$ are the lifts of the standard generators $l_n=-z^{n+1}\dd_z$ of the Witt algebra.

\textbf{Exercise:} check that the Lie brackets (\ref{l22 Virasoro}) or equivalently (\ref{l22 Virasoro via L}) satisfy the Jacobi identity.

%The definition above contains an implicit statement that commutation relations (\ref{l22 Virasoro}) or (\ref{l22 Virasoro via L}) do define a Lie algebra, i.e., satisfy Jacobi identity. This is a straightforward check which me leave as an exercise to the reader.

In fact, Virasoro algebra is the \emph{unique} (up to a choice of the value of the parameter $c$) central extension of the Witt algebra, which is the content of the following theorem.
\begin{thm}[Feigin-Fuchs] %\marginpar{Feigin-Fuchs?}
One has
\begin{equation}
H^2_\mr{Lie}(\mc{W},\CC)=\CC
\end{equation}
-- the second Lie algebra (Chevalley-Eilenberg) cohomology of the Witt algebra (with coefficients in the trivial module) has rank $1$. This cohomology is generated by the cohomology 
class of the Lie 2-cocycle 
\begin{equation}
\lambda(f(z)\dd_z,g(z)\dd_z)=\frac{1}{12}\oint_\gamma \frac{dz}{2\pi i } f'''(z)g(z).
\end{equation}
\end{thm}

\marginpar{Lecture 23,\\
10/14/2022}

\section{Axiomatic CFT on $\CC$. Action of Virasoro algebra on $\HH$}
We will start setting up general conformal field theory on $\CC$ as an axiomatic picture, following \cite{BPZ}.

In this picture, a CFT is the following collection of data. 
\begin{enumerate}[(I)]
\item \textbf{Space of states.} One has a complex vector space -- the space of states $\HH$ -- with a distinguished vector $|\vac\rangle\in \HH$.
\item \textbf{Space of fields, local operators.} One has a complex vector space of local observables (or ``space of composite fields'') $V$. For $z\in \CC$ we will denote $V_z$ a copy of $V$ placed at $z$;\footnote{I.e. we are thinking of a trivial vector bundle $\mc{V}=V\times \CC$ over $\CC$ with typical fiber $V$ and $V_z$ the fiber over a specific point $z$.} 
we denote a copy of an element $\Phi\in V$ placed at a point $z$ by $\Phi(z)\in V_z$. % of can be placed at any point $z\in \CC^*$

For $z\neq 0$, $\Phi(z)\in V_z$ is represented by a (possibly unbounded) operator $\wh\Phi(z)\in \mr{End}(\HH)$.\footnote{In other words, there is a map $Y\colon  V\times \CC^* \ra \mr{End}(\HH)$, linear in $V$ and smooth on $\CC^*$. We denote $Y(\Phi,z)$ by $\wh\Phi(z)$.}
%If $\Phi$ is a real element, then $\wh{\Phi}(z)$ is a hermitian operator. 
% Nonzero elements of $V$ are represented by nonzero operators.

 %Products of local operators $\wh{\Phi}_1(z_1)\cdots \wh{\Phi}_n(z_n)$ are assumed to be well defined if all $z_i$'s are distinct and the product is radially ordered, $|z_1|\geq \cdots \geq |z_n|$. Operators $\wh{\Phi}_1(z)$, $\wh\Phi_2(w)$ are assumed to commute if $z\neq w$ and $|z|=|w|$.
 \item \textbf{Field-state correspondence.}  One has a linear isomorphism 
\begin{equation}\label{l23 field-state correspondence}
\mathsf{s}\colon V \xra{\sim} \HH
\end{equation} 
mapping a field $\Phi\in V$ to the state $\lim_{z\ra 0} \wh{\Phi}(z)|\vac\rangle $. (In particular, such a limit is required to exist for any $\Phi$ and determine an isomorphism between fields and states.) See Section \ref{sss: field-state correspondence} below for an example.
\item \label{l23 Axiom inner product} \textbf{Inner product.} Both $\HH$ and $V$ carry  real structures and
 %sesquilinear 
%hermitianinner products 
nondegenerate hermitian forms $\langle,\rangle$
(intertwined by $\fs$).
If the hermitian forms are additionally positive-definite, the CFT is called \emph{unitary}.
 For the hermitian conjugate of a local operator 
%$\Phi$ of conformal weight $(h,\bar{h})$ 
one has 
%\marginpar{Edit! Should there be a bar in the r.h.s.?}
\begin{equation}\label{l23 Phi^+}
(\wh\Phi(z))^+ = \bar{z}^{-2h} z^{-2\bar{h}} \wh{\Phi^*}(1/\bar{z}).
\end{equation}
Here $*$ denotes the complex conjugation in $V$ and $(h,\bar{h})$ is the conformal weight of the field $V$ (see Definition \ref{l25 def: conformal dimension} below).%\marginpar{Put a link to the discussion of conformal weights.}
\item \label{l23 Axiom radial ordering} \textbf{Radial ordering, domains of field operators, same-time commutativity.} For any $n$-tuple of elements $\Phi_1,\ldots,\Phi_n\in V$, the vector 
\begin{equation}\label{l23 vector}
\wh{\Phi}_1(z_1)\cdots \wh{\Phi}_n(z_n)|\vac\rangle 
\end{equation} 
is assumed to be well-defined if $z_i$'s are radially ordered, 
$|z_1|\geq \cdots \geq |z_n|$. %\marginpar{Problem: if $|z_1|\geq 1$, the $L_2$ norm might be infinite!}
As a consequence (using the same axiom for a string of $n+1$ local operators) the vector (\ref{l23 vector}) is in the domain of $\wh\Phi(z)$ if $|z|\geq |z_1|$ and $z\neq z_i$ for $i=1,\ldots,n$.
Operators $\wh{\Phi}_1(z)$, $\wh\Phi_2(w)$ are assumed to commute if $z\neq w$ and $|z|=|w|$.
\item \textbf{Correlators.} For an $n$-tuple of elements $\Phi_1,\ldots,\Phi_n\in V$, the correlator is defined as
\begin{equation}\label{l23 correlator via radial ordering}
\langle \Phi_1(z_1)\cdots \Phi_n(z_n) \rangle\colon = \langle \vac|
\mc{R}\Big(\wh\Phi_1(z_1)\cdots \wh\Phi_n(z_n) \Big)
|\vac\rangle.
\end{equation}
where $\langle\vac| \colon= \big\langle |\vac\rangle,- \big\rangle_\HH \in \HH^*$ is the covector dual to the vector $|\vac\rangle$.
The correlator (\ref{l23 correlator via radial ordering}) is a smooth function on $C_n(\CC)$ -- the open configuration space of $n$ points on $\CC$, depending %$n$-multilinearly 
linearly
on the fields $\Phi_1,\ldots,\Phi_n$.\footnote{If one of the points $z_i$ in the l.h.s. of (\ref{l23 correlator via radial ordering}) is zero, one understands the r.h.s. as a limit $z_i\ra 0$.}
\item \textbf{Identity field and stress-energy tensor.} %$V_z$ contains the special element $\mathbb{1}$ acting on $\HH$ by identity and the special elements $T(z)$, $\ol{T}(z)$
$V$ contains a  element $\mathbb{1}$ acting on $\HH$ by identity and special elements $T$, $\ol{T}$
 satisfying holomorphicity/antiholomorphicity
\begin{equation}\label{l23 dbar T=0}
\bar\dd \wh{T}(z)=0,\quad \dd\wh{\ol{T}}(z)=0
\end{equation}
and the OPEs
\begin{eqnarray}
\label{l23 TT} \mc{R} \wh{T}(z) \wh{T}(w) &\sim & \frac{\frac{c}{2}\wh{\mathbb{1}}}{(z-w)^4} + \frac{2 \wh{T}(w)}{(z-w)^2} + \frac{\dd \wh{T}(w)}{z-w} +\mr{reg.}
\\ \label{l23 Tbar Tbar}
\mc{R} \wh{\ol{T}}(z) \wh{\ol{T}}(w) &\sim & \frac{\frac{\bar{c}}{2}\wh{\mathbb{1}}}{(\bar{z}-\bar{w})^4} + \frac{2 \wh{\ol{T}}(w)}{(\bar{z}-\bar{w})^2} + \frac{\bar\dd\, \wh{\ol{T}}(w)}{\bar{z}-\bar{w}} +\mr{reg.}  \\ \label{l23 T Tbar}
\mc{R} \wh{T}(z) \wh{\ol{T}}(w) &\sim & \mr{reg.}
\end{eqnarray}
with $c,\bar{c}$ some complex numbers (the holomorphic and antiholomorphic central charges).

Elements $\mathbb{1},T,\ol{T}\in V$ are real (with respect to the real structure on $V$).
\item \textbf{Projective action of conformal vector fields on states.} One has a projective representation $\rho$ of the Lie algebra of conformal vector fields on $\CC^*$ on $\HH$, where
the conformal vector field $v=u(z)\dd_z+\ol{u(z)}\dd_{\bar{z}}$ on $\CC^*$  (with $u$ a meromorphic function on $\CC$ with pole allowed only at $z=0$) is represented by the operator
\begin{equation}\label{l23 rho}
\rho(
%u(z)\dd_z+\ol{u(z)} \bar\dd_{\bar{z}}
u\dd+\bar{u}\bar\dd
) \colon= -\frac{1}{2\pi i}\oint_{\gamma} \left(dz\, u(z)\, \wh{T}(z)-d\bar{z}\, \ol{u(z)}\, \wh{\ol{T}}(z) \right) \qquad \in \mr{End}(\HH)
\end{equation}
where $\gamma\in \CC^*$ is a closed contour %(1-cycle) 
going around zero once counterclockwise.\footnote{
If we combine $\wh{T}(z)$ and $\wh{\ol{T}}(z)$ into a single object -- the total stress-energy operator $\wh{T}^\mr{total}(z)=\wh{T}(z)(dz)^2+\wh{\ol{T}}(z) (d\bar{z})^2$ -- a quadratic differential on $\CC^*$ valued in $\mr{End}(\HH)$, we can phrase (\ref{l23 rho}) as
$$\rho(v)=-\frac{1}{2\pi i}\oint_\gamma \iota_v \wh{T}^\mr{total}.$$ 
Here the contraction with the vector field $v$ converts the total stress-energy tensor from a quadratic differential into an (operator-valued) 1-form, which can then be integrated over the 1-cycle $\gamma$.
}
 In particular the standard generators of the Witt algebra $\mc{W}$, $l_n=-z^{n+1}\dd_z$, are represented by
\begin{equation}\label{l23 L via T}
\wh{L}_n\colon=\rho(-z^{n+1}\dd_z)=\frac{1}{2\pi i}\oint_\gamma dz\, z^{n+1} \wh{T}(z) \qquad \in\mr{End}(\HH)
\end{equation}
and likewise for the generators of the antiholomorphic copy $\ol{\mc{W}}$ of the Witt algebra:
\begin{equation}\label{l23 Lbar via Tbar}
\wh{\ol{L}}_n\colon=\rho(-\bar{z}^{n+1}\dd_{\bar{z}})=\frac{1}{2\pi i}\oint_\gamma d\bar{z}\, \bar{z}^{n+1} \wh{\ol{T}}(\bar{z}) \qquad \in\mr{End}(\HH).
\end{equation}
We remark that the inverse formulae for (\ref{l23 L via T}) and (\ref{l23 Lbar via Tbar}), expressing the stress-energy tensor in terms of operators $\wh{L}_n$, $\wh{\ol{L}}_n$ are:
\begin{equation}\label{l23 T via L, Tbar via Lbar}
\wh{T}(z)=\sum_{n\in\ZZ} z^{-n-2} \wh{L}_n,\quad \wh{\ol{T}}(z)=\sum_{n\in\ZZ} \bar{z}^{-n-2} \wh{\ol{L}}_n.
\end{equation}
I.e., essentially (and up to a shift in numbering), operators $\wh{L}_n$ are the Fourier modes of the field $\wh{T}(z)$ restricted to a circle.
\end{enumerate}

\begin{lemma}\label{l23 lemma: Virasoro from TT OPE}
\begin{enumerate}[(i)]
\item  As a consequence of the $TT$ OPE (\ref{l23 TT}), operators $\wh{L}_n$ satisfy the Virasoro commutation relation (\ref{l22 Virasoro via L}) with central charge $c$ (the coefficient in the fourth order pole  in (\ref{l23 TT})). 
\item Similarly, as a consequence of $\ol{T}\,\ol{T}$ OPE (\ref{l23 Tbar Tbar}), operators $\wh{\ol{L}}$ satisfy the Virasoro commutation relation with central charge $\bar{c}$.
\item As a consequence of $T\ol{T}$ OPE (\ref{l23 T Tbar}), the generators of the holomorphic and antiholomorphic copies of the Virasoro algebra commute: $[\wh{L}_n,\wh{\ol{L}}_m]=0$.
\end{enumerate}
\end{lemma}
We will prove this lemma in Section \ref{sss: Virasoro from TT OPE} below.

\begin{remark}\label{l23 rem: H^big, H^small} In the axioms above, it would be more correct to %introduce 
distinguish
two versions of the space of states: 
\begin{itemize}
\item 
$\HH=\HH^\mr{small}$ -- the one identified with $V$ by the field-state correspondence  (\ref{l23 field-state correspondence}), containing the vector $|\vac\rangle$ and carrying a %sesquilinear
hermitian inner product.
\item A completion $\HH^\mr{big}$ of $\HH^\mr{small}$ on which the field operators $\wh{\Phi}(z)$ act (it should be a completion containing all vectors of the form (\ref{l23 vector})). 
\end{itemize}

E.g. in the scalar field theory, we can define $\HH^\mr{small}$ to be the set of all \emph{finite} linear combinations of basis vectors (\ref{l17 basis vector in H}), while $\HH^\mr{big}$ is spanned by the same basis but has to contain certain infinite linear combinations. 

We note that even if the hermitian form is positive definite, neither $\HH^\mr{small}$ nor $\HH^\mr{big}$ is a Hilbert space: $\HH^\mr{small}$ carries a %sesquilinear 
hermitian
form, but is not complete with respect to it, while $\HH^\mr{big}$ contains vectors (\ref{l23 vector}) which generally have infinite $L^2$ norm if $|z_1|\geq 1$.

In the case of a positive definite hermitian form, one can also consider the $L^2$ completion $\HH^\mr{Hilb}$ of $\HH^\mr{small}$. As follows from the axioms (\ref{l23 Axiom inner product}) and (\ref{l23 Axiom radial ordering}), $\HH^\mr{Hilb}$ is guaranteed to contain the vector (\ref{l23 vector}) only if $z_i$'s are distinct, radially-ordered \emph{and are contained in the open unit disk $\{z\in \CC\;|\; |z|<1\}$}.\footnote{
Indeed, for the square of the $L^2$ norm of the vector (\ref{l23 vector}) we have 
$|| \wh\Phi_1(z_1)\cdots \wh\Phi_n(z_n)|\vac\rangle ||^2 =
\langle \vac|\big(\wh\Phi_n(z_n)\big)^+\cdots \big(\wh\Phi_1(z_1)\big)^+\wh\Phi_1(z_1)\cdots \wh\Phi_n(z_n)|\vac\rangle = 
\prod_{i=1}^n \bar{z}_i^{-2h_i} z_i^{-2\bar{h}_i} \cdot \langle\vac|
\wh{\Phi}^*_n(1/\bar{z}_n)\cdots \wh{\Phi}^*_1(1/\bar{z}_1)\wh\Phi_1(z_1)\cdots \wh\Phi_n(z_n)|\vac\rangle
$. The correlator on the right is certain to exist only if the insertion points of the operators 
$1/\bar{z}_n,\ldots,1/\bar{z}_1,z_1,\ldots,z_n$
are distinct and the sequence is radially ordered. This implies that all $z_i$'s must be in the open unit disk. (Note that if $|z_1|=1$ then $1/\bar{z}_1=1/z_1$, thus the sequence is radially ordered but not all points are  distinct.)
%$|1/\bar{z}_n|\geq \cdots\geq  |1/\bar{z}_1|\geq |z_1|\geq \cdots\geq |z_n|$.
}
\end{remark}

\begin{remark}\label{l23 rem: L_n^+}
The hermitian conjugate of the Virasoro generator $\wh{L}_n$ (\ref{l23 L via T}) is readily computed from (\ref{l23 Phi^+}):
\begin{equation}
\wh{L}_n^+=\frac{-1}{2\pi i} \oint_{\gamma} d\bar{z}\, \bar{z}^{n+1}\underbrace{\bar{z}^{-4}\wh{T}(1/\bar{z})}_{\wh{T}(z)^+}
\underset{w=1/\bar{z}}{=} \frac{1}{2\pi i}\oint_{\gamma'} dw\,w^{-n+1} \wh{T}(w) = \wh{L}_{-n}
\end{equation}
Here $\gamma'$ is the image of the contour $\gamma$ under the inversion map $z\mapsto 1/\bar{z}$, which is again a contour going once around zero in positive direction. We have used the fact that $T$ has conformal weight $(2,0)$, see Example \ref{l25 ex: T conformal weight}. Similarly, one proves
\begin{equation}
\wh{\ol{L}}_n^+=\wh{\ol{L}}_{-n}.
\end{equation}
\end{remark}

The following is also an immediate consequence of (\ref{l23 Phi^+}).
\begin{lemma}\label{l23 lemma: conjugation of a correlator}
For any fields $\Phi_1,\ldots, \Phi_n$ and $z_1,\ldots,z_n\in \CC\backslash \{0\}$ an $n$-tuple of distinct points one has
\begin{equation}
\ol{\langle \Phi_1(z_1)\cdots \Phi_n(z_n)\rangle }=
\prod_{i=1}^n \bar{z}_i^{-2h_i} z_i^{-2\bar{h}_i}\cdot
\langle \Phi_1^*(1/\bar{z}_1)\cdots \Phi_n^*(1/\bar{z}_n) \rangle
\end{equation}
where $(h_i,\bar{h}_i)$ is the conformal weight of $\Phi_i$. The bar over the correlator in the l.h.s. stands for complex conjugation.
\end{lemma}

\begin{proof}
Without loss of generality we may assume that points $z_i$ are radially ordered, $|z_1|\geq \cdots \geq |z_n|$.
We have
\begin{multline}
\ol{\langle \Phi_1(z_1)\cdots \Phi_n(z_n)\rangle }
=\ol{
\langle\vac|\wh\Phi_1(z_1)\cdots \wh\Phi_n(z_n)  |\vac\rangle  
}=\\
=\langle\vac|\Big(\wh\Phi_1(z_1)\cdots \wh\Phi_n(z_n)\Big)^+  |\vac\rangle 
=\langle\vac|\wh\Phi_n(z_n)^+\cdots \wh\Phi_1(z_1)^+  |\vac\rangle \\
\underset{(\ref{l23 Phi^+})}{=}
\prod_{i=1}^n \bar{z}_i^{-2h_i} z_i^{-2\bar{h}_i}\cdot
\langle\vac| \wh\Phi_n^*(1/\bar{z}_n)\cdots \wh\Phi_1^*(1/\bar{z}_1)  |\vac\rangle\\=
\prod_{i=1}^n \bar{z}_i^{-2h_i} z_i^{-2\bar{h}_i}\cdot
\langle \Phi_1^*(1/\bar{z}_1)\cdots \Phi_n^*(1/\bar{z}_n) \rangle
\end{multline}
\end{proof}

\subsection{Example: action of Virasoro algebra on $\HH$ in the scalar field theory. Abelian Sugawara construction.}
 
In the example of the free scalar field theory we know the stress-energy tensor (\ref{l20 T components}):
\begin{equation}\label{l23 T scalar}
\wh{T}(z)=-\frac12 :\dd\wh\phi(z) \dd\wh\phi(z):=\frac12\sum_{j,k\in\ZZ}z^{-j-k-2} :\wh{a}_j \wh{a}_k:
\end{equation}
where we used the expansion (\ref{l20 dd phi, dbar phi via a abar}) of $\dd\wh\phi(z)$ in terms of creation-annihilation operators. In particular, $\wh{T}(z)$ has no dependence on $\bar{z}$, i.e. the holomorphicity axiom (\ref{l23 dbar T=0}) holds (we skip the computations for $\wh{\ol{T}}(z)$ -- they are similar). OPEs (\ref{l23 TT}), (\ref{l23 Tbar Tbar}), (\ref{l23 T Tbar}) hold with the central charges $c=\bar{c}=1$ -- we know this from the explicit computation in Example \ref{l21 example: TT OPE}. From (\ref{l23 L via T}) and (\ref{l23 T scalar}) we find the operators $\wh{L}_n$ to be
\begin{equation}\label{l23 L_scalar via a a}
\wh{L}_n=\frac12 \sum_{k\in\ZZ} :\wh{a}_k\wh{a}_{n-k}:
\end{equation}
and similarly
\begin{equation}\label{l23 Lbar via abar abar}
\wh{\ol{L}}_n=\frac12 \sum_{k\in\ZZ} :\wh{\bar{a}}_k\wh{\bar{a}}_{n-k}:
\end{equation}
Note that the normal ordering is only relevant for $\wh{L}_n$, $\wh{\ol{L}}_n$ with $n=0$, as for $n\neq 0$ the operators $\wh{a}_k,\wh{a}_{n-k}$ commute for any $k$, and likewise for $\wh{\bar{a}}_k,\wh{\bar{a}}_{n-k}$.

\textbf{Exercise:}  Show by a direct computation that the operators (\ref{l23 L_scalar via a a}) satisfy Virasoro commutation relations with $c=1$, from the commutation relations (\ref{l17 a,abar comm rel}) for the creation/annihilation operators.

Equality (\ref{l23 L_scalar via a a}) expresses the generators of Virasoro algebra with central charge $c=1$ as quadratic polynomials in generators of the Heisenberg Lie algebra  (\ref{l17 Heisenberg Lie algebra}). Thus, we have an inclusion 
\begin{equation}
\mr{Vir}_{c=1}\hra U^{(2)}\mr{Heis},
\end{equation}
where $U^{(2)}$ means the subspace of (at most) quadratic elements in the universal enveloping algebra (of the Heisenberg Lie algebra). This inclusion is the abelian version of the Sugawara construction, realizing Virasoro algebra (at certain other values of $c$) inside the quadratic part of the universal enveloping algebra of the affine Lie algebra (a.k.a. Kac-Moody algebra)  $\wh{\g}$. We will come to the non-abelian Sugawara construction later, when talking about Wess-Zumino-Witten model.

\begin{remark} Comparing (\ref{l23 L_scalar via a a}) and (\ref{l23 Lbar via abar abar}) with  (\ref{l18 :H:}), (\ref{l18 :P:}), we observe the equalities
\begin{equation}\label{l23 H, P via L_0, Lbar_0}
\wh{L}_0+\wh{\ol{L}}_0 = \wh{H},\qquad \wh{L}_0-\wh{\ol{L}}_0 = \wh{P},
\end{equation}
expressing the quantum Hamiltonian and the total momentum operators in terms of Virasoro generators $\wh{L}_0$, $\wh{\ol{L}}_0$. In a general CFT, formulae (\ref{l23 H, P via L_0, Lbar_0}) become the \emph{definitions} of the Hamiltonian and the total momentum operators.

Note that due to (\ref{l23 rho}), the operator $\wh{H}=\wh{L}_0+\wh{\ol{L}}_0$ represents on $\HH$ the vector field $-z\dd_z-\bar{z}\dd_{\bar{z}}$ or, in terms of coordinates $\tau,\sigma$ on the cylinder, the vector field $-\dd_\tau$. Likewise, the operator $\wh{P}=\wh{L}_0-\wh{\ol{L}}_0$ represents the vector field $-z\dd_z+\bar{z}\dd_{\bar{z}}$ or, in terms of the cylinder, $i\dd_\sigma$. Ultimately, the operators represent infinitesimal translations along the cylinder and rotations of the cylinder, as the Hamiltonian and total momentum should, cf. Remark \ref{l16 Rem H and P via T}.
\end{remark}

\subsection{Virasoro commutation relations from $TT$ OPE (contour integration trick)}
\label{sss: Virasoro from TT OPE}

Let us prove Lemma \ref{l23 lemma: Virasoro from TT OPE}. We will focus on the case (i):
assuming that the $TT$ OPE (\ref{l23 TT}) is known, let us calculate the commutator of operators $\wh{L}_n$, $\wh{L_m}$ using their definition via the stress-energy tensor (\ref{l23 L via T}):
\begin{multline}\label{l23 [Ln,Lm] computation}
[\wh{L}_n,\wh{L}_m]=\wh{L}_n\wh{L}_m - \wh{L}_m\wh{L}_n=\\
=
\oint_{\gamma_{0,R}}\frac{dz}{2\pi i} \oint_{\gamma_{0,r}} \frac{dw}{2\pi i} z^{n+1}w^{m+1} \wh{T}(z)\wh{T}(w)-
\oint_{\gamma_{0,R}}\frac{dw}{2\pi i} \oint_{\gamma_{0,r}} \frac{dz}{2\pi i} w^{m+1}z^{n+1} \wh{T}(w)\wh{T}(z)\\
=\oint_{\gamma_{0,R}}\frac{dw}{2\pi i} \oint_{\Gamma} \frac{dz}{2\pi i} z^{n+1}w^{m+1} \mc{R}\left(\wh{T}(z)\wh{T}(w)\right).
\end{multline}
Here we denoted $\gamma_{z,r}$ the circle of radius $r$ centered at $z$, with counterclockwise orientation; we assume the two radii to satisfy $0<r<R$;
%where $\gamma_r,\gamma_R$ are circles centered at zero, oriented counterclockwise, of radii $r$ and $R$ with $0<r<R$;
 $\Gamma$ is the 1-cycle $\gamma_{R'}-\gamma_{r}$ with $R'>R$. We are exploiting the freedom to deform the integration contour, due to holomorphicity of the integrand for $z\neq w$ and $z,w\neq 0$ (in particular, the property (\ref{l23 dbar T=0})). We can then further deform the contour $\Gamma$ to the circle $\gamma_{w,\epsilon}$ centered at $w$, of radius $0<\epsilon<R$.
 
%\textcolor{red}{PICTURE: contour}
\begin{figure}[H]
%$$\vcenter{\hbox{ \includegraphics[scale=0.5]{l22_OPE.eps} }} $$
\begin{center}
\includegraphics[scale=0.7]{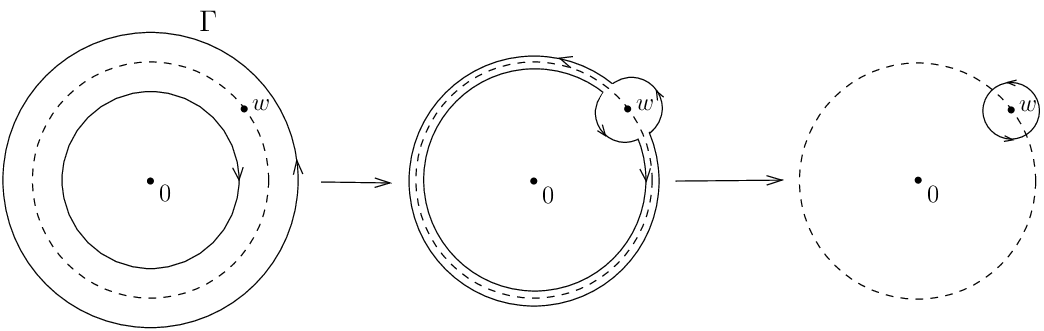}
\end{center}
\caption{Deformation of the integration contour for the integral over $z$ (solid curve). The dashed circle is the (fixed) integration contour for $w$.}
\end{figure}
 
Replacing the radially ordered product of stress energy tensors with the OPE (\ref{l23 TT}), we have then
\begin{multline}
[\wh{L}_n,\wh{L}_m]=
\oint_{\gamma_{0,R}}\frac{dw}{2\pi i} \oint_{\gamma_{w,\epsilon}} \frac{dz}{2\pi i} z^{n+1}w^{m+1} \left(\frac{\frac{c}{2}\wh{\mathbb{1}}}{(z-w)^4}+\frac{2\wh{T}(w)}{(z-w)^2}+\frac{\dd\wh{T}(w)}{z-w}+\mr{reg.}\right)\\
\hspace{-1cm}
\underset{
\begin{array}{c}
z=w+\alpha,\\
\mr{expand\;in\;}\alpha
\end{array}
}{=}
\oint_{\gamma_{0,R}}\frac{dw}{2\pi i} \oint_{\gamma_{0,\epsilon}} \frac{d\alpha}{2\pi i} \left(w^{n+1}+(n+1)w^n\alpha+\frac{(n+1)n}{2}w^{n-1}\alpha^2+\frac{(n+1)n(n-1)}{6}\alpha^3+\cdots\right)\cdot\\
\cdot
w^{m+1} \left(\frac{\frac{c}{2}\wh{\mathbb{1}}}{\alpha^4}+\frac{2\wh{T}(w)}{\alpha^2}+\frac{\dd\wh{T}(w)}{\alpha}+\mr{reg.}\right).
\end{multline}
Here the integral over $\alpha$ simply computes the residue at $\alpha=0$ of the integrand, i.e., the coefficient of $\alpha^{-1}$. Thus, continuing the computation we have
\begin{multline}
[\wh{L}_n,\wh{L}_m]=
\oint_{\gamma_{0,R}}\frac{dw}{2\pi i}\Big( \underbrace{w^{n+m+2} \dd\wh{T}(w)}_{\mr{integrate\;by\;parts}}+2(n+1)w^{n+m+1}\wh{T}(w)+\frac{(n+1)n(n-1)}{6}\frac{c}{2} w^{n+m-1}\wh{\mathbb{1}} \Big)\\
=\oint_{\gamma_{0,R}}\frac{dw}{2\pi i}\Big( 
\underbrace{(2(n+1)-(n+m+2))}_{n-m} w^{n+m+1}\wh{T}(w) + \frac{c}{12}(n^3-n)w^{n+m-1}\wh{\mathbb{1}}
\Big)\\
\underset{(\ref{l23 L via T})}{=}(n-m)\wh{L}_{n+m}+\frac{c}{12}(n^3-n)\delta_{n,-m}\wh{\mathbb{1}}.
\end{multline}
This is indeed the Virasoro commutation relation (\ref{l22 Virasoro via L}).
This proves case (i) of Lemma \ref{l23 lemma: Virasoro from TT OPE}. The other two cases are proved similarly.

\subsection{Digression: path integral heuristics, variation of a correlator in metric as an insertion of the stress-energy tensor, trace anomaly}
In the context of path integral quantization on a Riemann surface $\Sigma$, a correlator is represented by averaging over the space of classical fields with measure $e^{-S(\phi)}$:
\begin{equation}
\langle \Phi_1(z_1)\cdots \Phi_n(z_n) \rangle \colon= \int_{\F_\Sigma} \mc{D}\phi\, e^{-S(\phi)} \Phi(z_1)\cdots \Phi(z_n).
\end{equation}
We denote the classical field $\phi$.
The variation of this expression w.r.t. metric on $\Sigma$ is given by
\begin{multline}\label{l23 delta_g <...>}
\delta_g  \langle \Phi_1(z_1)\cdots \Phi_n(z_n) \rangle=
\delta_g\int_{\F_\Sigma} \mc{D}\phi \, e^{-S_g(\phi)}  \Phi(z_1)\cdots \Phi(z_n)\\
\underset{\mr{naively}}{=}\int_{\F_\Sigma} \mc{D}\phi \, e^{-S_g(\phi)}\,\underbrace{(-\delta_g S_g(\phi))}_{\frac{1}{2\pi}\int_\Sigma d^2z(T\mu +\ol{T}\bar\mu)}  \Phi(z_1)\cdots \Phi(z_n)\\
=\langle \frac{1}{2\pi}\int_\Sigma d^2z(T(z)\mu(z) +\ol{T}(z)\bar\mu(z))\; \Phi_1(z_1)\cdots \Phi_n(z_n) \rangle.
\end{multline}
Here we used the parametrization of a variation of metric (\ref{l14 variation of g}) by Beltrami differentials and Weyl factor, %(with the latter dropping out), 
and we used formula (\ref{l14 delta_g S via Beltrami}). Computation (\ref{l23 delta_g <...>}) tells us that the variation of a correlator in metric is given by the insertion of an extra field in the correlator -- the stress-energy tensor contracted with the Beltrami differential.

In fact, if the background metric on $\Sigma$ is not flat, there is a correction to the result (\ref{l23 delta_g <...>}) due to conformal anomaly (\ref{l2 change of Z under Weyl}):
\begin{equation}\label{l23 trace anomaly}
\delta_g  \langle \Phi_1(z_1)\cdots \Phi_n(z_n) \rangle=
\langle \frac{1}{2\pi}\int_\Sigma d^2z\Big(T(z)\mu(z) +\ol{T}(z)\bar\mu(z)\boxed{+c\frac{R_g(z)}{24}\omega(z)\mathbb{1}}\Big)\; \Phi_1(z_1)\cdots \Phi_n(z_n) \rangle
\end{equation}
where $\omega$ is the infinitesimal Weyl factor (cf. (\ref{l14 variation of g})) and $R_g$ is the scalar curvature of the metric. Heuristically, this correction can be  attributed to the dependence of the path integral measure in (\ref{l23 delta_g <...>}) on the metric. 

The correction term in (\ref{l23 trace anomaly}) corresponds to the fact that although classically the stress-energy tensor is traceless, in the quantum theory on a curved manifold the trace $\tr\, T= 4T_{z\bar{z}}$ has nonzero expectation value
\begin{equation}
\langle \tr\, T \rangle = c \frac{R(z)}{6}.
\end{equation}
This phenomenon is known as \emph{trace anomaly}. For instance, in the free boson theory one can obtain this result by calculating the variational derivative of the zeta-regularized determinant of the Laplacian in metric, cf. e.g. \cite[Appendix 5.A]{DMS}.

\marginpar{Lecture 24,\\
10/22/2022}

\section{Field-state correspondence in the example of the scalar field CFT}\label{sss: field-state correspondence}
%\subsection{Field-state correspondence for the scalar field CFT}

Let us examine the field-state correspondence is the map (\ref{l23 field-state correspondence}),
\begin{equation}\label{l24 field-state correspondence}
\begin{array}{cccc}
\mathsf{s}\colon & V &\ra & \HH\\
& \Phi & \mapsto & \displaystyle \lim_{z\ra 0} \wh{\Phi}(z)|\vac\rangle
\end{array}
\end{equation}
in the case of the scalar field theory. We start with simple examples.

For $\Phi=i\dd\phi$, we have 
\begin{equation}\label{l24 eq1}
\mathsf{s}(i\dd\phi)\colon=\lim_{z\ra 0}i\dd\wh\phi(z)|\vac\rangle = \lim_{z\ra 0} \sum_{n\in \ZZ} z^{-n-1}\wh{a}_n |\vac\rangle
\end{equation}
where we used (\ref{l20 dd phi, dbar phi via a abar}) to express the derivative of the field operator in terms of creation/annihilation operators. Notice that for $n\geq 0$ one has $\wh{a}_n|\vac\rangle =0$, while for $n\leq -2$ one has $z^{-n-1}\underset{z\ra 0}{\ra} 0$. So, the only surviving term in the r.h.s. of (\ref{l24 eq1}) is $n=-1$:
\begin{equation}
\mathsf{s}(i\dd\phi)= \wh{a}_{-1}|\vac\rangle
\end{equation}
-- a state with a single left-mover of energy-momentum $(1,1)$.

For higher derivatives of the fundamental field $\phi$ we find
\begin{equation}
\mathsf{s}(i\dd^p\phi)=\lim_{z\ra 0} i\dd^p \wh\phi(z) |\vac\rangle = 
\lim_{z\ra 0} \sum_{n\in\ZZ} (-n-1)(-n-2)\cdots (-n-p+1)z^{-n-p} \wh{a}_n|\vac\rangle
\end{equation}
where $p\geq 1$.
In the r.h.s. the summand satisfies the following:
\begin{itemize}
\item vanishes for $n\geq 0$, since then $\wh{a}_n|\vac\rangle =0$,
\item vanishes as $z\ra 0$ for $n\leq -p-1$, since then $\lim_{z\ra 0}z^{-n-p}=0$,
\item vanishes for $n=-1,-2,\ldots,-p+1$, since then the product ${(-n-1)(-n-2)\cdots (-n-p+1)}$ vanishes.
\end{itemize}
Thus, the only surviving term is $n= -p$:
\begin{equation}
\mathsf{s}(i\dd^p\phi)=(p-1)!\, \wh{a}_{-p}|\vac\rangle
\end{equation}
-- a state with a single left-mover of energy-momentum $(p,p)$.

\begin{remark}\label{l24 rem s(phi)}
Note that 
\begin{equation}
\fs(\phi)=\lim_{z\ra 0} \wh\phi(z)|\vac\rangle \underset{(\ref{l19 phi})}{=} \lim_{z\ra 0} \wh\phi_0 |\vac\rangle
\end{equation}
is ill-defined. This absence of the image of $\phi$ under field-state correspondence (together with the fact that correlators of $\phi$ are ill-defined) reinforces the point that $\phi$ should not be considered as an element of $V$ (while derivatives of $\phi$ are in $V$).
\end{remark}

As a more complicated example, consider the normally ordered differential monomial $\Phi=:i\dd\phi\,\dd\phi:$,
\begin{equation}\label{l24 eq2}
\fs(:i\dd\phi\, i\dd\phi:)=\lim_{z\ra 0} :i\dd\wh\phi(z)\, i\dd\wh\phi(z):|\vac\rangle = \sum_{n,m\in\ZZ} z^{-n-m-2} :\wh{a}_n\wh{a}_m:|\vac\rangle = \wh{a}_{-1}\wh{a}_{-1} |\vac\rangle
\end{equation}
-- the state with two left-moving quanta of energy-momentum $(1,1)$. Here the only surviving term in the double sum is $n=m=-1$, similarly to the situations above.

In particular, since the quantum stress-energy tensor (as an element of $V$) $T=-\frac12 :\dd\phi \dd\phi:$, we have
\begin{equation}
\fs(T)=\frac12 \wh{a}_{-1}\wh{a}_{-1} |\vac\rangle.
\end{equation}
Note that using (\ref{l23 L_scalar via a a}), we can write the r.h.s. as $\wh{L}_{-2}|\vac\rangle$.

\begin{remark}
In fact, in any CFT one has
\begin{equation}
\fs(T)=\wh{L}_{-2}|\vac\rangle.
\end{equation}
This, together with properties $\wh{L}_{\geq -1}|\vac\rangle =0$ and $\wh{L}_{-2-p}=\fs(\frac{1}{p!}\dd^p T)$ for $p\geq 0$, is a consequence of (\ref{l23 T via L, Tbar via Lbar}).
\end{remark}

A generalization of the examples above is the case where $\Phi$ is a general normally-ordered differential monomial in $\phi$:
\begin{multline}\label{l24 s(diff poly)}
\fs\Big(:\left( \prod_{j=1}^r \frac{i\dd^{n_j}\phi}{(n_j-1)!}\right) \left( \prod_{k=1}^s \frac{i\bar\dd^{\bar{n}_k}\phi}{(\bar{n}_k-1)!}\right)   :\Big) =\\ =
\wh{a}_{-n_1}\cdots \wh{a}_{-n_r} \wh{\ol{a}}_{-\bar{n}_1}\cdots \wh{\ol{a}}_{-\bar{n}_s}|\vac\rangle \underset{(\ref{l17 basis vector in H})}{=} \left|0;\{n_i\};\{\bar{n}_j\}\right\rangle
\end{multline}
where $1\leq n_1\leq\cdots\leq n_r$, $1\leq \bar{n}_1\leq\cdots \leq \bar{n}_s$. The computation is similar to the computations above (only a single term in the $(r+s)$-fold sum survives). Note that we identified all basis vectors  (\ref{l17 basis vector in H}) of $\HH$ \ul{with $\pi_0=0$} as images of particular vectors in $V$ (differential monomials), under the field-state correspondence. 

Since the map (\ref{l24 field-state correspondence}) is supposed to be an isomorphism, this means that $V$ should contain some more elements in addition to differential polynomials in $\phi$,\footnote{
We mean normally-ordered differential polynomials, where $\phi$ is not allowed to appear without derivatives, cf. Remark \ref{l24 rem s(phi)}.
} with images of these extra elements giving the states with $\pi_0\neq 0$.

\subsection{Vertex operators (in the scalar field theory)}\label{sss vertex operators}
A vertex operator is defined as
\begin{equation}\label{l24 Vhat}
\wh{V}_\alpha(z)\colon= \;\; :e^{i\alpha\wh\phi(z)}:
\end{equation}
where $\alpha\in \RR$ is a parameter (``charge'').
We emphasize that the vertex operator is a construction specific to the scalar field theory. We understand the operator (\ref{l24 Vhat}) as a local operator acting on $\HH$, corresponding to an abstract field $V_\alpha\in V$ placed at a point $z\in\CC$.

Let us find the state corresponding to the vertex operator $V_\alpha$:
\begin{multline}\label{l24 eq3}
\fs(V_\alpha)= \lim_{z\ra 0} \wh{V}_\alpha |\vac\rangle = \lim_{z\ra 0}:e^{i\alpha \wh\phi(z)}: |\vac\rangle  =
\\ \underset{(\ref{l19 phi})}{=} 
e^{i\alpha \sum_{n<0} \frac{i}{n}(\wh{a}_n z^{-n}+\wh{\ol{a}}_n \bar{z}^{-n}) } 
e^{i\alpha \sum_{n>0} \frac{i}{n}(\wh{a}_n z^{-n}+\wh{\ol{a}}_n \bar{z}^{-n}) } 
e^{i\alpha\wh\phi_0}   \cancel{e^{\alpha \wh\pi_0 \log(z\bar{z})}} |\vac\rangle
\end{multline}
Here the last exponential acting on $|\vac\rangle$ acts as identity, since $\wh\pi_0|\vac\rangle =0$. The next observation is that in Schr\"odinger representation of the quantum free particle system (corresponding to the zero-mode $\phi_0,\pi_0$), with states being $L^2$ functions of $\pi_0$, one has $\wh\pi_0 = \pi_0\cdot$ a multiplication operator and $\wh\phi_0 = i \frac{\dd}{\dd \pi_0}$ a derivation operator (cf. the discussion in  Example \ref{l19 ex <phi phi>}). Thus, the exponential 
$e^{i\alpha \wh\phi_0}\colon{ \psi(\pi_0)\mapsto \psi(\pi_0-\alpha)}$ is the shift operator. In particular, it maps the vacuum $|\pi_0=0\rangle$ represented by the delta-function centered at zero to the delta-function centered at $\alpha$, i.e., the vector $|\pi_0=\alpha\rangle$. In other words, $e^{i\alpha \wh\phi_0}$ maps the vacuum $|\vac\rangle$ to the pseudo-vacuum $|\pi_0=\alpha\rangle$ with zero-mode momentum $\alpha$. Thus, continuing the computation (\ref{l24 eq3}), we have
\begin{equation}
\fs(V_\alpha)=e^{i\alpha \sum_{n<0} \frac{i}{n}(\wh{a}_n z^{-n}+\wh{\ol{a}}_n \bar{z}^{-n}) } 
e^{i\alpha \sum_{n>0} \frac{i}{n}(\wh{a}_n z^{-n}+\wh{\ol{a}}_n \bar{z}^{-n}) }  |\pi_0=\alpha\rangle
\end{equation}
Here the right exponential acts by identity, since the annihilation operators in the exponent kill the pseudovacuum. The left exponential becomes identity as $z\ra 0$, thus one has
\begin{equation}
\fs(V_\alpha)= |\pi_0=\alpha\rangle
\end{equation}

So, the image of a vertex operator under the field-state correspondence is a pseudovacuum.
Combining this computation with the computation with (\ref{l24 s(diff poly)}), we have
\begin{equation}
\fs(: (\mr{differential\; monomial\;in\;}\phi)\cdot V_\alpha:) = |\alpha; \{n_i\};\{\bar{n}_j\}\rangle
\end{equation}
with differential monomial as in the l.h.s. of (\ref{l24 s(diff poly)}). Thus are recovering all basis vectors of $\HH$ as images of elements of $V$, once we have adjoined the vertex operators. Put another way, for the field-state correspondence to be an isomorphism, we should set the space of local fields in the scalar field theory to be
\begin{equation}
V= \mr{span}_\CC\{ :(\mr{differential\; polynomials\;in\;}\phi)\cdot V_\alpha:\;\;|\;\; \alpha\in \RR\}
\end{equation}
where as usual differential polynomials are not allowed to contain $\phi$ without derivatives.

\marginpar{Lecture 25,\\
10/26/2022}

\section{Local Virasoro action at a puncture
%on $V_z$
}
We continue with the CFT data/axioms list:
\begin{enumerate}[(I)]
\setcounter{enumi}{8}
\item \textbf{Local projective action of conformal vector fields on fields at a point $z$.}
Similarly to the projective action (\ref{l23 rho}) of conformal vector fields on states, one has a projective action of conformal vector fields with singularities (vector fields of the form $v=u(w)\dd_w+\ol{u(w)}\dd_{\bar{w}}$ with $u$ a meromorphic function)
 on fields at a point $z\in \CC\backslash\{0\}$,
\begin{equation}
\rho^{(z)}\colon \conf_\mr{sing}(\CC) \ra \mr{End}(V_z)
\end{equation}
given by 
\begin{equation}\label{l25 rho}
\rho^{(z)}(u\dd+\bar{u}\bar\dd)\circ \Phi(z)\colon =-\frac{1}{2\pi i}\oint_{\gamma_z}\left( 
dw\, u(w)\, T(w)\,\Phi(z) 
+ d\bar{w}\, \ol{u(w)}\, \ol{T}(w)\,\Phi(z) 
\right)
,
\end{equation}
for any field $\Phi(z)\in V_z$. Here $\gamma_z$ is a contour going around $z$ once in a positive direction (and small enough so that it does not enclose any poles of $u$ apart from $z$). We understand the r.h.s. of (\ref{l25 rho}) as defining a new local field at $z$.
Equality (\ref{l25 rho}) is understood either (a) as an equality under a correlator with an arbitrary collection of test field, or (b) as equality of local operators (then we put hats on $T$, $\Phi$ and the l.h.s., and we radially order the operator product in the r.h.s.).
\end{enumerate}

In particular, the vector fields $-(w-z)^{n+1}\dd_w$  (standard meromorphic vector fields generating the Witt algebra, centered at $z$ instead of the origin) correspond to certain operators $L_n^{(z)}$ acting on $V_z$:
\begin{equation}\label{l25 L_n loc}
L_n^{(z)}\Phi(z)\colon =\rho^{(z)} (-(w-z)^{n+1}\dd_w)\Phi(z)=\frac{1}{2\pi i}\oint_{\gamma_z} dw\,(w-z)^{n+1} T(w)\, \Phi(z)
\end{equation}
We will also write the l.h.s. as $(L_n\Phi)(z)$.
Calculating the integral on the right as a residue, we observe that the fields $L_n\Phi$ are the coefficients of the OPE $T(w)\Phi(z)$:
\begin{multline}\label{l25 T Phi via L_n}
T(w)\Phi(z)\sim \sum_{n\in \ZZ} (w-z)^{-n-2} (L_n\Phi)(z)=\\
= \cdots+\frac{(L_{1}\Phi)(z)}{(w-z)^3} +\frac{(L_{0}\Phi)(z)}{(w-z)^2}+\frac{(L_{-1}\Phi)(z)}{w-z}+\underbrace{(L_{-2}\Phi)(z)+(w-z)(L_{-3}\Phi)(z)+\cdots}_{\mr{reg.}}
\end{multline}

By an argument similar to  Lemma \ref{l23 lemma: Virasoro from TT OPE} (and the computation of Section \ref{sss: Virasoro from TT OPE}), operators $L_n^{z}$ acting on $V_z$ satisfy Virasoro commutation relations.

Similarly to (\ref{l25 L_n loc}), one defines operators $\ol{L}_n^{(z)}$ acting on $V_z$, corresponding to the terms in the OPE $\ol{T}(w)\Phi(z)$.

%\marginpar{Explain more/prove?}
Remark \ref{l23 rem: L_n^+} has an analog for the hermitian conjugates of  the operators $L_n^{(z)},\ol{L}_n^{(z)}$:
\begin{equation}\label{l25 L_n^+ local}
(L_n^{(z)})^+=L_{-n}^{(z)},\quad (\ol{L}_n^{(z)})^+=\ol{L}_{-n}^{(z)}.
\end{equation}
This follows from Remark \ref{l23 rem: L_n^+} by field-state correspondence.

\begin{remark}
Consider the OPE (\ref{l25 T Phi via L_n}) for $\Phi=\mathbb{1}$ the identity field. One has
\begin{equation}\label{l25 T 1}
T(w)\mathbb{1}(z)=T(w) =\sum_{n\geq 0}\frac{1}{n!}(w-z)^n \dd^n T(z)
\end{equation}
Where on the right we have the Taylor expansion of $T(w)$ centered at $z$; the sum is convergent for $w$ sufficiently close to $z$.\footnote{
More precisely, under a correlator with test fields $\Phi_1(z_1),\ldots,\Phi_n(z_n)$, the field $T(w)$ can be replaced with the r.h.s. of (\ref{l25 T 1}) -- and the sum is convergent -- if $|w-z|<|z_i-z|$ for all $i$.
}
Comparing the coefficients in (\ref{l25 T 1}) and in (\ref{l25 T Phi via L_n}) with $\Phi=\mathbb{1}$, we obtain 
\begin{equation}
%(L_{n} T)(z)=\left\{ 
%\begin{array}{cc}
%0, & n\geq -1,\\
%\frac{1}{k!}\dd^k T(z), & n=-2-k\;\mr{with}\; k\geq 0
%\end{array}
%\right.
\ldots,L_1 \mathbb{1}=0,\; L_0 \mathbb{1}=0,\; L_{-1} \mathbb{1} =0,\; \boxed{L_{-2} \mathbb{1}=T},\; L_{-3} \mathbb{1}=\dd T,\; L_{-4}\mathbb{1}=\frac{1}{2!}\dd^2 T\,\ldots
\end{equation}
One has similar formulae for $\ol{L}_n\mathbb{1}$, in particular, one has $\ol{L}_{-2}\mathbb{1}=\ol{T}$.
\end{remark}

\begin{enumerate}[(I)]
\setcounter{enumi}{9}
\item \textbf{$L_{-1}$ axiom.}\footnote{Informally, the axiom 
can be phrased as ``$L_{-1}$ acts by infinitesimally moving the puncture $z$.''
} 
For any $\Phi\in V$ one has
\begin{eqnarray}
(L_{-1}\Phi) (z) &=& \dd \Phi(z),\\
(\ol{L}_{-1}\Phi) (z) &=& \bar\dd \Phi(z).
\end{eqnarray}
Here one understands that the field $\dd\Phi(z)$ is defined by its behavior under a correlator with test fields: $\langle \dd\Phi(z) \Phi_1(z_1)\cdots \Phi_n(z_n) \rangle = \dd_z \langle \Phi(z)\Phi_1(z_1)\cdots \Phi_n(z_n) \rangle$ (or in the language of field operators, $\wh{\dd\Phi}(z)\colon= \frac{\dd}{\dd z} \wh\Phi(z)$). The case of $\bar\dd\Phi$ is similar.
\end{enumerate}

\begin{remark} If $v=u\dd+\bar{u}\bar\dd$ is a conformal vector field on $\CC$ without singularities (except possibly at zero), then the field operator corresponding to (\ref{l25 rho}) is
\begin{equation}\label{l25 rho=[rho,-]}
\wh{\rho^{(z)}(v)\circ\Phi(z)}=[\rho(v),\wh\Phi(z)]
\end{equation}
where the r.h.s. is the commutator of the field operator with the operator (\ref{l23 rho}) representing the vector field $v$ on the space of states $\HH$. Equality (\ref{l25 rho=[rho,-]}) is proven by a contour integration trick similar to one of Section \ref{sss: Virasoro from TT OPE}: the r.h.s. of (\ref{l25 rho=[rho,-]}) is an integral over a cycle $\Gamma$ -- the difference of two circles, one of radius $R>|z|$ and one of radius $r<|z|$; this contour can be deformed to a single circle centered at $z$, which yields the l.h.s. of (\ref{l25 rho=[rho,-]}).
\end{remark}

\begin{definition}\label{l25 def: conformal dimension}
We say that a field $V\in \Phi$ has \emph{conformal weight} (or \emph{conformal dimension}) $(h,\bar{h})\in \RR^2$ if one has
\begin{equation}
(L_0\Phi)(z) = h\Phi(z),\qquad (\ol{L}_0\Phi)(z)=\bar{h}\Phi(z),
\end{equation}
i.e., $\Phi$ is an eigenvector of operators $L_0,\ol{L}_0$ simultaneously, with eigenvalues $h,\bar{h}$.
\end{definition}

\begin{example} \label{l25 ex: T conformal weight}
Consider (\ref{l25 T Phi via L_n}) for $\Phi=T$ and compare with the standard $TT$ OPE (\ref{l23 TT}). We obtain
\begin{equation}
L_{\geq 3}T=0,\quad L_2 T = \frac{c}{2}\mathbb{1},\quad L_1 T=0,\quad L_0 T=2 T,\quad L_{-1}T=\dd T
\end{equation}
Likewise, from $\ol{T} T$ OPE (\ref{l23 T Tbar}) we have
\begin{equation}
\ol{L}_{\geq -1}T=0.
\end{equation}
In particular, we see that $T$ has conformal weight $(h,\bar{h})=(2,0)$. Similarly, $\ol{T}$ has conformal weight $(0,2)$.
\end{example}

We will be assuming that $L_0,\ol{L}_0$ are simultaneously diagonalizable on $V$\footnote{
There are interesting examples of CFTs where this diagonalizability assumption fails. Such CFTs are called ``logarithmic.''
} (this assumption is in fact a part of the highest weight axiom (\ref{l25 highest weight axiom}) below),
thus
the space $V$ is bi-graded by conformal weight:
\begin{equation}
V=\bigoplus_{h,\bar{h}\subset \Delta}V^{h,\bar{h}}
\end{equation}
where $\Delta\subset \RR^2$ is some set of admissible conformal weights.

The action of a Virasoro generator $L_{-n}$ changes the conformal weight of a field as\footnote{
This is a consequence of the relation $[L_0,L_{-n}]=nL_{-n}$ in Virasoro algebra: if $L_0\Phi=h\Phi$, then one has $L_0 (L_{-n}\Phi)=L_{-n}(L_0+n)\Phi=(h+n)L_{-n}\Phi$. Likewise, $[\ol{L}_0,L_{-n}]=0$ implies that the eigenvalue of $\ol{L}_0$ does not change under the action of $L_{-n}$.
} 
\begin{equation}\label{l25 h -> h+n}
(h,\bar{h})\ra (h+n, \bar{h}).
\end{equation}
Similarly, the action of $\ol{L}_{-n}$ changes the conformal weight as 
\begin{equation}\label{l25 hbar -> hbar+n}
(h,\bar{h})\ra (h,\bar{h}+n).
\end{equation}

The following is a standard assumption on admissible conformal weights.
\begin{assumption}\label{l25 assump: (h,hbar) properties}\leavevmode
%\begin{align}
%\label{l25 conf weight assump: positive} \Delta\subset \RR_{\geq 0}\times \RR_{\geq 0},\\
%\label{l25 conf weight assump: h-hbar integer}
%\mbox{If $(h,\bar{h})\in \Delta$ then $h-\bar{h}\in \ZZ$},\\
%\label{l25 conf weight assump: V^00 spanned by 1}
%V^{0,0}=\mr{Span}(\mathbb{1}).
%\end{align}
\begin{enumerate}[(a)]
\item \label{l25 assump (a)} $\Delta\subset \RR_{\geq 0}\times \RR_{\geq 0}$,\footnote{Otherwise the 2-point correlator $\langle \Phi(z) \Phi(w)\rangle$ can grow as points $z$ and $w$ become farther and farther apart, which contradicts the physical intuition of local interactions.}
\item \label{l25 assump (b)} If $(h,\bar{h})\in \Delta$ then\footnote{Needed for single-valuedness of correlators, cf. Remark \ref{l4 rem: h-hbar in Z}.}
\begin{equation}\label{l25 h-hbar in Z}
h-\bar{h}\in \ZZ,
\end{equation}
\item \label{l25 assump (c)} $V^{0,0}=\mr{Span}(\mathbb{1})$.
\end{enumerate}
\end{assumption}

\begin{remark}
Assumptions (\ref{l25 assump (a)}), (\ref{l25 assump (c)}) above pertain to unitary CFTs. E.g., (\ref{l25 assump (a)}) is implied by unitarity and the highest weight axiom (\ref{l25 highest weight axiom}).\footnote{
Indeed, for $\Phi$ a primary field of conformal weight $(h,\bar{h})$, one has $0\leq ||L_{-1}\Phi ||^2= \langle L_{-1}\Phi,\L_{-1}\Phi \rangle_V= \langle \Phi, L_1 L_{-1} \Phi \rangle_V =
\langle \Phi, (2L_0+L_{-1}L_1)\Phi \rangle = 2h ||\Phi ||^2$, which implies $h\geq 0$. Here we used that $L_1\Phi=0$, since $\Phi$ is assumed to be primary. By a similar argument, $\bar{h}\geq 0$. By Axiom (\ref{l25 highest weight axiom}), any field is a descendant of a primary field (or a linear combination of such). Hence, knowing that conformal weights are nonnegative for primary fields ensures that they are nonnegative for all fields.
} Assumption (\ref{l25 assump (c)}) does not follow from unitarity, but is rather an axiom in a unitary CFT -- the uniqueness (nondegeneracy) of vacuum.

  There are interesting non-unitary CFTs where assumptions (\ref{l25 assump (a)}), (\ref{l25 assump (c)}) fail, e.g., the $bc$ system, see Section \ref{s bc system}.\footnote{
In the $bc$ system, the ghost field $c$ has conformal weight $(-1,0)$, thus violating (\ref{l25 assump (a)}). Also, the field $\dd c$ has conformal weight $(0,0)$ and is linearly independent from $\mathbb{1}$, thus violating (\ref{l25 assump (c)}).
}

Assumption (\ref{l25 assump (b)}) holds in ``ordinary'' CFTs, with single-valued correlators. However, it is useful to consider 
(as auxiliary objects) 
a class of CFTs where (\ref{l25 assump (b)}) fails -- the so-called chiral CFTs. They arise in the holomorphic factorization of the correlators of ordinary (single-valued) CFTs, cf. Sections \ref{ss chiral boson}, \ref{ss space of states for chiral fermion}, \ref{ss WZW conf blocks}.
\end{remark}

%\begin{lemma}
%If $\Phi\in V$ is a field  of conformal weight $(h,0)$ (i.e. $\bar{h}=0$), then $\bar\dd\Phi=0$. In particular, correlation functions of the form $\langle \Phi(z) \Phi_1(z_1)\cdots \Phi_n(z_n) \rangle$ are holomorphic in $z$. Similarly, if $\Phi$ has conformal weight $(0,\bar{h})$, then $\dd\Phi=0$.
%\end{lemma}
%
%\begin{proof}
%We have 
%\end{proof}

%\begin{remark} Stress-energy tensor can be interpreted as a 1-form on the moduli space 
%\end{remark}

\section{Primary fields}
\begin{definition}
A field $V\in \Phi$ is said to be primary, of conformal weight $(h,\bar{h})$ if it satisfies the OPE
\begin{equation}\label{l25 primary def via OPE}
T(w)\Phi(z) \sim \frac{h \Phi(z)}{(w-z)^2}+\frac{\dd \Phi(z)}{w-z}+\mr{reg.},\quad  
\ol{T}(w)\Phi(z) \sim \frac{\bar{h} \Phi(z)}{(\bar{w}-\bar{z})^2}+\frac{\bar\dd \Phi(z)}{\bar{w}-\bar{z}}+\mr{reg.}
\end{equation} 
\end{definition}
Equivalently, $\Phi\in V$ is primary, with conformal weight $(h,\bar{h})$, if
\begin{equation}
\begin{array}{ll}
L_{>0}\Phi=0, & \ol{L}_{>0}\Phi=0,\\
L_0\Phi =h\Phi, & \ol{L}_0\Phi=\bar{h}\Phi.
\end{array}
\end{equation}
Put another way, a primary field is a highest weight vector of $V$ as a module over $\mr{Virasoro}\oplus \ol{\mr{Virasoro}}$, of weight $(h,\bar{h})$.

For $\Phi$ a primary field, fields obtained from it by repeated application of negative Virasoro generators $L_{<0},\ol{L}_{<0}$, i.e., fields of the form 
\begin{equation}\label{l25 descendant}
L_{-k_r}\cdots L_{-k_1}\ol{L}_{-l_s}\cdots \ol{L}_{-l_1}\Phi
\end{equation}
with $k_1,\ldots,k_r,l_1,\ldots, l_s\geq 1$, are called ``descendants'' of $\Phi$. If $\Phi$ has conformal weight $(h,\bar{h})$ then the descendant (\ref{l25 descendant}) has conformal weight $(h+\sum_i k_i, \bar{h}+\sum_j l_j)$ (cf. (\ref{l25 h -> h+n}), (\ref{l25 hbar -> hbar+n})).
The subspace of $V$ spanned by all descendants of a primary field $\Phi$ is called the ``conformal family'' of $\Phi$.
%\marginpar{add a table of first descendants?}

\begin{enumerate}[(I)]
\setcounter{enumi}{10}
\item \textbf{Highest weight axiom.} The space of fields $V$ splits as a direct sum of irreducible highest weight modules of the Lie algebra $\mr{Vir}\oplus\ol{\mr{Vir}}$ with primary fields being the highest weight vectors:
\begin{equation}\label{l25 highest weight axiom}
V=\bigoplus_\alpha V^{(\Phi_\alpha)}.
\end{equation}
Here the sum is over species of primary fields (i.e. over a basis in the subspace of primary fields in $V$); $V^{(\Phi_\alpha)}$ is the conformal family of $\Phi_\alpha$.
\end{enumerate}

\begin{remark}\label{l25 rem linear dependencies between descendants}
There can be linear dependencies between descendants of a given $\Phi_\alpha$.\footnote{For instance, in any CFT one has $L_{-1}\mathbb{1}=0$. Also, see (\ref{l25 vanishing descendant}) below for a nontrivial example in scalar field theory.}
More precisely, one can consider a \emph{Verma module} $\mathbb{V}^{h,\bar{h}}$ (free highest weight module) of the Lie algebra $\mr{Vir}\oplus\ol{\mr{Vir}}$ -- the span of formal expressions $L_{-k_r}\cdots L_{-k_1}\ol{L}_{-l_s}\cdots \ol{L}_{-l_1}\Phi_\alpha$ with $1\leq k_1\leq\cdots \leq k_r$, $1\leq l_1\leq\cdots\leq l_s$ (i.e. all ordered descendants are considered to be independent), with $\Phi_\alpha$ a vector of weight $(h,\bar{h})$ and annihilated by $L_{>0},\ol{L}_{>0}$. Then $V^{(\Phi_\alpha)}$ is the quotient of the Verma module $\mathbb{V}^{h,\bar{h}}$ by a submodule,
\begin{equation}
V^{(\Phi_\alpha)}\simeq\mathbb{V}^{(h,\bar{h})}/\mathsf{N}
\end{equation}
The submodule $\mathsf{N}$ that one quotients out 
is the kernel of the sesquilinear form $\langle,\rangle$ on $\mathbb{V}^{h,\bar{h}}$, defined in such a way that one has $L_n^+=L_{-n}, \ol{L}_n^+=\ol{L}_{-n}$ and the highest vector has norm $1$ (in particular, vectors in $\mathsf{N}$ have zero norm). Also, $\mathsf{N}\subset \mathbb{V}^{h,\bar{h}}$ is the submodule generated by ``null vectors'' $\chi\in\mathbb{V}^{h,\bar{h}}$ -- vectors with with the property $L_{>0}\chi=0, \ol{L}_{>0}\chi=0$.
%consists of vectors of zero $L^2$ norm in $\mathbb{V}^{h,\bar{h}}$, where the hermitian form is such that $L_n^+=L_{-n}, \ol{L}_n^+=\ol{L}_{-n}$ and the highest vector has norm $1$.\marginpar{Is this right, or should I better say that we are quotienting out vectors orthogonal to everything in $\mathbb{V}$? Talk about singular vectors?}
We refer to Chapter \ref{ch Vir rep theory} for more details.
\end{remark}

\subsection{Transformation property of a primary field}\label{sss: tranformation property of primary fields}
Let us fix a conformal vector field $v=u(w)\dd_w+\ol{u(w)}\dd_{\bar{w}}$ regular at $z$. For $\Phi\in V$ a primary of conformal weight $(h,\bar{h})$, by (\ref{l25 rho}) and (\ref{l25 primary def via OPE})  we have
\begin{multline}\label{l25 delta_v Phi computation}
\rho^{(z)}(u\dd + \bar{u}\bar\dd)\Phi(z)=\\
=
-\frac{1}{2\pi i}\oint_{\gamma_z} dw\, \underbrace{u(w)}_{u(z)+(w-z)\dd u(z)+\cdots} \underbrace{T(w) \Phi(z)}_{\frac{h\Phi(z)}{(w-z)^2}+\frac{\dd\Phi(z)}{w-z}+\mr{reg.}} + d\bar{w} \underbrace{\ol{u(w)}}_{\ol{u(z)}+(\bar{w}-\bar{z})\bar\dd \ol{u(z)}+\cdots}\underbrace{\ol{T}(w)\Phi(z)}_{\frac{\bar{h}\Phi(z)}{(\bar{w}-\bar{z})^2}+\frac{\bar\dd \Phi(z)}{\bar{w}-\bar{z}}+\mr{reg.}}\\
=
-u(z)\dd\Phi(z)-\ol{u(z)}\bar\dd\Phi(z)-h\dd u(z)\Phi(z)-\bar{h} \bar\dd\,\ol{u(z)}\,\Phi(z)
\end{multline}
-- a computation of the contour integral as a residue.

\subsubsection{Finite version, interpretation \#1: ``active transformations.''}
Formula (\ref{l25 delta_v Phi computation}) expresses the change of a field under an infinitesimal conformal map. For a finite conformal (holomorphic) map $z\mapsto w(z)$, it implies that the field transforms as
\begin{equation}\label{l25 primary field finite transf}
%\Phi(z) \mapsto \Phi(w) \left(\frac{\dd w}{\dd z}\right)^h  \left(\frac{\dd \bar{w}}{\dd \bar{z}}\right)^{\bar{h}}
%
\Phi(z) \mapsto \Phi'(w)= \left(\frac{\dd w}{\dd z}\right)^{-h}  \left(\frac{\dd \bar{w}}{\dd \bar{z}}\right)^{-\bar{h}} \Phi(z)
\end{equation}
As a check of compatibility with (\ref{l25 delta_v Phi computation}), take a map close to identity, $w(z)=z+\epsilon u(z)$. Then in the first order in $\epsilon$ we have
\begin{equation}
\delta \Phi(z) = \Phi'(z)-\Phi(z)=\epsilon\cdot(\mr{r.h.s.\;of\;(\ref{l25 delta_v Phi computation})})
\end{equation}
In Section \ref{sss constraints on corr from global conf sym} below
%\textcolor{red}{[LINK]} 
we will see %\marginpar{put a link to Ward identity} 
that (\ref{l25 primary field finite transf}) will become an equivariance property of correlators of primary fields under the diagonal action of a global conformal map on all field insertion points.
%In particular,  in a correlator of several primary fields at points $z_1,\ldots,z_n$, one can replace them with the correlator of the same types of fields at points $w_1=w(z_1),\ldots, w_n=w(z_n)$, multiplied by the product of Jacobians, as in the r.h.s. of (\ref{l25 primary field finite transf}) -- the correlator is invariant under such changes (this is a special case of the Ward identity, see below).\marginpar{put a link to Ward identity}

\subsubsection{Finite version, interpretation \#2: ``passive transformations''} Instead of moving points on the surface $\Sigma=\CC$, we can think about $z\mapsto w(z)$ as a change of local coordinate. We will use $z,w$ as names of local coordinate charts and call $p$  (previously $z$) the point on $\Sigma$.
 Think of the vector bundle $\mc{V}$ of fields over $\Sigma$; it has typical fiber $V$ and its local trivialization at a point $p$ depends on a choice of local coordinate $z$ or $w$ around $p$. Thus there is an isomorphism $V\ra V_p$ from the standard fiber to the particular fiber over the point depending on a choice of a local coordinate near $p$.
 Fix $\Phi\in V$ a field and denote its image in $V_p$ using 
% the trivialization induced by 
the chart $z$ by $\Phi_{(z)}(p)$. Then we have 
%As we change the coordinate chart around $p$ to a new one $w$, we have an isomorphism (transition function) $V_p\ra V_p$,
%$V_p^{(z)}\ra V_p^{(w)}$ (superscipts mean ``in the $z$-trivialization'', ``in the $w$-trivialization''),
\begin{equation}
\Phi_{(w)}(p)=\left(\left.\frac{\dd w}{\dd z}\right|_p\right)^{-h}  \left(\left.\frac{\dd \bar{w}}{\dd \bar{z}}\right|_p\right)^{-\bar{h}}\Phi_{(z)}(p)
\end{equation} 
Thus, the Jacobian on the right hand side is the transition function of the vector bundle.

Put another way, the combination
\begin{equation}\label{l25 Phi underline}
\mathbf{\Phi}(z)=\Phi(z) (dz)^h (d\bar{z})^{\bar{h}}
\end{equation}
is a coordinate-independent object valued in the line bundle 
\begin{equation}\label{l25 line bun}
K^{h,\bar{h}}\colon=K^{\otimes h}\otimes \ol{K}^{\otimes \bar{h}}
\end{equation}
over $\Sigma$. Here $K=(T^{1,0})^*\Sigma$ is the line bundle of $(1,0)$-forms and $\ol{K}=(T^{0,1})^*\Sigma$ is the line bundle of $(0,1)$-forms. For instance (see below), a correlation function of primary fields is a section of the product of several line bundles (\ref{l25 line bun}) pulled back to the configuration space of points on $\Sigma$.

\subsection{Examples of primary fields in scalar field theory}\label{sss primary fields in scalar field CFT}
In the scalar field theory, we have the following.
\begin{itemize}
\item The field $\dd\phi(z)$ is primary, with $(h,\bar{h})=(1,0)$. This follows from the OPEs (\ref{l21 T dd phi OPE}), (\ref{l21 T dbar phi OPE}). %\footnote{Note that in the scalar field theory one has $\bar\dd\dd\phi(z)=0$ (one can see it e.g. at the level of field operators from (\ref{l19 phi})). So the fact that $\ol{T}(w)\dd\phi(z)$ is regular is not in contradiction with the expectation that the first order pole should be $\bar\dd\dd\phi(z)$.}
Similarly, $\bar\dd\phi(z)$ is $(0,1)$-primary.
\item The stress-energy tensor $T$ is a field of conformal weight $(2,0)$, but it is not primary, since $T(w)T(z)$ OPE contains a fourth-order pole (\ref{l21 TT OPE computation}).
\item The field $\dd^2\phi$ has conformal weight $(2,0)$  but is not primary: differentiating (\ref{l21 T dd phi OPE}) we have
\begin{equation}
T(w) \dd^2\phi(z)\sim \frac{2\dd\phi(z)}{(w-z)^3}+\frac{2\dd^2\phi(z)}{(w-z)^2}+\frac{\dd^3\phi(z)}{w-z}+\mr{reg.}
\end{equation}
-- contains a third-order pole.
%\item The vertex operator $V_\alpha=:e^{i\alpha\phi}:$, with $\alpha$ any real number, is primary, of conformal weight $(h,\bar{h})=(\frac{\alpha^2}{2},\frac{\alpha^2}{2})$. Note that $h,\bar{h}$ are (generally) not integers! (Thus, in particular, we really do need real tensor powers of the line bundles in (\ref{l25 line bun})).
\end{itemize}

Here is another example.
\begin{lemma}
The vertex operator $V_\alpha=:e^{i\alpha\phi}:$, with $\alpha$ any real number, is primary, of conformal weight $(h,\bar{h})=(\frac{\alpha^2}{2},\frac{\alpha^2}{2})$.
\end{lemma}
Note that $h,\bar{h}$ are (generally) not integers! (Thus, in particular, we really do need real tensor powers of the line bundles in (\ref{l25 line bun})).

\begin{proof}
Let us calculate the OPE $T(w)V_\alpha(z)$ in the language of field operators:
\begin{multline}\label{l25 TV OPE computation}
\mc{R} \wh{T}(w)\wh{V}_\alpha(z) =\mc{R} :-\frac12 \dd\wh\phi(w)\dd\wh\phi(w):\;\sum_{n\geq 0} \frac{(i\alpha)^n}{n!} :\underbrace{\wh\phi(z)\wh\phi(z)\cdots \wh\phi(z)}_n:=
\\ \underset{\mr{Wick}}{=}
\wick{ :-\frac12 \dd\c1{\wh\phi(w)}\dd\c2{\wh\phi(w)}\;\sum_{n\geq 0} \frac{(i\alpha)^n}{n!} n(n-1) \c1{\wh\phi(z)}\c2{\wh\phi(z)}\underbrace{\wh\phi(z)\cdots \wh\phi(z)}_{n-2}:}+\\
+\wick{ :- \dd\wh\phi(w)\dd\c1{\wh\phi(w)}\;\sum_{n\geq 0} \frac{(i\alpha)^n}{n!} n \c1{\wh\phi(z)}\underbrace{\wh\phi(z)\cdots \wh\phi(z)}_{n-1}:}+ 
\underbrace{:\wh{T}(w) \wh{V}_\alpha(z):}_{\mr{reg.}}\\
\sim-\frac12 \frac{1}{(w-z)^2}(i\alpha)^2 \wh{V}_\alpha(z)+
\frac{1}{w-z}:\dd\wh\phi(w)(i\alpha)e^{i\alpha\wh{\phi}(z)}:+\mr{reg.}\\
\sim \frac{\frac{\alpha^2}{2}\wh{V}_\alpha(z)}{(w-z)^2}+\frac{\dd \wh{V}_\alpha(z)}{w-z}+\mr{reg.}
\end{multline}
The OPE $\ol{T}(w)V_\alpha(z)$ is computed similarly. Comparing with (\ref{l25 primary def via OPE}) we see that $V_\alpha$ is primary (no cubic or higher poles in the OPE with $T$, $\ol{T}$), and the conformal weight is $h=\bar{h}=\frac{\alpha^2}{2}$ as claimed.
\end{proof}

%\textbf{Exercise.} 
\begin{exercise}
\begin{enumerate}[(a)]
\item
Show that the field 
\begin{equation}\label{l25 h=4 primary field}
:2(\dd \phi)^4-3 (\dd^2\phi)^2+2\dd\phi\, \dd^3\phi:
\end{equation} 
is primary, of conformal weight $(4,0)$. Or, equivalently, check that the corresponding by (\ref{l24 s(diff poly)}) state $(2\wh{a}_{-1}^4+3 \wh{a}_{-2}^2-4 \wh{a}_{-3}\wh{a}_{-1})|\vac\rangle \in \HH$ is annihilated by operators $\wh{L}_{>0}$ and has eigenvalue $4$ w.r.t. $\wh{L}_0$.
\item Show that one has %the Virasoro descendant
\begin{equation}\label{l25 vanishing descendant}
(2L_{-3}-4 L_{-2}L_{-1}+L_{-1}^3)\dd\phi =0,
\end{equation}
i.e., this particular Virasoro descendant of the primary field $\dd\phi$ in free scalar theory vanishes (in terms of Remark \ref{l25 rem linear dependencies between descendants}, this descendant belongs to $\mathsf{N}$ -- the quotiented-out submodule).
\end{enumerate}
\end{exercise}

\begin{remark} %\marginpar{EDIT}
Classification of all primary fields in scalar field theory is a nontrivial problem; the answer is known as a corollary of a theorem by Feigin-Fuchs \cite{Feigin-Fuchs} (Theorem \ref{l34 thm FF maps of Verma modules} below).
 
First note that the space of fields (or space of states) of the free scalar theory is
\begin{equation}
V=\bigoplus_{\alpha\in\RR} \mathbb{V}^\mr{Heis}_\alpha\otimes \mathbb{V}^{\ol{\mr{Heis}}}_\alpha
\end{equation}
where $\mathbb{V}_\alpha^\mr{Heis}$ be the Verma module of the Heisenberg Lie algebra, with highest vector of $\wh{a}_0$-weight $\alpha$. 

Let  $M_{h}$ be the highest weight irreducible Virasoro module with $L_0$-highest weight $h$ and central charge $c=1$ and let $\mathbb{V}^\mr{Vir}_h$ be the highest weight Verma module (possibly reducible) of the Virasoro algebra with $L_0$-highest weight $h$ and central charge $c=1$. One has:
\begin{enumerate}
\item If $\alpha\not\in \frac{1}{\sqrt{2}}\ZZ$ then
\begin{equation}
\mathbb{V}_\alpha^\mr{Heis}\simeq M_{\frac{\alpha^2}{2}}=\mathbb{V}^\mr{Vir}_{\frac{\alpha^2}{2}}
\end{equation} 
is a single irreducible representation of Virasoro and contains no null vectors.
%$h\not\in \{\frac14 N^2 \}_{N\in \ZZ}$, 
%then $\mathbb{V}_h=M_h$ is irreducible and has no singular vectors.
\item If $\alpha=\pm \frac{N}{\sqrt{2}}$ for some $N=0,1,2,\ldots$, then one has
\begin{equation}
\mathbb{V}_\alpha^\mr{Heis}\simeq M_{\frac{N^2}{4}}\oplus M_{\frac{(N+2)^2}{4}} \oplus M_{\frac{(N+4)^2}{4}}  \oplus \cdots
\end{equation}
For instance, $\mathbb{V}^\mr{Heis}_0$ contains an infinite sequence of Virasoro-highest weight (primary) vectors $\chi_0=\mathbb{1},\chi_1=i\dd\phi,\chi_2,\chi_3,\ldots$, with $\chi_n$ having conformal weight $h=n^2$; $\chi_2$ is given explicitly by (\ref{l25 h=4 primary field}).

In the full scalar field theory, the Verma module $\mathbb{V}_{0,0}^{\mr{Heis}\oplus\ol{\mr{Heis}}}=\mathbb{V}^\mr{Heis}_0\otimes \mathbb{V}^{\ol{\mr{Heis}}}_0$ of the two copies of Heisenberg algebra contains a two-parameter family of Virasoro-highest weight vectors (primary fields) $\chi_{n,\bar{n}}$ with $n,\bar{n}=0,1,2,\ldots$, with conformal weights $(h=n^2,\bar{h}=\bar{n}^2)$. 
\item  A related point to the above is that the Virasoro Verma module $\mathbb{V}^\mr{Vir}_h$ for $h=\frac{N^2}{4}$ is reducible and contains null vectors at levels $\frac{(N+2k)^2}{4}-\frac{N^2}{4}$ with $k=1,2,\ldots$. Vanishing descendant (\ref{l25 vanishing descendant}) above gives an example of a null vector at level $3$ in $\mathbb{V}_{h=1}^\mr{Vir}$; here $N=2$, $k=1$. Also, $\mathbb{V}^\mr{Vir}_h$ fits (depending on parity of $N$) into one of the two sequences of maps between Virasoro Verma modules
\begin{equation}
\begin{gathered}
\mathbb{V}^\mr{Vir}_0\leftarrow \mathbb{V}^\mr{Vir}_1 \leftarrow \mathbb{V}^\mr{Vir}_{2^2}\leftarrow \mathbb{V}^\mr{Vir}_{3^2}\leftarrow \cdots\\
\mathbb{V}^\mr{Vir}_{\left(\frac{1}{2}\right)^2}\leftarrow \mathbb{V}^\mr{Vir}_{\left(\frac{3}{2}\right)^2} \leftarrow \mathbb{V}^\mr{Vir}_{\left(\frac{5}{2}\right)^2}\leftarrow \mathbb{V}^\mr{Vir}_{\left(\frac{7}{2}\right)^2}\leftarrow \cdots
\end{gathered}
\end{equation}
For each map here, the image of the highest vector or a null vector is a null vector in the target module (and each null vector arises that way -- ultimately comes from the highest vector of one of the modules to the right in the sequence).
Also, one has that the irreducible Virasoro module 
\begin{equation}
M_{\frac{N^2}{4}}=\mathbb{V}^\mr{Vir}_{\frac{N^2}{4}}/\mathbb{V}^\mr{Vir}_{\frac{(N+2)^2}{4}}
\end{equation} 
is the quotient of the Verma module by the submodule generated by the first null vector (all subsequent null vectors are already in that submodule).
\end{enumerate}

\end{remark}

%\marginpar{Comment more on the classification of Virasoro-primary fields in scalar field CFT}

\marginpar{Lecture 26,\\ 10/28/2022}

\subsection{More on vertex operators}\label{sss more on vertex operators}

Here are some other interesting properties of vertex operators in free scalar theory.
\begin{itemize}
\item The 2-point correlator of vertex operators is
\begin{equation}\label{l26 <VV>}
\langle V_\alpha(w) V_\beta(z) \rangle =\left\{
\begin{array}{cl}
|w-z|^{-2\alpha^2} & \mr{if}\; \beta=-\alpha,\\
0& \mr{otherwise}
\end{array}
\right.
\end{equation}
More generally, the $n$-point correlator of vertex operators is
\begin{equation}\label{l26 <V...V>}
\langle V_{\alpha_1}(z_1)\cdots V_{\alpha_n}(z_n) \rangle =\left\{
\begin{array}{cl}\displaystyle
\prod_{1\leq j<k\leq n} |z_j-z_k|^{2\alpha_j\alpha_k} &\mr{if}\; \alpha_1+\cdots+\alpha_n=0,\\
0& \mr{otherwise}
\end{array}
\right.
\end{equation}
\item Vertex operators satisfy the OPE
\begin{equation}\label{l26 VV OPE}
V_\alpha(w)V_\beta(z)\sim |w-z|^{2\alpha\beta}V_{\alpha+\beta}(z)+(\mr{less\;singular\;terms}).
\end{equation}
\item One has the OPE
\begin{equation}\label{l26 dd phi V OPE}
i\dd\phi(w) V_\alpha(z)\sim \frac{\alpha}{w-z}V_\alpha(z)+\mr{reg.}
\end{equation}
\end{itemize}
All these properties follow from the explicit formula for the vertex operator (\ref{l24 Vhat}) and Wick's lemma. For instance, let us prove (\ref{l26 <VV>}). We apply Wick's lemma to the product of two vertex operators (as operators on $\HH$:
%For simplicity, assume $|w|>|z|$. We have
%\begin{multline}
%\wh{V}_\alpha(w)\wh{V}_\beta(z)=e^{
%i\alpha \wh\phi_0-\alpha \sum_{n<0}\frac{1}{n} (\wh{a}_n w^{-n}+\wh{\ol{a}}_n \bar{w}^{-n})}\circ\\ \circ
%\underbrace{e^{\alpha \wh\pi_0\log(w\bar{w})-\alpha\sum_{n>0}(\wh{a}_nw^{-n}+\wh{\ol{a}}_n\bar{w}^{-n})}}_{e^A}\circ
%\underbrace{e^{i\beta \wh\phi_0-\alpha \sum_{m<0}\frac{1}{m} (\wh{a}_m z^{-m}+\wh{\ol{a}}_m \bar{z}^{-m})}}_{e^B}\circ\\\circ
%e^{\beta \wh\pi_0\log(z\bar{z})-\alpha\sum_{m>0}(\wh{a}_nz^{-m}+\wh{\ol{a}}_m \bar{z}^{-m})}
%\end{multline}
For simplicity, assume $|w|>|z|$.
We have 
\begin{multline}\label{l26 VV product computation}
\wh{V}_\alpha(w)\wh{V}_\beta(z)=\sum_{n,m\geq 0}\frac{1}{n!m!}(i\alpha)^n(i\beta)^m :\wh\phi(w)^n:\; :\wh\phi(z)^m:=\\
\underset{\mr{Wick}}{=} \sum_{k\geq 0} \sum_{n,m\geq k}
\frac{1}{n!m!} \underbrace{
\left( \begin{array}{c}n\\ k\end{array}\right)
\left( \begin{array}{c}m\\ k\end{array}\right) k!}_{\mbox{\#($k$-fold Wick contractions)}} (i\alpha)^n (i\beta)^m (-2\log |w-z|)^k :\wh\phi(w)^{n-k}\wh\phi(z)^{m-k}:\\
=\sum_{k\geq 0} \sum_{n,m\geq k} \frac{1}{(n-k)!(m-k!)k!} (i\alpha)^n (i\beta)^m (-2\log |w-z|)^k :\wh\phi(w)^{n-k}\wh\phi(z)^{m-k}:\\
\underset{n'=n-k,\,m'=m-k}{=}
\sum_{k\geq 0} \frac{(2\alpha\beta)^k}{k!}(\log|w-z|)^k \sum_{n',m'\geq 0}  \frac{(i\alpha)^{n'}(i\beta)^{m'}}{n'!m'!}
:\wh\phi(w)^{n'}\wh\phi(z)^{m'}:\\
=e^{2\alpha\beta \log |w-z|} :e^{i\alpha \wh\phi(w)}e^{i\beta \wh\phi(z)}: = \boxed{ |w-z|^{2\alpha\beta} :\wh{V}_\alpha(w)\wh{V}_\beta(z):}
\end{multline}
The normally ordered product of vertex operators on the right can be written as $e^{i(\alpha+\beta)\wh\phi_0}(1+\cdots)$ where $\cdots$ are normally ordered terms with zero VEV (vacuum expectation value). The operator $e^{i(\alpha+\beta)\wh\phi_0}$ shift the vacuum $|\vac\rangle$ to a pseudovacuum $|\pi_0=\alpha+\beta\rangle$, so it has expectation value zero unless $\alpha+\beta=0$, and in the latter case the VEV is $1$. 
Thus,
\begin{equation}
\langle\vac|\wh{V}_\alpha(w) \wh{V}_\beta(z)  |\vac\rangle=
|w-z|^{2\alpha \beta}\cdot\left\{
\begin{array}{cc}
1& \mr{if}\; \alpha+\beta=0,\\
0& \mr{otherwise}
\end{array}
\right.
\end{equation}
This finishes the proof of (\ref{l26 <VV>}).

Note that the computation (\ref{l26 VV product computation}) also implies the OPE (\ref{l26 VV OPE}):
\begin{multline}
\mc{R}\wh{V}_\alpha(w) \wh{V}_\beta(z) = |w-z|^{2\alpha\beta}:\underbrace{\wh{V}_\alpha(w)}_{\mr{expand\;around\;}z}\wh{V}_\beta(z) : =\\
=|w-z|^{2\alpha\beta} \sum_{k,l\geq 0} \frac{(w-z)^k(\bar{w}-\bar{z})^l}{k!l!} :\dd^k\bar\dd^l \wh{V}_\alpha(z)\,\wh{V}_\beta(z):\\=|w-z|^{2\alpha\beta} \wh{V}_{\alpha+\beta}(z)+ O(|w-z|^{2\alpha\beta+1}),
\end{multline}
where we used the property $:\wh{V}_\alpha(z)\wh{V}_\alpha(z):=\wh{V}_{\alpha+\beta}(z)$, obvious from the definition of the vertex operator (\ref{l24 Vhat}).

The correlator (\ref{l26 <V...V>}) is also obtained from Wick's lemma, see \cite[section 9.1.1]{DMS}.

The OPE (\ref{l26 dd phi V OPE}) is obtained by a computation similar to (\ref{l25 TV OPE computation}) (actually simpler, as there are only single Wick contractions).

\section{Conformal Ward identity (via contour integration trick)}
\label{ss: Ward identity via contour integration}

In any CFT on $\CC$ one the following.
\begin{thm}[Conformal Ward identity]
Fix a collection of  fields $\Phi_1,\ldots,\Phi_n\in V$, a collection of distinct points $z_1,\ldots,z_n\in \CC$, a conformal vector field $v=u(w)\dd_w+\ol{u(w)}\dd_{\bar{w}}$ with $u(w)\dd_w$ a meromophic vector field  on $\CP^1$
%and an antimeromorphic vector field $\bar{u}(\bar{w})\dd_{\bar{w}}$ 
%where $u,\bar{u}$ are allowed to (but do not have to) have
with poles allowed only at the points $z_1,\ldots,z_k$ (in particular we are assuming that $w=\infty$ is a regular point of $u\dd$). Then one has
\begin{equation}\label{l26 Ward identity}
\sum_{k=1}^n \langle \Phi_1(z_1)\cdots \rho^{(z_k)}(v)\circ \Phi_k(z_k) \cdots \Phi_n(z_n)\rangle =0
\end{equation}
where $\rho^{(z_k)}(v)\circ \Phi_k(z_k)$ is the action of the vector field $v$ on the field $\Phi_k$ defined via (\ref{l25 rho}).
\end{thm}
We denote the l.h.s. of (\ref{l26 Ward identity}) by $\delta_v \langle \Phi_1(z_1)\cdots \Phi_n(z_n) \rangle$ -- the action of the vector fields on the correlator (via acting on individual fields). Thus, the Ward identity says that the action of a conformal vector field on a correlator vanishes.

Note that (by complexification) we can treat $u(w)\dd_w$ and $\ol{u(w)}\dd_{\bar{w}}$ in (\ref{l26 Ward identity}) as independent meromorphic and antimeromorphic vector fields (not complex conjugate to one another).

\begin{proof}
Consider the action of a meromorphic vector field $u(w)\dd_w$ on a correlator. Let $\Gamma=C_{0,R}$ be a circle centered at the origin of a large radius $R$ (in particular, large enough that it encloses all $z_i$'s). Then we have
\begin{multline}\label{l26 Ward computation}
 -\frac{1}{2\pi i}\oint_\Gamma dw\, u(w) \Big\langle T(w) \Phi_1(z_1)\cdots \Phi_n(z_n) \Big\rangle = \\
 \underset{\mr{deformation\;of\;contour}}{=}
\sum_{k=1}^n -\frac{1}{2\pi i}\oint_{\gamma_k} dw\, u(w) \Big\langle T(w) \Phi_1(z_1)\cdots \Phi_n(z_n) \Big\rangle  \\
\underset{(\ref{l25 rho})}{=}
\sum_{k=1}^n \langle \Phi_1(z_1)\cdots \rho^{(z_k)}(u\dd)\circ \Phi_k(z_k) \cdots \Phi_n(z_n)\rangle =
\delta_{u\dd} \langle \Phi_1(z_1)\cdots \Phi_n(z_n) \rangle,
\end{multline}
Here $\gamma_k=C_{z_k,r_k}$ is a circle around $z_k$ of radius $r_k$ small enough that $\gamma_k$ does not enclose any $z_i$ with $i\neq k$. We used the fact that the correlator with $T(w)$ is meromorphic in $w$, with possible poles at $w=z_1,\ldots,z_n$, to deform the integration contour $\Gamma$ to $\gamma_1\cup\cdots\cup\gamma_n$.

%\textcolor{red}{PICTURE}
\begin{figure}[H]
$$\vcenter{\hbox{ \includegraphics[scale=0.8]{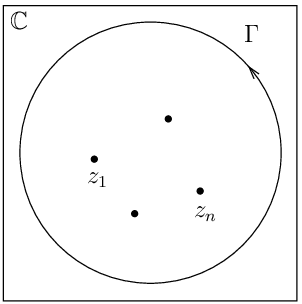} }} 
\longrightarrow
\vcenter{\hbox{ \includegraphics[scale=0.8]{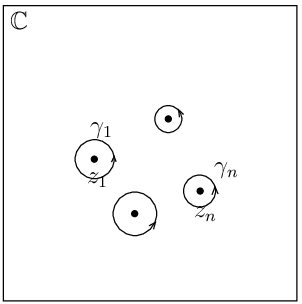} }} 
$$
\caption{Deformation of the integration contour $\Gamma$ (large circle) into a collection of small circles $\gamma_1,\ldots,\gamma_n$ around punctures $z_1,\ldots,z_n$.}
\end{figure}

It remains to show that the l.h.s. of (\ref{l26 Ward computation}) vanishes. For that, let us use  Lemma \ref{l23 lemma: conjugation of a correlator}:
\begin{multline}\label{l26 Ward computation 2}
\ol{-\frac{1}{2\pi i}\oint_\Gamma dw\, u(w) \Big\langle T(w) \Phi_1(z_1)\cdots \Phi_n(z_n) \Big\rangle}
=\\
=\frac{1}{2\pi i}\oint_{\Gamma\ni w} d\bar{w} \ol{u(w)} \bar{w}^{-4} \langle T(1/\bar{w}) \Phi_1^*(1/\bar{z}_1)\cdots \Phi_n^*(1/\bar{z}_n) \rangle\cdot \prod_{i=1}^n \bar{z}_i^{-2h_i}z_i^{-2\bar{h}_i}\\
\underset{y=1/\bar{w}}{=}
-\frac{1}{2\pi i}\oint_{\Gamma'\ni y} dy\, u_y(y)  
\langle T(y) \Phi_1^*(1/\bar{z}_1)\cdots \Phi_n^*(1/\bar{z}_n) \rangle\cdot \prod_{i=1}^n \bar{z}_i^{-2h_i}z_i^{-2\bar{h}_i}
\end{multline}
where $u_y(y)=\ol{u(w)}/\bar{w}^2$ is regular at $y= 0$, since the vector field $u(w)\dd_w$ was required to be regular at $w=\infty$; $\Gamma'$ is a circle around zero of small radius $1/R$. 
%Taking the limit $R\ra \infty$ (i.e., shrinking the circle $\Gamma'$ -- the integration contour for $y$ -- onto zero)
The integrand in the r.h.s. of (\ref{l26 Ward computation 2}) is a  meromorphic function in $y$ and $\Gamma'$ does not enclose any poles (in particular $y=0$ is a regular point), hence (\ref{l26 Ward computation 2}) vanishes. This proves that the r.h.s. of (\ref{l26 Ward computation}) is zero.

The case of the action of an antimeromorphic vector field on a correlator is similar.
\end{proof}

Informally, the argument is: take the integral in the l.h.s. (\ref{l26 Ward computation}) over a contour around $w=\infty$ in $\CP^1$. One the one hand the integral vanishes, since integrand is holomorphic around $w=\infty$. On the other hand, the contour can be deformed into a union of small circles around field insertions $z_i$, which yields $\delta_v$ of the correlator.

\begin{figure}[H]
$$\vcenter{\hbox{ \includegraphics[scale=0.8]{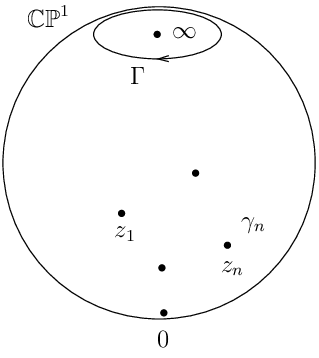} }} 
\longrightarrow
\vcenter{\hbox{ \includegraphics[scale=0.8]{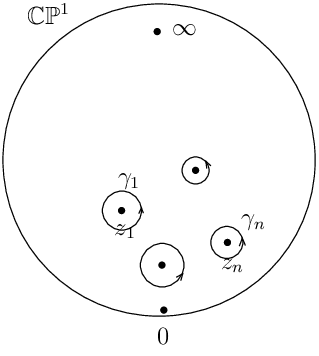} }} 
$$
\caption{Deformation of the integration contour on $\CP^1$.}
\end{figure}

\begin{example}
Let $u(w)\dd_w=\frac{-\dd_w}{w-z_0}$ -- a meromorphic vector field with a simple pole at $z_0$. Assume that $\Phi_1,\ldots,\Phi_n$ are \emph{primary} fields with conformal weights $(h_i,\bar{h}_i)$. Applying (\ref{l26 Ward identity}) to the correlator $\langle \mathbb{1}(z_0)\Phi_1(z_1)\cdots \Phi_n(z_n)\rangle$,\footnote{We inserted $\mathbb{1}(z_0)$, which does not affect the correlator, since we required that the vector field only has poles at the points where fields are inserted.} we obtain
\begin{multline}
0=\langle \rho^{z_0}\left(\frac{-\dd_w}{w-z_0}\right)\circ\mathbb{1}(z_0)\; \Phi_1(z_1)\cdots \Phi_n(z_n) \rangle+\\+
\sum_{k=1}^n \langle \cancel{\mathbb{1}(z_0)} \Phi_1(z_1)\cdots
\rho^{z_k}\Big(\underbrace{\frac{-\dd_w}{w-z_0}}_{\mr{expand\;at\;}z_k}\Big)\circ \Phi_k(z_k)\cdots \Phi_n(z_n)
 \rangle\\
 \underset{(\ref{l25 L_n loc})}{=}
\langle (L_{-2}\mathbb{1})(z_0) \Phi_1(z_1)\cdots \Phi_n(z_n) \rangle+\\
+
\sum_{k=1}^n \langle
\Phi_1(z_1)\cdots  \rho\left(-\frac{1}{z_k-z_0}\dd_w+\frac{w-z_k}{(z_k-z_0)^2} \dd_w -\frac{(w-z_k)^2}{(z_k-z_0)^3} \dd_w+\cdots\right)\circ \Phi_k(z_k)\cdots \Phi_n(z_n)
 \rangle\\
 =\langle \underbrace{(L_{-2}\mathbb{1})(z_0)}_{T(z_0)} \Phi_1(z_1)\cdots \Phi_n(z_n) \rangle+\\
 +
 \sum_{k=1}^n \langle \Phi_1(z_1)\cdots \Big(\frac{1}{z_k-z_0}L_{-1}-\frac{1}{(z_k-z_0)^2}L_0+\underbrace{\cancel{\frac{1}{(z_k-z_0)^3}L_{1}-\cdots}}_{\mr{since\;}\Phi_k\;\mr{is\;primary}}\Big)\circ\Phi_k(z_k)\cdots \Phi_n(z_n) \rangle\\
=\langle T(z_0) \Phi_1(z_1)\cdots \Phi_n(z_n) \rangle
+\sum_{k=1}^n \langle \Phi_1(z_1)
\cdots \left(\frac{1}{z_k-z_0}\,\frac{\dd}{\dd z_k}-\frac{h_k}{(z_k-z_0)^2}\right)\Phi_k(z_k) \cdots \Phi_n(z_n)
\rangle .
\end{multline}
Or, written another way:
\begin{equation}
\langle T(z_0)\Phi_1(z_1)\cdots \Phi_n(z_n) \rangle =
\left(\sum_{k=1}^n \frac{h_k}{(z_k-z_0)^2}-\frac{1}{z_k-z_0}\,\frac{\dd}{\dd z_k}\right) \circ \langle \Phi_1(z_1)\cdots \Phi_n(z_n) \rangle.
\end{equation}
Thus, the correlator of the stress-energy  with a collection of primary fields is expressed as a certain differential operator acting on the correlator of just the primary fields.
\end{example}

\begin{example}\label{l26 ex: corr of a descendant}
If the correlator of primary fields $\Phi_1,\ldots,\Phi_n$ is known then any correlator of their descendants can be recovered as a certain differential operator acting on $\langle \Phi_1\cdots \Phi_n \rangle$. Such an expression is obtained from Ward identity by repeatedly applying meromorphic vector fields of the form $-(w-z_k)^{-r+1}\dd_w$ to the correlator of the primary fields. 

For instance applying the vector field $u\dd=-(w-z_1)^{-r+1}\dd_w$ (for some $r\geq 1$) to $\langle\Phi_1(z_1)\cdots \Phi_n(z_n)\rangle$ we find
\begin{multline}\label{l26 corr of descendant}
0=\delta_{u\dd}\langle\Phi_1(z_1)\cdots \Phi_n(z_n)\rangle=\\
=
\langle (L_{-r}\Phi_1)(z_1)\, \Phi_2(z_2)\cdots \Phi_n(z_n) \rangle+\\
+\underbrace{\left(\sum_{k=2}^n (z_k-z_1)^{-r+1}\dd_{z_k}-(r-1)(z_k-z_1)^{-r}h_k\right)}_{-\mc{D}}\circ \langle \Phi_1(z_1)\Phi_2(z_2)\cdots \Phi_n(z_n) \rangle.
\end{multline}
Thus, one has
\begin{equation}
\langle (L_{-r}\Phi_1)(z_1)\, \Phi_2(z_2)\cdots \Phi_n(z_n) \rangle=\mc{D}\langle \Phi_1(z_1)\Phi_2(z_2)\cdots \Phi_n(z_n) \rangle
\end{equation}
with $\mc{D}$ the differential operator appearing in (\ref{l26 corr of descendant}). Here we were assuming that $\Phi_1,\ldots,\Phi_n$ are primary.
% (we did not use the primary property of $\Phi_1$).
\end{example}

\subsection{
%Correlators of primary fields: constraints from global conformal symmetry
Constraints on correlators from global conformal symmetry
}\label{sss constraints on corr from global conf sym}
Let us explore the consequences of the Ward identity (\ref{l26 Ward identity}) with $v$ a conformal vector field on $\CP^1$ without singularities.

For $\Phi_1,\ldots,\Phi_n\in V$ primary and $v=u\dd+\bar{u}\bar\dd$ a conformal vector field without singularities, the Ward identity reads
\begin{multline}\label{l26 Ward for primary fields}
0=\delta_v\langle \Phi_1(z_1)\cdots \Phi_n(z_n) \rangle=\\
=
\sum_{k=1}^n \langle \Phi_1(z_1)\cdots 
\Big(-u(z_k)\dd_{z_k}-\ol{u(z_k)} \dd_{\bar{z}_k} -h_k \dd u(z_k) -\bar{h}_k \ol{\dd u(z_k)}\Big)\Phi_k(z_k)
\cdots\Phi_n(z_n) \rangle 
\end{multline}
The ``finite'' (or ``integrated'') version is then as follows: for $z\mapsto w(z)$ 
a holomorphic map $\CP^1\ra \CP^1$ (i.e., a M\"obius transformation) one has
\begin{equation}\label{l26 Ward idenity for primary fields, finite version}
\langle \Phi_1(w(z_1))\cdots  \Phi_n(w(z_n)) \rangle=
\prod_{i=1}^n\left(\frac{\dd w}{\dd z}(z_i)\right)^{-h_i}
\left(\frac{\dd \bar{w}}{\dd \bar{z}}(z_i)\right)^{-\bar{h}_i}\cdot
\langle \Phi_1(z_1)\cdots  \Phi_n(z_n) \rangle
\end{equation}
Put another way, one has an equality 
\begin{multline}
\langle \Phi_1(z_1)(dz_1)^{h_1}(d\bar{z}_1)^{\bar{h}_1}\cdots
\Phi_n(z_n)(dz_n)^{h_n}(d\bar{z}_n)^{\bar{h}_n} \rangle =\\
=
\langle \Phi_1(w_1)(dw_1)^{h_1}(d\bar{w}_1)^{\bar{h}_1}\cdots
\Phi_n(w_n)(dw_n)^{h_n}(d\bar{w}_n)^{\bar{h}_n} \rangle
\end{multline}
Using the notation (\ref{l25 Phi underline}), 
\begin{equation}\label{l26 Phi underline}
\mathbf{\Phi}(z)\colon= \Phi(z)(dz)^h(d\bar{z})^{\bar{h}} \quad \in V\otimes K^{h,\bar{h}}_z,
\end{equation}
the $n$-point correlator of primary fields is a section of a certain line bundle on the open configuration space of points on $\CP^1$, invariant under the
M\"obius group (the latter being the statement of the Ward identity):
\begin{equation}\label{l26 corr of primary fields as a section of a line bundle}
\langle \mathbf{\Phi}_1(z_1)\cdots \mathbf{\Phi}_n(z_n) \rangle \in 
\Gamma\left(C_n(\CP^1),\bigotimes_{i=1}^n \pi_i^* K^{h_i,\bar{h}_i}\right)^{PSL_2(\CC)}
\end{equation}
where $\pi_i\colon C_n(\CP^1)\ra \CP^1$ is the map selecting the $i$-th point of the $n$-tuple; $K^{h_i,\bar{h}_i}$ is the line bundle (\ref{l25 line bun}) on $\CP^1$.

\begin{remark}\label{l26 rem: Ward identity for translations, rotations, scalings}
If the vector field $v$ is at most linear in coordinates, then (\ref{l26 Ward for primary fields}) holds without assuming that fields $\Phi_1,\ldots,\Phi_n$ are primary. At the level of ``finite'' conformal maps, the identity (\ref{l26 Ward idenity for primary fields, finite version}) holds for $z\mapsto w(z)$ translations, rotations and dilations, without assuming that the fields are primary.
\end{remark}

\begin{lemma}\label{l26 lemma OPE exponents}
%A corollary of the previous remark \ref{l26 rem: Ward identity for translations, rotations, scalings} is the following: 
If the OPE of fields $\Phi_1,\Phi_2\in V$ contains the term
\begin{equation}\label{l26 OPE term}
\frac{C}{(w-z)^\alpha (\bar{w}-\bar{z})^{\bar\alpha}} \Phi(z)
\end{equation}
with some field $\Phi\in V$ and $C$ a constant, then the exponents  in (\ref{l26 OPE term}) satisfy 
\begin{equation}\label{l26 OPE exponents}
h(\Phi_1)+h(\Phi_2)=\alpha+h(\Phi), \quad \bar{h}(\Phi_1)+\bar{h}(\Phi_2)=\bar\alpha+\bar{h}(\Phi),
\end{equation}
where $h,\bar{h}$ are the conformal weights of the fields involved. 
\end{lemma}

\begin{proof}
This is a consequence of (\ref{l26 Ward idenity for primary fields, finite version}) and Remark \ref{l26 rem: Ward identity for translations, rotations, scalings}: one considers the correlator $\langle \Phi_1(w) \Phi_2(z) \Phi_3(z_3)\cdots \Phi_n(z_n) \rangle$, with $\Phi_3,\ldots,\Phi_n\in V$ arbitrary test fields,  and acts on it with rotation and dilation around $z$. For simplicity, set $z=0$ and consider the map $z\mapsto \lambda z$ with $\lambda\in \CC^*$. Then we have, in the asymptotics $w\ra 0$,
{\small
\begin{equation}
\hspace{-2cm}
\begin{CD}
\langle \Phi_1(\lambda w)\Phi_2(0)\Phi_3(\lambda z_3)\cdots \Phi_n(\lambda z_n)\rangle @=
\prod_{i=1}^n \lambda^{-h_i} \bar\lambda^{-\bar{h}_i} 
\langle \Phi_1(w) \Phi_2(0)\Phi_3(z_3)\cdots \Phi_n(z_n) \rangle\\
\mr{OPE}@| \mr{OPE}@| \\
\frac{C}{\lambda^\alpha\bar\lambda^{\bar\alpha}w^\alpha \bar{w}^{\bar{\alpha}}}\langle \Phi(0)\Phi_3(\lambda z_3)\cdots \Phi_n(\lambda z_n) \rangle+\cdots
@.
\frac{C}{w^\alpha \bar{w}^{\bar{\alpha}}} \prod_{i=1}^n \lambda^{-h_i} \bar\lambda^{-\bar{h}_i} 
\langle \Phi(0)\Phi_3(z_3)\cdots \Phi_n(z_n) \rangle +\cdots\\
(\ref{l26 Ward idenity for primary fields, finite version})@| @|\\
\frac{C\lambda^{-h(\Phi)}\bar{\lambda}^{-\bar{h}(\Phi)}}{\lambda^\alpha\bar\lambda^{\bar\alpha}w^\alpha \bar{w}^{\bar{\alpha}}}
\prod_{i=3}^n \lambda^{-h_i}\bar\lambda^{-\bar{h}_i}
\langle \Phi(0)\Phi_3(z_3)\cdots \Phi_n(z_n) \rangle+\cdots
@.
\frac{C}{w^\alpha \bar{w}^{\bar{\alpha}}} \prod_{i=1}^n \lambda^{-h_i} \bar\lambda^{-\bar{h}_i} 
\langle \Phi(0)\Phi_3(z_3)\cdots \Phi_n(z_n) \rangle +\cdots
\end{CD}
\end{equation}
}
%for any $\lambda\in \CC^*$; 
Here $\cdots$ stands for the other terms in the OPE. 
%Correlators in the bottom row are again related by (\ref{l26 Ward idenity for primary fields, finite version}) for the map $z\mapsto \lambda z$.
Equality in the last row implies the claimed relation on the OPE exponents (\ref{l26 OPE exponents}).
\end{proof}

\subsubsection{One-point correlators.}
\begin{lemma}\label{l26 lemma 1-point corr}
Let $\Phi\in V$ be a field (not necessarily primary) of conformal weight $(h,\bar{h})$. Then
\begin{equation}
\langle \Phi(z) \rangle =
\left\{
\begin{array}{cc}
C_\Phi& \mr{if}\; h=\bar{h}=0,\\
0&\mr{otherwise} 
\end{array}
\right.
\end{equation}
where $C_\Phi$ is a constant function. 
(the value of the constant depends on $\Phi$).
\end{lemma}
\begin{proof}
Using the Ward identity with $v$ a constant vector field $a\dd_w+\bar{a} \dd_{\bar{w}}$ (with arbitrary coefficients $a,\bar{a}\in\CC$), we find that the one-point correlator satisfies $(a\dd_z+\bar{a} \dd_{\bar{z}})\langle \Phi(z) \rangle =0$, i.e., the correlator is a constant function. Applying the vector field $v=b(w-z)\dd_w+\bar{b}(\bar{w}-\bar{z})\dd_{\bar{w}}$ to the correlator, we see that it satisfies 
\begin{equation}
(b h + \bar{b} \bar{h})\langle \Phi(z) \rangle=0
\end{equation}
for any $b,\bar{b}\in\CC$. Thus, the one-point correlator must vanish unless $h=\bar{h}=0$.
\end{proof}

\subsubsection{Two-point correlators.}
\begin{lemma}\label{l26 lemma 2-point}
Let $\Phi_1,\Phi_2\in V$ be two fields of conformal weights $(h_i,\bar{h}_i)$, $i=1,2$. 
\begin{enumerate}[(a)]
\item \label{l26 lemma 2-point (a)} One has
\begin{equation}\label{l26 2-point function}
\langle \Phi_1(z_1) \Phi_2(z_2) \rangle = C_{\Phi_1\Phi_2} \frac{1}{(z_1-z_2)^{h_1+h_2}(\bar{z}_1-\bar{z}_2)^{\bar{h}_1+\bar{h}_2}}
\end{equation}
with $C_{\Phi_1\Phi_2}$ some constant depending on $\Phi_1,\Phi_2$.
\item \label{l26 lemma 2-point (b)} If $\Phi_1,\Phi_2$ are primary, then the constant $C_{\Phi_{1},\Phi_2}$ in (\ref{l26 2-point function}) vanishes unless one has
\begin{equation}\label{l26 lemma condition on conf weights}
h_1=h_2,\;\; \bar{h}_1=\bar{h}_2.
\end{equation}
\item \label{l26 lemma 2-point (c)} For $\Phi_1,\Phi_2$ 
two fields satisfying condition (\ref{l26 lemma condition on conf weights}) on conformal weights,
%primary, 
the constant $C_{\Phi_1\Phi_2}$ in (\ref{l26 2-point function}) is related to the hermitian inner product on $V$ (cf. Axiom (\ref{l23 Axiom inner product})) by
\begin{equation}\label{l26 lemma 2-point (c) eq}
C_{\Phi_1\Phi_2}=\langle \Phi_1^*,\Phi_2  \rangle_V.
\end{equation}
\end{enumerate}
\end{lemma}

\begin{proof} Part (\ref{l26 lemma 2-point (a)}) follows from (\ref{l26 Ward idenity for primary fields, finite version}) for translations, rotations and dilations (we exploit Remark \ref{l26 rem: Ward identity for translations, rotations, scalings}).

For (\ref{l26 lemma 2-point (b)}), let us fix the two points at $z_1=0$ and $z_2=1$ and act on the correlator with the vector field $u\dd_w=w(1-w)\dd_w$ -- a holomorphic vector field on the entire $\CP^1$. 
%Note that near $z_1=0$ the vector field is $(w-w^2)\dd_w$ and near $z_2=1$ it is $(-w'-(w')^2)\dd_{w'}$, with $w'=w-1$. Thus, the Ward identity (\ref)
The Ward identity (\ref{l26 Ward for primary fields}) in this case reads
\begin{equation}
0=\langle -h_1\Phi_1(z_1)\Phi_2(z_2) \rangle + \langle \Phi_1(z_1)h_2 \Phi_2(z_2) \rangle = (h_2-h_1)\langle\Phi_1(z_1)\Phi_2(z_2) \rangle.
\end{equation}
Thus unless $h_1=h_2$, the 2-point correlator vanishes. Likewise, acting with the vector field $\bar{w}(1-\bar{w})\dd_{\bar{w}}$, we find that unless $\bar{h}_1=\bar{h}_2$, the correlator also has to vanish.

For (\ref{l26 lemma 2-point (c)}), we calculate the r.h.s. of (\ref{l26 lemma 2-point (c) eq}) exploiting  the state-field correspondence:
\begin{multline}
\langle \Phi_1^*,\Phi_2 \rangle_V= \lim_{w,z\ra 0}\Big\langle 
\wh{\Phi}_1^*(w)|\vac\rangle, \wh{\Phi}_2(z)|\vac\rangle
\Big\rangle_\HH=\lim_{w,z\ra 0} \langle\vac| 
\wh\Phi_1^*(w)^+ \wh\Phi_2(z)
|\vac\rangle\\
\underset{(\ref{l23 Phi^+})}{=}
\lim_{w,z\ra 0} \bar{w}^{-2h_1}w^{-2\bar{h}_1}\langle \vac| \wh\Phi_1(1/\bar{w}) \wh\Phi_2(z) |\vac\rangle\\
\underset{(\ref{l26 2-point function})}{=}
C_{\Phi_1\Phi_2} \lim_{w,z\ra 0} \bar{w}^{-2h_1}w^{-2\bar{h}_1}
\left(1/\bar{w}-z\right)^{-h_1-h_2}(1/w-\bar{z})^{-\bar{h}_1-\bar{h}_2}\\
=C_{\Phi_1\Phi_2}.
\end{multline}
Here in the last step we used the condition (\ref{l26 lemma condition on conf weights}).
\end{proof}

\begin{example}
In scalar field theory, the correlators 
\begin{equation}
\langle \dd\phi(w) \dd\phi(z) \rangle = -\frac{1}{(w-z)^2}, \qquad 
\langle V_\alpha(w) V_\beta(z)\rangle = \left\{
\begin{array}{cl}
\frac{1}{|w-z|^{2\alpha^2}}, & \alpha=\beta \\
0,& \alpha\neq \beta
\end{array}\right.
\end{equation}
(cf. (\ref{l20 <dd phi  dd phi>}), (\ref{l26 <VV>})) are examples of two-point correlators of primary fields (of weight $(1,0)$ in the first case and of weight $(\frac{\alpha^2}{2},\frac{\alpha^2}{2})$ in the second case). They are clearly consistent with %Lemma \ref{l26 lemma 2-point}.
the general ansatz (\ref{l26 2-point function}).
\end{example}

\begin{example}
The $TT$ OPE (\ref{l23 TT}) and the ansatz (\ref{l26 2-point function}) imply that the two-point correlator of the stress-energy tensor is
\begin{equation}
\langle T(w) T(z)\rangle=\frac{c/2}{(w-z)^4}.
\end{equation}
With (\ref{l26 lemma 2-point (c) eq}) this implies 
\begin{equation}
\langle T,T \rangle_V = \frac{c}{2}.
\end{equation}
Recall that for a \emph{unitary} CFT the inner product on $V$ is assumed to be positive definite. This means that the central charge $c$ must be a positive number.\footnote{
%Positivity of the inner product on $V$ is a part of the unitarity assumption for a CFT. 
There are interesting examples of non-unitary CFTs (e.g. the so-called ghost system or $bc$ system, see Section \ref{s bc system}) where the central charge can be negative. For instance, in the $bc$ system one has $c=-26$.
}
\end{example}

\marginpar{Lecture 27,\\ 10/31/2022}
\subsubsection{Three-point correlators of primary fields.}
\begin{lemma}\label{l27 lemma: 3-point fun}
For any three primary fields $\Phi_1,\Phi_2,\Phi_3\in V$, with $\Phi_i$ of conformal weights $(h_i,\bar{h}_i)$, one has
\begin{equation}\label{l27 3-point fun}
\langle \Phi_1(z_1)\Phi_2(z_2)\Phi_3(z_3) \rangle = C_{\Phi_1\Phi_2\Phi_3} \prod_{1\leq i<j \leq 3}\frac{1}{(z_i-z_j)^{2\alpha_{ij}}(\bar{z}_i-\bar{z}_j)^{2\bar\alpha_{ij}}},
\end{equation}
where $C_{\Phi_1\Phi_2\Phi_3}$ is a constant (depending on the fields but not on the points $z_1,z_2,z_3$) and the exponents are expressed in terms of conformal weights of the fields:
\begin{equation}\label{l27 alpha_ij}
\begin{gathered}
\alpha_{12}=\frac12 (h_1+h_2-h_3),\;\alpha_{13}=\frac12 (h_1+h_3-h_2),\; \alpha_{23}=\frac12 (h_2+h_3-h_1),\\
\bar{\alpha}_{12}=\frac12 (\bar{h}_1+\bar{h}_2-\bar{h}_3),\;
\bar{\alpha}_{13}=\frac12 (\bar{h}_1+\bar{h}_3-\bar{h}_2),\;
\bar{\alpha}_{23}=\frac12 (\bar{h}_2+\bar{h}_3-\bar{h}_1).
\end{gathered}
\end{equation}
\end{lemma}

\begin{proof}[Proof \#1 (idea)]
Take the unique M\"obius transformation $f\colon\CP^1\ra \CP^1 $ that maps points $z_1,z_2,z_3$ to $0,1,2$. Then the Ward identity 
(\ref{l26 Ward idenity for primary fields, finite version}) allows one to write the 3-point correlator as
\begin{equation}
\langle \Phi_1(z_1)\Phi_2(z_2)\Phi_3(z_3) \rangle =\prod_{i=1}^3 (\dd f(z_i))^{h_i} (\ol{\dd f(z_i)})^{\bar{h}_i}\cdot \underbrace{\langle \Phi_1(0) \Phi_1(1)\Phi_3(2) \rangle}_{\til{C}}
\end{equation}
with $\til{C}$ some constant. Computing explicitly the derivatives in the r.h.s., one obtains (\ref{l27 3-point fun}).
\end{proof}

Let us introduce the notation
\begin{equation}\label{l27 mu}
\mu=\frac{dz_1\wedge dz_2}{(z_1-z_2)^2} \qquad \in \Gamma(C_2(\CP^1),\pi_1^*K\otimes \pi_2^*K) \subset \Omega^2(C_2(\CP^1)).
\end{equation}
with $\pi_i$ as in (\ref{l26 corr of primary fields as a section of a line bundle}). We will call $\mu$ the \emph{Szeg\"o kernel}.\footnote{In the standard terminology, it is the square root of $\mu$ that is called the Szeg\"o kernel.}

\begin{lemma}
The Szeg\"o kernel defined by (\ref{l27 mu}) is the unique (up to normalization) nowhere vanishing \emph{M\"obius-invariant} holomorphic $2$-form on the configuration space of two points on $\CP^1$. 
\end{lemma}

\begin{proof} To check that $\mu$ is M\"obius-invariant, we observe that it is invariant under (a) translations $z\mapsto z+a$, (b) rotation and dilation $z\mapsto \lambda z$  (since $\mu$ is homogeneous of degree zero), (c) the map $i\colon z\mapsto 1/z$ (indeed, 
$i^*\mu= \frac{\frac{-dz_1}{z_1^2}\frac{-dz_2}{z_2^2}}{(\frac{1}{z_1}-\frac{1}{z_2})^2}=\mu$). These transformation generate all M\"obius transformations, thus $\mu$ is M\"obius-invariant. The fact that $\mu$ is nowhere vanishing is obvious if $z_1,z_2\neq \infty$. For  $z_1=\infty$  we switch for the point $z_1$ to the coordinate chart $w_1=1/z_1$ near the point $\infty\in\CP^1$. We have then  $\mu=-\frac{dw_1\wedge dz_2}{(1-w_1 z_2)^2}$ -- it is nonvanishing at $w_1=0$. The case $z_2=\infty$ is similar.

If $\nu$ is some other M\"obius-invariant section of the line bundle $\pi_1^* K\otimes \pi_2^* K$ over $C_2(\CP^1)$, we must have $\nu=f \mu$ for some M\"obius-invariant function $f$ on $C_2(CP^1)$. Such a function has to be constant, since any two points on $\CP^1$ can be moved to $0,1$ by a M\"obius transformation (and thus $f(z_1,z_2)=f(0,1)$ for any $z_1\neq z_2\in \CP^1$). This proves uniqueness of $\mu$ up to a multiplicative constant.
\end{proof}

In terms of the Szeg\"o kernel, the three-point function of primary fields (\ref{l27 3-point fun}) admits an equivalent expression:
\begin{equation}\label{l27 3-point fun via Szego}
\langle \mathbf{\Phi}_1(z_1) \mathbf{\Phi}_2(z_2) \mathbf{\Phi}_3(z_3) \rangle = C_{\Phi_1\Phi_2\Phi_3} \prod_{1\leq i< j\leq 3} (\pi_{ij}^* \mu)^{\alpha_{ij}} (\pi_{ij}^* \bar\mu)^{\bar\alpha_{ij}}
\end{equation}
where $\pi_{ij}\colon C_3(\CP^1)\ra C_2(\CP^1)$ maps $(z_1,z_2,z_3)\mapsto (z_i,z_j)$ and we used the notation (\ref{l26 Phi underline}). The exponents %$\alpha_{ij}$, $\bar{\alpha}_{ij}$ 
(\ref{l27 alpha_ij})
are chosen in such a way that the r.h.s.  of (\ref{l27 3-point fun via Szego}) is the section of the same line bundle over $C_3(\CP^1)$ as the l.h.s., i.e., so that the power of $dz_i$ is the same on both sides for $i=1,2,3$: 
\begin{equation}
h_1=\alpha_{12}+\alpha_{13},\; h_2=\alpha_{12}+\alpha_{23},\;h_3= \alpha_{13}+\alpha_{23},
\end{equation}
and similarly for powers of $d\bar{z}_i$.

\begin{proof}[Proof \#2 of Lemma \ref{l27 lemma: 3-point fun}]
Denote the r.h.s. of (\ref{l27 3-point fun via Szego}) without the factor $C_{\Phi_1\Phi_2\Phi_3}$ by $A$.
The l.h.s. of (\ref{l27 3-point fun via Szego}) and $A$ both are sections of the line bundle $\bigotimes_{i=1}^3 \pi_i^* K^{h_i,\bar{h}_i}$ over $C_3(\CP^1)$. Moreover, both are M\"obius invariant (the ll.h.s by Ward identity and $A$ by M\"obius-invariance of Szeg\"o kernel) and $A$ is nonvanishing. Therefore, one has
\begin{equation}
(\mbox{l.h.s. of (\ref{l27 3-point fun via Szego})}) = f\cdot A
\end{equation}
where $f$ is a M\"obius-invariant function on $C_3(\CP^1)$. Since M\"obius group acts 3-transitively on $\CP^1$, such a function has to be constant.
\end{proof}

\subsubsection{Correlators of $n\geq 4$ primary fields}
\begin{lemma}\label{l27 lemma: n-point fun from global conformal invariance}
For $\Phi_1,\ldots, \Phi_n\in V$ a collection of $n\geq 4$ primary fields, with $\Phi_i$ of conformal dimension $(h_i,\bar{h}_i)$, one has 
\begin{equation}\label{l27 n-point fun}
\langle \mathbf{\Phi}_1(z_1)\cdots \mathbf{\Phi}_n(z_n) \rangle = \prod_{1\leq i<j \leq n} (\pi_{ij}^*\mu)^{\alpha_{ij}}(\pi_{ij}^*\bar\mu)^{\bar{\alpha}_{ij}}\cdot \mc{F}_{\Phi_1\cdots\Phi_n}(\lambda_1,\ldots,\lambda_{n-3}),
\end{equation}
where $\mu$ is the Szeg\"o kernel (\ref{l27 mu}), $\lambda_i=[z_1,z_2:z_3,z_{i+3}]$ for $i=1,\ldots,n-3$ are cross-ratios, the exponents $\alpha_{ij}$, $\bar\alpha_{ij}$ are
\begin{equation}\label{l27 alpha_ij for n-point fun}
\alpha_{ij}=\frac{1}{n-2}(h_i+h_j-\frac{1}{n-1}\sum_{k=1}^n h_k),\quad \bar\alpha_{ij}=\frac{1}{n-2}(\bar{h}_i+\bar{h}_j-\frac{1}{n-1}\sum_{k=1}^n \bar{h}_k)
\end{equation}
and $\mc{F}_{\Phi_1\cdots\Phi_n}$ is some smooth function on $C_{n-3}(\CP^1\backslash\{0,1,\infty\})$ (it cannot be determined from the global conformal symmetry).
\end{lemma}

Put another way, the result is that any M\"obius-invariant section of the line bundle in the r.h.s. of (\ref{l26 corr of primary fields as a section of a line bundle})  is built out of two types of ``building blocks'' -- cross-ratios and Szeg\"o kernels.

\begin{proof}
The proof is similar to the proof \#2 of Lemma \ref{l27 lemma: 3-point fun} above: the l.h.s. of (\ref{l27 n-point fun}) and $B\colon=  \prod_{1\leq i<j \leq n} (\pi_{ij}^*\mu)^{\alpha_{ij}}(\pi_{ij}^*\bar\mu)^{\bar{\alpha}_{ij}}$ are both M\"obius-invariant sections of the line bundle\footnote{
Note that the exponents (\ref{l27 alpha_ij for n-point fun}) are chosen in such a way that one has $\alpha_{ij}=\alpha_{ji}$ and $\sum_{j\neq i}\alpha_{ij}=h_i$ (and similarly for $\bar\alpha_{ij}$), which implies that  both sides of (\ref{l27 n-point fun}) are sections of the same line bundle.
} $\bigotimes_{i=1}^n \pi_i^* K^{h_i,\bar{h}_i}$  over $C_n(\CP^1)$ and $B$ is nonvanishing, therefore one has
\begin{equation}
(\mbox{l.h.s. of (\ref{l27 n-point fun})}) = g\cdot B
\end{equation}
with $g$ a M\"obius-invariant function on $C_n(\CP^1)$. Choosing a M\"obius transformation that maps $(z_1,\ldots,z_n)$ to $(1,0,\infty,\lambda_1,\ldots,\lambda_{n-3})$, we obtain 
\begin{equation}
g(z_1,\ldots,z_n)=g(1,0,\infty,\lambda_1,\ldots,\lambda_{n-3})=: \mc{F}(\lambda_1,\ldots,\lambda_{n-3}).
\end{equation}
\end{proof}

\section{Holomorphic fields, mode operators}
%\marginpar{can move this subsubsection to the very end of Section 5}
\subsection{Holomorphic fields}
\begin{definition}
We call a (not necessarily primary) field $\Phi\in V$ ``holomorphic'' if it satisfies $\bar\dd\Phi=0$.  Then in particular, correlation functions of the form $\langle \Phi(z) \Phi_1(z_1)\cdots \Phi_n(z_n) \rangle$ are holomorphic in $z$ (for $z$ distinct from $z_1,\ldots,z_n$).
Similarly, we call $\Phi\in V$ ``antiholomorphic'' if it satisfies $\dd\Phi=0$.
\end{definition}

\begin{lemma}\label{l25 lemma: holomorphic fields}
%Assume that the CFT is unitary.
%and satisfies Assumption \ref{l25 assump: (h,hbar) properties} (\ref{l25 assump (a)}) on conformal weights.
Then if a field $\Phi\in V$ %is holomorphic if and only if it 
in a \emph{unitary} CFT has conformal weight of the form $(h,0)$ (i.e. $\bar{h}=0$), then it is holomorphic. 
%then $\bar\dd\Phi=0$. In particular, correlation functions of the form $\langle \Phi(z) \Phi_1(z_1)\cdots \Phi_n(z_n) \rangle$ are holomorphic in $z$. 
Similarly, if $\Phi$ %is antiholomorphic if and only if it 
has conformal weight $(0,\bar{h})$ then it is antiholomorphic.
\end{lemma}

\begin{proof}
Consider a field  
$\Phi\in V$ of conformal weight $(h,\bar{h}=0)$. %has $\bar{h}=0$.
 Computing the square of the norm of $\ol{L}_{-1}\Phi$ we find
\begin{equation}
\Big\langle \ol{L}_{-1}\Phi,\ol{L}_{-1}\Phi \Big\rangle \underset{(\ref{l25 L_n^+ local})}{=}
\Big\langle \Phi,\ol{L}_1\ol{L}_{-1}\Phi \Big\rangle=
\Big\langle \Phi,(2\ol{L}_0+\ol{L}_{-1}\ol{L}_1)\Phi \Big\rangle=2\bar{h} \langle\Phi,\Phi\rangle=0.
\end{equation}
Here we used that $\ol{L}_1\Phi=0$, since if it were nonzero it would be a field of conformal weight $(h,-1)$, and by Assumption \ref{l25 assump: (h,hbar) properties} (\ref{l25 assump (a)}) (implied by unitarity) negative conformal weights are inadmissible.
Since the hermitian form on $V$ is %assumed to be 
positive definite (again by unitarity),
%nondegenerate, 
this implies %that if $\bar{h}=0$ then
\begin{equation}
\ol{L}_{-1}\Phi=\bar\dd\Phi=0,
\end{equation}
i.e., $\Phi$ is a holomorphic field. 
%Conversely, if $\Phi$ is holomorphic, the l.h.s. above vanishes, hence $\bar{h}=0$.
\end{proof}

For example, in any CFT, the stress-energy tensor $T$ is a $(2,0)$-field and therefore is holomorphic.\footnote{We already included holomorphicity of $T$ as a part of axiomatics in (\ref{l23 dbar T=0}). Lemma \ref{l25 lemma: holomorphic fields} provides another explanation why $T$ should be holomorphic.} In the scalar field theory, $\dd\phi$ is a $(1,0)$-field and thus is holomorphic.

\subsection{Mode operators}
\begin{definition} \label{l25 def: mode operators}
Let $\Xi\in V$ be a holomorphic field of conformal weight $(h,0)$, with $h\in\ZZ$. 
%Then $\Xi$ gives rise to a family of 
One defines the
``mode operators'' associated with $\Xi$ as the operators 
$\Xi_{(n)}\in \mr{End}(V)$, with $n\in\ZZ$, defined by
\begin{equation}\label{l25 mode operator}
\Xi_{(n)}\Phi(z)=\frac{1}{2\pi i}\oint_{\gamma_z} dw (w-z)^{n+h-1} \Xi(w) \Phi(z)
\end{equation}
for any test field $\Phi\in V$,
with $\gamma_z$ the contour going around $z$. Put another way, operators $\Xi_{(n)}$ yield terms in the OPE of $\Xi$ with the test field:
\begin{equation}
\Xi(w)\Phi(z)\sim \sum_{n\in\ZZ} \frac{\Xi_{(n)}\Phi(z)}{(w-z)^{n+h}}.
\end{equation}
%
%Similarly, one also has the ``centered-at-zero version'': for $\Xi\in V$ a holomorphic field, one has mode operators $\wh\Xi_{(n)}$ acting on the space of states $\HH$ defined by
%\begin{equation}
%\wh\Xi_{(n)}=\frac{1}{2\pi i}\oint_{\gamma_0}dw\, w^{n+h-1}\wh{\Xi}(w)
%\end{equation}
%with $\gamma_0$ a contour around zero, or equivalently:
%\begin{equation}
%\wh\Xi(w)=\sum_{n\in\ZZ} \frac{\wh\Xi_{(n)}}{w^{n+h}}.
%\end{equation}
\end{definition}

For instance, the mode operators for the stress-energy tensor $T$ are the Virasoro generators $L_n$, cf. (\ref{l25 L_n loc}). Another example: mode operators for the identity field $\mathbb{1}$ are $\mathbb{1}_{(n)}=\delta_{n,0}\,\mr{id}_V$.

The shift by $h$ in the definition (\ref{l25 mode operator}) is designed in such a way that the operator $\Xi_{(-n)}$ shifts the conformal weight by $(n,0)$.

%\marginpar{Add a remark/lemma on commutation relations for mode operators for holomorphic fields from their OPE}

\subsection{The Lie algebra of mode operators.}
\begin{lemma}\label{l25 lemma: algebra of mode operators}
Assume that the CFT contains a collection of holomorphic fields $\{\Phi_i\}_{i\in I}$ (with $I$ an indexing set) of conformal weights $(h_i,0)$ satisfying  the OPEs
\begin{equation}
\Phi_i(w)\Phi_j(z)\sim \sum_{k\in I} f_{ijk}\frac{\Phi_k(z)}{(z-w)^{h_i+h_j-h_k}}+\mr{reg.}
\end{equation}
with $f_{ijk}$ some constants (note that the exponents in the OPE are fixed by Lemma \ref{l26 lemma OPE exponents}). Then the mode operators of fields $\Phi_i$ satisfy the commutation relations
\begin{equation}\label{l25 comm relations of mode operators}
[\Phi_{i(n)},\Phi_{j(m)}]=\sum_{k\in I} f_{ijk} 
\left(
\begin{array}{c}
n+h_i-1 \\ h_i+h_j-h_k-1
\end{array}
 \right)
 \Phi_{k(n+m)}.
\end{equation}
\end{lemma}
The proof is similar to the proof of Virasoro commutation relations from $TT$ OPE in Section \ref{sss: Virasoro from TT OPE}.

\begin{remark}
Similarly to Definition \ref{l25 def: mode operators}, one also has the ``centered-at-zero version'' of mode operators: for $\Xi\in V$ a holomorphic field, one has mode operators $\wh\Xi_{(n)}$ acting on the space of states $\HH$ defined by
\begin{equation}
\wh\Xi_{(n)}=\frac{1}{2\pi i}\oint_{\gamma_0}dw\, w^{n+h-1}\wh{\Xi}(w)
\end{equation}
with $\gamma_0$ a contour around zero, or equivalently:
\begin{equation}
\wh\Xi(w)=\sum_{n\in\ZZ} \frac{\wh\Xi_{(n)}}{w^{n+h}}.
\end{equation}

For example, in the scalar field theory, for the holomorphic field $\Xi=i\dd\phi$, the corresponding mode operators acting on states are the creation/annihilation operators:
\begin{equation}
(\wh{i\dd\phi})_{(n)}=\wh{a}_n,
\end{equation}
as follows from (\ref{l20 dd phi, dbar phi via a abar}).
\end{remark}

\subsection{Ward identity associated with a holomorphic field.}
\begin{lemma}\label{l27 lemma: Ward identity for holom field Xi}
Assume that the CFT contains a holomorphic $\Xi$ of conformal weight $(h,0)$. Then one has the corresponding Ward identity: for any collection of fields $\Phi_1,\ldots,\Phi_n\in V$ and meromorphic section $f=f(w)(\dd_w)^{h-1}$ of the line bundle $K^{\otimes (1-h)}$ over $\CP^1$ with singularities allowed at $z_1,\ldots,z_n$, one has
\begin{equation}\label{l27 Ward identity for holom field Xi}
\sum_{k=1}^n \langle \Phi_1(z_1)\cdots \rho_\Xi^{(z_k)}(f)\circ \Phi_k(z_k)\cdots \Phi_n(z_n) \rangle =0
\end{equation}
where the action of $f$ on $V_{z}$ is given by the contour integral around $z$,
\begin{equation}
\rho_\Xi^{(z)}(f)\circ\Phi(z)\colon=\frac{1}{2\pi i}\oint_{\gamma_z} \underbrace{dw\, f(w) \Xi(w)}_{\iota_f (\Xi(w)(dw)^h)} \Phi(z).
\end{equation}
\end{lemma}
The proof is completely analogous to the proof of the conformal Ward identity (\ref{l26 Ward identity}).

\begin{example}
In the scalar field theory, take $\Xi=i\dd\phi$ and $\Phi_1=V_{\alpha_1},\ldots, \Phi_n=V_{\alpha_n}$ vertex operators, and set $f=1$. Then the Ward identity (\ref{l27 Ward identity for holom field Xi}) reads
\begin{equation}
(\alpha_1+\cdots+\alpha_n) \langle V_{\alpha_1}(z_1)\cdots V_{\alpha_n}(z_n)\rangle =0
\end{equation}
where we used the OPE (\ref{l26 dd phi V OPE}). This implies the result that the correlator of vertex operators can be nonzero only if the sum of their charges $\alpha_i$ vanishes (cf. (\ref{l26 <V...V>})).
\end{example}

\section{Transformation law for the stress-energy tensor}
The action of a holomorphic vector field $u(w)\dd_w$ on the stress-energy tensor is given by
\begin{multline}\label{l27 T infinitesimal transformation}
\rho^{(z)}(u\dd)T(z)\underset{(\ref{l25 rho})}{=}-\frac{1}{2\pi i}\oint_{\gamma_z} dw u(w) T(w) T(z) =\\
= -\frac{1}{2\pi i}\oint_{\gamma_z} dw (u(z)+(w-z)\dd u(z)+\frac12 (w-z)^2 \dd^2u(z)+\frac{1}{6}(w-z)^3 \dd^3u(z)+\cdots)\cdot \\
\cdot \left(\frac{c/2}{(w-z)^4}+\frac{2T(z)}{(w-z)^2}+\frac{\dd T(z)}{w-z}+\mr{reg}\right)\\
=-u(z) \dd T(z)-2\dd u(z)T(z)-\frac{c}{12} \dd^3 u(z)
\end{multline}
If not for the last term, this would have been the transformation law of a $(2,0)$-primary field (cf. (\ref{l25 delta_v Phi computation})). The last term in (\ref{l27 T infinitesimal transformation}) is a certain correction due to the projective property of CFT (a manifestation of conformal anomaly).  We note that the action of an antiholomorphic vector field on $T$ is zero,
\begin{equation}
\rho^{(z)}(\bar{u}\bar\dd)T(z)=0,
\end{equation}
since $\bar{T} T$ OPE is  regular.

The calculation (\ref{l27 T infinitesimal transformation}) expresses the infinitesimal transformation of $T$ under a conformal vector field (seen as an infinitesimal conformal map). Its counterpart for a ``finite'' conformal (holomorphic) transformation $z\mapsto w(z)$ is:
\begin{equation}\label{l27 T finite transformation}
T_{(z)}(z) \mapsto T_{(w)}(w)=\left(\frac{\dd w}{\dd z}\right)^{-2} (T_{(z)}(z)-\frac{c}{12}S(w,z))
\end{equation}
where 
\begin{equation}\label{l27 Schwarzian}
S(w,z)\colon= \frac{\dd_z^3 w}{\dd_z w}-\frac{3}{2} \left(\frac{\dd_z^2 w}{\dd_z w}\right)^2
\end{equation}
is the so-called \emph{Schwarzian derivative} of the holomorphic map $f\colon z\mapsto w(z)$ (we will also use the notation $S(f)$ for the Schwarzian derivative).

Here are some properties of the Schwarzian derivative:
\begin{enumerate}[(a)]
\item \label{l27 S properties (a)} $S$ vanishes on M\"obius transformations,
\item $S$ satisfies a chain-like rule
\begin{equation}\label{l27 S chain-rule}
S(f\circ g)=(S(f)\circ g)\cdot (g')^2+S(g).
\end{equation}
In particular, combining with (\ref{l27 S properties (a)}), we have that for $f$ a M\"obius transformation and $g$ any holomorphic map, $S(f\circ g)=S(g)$.
\item $S$ can be restricted to smooth maps $S^1\ra S^1$. This restriction can be understood as a degree 1 group cocycle of diffeomorphisms of the circle with coefficients in the module of densities of weight $2$:
\begin{equation}
S\in H^1(\mr{Diff}(S^1),\mr{Dens}^2(S^1)).
\end{equation}
This is ultimately a consequence of the ``chain rule'' (\ref{l27 S chain-rule}).
\end{enumerate}

As in Section (\ref{sss: tranformation property of primary fields}), the transformation law (\ref{l27 T finite transformation}) can either be understood in ``active way'' (moving points on the surface $\Sigma$) or ``passive way'' (action of a coordinate transformation).

\begin{example}
Consider $w=\log(z)$ as a holomorphic map from the punctured plane to the cylinder
\begin{equation}
\begin{array}{ccc}
\CC\backslash \{0\} & \ra & \CC/2\pi i \ZZ \\
z & \mapsto & w=\log(z)
\end{array}
\end{equation}
From (\ref{l27 Schwarzian}) one finds
\begin{equation}
S(w,z)=\frac{1}{2z^2}.
\end{equation}
In particular, (\ref{l27 T finite transformation}) becomes
\begin{equation}\label{l27 T plane to cyl}
T_{(z)}(z) \mapsto T_{(w)}(w)=z^2 T_{(z)}(z)-\frac{c}{24}.
\end{equation}
In particular, on $\CC$ one has $\langle T(z) \rangle_\mr{plane}=0$ (this is a consequence of e.g. Lemma \ref{l26 lemma 1-point corr}). Therefore, on the cylinder one has
\begin{equation}
\langle T(w) \rangle_\mr{cylinder}=-\frac{c}{24}.
\end{equation}
Thus, the vacuum energy on the cylinder should be $-\frac{c+\bar{c}}{24}$ instead of zero. 
In physics this %mysterious effect has the name
effect is known as the 
Casimir energy associated with periodic boundary conditions.
\end{example}

\chapter{More free CFTs}

\section{Free scalar field with values in $S^1$}\label{s free scalar field with values in S^1}
An important deficiency of the free scalar field, our main (and only) example of a CFT so far, is that the evolution operator it assigns to a cylinder (or annulus) is not trace-class, which leads to the genus one partition function being ill-defined. This is remedied if we consider free scalar field with values in a circle (instead of values in $\RR$). This model is also known as ``free boson compactified on $S^1$'' (compactification refers to the target) or ``compactified free boson.''

\subsection{Classical theory}
We will introduce the model and quickly retrace our steps in Sections \ref{ss: free scalar on Minkowski cylinder}, \ref{sss: free scalar Minkowski cylinder, Eucl cylinder, C^*}, pointing out where the change of target from $\RR$ to $S^1$ changes the story.

Classically, the model on a Minkowski cylinder $\Sigma=\RR\times S^1$ is defined  by the action functional
\begin{equation}\label{l27 S for boson on S^1}
S(\phi)=  \int dt \underbrace{\int d\sigma \, \frac{\kappa}{2}((\dd_t \phi)^2-(\dd_\sigma \phi)^2)}_{\mathsf{L}}
\end{equation}
(the same formula as (\ref{l16 S})) where $\phi$ now is a smooth map $\Sigma\ra S^1_\mr{target}$ where $S^1_\mr{target}=\RR/2\pi r\ZZ$ is a circle of a fixed radius $r$. Such a maps $\phi$ fall into homotopy classes, classified by a winding number $\m\in\ZZ$: a map with winding $\m$ satisfies $\phi(t,\sigma+2\pi)=\phi(t,\sigma)+2\pi r \m$. We included the conventional normalization $\kappa=\frac{1}{4\pi}$ in (\ref{l27 S for boson on S^1}).

Thus, the space of fields splits as a disjoint union
of spaces of maps to $S^1_\mr{target}$ with a given winding number:
\begin{equation}
\F=\mr{Map}(\Sigma,S^1_\mr{target})=\bigsqcup_{\m\in \ZZ}\underbrace{\mr{Map}_\m(\Sigma,S^1_\mr{target})}_{\mr{maps\;with\;winding\;number\;}\m}
\end{equation}

One can then consider this model as classical mechanics with target 
\begin{equation}
X=\bigsqcup_{\m\in \ZZ}\mr{Map}_\m(S^1,S^1_\mr{target})
\end{equation}
with Lagrangian $\mathsf{L}$ as in (\ref{l27 S for boson on S^1}). A field $\phi\in X_\m$ with winding $\m$ can be expanded in Fourier modes, plus a shift linear in $\Sigma$, accounting for the winding:
\begin{equation}\label{l27 phi via modes}
\phi(\sigma)=%\phi_0+
\m r \sigma+\sum_{n\in\ZZ} \phi_n e^{in\sigma}
\end{equation}

Transitioning to the Hamiltonian formalism (by Legendre transform), we have the phase space  
\begin{equation}\label{l27 Phi splitting by winding}
\Phi=T^*X=\bigsqcup_{\m\in \ZZ} \underbrace{T^* X_\m}_{\Phi_\m}
\end{equation}
parameterized in $\m$-th sector by the field $\phi(\sigma)$ and the Darboux-conjugate ``momentum'' $\pi(\sigma)=\frac{1}{2\pi}\sum_{n\in\ZZ} \pi_n e^{in\sigma}$. The modes satisfy the standard Poisson brackets %$\{\pi_j,\phi_k\}=\delta_{j,-k}$.  
(\ref{l16 Poisson brackets between modes}).
The Hamiltonian on $\Phi_\m$ in terms of Fourier modes  is 
\begin{equation}
H=\pi_0^2+\left(\frac{\m r}{2}\right)^2+\sum_{n\neq 0}(\pi_n \pi_{-n}+\frac14 n^2 \phi_n \phi_{-n}).
\end{equation}
Note that this differs from the Hamiltonian (\ref{l16 H via modes}) by a shift $\left(\frac{\m r}{2}\right)^2$ which arises from the $\sigma$-linear term in (\ref{l27 phi via modes}).

\subsection{Canonical quantization}
We proceed with canonical quantization of the theory. The splitting (\ref{l27 Phi splitting by winding}) of the phase space means that the space of states splits as a direct sum
\begin{equation}
\HH=\bigoplus_{\m\in \ZZ} \HH_\m,
\end{equation}
where $\HH_\m$ consists of states with winding number $\m$.

Let $\wh{\m}$ be the operator on $\HH$ which has eigenvalue $\m$ on $\HH_\m$. The quantum hamiltonian is 
\begin{equation}
\wh{H}=\wh\pi_0^2+\left(\frac{\wh{\m}r}{2}\right)^2+\sum_{n\neq 0}(\wh\pi_n \wh\pi_{-n}+\frac14 n^2 \wh\phi_n \wh\phi_{-n}).
\end{equation}
Similarly to Section \ref{sss free scalar field can quant}, 
the Hamiltonian splits into 
\begin{itemize}
\item 
A collection of harmonic oscillators (one for each $n\neq 0$, with frequency $\omega_n=|n|$). For the oscillators we introduce creation/annihilation operators $\wh{a}_n$, $\wh{\ol{a}}_n$, $n\neq 0$, exactly as before (\ref{l17 phi, pi via a, abar});  they satisfy the usual  commutation relations (\ref{l17 a,abar comm rel}).
\item  A free particle of mass $\mu=\frac12$ \ul{with values in $S^1_\mr{target}$} (described by $\wh\phi_0$, $\wh\pi_0$).
\item  A shift by a constant depending on winding, $\left(\frac{\wh{\m}r}{2}\right)^2$.
\end{itemize}

For a free quantum particle on $S^1_\mr{target}$ the space of states in Schr\"odinger representation is $L^2(S^1_\mr{target})$ (the space of $2\pi r$-periodic $L^2$ functions $\psi(\phi_0)$) with $\wh\phi_0$ acting by multiplication $\psi(\phi_0)\mapsto \phi_0 \psi(\phi_0)$ and $\wh\pi_0=-i\frac{\dd}{\dd \phi_0}$ the derivation. Two important points here (in comparison with Section \ref{sss free particle}):
\begin{itemize}
\item The eigenvectors of $\wh\pi_0$ are functions $\psi_\e(\phi_0)=e^{\frac{i \e}{r} \phi_0}$ with $\e\in \ZZ$, the corresponding eigenvalue is $\frac{\e}{r}$. In particular, the eigenvalue spectrum of $\wh\pi_0$ is \emph{discrete}: 
\begin{equation}
\left\{\frac{\e}{r}\right\}_{\e\in\ZZ}=\frac{1}{r}\ZZ,
\end{equation}
unlike the case of a free particle on $\RR$ where the spectrum of momentum operator is $\RR$.
\item ``Operator'' $\wh\phi_0$ is multi-valued (defined modulo $2\pi r \ZZ\cdot \mr{Id}$). In particular, it is not a well-defined operator in the usual sense, though exponentials $\wh{v}^n\colon=e^{i\frac{n}{r}\wh\phi_0}$ are well-defined operators for $n\in \ZZ$.\footnote{In $\wh{v}^n$, the superscript can be read either as index or as a power (of the operator $\wh{v}=\wh{v}^1$).} They satisfy the commutation relation 
\begin{equation}\label{l27 [pi_0,v^n]}
[\wh\pi_0,\wh{v}^n]=\frac{n}{r} \wh{v}^n.
\end{equation}
\end{itemize} 

Retracing our steps with the scalar field, we proceed with the canonical quantization, construct the Heisenberg field operator, switch to Euclidean cylinder by Wick rotation and map to $\CC\backslash\{0\}$ by the exponential map, arriving at the Heisenberg field operator
\begin{equation}\label{l27 phi}
\wh\phi(z)=\wh\phi_0-i\frac{\wh{\m}r}{2}\log\frac{z}{\bar{z}}-i\wh\pi_0 \log(z\bar{z})+\sum_{n\neq 0} \frac{i}{n}(\wh{a}_n z^{-n}+\wh{\ol{a}}_n \bar{z}^{-n})
\end{equation}

As we discussed above, $\wh\pi_0$ has eigenvalue spectrum $\frac{1}{r}\ZZ$. So, we introduce the operator $\wh{\e}\colon=r\wh\pi_0$ which has integer eigenvalues. In terms of this new notation, the field operator (\ref{l27 phi}) is
\begin{equation}
\wh\phi(z)=\wh\phi_0-i\frac{\wh{\m}r}{2}\log\frac{z}{\bar{z}}-i\frac{\wh{\e}}{r} \log(z\bar{z})+\sum_{n\neq 0} \frac{i}{n}(\wh{a}_n z^{-n}+\wh{\ol{a}}_n \bar{z}^{-n})
\end{equation}

\marginpar{Lecture 28,\\ 11/2/2022}

The derivatives of the field operator are
\begin{equation}
i\dd\wh\phi(z)=\sum_{n\in \ZZ}\wh{a}_n z^{-n-1},\quad i\bar\dd\wh\phi(z)=\sum_{n\in \ZZ}\wh{\ol{a}}_n \bar{z}^{-n-1}
\end{equation}
(same formulae as (\ref{l20 dd phi, dbar phi via a abar})),
where we defined 
\begin{equation}
\wh{a}_0\colon = \frac{\wh{\e}}{r}+\frac{\wh{\m}r}{2},\quad \wh{\ol{a}}_0\colon = \frac{\wh{\e}}{r}-\frac{\wh{\m}r}{2}.
\end{equation}

The stress-energy tensor is given by the same formula as for the scalar field valued in $\RR$, $T=:-\frac12\dd\phi \dd\phi:$ (the normal ordering is defined as usual, putting the operator $\wh{a}_{\geq 0}$, $\wh{\ol{a}}_{\geq 0}$ to the right). Thus, the Virasoro generators are again given by  (\ref{l23 L_scalar via a a}), (\ref{l23 Lbar via abar abar}) and the Hamiltonian and total momentum operators are given by (\ref{l23 H, P via L_0, Lbar_0}):
\begin{equation}\label{l28 H,P via a}
\begin{gathered}
\wh{H}=\wh{L}_0+\wh{\ol{L}}_0=\frac12 \sum_{n\in\ZZ}:\wh{a}_n \wh{a}_{-n}+\wh{\ol{a}}_{n}\wh{\ol{a}}_{-n}:,\\
\wh{P}=\wh{L}_0-\wh{\ol{L}}_0=\frac12 \sum_{n\in\ZZ}:\wh{a}_n \wh{a}_{-n}-\wh{\ol{a}}_{n}\wh{\ol{a}}_{-n}:
\end{gathered}
\end{equation}

\subsection{Space of states}
The space of states of the scalar field with values in $S^1_\mr{target}$ (the Fock space) is 
\begin{equation}\label{l28 Fock space}
\HH=\mr{Span}_\CC \Big\{\wh{a}_{-n_r}\cdots \wh{a}_{-n_1}\wh{\ol{a}}_{-\bar{n}_s}\cdots\wh{\ol{a}}_{-\bar{n}_1}|\e,\m\rangle\Big| 
\begin{array}{c}
1\leq n_1\leq\cdots \leq n_r,\\
1\leq \bar{n}_1\leq\cdots \leq \bar{n}_s,\\
(\e,\m)\in\ZZ^2
\end{array}
\Big\}
\end{equation}
%(cf. the space of states for the usual free scalar field (\ref{l17 Fock space})).
The vector $|\e,\m\rangle\in \HH$ (``pseudovacuum'') is annihilated by the annihilation operators $\wh{a}_{>0}$, $\wh{\ol{a}}_{>0}$ and is an eigenvector of $\wh{a}_0, \wh{\ol{a}}_0$:
\begin{equation}\label{l28 a_0 abar_0 on pseudovacua}
\wh{a}_0 |\e,\m\rangle = \underbrace{\left(\frac{\e}{r}+\frac{\m r}{2}\right)}_{\alpha_{\e,\m}}|\e,\m\rangle, \qquad 
\wh{\ol{a}}_0 |\e,\m\rangle = \underbrace{\left(\frac{\e}{r}-\frac{\m r}{2}\right)}_{\bar\alpha_{\e,\m}}|\e,\m\rangle
\end{equation}
where we introduced the notations $\alpha_{\e,\m}$, $\bar\alpha_{\e,\m}$ for the respective eigenvalues.

Another way to express the the space of states is as a direct sum of Verma modules of the Lie algebra $\mr{Heis}\oplus \ol{\mr{Heis}}$ (the direct sum of two Heisenberg Lie algebras (\ref{l17 Heisenberg Lie algebra})) with highest weights (eigenvalues of $\wh{a}_0,\wh{\ol{a}}_0$) given by pairs  $(\alpha_{\e,\m},\bar\alpha_{\e,\m})$:
\begin{equation}
\HH=\bigoplus_{(\e,\m)\in \ZZ^2} \underbrace{\mathbb{V}^{\mr{Heis}\oplus \ol{\mr{Heis}}}_{(\alpha_{\e,\m},\bar\alpha_{\e,\m})} }_{\HH_{\e,\m}}
\end{equation} 

Note that the main distinction from the case of the usual free scalar theory (\ref{l17 Fock space}) is the structure of pseudovacua: previously we had a \emph{continuum family} of pseudovacua $|\pi_0\rangle$ characterized by the value of the zero-mode momentum $\pi_0\in \RR$, whereas now we have a \emph{lattice} of pseudovacua $|\e,\m\rangle$ characterized by the (integer) zero-mode momentum $\e$ and the winding number $\m$.\footnote{
The notations $\e,\m$ correspond to ``electric'' and ``magnetic'' number.
}

The energy and total momentum of pseudovacua are found from (\ref{l28 H,P via a}):
\begin{equation}
\begin{gathered}
\wh{H}|\e,\m\rangle=\frac12(\alpha_{\e,\m}^2+\bar\alpha_{\e,\m}^2)|\e,\m\rangle =
\left(\left(\frac{\e}{r}\right)^2+\left(\frac{\m r}{2}\right)^2 \right) |\e,\m \rangle, \\
\wh{P}|\e,\m\rangle=\frac12(\alpha_{\e,\m}^2-\bar\alpha_{\e,\m}^2)|\e,\m\rangle =
\e\m |\e,\m \rangle
\end{gathered}
\end{equation}
Note that while the eigenvalue of $\wh{H}$ is a non-negative real number, the eigenvalue of $\wh{P}$ is always an integer.
Also note that the only pseudovacuum with zero energy (eigenvalue of $\wh{H}$) is $|\e=0,\m=0\rangle$. It also has zero total momentum and we identify this particular state as the ``true'' (as opposed to ``pseudo-'') vacuum, $|\vac\rangle\colon = |0,0\rangle$

As in the ordinary free scalar theory, we have that 
\begin{itemize}
\item 
applying $\wh{a}_{-n}$  to a state changes energy-momentum by $(n,n)$ (creates a ``left-mover''), 
\item applying $\wh{\ol{a}}_{-n}$ to a state changes energy-momentum by $(n,-n)$ (creates a ``right-mover''),
\end{itemize}
where we assume $n>0$. 

The pseudovacuum $|\e,\m\rangle$ is also an eigenvector of the Virasoro generators $\wh{L}_0$, $\wh{\ol{L}}_0$ with
\begin{equation}\label{l28 L_0 acting on pseudovacuum}
\wh{L}_0 |\e,\m\rangle =\underbrace{\frac12 \alpha_{\e,\m}^2}_{h_{\e,\m}} |\e,\m\rangle,\quad
\wh{\ol{L}}_0 |\e,\m\rangle =\underbrace{\frac12 \bar\alpha_{\e,\m}^2}_{\bar{h}_{\e,\m}} |\e,\m\rangle.
\end{equation}

\subsection{Vertex operators}
The counterpart of pseudovacua $|\e,\m\rangle$ via the field-state correspondence are the vertex operators $V_{\e,\m}\in V$, constructed somewhat differently than in the non-compactified scalar field theory.

Let us introduce an ``operator'' $\wh\mu$ on $\HH$, defined modulo $2\pi\ZZ\cdot \mr{Id}$ (similarly to the operator $\wh\phi_0$) satisfying $[\wh\mu,\wh{m}]=i$ and commuting with $\wh{a}_{\neq 0}, \wh{\ol{a}}_{\neq 0},\wh{e},\wh{\phi}_0$. Then for $k\in\ZZ$ the exponential $e^{ik\wh\mu}$ is a well-defined operator  on $\HH$ satisfying 
\begin{equation}
[\wh{m},e^{ik\wh{\mu}}]=k e^{ik\wh\mu},
\end{equation}
cf. (\ref{l27 [pi_0,v^n]}), i.e., the operator 
\begin{equation}\label{l28 exp(ik mu)}
\begin{array}{cccc}
e^{ik\wh\mu}\colon &\HH_{\e,\m}&\ra &\HH_{\e,\m+k} \\
& 
%\prod_{i=1}^r \wh{a}_{-n_i} \prod_{j=1}^s \wh{\ol{a}}_{-\bar{n}_j}
 \wh{a}_{-n_r}\cdots \wh{a}_{-n_1}\wh{\ol{a}}_{-\bar{n}_s}\cdots\wh{\ol{a}}_{-\bar{n}_1}
 |\e,\m\rangle &\mapsto &
\wh{a}_{-n_r}\cdots \wh{a}_{-n_1}\wh{\ol{a}}_{-\bar{n}_s}\cdots\wh{\ol{a}}_{-\bar{n}_1}|\e,\m+k\rangle
\end{array}
\end{equation}
shifts the magnetic (or winding) number $\m$ by $k$.\footnote{Instead of introducing the operator $\wh\mu$, one can 
%directly introduce a family of operators 
treat (\ref{l28 exp(ik mu)}) as the definition of a family of operators on $\HH$, formally denoted $e^{ik\wh\mu}$. From this viewpoint, $\wh\mu$ is a purely notational device, only meaningful in the combination $e^{ik\wh\mu}$.
} Similarly, due to (\ref{l27 [pi_0,v^n]}), the operator $e^{il\wh\phi_0/r}$ shifts the electric number $\e$ by $l$:
\begin{equation}
\begin{array}{cccc}
e^{il\wh\phi_0/r}\colon &\HH_{\e,\m}&\ra &\HH_{\e+l,\m} \\
& \wh{a}_{-n_r}\cdots \wh{a}_{-n_1}\wh{\ol{a}}_{-\bar{n}_s}\cdots\wh{\ol{a}}_{-\bar{n}_1}|\e,\m\rangle &\mapsto &
\wh{a}_{-n_r}\cdots \wh{a}_{-n_1}\wh{\ol{a}}_{-\bar{n}_s}\cdots\wh{\ol{a}}_{-\bar{n}_1}|\e+l,\m\rangle
\end{array}.
\end{equation}

Further, let us introduce the following two multivalued operators (the ``holomorphic/antiholomorphic parts of $\wh\phi$''):
\begin{equation}\label{l28 chi, chibar}
\wh\chi(z)=\frac12 \wh\phi_0+\frac{\wh{\mu}}{r}-i\wh{a}_0 \log z+\sum_{n\neq 0}\frac{i}{n} \wh{a}_n z^{-n},\quad
\wh{\ol\chi}(z)=\frac12 \wh\phi_0-\frac{\wh{\mu}}{r}-i\wh{\ol{a}}_0 \log \bar{z}+\sum_{n\neq 0}\frac{i}{n} \wh{\ol{a}}_n \bar{z}^{-n}.
\end{equation}
In particular, one has $\wh\phi(z)=\wh\chi(z)+\wh{\ol{\chi}}(z)$.

\begin{definition}
The vertex operator $V_{\e,\m}$ in the compactified free scalar field CFT is defined by
%One defines the vertex operator  as
\begin{equation}\label{l28 vertex operator}
\wh{V}_{\e,\m}(z)\colon = \; 
:
e^{i\alpha_{\e,\m}\wh\chi(z)} e^{i \bar\alpha_{\e,\m} \wh{\ol\chi}(z)}
:
\end{equation}
Here the parameters $\e,\m$ are integers and $\alpha_{\e,\m}$, $\bar\alpha_{\e,\m}$ are as in (\ref{l28 a_0 abar_0 on pseudovacua}).
\end{definition}
The normal ordering puts operators $\wh{a}_{\geq 0},\wh{\ol{a}}_{\geq 0}$ to the right and operators $\wh{a}_{<0}, \wh{\ol{a}}_{<0}, \wh\phi_0,\wh\mu$ to the left. Written more explicitly, the vertex operator is
\begin{multline}
\wh{V}_{\e,\m}(z) = e^{i\e\wh\phi_0/r}e^{i\m \wh\mu}e^{-\sum_{n<0} \frac{1}{n}(\alpha_{\e,\m} \wh{a}_n z^{-n}+\bar\alpha_{\e,\m} \wh{\ol{a}}_n \bar{z}^{-n})}\cdot \\
\cdot e^{-\sum_{n>0} \frac{1}{n}(\alpha_{\e,\m} \wh{a}_n z^{-n}+\bar\alpha_{\e,\m} \wh{\ol{a}}_n \bar{z}^{-n})}e^{\alpha_{\e,\m}\wh{a}_0\log z+\bar\alpha_{\e,\m} \wh{\ol{a}}_0 \log\bar{z}}.
\end{multline}
Somewhat non-obviously, this is a single-valued operator: the multi-valued operators $\wh\phi_0$, $\wh\mu$ are only present in single-valued exponential expressions; the last exponential is single valued when acting on $\HH_{\e',\m'}$ since one has
\begin{equation}\label{l28 alpha alpha - alpha alpha = integer}
\alpha_{\e,\m} \alpha_{\e',\m'}-\bar\alpha_{\e,\m}\bar\alpha_{\e',\m'}=\e\m'+\m \e' \in \ZZ.
\end{equation}

Performing computations similar to those of Section \ref{sss vertex operators}, \ref{sss primary fields in scalar field CFT}, \ref{sss more on vertex operators}, one proves the following properties of vertex operators:
\begin{itemize}
\item $V_{\e,\m}$ is a primary field of conformal weight
\begin{equation}\label{l28 V conf weights}
h_{\e,\m}=\frac12 \left(\frac{\e}{r}+\frac{\m r}{2}\right)^2,\quad\bar{h}_{\e,\m}=\frac12 \left(\frac{\e}{r}-\frac{\m r}{2}\right)^2
\end{equation} 
-- same  $h_{\e,\m}, \bar{h}_{\e,\m}$ as in (\ref{l28 L_0 acting on pseudovacuum}).
\item One has 
\begin{equation}
\lim_{z\ra 0} \wh{V}_{\e,\m}(z)|\vac\rangle =|\e,\m\rangle,
\end{equation}
i.e., as claimed in the beginning of this section, the state corresponding to the vertex operator $V_{\e,\m}$ by the field-state correspondence is the pseudovacuum $|\e,\m\rangle$. More generally, one has
\begin{equation}
\lim_{z\ra 0}  : \prod_{j=1}^r \frac{i\dd^{n_j}\wh\phi(z)}{(n_j-1)!} \prod_{k=1}^s \frac{i\bar\dd^{\bar{n}_k}\wh\phi(z)}{(\bar{n}_k-1)!}
 \wh{V}_{\e,\m}(z):|\vac\rangle =
 \wh{a}_{-n_r}\cdots \wh{a}_{-n_1}\wh{\ol{a}}_{-\bar{n}_s}\cdots \wh{\ol{a}}_{\bar{n}_1}|\e,\m \rangle,
\end{equation}
i.e., the fields corresponding to basis states of $\HH$ are the vertex operators multiplied by differential polynomials in $\phi$.
\item The correlator of $n$ vertex operators is
\begin{equation} \label{l28 <V..V>}
\hspace{-1cm}
\left\langle \prod_{k=1}^n V_{\e_k,\m_k}(z_k) \right\rangle=\left\{
\begin{array}{cl}\displaystyle 
\prod_{1\leq i<j\leq n} (z_i-z_j)^{\alpha_{\e_i,\m_i}\alpha_{\e_j,\m_j}} (\bar{z}_i-\bar{z}_j)^{\bar\alpha_{\e_i,\m_i}\bar\alpha_{\e_j,\m_j}},&
\mr{if}\; \sum_{i=1}^n \e_i =\sum_{i=1}^n \m_i =0,\\
0,& \mr{otherwise}
\end{array}
\right.
\end{equation}
Despite the real exponents appearing here, the entire expression on the right is in fact a single-valued function on $C_n(\CP^1)$, due to (\ref{l28 alpha alpha - alpha alpha = integer}). For instance, for $n=2$ one has
\begin{equation}
\left\langle V_{\e,\m}(w) V_{-\e,-\m}(z) \right\rangle = |w-z|^{-2\left(
\left(\frac{\e}{r}\right)^2+ \left(\frac{\m r}{2}\right)^2
\right)} \left(\frac{w-z}{\bar{w}-\bar{z}}\right)^{-\e\m}
\end{equation}
-- note that the first exponent on the right is real while the second is an integer, making the expression single-valued.
\end{itemize}

\subsection{Torus partition function in a general CFT}

Consider the torus $\mathbb{T}$ obtained from the annulus $\{z\in \CC\; |\; r_\mr{in} \leq|z|\leq r_\mr{out}\}$ by identifying the inner and outer circles via the identification $r_\mr{in} e^{i\sigma} \sim r_\mr{out} e^{i\sigma}$. Equivalently, we map the annulus by the map $z\mapsto \zeta=\log z$ to the cylinder 
%$[\log r_\ii,\log r_\oo]\times i\RR/$
\begin{equation}
\mr{cyl}=\{\zeta=t+i\sigma \in \CC/2\pi i\ZZ\;|\; \log r_\ii\leq t\leq \log r_\oo\}
\end{equation} 
and identify the boundary circles by $\log r_\ii+i\sigma\sim \log r_\oo +i\sigma$. This yields a complex torus with modular parameter 
\begin{equation}\label{l28 tau}
\tau=\frac{i}{2\pi} T 
\end{equation}
with $T=\log\frac{r_\oo}{r_\ii}$.

%\textcolor{red}{PICTURE}
\begin{figure}[H]
%$$\vcenter{\hbox{ \includegraphics[scale=0.8]{l26_contours_1.eps} }} 
%\longrightarrow
%\vcenter{\hbox{ \includegraphics[scale=0.8]{l26_contours_2.eps} }} 
%$$
\begin{center}
\includegraphics[scale=0.85]{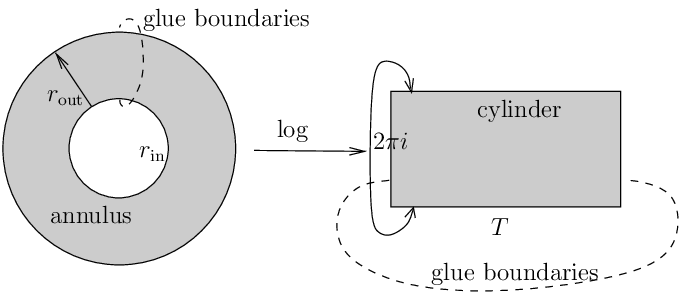}
\end{center}
\caption{Torus obtained from annulus or cylinder by identifying the boundary circles.}
\end{figure}

The evolution operator for the cylinder of Euclidean length $T$ is 
\begin{equation}
Z(\mr{cyl}_T)=e^{-T\wh{H} } = e^{-T(\wh{L}_0+\wh{\ol{L}}_0)} 
\end{equation}
The partition function for the torus is the trace of this evolution operator over the space of states,
\begin{equation}
Z(\mathbb{T}_\tau)=\tr_\HH e^{-T\wh{H} } = \tr_\HH e^{2\pi i \tau (\wh{L}_0+\wh{\ol{L}}_0)}
\end{equation}
with $\tau$ the modular parameter (\ref{l28 tau}).

\marginpar{Lecture 29,\\
11/4/2022}

\subsubsection{Gluing with a twist by angle $\theta$.}
More generally, one can glue the inner and outer boundary circles of the annulus with a twist by angle $\theta$:  $r_\ii e^{i\sigma}\sim r_\oo e^{i\sigma+\theta}$, or equivalently identify the boundary circles of the cylinder as $\log r_\ii +i\sigma \sim \log r_\oo +i (\sigma+\theta)$.
Denote $\mr{cyl}_{T,\theta}$ the mapping cylinder (\ref{l1 mapping cylinder}) of length $T$ (understood as a cobordism $S^1\ra S^1$), associated with mapping $\rho_r\colon S^1\ra S^1$ rotating the circle by angle $\theta$. Then one has
\begin{equation}\label{l29 Z(mapping cylinder)}
Z(\mr{cyl}_{T,\theta})= e^{-T\wh{H}-i\theta \wh{P}} =  e^{2\pi i \tau \wh{L}_0-2\pi i \bar\tau \wh{\ol{L}}_0} = q^{\wh{L}_0}\bar{q}^{\wh{\ol{L}}_0}
\end{equation}
with $\wh{P}$ the total momentum operator. 
In (\ref{l29 Z(mapping cylinder)})
we denoted
\begin{equation}
q=e^{2\pi i \tau}
\end{equation}
and $\bar{q}$ is its complex conjugate; note that since $\mr{Im}(\tau)>0$, one has $|q|<1$. We used the expressions (\ref{l23 H, P via L_0, Lbar_0}) for the total energy/momentum as $\wh{L}_0\pm \wh{\ol{L}}_0$.

This yields a complex torus with modular parameter $\tau=\frac{i}{2\pi} (T+i\theta)$ and the corresponding partition function is
\begin{equation}\label{l29 Z(torus)}
Z(\mathbb{T}_\tau)=\tr_\HH Z(\mr{cyl}_{T,\theta}) = \tr_\HH q^{\wh{L}_0}\bar{q}^{\wh{\ol{L}}_0}
\end{equation}

\subsubsection{Correction due to central charge.}
In fact, one needs to introduce a correction in (\ref{l29 Z(torus)}):
\begin{equation}\label{l29 Z(torus) corrected}
Z(\mathbb{T}_\tau)=\tr_\HH Z(\mr{cyl}_{T,\theta}) = \tr_\HH q^{\wh{L}_0\textcolor{red}{-\frac{c}{24} }}\bar{q}^{\wh{\ol{L}}_0\textcolor{red}{-\frac{\bar{c}}{24}}},
\end{equation}
with $(c,\bar{c})$ the holomorphic/antiholomorphic central charge of the CFT (one also needs similar correction in (\ref{l29 Z(mapping cylinder)})). The reason for this correction can be explained in several ways:
\begin{enumerate}[(i)]
\item The correction in (\ref{l29 Z(torus) corrected}) arises from the Schwarzian derivative correction in the transformation law of the stress-energy tensor (\ref{l27 T finite transformation}), (\ref{l27 T plane to cyl}), which implies 
\begin{equation}
\wh{H}_\mr{cyl}=\frac{1}{2\pi}\oint_{t=\mr{const}} \iota_{\frac{\dd}{\dd t}}(T_\mr{cyl}(d\zeta)^2+\ol{T}_\mr{cyl}(d\bar\zeta)^2)=\wh{H}_\mr{plane}-\frac{c+\bar{c}}{24}=\wh{L}_{0}+\wh{\ol{L}}_0-\frac{c+\bar{c}}{24}.
\end{equation}
and similarly for the total momentum operator; Virasoro generators $\wh{L}_0$, $\wh{\ol{L}}_0$ are understood as pertaining to the plane and to the radial quantization picture (thus, when mapping to the cylinder by the map $z\mapsto \zeta=\log(z)$ they receive the Schwarzian correction).
\item \label{l29 (ii)}
Expression (\ref{l29 Z(torus) corrected}) is the partition function for a torus with \emph{flat metric} (obtained from the flat metric on the cylinder), whereas (\ref{l29 Z(torus)}) is the partition function for the torus with a singular metric obtained by taking the flat annulus and identifying the two boundary circles (the glued surface has a metric which is flat almost everywhere, except at the circle where the gluing was performed -- there the metric is singular, since e.g. the identified circles had different lengths). Conformal anomaly means that the partition function has a dependence on the metric within the conformal class (\ref{l2 change of Z under Weyl}). Thus, the factor $q^{-\frac{c}{24}}\bar{q}^{-\frac{\bar{c}}{24}}$ in (\ref{l29 Z(torus) corrected}) is the explonential of the Liouville action in (\ref{l2 change of Z under Weyl}) corresponding to the change from the singular metric on $\mathbb{T}$ coming from the annulus to the flat metric.
\item Pragmatic viewpoint: the partition function for the torus is expected to be modular invariant, in particular, it should be invariant under $\tau\mapsto -\frac{1}{\tau} $. As we will see in the example of the free scalar field with values in $S^1$, expression (\ref{l29 Z(torus) corrected}) has this property, while (\ref{l29 Z(torus)}) does not. This is connected with item (\ref{l29 (ii)}) above: flat tori with modular parameters $\tau$ and $-1/\tau$ are connected by a constant Weyl transformation, for which the Liouville action in (\ref{l2 change of Z under Weyl}) is zero. %On the other hand, the singular metric coming from the annulus on conformal map $\mathbb{T}_\tau\ra \mathbb{T}_{-1/\tau}$
For the singular metric coming from the annulus, this is not true: the metric tori $\mathbb{T}_\tau$, $\mathbb{T}_{-1/\tau}$ have ``scars'' -- singular loci of the metric -- and they are not intertwined by the conformal map $\mathbb{T}_\tau\ra \mathbb{T}_{-1/\tau}$.
\end{enumerate}

\subsection{Torus partition function for the 
free scalar field with values in $S^1$}\label{sss torus Z for compactified free boson}
In our case the central charge is $c=\bar{c}=1$ and the formula (\ref{l29 Z(torus) corrected}) becomes
\begin{multline}\label{l29 Z(torus) computation}
%Z(\mathbb{T}_\tau) 
Z(\tau)
= \tr_\HH q^{\wh{L}_0-\frac{1}{24}}\bar{q}^{\wh{\ol{L}}_0-\frac{1}{24}}=\\
=
\sum_{(\e,\m)\in \ZZ^2}\sum_{1\leq n_1\leq\cdots \leq n_r,\; 1\leq \bar{n}_1\leq \cdots \leq \bar{n}_s} q^{h_{\e,\m}+\sum_{i=1}^r n_i-\frac{1}{24}}
\bar{q}^{\bar{h}_{\e,\m}+\sum_{j=1}^s \bar{n}_j -\frac{1}{24}}
\end{multline}
For brevity we denote the partition function of the torus with modular parameter $\tau$ simply as $Z(\tau)$.
Here we used that the operators $\wh{L}_0$, $\wh{\ol{L}}_0$ are diagonal in the basis (\ref{l28 Fock space}); the exponents in the r.h.s. of (\ref{l29 Z(torus) computation}) are the corresponding eigenvalues shifted by $-\frac{1}{24}$; $h_{\e,\m}$, $\bar{h}_{\e,\m}$ are the conformal weights of the pseudovacua (\ref{l28 L_0 acting on pseudovacuum}), (\ref{l28 V conf weights}). Continuing the computation, we have
\begin{equation}\label{l29 Z(torus) computation 2}
Z(\tau) = 
\sum_{(\e,\m)\in \ZZ^2} q^{h_{\e,\m}}\bar{q}^{\bar{h}_{\e,\m}} (q\bar{q})^{\frac{1}{24}} \sum_{k,l\geq 0} P(k)P(l)q^k \bar{q}^l,
\end{equation}
where $P(k)$ is the number of partitions of $k$, i.e., the number of nondecreasing sequences $1\leq n_1\leq \cdots n_r$ such that $k=n_1+\cdots + n_r$, for some $r\geq 1$. For instance, one has
\begin{equation}
\begin{aligned}
4&=1+1+1+1\\ &= 1+1+2 \\ &=2+2 \\ &=1+3 \\&= 4, 
\end{aligned}
\end{equation}
thus, $P(4)=5$. In (\ref{l29 Z(torus) computation 2}), the left factor is the sum over pseudovacua, the middle factor is the central charge correction, and the right factor accounts for the contributions of $\mr{Heis}\oplus\ol{\mr{Heis}}$-descendants of the pseudovacuum (and $P(k)P(l)$ is the count of descendants of conformal weight $(h_{\e,\m}+k,\bar{h}_{\e,\m}+l)$).

The generating function for the numbers of partitions is a well-studied object of  combinatorics,
\begin{equation}
\sum_{k\geq 0} P(k) q^k = 
\frac{1}{\prod_{n\geq 1} (1-q^n)}=\frac{q^{\frac{1}{24}}}{\eta(\tau)}
\end{equation}
where 
\begin{equation}\label{l29 Dedekind eta}
\eta(\tau)=q^{ \frac{1}{24}
}\prod_{n\geq 1}(1-q^n)
\end{equation}
is the Dedekind eta-function which satisfies the modular equivariance properties\footnote{Property (\ref{l29 eta modular property 1}) is obvious from the definition (\ref{l29 Dedekind eta}). Property (\ref{l29 eta modular property 2}) follows from the Euler's identity $\prod_{n\geq 1}(1-q^n)=\sum_{j\in \ZZ}(-1)^j q^{\frac{3j^2-j}{2}}$ by applying Poisson summation formula (cf. footnote \ref{l2 footnote: Poisson summation}).}
\begin{eqnarray}
\eta(\tau+1)&=&e^{i\pi/12} \eta(\tau), \label{l29 eta modular property 1} \\
\eta(-1/\tau)&=& (-i\tau)^{\frac12} \eta(\tau) \label{l29 eta modular property 2}
\end{eqnarray}

Finally, the partition function (\ref{l29 Z(torus) computation 2}) can be written in the form
\begin{equation} \label{l29 Z(tau) explicitly}
Z(\tau)= \frac{1}{\eta(\tau)\eta(\bar\tau)}\sum_{(\e,\m)\in \ZZ^2}
q^{\frac12 \left(\frac{\e}{r}+\frac{\m r}{2}\right)^2}
\bar{q}^{\frac12 \left(\frac{\e}{r}-\frac{\m r}{2}\right)^2}
\end{equation}
When we are interested in the dependence of the partition function on the radius of the target circle, we will write it as a function of two arguments $Z(\tau,r)$.

\begin{lemma}[Properties of $Z(\tau)$]
The torus partition function (\ref{l29 Z(tau) explicitly}) satisfies the following properties.
\begin{enumerate}[(a)]
\item \label{l29 lemma (a)} Modular invariance:
\begin{eqnarray}\label{l29 Z modular invariance 1}
Z(\tau+1)&=& Z(\tau),\\
Z(-1/\tau) &=& Z(\tau). \label{l29 Z modular invariance 2}
\end{eqnarray}
\item \label{l29 lemma (b)} ``T-duality'':
\begin{equation}\label{l29 Z T-duality}
Z(\tau,r)=Z(\tau,2/r)
\end{equation}
%where we are indicating explicitly the radius of the target circle as the second argument of $Z$.
%$Z(\tau,r)$ is the partition function (\ref{l29 Z(tau) explicitly})
\item \label{l29 lemma (c)} Large-radius asymptotics
\begin{equation}\label{l29 Z large r asymptotics}
Z(\tau,r)\underset{r\ra \infty}{\sim} r \frac{1}{\sqrt{\mr{Im}(\tau)}\,\eta(\tau)\eta(\bar\tau)}
\end{equation}
\end{enumerate}
\end{lemma}

Modular invariance (\ref{l29 Z modular invariance 1}), (\ref{l29 Z modular invariance 2}) means that the genus one partition function  belongs to $C^\infty(\Pi_+)^{PSL_2(\ZZ)}$, i.e., descends to a smooth function on the moduli space of complex tori $\mc{M}_{1,0}$ -- which is the general feature expected in any CFT, cf. Section \ref{sss: modular invariance}.

``T-duality'' (or ``target-space duality'') is a term originating in string theory. %``T'' stands for ``target.'' 
T-duality means that there is an equivalence of sigma-models with target a circle of radius $r$ and target a circle of radius $2/r$.

Property (\ref{l29 Z large r asymptotics}) means in particular that if we think of the scalar field with target $\RR$ as a limit of the scalar field with target $S^1$ of radius $r$, as $r\ra\infty$, we are seeing explicitly how the partition function diverges (as the volume of the target). This gives us a better understanding of the claim made in the very beginning of Section \ref{s free scalar field with values in S^1} that the genus one partition function of the $\RR$-valued free scalar theory diverges.

\begin{proof}
Item (\ref{l29 lemma (a)}) is proven by Poisson summation in $\e,\m$.
%\marginpar{add the detailed computation?}

For the item (\ref{l29 lemma (b)}), we notice that the exponents in (\ref{l29 Z(tau) explicitly}) satisfy
\begin{equation}
h_{\e,\m}(r)=h_{\m,\e}(2/r),\; \bar{h}_{\e,\m}(r)=\bar{h}_{\m,\e}(2/r)
\end{equation}
where we indicate explicitly the dependence of the exponents (conformal weights of the pseudovacuum $|\e,\m\rangle$) on $r$. From this observation, the equality (\ref{l29 Z T-duality}) is obvious. (Interestingly, the inversion of the target radius $r\mapsto 2/r$ is compensated by the interchange of the electric and magnetic numbers $(\e,\m)\mapsto (\m,\e)$.)

For the item (\ref{l29 lemma (c)}) one applies Poisson summation
%also follows from Poisson summation formula, 
just in the variable $\e$ to (\ref{l29 Z(tau) explicitly}): one has
\begin{equation}\label{l29 Z as a sum over p and m}
Z(\tau,r)=%\frac{1}{\eta(\tau)\eta(\bar\tau)}\sum_{(\e,\m)\in \ZZ^2} q^{h_{\e,\m}} \bar{q}^{\bar\alpha_{\e,\m}} =
\frac{1}{\eta(\tau)\eta(\bar\tau)}\sum_{(p,\m)\in \ZZ^2}
\frac{r}{\sqrt{\mr{Im}(\tau)}} e^{-\frac{\pi^2}{2}r^2 \left(
\frac{(p+\m \mr{Re}(\tau))^2}{\mr{Im}(\tau)}+\m^2
\right)},
\end{equation}
where we denoted $p$ the dual variable to $\e$ (w.r.t. Poisson summation). In the sum (\ref{l29 Z as a sum over p and m}), the asymptotics as $r\ra \infty$ is given by the term $p=\m=0$ (and it is the r.h.s. of (\ref{l29 Z large r asymptotics})), while the sum of all other terms is exponentially suppressed -- it behaves as $O(r e^{-A r^2})$ with some constant $A>0$.

\end{proof}

As mentioned above, T-duality (\ref{l29 Z T-duality}) extends to an equivalence of free boson CFTs corresponding to target radii $r$ and $2/r$. In particular, one has an isomorphism of the respective spaces of states:
%At the level of spaces of states, this amounts to an isomorphism
\begin{equation}
\begin{array}{ccc}
\HH_r & \simeq & \HH_{2/r} \\
%\wh{a}_{-n_r}\cdots \wh{a}_{-n_1} \wh{\ol{a}}_{-\bar{n}_s}\cdots \wh{\ol{a}}_{-\bar{n}_1}
W
|\e,\m\rangle_r & \mapsto &
%\wh{a}_{-n_r}\cdots \wh{a}_{-n_1} \wh{\ol{a}}_{-\bar{n}_s}\cdots \wh{\ol{a}}_{-\bar{n}_1}
W|\m,\e\rangle_{2/r}
\end{array}
\end{equation}
where $W$ is any word in creation operators.

\subsection{Path integral approach to the torus partition of the free scalar field with values in $S^1$}
In this part we follow K. Gawedzki \cite{Gawedzki_IAS}, we refer the reader to this source for more details.

In the path integral approach, the partition function of the torus $\Sigma=\mathbb{T}_\tau$ is represented by the integral over smooth maps $\phi\colon \Sigma \ra  S^1_\mr{target}$:
\begin{equation}\label{l29 Z PI}
Z^\mr{PI}(\Sigma)=\int_{\mr{Map}(\Sigma,S^1)} \mc{D}\phi\; e^{-S(\phi)}
\end{equation}
with $S(\phi)$ the classical (Euclidean) action 
%(\ref{l18 S_Eucl}).
of the model,
\begin{equation}\label{l29 S}
S(\phi)=\frac{1}{8\pi} \int_\Sigma d\phi \wedge * d\phi =
\frac{1}{8\pi} \int_\Sigma dt d\sigma ((\dd_t \phi)^2+(\dd_\sigma \phi)^2).
\end{equation}

Note that we have $\pi_0 \mr{Map}(\Sigma,S^1_\mr{target})\simeq \ZZ^2$. More specifically, maps $\phi$ fall into classes of homotopy equivalent maps, according to the pair of winding numbers $(n_1,n_2)\in \ZZ^2$ of $\phi$ around two closed curves $\gamma_{1,2}\subset \Sigma$ -- the generators of $\pi_1(\Sigma)$.

%\textcolor{red}{PICTURE}
\begin{figure}[H]
%$$\vcenter{\hbox{ \includegraphics[scale=0.8]{l26_contours_1.eps} }} 
%\longrightarrow
%\vcenter{\hbox{ \includegraphics[scale=0.8]{l26_contours_2.eps} }} 
%$$
\begin{center}
\includegraphics[scale=0.85]{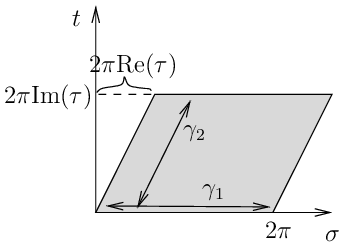}
\end{center}
\caption{Torus with modular parameter $\tau$ with two generators of $\pi_1$.}
\end{figure}

 Thus, the mapping space breaks into connected components
\begin{equation}
\mr{Map}(\Sigma,S^1_\mr{target})=\bigsqcup_{(n_1,n_2)\in \ZZ^2} \mr{Map}_{n_1,n_2}(\Sigma,S^1_\mr{target})
\end{equation}
where $\mr{Map}_{n_1,n_2}$ consists of maps with prescribed winding numbers $n_1,n_2$. Therefore, we can rewrite (\ref{l29 Z PI}) as
\begin{equation}\label{l29 Z PI 2}
Z^\mr{PI}(\Sigma)=\sum_{(n_1,n_2)\in\ZZ^2} 
\int_{\mr{Map}_{n_1,n_2}(\Sigma,S^1)} \mc{D}\phi\; e^{-S(\phi)}
\end{equation}

Notice that for each pair $(n_1,n_2)\in \ZZ$ there exists a unique (up to a constant shift) solution of the Euler-Lagrange equation $\Delta\phi=0$ with winding numbers $(n_1,n_2)$. Explicitly it can be represented by the function
\begin{equation}\label{l29 phi^cl}
\phi_{n_1,n_2}^\mr{cl}(\sigma,t)=r\cdot \left(
n_1\sigma + \frac{n_2-n_1 \mr{Re}(\tau)}{\mr{Im}(\tau)}t
\right).
\end{equation}
Note that it is a linear function in coordinates $\sigma,t$ on the torus.
The classical action evaluated on the classical solution (\ref{l29 phi^cl}) is
\begin{equation}\label{l29 S(phi^cl)}
S(\phi_{n_1,n_2}^\mr{cl})=\frac{\pi r^2}{2} \frac{|n_2-\tau n_1|^2}{\mr{Im}(\tau)}
\end{equation}

A general smooth map $\phi\in \mr{Map}_{n_1,n_2}(\Sigma,S^1_\mr{target})$ can be uniquely decomposed as
\begin{equation}\label{l29 phi decomposition}
\phi= \phi_0 +\phi_{n_1,n_2}^\mr{cl}+\til\phi,
\end{equation} 
where
\begin{itemize}
\item $\phi_0$ is a constant function valued in $S^1_\mr{target}$ (the constant shift of a classical solution),
\item $\phi_{n_1,n_2}^\mr{cl}$ is the ``standard'' classical solution with given winding numbers (\ref{l29 phi^cl}),
\item the ``fluctuation'' $\til\phi$ is a smooth function with no winding  (i.e. lifting to a function $\Sigma\ra \RR$) and satisfying the condition 
\begin{equation}\label{l29 zero integral condition}
\int_\Sigma dt\,d\sigma\,\til\phi=0
\end{equation} 
(this condition is imposed to have uniqueness of the decomposition (\ref{l29 phi decomposition})). We denote the space of maps $\phi\colon \Sigma\ra \RR$ satisfying (\ref{l29 zero integral condition}) by $\mr{Map}'(\Sigma,\RR)$ (it is the orthogonal complement of constant maps).
\end{itemize}
Note that the first two terms in (\ref{l29 phi decomposition}) together give the general classical solution with given winding numbers. Substituting the decomposition (\ref{l29 phi decomposition}) into the action (\ref{l29 S}), we obtain
\begin{equation}
S(\phi)= S(\phi_{n_1,n_2}^\mr{cl})+S(\til\phi).
\end{equation}
Thus, the path integral (\ref{l29 Z PI 2}) is
\begin{equation}\label{l29 Z PI 3}
Z^\mr{PI}(\Sigma)=\sum_{(n_1,n_2)\in \ZZ^2} \underbrace{\oint_{S^1_\mr{target}} d\phi_0}_{2\pi r} \underbrace{\int_{\mr{Map}'(\Sigma,\RR)}\mc{D}\til\phi\; e^{-S(\til\phi)}}_{(\det'\Delta_\Sigma)^{-\frac12}} \cdot e^{-S(\phi_{n_1,n_2}^\mr{cl})}
\end{equation}
The integral over $\phi_0\in S^1_\mr{target}$ here is the integral over the space of classical solutions. The Gaussian functional integral in the middle is formally evaluated to the determinant-prime (i.e. excluding the zero eigenvalue) of the Laplacian on $\Sigma$ raised to the power $-\frac12$, cf. Section \ref{sss free scalar in PI formalism}. This determinant can be calculated explicitly in the sense of zeta-function regularization (this is a rather nontrivial computation for which we refer the reader again to Gawedzki \cite{Gawedzki_IAS}), yielding
\begin{equation}
{\det}' \Delta_\Sigma =(2\pi)^2 \mr{Im}(\tau)\, |\eta(\tau)|^4,
\end{equation}
where the Dedekind eta-function makes an appearance. Thus, continuing (\ref{l29 Z PI 3}), we have
\begin{equation}
Z^\mr{PI}(\Sigma)=\sum_{(n_1,n_2)\in \ZZ^2} \cancel{2\pi} r \frac{1}{\cancel{2\pi} \sqrt{\mr{Im}(\tau)} |\eta(\tau)|^2} e^{-\frac{\pi r^2}{2} \frac{|n_2-\tau n_1|^2}{\mr{Im}(\tau)}}
\end{equation}
This expression coincides with result of the operator formalism in the form (\ref{l29 Z as a sum over p and m})!

To see this coincidence, we identify $n_1$ with $\m$ (which is not surprising, since $\m$ was the winding number along the fixed-time circle) and $n_2$ with $p$ (i.e., the second winding number gets identified with the Poisson-dual variable to $\e$ -- the zero-mode momentum).

Ultimately, we obtained a check that the operator formalism of CFT (relying on the study of the space of states) and the path integral formalism yield the same answer for the genus one partition function.

We remark that in the path integral formalism, the modular invariance of the torus partition function is manifest (unlike the operator formalism where it is a nontrivial consequence of Poisson summation). Indeed, the values of the action evaluated on  classical solutions (\ref{l29 S(phi^cl)}) on the tori $\Sigma=\mathbb{T}_\tau$ and $\Sigma'=\mathbb{T}_{-1/\tau}$ are the same (if one identifies the winding numbers as $(n_1,n_2)\leftrightarrow (n_2,-n_1)$). Likewise, the eigenvalue spectra of Laplacians on on $\Sigma,\Sigma'$ are the same, and hence the determinants are the same. Put another way, in the path integral formalism modular invariance is manifest, because the classical (Lagrangian) theory is conformally invariant.\footnote{In this form this argument is a bit formal  and implicitly assumes conformal invariance of the path integral measure.}
%\marginpar{This classical argument sounds a bit naive.. say something about the possible anomaly?}

\marginpar{Lecture 30,\\ 11/7/2022 }

\section{Aside: conformal blocks}
In a general CFT on a surface $\Sigma$ (e.g. $\Sigma=\CC$ or $\CP^1$), for a collection of fields $\Phi_1,\ldots,\Phi_n\in V$  one is interested in writing the correlator as a sum of products of holomorphic and antiholomorphic functions
\begin{equation}\label{l30 conf blocks}
\langle \Phi_1(z_1)\cdots \Phi_n(z_n) \rangle = \sum_{\rho \in I(\Phi_1,\ldots,
\Phi_n)} F_\rho (z_1,\ldots, z_n) F'_\rho(\bar{z}_1,\ldots,\bar{z}_n).
\end{equation}
Here 
\begin{itemize}
\item The correlator in the l.h.s. is a smooth single-valued function on the open configuration space $C_n(\Sigma)$.
\item In the r.h.s. the index $\rho$ ranges over some set $I(\Phi_1,\ldots,\Phi_n)$ depending on the input fields (in nice cases it is a finite set, but generally does not have to be). 
\item $F_\rho$, $F'_\rho$ are respectively holomorphic and antiholomorphic (possibly multivalued\footnote{
In particular, $F_\rho, F'_\rho$ are allowed to have monodromy as one puncture goes around another one. Put another way, $F_\rho,F'_\rho$ are single-valued holomorphic/antiholomorphic functions on some covering space of $C_n(\Sigma)$.
})  functions on $C_n(\Sigma)$; they are called the ``conformal blocks'' for the correlator in the l.h.s. of (\ref{l30 conf blocks}). %Functions $F_\rho, F'_\rho$ are allowed to be multivalued (in particular allowed to have monodromy as one puncture goes around another one).
\end{itemize}

Similarly, the genus one partition function can be written as
\begin{equation}\label{l30 Z(tau) conf blocks}
Z(\tau)=\sum_{\rho\in I_{1,0}} \chi_\rho(\tau) \chi'_\rho(\bar\tau)
\end{equation}
with  $\chi_\rho,\chi'_\rho$ -- the ``conformal blocks for the torus partition function'' -- respectively holomorphic and antiholomorphic multivalued functions on the moduli space $\MM_{1,0}$.

\subsection{Chiral (holomorphic) free boson with values in $S^1$}
\label{ss chiral boson}
Consider the version of the compactified free boson theory where one only considers one copy of the Heisenberg algebra (generated by $\wh{a}_n$, but not $\wh{\ol{a}}_n$), and the space of states is the sum of Verma modules for this single Heisenberg algebra:
\begin{equation}
\HH^\mr{chiral}=\bigoplus_{(\e,\m)\in \ZZ^2} \mathbb{V}^\mr{Heis}_{\e,\m}  = \mr{Span}\Big\{\wh{a}_{-k_r}\cdots \wh{a}_{-k_1}|\e,\m\rangle\; \Big|\; (\e,\m)\in \ZZ^2,\; 1\leq k_1\leq\cdots \leq k_r \Big\}
\end{equation}

In this model, one can consider the chiral vertex operator
\begin{equation}\label{l30 chiral vertex operator}
\wh{V}^\mr{chiral}_{\e,\m}(z)= :e^{i\alpha_{\e,\m}\wh\chi(z)}:,
\end{equation}
with $\wh\chi(z)$ as in (\ref{l28 chi, chibar}) -- the ``holomorphic part'' of the field operator $\wh\phi(z)$. The expression (\ref{l30 chiral vertex operator}) should be thought of as the ``holomorphic half'' of the vertex operator (\ref{l28 vertex operator}) of the full (non-chiral) theory.

From Wick's lemma one obtains correlators
\begin{equation}\label{l30 corr of chiral vertex operators}
\langle \prod_{i=1}^n {V}^\mr{chiral}_{\e_i,\m_i}(z_i) \rangle = 
\left\{\begin{array}{cl}
\prod_{1\leq i<j \leq n} (z_i-z_j)^{\alpha_{\e_i\m_i}\alpha_{\e_j\m_j}}, & \mr{if}\; \sum \e_i =\sum \m_i=0,\\
0 ,& \mr{otherwise}
\end{array}
\right.
\end{equation}
This expression is holomorphic and multivalued (has monodromies) on $C_n(\CC)$. If the radius of the target circle satisfies $r^2\in \mathbb{Q}$, then the monodromies are rational and the correlator lifts to as a single-valued function on a finite-degree covering space of $C_n(\CC)$.

We remark that multivaluedness of correlators is linked to the fact that the conformal weights of chiral vertex operators $(h,\bar{h})=(\frac12 \alpha_{\e,\m}^2,0)$ fail the assumption (\ref{l25 h-hbar in Z}). In the chiral theory the antiholomorphic stress-energy tensor vanishes identically $\ol{T}=0$ and any field has $\bar{h}=0$, so $V$ and $\HH$ are graded just by the holomorphic conformal weight $h$.

The correlator (\ref{l28 <V..V>}) of vertex operators in the full compactified free boson theory factorizes as the correlator of holomorphic chiral vertex operators (\ref{l30 corr of chiral vertex operators}) times the correlator of (analogous) antiholomorphic chiral vertex operators:\footnote{In the antiholomorphic chiral theory, one only retains the creation/annihilation operators $\wh{\ol{a}}_n$, all fields have conformal weight of the form $(0,\bar{h})$ and $T=0$. Correlators are antiholomorphic and multivalued.}
\begin{equation}
\langle \prod_{i=1}^n {V}^\mr{non-chiral}_{\e_i,\m_i}(z_i) \rangle =
\langle \prod_{i=1}^n {V}^\mr{chiral}_{\e_i,\m_i}(z_i) \rangle\cdot 
\langle \prod_{i=1}^n \ol{V}^\mr{chiral}_{\e_i,\m_i}(\bar{z}_i) \rangle
\end{equation}
Comparing with (\ref{l30 conf blocks}), we can say that correlators of holomorphic/antiholmorphic chiral vertex operators in the respective chiral compactified free boson theories yield the conformal blocks for the correlator of vertex operators in the full (non-chiral) compactified free boson theory. In particular, in this example the indexing set $I$ of (\ref{l30 conf blocks}) is a single-element set.

The genus one partition function (\ref{l29 Z(tau) explicitly}) of the compactified free boson admits the representation (\ref{l30 Z(tau) conf blocks}) with $I_{1,0}$ a finite set if and only if the target radius satisfies $r^2\in \mathbb{Q}$.

For example, for $r=\sqrt{2}$ (the so-called self-dual radius, since it is a stationary point of T-duality (\ref{l29 Z T-duality})), one has
\begin{equation}
Z(\tau)=\Big( \frac{1}{\eta(\tau)}\sum_{k\in \ZZ}q^{k^2} \Big)\Big( \frac{1}{\eta(\bar\tau)}\sum_{l\in \ZZ}\bar{q}^{l^2} \Big) +
\Big( \frac{1}{\eta(\tau)}\sum_{k\in \ZZ+\frac12}q^{k^2} \Big)\Big( \frac{1}{\eta(\bar\tau)}\sum_{l\in \ZZ+\frac12}\bar{q}^{l^2} \Big)
\end{equation}
I.e., here $I_{1,0}$ is a 2-element set: one has two holomorphic and two antiholomorphic conformal blocks.

\section{Free fermion}

\subsection{Classical Lagrangian theory on a surface}
As a Lagrangian field theory, 2d free fermion on a Riemannian surface $\Sigma$ is defined by the classical action
\begin{equation}\label{l30 S free fermion}
S=\frac{i}{4\pi}\int_\Sigma \ppsi \bar{\boldsymbol\dd} \ppsi- \bar\ppsi \boldsymbol\dd \bar\ppsi = \frac{1}{2\pi}\int_\Sigma d^2z (\psi \bar\dd \psi+\bar\psi \dd \bar\psi)
\end{equation}
Here 
%$\kappa=\frac{1}{4\pi}$ is a convenient normalization factor, 
$\bdd=dz\, \dd$, $\bar\bdd =\bar{dz}\,\bar\dd$ are the holomorphic/antiholomorphic Dolbeault differentials, $z,\bar{z}$ refers to a local complex coordinate on $\Sigma$ and $d^2z=\frac{i}{2}dz\wedge d\bar{z}$ is the coordinate area element. The fields of the model are fermions (spinors)\footnote{One understands $\psi,\bar\psi$ as two independent fields.}
\begin{equation}\label{l30 psi}
\ppsi = \psi (dz)^{1/2}\in \Gamma (\Sigma, K^{\otimes \frac12}),\qquad 
\bar\ppsi =\bar\psi (d\bar{z})^{1/2} \in \Gamma(\Sigma, \ol{K}^{\otimes \frac12}).
\end{equation}
Here $K,\ol{K}$ are the line bundles $(T^{1,0})^*\Sigma, (T^{0,1})^*\Sigma$. Two important points:
\begin{itemize}
\item To define the square root of these line bundles, one needs to choose the sign of the root of the transition function. This choice of sign is known as the spin structure on $\Sigma$.\footnote{Put another way, it is a choice of a consistent set of periodicity/antiperiodicity conditions for the fermion field $\psi,\bar\psi$, as one traverses a closed curve $\gamma$ on $\Sigma$.}
\item One treats the values of the fields $\psi,\bar\psi$ as \emph{anticommuting} (or ``odd'' or ``Grassmann'') variables.
\end{itemize}
Thus, the space of fields of the model is (purely odd) vector superspace
\begin{equation}
\F_\Sigma= \bigoplus_{s}\Gamma(\Sigma, \Pi  K_s^{\otimes \frac12}\oplus \Pi\ol{K}_s^{\otimes \frac12}))
\end{equation}
where the sum is over the spin structures $s$ on $\Sigma$;\footnote{Generally, spin structures form a torsor over $H^1(\Sigma,\ZZ_2)$, thus there are $2^{B_1}$ spin structures on a surface with first Betti number $B_1$.} 
%\marginpar{refer to Cimasoni-Reshetikhin?}
$\Pi$ is the parity reversal symbol, implying that $\F_\Sigma$ is the space of sections of a supervector bundle with purely odd fiber.

The Euler-Lagrange equation for the action reads
\begin{equation}
\bar\bdd \ppsi=0,\quad \bdd\bar\ppsi =0,
\end{equation}
or equivalently, in a local complex coordinate,
\begin{equation}
\bar\dd \psi=0,\quad \dd\bar\psi =0.
\end{equation}

\begin{remark}\label{l30 rem: Majorana and chiral fermions}
The system described by the action functional (\ref{l30 S free fermion}), with fields $\psi,\bar\psi$ is called the free Majorana  fermion.\footnote{Majorana fermion is ``uncharged'' as opposed to Dirac fermion, which is ``charged'' -- possesses an extra $U(1)$-symmetry $\psi\ra e^{i\theta}\psi$. Majorana and Dirac fermions are also referred to as ``real'' and ``complex'' fermions, respectively.}
 One can also consider the system with only field $\psi$ (or only $\bar\psi$), with the action $S_\mr{chiral}=\frac{1}{2\pi}\int_\Sigma d^2 z\, \psi\bar\dd\psi$ (respectively, $\frac{1}{2\pi}\int_\Sigma d^2 z\, \bar\psi\dd\bar\psi$) -- it is called the chiral or Weyl fermion. When one wants to distinguish between the chiral fermion $\psi$ and the chiral fermion $\bar\psi$, they are called respectively left- and right-chiral fermions.
\end{remark}

\subsection{Hamiltonian picture}
As a Hamiltonian theory on a cylinder, the model has phase space -- the purely odd vector superspace
\begin{equation}
\Phi = \bigsqcup_{s\in \{P,A\}} C^\infty_s (S^1)\otimes \CC^{0|2}
\end{equation}
where $\CC^{0|2}$ is another notation for the odd two-dimensional complex space $\Pi \CC^2$; $s\in \{P,A\}$ is a choice of spin structure on the cylinder -- a choice of either periodic (P) or antiperiodic (A) boundary conditon. Elements of $\Phi$ are pairs $(\psi,\bar\psi)$ of functions on $S^1$ satisfying simultaneously either P or A condition, 
\begin{equation}
\psi(\sigma+2\pi)=\epsilon \psi(\sigma), \quad \bar\psi(\sigma+2\pi)=\epsilon \bar\psi(\sigma)
\end{equation}
with $\epsilon=+1$ is $s=P$ and $\epsilon=-1$ if $s=A$.

One has the symplectic form on the phase space,
\begin{equation}
\omega = \frac{i}{4\pi} \oint_{S^1} d\sigma \Big( \delta\psi\wedge \delta \psi+\delta \bar\psi\wedge \delta \bar\psi \Big)\qquad \in \Omega^2(\Phi).
\end{equation}

The corresponding Poisson (anti-)brackets\footnote{Instead of being skew-symmetric, they are symmetric} are
\begin{equation}\label{l30 Poisson brackets}
\{\psi(\sigma),\psi(\sigma')\} = 2\pi i \delta(\sigma-\sigma'),\; \{\bar\psi(\sigma),\bar\psi(\sigma')\} = 2\pi i \delta(\sigma-\sigma'),\; 
\{\psi(\sigma),\bar\psi(\sigma')\} = 0.
\end{equation}

The Hamiltonian of the model is 
\begin{equation}
H=\frac{i}{4\pi} \oint_{S^1} d\sigma\, (\psi\dd_\sigma \psi - \bar\psi\dd_\sigma\bar\psi).
\end{equation}
It is obtained by writing the action functional on the Minkowski cylinder 
\begin{equation}
S_\mr{Mink}=\int dt \underbrace{\frac{i}{4\pi}\oint d\sigma (\psi \dd_t \psi -\psi \dd_\sigma \psi + \bar\psi \dd_t \bar\psi +\bar\psi\dd_\sigma \bar\psi)}_{\mathsf{L}}
\end{equation}
and performing the Legendre transform. Since the Lagrangian $\mathsf{L}$ is linear in velocities $\dot\psi(\sigma)$, the corresponding momenta $\pi(\sigma)=\frac{\delta}{\delta \dot\psi(\sigma)}\mathsf{L}=-\frac{i}{4\pi}\psi(\sigma)$ are not independent and drop out of the Legendre transform.

%\marginpar{remark on Legendre transform and the system being first-order}

One can expand the fields in Fourier modes,
\begin{equation}
\psi(\sigma)=\sum_{n} e^{-in\sigma} b_n,\quad \bar\psi(\sigma)=\sum_n e^{-in\sigma} \bar{b}_n
\end{equation}
where $n$ ranges over integers if $s=P$ and over half-integers ($n\in\ZZ+\frac12$) if $s= A$. 
%Let us for future convenience introduce the notation $\ZZ_P\colon=\ZZ,\quad \ZZ_A\colon=\ZZ+\frac12.$
Poisson brackets (\ref{l30 Poisson brackets}) imply to following Poisson brackets for the Fourier modes:
\begin{equation}\label{l30 Poisson brackets for modes}
\{b_n,b_m\}=i\delta_{n,-m},\quad \{\bar{b}_n,\bar{b}_m\}=i\delta_{n,-m},\quad \{b_n,\bar{b}_m\}=0.
\end{equation}

\subsection{Canonical quantization}\label{ss fermion can quantization}
Proceeding to canonical quantization, one replaces coordinates $b_n$, $\bar{b}_n$ on the phase space with operators $\wh{b}_n,\wh{\bar{b}}_n$ acting on some space of states $\HH$ (to be described), subject to the following anticommutation relations
%\marginpar{Expand on canonical quantization prescription?} 
(obtained from (\ref{l30 Poisson brackets for modes}) by the canonical quantization prescription):
\begin{equation}\label{l30 b_n anticommutation relations}
[\wh{b}_n,\wh{b}_m]_+ = \delta_{n,-m} \wh{\mathbb{1}}, \quad [\wh{\bar{b}}_n,\wh{\bar{b}}_m]_+ = \delta_{n,-m} \wh{\mathbb{1}},\quad
[\wh{b}_n,\wh{\bar{b}}_m]_+ = 0,
\end{equation}
where $[A,B]_+\colon=AB+BA$ is the anticommutator.

\begin{remark} \label{l30 rem Cliff}
Generally, given a vector space $W$ with an inner product $g$, one can form the Clifford algebra $\mr{Cl}(W,g)$ -- the associative unital algebra generated by the elements of $W$ subject to the relation
\begin{equation}
uv+vu = g(u,v) \mathbb{1}
\end{equation}
for any $u,v\in W$.
%\begin{equation}
%\mr{Cl}(W,(,))\colon = \CC \langle W \rangle
%\end{equation}
Then, the algebra spanned by the operators $\wh{b}_n$ above (with $n\in \ZZ$ for $s=P$ and $n\in \ZZ+\frac12$ for $s=A$) is the Clifford algebra for the vector space $W=C^\infty_s(S^1)$ with inner product $g(u,v)=\oint d\sigma u(\sigma) v(\sigma)$.\footnote{
Or, more invariantly, one should set $W=\Gamma(S^1,(T^*S^1)^{\otimes \frac12}_s)$ -- the space of half-densities on $S^1$ with periodicity condition $s\in \{P,A\}$. Then one has $g(u,v)=\oint \mathbf{u}\mathbf{v}$ for $\mathbf{u},\mathbf{v}\in W$ two half-densities. 
} 
Thus, the Clifford algebra for $W$ plays a similar role in the free fermion theory to the role of the Weyl algebra in the free boson theory. We will denote these two Clifford algebras $\mr{Cl}_s$ with $s\in\{P,A\}$:
\begin{equation}\label{l30 Cl_P, Cl_A}
\begin{gathered}
\mr{Cl}_P=\CC\langle \ldots,\wh{b}_{-1},\wh{b}_0,\wh{b}_1,\ldots \rangle\Big/(\wh{b}_n\wh{b}_m+\wh{b}_m\wh{b}_n=\delta_{n,-m}\wh{\mathbb{1}}),\\
\mr{Cl}_A=\CC\langle \ldots,\wh{b}_{-3/2},\wh{b}_{-1/2},\wh{b}_{1/2},\wh{b}_{3/2}\ldots \rangle\Big/(\wh{b}_n\wh{b}_m+\wh{b}_m\wh{b}_n=\delta_{n,-m}\wh{\mathbb{1}})
\end{gathered}
\end{equation}
\end{remark}

The Heisenberg field operator on the cylinder is
\begin{equation}\label{l30 psihat cylinder}
\wh{\psi}(\zeta)=\sum_{n\in \ZZ_s} e^{-n\zeta} \wh{b}_n, \quad 
\wh{\bar{\psi}}(\zeta)=\sum_{n\in \ZZ_s} e^{-n\bar\zeta} \wh{\bar{b}}_n,
\end{equation}
where $\zeta=t+i\sigma$, with $t$ the Euclidean time, and where we denoted $\ZZ_P\colon=\ZZ,\quad \ZZ_A\colon=\ZZ+\frac12$.

Mapping from the cylinder to the punctured plane by $\exp\colon \CC/2\pi i\ZZ\ra \CC\backslash\{0\}$, $\zeta\mapsto z=e^\zeta$, we have
$
\psi_\mr{plane} (z)(dz)^{\frac12}= \psi_\mr{cyl}(\zeta) (d\zeta)^{\frac12}
$
and thus
\begin{equation}\label{l30 cyl->plane 1}
\psi_\mr{plane}(z)=\underbrace{z^{-\frac12}}_{\left(\frac{d z}{d\zeta}\right)^{-\frac12}} \psi_\mr{cyl}(\zeta)
\end{equation}
where the power of derivative is minus the power of $K$ in (\ref{l30 psi}), cf. also (\ref{l25 primary field finite transf}).
Similarly, one has
\begin{equation}\label{l30 cyl->plane 2}
\bar\psi_\mr{plane}(z)=\bar{z}^{-\frac12} \bar\psi_\mr{cyl}(\zeta).
\end{equation}

By this reasoning, Heisenberg field operators on the cylinder (\ref{l30 psihat cylinder}) mapped to the punctured plane become
\begin{equation}
\wh\psi(z)=\sum_{n\in \ZZ_s} \wh{b}_n z^{-n-\frac12},\quad 
\wh{\bar\psi}(z)=\sum_{n\in \ZZ_s} \wh{\bar{b}}_n \bar{z}^{-n-\frac12}
\end{equation}
where the $-\frac12$ shift in the exponent comes from (\ref{l30 cyl->plane 1}), (\ref{l30 cyl->plane 2}).

Periodic boundary condition (P) on the cylinder ($\psi_\mr{cyl}(\sigma+2\pi)=\psi_\mr{cyl}(\sigma)$) maps to the antiperiodic condition on the plane, 
\begin{equation}
\psi_\mr{plane}(e^{2\pi i} z)=e^{-\frac12 2\pi i}\psi_\mr{plane}(z)=- \psi_\mr{plane}(z),
\end{equation}
i.e., when travelling along a closed simple contour around zero, the field $\psi_\mr{plane}(z)$ changes sign. This spin structure on $\CC\backslash\{0\}$ (or ``sector'' of the phase space/space of states) is called  ``Ramond sector.'' Thus, in P or Ramond sector one has
$\wh\psi(z)=\sum_{n\in \ZZ} \wh{b}_n z^{-n-\frac12}$ and similarly for $\wh{\bar\psi}(z)$.

Similarly, antiperiodic condition (A)  on the cylinder becomes periodic condition on the plane, $\psi_\mr{plane}(e^{2\pi i}z)=+\psi_\mr{plane}(z)$. This is the so-called ``Neveu-Schwarz spin structure/sector.'' Thus, in A or Neveu-Schwarz sector one has
$\wh\psi(z)=\sum_{n\in \ZZ+\frac12} \wh{b}_n z^{-n-\frac12}$ and similarly for $\wh{\bar\psi}(z)$.

\marginpar{Lecture 31,\\
11/9/2022}

\subsection{Space of states for the chiral fermion}\label{ss space of states for chiral fermion}

Let us restrict our attention to the chiral fermion $\psi$, %(i.e. consider the system without $\bar\psi$), 
cf. Remark \ref{l30 rem: Majorana and chiral fermions}. 

The space of states splits into P- and A-sectors:
\begin{equation}
\HH=\HH_P\oplus \HH_A
\end{equation}
with $\HH_P$ a highest weight $\mr{Cl}_P$-module (cf. (\ref{l30 Cl_P, Cl_A})) generated by the highest vector $|\vac_P\rangle$ satisfying $\wh{b}_{>0} |\vac_P\rangle=0 $.
Similarly, $\HH_A$ a highest weight $\mr{Cl}_A$-module generated by the highest vector $|\vac_A\rangle$ satisfying $\wh{b}_{>0} |\vac_A\rangle=0 $. Thus, one has
\begin{equation}
\HH_P=\mr{Span}\Big\{\cdots \wh{b}_{-2}^{p_2}\wh{b}_{-1}^{p_1}\wh{b}_0^{p_0}|\vac_P\rangle\; \Big| \; 
\begin{array}{c}
p_0,p_1,p_2,\ldots \in \{0,1\}, \\
\mr{finitely\;many\;}p_n\mr{\;are\;nonzero}
\end{array}
  \Big\}
\end{equation}
Fermionic occupation numbers $p_0,p_1,\ldots$ are in $\{0,1\}$ since from the anticommutation relations (\ref{l30 b_n anticommutation relations}) one has $(\wh{b}_n)^2=0$ for $n\neq 0$ and $(\wh{b}_0)^2=\frac12\wh{\mathbb{1}}$. Similarly, one has
\begin{equation}
\HH_A=\mr{Span}\Big\{\cdots \wh{b}_{-5/2}^{p_{5/2}}\wh{b}_{-3/2}^{p_{3/2}}\wh{b}_{1/2}^{p_{1/2}}|\vac_A\rangle\; \Big| \; 
\begin{array}{c}
p_{1/2},p_{3/2},p_{5/2},\ldots \in \{0,1\}, \\
\mr{finitely\;many\;}p_n\mr{\;are\;nonzero}
\end{array}
  \Big\}
\end{equation}

\subsection{2-point function $\langle \psi \psi \rangle$}
Tu understand which of the Clifford highest vectors $|\vac_P\rangle$, $|\vac_A\rangle$ is the true vacuum of the system,
let us calculate the correlation function $\langle \psi(w) \psi(z) \rangle$ in the operator formalism.  Assume for simplicity $|w|>|z|>0$.

In the P-sector we have
\begin{multline}\label{l31 <psi psi>_P beginning}
\langle \psi(w)\psi(z) \rangle_P\colon = \langle \vac_P | \wh\psi(w) \wh\psi(z) |\vac_P\rangle =
\sum_{n,m\in \ZZ} \langle \vac_P | \wh{b}_n \wh{b}_m |\vac_P\rangle w^{-n-\frac12}z^{-m-\frac12}
\end{multline}
From the fact that $\wh{b}_{>0}|\vac_P\rangle=0$, $\langle\vac_P| \wh{b}_{<0}=0$ and from the anticommutation relation (\ref{l30 b_n anticommutation relations}) we see that the only surviving terms are $n=m=0$ and $n=-m>0$, i.e., one has
\begin{multline}\label{l31 <psi psi>_P}
\langle \psi(w)\psi(z) \rangle_P= 
\underbrace{\langle\vac_P| \wh{b}_0\wh{b}_0 |\vac_P\rangle}_{\frac12} w^{-\frac12} z^{-\frac12}+\sum_{n=1}^\infty \underbrace{\langle\vac_P| \wh{b}_n \wh{b}_{-n} |\vac_P \rangle}_{1} w^{-n-\frac12}z^{n-\frac12} =\\= (wz)^{-\frac12}(\frac12+\sum_{n=1}^\infty\left(\frac{z}{w}\right)^n)=
\frac12 \frac{\left(\frac{w}{z}\right)^{\frac12}+\left(\frac{z}{w}\right)^{\frac12}}{w-z}
\end{multline}

By a similar computation, in the A-sector we have
\begin{multline}\label{l31 <psi psi>_A}
\langle \psi(w)\psi(z) \rangle_A= \sum_{n\in \ZZ+\frac12,\,n>0} \underbrace{\langle \vac_A| \wh{b}_n\wh{b}_{-n}|\vac_A\rangle}_1 w^{-n-\frac12} z^{n-\frac12} =\\
=\frac{1}{w}+\frac{z}{w^2}+\frac{z^2}{w^3}+\cdots = \frac{1}{w-z}
\end{multline}

Note that the expression (\ref{l31 <psi psi>_A}) is translation-invariant as expected of a 2-point correlator in any CFT (cf. Lemma \ref{l26 lemma 2-point}),  while (\ref{l31 <psi psi>_P}) is not translation-invariant. This suggests that we should identify $|\vac\rangle\colon=|\vac_A\rangle$ as the true vacuum vector in $\HH$, while $|\vac_P\rangle$ is a pseudovacuum (similar to the states $|\pi_0\rangle$ with $\pi_0\neq 0$ in the scalar field theory). In particular, (\ref{l31 <psi psi>_A}) should be understood as the actual 2-point correlator 
\begin{equation}\label{l31 <psi psi>}
\langle \psi(w) \psi(z) \rangle=\frac{1}{w-z}
\end{equation} 
On the other hand, the computation (\ref{l31 <psi psi>_P}) should be understood as a 4-point correlator on $\CP^1$,\footnote{The field $\sigma(\infty)$ here is with respect to the coordinate chart at $\infty\in \CP^1$. Writing this correlator in terms of $\CC$, and using the result from further along this section that $\sigma$ has conformal weight $(\frac{1}{16},0)$, one should write
$\lim_{y\ra \infty} y^{\frac18} \langle \sigma(y) \psi(w)\psi(z)\sigma(0) \rangle_\CC$.
}
%\marginpar{correct this formula by scaling correctly the field at $\infty$ with $z^{2h}$.}
\begin{equation}
\langle \sigma(\infty)\psi(w)\psi(z)\sigma(0)\rangle_{\CP^1},
\end{equation} 
with a certain field $\sigma$ (so-called ``twist field,'' to be discussed later), corresponding by field-state correspondence to $|\vac_P\rangle$, inserted at the points $0$ and $\infty$. This explains why we don't see translation invariance in (\ref{l31 <psi psi>_P}) -- because ``secretly'' it is a 4-point function and a translation would displace the field $\sigma$ away from the origin.

%As we will discuss later, under the field-state correspondence, the vector $|\vac_P\rangle$ corresponds to a certain field $\sigma(z)$.

In the free fermion model, the space of states $\HH$ and the space of fields $V$ are $\ZZ_2$-graded and we understand that when the radial ordering is applied, we have a sign when we have to permute field operators:
\begin{equation}
\mc{R} \wh{\Phi}_1(w) \wh{\Phi}_2(z) = \left\{
\begin{array}{cc}
\wh{\Phi}_1(w)\wh{\Phi}_2(z), & \mr{if}\; |w|\geq |z| \\
(-1)^{|\Phi_1|\cdot |\Phi_2|} \wh{\Phi}_2(z)\wh{\Phi}_2(w), & \mr{if}\; |z|\geq |w|
\end{array}
\right.
\end{equation}
Here $|\Phi|\in \ZZ_2$ is the parity of the field. With this prescription, for instance, the computation (\ref{l31 <psi psi>_A}) extends to the case $|w|\leq |z|$, yielding the same formula:
\begin{equation}\label{l31 <psi psi> 2}
\langle \phi(w)\psi(z) \rangle\colon = \langle \vac_A| \mc{R} \wh{\psi}(w) \wh\psi(z) |\vac_A \rangle = \frac{1}{w-z}.
\end{equation}
with any $w\neq z\in \CC\backslash\{0\}$.

Note that the 2-point function (\ref{l31 <psi psi> 2})  satisfies 
\begin{equation}
\langle \psi(w) \psi(z)\rangle = - \langle \psi(z) \psi(w) \rangle
\end{equation}
-- the correlation function is antisymmetric under swapping the positions of fermions (as expected in Fermi statistics).

\subsection{Stress-energy tensor}
Classically, the stress-energy tensor (computed as a variation of the action w.r.t. metric) for the chiral fermion is
\begin{equation}\label{l31 T_cl}
T(z)=-\frac12 \psi(z) \dd\psi(z)
\end{equation}
for the holomorphic component and $\ol{T}(z)=0$ for the antiholomorphic component. 

For the corresponding quantum object -- an operator on $\HH$, we consider separately $\wh{T}(z)$ as an operator on $\HH_A$ and on $\HH_P$.

\subsubsection{A-sector.} Set 
\begin{equation}\label{l31 T in A-sector by normal ordering}
\wh{T}(z)\colon= -\frac12:\wh\psi(z) \dd \wh\psi(z):
\end{equation}
where the normal ordering puts fermion annihilation operators $\wh{b}_{>0}$ to the right and fermion creation operators $\wh{b}_{<0}$ to the left; we understand that when we interchange two $\wh{b}$'s, the sign of the expression is flipped.

From Wick's lemma (or rather its obvious adaptation to the Clifford algebra) we find the standard OPE
\begin{equation}\label{l31 TT}
\mc{R}\wh{T}(w) \wh{T}(z) = \frac{\frac{c}{2}\wh{\mathbb{1}}}{(w-z)^4}+\frac{2\wh{T}(z)}{(w-z)^2} + \frac{\dd \wh{T}(z)}{w-z}+\mr{reg.}
\end{equation}
(cf. (\ref{l23 TT})) with holomorphic central charge $c=\frac12$. Since $\ol{T}=0$, the $T\ol{T}$ and $\ol{T}\ol{T}$ OPEs are satisfied trivially, with antiholomorphic central charge $\bar{c}=0$.

\subsubsection{P-sector.} Set
\begin{equation}
\wh{T}^\mr{naive}(z)\colon= -\frac12:\wh\psi(z) \dd \wh\psi(z):
\end{equation}
with the same definition of normal ordering as above. Interestingly, it does not satisfy the expected OPE (\ref{l23 TT}), thus it fails a basic axiom of a CFT (in particular its modes do not satisfy the Virasoro algebra relations). It turns out that a good definition is as follows:
\begin{equation}\label{l31 T via subtraction of sing part of OPE}
\wh{T}(z)\colon=\lim_{w\ra z}\left(-\frac12 \mc{R} \wh\psi(w)\dd\wh\psi(z)+\frac12 \frac{\wh{\mathbb{1}}}{(w-z)^2}\right)
\end{equation}
-- we split the two points in the definition of the stress-energy tensor (\ref{l31 T_cl}) and subtract the (translation-invariant) singular part of OPE, $-\frac12 \psi(w)\dd\psi(z)-[-\frac12 \psi(w)\dd\psi(z)]_\mr{sing}$. Then one has\footnote{
Indeed, repeating the computation (\ref{l31 <psi psi>_P beginning}), (\ref{l31 <psi psi>_P}), without pairing to $|\vac_P\rangle$, we have $\mc{R}\wh\psi(w)\wh\psi(z)-:\wh\psi(w)\wh\psi(z):=\frac12 \frac{\left(\frac{w}{z}\right)^{\frac12}+\left(\frac{z}{w}\right)^{\frac12}}{w-z}\wh{\mathbb{1}} =(\frac{1}{w-z}+ \frac{1}{8z^2} (w-z)+O((w-z)^2))\wh{\mathbb{1}}$. Hence, $-\frac12\mc{R}\wh\psi(w)\dd\wh\psi(z)=-\frac12:\wh\psi(w)\dd\wh\psi(z):+(-\frac12 \frac{1}{(w-z)^2}+\frac{1}{16z^2}+O(w-z))\wh{\mathbb{1}}$, 
or equivalently $-\frac12\mc{R}\wh\psi(w)\dd\wh\psi(z)+\frac12 \frac{\wh{\mathbb{1}}}{(w-z)^2}=-\frac12:\wh\psi(w)\dd\wh\psi(z):+\frac{\wh{\mathbb{1}}}{16z^2}+O(w-z)$. Taking the limit $w\ra z$, we obtain (\ref{l31 T = T^naive + 1/16}).
}
\begin{equation}\label{l31 T = T^naive + 1/16}
\wh{T}(z)= \wh{T}^\mr{naive}(z)+\frac{\wh{\mathbb{1}}}{16z^2}
\end{equation}
-- with this $\frac{\wh{\mathbb{1}}}{16z^2}$ shift included, $\wh{T}$ does satisfy the desired OPE (\ref{l31 TT}), again with $c=\frac12$.

%\marginpar{remove?}
In particular, we have nonzero expectation value of the stress-energy tensor in P-sector
\begin{equation}
\langle T(z) \rangle_P\colon=\langle \vac_P| \wh{T}(z) |\vac_P\rangle=\frac{1}{16z^2}.
\end{equation}
%thus the pseudovacuum $|\vac_P\rangle$ has nonzero energy.

We remark that in A-sector, prescription (\ref{l31 T via subtraction of sing part of OPE}) is compatible with the construction via normal ordering (\ref{l31 T in A-sector by normal ordering}). Thus, (\ref{l31 T via subtraction of sing part of OPE}) can be taken as a universal recipe for the fermion stress-energy tensor (applies to both A- and P-sector).

\subsubsection{Virasoro generators.} Virasoro generators can be obtained from the  stress-energy tensor $\wh{T}(z)=\sum_{n\in \ZZ} z^{-n-2}\wh{L}_n$. Thus, from (\ref{l31 T in A-sector by normal ordering}) and (\ref{l31 T = T^naive + 1/16}) one obtains:
\begin{equation}\label{l31 L_n via b_n}
\begin{array}{cl}
\mbox{A-sector:} &\displaystyle \wh{L}_n=\sum_{m\in \ZZ+\frac12} \left(\frac{m}{2}+\frac14\right) :\wh{b}_{n-m}\wh{b}_m: ,\\
\mbox{P-sector:} &\displaystyle  \wh{L}_n=\sum_{m\in \ZZ} \left(\frac{m}{2}+\frac14\right) :\wh{b}_{n-m}\wh{b}_m:+\delta_{n,0}\frac{\wh{\mathbb{1}}}{16}.
\end{array}
\end{equation}
All operators $\ol{L}_n$ vanish identically.

In particular, one has
\begin{equation}
\wh{L}_0 |\vac_A\rangle =0,\qquad \wh{L}_0 |\vac_P\rangle = \frac{1}{16} |\vac_P\rangle
\end{equation}
In particular, the true vacuum $|\vac_A\rangle$ has zero energy and total momentum, while $|\vac_P\rangle$ has both energy and total momentum $\frac{1}{16}$.

One also has 
\begin{equation}
[\wh{L}_0,\wh{b}_{-n}]=n \wh{b}_n
\end{equation}
in both A- and P-sectors.
I.e., applying $\wh{b}_{-n}$, one increases the $\wh{L}_0$-eigenvalue (conformal weight) by $n$.

\subsection{Back to the space of states}
Let us list the states in A- and P-sectors with small conformal weights $h$ (i.e., $\wh{L}_0$-eigenvalues).

%\ul{A-states}

\begin{tabular}{c|l}
$h$ & state \\ \hline
$0$& $|\vac_A\rangle$ \\
$\frac12$ & $\wh{b}_{-\frac12}|\vac_A\rangle$\\
$1$ & $\varnothing$ \\
$\frac32$ & $\wh{b}_{-\frac32}|\vac_A\rangle$ \\
$2$ & $\wh{b}_{-\frac32}\wh{b}_{-\frac12}|\vac_A\rangle$ \\
$\frac52$ & $\wh{b}_{-\frac52}|\vac_A\rangle$\\
$3$ & $\wh{b}_{-\frac52}\wh{b}_{-\frac12}|\vac_A\rangle$\\
$\cdots$& $\cdots$
\end{tabular}  \qquad
\begin{tabular}{c|l}
$h$ & state \\ \hline
$\frac{1}{16}$ & $|\vac_P\rangle$, $\wh{b}_0|\vac_P\rangle$ \\
$1+\frac{1}{16}$ & $\wh{b}_{-1}|\vac_P\rangle$, $\wh{b}_{-1}\wh{b}_0|\vac_P\rangle$ \\
$2+\frac{1}{16}$ & $\wh{b}_{-2}|\vac_P\rangle$, $\wh{b}_{-2}\wh{b}_0|\vac_P\rangle$ \\
$3+\frac{1}{16}$ & $\wh{b}_{-3}|\vac_P\rangle$, $\wh{b}_{-3}\wh{b}_0|\vac_P\rangle$, \\
& $\wh{b}_{-2}\wh{b}_{-1}|\vac_P\rangle$, $\wh{b}_{-2}\wh{b}_{-1}\wh{b}_0|\vac_P\rangle$ \\
$\cdots$&$\cdots$
\end{tabular}

Here we have states in A-sector on the left and states in P-sector on the right.

States 
\begin{equation}\label{l31 four primary states}
|\vac_A\rangle, \quad \wh{b}_{-\frac12}|\vac_A\rangle,\quad |\vac_P\rangle,\quad \wh{b}_0 |\vac_P\rangle
\end{equation}
are Virasoro-primary (annihilated by $\wh{L}_{>0}$) -- and they are the only Virasoro-primary states in $\HH$. We will also denote these four states according to their conformal weight by $|0\rangle$, $|\frac12 \rangle$, $|\frac{1}{16}\rangle_+$, $|\frac{1}{16}\rangle_-$. Their $\ZZ_2$-grading is, respectively, even, odd, even, odd.\footnote{The logic with $\ZZ_2$ grading is that vectors $|\vac_A\rangle$, $|\vac_P\rangle$ are even, while action by any single Clifford generator $\wh{b}$ changes the parity of the vector.} 

Thus, the space of states of the chiral fermion splits into four conformal families (irreducible representations of Virasoro algebra):
\begin{equation}
\HH= \underbrace{M_{0}\oplus \Pi M_{\frac12}}_{\HH_A}\oplus \underbrace{ M_{\frac{1}{16}}\oplus \Pi M_{\frac{1}{16}} }_{\HH_P}
\end{equation}
where $M_h$ is the irreducible Virasoro highest weight module with central charge $\frac12$ and highest weight $h$; $\Pi$ is the parity reversal symbol (i.e. $M_h$ is an even vector space and $\Pi M_h$ is an odd (super)vector space).

By the field-state correspondence, the four primary states (\ref{l31 four primary states}) correspond to four primary fields
\begin{equation}
\mathbb{1},\quad \psi(z),\quad \sigma(z),\quad \mu(z)
\end{equation} 
%In particular $|\vac_A\rangle$ corresponds to the identity field $\mathbb{1}$.
%\begin{tabular}{c|cccc}
%primary field & $\mathbb{1}$ & $\psi(z)$ & $\sigma(z)$ & $\mu(z)$ \\
%$(h,\bar{h})$ & $(0,0)$ & $(\frac12,0)$ & $(\frac{1}{16},0)$ &
%$(\frac{1}{16},0)$
%\end{tabular}
with conformal weight $h$ being $0,\frac12,\frac{1}{16},\frac{1}{16}$, respectively (and $\bar{h}=0$ for all fields in the chiral theory). Fields $\sigma,\mu$ are the so-called ``twist fields.'' One has for instance the OPE
\begin{equation}
\psi(w) \sigma(z)\sim (w-z)^{-\frac12}\mu(z)+\mr{reg}.
\end{equation}
%\marginpar{Is it phrased ok?}
In particular, the insertion of the twist field $\sigma(z)$ creates a monodromy $-1$ around $z$ for the fermion $\psi(w)$.

\marginpar{Lecture 32,\\
11/11/2022}

\subsection{Non-chiral (Majorana) fermion}
We pair the left- and right- (or holomorphic/antiholomorphic) chiral fermion CFTs, with the following conventions:
\begin{itemize}
\item We require that the P/A boundary condition is the same for $\psi$ and $\bar\psi$.
\item We impose $\wh{b}_0=\wh{\bar{b}}_0$ (cf. (\ref{l17 a_0})).
\end{itemize}

The space of states splits as a sum of irreducible highest weight modules of $\mr{Vir}\oplus \ol{\mr{Vir}}$ with central charge $c=\bar{c}=\frac12$:
\begin{equation}\label{l32 H non-chiral fermion}
\HH^\mr{non-chiral}=\underbrace{M_{0,0}\oplus \Pi M_{\frac12,0}\oplus \Pi M_{0,\frac12} \oplus M_{\frac12,\frac12}}_{\HH^\mr{non-chiral}_A} \oplus \underbrace{M_{\frac{1}{16},\frac{1}{16}} \oplus \Pi  M_{\frac{1}{16},\frac{1}{16}}}_{\HH^\mr{non-chiral}_P}
\end{equation}
where the two indices of $M$ are the highest weight (conformal weight) $(h,\bar{h})$ of the highest vector. The highest weight vectors themselves and the corresponding primary fields are, respectively: 

\vspace{0.3cm}
\hspace{-1cm}
\begin{tabular}{c|cccccc}
highest vector & $|\vac_A\rangle$ & $\wh{b}_{-\frac12} |\vac_A\rangle$& $\wh{\bar{b}}_{-\frac12} |\vac_A\rangle$ & $\wh{b}_{-\frac12}\wh{\bar{b}}_{-\frac12}|\vac_A\rangle$ & $|\vac_P\rangle$& $\wh{b}_0 |\vac_P\rangle$ \\
primary field & $\mathbb{1}$ & $\psi(z)$ & $\bar\psi(z)$ & $\epsilon(z)=\psi(z)\bar\psi(z)$ & $\bbsigma(z)$ & $\bbmu(z)$ \\
$(h,\bar{h})$ & $(0,0)$ & $(\frac12,0)$ & $(0,\frac12)$ & $(\frac12,\frac12)$ & $(\frac{1}{16},\frac{1}{16})$ & $(\frac{1}{16},\frac{1}{16})$ \\
$\ZZ_2$-parity & even & odd & odd & even & even & odd
\end{tabular}
%We use a different font for the $(\frac{1}{16},0)$-field $\sigma$ in the chiral fermion CFT and the $(\frac{1}{16},\frac{1}{16})$-field $\bbsigma$ in the non-chiral (full) fermion CFT, and similarly for the field $\mu$ vs $\bbmu$.

%are, respectively,
%\begin{equation}\label{l32 primary states}
%|\vac_A\rangle, \quad \wh{b}_{-\frac12} |\vac_A\rangle,\quad \wh{\bar{b}}_{-\frac12} |\vac_A\rangle, \quad \wh{b}_{-\frac12}\wh{\bar{b}}_{-\frac12}|\vac_A\rangle,\quad |\vac_P\rangle,\quad \wh{b}_0 |\vac_P\rangle
%\end{equation}
%and the corresponding primary fields are
%\begin{equation}\label{l32 primary fields}
%\underset{(0,0)}{\mathbb{1}},\quad \underset{(\frac12,0)}{\psi(z)},\quad \underset{(0,\frac12)}{\bar\psi(z)}, \quad \underset{(\frac12,\frac12)}{\epsilon(z)=\psi(z)\bar\psi(z)},\quad \underset{(\frac{1}{16},\frac{1}{16})}{\bbsigma(z)}, \quad \underset{(\frac{1}{16},\frac{1}{16})}{\bbmu(z)},
%\end{equation}
%where we again indicated the conformal weight $(h,\bar{h})$ under each field.
%\footnote{We use a different font for the $(\frac{1}{16},0)$-field $\sigma$ in the chiral fermion CFT and the $(\frac{1}{16},\frac{1}{16})$-field $\bbsigma$ in the non-chiral (full) fermion CFT, and similarly for the field $\mu$ vs $\bbmu$.} 
%Parities of the fields (\ref{l32 primary fields}) or correspondingly states (\ref{l32 primary states}) are: even, odd, odd,

\begin{remark} Free Majorana fermion is the CFT model corresponding to the Ising model at critical temperature, see \cite{BPZ} and \cite{DMS} for a detailed discussion. In particular, correlation functions of the spin field in Ising model can be recovered as correlation functions of the field $\bbsigma$ in the free fermion CFT.
\end{remark}

\subsection{Examples of correlators}\label{sss free fermion correlators}
From the computation (\ref{l31 <psi psi>_A}) we know the 2-point correlator
\begin{equation}
\langle \psi(w) \psi(z)\rangle = \frac{1}{w-z}.
\end{equation}
The correlator of any number of fields $\psi$, $\bar\psi$ can be computed by Wick's lemma, as a sum over perfect matchings (where one needs to be careful with signs incurred when moving $\hat\psi$ over other $\hat\psi$'s.) For the correlator of several $\psi$ fields, this sum over perfect matchings can written as a Pfaffian formula
\begin{equation}
\langle \psi(z_1)\cdots \psi(z_n) \rangle = \left\{
\begin{array}{cl}
\mr{Pf}\left(\frac{1}{z_i-z_j}\right) & \mr{if}\; n\; \mr{is\;even},\\
0 & \mr{if}\; n\; \mr{is\; odd}
\end{array}
 \right.
\end{equation}

For example, for $n=4$ one has
\begin{equation}
\langle \psi(z_1) \psi(z_2) \psi(z_3) \psi(z_4) \rangle = \frac{1}{z_{12}z_{34}}-\frac{1}{z_{13}z_{24}} +\frac{1}{z_{14}z_{23}},
\end{equation}
where $z_{ij}=z_i-z_j$.

The 2-point correlator $\langle \bbsigma(w) \bbsigma(z) \rangle$ cannot be found from Wick's lemma (we don't have an explicit description of the field $\bbsigma$ in terms of Clifford generators $\wh{b}_n,\wh{\bar{b}}_n$ at our disposal), however we have an ansatz for it from global conformal symmetry, cf. Lemma \ref{l26 lemma 2-point}:
\begin{equation}\label{l32 <sigma sigma>}
\langle \bbsigma(w) \bbsigma(z) \rangle = C\frac{1}{(w-z)^{\frac{1}{16}+\frac{1}{16}}}\cdot \frac{1}{(\bar{w}-\bar{z})^{\frac{1}{16}+\frac{1}{16}}} = C\frac{1}{|w-z|^\frac14}
\end{equation}
with $C$ some constant. By choosing a convenient normalization for the field $\sigma$, we can assume $C=1$.\footnote{This normalization agrees with the convention that the  state corresponding to $\bbsigma$, $|\vac_P\rangle$, has unit norm $\langle \vac_P|\vac_P \rangle=1$, cf. (\ref{l26 lemma 2-point (c) eq}).} 

The exponent $\frac14$ in (\ref{l32 <sigma sigma>}) is exactly the one appearing in the spin-spin correlator in Ising model at critical temperature (as known from the explicit solution of 2d Ising model), thus corroborating the free fermion-Ising correspondence.

\subsubsection{4-point correlator of $\bbsigma$ fields.}
%$\langle \bbsigma\bbsigma\bbsigma\bbsigma \rangle$.
As the next example, consider the 4-point function of $\bbsigma$ fields. From global conformal invariance (cf. Lemma \ref{l27 lemma: n-point fun from global conformal invariance}) one has
\begin{equation}\label{l32 sigma corr ansatz}
\langle \bbsigma(z_1) \bbsigma(z_2) \bbsigma(z_3) \bbsigma(z_4) \rangle = \left| \frac{z_{13} z_{24}}{z_{12}z_{23}z_{34}z_{41}}\right|^{\frac14} F(\lambda),
\end{equation}
where $F(\lambda)$ is some smooth function of the cross-ratio $\lambda=\frac{z_{12}z_{34}}{z_{13}z_{24}}\in \CP^1\backslash \{0,1,\infty\}$. To fix the function $F$, we need some other idea than just global conformal invariance.

In the free fermion theory one has a vanishing descendant of the state $|\vac_P\rangle$ at level $2$:
\begin{equation}
(\wh{L}_{-2}-\frac43 \wh{L}_{-1}^2) |\vac_P\rangle =0
\end{equation}
-- this can be verified by using the expressions (\ref{l31 L_n via b_n}) for Virasoro generators in terms of Clifford generators.\footnote{
In fact, it is true generally that in the Verma module $\mathbb{V}_{c,h}$ for the Virasoro algebra at central charge $c$ with highest weight $h$ one has a singular vector at level $2$ (cf. Remark \ref{l25 rem linear dependencies between descendants}), of the form $|\chi\rangle=(L_{-2}+\alpha L_{-1}^2)|h\rangle$, if and only if one has 
$\left| \begin{array}{cc}
3& 4h+2 \\ \frac{c}{2}+4h & 6h 
\end{array} \right| =0
$ and then $|\chi\rangle$ is a singular vector if $\alpha=-\frac{3}{4h+2}$. In particular, the pair $c=\frac12$, $h=\frac{1}{16}$ satisfies the determinant condition and gives $\alpha=-\frac43$, i.e., $(L_{-2}-\frac{4}{3}L_{-1}^2)|\frac{1}{16}\rangle$ is a singular vector in the Verma module. Thus, in the irreducible Virasoro module it has to be set to zero.
} Thus, the corresponding primary field also has a vanishing descendant:
\begin{equation}
(L_{-2}-\frac43 L_{-1}^2) \bbsigma(z)=0.
\end{equation}
Thus, by Ward identity (cf. Example \ref{l26 ex: corr of a descendant}) one has
\begin{equation}
0=\langle (L_{-2}-\frac43 L_{-1}^2) \bbsigma(z_1)\; \bbsigma(z_2) \bbsigma(z_3) \bbsigma(z_4) \rangle =\mc{D} \langle \bbsigma(z_1) \bbsigma(z_2) \bbsigma(z_3) \bbsigma(z_4) \rangle
\end{equation}
with $\mc{D}$ some differential operator in $z_i$'s. Substituting the ansatz (\ref{l32 sigma corr ansatz}), we obtain a differential equation on the function $F(\lambda)$ -- the hypergeometric equation
\begin{equation}
\left(\lambda (1-\lambda)\frac{\dd^2}{\dd \lambda^2}+(\frac12-\lambda)\frac{\dd}{\dd \lambda}+\frac{1}{16}\right) F(\lambda) = 0.
\end{equation}
This equation has two independent solutions 
\begin{equation}
f_{1,2}(\lambda)=(1\pm \sqrt{1-\lambda})^{\frac12}
\end{equation}
and the general solution has the form $f_1(\lambda) g_1(\bar\lambda)+f_2(\lambda) g_2(\bar\lambda)$ with $g_{1,2}$ some antiholomorphic functions. Using the conditions that $F$ should be a real, single valued function, fixes the solution to the form
\begin{equation}
F(\lambda) = a (f_1(\lambda) f_1(\bar\lambda)+f_2(\lambda)f_2(\bar\lambda))
\end{equation}
with $a$ a constant. Using additionally the OPE $\bbsigma(w)\bbsigma(z)\sim \frac{\mathbb{1}}{|w-z|^{\frac14}}+\cdots$ (where the normalization follows from $C=1$ in (\ref{l32 <sigma sigma>})), one obtains $a=\frac12$. Thus, putting everything together, one has
\begin{equation}\label{l32 corr of 4 sigmas}
\langle \bbsigma(z_1) \bbsigma(z_2) \bbsigma(z_3) \bbsigma(z_4) \rangle=\frac12  \left| \frac{z_{13} z_{24}}{z_{12}z_{23}z_{34}z_{41}}\right|^{\frac14}
(|1+\sqrt{1-\lambda}|+|1-\sqrt{1-\lambda}|).
\end{equation}

\subsection{Torus partition function for the Majorana fermion}\label{ss Majorana fermion Z(torus)}
Denote $(-1)^F$ the operator on the space of states with eigenvalue $+1$ on even vectors and $-1$ on odd vectors.

The partition function of the free Majorana fermion on a torus is given by (\ref{l29 Z(torus) corrected}) with the following correction: $\tr_\HH(\cdots)$ should be replaced by 
\begin{equation}\label{l32 tr str}
\tr_{\HH^\mr{even}}(\cdots)=\tr_{\mr{\HH}} \frac{1+(-1)^F}{2}(\cdots)=\frac12\left(\mr{tr}_\HH(\cdots)+\mr{Str}_\HH(\cdots)\right),
\end{equation}
where $\mr{Str}$ is the supertrace\footnote{
Generally, for a $\ZZ_2$-graded vector space $W=W^\mr{even}\oplus W^\mr{odd}$ and $A\colon W\ra W$ a linear map, the supertrace is defined as $\tr_{W^\mr{even}}A - \tr_{W^\mr{odd}}A=\mr{tr}\,A^\mr{even,even}-\mr{tr}\, A^\mr{odd,odd}$. Where in the last form we are referring to the diagonal blocks of $A$ seen as a $2\times 2$ block matrix.
} and $(\cdots)=q^{\wh{L}_0-\frac{c}{24}}\bar{q}^{\wh{\ol{L}}_0-\frac{\bar{c}}{24}}$. Recall that for the Majorana fermion, the central charge is $c=\bar{c}=\frac12$. The operator $\frac{1+(-1)^F}{2}$ is the projector to the even part of the space of states.

Averaging over trace and supertrace in (\ref{l32 tr str}) corresponds to averaging over spin structures (boundary conditions for the fermions) in the ``time direction'' on the torus.
In fact, the \emph{supertrace} is the more natural extension of the notion of trace to $\ZZ_2$-graded vector spaces (it satisfies the natural cyclicity property with Koszul sign). As we will see below, from comparison with the path integral approach, the supertrace term in the r.h.s. of (\ref{l32 tr str}) corresponds to \emph{periodic} boundary condition for the fermions in the time direction, whereas the trace term corresponds to the \emph{antiperiodic} boundary condition.

In view of (\ref{l32 tr str}), the torus partition function of the Majorana fermion is
\begin{equation}\label{l32 Z(torus) via tr, str}
Z(\tau)=(q\bar{q})^{-\frac{1}{48}} \tr_{\HH^\mr{even}} q^{\wh{L}_0} \bar{q}^{\wh{\ol{L}}_0}
= \frac12 (q\bar{q})^{-\frac{1}{48}}  \left( \tr_{\HH_A} +\mr{Str}_{\HH_A}+\tr_{\HH_P}+\underbrace{\mr{Str}_{\HH_P}}_{=0} \right) q^{\wh{L}_0} \bar{q}^{\wh{\ol{L}}_0}
\end{equation}
Here the splitting of the space of states into A- and P-parts is as in (\ref{l32 H non-chiral fermion}). The supertrace over $\HH_P$ vanishes, since for each eigenstate $\alpha\in \HH_P$ of conformal weight $(h,\bar{h})$ there is a second state $\wh{b}_0\alpha\in \HH_P$  with the same conformal weight but opposite parity.  In the supertrace $\mr{Str}_{\HH_P}$, The contributions of $\alpha$ and $\wh{b}_0\alpha$ cancel out. On the other hand, in $\tr_{\HH_P}$ such contributions enter with the same sign. 

From the description of $\HH_A,\HH_P$ as Verma modules over Clifford algebras $\mr{Cl}_A\otimes \ol{\mr{Cl}}_A$ and $\mr{Cl}_P\otimes \ol{\mr{Cl}}_P$ (cf. (\ref{l30 Cl_P, Cl_A})), we can write explicit formulae for the terms of (\ref{l32 Z(torus) via tr, str}):
\begin{multline}\label{l32 Z(torus) via products}
Z(\tau)=\frac12 (q\bar{q})^{-\frac{1}{48}} 
\left(
\prod_{n\geq 1}(1+q^{n-\frac12})(1+\bar{q})^{n-\frac12}+
\prod_{n\geq 1}(1-q^{n-\frac12})(1-\bar{q})^{n-\frac12}+\right. \\
+ \left.
2 (q\bar{q})^{\frac{1}{16}}\prod_{n\geq 1} (1+q^n)(1+\bar{q}^n)
\right).
\end{multline}
In the last term, the factor $2$ comes from the doubling mechanism described above (contributions of $\alpha$ and $\wh{b}_0\alpha$); the exponent $\frac{1}{16}$ is the eigenvalue of $\wh{L}_0,\wh{\ol{L}}_0$ for the highest vector $|\vac_P\rangle$.

\subsubsection{Aside: Jacobi triple product identity}
\begin{thm}[Jacobi]
For any $q,t\in \CC$ with $|q|<1$ and $t\neq 0$ one has the equality
\begin{equation}\label{l32 Jacobi triple product}
\prod_{n=1}^\infty (1-q^n)(1+t q^{n-\frac12}) (1+t^{-1}q^{n-\frac12})=\sum_{k=-\infty}^\infty t^k q^{\frac{k^2}{2}}.
\end{equation}
\end{thm}
The r.h.s. of (\ref{l32 Jacobi triple product}) is denoted
\begin{equation}
\theta_3(w;\tau),
\end{equation} 
where 
\begin{equation}
t=e^{2\pi i w},\;\; q=e^{2\pi i \tau}
\end{equation}
One also defines
\begin{equation}
\begin{gathered}
\theta_1(w;\tau)\colon= -i \theta_3(w+\frac12+\frac{\tau}{2};\tau)\cdot q^{\frac18}t^{\frac12},\\
\theta_2(w;\tau)\colon= \theta_3(w+\frac{\tau}{2};\tau)\cdot q^{\frac18}t^{\frac12},\\
\theta_4(w;\tau)\colon= \theta_3(w+\frac12;\tau). 
\end{gathered}
\end{equation}
The function $\theta_i$, $i=1,\ldots,4$ are known as Jacobi theta functions. Of importance to us are their values at $w=0$. We denote them
\begin{equation}
\theta_i(\tau)\colon= \theta_i(0;\tau),\quad i=1,\ldots 4.
\end{equation}

One has the following important special cases of the Jacobi triple product identity (\ref{l32 Jacobi triple product}):
\begin{equation}\label{l32 products via theta}
\begin{array}{*6{>{\displaystyle}l}
%llllll
}
\displaystyle
t=1: & \prod_{n\geq 1}(1-q^n)(1+q^{n-\frac12})^2 &=&\sum_{k\in \ZZ}q^{\frac{k^2}{2}} &=&\theta_3(\tau),\\
\displaystyle
t=-1: & \prod_{n\geq 1}(1-q^n)(1-q^{n-\frac12})^2 &=&\sum_{k\in \ZZ}(-1)^kq^{\frac{k^2}{2}} &=&\theta_4(\tau),\\
\displaystyle
t=q^{\frac12}: & 2\prod_{n\geq 1}(1-q^n)(1+q^{n})^2 &=&\sum_{k\in \ZZ}q^{\frac{k (k+1)}{2}} &=&\theta_2(\tau)\cdot q^{-\frac18},
\end{array}
\end{equation}

\subsubsection{Back to torus partition function}
Evaluating the terms of (\ref{l32 Z(torus) via products}) using the identities (\ref{l32 products via theta}), we arrive to the following expression for the torus partition function of the free Majorana fermion:
\begin{equation}\label{l32 Z(torus) answer}
Z(\tau)=\frac{1}{2|\eta(\tau)|}\left( |\theta_3(\tau)|+|\theta_4(\tau)|+|\theta_2(\tau)| \right),
\end{equation}
where $\eta(\tau)$ is the Dedekind eta function. The function (\ref{l32 Z(torus) answer}) satisfies modular invariance:
\begin{equation}
Z(\tau+1)=Z(\tau),\quad Z(-\frac{1}{\tau})=Z(\tau).
\end{equation}
This can be shown directly, from modular transformation properties of Jacobi theta functions (which in turn are proven by Poisson summation).

\begin{remark}
One can also write the expression (\ref{l32 Z(torus) answer}) in terms of just the Dedekind eta function (without theta functions):
\begin{equation}
Z(\tau)=\frac12\left( \left|\frac{\eta(\tau)^2}{\eta(\frac{\tau}{2})\eta(2\tau)}\right|^2  + \left|\frac{\eta(\frac{\tau}{2})}{\eta(\tau)}\right|^2+2 \left|\frac{\eta(2\tau)}{\eta(\tau)}\right|^2\right)
\end{equation}
\end{remark}

\subsubsection{Path integral formalism}
In the path integral formalism, the torus partition function is given by a sum over the four spin-structures on the torus:
\begin{multline}\label{l32 path integral}
Z(\tau)=\sum_{\epsilon_1=\pm 1,\; \epsilon_2=\pm 1} \underset{
\begin{array}{c}\small
\psi(\zeta+2\pi i)=\epsilon_1 \psi(\zeta),\\
\bar\psi(\zeta+2\pi i)=\epsilon_1 \bar\psi(\zeta),\\
\psi(\zeta+2\pi i \tau)=\epsilon_2 \psi(\zeta),\\
\bar\psi(\zeta+2\pi i \tau) = \epsilon_2 \bar\psi(\zeta)
\end{array}
}{\int}
\mc{D}\psi\,\mc{D} \bar\psi \; e^{-S(\psi,\bar\psi)}\\
=
\mr{Pf}_{AA}(\bar\dd) \mr{Pf}_{AA}(\dd)+\mr{Pf}_{AP}(\bar\dd) \mr{Pf}_{AP}(\dd)+
\mr{Pf}_{PA}(\bar\dd) \mr{Pf}_{PA}(\dd)+
\mr{Pf}_{PP}(\bar\dd) \mr{Pf}_{PP}(\dd)
\\
=|{\det}_{AA}(\bar\dd)|+|{\det}_{AP}(\bar\dd)|+|{\det}_{PA}(\bar\dd)|+|\underbrace{{\det}_{PP}(\bar\dd)}_0|.
\end{multline}
This is a fermionic Gaussian integral (the quadratic action $S$ is (\ref{l30 S free fermion})), which can be expressed in terms of zeta-regularized  Pfaffians of the operators $\bar\dd \;\calt \; \Gamma(\Sigma,K^{\frac12})$, $\dd \;\calt \; \Gamma(\Sigma,\ol{K}^{\frac12})$ 
%determinants of the operator $\bar\dd \;\calt \; \Gamma(\Sigma,K^{\frac12})$ 
acting on spinors on the torus with chosen spin structure. E.g., subscript $AP$ means that we consider spinors antiperiodic in the ``space'' direction ($\psi(\zeta+2\pi i)=-\psi(\zeta)$) and periodic in the ``time'' direction ($\psi(\zeta+2\pi i \tau)=+ \psi(\zeta)$).  Products of complex conjugate Pfaffians in turn become determinants.   The determinant for the PP spin structure vanishes, since in that case the operator $\bar\dd$ has a zero mode given by constant spinors.

The four terms in the r.h.s. of (\ref{l32 path integral}) correspond in the operator language to the four terms in the 
r.h.s. of (\ref{l32 Z(torus) via tr, str}).

\begin{remark}
The mapping class group of the torus (the modular group) $PSL_2(\ZZ)$ acts on the spin structures on the torus. This action has two orbits: $\{PP\}$ and $\{AA,AP,PA\}$. More explicitly, in terms of the standard generators $T\colon \tau \mapsto \tau+1$, $S\colon \tau\mapsto - 1/\tau$ of $PSL_2(\ZZ)$, one has the following action on spin structures:
\begin{equation}
S\;\;\calt \;\;AA \stackrel{T}{\longleftrightarrow} AP \stackrel{S}{\longleftrightarrow} PA \;\;\car\;\; T,\qquad
S\;\; \calt\;\; PP\;\; \car\;\; T.
\end{equation}
Symbolically denoting the contributions of the four spin structures to the path integral (\ref{l32 path integral}) by $Z_{AA},Z_{AP},Z_{PA},Z_{PP}$ we have that the general modular invariant linear combination is
\begin{equation}
C_1(Z_{AA}+Z_{AP}+Z_{PA})+C_2 Z_{PP},
\end{equation}
with $C_{1,2}$ arbitrary constants. The actual partition function we are computing has $C_1=C_2=1$ (and $C_2$ is in fact irrelevant, since $Z_{PP}=0$).
\end{remark}

\subsection{Bosonization and Dirac fermion}
Bosonization is a mechanism allowing one to compute correlators of the free real (Majorana) fermion by reducing the problem to correlators in the free boson theory. This is particularly useful, since not all correlators can be computed from Wick's lemma (e.g. the correlators of the twist fields), whereas in the free boson theory all correlators are computable via Wick's lemma.

Roughly, the idea is that the system of two Majorana fermions (with $c=\frac12$ each) is equivalent to a single $c=1$ free boson.

\subsubsection{Dirac fermion}
The system of two Majorana fermions $\{\psi_a,\bar\psi_a\}_{a=1,2}$ is equivalent to a single Dirac (or ``complex,'' or ``charged'') fermion:
\begin{equation}
S^\mr{Dirac}(\psi_\pm,\bar\psi_\pm) = \frac{1}{\pi}\int_\Sigma d^2 z \left( \psi_- \bar\dd \psi_+ + \bar\psi_- \dd \bar\psi_+\right) = S^\mr{Majorana}(\psi_1,\bar\psi_1)+S^\mr{Majorana}(\psi_2,\bar\psi_2),
\end{equation}
where $S^\mr{Majorana}$ is the action (\ref{l30 S free fermion}) and the Dirac field is
\begin{equation}
\psi_\pm =\frac{\psi_1\pm i \psi_2}{\sqrt{2}},\quad 
\bar\psi_\pm =\frac{\bar\psi_1\mp i \bar\psi_2}{\sqrt{2}}.
\end{equation}
We understand that $\psi_\pm (dz)^{\frac12}$ are odd sections of $K^{\frac12}$ (left Weyl spinors) and $\bar\psi_\pm (d\bar{z})^{\frac12}$ are odd sections of $\ol{K}^{\frac12}$ (right Weyl spinors). We are assuming that the spin structures for $\psi_{1,2}$ are synchronized (thus, the field $\psi^\pm$ satisfies either periodic or antiperiodic  condition around a puncture). 

The space of states of Dirac fermion is
\begin{equation}\label{l32 H^Dirac via H^Majorana}
\HH^\mr{Dirac}=\HH_A^\mr{Majorana}\otimes \HH_A^\mr{Majorana}\oplus \HH_P^\mr{Majorana}\otimes \HH_P^\mr{Majorana},
\end{equation}
where the factors in each summand correspond to $\psi_1,\psi_2$, cf. (\ref{l32 H non-chiral fermion}).

\subsubsection{$U(1)$-current}\label{sss Dirac U(1) current}
Dirac fermion CFT contains ``Dirac $U(1)$-current''\footnote{
At the level of classical field theory, it is the Noether current associated with the $U(1)$-symmetry of the theory $\psi_\pm(z)\mapsto e^{\pm i\alpha} \psi_\pm(z)$, with $e^{i\alpha}$ a constant phase (in fact, it is also a symmetry for $\alpha$ a holomorphic function).
} -- the holomorphic $(1,0)$-field
\begin{equation}
j(z)=:\psi_+(z)\psi_-(z):=-i \psi_1(z)\psi_2(z)
\end{equation}
satisfying the OPE
\begin{equation}
j(w) j(z)\sim \frac{1}{(w-z)^2}+\mr{reg.}
\end{equation}
similar to the OPE satisfied by the field $i\dd\phi$ in the free boson theory. By Lemma \ref{l25 lemma: algebra of mode operators}, modes operators of the field $j$ satisfy the Heisenberg Lie algebra relations.  Similarly, one has a complex conjugate field $\bar{j}=:\bar\psi_+\bar\psi_-:$. 

Jointly, modes of $j$ and $\bar{j}$ endow the space of states of Dirac fermion with the structure of a 
$\mr{Heis}\oplus \ol{\mr{Heis}}$-module (again, similarly to the free boson theory).

The stress-energy tensor of the model is
\begin{equation}
T(z)=\frac12:j(z)j(z):=\frac12:\psi_+\psi_-\psi_+\psi_-: = -\frac12 :\psi_1 \dd\psi_1+\psi_2\dd\psi_2:
\end{equation}
where all fields are at $z$. 
Note that the formal substitution $j\mapsto i\dd\phi$ converts this expression into the stress-energy tensor of the free boson CFT.  
This implies that the Virasoro action on the space of states is expressed in terms of the Heisenberg action by the usual formula (\ref{l23 L_scalar via a a}), where operators $\wh{a}_n$ should be %replaced with 
understood as the
mode operators of $j$.
%Note that this expression is compatible with correspondence $j\leftrightarrow i\dd\phi$ between Dirac fermion and free boson CFTs: it 

\subsubsection{Torus partition function for Dirac fermion}
The partition function of the Dirac fermion on a torus is computed by the technique of Section \ref{ss Majorana fermion Z(torus)}:
\begin{multline}\label{l32 Z^Dirac(torus)}
Z^\mr{Dirac}(\tau)=\frac12 (\tr_{\HH_A\otimes \HH_A}+\mr{Str}_{\HH_A\otimes \HH_A}+\tr_{\HH_P\otimes \HH_P}+\underbrace{\mr{Str}_{\HH_P\otimes \HH_P}}_{0}) q^{-\frac{c}{24}+\wh{L}_0}\bar{q}^{-\frac{\bar{c}}{24}+\wh{\ol{L}}_0}=\\
=\frac{1}{2}\left(
(q\bar{q})^{-\frac{1}{24}} \prod_{n\geq 1} (1+q^{n-\frac12})^2 (1+\bar{q}^{n-\frac12})^2+
 (q\bar{q})^{-\frac{1}{24}} \prod_{n\geq 1} (1-q^{n-\frac12})^2 (1-\bar{q}^{n-\frac12})^2+\right.\\
\left. +
 4 (q\bar{q})^{\frac{1}{12}} \prod_{n\geq 1} (1+q^{n})^2 (1+\bar{q}^{n})^2
\right)\\
\underset{\mr{Jacobi\;triple\;product\;}(\ref{l32 products via theta})}{=}
\frac{1}{\eta(\tau)\eta(\bar\tau)}\cdot \frac12\sum_{(k,l)\in \ZZ^2}\Big( q^{\frac{k^2}{2}}\bar{q}^{\frac{l^2}{2}}+ 
(-1)^{k+l}q^{\frac{k^2}{2}}\bar{q}^{\frac{l^2}{2}}+\\
+
q^{\frac12 (k+\frac12)^2}\bar{q}^{\frac12 (l+\frac12)^2}+
\underbrace{(-1)^{k+l} q^{\frac12 (k+\frac12)^2}\bar{q}^{\frac12 (l+\frac12)^2}}_{=0\;\mr{as\;it\;changes\;sign\;under\;}k\ra-k-1}
\Big)\\
=\frac{1}{\eta(\tau)\eta(\bar\tau)}\sum_{(\e,\m)\in\ZZ^2} q^{\frac12 (\frac{\e}{2}+\m)^2}  \bar{q}^{\frac12 (\frac{\e}{2}-\m)^2}  
\end{multline}
In the last expression we recognize the torus partition function of the free boson with values in a circle of radius $r=2$, (\ref{l29 Z(tau) explicitly}):\footnote{
Equivalently, we could talk about the free boson on a circle of radius $r=1$: by $T$-duality (\ref{l29 Z T-duality}), $r=1$ and $r=2$ theories are equivalent.
}
\begin{equation}\label{l32 Z^Dirac=Z^boson}
Z^\mr{Dirac}(\tau)=Z^{r=2\;\mr{free\;boson}}(\tau).
\end{equation}

\subsubsection{Correspondence between Dirac fermion states and boson states}
The lattice 
\begin{equation}
\Lambda=\left\{\left(\frac{\e}{2}+\m,\frac{\e}{2}-\m\right)\right\}_{\e,\m\in \ZZ^2}\quad
\subset \RR^2
\end{equation} 
appearing in the r.h.s. of (\ref{l32 Z^Dirac(torus)}) can be split as a union of two lattices:
\begin{itemize}
\item $\Lambda_1$, with $\e$ even and any $\m$, 
\item $\Lambda_2$, with $\e$ odd and any $\m$.
\end{itemize}

One has the following refinement of the observation (\ref{l32 Z^Dirac=Z^boson}). 
\begin{thm}[Bosonization correspondence]
One has an isomorphism of $\mr{Heis}\oplus \ol{\mr{Heis}}$-modules (and, a fortiori, of
%which can in particular be seen as 
$\mr{Vir}\oplus \ol{\mr{Vir}}$-modules):
\begin{equation}\label{l32 H^Dirac=H^boson}
(\HH^\mr{Dirac})^\mr{even} \simeq \HH^{r=2\;\mr{free\;boson}}.
\end{equation}
More specifically, restricting to the summands in the r.h.s. of (\ref{l32 H^Dirac via H^Majorana}), one has isomorphisms
\begin{equation}\label{l32 H^Dirac=H^boson A,P}
\begin{gathered}
\left(\HH_A^\mr{Majorana}\otimes \HH_A^\mr{Majorana} \right)^\mr{even} \simeq \bigoplus_{(\alpha,\bar\alpha)\in \Lambda_1} V^{\mr{Heis}\oplus \ol{\mr{Heis}}}_{(\alpha,\bar\alpha)},\\
\left(\HH_P^\mr{Majorana}\otimes \HH_P^\mr{Majorana} \right)^\mr{even} \simeq \bigoplus_{(\alpha,\bar\alpha)\in \Lambda_2} V^{\mr{Heis}\oplus \ol{\mr{Heis}}}_{(\alpha,\bar\alpha)},
\end{gathered}
\end{equation}
where the terms in the r.h.s. are the Verma modules of $\mr{Heis}\oplus \ol{\mr{Heis}}$ with highest weight $(\alpha,\bar\alpha)$ (w.r.t. the operators $\wh{a}_0,\wh{\bar{a}}_0$).
\end{thm}

%In fact, (\ref{l32 H^Dirac=H^boson}) follows from 
%\begin{enumerate}[1.]
%\item 
%The fact that a highest weight module over Heisenberg Lie algebra uniquely splits as a sum of Verma modules (and the same applies to modules over $\mr{Heis}\oplus \ol{\mr{Heis}}$)
%\item The coincidence of characters of both sides of (\ref{l32 H^Dirac=H^boson}), which is 
%\end{enumerate}
%
\begin{proof}[Sketch of proof]
Both sides of (\ref{l32 H^Dirac=H^boson}) can be split as a sum of Verma modules over $\mr{Heis}\oplus \ol{\mr{Heis}}$, with highest weights in the lattice $\Lambda$ for the r.h.s. of (\ref{l32 H^Dirac=H^boson}) and in some set $S\subset \RR^2$ for the l.h.s. of (\ref{l32 H^Dirac=H^boson}):
\begin{equation}\label{l32 thm proof eq1}
(\HH^\mr{Dirac})^\mr{even}\simeq\bigoplus_{(\alpha,\bar\alpha)\in S} V^{\mr{Heis}\oplus \ol{\mr{Heis}}}_{(\alpha,\bar\alpha)}.
\end{equation}
 It suffices to show that $S=\Lambda$ (since a splitting of a 
$\mr{Heis}\oplus \ol{\mr{Heis}}$-module as a sum of Verma modules is unique).  We decompose the set $S$ into two subsets $S=S_1\sqcup S_2$ according to the contributions of $A$- and $P$-sector into (\ref{l32 thm proof eq1}).

The splitting (\ref{l32 thm proof eq1}) implies that the torus partition function of the Dirac fermion is
\begin{equation}
Z^\mr{Dirac}(\tau)= \tr_{(\HH^\mr{Dirac})^\mr{even}} q^{-\frac{1}{24}+\wh{L}_0}\bar{q}^{-\frac{1}{24}+\wh{\ol{L}}_0} = \frac{1}{\eta(\tau)\eta(\bar\tau)}\sum_{(\alpha,\bar\alpha)\in S} q^{\frac{\alpha^2}{2}}\bar{q}^{\frac{\bar\alpha^2}{2}}.
\end{equation}
Comparing with (\ref{l32 Z^Dirac(torus)}), we find that the respective sets of exponents coincide
\begin{equation}\label{l32 thm proof eq2}
\{(\frac12 \alpha^2,\frac12 \bar\alpha^2)\}_{(\alpha,\bar\alpha)\in S}=\{(\frac12 \alpha^2,\frac12 \bar\alpha^2)\}_{(\alpha,\bar\alpha)\in \Lambda}.
\end{equation}

Consider the map 
\begin{equation}
\begin{array}{cccc}
f\colon & \RR^2 & \ra & \RR^2 \\
&(\alpha,\bar\alpha)&\mapsto & (\frac12\alpha^2,\frac12 \bar\alpha^2)
\end{array}
\end{equation} 
Note that $f$ is four-to-one on $\Lambda_1$ and is two-to-one on $\Lambda_2$. To infer $S=\Lambda$ from $f(S)=f(\Lambda)$,  we need to explain this quadruple/double degeneracy on the side of the Dirac fermion.

Dirac fermion theory has the following two discrete symmetries: (left/right charge conjugation):
\begin{equation}
C_L\colon   \psi_+\longleftrightarrow \psi_-
%\begin{array}{ccc}
%\psi_+ &\longleftrightarrow& \psi_-,\\
%\bar\psi_+ &\longleftrightarrow& \bar\psi_-
%\end{array}
\qquad 
C_R\colon  \bar\psi_+ \longleftrightarrow \bar\psi_-
%\begin{array}{ccc}
%\psi_+ &\longleftrightarrow& \bar\psi_+,\\
%\psi_- &\longleftrightarrow& \bar\psi_-
%\end{array}
\end{equation}
%\marginpar{EDIT}
More precisely, the $A$-sector of $\HH^\mr{Dirac}$ is invariant under $C_L,C_R$ separately, while the $P$-sector is invariant only under the composition $C_L C_R$.

The involutions $C_{L},C_R$ act on the $A$-part of the space of states -- and in particular on the set $S_1$ -- as $(\alpha,\bar\alpha)\leftrightarrow (-\alpha,-\bar\alpha)$ and $(\alpha,\bar\alpha)\leftrightarrow (\bar\alpha,\alpha)$. Thus, each highest weight $(\alpha,\bar\alpha)$ appears in $S_1$ as a part of the quadruplet $(\pm \alpha,\pm\bar\alpha)$. Similarly, due to the action of $C_LC_R$, in $S_2$ each highest weight appears as a part of a doublet
$(\alpha,\bar\alpha),(-\alpha,-\bar\alpha)$.
 This together with (\ref{l32 thm proof eq2}) proves that $S_1=\Lambda_1$, $S_2=\Lambda_2$.
\end{proof}

We proceed to give some examples of the correspondence (\ref{l32 H^Dirac=H^boson}) at the level of fields (rather than states).
\begin{itemize}
\item 
 Informally, the odd fields $\psi_\pm(z)$ correspond to the chiral vertex operators $:e^{\pm i \chi(z)}:$ in the free boson theory, cf. (\ref{l28 chi, chibar}), (\ref{l28 vertex operator}):
 \begin{equation}
 \psi_\pm(z) \longleftrightarrow :e^{\pm i\chi(z)}:,\quad
 \bar\psi_\pm(z) \longleftrightarrow :e^{\pm i\bar\chi(z)}:
 \end{equation}
The  right hand sides here are not in fact elements of the space of fields of the compactified free boson with $r=2$ (and the elements in the l.h.s. are odd, so this example is outside of the correspondence (\ref{l32 H^Dirac=H^boson})). However \emph{even} composites of $\psi_\pm,\bar\psi_\pm$ are mapped to legitimate fields of the $r=2$ boson theory.
 \item The $U(1)$-current of the Dirac fermion CFT is mapped to the Heisenberg current of the free boson CFT:
 \begin{equation}
 j=:\psi_+\psi_-: \longleftrightarrow i\dd\phi,\quad
 \bar{j}=:\bar\psi_+\bar\psi_-: \longleftrightarrow i\bar\dd\phi.
 \end{equation}
The stress-energy tensor is mapped to the stress-energy tensor:
\begin{equation}
\frac12 :jj: \longleftrightarrow -\frac12 :\dd\phi\,\dd\phi:,
\quad
\frac12 :\bar{j}\bar{j}: \longleftrightarrow -\frac12 :\bar\dd\phi\,\bar\dd\phi:.
\end{equation}
One also has
\begin{equation}
j \bar{j} = \psi_1\bar\psi_1\psi_2\bar\psi_2 \longleftrightarrow - \dd\phi\,\bar\dd\phi.
\end{equation}
\item Consider the pair of fields
\begin{equation}\label{l32 sigma_pm}
\sigma_\pm(z) =\frac{1}{\sqrt{2}}(\bbsigma_1(z)\bbsigma_2(z)\pm i \bbmu_1(z)\bbmu_2(z))
\end{equation}
in the Dirac fermion theory. They satisfy the OPEs 
\begin{equation}
j(w) \sigma_\pm(z)=\frac{\pm 1/2}{w-z}+\mr{reg.},\quad 
\bar{j}(w) \sigma_\pm(z)=\frac{\pm 1/2}{\bar{w}-\bar{z}}+\mr{reg.}
\end{equation}
Hence, $\sigma_\pm$ are highest vectors w.r.t. $\mr{Heis}\oplus \ol{\mr{Heis}}$ with weights $(\alpha=\pm\frac12,\bar\alpha=\pm\frac12)$ (note that these two points belong to the lattice $\Lambda_2$). In the $r=2$ free boson theory these fields correspond to particular vertex operators:
\begin{equation}\label{l32 sigma to vertex operator}
\sigma_\pm(z) \longleftrightarrow :e^{\pm \frac{i\phi(z)}{2}}:=V_{\pm \frac12}(z).
\end{equation}
\end{itemize}

\begin{figure}[H]
\begin{center}
\includegraphics[scale=1]{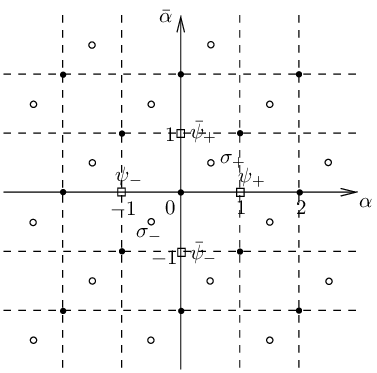}
\end{center}
\caption{Lattice $\Lambda$ of $\mr{Heis}\oplus \ol{\mr{Heis}}$ highest weights in the bosonization correspondence (\ref{l32 H^Dirac=H^boson}). Dots and circles correspond to sublattices $\Lambda_1$ and $\Lambda_2$, respectively. Four little boxes do not belong to the lattice $\Lambda$ but correspond to the fundamental (odd) fields $\psi_\pm,\bar\psi_\pm$ in the Dirac fermion CFT.}
\end{figure}

\subsubsection{Correlators of Majorana fermion theory via bosonization}
The basic idea of using the bosonization correspondence (\ref{l32 H^Dirac=H^boson}) to compute correlators in the free Majorana fermion theory is as follows. Let $\Phi(z)\in V^\mr{Majorana}_z$ be some field. We can consider product of two copies of this field (a tensor square) as a field in the Dirac fermion CFT, $\Phi_1(z)\Phi_2(z)\in V^\mr{Dirac}$, which corresponds by (\ref{l32 H^Dirac=H^boson}) to some field $X\in V^\mr{boson}$ in the free boson theory,
\begin{equation}\label{l32 Phi1 Phi2 = X}
\Phi_1(z)\Phi_2(z) \longleftrightarrow X(z).
\end{equation}
This leads to relations between correlators of the form
\begin{equation}\label{l32 bosonization for a correlator}
\left(\langle \Phi_{(1)}(z_1)\cdots \Phi_{(n)}(z_n) \rangle_\mr{Majorana}\right)^2 = \langle X_{(1)}(z_1)\cdots X_{(n)}(z_n) \rangle_\mr{boson},
\end{equation}
with $\Phi_{(i)}$ some fields in the Majorana fermion theory and $X_{(n)}$ the corresponding fields in the free boson theory.

\begin{example} For $\Phi=\psi$ the fermion field itself, the corresponding (in the sense of (\ref{l32 Phi1 Phi2 = X}) field in the boson theory  is $-\dd\phi$. The relation (\ref{l32 bosonization for a correlator}) becomes
\begin{multline}
\underbrace{\langle \psi_1(z_1)\cdots \psi_{2n}(z_{2n})  \rangle_\mr{Majorana}^2}_{=\mr{Pf\left(\frac{1}{z_i-z_j}\right)^2}} = \langle \dd\phi(z_1)\cdots \dd\phi(z_{2n}) \rangle_\mr{boson}=\\
\underset{\mr{Wick}}{=} \frac{1}{2^{2n}n!}\sum_{\pi \in S_{2n}} 
(z_{\pi(1)}-z_{\pi(2)})^{-2}\cdots (z_{\pi(2n-1)}-z_{\pi(2n)})^{-2}
\end{multline}
In this case we know the fermion correlator and the boson correlator separately from Wick's lemma and we obtain an interesting equality of rational functions on the configuration space. E.g. for $n=2$ this equality is
\begin{multline}
\left(\frac{1}{z_{12}z_{34}}-\frac{1}{z_{13}z_{24}}+\frac{1}{z_{14}z_{23}}\right)^2=\langle \psi(z_1)\cdots \psi(z_4) \rangle^2 =\\
= \langle \dd\phi(z_1)\cdots \dd\phi(z_4) \rangle=
\frac{1}{z_{12}^2z_{34}^2}+\frac{1}{z_{13}^2z_{24}^2}+\frac{1}{z_{14}^2z_{23}^2},
\end{multline}
where $z_{ij}\colon=z_i-z_j$.
\end{example}

\begin{example} For $\Phi=\bbsigma$ the twist field, the corresponding field in the boson theory is 
\begin{equation}
\frac{1}{\sqrt{2}}(V_{\frac12}(z)+V_{-\frac12}(z)),
\end{equation}
cf. (\ref{l32 sigma_pm}), (\ref{l32 sigma to vertex operator}).
Thus, the equality (\ref{l32 bosonization for a correlator}) becomes
\begin{multline}\label{l32 <sigma..sigma>^2}
\langle \bbsigma(z_1)\cdots \bbsigma(z_n) \rangle^2_\mr{Majorana}=
2^{-\frac{n}{2}} \langle (V_{\frac12}(z_1)+V_{-\frac12}(z_1))\cdots
(V_{\frac12}(z_n)+V_{-\frac12}(z_n))
 \rangle_\mr{boson}=\\
 =2^{-\frac{n}{2}}\sum_{k_1,\ldots,k_n\in \{+1,-1\},\; \mr{s.t.}\;
 k_1+\cdots+k_n=0}\;\;
 \prod_{1\leq i<j\leq n} |z_i-z_j|^{\frac{k_ik_j}{2}}.
\end{multline}
Here in the last step we used the fact that we know the correlator of a collection of vertex operators, cf. (\ref{l26 <V...V>}).

For example, for $n=2$ (\ref{l32 <sigma..sigma>^2}) yields
\begin{equation}
\langle \bbsigma(z_1)\bbsigma(z_2) \rangle_\mr{Majorana} = |z_1-z_2|^{-\frac14},
\end{equation}
cf. (\ref{l32 <sigma sigma>}). 

For $n=4$, (\ref{l32 <sigma..sigma>^2}) yields an equivalent form of the result (\ref{l32 corr of 4 sigmas}) -- but obtained from a completely different idea (bosonization vs. differential equation on the correlator arising from a null descendant).
\end{example}

Generally, bosonization allows one to determine (up to sign) any correlator in the Majorana fermion CFT via (\ref{l32 bosonization for a correlator}), reducing it to a computation by Wick's lemma of the corresponding correlator in the free boson theory.

\section{$bc$ system}\label{s bc system}
The $bc$ system (or ``reparametrization ghost system'') is a CFT classically defined on a Riemannian surface $\Sigma$ by the action functional%\marginpar{signs? normalization?}
\begin{equation}\label{l32 S_bc}
S_{bc}=\frac{i}{2\pi}\int_\Sigma -\mathbf{b} \bar\bdd \mathbf{c}+
\bar{\mathbf{b}} \bdd \bar{\mathbf{c}}
=\frac{1}{\pi}\int_\Sigma d^2 z\, (b\bar\dd c+\bar{b} \dd\bar{c}),
\end{equation}
where the fields are a $(1,0)$-vector field and a quadratic differential
\begin{equation}
\mb{c}=c\dd_z \in \Gamma(\Sigma,\underbrace{K^{-1}}_{T^{1,0}}),
\quad \mb{b} = b (dz)^2 \in \Gamma (\Sigma, K^{\otimes 2})
\end{equation}
and their antiholomorphic counterparts
\begin{equation}
\bar{\mb{c}}=\bar{c} \dd_{\bar{z}} \in \Gamma(\Sigma, \underbrace{\ol{K}^{-1}}_{T^{0,1}}),\quad \bar{\mb{b}}=\bar{b} (d\bar{z})^2 \in \Gamma (\Sigma, \ol{K}^{\otimes 2}).
\end{equation}
Fields $b,c,\bar{b},\bar{c}$ are understood as odd (anticommuting). Since no fractional powers of $K$ appear in the definition of the fields, there is no choice of a spin structure/boundary condition involved.

\begin{comment}
One proceeds with canonical quantization as in the other free models we considered, arriving to the Heisenberg fields on $\CC\backslash\{0\}$,
\begin{equation}
\wh{c}(z)=\sum_{n\in \ZZ} \wh{c}_n z^{-n+1}, \quad 
\wh{b}(z)=\sum_{n\in \ZZ} \wh{b}_n z^{-n-2}
\end{equation}
with operators $\wh{c}_n,\wh{b}_n$ subject to the anticommutation relations
\begin{equation}
[\wh{b}_n,\wh{c}_m]_+= \delta_{n,-m} \wh{\mathbb{1}},\quad
[\wh{b}_n,\wh{b}_m]_+=0,\quad 
[\wh{c}_n,\wh{c}_m]_+=0.
\end{equation}
One has similar mode expansions and anticommutation relations for $\bar{b},\bar{c}$.

\marginpar{edit/correct: what are exactly cr/annih operators here?}
%Since no fractional powers of $K$ appear in the definition of the fields, there is no choice of a spin structure/boundary condition involved. 
So, the space of states $\HH$ is generated by a unique vacuum vector $|\vac\rangle$ (annihilated by $\wh{b}_{>0},\wh{c}_{>0}, \wh{\bar{b}}_{>0}, \wh{\bar{c}}_{>0}$), by acting on it repeatedly with creation operators  $\wh{b}_{\leq 0},\wh{c}_{\leq 0}, \wh{\bar{b}}_{\leq 0}, \wh{\bar{c}}_{\leq 0}$.
\end{comment}

It is easier to analyze the model in the path integral formalism.
%From a computation in the operator formalism, one finds 
One finds the 2-point function 
\begin{equation}\label{l32 <bc>}
\langle b(w) c(z) \rangle = \frac{1}{w-z}
\end{equation}
as the Green's function for the operator $\frac{1}{\pi}\bar\dd$.
Similarly, by the method of Section \ref{sss free scalar in PI formalism} one finds the OPE
\begin{equation}
b(w) c(z) \sim \frac{\mathbb{1}}{w-z}+\mr{reg.}
\end{equation}

The stress-energy tensor is
\begin{equation}
T(z)=:2\dd c(z) \,b(z)+c(z)\dd b(z):.
\end{equation}
and similarly for $\ol{T}$.
The normal ordering here means that inside a correlator Wick contractions of fields inside $:\cdots:$ are prohibited.
%is defined in the usual way (mode operators $\cdots_{\geq 0}$ to the right and operators $\cdots_{<0}$ to the left), supplemented by the prescription that $\wh{b}_0$ goes to the right of $\wh{c}_0$. 
Using Wick's lemma as in Section \ref{sss free scalar in PI formalism}, one computes the OPEs of $b(z)$, $c(z)$, $T(z)$ with $T(w)$ or $\ol{T}(w)$ and finds that:
\begin{itemize}
\item $c$ is a primary field of conformal weight $(-1,0)$ (similarly, $\bar{c}$ is $(0,-1)$-primary),
\item $b$ is a primary field of conformal weight $(2,0)$ (similarly, $\bar{b}$ is $(0,2)$-primary),
\item one has the standard OPE of the stress-energy with itself (\ref{l23 TT}), (\ref{l23 T Tbar}), (\ref{l23 Tbar Tbar}) with central charge
\begin{equation}
c=\bar{c}=-26.
\end{equation}
\end{itemize}

\begin{remark}\label{l32 rem: bc system with j-parameter}
 One can consider a modified ghost system, with fields as above and with modified stress-energy tensor 
\begin{equation}
T(z)=:\dd c(z) b(z)+j\dd \big(c(z)b(z)\big):
\end{equation}
with $j\in \RR$ a parameter of the system (the case of reparametrization ghosts corresponds to $j=1$). Then one obtains by similar computations to the above that $c$ is $(-j,0)$-primary, $b$ is $(j+1,0)$-primary and the central charge is 
\begin{equation}
c=-12j^2-12j-2.
\end{equation}

The case $j=0$ in the terminology of \cite{DMS} is called the ``simple ghost system.'' By the formula above, this system has central charge $c=-2$.
\end{remark}

\subsection{Correlators on $\CP^1$, soaking field, ghost number anomaly}\label{sss: bc soaking field, ghost anomaly}
%\begin{remark}
Note that the correlator (\ref{l32 <bc>}) seems to contradict Lemma \ref{l26 lemma 2-point} (\ref{l26 lemma 2-point (b)}): we have a nonvanishing correlator of two primary fields of \emph{different} conformal weight ($2$ and $-1$). The answer to this seeming paradox is that the field $c$ on $\CP^1$ has zero-modes: there is a 3-dimensional space of holomorphic vector fields on $\CP^1$. When we wrote the Green's function (\ref{l32 <bc>}), we implicitly imposed the condition that the vector field $c(z)\dd_z$ vanishes together with its first and second derivatives at the point $\infty\in \CP^1$. This is tantamount to inserting a certain field $s$ (``zero-mode soaking field'') of conformal weight $(h,\bar{h})=(0,0)$ at $z=\infty$. So, the correlator (\ref{l32 <bc>}) is ``secretly'' a 3-point function 
\begin{equation}\label{l32 <sbc> on CP^1}
\langle s(\infty)b(w)c(z)\rangle_{\CP^1}.
\end{equation}
From this standpoint, there is no contradiction in the fact that the correlator is nonzero. For three arbitrary points on $\CP^1$, the correlator (\ref{l32 <sbc> on CP^1}) becomes a M\"obius-invariant expression
\begin{equation}
\langle b(z_1)(dz_1)^2\; c(z_2)\dd_{z_2} \; s(z_3) \rangle = \nu_{12} \nu_{13}^3 \nu_{23}^{-3},
\end{equation}
where
\begin{equation}
\nu_{ij}\colon=\frac{d^{\frac12}z_i d^{\frac12}z_j}{z_i-z_j}
\end{equation}
is (the square root of) the Szeg\"o kernel (\ref{l27 mu}). The soaking field $s$ can be written as 
\begin{equation}\label{l32 soaking operator}
s= \frac14( c\,\dd c\,\dd^2 c)(\bar{c}\, \bar\dd \bar{c} \,\bar{\dd}^2\bar{c}).
\end{equation} 
We refer to \cite[Section 10]{Witten_superstring} and \cite[Section 2.4]{LMY2} for details on soaking fields.

The presence of zero-modes also means that for instance one has 
\begin{equation}
\langle 1\rangle_{\CP^1}=\langle \vac |\vac\rangle=0
\end{equation} 
(which means that the theory does not satisfy the usual BPZ axiomatics). On the other hand, 
\begin{equation}\label{l32 <s>}
\langle s(\infty) \rangle_{\CP^1}=\langle s|\vac\rangle=1.
\end{equation}

One can assign the ``left ghost number'' $+1$ to the field $c$ and $-1$ to $b$ and likewise ``right ghost number'' $+1$ to $\bar{c}$ and $-1$ to $\bar{b}$. Then
for a correlator on $\CP^1$ of some collection of differential monomials inserted at points $z_1,\ldots,z_n\in \CP^1$ to be possibly nonzero, one needs the following selection rule to hold:
the total left ghost number and the total right ghost number (of the entire expression under the correlator) should both be $+3$: 
\begin{equation}\label{l32 ghost anomaly}
\#c-\#b=3,\quad \#\bar{c}-\#\bar{b}=3
\end{equation}
%\end{remark}
This %is an instance of the so-called 
phenomenon is known as the
``ghost number anomaly.'' For example, one has %\marginpar{finish}
\begin{equation}\label{l32 <ccc>}
\langle c(z_1) c(z_2) c(z_3) \rangle_{\CP^1}^{\mr{chiral}} = z_{12}z_{13}z_{23}
\end{equation}
Here for brevity we wrote the correlator in the chiral $bc$ system (ignoring the fields $\bar{b},\bar{c}$). Taking the limit $\lim_{z_2\ra z_1}\frac{1}{z_{12}}(\cdots)$ in (\ref{l32 <ccc>}), replacing $c(z_2)$ with its Taylor expansion around $z_1$, we have
\begin{equation}
\langle (c\,\dd c)(z_1) c(z_3) \rangle_{\CP^1}^{\mr{chiral}}=-z_{13}^2.
\end{equation}
Taking here the limit $\lim_{z_3\ra z_1}\frac{1}{z_{13}^2}\cdots$, replacing $c(z_3)$ with its expansion around $z_1$, we obtain 
\begin{equation}
\langle (\frac{-1}{2} c\,\dd c\,\dd^2 c)(z_1) \rangle_{\CP^1}^{\mr{chiral}}=1,
\end{equation} 
which is the chiral counterpart of (\ref{l32 <s>}).
%and 
%\begin{equation}
%\langle b(z_1) c(z_2) c(z_3)  c(z_4) c(z_5)\rangle_{\CP^1} = \cdots
%\end{equation}

For a surface $\Sigma$ of genus $g$, the ghost number anomaly (\ref{l32 ghost anomaly}) is given by Riemann-Roch theorem, as the dimension of the space of holomorphic vector fields minus the dimension of the space of holomorphic quadratic differentials:
\begin{equation}
\dim H^0_{\bar\dd}(\Sigma, K^{-1})-\dim H^0_{\bar\dd}(\Sigma,K^{\otimes 2}) = 3-3g.
\end{equation}

\subsection{Operator formalism for the $bc$ system}
One can develop the canonical quantization picture for the $bc$ system, similarly to how we did it for the other free field models before.
Then one obtains the Heisenberg fields on $\CC\backslash\{0\}$,
\begin{equation}
\wh{c}(z)=\sum_{n\in \ZZ} \wh{c}_n z^{-n+1}, \quad 
\wh{b}(z)=\sum_{n\in \ZZ} \wh{b}_n z^{-n-2}
\end{equation}
with operators $\wh{c}_n,\wh{b}_n$ subject to the anticommutation relations
\begin{equation}\label{l32 bc anticomm rel}
[\wh{b}_n,\wh{c}_m]_+= \delta_{n,-m} \wh{\mathbb{1}},\quad
[\wh{b}_n,\wh{b}_m]_+=0,\quad 
[\wh{c}_n,\wh{c}_m]_+=0.
\end{equation}
One has similar mode expansions and anticommutation relations for $\bar{b},\bar{c}$.  Here the 
%creation operators are $\wh{b}_{\leq -2}, \wh{\bar{b}}_{\leq -2}, \wh{c}_{\leq 1},\wh{\bar{c}}_{\leq 1}$
the splitting of the mode operators into creation and annihilation operators is as follows:
\begin{equation}
\underbrace{\ldots,\wh{c}_{-1},\wh{c}_{0},\wh{c}_{1}}_{\mr{creation}},\underbrace{\wh{c}_{2},\wh{c}_{3},\ldots}_{\mr{annihilation}},\qquad 
\underbrace{\ldots,\wh{b}_{-3},\wh{b}_{-2}}_{\mr{creation}},\underbrace{\wh{b}_{-1},\wh{b}_{0},\wh{b}_1,\ldots}_{\mr{annihilation}}
\end{equation}
and similarly for $\wh{\bar{b}}_n,\wh{\bar{c}}_n$.\footnote{
This nontrivial splitting of modes into creation and annihilation operators is forced by the field-state correspondence: one wants limits $\lim_{z\ra 0}\wh{\Phi}(z)|\vac\rangle$ to be well-defined and nonzero for $\Phi=b,c,\bar{b},\bar{c}$.
} The vacuum vector $|\vac\rangle$ is killed by annihilation operators, while creation operators produce nonzero vectors out of $|\vac\rangle$. The hermitian conjugates are $(\wh{b}_n)^+=\wh{b}_{-n}$, $(\wh{c}_n)^+=\wh{c}_{-n}$. The special vector $|s\rangle$ corresponding to the soaking field (\ref{l32 soaking operator}) is
\begin{equation}
|s\rangle = \wh{c}_{-1}\wh{c}_0\wh{c}_1 \wh{\bar{c}}_{-1}\wh{\bar{c}}_{0} \wh{\bar{c}}_{0} |\vac\rangle.
\end{equation}

The space of states $\HH$ is generated freely by acting on the vector $|\vac\rangle$ repeatedly with creation operators (i.e., $\HH$ is a Verma module for the Clifford algebra (\ref{l32 bc anticomm rel}), tensored with the conjugate one).

Reproducing the 2-point correlation function (\ref{l32 <bc>}) in the language of operator quantization, we have (assuming $|w|>|z|$ for simplicity):
\begin{equation}\label{l32 <bc> operator computation}
\langle b(w) c(z) \rangle = \langle s | \wh{b}(w) \wh{c}(z) |\vac\rangle = \sum_{m,n\in\ZZ} \langle s|\wh{b}_n \wh{c}_m |\vac\rangle w^{-n-2} z^{-m+1}.
\end{equation}
Here we notice that the expression $\langle s| \wh{b}_n \wh{c}_m|\vac\rangle$ has the following properties:
\begin{itemize}
\item Vanishes for $\wh{c}_m$ an annihilation operator (since then $\wh{c}_n|\vac\rangle =0$), i.e., for $m\geq 2$.
\item Vanishes for $n\neq -m$ and $\wh{b}_n$ an annihilation operator (since  $\wh{b}_n$ commutes past $\wh{c}_m$ and acts on $|\vac\rangle$), i.e., for $n\neq -m$, $n\geq -1$.
\item Vanishes for $\wh{b}_n$ a creation operator (in this case $\langle s|\wh{b}_n=0$), i.e., for $n\leq -2$.
\item Vanishes for $\wh{c}_m$ a creation operator if $n\neq -m$ (in this case, $\wh{c}_m$ commutes to the left past $\wh{b}_n$ and annihilates $\langle h|$), i.e., for $n\neq -m$, $m\leq 1$.
\end{itemize}
Thus, the only surviving terms in (\ref{l32 <bc> operator computation}) are $n=-m\geq -1$, i.e., we have
\begin{multline}
\langle b(w) c(z) \rangle = \sum_{n\geq -1} \langle s|\underbrace{\wh{b}_n \wh{c}_{-n}}_{\wh{\mathbb{1}}-\wh{c}_{-n}\wh{b}_n}|\vac\rangle w^{-n-2}z^{n+1} = \sum_{n\geq -1}w^{-n-2}z^{n+1}=\\
=\frac{1}{w}\left(1+\frac{z}{w}+\left(\frac{z}{w}\right)^2+\cdots\right)=\frac{1}{w-z}
\end{multline}

\section{Bosonic string}

We start with outlining the heuristic idea of bosonic string theory. One wants to integrate over maps $\phi$ of a smooth surface $\Sigma$ (worldsheet) to the target $\RR^D$ (for some dimension $D\geq 1$):
\begin{equation}\label{l32 Z_string}
Z_\mr{string}(\Sigma,\RR^D)=\int_{\mr{Met}(\Sigma)}\mc{D}g \int_{\mr{Map}(\Sigma,\RR^D)}
\mc{D}\phi\,
e^{-S_\mr{Polyakov}(g,\phi)}
\end{equation}
where
\begin{equation}
S_\mr{Polyakov}(g,\phi)=\sum_{k=1}^D\frac12 \int_\Sigma \dvol_g d\phi^k\wedge *d\phi^k
\end{equation}
is the action for $D$ non-interacting free bosons $\phi^1,\ldots,\phi^D$ on $\Sigma$; the action depends on a choice of Riemannian metric $g$ on the surface, and this choice is averaged over in  (\ref{l32 Z_string}). The integrand in (\ref{l32 Z_string}) is invariant under diffeomorphisms of $\Sigma$, and one wants to switch to integration over the quotient $\mr{Met}(\Sigma)\times \mr{Map}(\Sigma,\RR^D)/\mr{Diff}(\Sigma)$.\footnote{Heuristically, transitioning to integration over the quotient rescales the result by an ``infinite constant'' -- the volume of $\mr{Diff}(\Sigma)$.} 
Next, one writes the metric  as 
\begin{equation}
g=e^{2\sigma}g_0
\end{equation} 
where $g_0$ is the canonical ``uniformization'' metric of constant scalar curvature $K\in \{0,\pm1\}$ representing the conformal class of $g$ -- the metric arising from uniformization theorem; %(e.g. hyperbolic metric of scalar curvature $K=-1$).
$\Omega=e^{2\sigma}$ with $\sigma\in C^\infty(\Sigma)$ is the Weyl factor, transforming $g_0$ into $g$; one calls $\sigma$ the Liouville field. With this in mind, the path integral (\ref{l32 Z_string}) becomes the integral over
\begin{multline}
\{\mr{conformal\;structures\;on\;}\Sigma\}\times \{\mr{Weyl\;factors\;}\Omega=e^{2\sigma}\}\times \mr{Map}(\Sigma,\RR^D) /\mr{Diff}(\Sigma)\simeq\\
\simeq 
\frac{\{\mr{conformal\;structures\;on\;}\Sigma\}}{\mr{Diff}(\Sigma)}\times \{\mr{Weyl\;factors\;}\Omega=e^{2\sigma}\}\times \mr{Map}(\Sigma,\RR^D) /\mr{Diff}(\Sigma)
\end{multline}
where in the first factor on the r.h.s. we recognize the moduli space of conformal structures $\MM_\Sigma$. For the integral over the quotient by diffeomorphisms, one employs the Faddeev-Popov gauge-fixing mechanism, which results in the path integral
\begin{equation}
\int_{\MM_\Sigma}\mc{D}\xi \int_{C^\infty(\Sigma)}\mc{D}\sigma \int_{\Pi \mathfrak{X}(\Sigma)\times \Pi \Gamma(\Sigma,K^{\otimes 2}\oplus \ol{K}^{\otimes 2})} \mc{D}c\mc{D}\bar{c}\mc{D}b\mc{D}\bar{b} e^{-S_{bc} %(b,\bar{b},c,\bar{c})
}\;
\int_{\mr{Map}(\Sigma,\RR^D)} e^{-S_\mr{Polyakov}}
\end{equation}
where the auxiliary fields $c\dd_z+\bar{c}\dd_{\bar{z}}$ (an odd vector field) and $b(dz)^2+\bar{b}(d\bar{z})^2$ (an odd quadratic differential) appear as Faddeev-Popov ghosts corresponding to the quotient by diffeomorphisms (or ``reparametrizations,'' hence the name ``reparametrization ghosts''); the action $S_{bc}$ is as in (\ref{l32 S_bc}). The Gaussian integral over ghosts is an integral representation of a Jacobian, canceling the dependence of the integral over a section of the quotient $\{\mr{conf.\;structures}\}/\mr{Diff}(\Sigma)$ on the choice of the section.

Exploiting the result (\ref{l2 change of Z under Weyl}), we have that the bosonic string path integral is
\begin{equation}
\int_{\MM_\Sigma}\mc{D}\xi \int_{C^\infty(\Sigma)} \mc{D}\sigma\; Z_\mr{CFT}\left(
\begin{array}{c}
D\mr{\; free\; bosons} \\
+ bc\mr{\;system}
\end{array}, \xi
\right) e^{ic S_\mr{Liouville}(\sigma)}
\end{equation}
where 
\begin{equation}
c=D-26
\end{equation} 
is the central charge of the CFT comprised of $D$ free bosons and a single $bc$ system. The case $D=26$ is special and corresponds to the so-called ``critical'' bosonic string -- in this case the central charge vanishes and the integrand is independent of the Liouville field $\sigma$.

In summary, bosonic string is the conformal field theory 
comprised of $D$ free bosons and a $bc$ system, 
with  classical action%\marginpar{fix normalizations}
\begin{equation}\label{l32 string CFT}
S_\mr{string}=\frac{1}{\pi}\int_\Sigma d^2 z \Big(\underbrace{\sum_{k=1}^D \frac12\dd\phi^k \bar\dd\phi^k}_{D\mr{\; free\;bosons}} + \underbrace{b\bar\dd c+\bar{b}\dd\bar{c}}_{bc\;\mr{system}}\Big)
\end{equation}
where to get the full string path integral one needs to integrate the CFT partition function (or correlator) over the moduli space $\mc{M}_\Sigma$ (and if $D\neq 26$, also factor in the Liouville path integral).\footnote{
In a jargon, one couples the CFT (\ref{l32 string CFT}) on $\Sigma$ with ``2d gravity on $\Sigma$.''
}

\marginpar{Lecture 33,\\ 11/14/2022}

\subsection{The BRST differential $Q$ in bosonic string}
Fix $D=26$. Consider the fields 
\begin{equation}
\begin{gathered}\label{l33 J bosonic string}
J=\;:c T_\mr{bosons}+\frac12 c T_{bc}+\frac32 \dd^2 c:\; =\; :\sum_{k=1}^D-\frac12 c \dd\phi^k \dd\phi^k + c\, \dd c\, b+\frac32 \dd^2 c:,\\
\bar{J}=\;:\bar{c} \ol{T}_\mr{bosons}+\frac12 \bar{c} \ol{T}_{bc}+\frac32 \bar\dd^2 \bar{c}:\;=\; :\sum_{k=1}^D-\frac12 \bar{c}\bar\dd\phi^k \bar\dd\phi^k + \bar{c}\, \bar\dd \bar{c}\, \bar{b}+\frac32 \bar\dd^2 \bar{c}:
\end{gathered}
\end{equation}
They satisfy the following properties.
\begin{itemize}
\item $J$ is a holomorphic $(1,0)$-primary field, $\bar{J}$ is an antiholomorphic $(0,1)$-primary field.
\item The  
OPE $J(w) J(z)$ does not contain a first-order pole\footnote{
\label{l33 footnote JJ}
This property relies on $D=26$. 
%For $D\neq 26$, the OPE $J(w)J(z)$ has a third order pole.
More explicitly, if one defines $J_\alpha=:cT_\mr{bosons}+\frac12 cT_{bc}+\alpha \dd^2 c:$, then one has the OPE 
$$J_\alpha(w)J_\alpha(z)\sim \frac{(3-\frac{D}{2}+4\alpha)\,c\dd c}{(w-z)^3}+\frac{(\frac32-\frac{D}{4}+2\alpha)\, c\dd^2 c}{(w-z)^2}+ \frac{(\frac23-\frac{D}{12}+\alpha)c\dd^3 c+(-\frac32+\alpha) \dd c\dd^2 c}{w-z}+\mr{reg.}$$
Here all fields on the right are at the point $z$.
In particular, for $D=26$ and $\alpha=\frac32$ one has
$J(w)J(z)\sim -\frac{4c\dd c}{(w-z)^3}-\frac{2c \dd^2 c}{(w-z)^2}+\mr{reg.}$
} (but contains second and third-order poles) and similarly for $\bar{J}(w)\bar{J}(z)$. The mixed OPE $J(w)\bar{J}(z)$ is regular.
%\begin{equation}
%J(w)J(z)\sim \mr{reg.},\quad J(w)\bar{J}(z)\sim \mr{reg.},\quad \bar{J}(w)\bar{J}(z)\sim \mr{reg.}
%\end{equation}
\item One can introduce an operator $Q\colon V_z\ra V_z$ given by%\marginpar{I think there should be no minus..}
\begin{equation}\label{l33 JJ OPE}
Q\colon \Phi(z)\mapsto \frac{1}{2\pi i}\oint_{\gamma_z} (dw J(w)+d\bar{w} \bar{J}(w)) \Phi(z),
\end{equation}
with $\gamma_z$ a contour around $z$.
This operator satisfies 
\begin{equation}
Q^2=0,
\end{equation}
as a consequence of (\ref{l33 JJ OPE}) (proven by the contour integration technique of Section \ref{sss: Virasoro from TT OPE}). One can equip $V$ with $\ZZ$-grading by \emph{(total) ghost number}, by prescribing the ghost numbers to elementary fields as follows:

\vspace{0.2cm}
\begin{center}
\begin{tabular}{c|ccccc}
field & $c$ & $b$ & $\bar{c}$ & $\bar{b}$ & $\phi^k$ \\ \hline
ghost number &$1$ & $-1$ & $1$ & $-1$ & $0$
\end{tabular}
\end{center}
\vspace{0.2cm}

-- This is the sum of the left and right ghost numbers of Section \ref{sss: bc soaking field, ghost anomaly}. According to this $\ZZ$-grading, $V$ is a cochain complex, with differential $Q$ (known as the ``BRST operator''), increasing the ghost number by $+1$.

\item The stress-energy tensor satisfies
\begin{equation}
T=Q(b),\quad \ol{T}=Q(\bar{b})
\end{equation}
-- the stress-energy tensor is $Q$-exact.
\end{itemize}

\begin{remark}
If one omits the $\frac32 \dd^2 c$ term in $J$ and likewise in $\bar{J}$ then the residue  of the first-order pole in $JJ$ OPE will be nonzero, but it will be exact, so the operator $Q$ would not change. Also, with this modification $J,\bar{J}$ would not be primary.
\end{remark}

\begin{remark} Fields $J,\bar{J}$ are also $Q$-exact:
\begin{equation}\label{l33 J=Q(bc)}
J=Q(:bc:),\quad \bar{J}=Q(:\bar{b}\bar{c}:).
\end{equation}
We also note that in the computation of the r.h.s. it is the double Wick contractions that result in $\frac32 \dd^2 c$ term in $J$ in the l.h.s.; in this sense, the term $\frac32 \dd^2 c$ should be regarded as a quantum (``1-loop'' in the language of Feynman diagrams) correction to $J$.\footnote{In a bit more detail: in the classical field theory defined by the action functional (\ref{l32 string CFT}) one has an odd symmetry $Q_{cl}\in \mathfrak{X}(\F_\Sigma)$ acting on the space of classical fields and squaring to zero. For this symmetry one has an associated Noether current $dz J_{cl}+d\bar{z} \bar{J}_{cl}$, where $J_{cl},\bar{J}_{cl}$ are given by the formulae (\ref{l33 J bosonic string}) without the $\frac32\dd^2 c$, $\frac{3}{2}\bar\dd^2\bar{c}$ terms and without normal ordering. Thus, the quantum fields (\ref{l33 J bosonic string}) are the ``naive'' quantization of $J_{cl},\bar{J}_{cl}$ (replacing a differential polynomial in free classical fields by a normally ordered expression), plus a ``quantum correction'' $\frac32 \dd^2 c$, $\frac{3}{2}\bar\dd^2\bar{c}$.
}
\end{remark}

\section{Topological conformal field theories}
\label{ss TCFT}

\begin{definition}
A CFT is called \emph{topological} (or TCFT\footnote{We refer the reader to the introductory part of \cite{Donald_PhD} for an introduction to topological conformal field theories.}) is
the space of fields $V$ (and the space of states $\HH$) is endowed with the structure a cochain complex with differential $Q$ of degree $+1$,
%(we will call the $\ZZ$-grading on $V$ the ghost number/degree), 
such that the stress-energy tensor is $Q$-exact, 
\begin{equation}\label{l33 T=Q(G)}
T=Q(G),\quad \ol{T}=Q(\ol{G}),
\end{equation}
with $G,\ol{G}$ some fields of cohomological degree $-1$, such that
\begin{enumerate}[(a)]
\item One has regular OPEs 
\begin{equation}\label{l33 GG OPE}
G(w)G(z),\quad G(w)\ol{G}(z),\quad \ol{G}(w) \ol{G}(z).
\end{equation}
\item \label{l33 def (b)}
$G$ is a holomorphic $(2,0)$-primary field, $\ol{G}$ is an antiholomorphic $(0,2)$-primary field.
\item There exist fields $J,\bar{J}\in V$ of degree $+1$ and conformal weights $(1,0)$ for $J$ and $(0,1)$ for $\bar{J}$, such that:
\begin{itemize}
\item The 1-form-valued field 
\begin{equation}
\mathbb{J}(z)=dz\,J+d\bar{z}\,\bar{J}\in V_z\otimes T^*_z \Sigma
\end{equation} 
is $d$-closed under the correlator, or equivalently 
\begin{equation}\label{l33 dbar J - d barJ=0}
\bar\dd J-\dd \bar{J}=0.
\end{equation}
%\footnote{\label{l33 footnote J, Jbar conservation}
%In some TCFTs one has that $\bar\dd J$ and $\dd \bar{J}$ vanish separately. This means that $Q$ splits into two commuting differentials $Q=Q_L+Q_R$ which square to zero separately. This extra symmetry is present, e.g., in bosonic string theory and in  A-model (Section \ref{s A-model}), but fails in some other examples, see e.g. \cite{LMY2}.
%}
\item The differential $Q$ is given by%\marginpar{sign?}
\begin{equation}\label{l33 Q = action of J}
Q\Phi(z) = \frac{1}{2\pi i}\oint_{\gamma_z} \mathbb{J}(w) \Phi(z).
\end{equation}
%where $\mathbb{J}(w)\colon= dw J(w)+d\bar{w} \ol{J}(w)$.
\item 
%the OPE $\mathbb{J}(w) \mathbb{J}(z)$ %is regular 
%does not contain a first-order pole, i.e., 
$\mathbb{J}$ satisfies
\begin{equation}\label{l33 oint JJ}
\oint_{\gamma_z} \mathbb{J}(w) \mathbb{J}(z)=0.
\end{equation}
This property implies $Q^2=0$.
\end{itemize}
\item The field $\mathbb{1}$ is not $Q$-exact (note that it is automatically $Q$-closed).
\end{enumerate}
%\marginpar{Do we need an axiom that $\mathbb{1}$ is not $Q$-exact? Do we need to ask that $\mathbb{J}$ is primary?}
\end{definition}

In particular, bosonic string with $D=26$ is an example of a TCFT.

\begin{remark}\label{l33 rem: chirally split TCFTs}
In some TCFTs a stronger version of property (\ref{l33 dbar J - d barJ=0}) holds: %one has that 
$\bar\dd J$ and $\dd \bar{J}$ vanish separately. This means that $Q$ splits into two commuting differentials $Q=Q_L+Q_R$ which square to zero separately. We will call such TCFTs ``chirally split.'' 
 This extra symmetry is present, e.g., in bosonic string theory and in 
 A-model (Section \ref{s A-model}), but fails in some other examples, see e.g. \cite{LMY2}. More generally, all so-called twisted $\mc{N}=(2,2)$ supersymmetric CFTs are chirally split -- the A-model belongs to this class, cf. Section \ref{ss SUSY sigma model}. The converse is not true, e.g., bosonic string is not a twisted supersymmetric theory.\footnote{
Twisted supersymmetric theories have an extra symmetry between $J$ and $G$ -- the so-called R-symmetry. Also, in a twisted supersymmetric theory, $J(w)J(z)$ OPE is purely regular (unlike in the case of bosonic string, see footnote \ref{l33 footnote JJ}).
 }.
\end{remark}

One can introduce mode operators for $G,\ol{G}$, defined similarly to (\ref{l25 L_n loc}):
\begin{equation}\label{l33 G modes}
G_n\Phi(z)=\frac{1}{2\pi i}\oint_{\gamma_z} dw\,(w-z)^{n+1} G(w) \Phi(z),\quad 
\ol{G}_n\Phi(z)=\frac{1}{2\pi i}\oint_{\gamma_z} d\bar{w} (\bar{w}-\bar{z})^{n+1} \ol{G}(w) \Phi(z)
\end{equation} 
Then the property (\ref{l33 T=Q(G)}) implies that one has\footnote{
We write $[A,B]=AB-(-1)^{|A|\cdot |B|} BA$ for the supercommutator of two operators $A,B$. It is the usual commutator if either $A$ or $B$ (or both) are even and it is the anticommutator if $A$ and $B$ are odd. 
}
\begin{equation}
L_n=[Q,G_n],\quad \ol{L}_n= [Q,\ol{G}_n]
\end{equation}
for $n\in\ZZ$,
i.e., Virasoro generators are $Q$-exact. In turn this implies that the central charge of the CFT must vanish (because the coefficient of the fourth-order pole in $TT$ OPE must be $Q$-exact; since it is proportional to identity, it must vanish\footnote{
An equivalent argument: the commutator $[L_n,L_m]=[Q,[G_n,[Q,G_m]]]$ has the form $[Q,-]$, so it cannot contain a nonzero term proportional to identity/central element (with is not of the form $[Q,-]$).
}):
\begin{equation}\label{l33 c=0}
c=\bar{c}=0.
\end{equation}
Property (\ref{l33 GG OPE}) implies
\begin{equation}
[G_n,G_m]=0,\quad [G_n,\ol{G}_m]=0, \quad [\ol{G}_n,\ol{G}_m]=0.
\end{equation}

From the OPEs between $T, \ol{T}$ and $G,\ol{G}$, which are encoded in the axiom (\ref{l33 def (b)}) above: 
\begin{equation}\label{l33 TG OPE}
\begin{gathered}
T(w) G(z)\sim \frac{2G(z)}{(w-z)^2}+\frac{\dd G(z)}{w-z}+\mr{reg.},\\
\ol{T}(w) \ol{G}(z)\sim \frac{2\ol{G}(z)}{(\bar{w}-\bar{z})^2}+\frac{\bar\dd \,\ol{G}(z)}{\bar{w}-\bar{z}}+\mr{reg.},\\
T(w) \ol{G}(z)\sim \mr{reg.}, \; \ol{T}(w) G(z)\sim \mr{reg.},
\end{gathered}
\end{equation}
one has
 the commutation relations
\begin{equation}\label{l33 [L,G]}
[L_n, G_m]=(n-m)G_{n+m}, \quad [\ol{L}_n, \ol{G}_m]=(n-m)\ol{G}_{n+m},\quad [L_n, \ol{G}_m] = [\ol{L}_n, G_m] =0.
\end{equation}

\begin{lemma}\label{l33 lemma: corr of Q-closed are constant}
In a TCFT, assume that $\Phi_1,\ldots,\Phi_n\in V$ are $Q$-closed elements. Then: 
\begin{enumerate}[(i)]
\item \label{l33 lemma (i)} The correlator on $\CP^1$
\begin{equation}\label{l33 corr of Q-closed}
\langle \Phi_1(z_1)\cdots \Phi_n(z_n)\rangle
\end{equation}
is a constant function on the configuration space $C_n(\CP^1)$.
\item \label{l33 lemma (ii)} For any $\Psi\in V$, one has
\begin{equation}
\langle Q\Psi(z_0)\;\Phi_1(z_1)\cdots \Phi_n(z_n)\rangle = 0
\end{equation}
-- the correlator of a $Q$-exact field with several $Q$-closed fields vanishes.
\end{enumerate}
\end{lemma}

\begin{proof}
For, (\ref{l33 lemma (i)}) consider the derivative of the correlator (\ref{l33 corr of Q-closed}) in $z_i$, $i=1,\ldots,n$. We have
\begin{multline}\label{l33 lemma computation}
\dd_{z_i} \langle \Phi_1(z_1)\cdots \Phi_n(z_n)\rangle = \langle \Phi_1(z_1)\cdots \underbrace{L_{-1}\Phi_i}_{Q G_{-1} \Phi_i}(z_i)\cdots \Phi_n(z_n)\rangle=\\
= \pm\frac{1}{2\pi i}\oint_{\gamma_{z_i}} \langle \mathbb{J}(w) \Phi_1(z_1) \cdots G_{-1}\Phi_i(z_i)\cdots \Phi_n(z_n) \rangle
\end{multline}
where $\gamma_{z_i}$ is a contour going around $z_i$ and not enclosing any other $z_j$'s. One then deforms $\gamma_{z_i}$ into a collection of contours going around $z_j$'s for $j\neq i$:  $\gamma_{z_i}\sim \sqcup_{i\neq j} -\gamma_{z_j}$ (cf. Section \ref{ss: Ward identity via contour integration}). Thus, one has 
\begin{equation}
\dd_{z_i} \langle \Phi_1(z_1)\cdots \Phi_n(z_n)\rangle  = \sum_{j\neq i} \langle \Phi_1(z_1)\cdots \underbrace{Q\Phi_j}_{=0}(z_j)\cdots G_{-1}\Phi_i \cdots \Phi_n(z_n) \rangle =0.
\end{equation}
So, we obtain that all holomorphic derivatives of the correlator vanish; by a similar argument, the antiholomorphic derivatives vanish too. Hence, the correlator is constant.

The proof of (\ref{l33 lemma (ii)}) is similar: one represents $Q$ acting on $\Psi$ by a contour integral around $z_0$ and then deform the contour to a collection of contours going around $z_j$, $j\neq 0$; those give correlators containing $Q\Phi_j=0$.
\end{proof}

%\marginpar{edit? (or maybe it's ok already..)}
\begin{remark} \label{l33 rem: higher genus}
The statement and proof of Lemma \ref{l33 lemma: corr of Q-closed are constant} actually extends to correlators on Riemannian surfaces $\Sigma$ of any genus $g$, 
since on any $\Sigma$ the 1-cycle $\gamma_{z_i}$ is \emph{homologous} (though not homotopic for $g>0$) to $\sqcup_{i\neq j} -\gamma_{z_j}$. Since one has $d\mathbb{J}=0$ (under a correlator, away from the punctures $z_j$), this homology statement is sufficient to justify the switch of contours in (\ref{l33 lemma computation}).
\begin{comment}
as a consequence of the fact 
that the partition function of a torus with one hole is $Q$-closed, or equivalently that the 1-point correlator of any $Q$-exact field on the torus vanishes:
\begin{equation}\label{l33 <Q Phi>_torus}
\langle Q\Phi(z) \rangle_\mr{torus} =0.
\end{equation} 
The latter follows from the fact that the integration contour $\gamma$ for the field $\mathbb{J}$ in (\ref{l33 <Q Phi>_torus}) -- a contour going around the puncture $z$ -- is null-homologous on the punctured torus. Knowing this and thinking of $\Sigma$ as a sphere with $g$ handles (one-holed tori) attached, the statement of Lemma \ref{l33 lemma: corr of Q-closed are constant} follows by repeating its proof and taking care that when deforming the contour $\gamma_{z_i}$ it also creates components going around the necks by which handles are attached. By (\ref{l33 <Q Phi>_torus}), the corresponding terms in the correlator vanish.
\end{comment}
\end{remark}

\begin{lemma}
If a field $\Phi\in V$ is $Q$-closed and has conformal weight $h\neq 0$, then $\Phi$ is $Q$-exact.
\end{lemma}
\begin{proof}
Since $\Phi$ has conformal weight $h$, we have
\begin{equation}
h\Phi = L_0\Phi  = (QG_0+\cancel{G_0 Q}) \Phi = QG_0\Phi.
\end{equation}
Thus, we have
\begin{equation}\label{l33 homotopy}
\Phi = Q(\frac{1}{h}G_0 \Phi).
\end{equation}
\end{proof}
Similarly, one shows that if a $Q$-closed field has $\bar{h}\neq 0$, then it is $Q$-exact. Therefore, a nontrivial $Q$-cocycle (homogeneous w.r.t. grading by conformal weight) must have $(h,\bar{h})=0$. By ``nontrivial'' we mean ``not $Q$-exact'' or equivalently defining a nonzero element in the cohomology of $Q$, $H_Q(V)$.

\begin{example}
Here is an example of a $Q$-cocycle in bosonic string: fix a unit ``momentum'' vector  $p\in (\RR^D)^*$, $||p||=1$. Then the field
\begin{equation}\label{l33 tachyon}
\Phi= :c\bar{c}\; e^{i\sqrt{2} \sum_{k=1}^D p_k \phi^k}:
\end{equation}
is a nontrivial $Q$-cocycle. (In string theory, in Lorentzian signature on $\RR^D$ it is called the ``tachyon field.'')
Note that the condition $||p||=1$ guarantees that $\Phi$ has confromal weight $(0,0)$.
\end{example}

\begin{example} In any TCFT one has $QJ =Q\bar{J}=0$, as a consequence of (\ref{l33 oint JJ}). Thus, using homotopy (\ref{l33 homotopy}) one has
\begin{equation}
J= Q(G_0(J)), \quad \bar{J}=Q(\ol{G}_0(\bar{J})),
\end{equation}
i.e., fields $J,\bar{J}$ are always $Q$-exact. This generalizes the observation (\ref{l33 J=Q(bc)}) in bosonic string theory.
\end{example}

\begin{remark}[1d version: topological quantum mechanics]
Two-dimensional TCFTs have a one-dimensional analog: topological quantum mechanics (TQM).  Topological quantum mechanics is defined by a $\ZZ$-graded vector space $\HH$ (the space of states of a point) equipped with a differential $Q$ of degree $+1$ and a second differential $G$ of degree $-1$ (both differentials are assumed to square to zero). The Hamiltonian is defined as the anticommutator 
\begin{equation}
H=[Q,G].
\end{equation}

An example of this structure is: $\HH=\Omega^\bt(X), Q=d_X$ -- the de Rham complex of a target manifold $X$. For $G$ one can choose:
\begin{enumerate}[(a)]
\item \label{l33 TQM (a)} The Hodge-de Rham codifferential $G=d^* = \pm *d*$  (assuming that $X$ is Riemannian). In this case, the Hamiltonian $H=\Delta$ is the Laplace-Beltrami operator on $X$.
\item  \label{l33 TQM (b)} Contraction with a vector a field $v\in \X(X)$, $G=\iota_X$. In this case, the Hamiltonian $H=\LL_v$ is the Lie derivative and the quantum-mechanical evolution operator $U(t)=e^{-t \LL_v}$ is the flow along $v$ on $X$ in a given time.
\item    \label{l33 TQM (c)} One can interpolate between cases (\ref{l33 TQM (a)}) and (\ref{l33 TQM (b)})  for $v=-\mr{grad}(f)$  the gradient vector field of a Morse function $f\in C^\infty(X)$  by setting
\begin{equation}
G=  e^{-\frac{f}{\epsilon}} d^* e^{\frac{f}{\epsilon}}
\end{equation}
with $\epsilon$ an interpolation parameter. This is the setting of  the seminal example of TQM \cite{Witten_Morse}.
\end{enumerate}
Topological quantum mechanics is also known as $\mc{N}=2$ supersymmetric quantum mechanics.\footnote{A small caveat: in $\mc{N}=2$ supersymmetric quantum mechanics, one only requires a $\ZZ_2$-grading on $\HH$ instead of $\ZZ$-grading. Then one just requires that $Q,G$ are odd operators, and hence $H$ is an even operator.}
 We refer to \cite{Losev_TQM} for details on topological quantum mechanics. 
%and to \cite{Witten_Morse} for the seminal example. 
\end{remark}

\subsection{Witten's descent equation}
%Given a $Q$-closed observable $\Phi=\Phi^{(0)}\in V$, one can try to construct a 1-form valued

Witten's descent equation is a sequence of equations on a tower of $p$-form valued fields $\Phi^{(0)}$, $\Phi^{(1)}$, $\Phi^{(2)}$, where\footnote{Recall that we already encountered a situation where it is convenient to consider form-valued observables, -- transformation of primary fields and Ward identity for primary fields, cf. (\ref{l25 Phi underline}), (\ref{l26 corr of primary fields as a section of a line bundle}).} 
\begin{equation}
\Phi^{(p)}(z)\in V_z^{(p)}=V_z\otimes \wedge^p T^*_z \Sigma
\end{equation}
(we denoted $V_z^{(p)}$ the space of $p$-form-valued fields at $z$).
Descent equation reads
\begin{equation}\label{l33 descent eq}
d\Phi^{(p-1)} = Q \Phi^{(p)},\qquad p=0,1,2.
\end{equation}
Here we understand that $\Phi^{(-1)}\colon=0$. Thus, explicitly, the equations are:
\begin{eqnarray}
Q\Phi^{(0)} & = & 0 \label{l33 descent 0}\\
Q\Phi^{(1)} &=& d\Phi^{(0)} \label{l33 descent 1}\\
Q\Phi^{(2)} &=& d\Phi^{(1)} \label{l33 descent 2}
\end{eqnarray}

One can think of this sequence as follows: one fixed a $Q$-cocycle $\Phi^{(0)}$ -- a ``0-observable,'' %(recall that ) 
then one wants to solve (\ref{l33 descent 1}) for the ``1-observable'' $\Phi^{(1)}$ and subsequently solve (\ref{l33 descent 2}) for the 2-observable $\Phi^{(2)}$. 

\begin{remark}
Descent equations (\ref{l33 descent eq}) are meaningful not just in dimension $2$ (then $p$ goes up to the dimension of the manifold). Originally, they appeared in the work of Witten on 4-dimensional Donaldson theory \cite{Witten_Donaldson}.
\end{remark}

From Lemma \ref{l33 lemma: corr of Q-closed are constant}, correlators of $Q$-closed $0$-observables $\langle \Phi^{(0)}_1(z_1)\cdots \Phi_n^{(0)}(z_n)\rangle$ are constant functions of positions $z_1,\ldots, z_n$ (as long as points are distinct).

Equation (\ref{l33 descent 1}) implies that one can construct an ``extended observable'' (localized on a 1-cycle rather than at a point)
\begin{equation}\label{l33 Phi^1 integrated}
%\Phi^{(1)}(\gamma)\colon=
\oint_\gamma \Phi^{(1)}
\end{equation}
with $\gamma$ some closed contour. Then (\ref{l33 descent 1}) implies by Stokes' theorem that (\ref{l33 Phi^1 integrated}) is $Q$-closed:\footnote{
We understand (\ref{l33 Phi^1 integrated}) as an element of $V$ -- in that sense it is clear what acting by $Q$ means.
Equivalently, the 
action of $Q$ on (\ref{l33 Phi^1 integrated}) can be understood as $\frac{1}{2\pi i}\int_{\dd U_\gamma \ni w} \mathbb{J}(w) 
%\Phi^{(1)}(\gamma)
\oint_{\gamma\ni z} \Phi^{(1)}(z)
$, with the integral being over the boundary of a thickening $U_\gamma$ of the contour $\gamma$.}
\begin{equation}
Q \oint_\gamma \Phi^{(1)}=0
\end{equation}
By repeating the argument of Lemma \ref{l33 lemma: corr of Q-closed are constant}, we have that, given $Q$-cocycles $\Phi^{(0)},\Phi_1^{(0)},\ldots, \Phi_n^{(0)}$, the correlator 
\begin{equation}\label{l33 top correlator}
\langle \oint_\gamma \Phi^{(1)}\;\; \Phi^{(0)}_1(z_1)\cdots \Phi^{(n)}_n(z_n) \rangle
\end{equation}
does not change when one moves points $z_i$ or deforms the contour $\gamma$ (as long as the points and the contxour keep disjoint), however it can change when some point $z_i$ crosses $\gamma$.

The correlator (\ref{l33 top correlator}) is an example of a ``topological correlator'' -- one invariant under small deformations of insertion points of fields (and the contour over which the 1-observable is integrated).

%\marginpar{Edit, think through more}
A 2-observable $\Phi^{(2)}$ gives rise to a $Q$-closed extended observable
\begin{equation}
%\Phi^{(2)}(\Sigma)=
\int_\Sigma \Phi^{(2)}
\end{equation}
and can be understood as defining an infinitesimal deformation of a TCFT, deforming the correlators as
\begin{multline}\label{l33 deformation by Phi^2}
\langle \Phi_1(z_1)\cdots \Phi_n(z_n) \rangle \mapsto \\
\mapsto 
\langle \Phi_1(z_1)\cdots \Phi_n(z_n) \rangle +\epsilon
\Big\langle \left(\int_{\Sigma-\sqcup_{i=1}^n D_i}\Phi^{(2)}\right)  \Phi_1(z_1)\cdots \Phi_n(z_n) \Big\rangle 
\end{multline}
Here $D_i$ is a small disk centered at $z_i$; $\epsilon$ is the infinitesimal deformation parameter. 
%Here again if $\Phi_i^{(0)}$ are $Q$-closed, then the deformed correlator is locally constant, again by similar argument to the above.\marginpar{Edit: That's not right: one needs to incorporate the deformation of $Q$.}
\marginpar{?}
This deformation should be accompanied by a deformation of the rest of TCFT data, $Q,\mathbb{J},G,\ol{G},T,\ol{T}$, so that the relations of TCFT hold (up to $O(\epsilon^2)$) for the deformed package.

The deformation (\ref{l33 deformation by Phi^2}) in the path integral language can be interpreted as the deformation of the action functional,
\begin{equation}
S\mapsto S- \epsilon \int_\Sigma \Phi^{(2)}.
\end{equation}

\subsubsection{Total descendant}
Given a solution $\Phi^{(0)}$, $\Phi^{(1)}$, $\Phi^{(2)}$ of the descent equation (\ref{l33 descent eq}), one can consider the ``total descendant''
\begin{equation}\label{l33 total descendant}
\til\Phi\colon =\Phi^{(0)}+\Phi^{(1)}+\Phi^{(2)}
\end{equation}
-- a field valued in nonhomogenous forms.
The descent equation can be written in terms of the total descendant as
\begin{equation}\label{l33 (d-Q) til Phi =0}
(d-Q)\til\Phi=0.
\end{equation}

\subsubsection{Closed forms on the configuration space from correlators of total descendants}
\label{l33 rem: closed forms from a TCFT}
\begin{lemma}
Given a collection of $Q$-cocycles $\Phi_1^{(0)},\ldots,\Phi_n^{(0)}$, the correlator of their total descendants (\ref{l33 total descendant}) is a \emph{closed} form on the open  configuration space:
\begin{equation}\label{l33 closed form on C_n}
\langle \til\Phi_1(z_1)\cdots \til\Phi_n(z_n) \rangle \in \Omega_\mr{closed}(C_n(\CP^1)).
\end{equation}
\end{lemma}
\begin{proof}
Indeed, one has
\begin{equation}\label{l33 closedness argument}
d\langle \til\Phi_1(z_1)\cdots \til\Phi_n(z_n) \rangle  =
\sum_{j=1}^n \langle \til\Phi_1(z_1)\cdots \underbrace{(d-Q)\til\Phi_j(z_j)}_0\cdots \til\Phi_n(z_n) \rangle =0.
\end{equation}
\end{proof}
In fact, the form (\ref{l33 closed form on C_n}) is $PSL_2(\CC)$-basic and thus descends to a closed form on the moduli space $\MM_{0,n}\simeq C_n(\CP^1)/PSL_2(\CC)$.

We remark also
that if in the correlator (\ref{l33 closed form on C_n}) we replace one of the fields $\til\Phi_j$ by a $(d-Q)$-\emph{exact}  form-valued field $(d-Q)\Xi^\bt$, with some $\Xi^\bt\in V\otimes \wedge^\bt T^*\Sigma$, the resulting correlator will be an exact form on the configuration space (rather than just a closed one):
\begin{equation}
\langle \til\Phi_1(z_1)\cdots (d-Q) \Xi^\bt(z_j) \cdots \til\Phi_n(z_n) \rangle = d \langle \til\Phi_1(z_1)\cdots \Xi^\bt(z_j) \cdots \til\Phi_n(z_n) \rangle.
\end{equation}
 This is proven by an argument similar to (\ref{l33 closedness argument}).

In the example of bosonic string theory the degree of the form (\ref{l33 closed form on C_n}) is 
\begin{equation}
\sum_{j=1}^n \mr{gh}(\Phi_j^{(0)})-6
\end{equation}
-- the sum of the ghost numbers of the fields $\Phi_j^{(0)}$ minus the total (left plus right) ghost number anomaly.

By Remark \ref{l33 rem: higher genus}, the correlator (\ref{l33 closed form on C_n}) can be considered on a surface $\Sigma$ of any genus, yielding again a closed form on the configuration space.

\subsection{Canonical solution of descent equations using the $G$-field}
In a TCFT, one can find a canonical solution of the equation (\ref{l33 descent eq}) starting from any $Q$-cocycle $\Phi^{(0)}$.

Consider the operator
\begin{equation}\label{l33 Gamma}
\Gamma=-dz\, G_{-1} - d\bar{z}\,\ol{G}_{-1} \colon  \; V_z^{(p)}\ra V_{z}^{(p+1)},
\end{equation}
where $G_{-1},\ol{G}_{-1}$ are particular mode operators of the fields $G,\ol{G}$, cf. (\ref{l33 G modes}). We will refer to $\Gamma$ as the \emph{descent operator}.  Note that the commutator of $\Gamma$ with $Q$ is  the de Rham operator:
\begin{equation}\label{l33 [Q,Gamma]=d}
[Q,\Gamma]=dz [Q,G_{-1}]+d\bar{z} [Q,\ol{G}_{-1}]= dz L_{-1}+d\bar{z} \ol{L}_{-1} = dz \dd_z + d\bar{z} \dd_{\bar{z}} = d.
\end{equation}

Note also that one has 
\begin{equation}\label{l33 [d,Gamma]=0}
[d,\Gamma]=0, 
\end{equation}
since operators $L_{-1},\ol{L}_{-1}$ commute with $G_{-1},\ol{G}_{-1}$, cf. (\ref{l33 [L,G]}).

\begin{lemma}
Given a $Q$-cocycle $\Phi^{(0)}$, the sequence 
\begin{equation}\label{l33 can descent}
\Phi^{(0)},\quad\Phi^{(1)}\colon= \Gamma \Phi^{(0)},\quad \Phi^{(2)}\colon= \frac12 \Gamma^2 \Phi^{(0)}
\end{equation}
solves the descent equation (\ref{l33 descent eq}).
\end{lemma}

\begin{proof}
The equation (\ref{l33 descent 0}) is given, since $\Phi^{(0)}$ is a $Q$-cocycle. For (\ref{l33 descent 1}) we have
\begin{equation}
Q \Gamma \Phi^{(1)} \underset{(\ref{l33 [Q,Gamma]=d})}{=} (d+\Gamma Q) \Phi^{(0)} = d\Phi^{(0)}.
\end{equation}
For (\ref{l33 descent 2}) we have
\begin{multline}
Q\frac12 \Gamma\Gamma \Phi^{(0)}= \frac12 (d+\Gamma Q) \Phi^{(1)}\underset{(\ref{l33 descent 1})}{=} \frac12 (d \Phi^{(1)}  +\Gamma d \Phi^{(0)}) =\\
\underset{(\ref{l33 [d,Gamma]=0})}{=}
\frac12 (d \Phi^{(1)}  + d \Gamma\Phi^{(0)})= d\Phi^{(1)}.
\end{multline}
\end{proof}

\begin{example} In bosonic string, starting with the $Q$-cocycle (\ref{l33 tachyon}) and applying the canonical descent construction (\ref{l33 can descent}), we obtain the descent sequence
\begin{equation}
\Phi^{(0)}=:c\bar{c} V_p:,\quad
\Phi^{(1)}=:( -dz \, \bar{c}+ dz\,c)V_p:,\quad
\Phi^{(2)}=:dz\,d\bar{z} V_p:,
\end{equation}
where we denoted $V_{p}=e^{i\sqrt{2} \sum_k p_k \phi^k}$, with the momentum satisfying $||p||=1$.
\end{example}

For $\Phi=\Phi^{(0)}$ a $Q$-cocycle, one can assemble the  descendants (\ref{l33 can descent}) into the canonical total descendant
%field valued in nonhomogeneous forms
\begin{equation}\label{l33 canonical total descendant}
\til\Phi\colon=e^\Gamma \Phi = \Phi+\Gamma \Phi+\frac12 \Gamma^2 \Phi
\end{equation}
%-- the ``total descendant.''
satisfying the equation (\ref{l33 (d-Q) til Phi =0}).
%Then the descent equation (\ref{l33 descent eq}) can be written as 
%\begin{equation}\label{l33 (d-Q) til Phi =0}
%(d-Q) \til\Phi =0.
%\end{equation}
More generally, for $\Phi$ not necessarily $Q$-closed, one has an easily proven identity
\begin{equation}\label{l33 (d-Q)e^Gamma}
(d-Q) e^\Gamma \Phi = -e^\Gamma (Q\Phi).
\end{equation}

\subsection{BV algebra structure on $Q$-cohomology}
\label{sss: BV algebra on Q-cohomology}

\begin{definition}\label{l33 def: BV algebras}
A Batalin-Vilkovisky algebra (or ``BV algebra'') is a $\ZZ$-graded supercommutative unital algebra $(W,\cdot, \mathbb{1})$ equipped additionally with:
\begin{itemize}
\item A degree $-1$ Poisson bracket\footnote{
The grading convention that we use here, with $(,)$ and $\Delta$ of degree $-1$, is adapted to BV algebras arising from 2d TCFT. In the setting where BV algebras originally appeared -- Batalin-Vilkovisky quantization of gauge theories -- the natural convention is to assign degree $+1$ to $(,)$ and $\Delta$ (the same degree as the operator $Q$, whereas in TCFT the degrees are opposite to the degree of $Q$).
}
 (or ``BV bracket,'' or ``antibracket'')
\begin{equation}
(,)\colon W\otimes W \ra W
\end{equation}
which is a derivation in both slots and satisfies (graded) Jacobi identity.
\item A degree $-1$ operator $\Delta\colon W\ra W$ (the ``BV Laplacian'')
satisfying the Leibniz identity for a second-order differential operator
\begin{equation}
\Delta(xyz)\pm \Delta(xy)z\pm \Delta(xz)y \pm \Delta(yz)x \pm xy \Delta(z)\pm xz\Delta(y) \pm yz \Delta(x) =0
\end{equation}
and the properties
\begin{gather}
\Delta(\mathbb{1})=0,\\
\Delta(xy)=\Delta(x)y+(-1)^{|x|}x \Delta(y) + (-1)^{|x|}(x,y).
\end{gather}
In particular, the BV bracket arises as the defect of the first order Leibniz identity for $\Delta$.
\end{itemize}
\end{definition}

In the setting of TCFT, we consider the graded vector space $W=H_Q(V)$ -- the cohomology of $Q$ (with grading by the ``ghost number''), the unit element is $\mathbb{1}$ -- the cohomology class of the identity field. 

The supercommutative product on $W$ is given by OPEs. Notice that if $\Phi_1,\Phi_2$ are two nontrivial $Q$-cocycles, %then all terms in their OPE are $Q$-cocycles and the terms that 
the OPE has the form
\begin{equation}
\Phi_1(w) \Phi_2(z)\sim \sum_\Phi (w-z)^{h(\Phi)}(\bar{w}-\bar{z})^{\bar{h}(\Phi)} \Phi(z)
\end{equation}
where we used that $\Phi_{1,2}$ must have conformal weight $(0,0)$  and used Lemma \ref{l26 lemma OPE exponents}. Terms in the right hand side of the OPE must also be $Q$-closed, and the ones containing nontrivial $Q$-cocycles $\Phi$ have to contribute with exponents $h(\Phi)=\bar{h}(\Phi)=0$. Therefore, for $\Phi_1,\Phi_2$ two $Q$-cocycles one has an OPE of the form
\begin{equation}\label{l33 Q-closed OPE}
\Phi_1(w)\Phi_2(z)\sim (\Phi_1\cdot \Phi_2)(z)\quad\mr{modulo}\;Q\mbox{-exact terms.}
\end{equation}
with $\Phi_1\cdot \Phi_2$ some $Q$-cocycle.
Thus, in $Q$-cohomology OPE, is always constant and induces a supercommutative product. 

The BV bracket is given by the following construction:
for $\Phi_1=\Phi_1^{(0)}$, $\Phi_2=\Phi_2^{(0)}$ two $Q$-cocycles, we set
\begin{equation}\label{l33 (,)}
(\Phi_1,\Phi_2)(z)=\frac{1}{2\pi i}\oint_{\gamma_z} \Phi_1^{(1)}(w)\Phi_2^{(0)}(z),
\end{equation}
where $\gamma_z$ is a contour around $z$ and $\Phi_1^{(1)}=\Gamma \Phi_1^{(0)}$ is the first descent of $\Phi_1$.

The BV Laplacian is constructed as the operator
\begin{equation}\label{l33 G_0,-}
\Delta\colon= G_0-\ol{G}_0,
\end{equation}
also denoted $G_{0,-}$, where $G_0,\ol{G}_0$ are particular mode operators of $G$, cf. (\ref{l33 G modes}).

We refer to \cite{LMY1} for an example of a TCFT with explicitly computed BV algebra structure on $Q$-cohomology.

\subsection{Action of the operad of framed little disks on $V$}
The BV algebra structure on $Q$-cohomology $H_Q(V)$ has a ``lift'' to the full space of fields $V$, as an ``algebra over the operad of framed little 2-disks.''

\begin{definition}
The operad of framed little 2-disks $E_2^\mr{fr}$ is a sequence of manifolds $(E_2^\mr{fr})_n$, where $(E_2^\mr{fr})_n$ is the space of configurations of $n\geq 0$ disjoint disks inside a unit disk in $\RR^2\simeq \CC$, each disk is equipped with a ``framing'' -- a marked point on the boundary circle.\footnote{In particular, $(E_2^\mr{fr})_n$ is a manifold of real dimension $4n$, parameterized by positions of centers of the $n$ disks, $n$ radii and $n$ angles (of the marked point).}
The marked point on the unit circle is fixed at $(1,0)$. 

One has composition maps 
\begin{equation}
\circ_i\colon (E_2^\mr{fr})_n \times (E_2^\mr{fr})_m \ra (E_2^\mr{fr})_{n+m-1}
\end{equation}
for $i=1,\ldots,n$, defined as follows.
A configuration of disks $o_2\in (E_2^\mr{fr})_m$ is
scaled and rotated so that its outer disk fits with $i$-th disk in  the configuration $o_1\in (E_2^\mr{fr})_n$ (and the marked points should coincide). Then the new configuration $o_1 \circ_i o_2$ consists of the rescaled/rotated configuration $o_2$ and all disks of $o_1$ except the $i$-th disk.
\end{definition}

It is convenient to think of a configuration of disks as ``holes'' in the unit disk. Then the composition map fits one $m$-holed inside the $i$-th hole of another $n$-holed disk.

%\textcolor{red}{PICTURE}
\begin{figure}[H]
%$$\vcenter{\hbox{ \includegraphics[scale=0.8]{l26_contours_1.eps} }} 
%\longrightarrow
%\vcenter{\hbox{ \includegraphics[scale=0.8]{l26_contours_2.eps} }} 
%$$
\begin{center}
\includegraphics[scale=0.8]{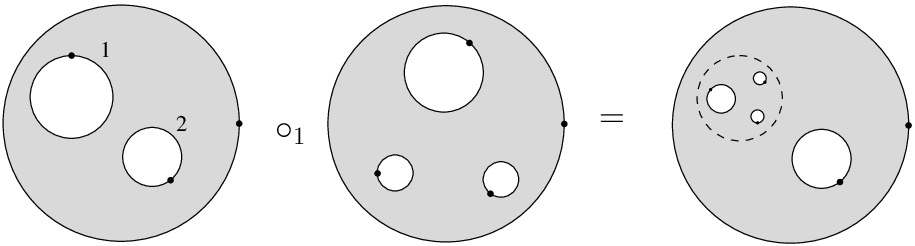}
\end{center}
\caption{Composition in the operad of framed little 2-disks.}
\end{figure}

Given a TCFT, one can construct a sequence of differential forms $\omega_n$ on $(E_2^\mr{fr})_n$ valued in $\mr{Hom}(V^{\otimes n},V)$, for $n\geq 1$, defined by
\begin{equation}\label{l33 form on E_2^fr}
\omega_n (\Phi_1,\ldots,\Phi_n)= \prod_{k=1}^n e^{\zeta_k L_0+\bar\zeta_k \ol{L}_0+ d\zeta_k G_0 + d\bar\zeta_k \ol{G}_0}e^\Gamma\Phi_k(z_k).
\end{equation}
Here $z_k$ are positions of the centers of disks, $\zeta_k=\log r_k +i\theta_k$, with $r_k$ the radii and $\theta_k$ the angles; $\Gamma$ is the descent operator (\ref{l33 Gamma}). The expression in the r.h.s. of (\ref{l33 form on E_2^fr}) is to be understood under a correlator with an arbitrary collection of test fields inserted outside the unit disk. Thus, the r.h.s. of (\ref{l33 form on E_2^fr}) is a ``multi-OPE.'' 

One has the property 
\begin{equation} \label{l33 (d-Q)omega=0}
(d-\mr{ad}_Q)\omega_n=0
\end{equation}
where $\mr{ad}_Q$ means the sum of terms  where $Q$  on an input or the output field of $\omega_n$. More explicitly,
\begin{multline}
(d-\mr{ad}_Q)\omega_n(\Phi_1,\ldots,\Phi_n)\colon= \\
=d\omega_n(\Phi_1,\ldots,\Phi_n)-Q\omega_n(\Phi_1,\ldots,\Phi_n)
+\sum_{k=1}^n \pm\omega_n(\Phi_1,\ldots,Q\Phi_k,\ldots, \Phi_n) =0
\end{multline}
This property  is a consequence of (\ref{l33 (d-Q)e^Gamma}). 

The property (\ref{l33 (d-Q)omega=0}) implies that one has a map of cochain complexes, from singular chains of the framed little disk operad to multilinear operators on $V$:
\begin{equation}
\begin{array}{ccc}
C_{-\bt}((E_2^\mr{fr})_n) &\ra& \mr{Hom}(V^{\otimes n},V)\\
\mr{chain} &\mapsto & \Big(\Phi_1\otimes\cdots\otimes\Phi_n \mapsto \int_\mr{chain} \omega_n(\Phi_1,\ldots,\Phi_n) \Big)
\end{array}
\end{equation}
Note that we put the reverse grading on chains, so that singular chains are seen as a cochain complex.
This map is a representation of the operad, i.e., is compatible with compositions.

In particular, passing to cohomology, we obtain a map from homology of $(E_2^\mr{fr})_n$ to $\mr{Hom}(W^{\otimes n},W)$, where $W=H_Q(V)$.

%\marginpar{What is the original reference?}
Here is a known fact (see \cite{Getzler_BV}): homology of the operad $E_2^\mr{fr}$ is the operad of BV algebras, with generators $\mathbb{1},\cdot, (,), \Delta$ (subject to the relations as in Definition \ref{l33 def: BV algebras}).\footnote{We think of the unit as an operation of ``arity zero,'' $\mathbb{1}\in\mr{Hom}(V^{\otimes 0},V)\simeq V$.}

More explicitly, the homology of $E_2^\mr{fr}$ is generated (using compositions $\circ_i$) by four homology classes:
\begin{enumerate}[(i)]
\item The tautological $0$-class in $H_0((E_2^\mr{fr})_0)$.
\item The 0-class in $H_0((E_2^\mr{fr})_2)$, represented by any configuration of two disjoint disks in the unit disk.
\item The 1-class in $H_1((E_2^\mr{fr})_2)$, represented by one disk moving a full circle around the other disk.
\item The 1-class in $H_1((E_2^\mr{fr})_1)$, represented by rotating the disk (or equivalently rotating the marked point on the boundary) a full circle.
\end{enumerate}
%\textcolor{red}{PICTURES}
\begin{figure}[H]
%$$\vcenter{\hbox{ \includegraphics[scale=0.8]{l26_contours_1.eps} }} 
%\longrightarrow
%\vcenter{\hbox{ \includegraphics[scale=0.8]{l26_contours_2.eps} }} 
%$$
\begin{center}
\includegraphics[scale=0.8]{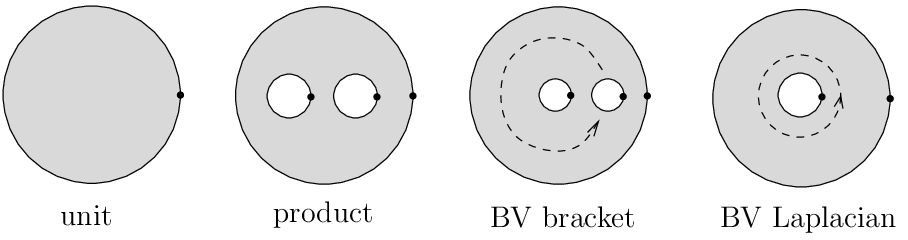}
\end{center}
\caption{Generators of homology of $E_2^\mr{fr}$.}
\end{figure}

These classes correspond to the elements $\mathbb{1},\cdot, (,), \Delta$ of the BV operad and are represented on $W$ by the corresponding operations: cohomology class of the unit field, (\ref{l33 Q-closed OPE}), (\ref{l33 (,)}), (\ref{l33 G_0,-}), and one can see that those formulae follow from integrating the form (\ref{l33 form on E_2^fr}) over the respective cycles. Thus, indeed, the BV algebra structure on $Q$-cohomology that we constructed in Section \ref{sss: BV algebra on Q-cohomology} is induced from the $E_2^\mr{fr}$-algebra structure on $V$ via passing to cohomology.

For details on the operadic viewpoint on TCFTs, we refer the reader to \cite{Getzler_BV}. For an explicit example, we refer to \cite{LMY1}.

\chapter{Elements of representation theory of Virasoro algebra}
\chaptermark{Representation theory of Virasoro algebra}

\marginpar{Lecture 34,\\ 11/16/2022}
\label{ch Vir rep theory}
\section{Verma modules of Virasoro algebra, null vectors}
Let $V_{c,h}$ be the Verma module\footnote{\label{l34 footnote: Verma}
Given any $\ZZ$-graded Lie algebra $A=A_\bt$, one defines the Verma module as follows. Let $W$ be a module over $A_{\geq 0}$ where $A_{>0}$ acts by zero. Then the Verma module is the $U(A)$-module induced from $W$, i.e., $U(A)\otimes_{U(A_{\geq 0})} W$, where $U(\cdots)$ is the enveloping algebra.
} of Virasoro algebra with central charge $c\in \CC$ and highest weight $h\in \CC$, i.e., it is generated by the highest weight vector which we denote $|h\rangle$ which satisfies
\begin{equation}
L_{>0} |h\rangle =0,\quad L_0 |h\rangle = h|h\rangle
\end{equation}
-- is killed by the positive part of the Virasoro algebra and is an eigenvector for $L_0$ with eigenvalue $h$. The Verma module is then
\begin{equation}
V_{c,h}=\mr{Span}_\CC \{L_{-n_r}\cdots L_{-n_1}|h\rangle\;|\; 1\leq n_1\leq \cdots\leq n_r,\; r\geq 0\}
\end{equation}
The descendant
\begin{equation}\label{l34 descendant}
L_{-n_r}\cdots L_{-n_1}|h\rangle
\end{equation}
has conformal weight ($L_0$-eigenvalue) $h+N$ where $N=n_1+\cdots+n_r$. One says that (\ref{l34 descendant}) is a ``level-$N$'' vector in $V_{c,h}$. One has a splitting of $V_{c,h}$ by level:
\begin{equation}
V_{c,h}=\bigoplus_{N\geq 0} V_{c,h}^N,
\end{equation}
where $V_{c,h}^N$ is the subspace of the Verma module spanned by level-$N$ vectors (i.e., it is the $(h+N)$-eigenspace of $L_0$). Note that the dimension of $V_{c,h}^N$ is
\begin{equation}
\dim V_{c,h}^N= P(N)
\end{equation}
-- the number of partitions of $N$ (cf. Section \ref{sss torus Z for compactified free boson}).

There is a unique sesquilinear form $\langle,\rangle$ on $V_{c,h}$ characterized by the properties
\begin{eqnarray}
\langle h|h\rangle &=& 1 \\
(L_n)^+ &=& L_{-n},\quad n\in \ZZ
\end{eqnarray}
Generally, $\langle,\rangle$ is not positive-definite and may be degenerate.

\begin{definition}
A vector $|\chi\rangle\neq |h\rangle$ in $V_{c,h}$ is called a ``singular vector'' or ``null vector''  if it satisfies
\begin{equation}\label{l34 L_>0 chi=0}
L_{>0} |\chi\rangle =0.
\end{equation}
\end{definition}

Note that a null vector is automatically orthogonal to the entire $V_{c,h}$, since one has
\begin{equation}\label{l34 chi perp to V}
\big\langle L_{-n_r}\cdots L_{-n_1}|h\rangle  ,|\chi\rangle\big\rangle = \langle h| L_{n_1}\cdots \underbrace{L_{n_r} |\chi\rangle}_{=0\;\mr{since\;}n_r>0} =0.
\end{equation}
In particular, a null vector has zero norm:
\begin{equation}
\langle \chi|\chi \rangle =0.
\end{equation}

Assume that there exists a null vector $|\chi\rangle$ at level $N$ in $V_{c,h}$.
Then Virasoro descendants of $|\chi\rangle$ form a submodule of $V_{c,h}$ isomorphic to the Verma module $V_{c,h+N}$:
\begin{equation}
\underbrace{\mr{Span}\{L_{-n_r}\cdots L_{-n_1}|\chi\rangle\}}_{\simeq V_{c,h+N}} \subset V_{c,h}
\end{equation}
In fact, this entire submodule is orthogonal to $V_{c,h}$, by an argument similar to (\ref{l34 chi perp to V}).

Let us consider when null vectors can appear at small levels $N$ (the full answer for general $N$ is given by Kac determinant formula in Section \ref{ss Kac determinant formula} below).
\begin{example} \label{l34 ex null vector at N=1}
Assume that $V_{c,h}$ contains a null vector at level $N=1$. That means $|\chi\rangle = L_{-1}|h\rangle$ (ignoring a possible normalization factor). Note that $L_{\geq 2}|\chi\rangle$ is a vector at level $-1$, so it automatically vanishes. The only case of (\ref{l34 L_>0 chi=0}) that needs checking is $L_1 |\chi\rangle$ :
\begin{equation}
L_1|\chi\rangle = L_1 L_{-1}|h\rangle =(2L_0-L_{-1}\underbrace{L_1)|h\rangle}_0 = 2h|h\rangle.
\end{equation}
Thus, $|\chi\rangle=L_{-1}|h\rangle$ is a null vector if and only if $h=0$.
\end{example}

\begin{example} Assume that $V_{c,h}$ contains a null vector at level $N=2$. This means
\begin{equation}
|\chi\rangle = (\alpha L_{-2}+\beta L_{-1}^2)|h\rangle
\end{equation}
with $\alpha,\beta\in\CC$ not simultaneously zero. By the same argument as above, $L_{\geq 3}|\chi\rangle$ vanishes automatically, so we only need to check $L_1|\chi\rangle$ and $L_2|\chi\rangle$.
We have
\begin{multline}
L_1(\alpha L_{-2}+\beta L_{-1}^2)|h\rangle = (\alpha 3 (L_{-1}+L_{-2}L_1)+\beta (2L_0 L_{-1}+L_{-1}L_1 L_{-1}) )|h\rangle=\\
=
(\alpha 3 L_{-1} +\beta (2L_{-1}+2 L_{-1}L_0+2L_{-1}L_0+L_{-1}^2 L_1))|h\rangle=
(3\alpha + (4h+2)\beta)|h\rangle,
\end{multline}
\begin{multline}
L_2(\alpha L_{-2}+\beta L_{-1}^2)|h\rangle = (\alpha (4L_0+\frac{c}{2}+L_{-2}L_2)+\beta(3L_1 L_{-1}+L_{-1}L_2L_{-1}))|h\rangle
=\\
=(\alpha (4h+\frac{c}{2})+\beta(6L_0+3L_{-1}L_1+3 L_{-1} L_1+L_{-1}^2 L_2 ))|h\rangle = ( (4h+\frac{c}{2})\alpha+ 6h \beta) |h\rangle.
\end{multline}
So, the equations on a null vector $L_1 |\chi\rangle=0$, $L_2|\chi\rangle$ are equivalent to a homogeneous system of two linear equations on two coefficients $\alpha,\beta$, 
\begin{equation}
3\alpha + (4h+2)\beta =0,\quad (4h+\frac{c}{2})\alpha+6h\beta =0,
\end{equation}
which has a nonzero solution if and only if the determinant of the coefficient matrix vanishes,
\begin{equation}
\left| 
\begin{array}{cc}
3 & 4h+2 \\
\frac{c}{2}+4h & 6h
\end{array}
\right| =0.
\end{equation}
This is a nontrivial quadratic relation on $c$ and $h$, and as we just showed, $V_{c,h}$ contains a null vector at level $N=2$ if and only if this relation is satisfied. 

For instance, this relation is satisfied for $c=\frac12, h=\frac{1}{16}$, which is what allowed us to find a hypergeometric equation on the four-point correlator of fields $\bbsigma$ in the free fermion CFT in Section \ref{sss free fermion correlators}.
\end{example}

\section{Kac determinant formula}\label{ss Kac determinant formula}

Consider the ``Gram matrix'' -- the matrix of inner products of level-$N$ descendants of the highest vector $|h\rangle$ in $V_{c,h}$:
\begin{equation}\label{l34 Gram matrix}
M^{(N)}=(\langle i|j \rangle)_{i,j}
\end{equation}
where $i,j$ run over the basis of vectors (\ref{l34 descendant}) in $V_{c,h}$. In particular, $M^{(N)}$ is a matrix of size $P(N)\times P(N)$.

\begin{thm}[Kac \cite{Kac}, Feigin-Fuchs \cite{Feigin-Fuchs}]
The determinant of the Gram matrix  (\ref{l34 Gram matrix}) is
\begin{equation}\label{l34 Kac formula}
\det M^{(N)} =\alpha_N \prod_{
\begin{array}{c}
p,q\geq 1\;\mr{s.t.}\\
pq\leq N
\end{array}
} (h-h_{p,q}(c))^{P(N-pq)}.
\end{equation}
Here
\begin{equation}
\alpha_N=\prod_{
\begin{array}{c}
p,q\geq 1\;\mr{s.t.}\\
pq\leq N
\end{array}
} ((2p)^q q!)^{P(N-pq)-P(N-p(q+1))}
\end{equation}
is a numerical factor and 
\begin{equation}\label{l34 h_p,q}
h_{p,q}(c)=\frac{((m+1)p-mq)^2-1}{4m(m+1)},
\end{equation}
where $m$ is related to the central charge $c$ by
\begin{equation}\label{l34 m via c}
m=-\frac12 \pm \sqrt{\frac{25-c}{1-c}}
\end{equation}
or equivalently
\begin{equation}\label{l34 c via m}
c=1-\frac{6}{m(m+1)}.
\end{equation}
\end{thm}

The importance of Kac determinant formula (\ref{l34 Kac formula}) is that says for which $c,h$ the Gram matrix at level $N$ vanishes, which means that $V_{c,h}$ contains a null vector at level $\leq N$. More precisely, Kac formula implies the following:
\begin{corollary}\label{l34 Corollary}
 If %and only if
 $h=h_{p,q}$ (as defined by (\ref{l34 h_p,q}))  for some integers $p,q\geq 1$ then $V_{c,h}$ contains a null vector at level $N=pq$.
\end{corollary}
In fact, null, vectors at other levels may also appear (i.e. this is not an ``if and only if'' statement), however every null vector in $V_{c,h}$ is either covered by Corollary \ref{l34 Corollary} or is a descendant of one.

\begin{example}
For any $c$ from (\ref{l34 h_p,q}) we have $h_{1,1}=0$. This corresponds to the fact that $V_{c,0}$ has a null vector at level $N=1$ for any $c$, cf. Example \ref{l34 ex null vector at N=1}.
\end{example}

\begin{example} Consider the case of central charge $c=1$. By (\ref{l34 m via c}), (\ref{l34 c via m}) it corresponds to the limiting case $m\ra \infty$. In this limit, (\ref{l34 h_p,q}) becomes
\begin{equation}
h_{p,q}=\frac{(p-q)^2}{4}.
\end{equation}
This implies that for $c=1$, $h=\frac{n^2}{4}$, with $n=0,1,2,3,\ldots$, the Verma module $V_{1,h}$ contains an infinite sequence of null vectors at levels $N=p\underbrace{(n+p)}_q$, with  $p=1,2,3,\ldots$, since $h=\frac{n^2}{4}$ equals $h_{p,q}$ for an infinite sequence of pairs $(p,q)$ of the form $(p,n+p)$.
\end{example}

The following is (a part of) a theorem of Feigin-Fuchs \cite{Feigin-Fuchs}.
\begin{thm}[Feigin-Fuchs]
\begin{itemize}
\item $V_{c,h}$ is irreducible if and only if it contains no null vectors and is reducible if and only if $h=h_{p,q}$ for some integers $p,q\geq 1$.
\item Proper submodules of $V_{c,h}$ are generated by null vectors.
\item The irreducible highest weight module $M_{c,h}$ for Virasoro algebra at central charge $c$ and with highest weight $h$ has the form
\begin{equation}
M_{c,h}=V_{c,h}/\mathsf{N}
\end{equation}
where $\mathsf{N}\subset V_{c,h}$ is the maximal proper submodule.
It can also be realized as the kernel of the sesquilinear form $\langle ,\rangle$ on $V_{c,h}$ or equivalently the orthogonal complement of $V_{c,h}$:
\begin{equation}
\mathsf{N}=\mr{ker} \langle,\rangle = V_{c,h}^\perp.
\end{equation}
\end{itemize}
\end{thm}

%\marginpar{Maybe recall the full Fegin-Fuchs theorem here?}
In Section \ref{ss FF: maps between Verma modules} we recall the second part of Feigin-Fuchs theorem giving the full classification of maps (inclusions) between Verma modules at a given $c$, which in particular yields formulae for characters of $M_{c,h}$ and therefore formulae for torus partition functions in a CFT, see Section \ref{sss Virasoro characters} and (\ref{l35 torus Z for M(P,Q)}).
%\footnote{The character of a representation $W$ of Virasoro algebra is defined as $\chi_W(q)=\tr_W q^{L_0-\frac{c}{24}}$; in particular $\chi_{M_{c,h}}$ arises as a building block for genus one partition function of any CFT containing $M_{c,h}$ as a part of the space of states.}

%\marginpar{I am being redundant here, this was discussed before..}
\begin{example}\label{l34 ex c=1 two sequences}
 For $c=1$, $V_{1,h}$ is reducible iff $h=\frac{n^2}{4}$ for some $n=0,1,2,\ldots$. In particular, for $h\neq \frac{n^2}{4}$, one has $M_{1,h}=V_{1,h}$. Reducible Verma modules for $c=1$ arrange into two sequences connected by inclusions of modules:
\begin{equation}\label{l34 c=1 two sequences}
\begin{gathered}
V_{1,0}\leftarrow V_{1,1}\leftarrow V_{1,4}\leftarrow V_{1,9}\leftarrow \cdots\\
V_{1,\frac14} \leftarrow V_{1,\frac94} \leftarrow V_{1,\frac{25}{4}} \leftarrow V_{1,\frac{49}{4}}\leftarrow \cdots
\end{gathered}
\end{equation}
Irreducible modules for these values of $h$ are obtained by taking the corresponding Verma module and quotienting out the module mapping into it. E.g., $M_{1,0}=V_{1,0}/V_{1,1}$. Null vectors in $V_{1,h}$ are images of highest vectors of Verma modules to the right in the respective sequence, i.e., mapping into $V_{1,h}$, possibly via a sequence of inclusions.
\end{example}

\begin{example}
Consider the case $c=\frac12$, which corresponds to $m=3$. One has \begin{equation}
h_{p,q}= \frac{(4p-3q)^2-1}{48}.
\end{equation}
The values of $h_{p,q}$ for small $p,q$ are the following.

\begin{center}
\begin{tabular}{c|cc}
$h_{p,q}$ 
& $p=1$ & $p=2$ \\ \hline 
$q=1$ & $0$ &  $\frac12$ \\
$q=2$ & $\frac{1}{16}$ & $\frac{1}{16}$ \\
$q=3$ & $\frac{1}{2}$ & $0$
\end{tabular}
\end{center}

We recognize these numbers $h_{1,1}=0$, $h_{2,1}=\frac12$, $h_{1,2}=\frac{1}{16}$ as precisely the conformal weights $h$ of primary field $\mathbb{1}$, $\psi$, $\sigma$ in the free fermion CFT. Thus, the corresponding conformal families $M_{\frac12,h}$ are the ones coming from Verma modules $V_{\frac12,h}$ containing a null vector, which allows one to write differential equations on correlators of the corresponding primary fields, as we did in Section \ref{sss free fermion correlators}.
\end{example}

\section{Maps between Verma modules}\label{ss FF: maps between Verma modules}
In this section we follow Feigin-Fuchs \cite{Feigin-Fuchs}.

Fix the central charge $c,h\in \RR$. Equation $h_{p,q}=h$ with $h_{p,q}$ defined by (\ref{l34 h_p,q}) determines two parallel lines on the $(p,q)$-plane related to one another by reflection $(p,q)\leftrightarrow (-p,-q)$. Pick one of those lines and denote it $l_{c,h}$. The slope of $l_{p,q}$ is $\frac{m+1}{m}$, in particular:
\begin{itemize}
\item If $c\leq 1$, the line is real, with positive slope. For $c=1$ the slope is $1$.
\item If $c\geq 25$, the line is real, with negative slope. For $c=25$, the slope is $-1$.
\item If $1<c<25$, the slope (and the line) is complex.
\end{itemize}

One is interested in integer points on $l_{p,q}$. The relevant cases (with nomenclature taken from \cite{Feigin-Fuchs}) is:
\begin{enumerate}
\item[I] $l_{c,h}$ has no integer points.
\item[II] $l_{c,h}$ has a single integer point $(a',a'')\in \ZZ^2$. One distinguishes the following subcases:
\begin{enumerate}
\item[$\mr{II}_+$] $a'a''>0$,
\item[$\mr{II}_0$] $a'a''=0$,
\item[$\mr{II}_-$] $a'a''<0$.
\end{enumerate}
\item[III] $l_{c,h}$ contains infinitely many integer points. In this case $m\in\mathbb{Q}$ and one has either $c\leq 1$ (subcase $\mr{III}_-$) or $c\geq 25$ (subcase $\mr{III}_+$). We further distinguish between the subcases according to whether $l_{p,q}$ intersects the coordinate axes at integer points.
\begin{enumerate}
\item[$\mr{III}^{0,0}_\pm$] $l_{c,h}$ intersects both coordinate axes $q=0$ and $p=0$ at integer points. Denote $P$ the middle point of the interval connecting these two intersection points. Enumerate the integer points of the upper half of $l_{c,h}$ (above $P$) as
\begin{equation}\label{l34 FF *}
\ldots, (a'_{-1},a''_{-1}), (a'_0,a''_0), (a'_1,a''_1),\ldots
\end{equation}
in such order that one has 
\begin{equation}\label{l34 FF **}
\cdots < a'_{-1}a''_{-1} < a'_0 a''_0=0 < a'_1a''_1 <\cdots
\end{equation}
In particular, in the case $\mr{III}_{-}^{0,0}$, the sequence (\ref{l34 FF *}) is finite on the left and infinite on the right, and vice versa in the case $\mr{III}_{+}^{0,0}$.
\item[$\mr{III}^0_\pm$] $l_{c,h}$ intersects one of the coordinate axes at an integer point. Then we enumerate all integer points of $l_{c,h}$ (not just half) as in (\ref{l34 FF *}), so that (\ref{l34 FF **}) holds.
\item[$\mr{III}_\pm$] $l_{c,h}$ intersects both coordinate axes at non-integer points. Then we enumerate all integer points of $l_{c,h}$ as in (\ref{l34 FF *}), so that
\begin{equation}\label{l34 FF III a order}
\cdots < a'_{-1}a''_{-1}<0<a'_0a''_0<a'_1a''_1<\cdots
\end{equation}
We also draw a second line $l'_{c,h}$ through the point $(-a'_0,a''_0)$ parallel to $l_{c,h}$ and enumerate its integer points as
\begin{equation}
\ldots, (b'_{-1},b''_{-1}), (b'_0,b''_0)=(-a'_0,a''_0), (b'_1,b''_1),\ldots
\end{equation}
so that one has
\begin{equation}\label{l34 FF III b order}
\cdots < b'_{-1}b''_{-1}<b'_0b''_0=-a'_0a''_0<0<b'_1b''_1 < \cdots
\end{equation}
\end{enumerate}
\end{enumerate}

\begin{thm}[Feigin-Fuchs]\label{l34 thm FF maps of Verma modules}
Fix $c,h\in\RR$ and a line $l_{c,h}$ as above. Then:
\begin{itemize}
\item In cases $\mr{I}$, $\mr{II}_0$: $V_{c,h}$ is irreducible and not a proper submodule of any Verma module.
\item In the case $\mr{II}_+$, $V_{c,h}$ has a single Verma submodule isomorphic to $V_{c,a'a''}$ (which is irreducible) and generated a by a null vector at level $a'a''$; $V_{c,h}$ is not a proper submodule of any Verma module. 
\item In the case $\mr{II}_-$, $V_{c,h}$ is irreducible but can be  embedded into $V_{c,h+a'a''}$ and is generated there by a null vector at level $-a'a''$. $V_{c,h}$ cannot be embedded into any other Verma module.
\item In the cases $\mr{III}^{0,0}_\pm$, $\mr{III}^{0}_\pm$, there is a sequence of embeddings
\begin{equation}\label{l34 FF III^0}
\cdots \ra V_{c,h+a'_1a''_1} \ra V_{c,h} \ra V_{c,h+a'_{-1}a''_{-1}}\ra \cdots
\end{equation}
Modules in this sequence are not related by morphisms with any Verma modules not from this sequence.
\item In the cases $\mr{III}_\pm$ there is a commutative diagram of embeddings of Verma modules
\begin{equation}\label{l34 FF III diagram}
\xymatrix{
%\vdots & \vdots \\
\stackrel{\vdots}{V_{c,h+a'_{-1}a''_{-1}}}  & \stackrel{\vdots}{V_{c,h+a'_{-2}a''_{-2}}}\\
V_{c,h}\ar[u]\ar[ur]&V_{c,h+b'_{-1}b''_{-1}+a'_0a''_0}\ar[u]\ar[ul] \\
V_{c,h+a'_1a''_1}\ar[u]\ar[ur] & V_{c,h+a'_0a''_0}\ar[u]\ar[ul] 
\\
\underset{\vdots}{V_{c,h+b'_2b''_2+a'_0a''_0}}\ar[u]\ar[ur] & \underset{\vdots}{V_{c,h+b'_{1}b''_{1}+a'_0a''_0}}\ar[u]\ar[ul]
}
\end{equation}
Modules in this diagram are not connected by homomorphisms with any other Verma modules. In each piece of the form
\begin{equation}
\xymatrix{
A & \\
B \ar[u] & C \ar[ul] \\
D \ar[u] \ar[ur] & E \ar[ul] \ar[u]
}
\end{equation}
the images of $B$ and $C$ in $A$ do not contain each other and their intersection is generated by images of $D$ and $E$ in $A$.
\end{itemize}
\end{thm}

\begin{example}
The case $c=1$, $h=0$ corresponds is $\mr{III}_-^{0,0}$ in Feigin-Fuchs classification and $c=1$, $h=\frac{n^2}{4}$ with $n=1,2,\ldots$ is $\mr{III}_-^0$. In these cases the sequence (\ref{l34 FF III^0}) is one of the two sequences (\ref{l34 c=1 two sequences}).
\end{example}

\begin{example}
For $c=\frac12$, $h=0$ we have the line $l_{\frac12,0}=\{(p,q)\;|\; 4p-3q=1\}$ corresponding to the case $\mr{III}_-$. The integer points on $l_{\frac12,0}$ are $(1+3k,1+4k)$ with $k\in\ZZ$; arranged in the order  (\ref{l34 FF III a order}) %(\ref{l34 FF **}) 
they are:  
\begin{equation}
\begin{array}{c|cccccc}
n & 0&1 & 2& 3&4&\cdots \\ \hline
(a'_n,a''_n) & (1,1) & (-2,-3) &(4,5)&(-5,-7)&(7,9)&\cdots
\end{array}
%(a'_0,a''_0)=(1,1),(-2,-3),(4,5),(-5,-7),(7,9),\ldots
\end{equation} 
The parallel line $l'_{\frac12,0}=\{(p,q)\;|\; 4p-3q=-7\}$, it has integer points $(-1+3k,1+4k)$, $k\in\ZZ$; arranged in the order (\ref{l34 FF III b order}) they are:
\begin{equation}
\begin{array}{c|cccccc}
n & 0&1 & 2& 3&4&\cdots \\ \hline
(b'_n,b''_n) & (-1,1) & (2,5) &(-4,-3)&(5,9)&(-7,-7)&\cdots
\end{array}
\end{equation} 
The diagram of embeddings (\ref{l34 FF III diagram}) becomes
\begin{equation}\label{l34 V_1/2,0 FF diagram}
\xymatrix{
V_{\frac12,0} & \\
V_{\frac12,6} \ar[u] & V_{\frac12,1} \ar[ul] \\
\underset{\vdots}{V_{\frac12, 11}} \ar[u] \ar[ur] & \underset{\vdots}{V_{\frac12,9}} \ar[u] \ar[ul]
}
\end{equation}
One has similar diagrams for $h=\frac12$ and for $h=\frac{1}{16}$ (all with $c=\frac12$).
\end{example}

\subsection{Characters of highest weight modules of Virasoro algebra}\label{sss Virasoro characters}
Given a module $W$ of Virasoro algebra with central charge $c$ is defined as
\begin{equation}\label{l34 character}
\chi_W(\qq)=\tr_W \qq^{L_0-\frac{c}{24}},
\end{equation}
with $\qq$ a complex parameter with $|\qq|<1$. For a Verma module $V_{c,h}$, one has
\begin{equation}
\chi_{V_{c,h}}(\qq)=\sum_{N\geq 0}P(N) \qq^{h+N-\frac{c}{24}} =\frac{\qq^{h+\frac{1-c}{24}}}{\eta(\tau)},
\end{equation}
where $P(N)$ is the number of partitions and $\eta(\tau)$ is the Dedekind eta-function; $\qq$ is related to $\tau\in \Pi_+$ by
\begin{equation}
\qq=e^{2\pi i \tau}.
\end{equation}

Characters of irreducible highest weight modules $M_{c,h}$ can be obtained using Theorem \ref{l34 thm FF maps of Verma modules}.

\begin{example}\label{l34 ex character of M_1/2,0}
The character of the irreducible module $M_{\frac12,0}$ can be obtained from the diagram (\ref{l34 V_1/2,0 FF diagram}):
\begin{multline}
\chi_{M_{\frac12,0}}(\qq)=\chi_{V_{\frac12,0}}(\qq) - \chi_{V_{\frac12,1}}(\qq) - \chi_{V_{\frac12,6}}(\qq) + \chi_{V_{\frac12,9}}(\qq)+\chi_{V_{\frac12,11}}(\qq)-\cdots =\\
=
\frac{\qq^\frac{1}{48}}{\eta(\tau)}(1-\qq-\qq^6+\qq^9+\qq^{11}-\cdots)=
\frac{\qq^{\frac{1}{48}}}{\eta(\tau)} \sum_{k\in \ZZ}\left(
\qq^{1+(-1+3k)(1+4k)}-\qq^{(1+3k)(1+4k)}
\right)
\end{multline}
\end{example}

Characters of irreducible modules $M_{c,h}$ are the conformal blocks for the torus partition function in a CFT. If the space of fields of a CFT with central charge $c,\bar{c}$ contains primary fields $\Phi_i$ with $i\in I$ (the indexing set for primary fields), with conformal weights $(h_i,\bar{h}_i)$, then the space of states (or space of fields) is
\begin{equation}
\HH=\bigoplus_{i\in I} M_{c,h_i}\otimes M_{\bar{c},\bar{h}_i}
\end{equation}
and the torus partition function is
\begin{equation}\label{l34 torus Z via chi}
Z(\tau)= \tr_\HH \qq^{L_0-\frac{c}{24}} \bar{\qq}^{\ol{L}_0-\frac{\bar{c}}{24}}=
\sum_{i\in I} \chi_{M_{c,h_i}}(\qq)\chi_{M_{\bar{c},\bar{h}_i}}(\bar{\qq})
\end{equation}

\section{Minimal models of CFT}
\subsection{Unitary minimal models}

The following theorem is due to Friedan-Qiu-Shenker (1984) and Goddard Kent-Olive (1986).
\begin{thm}
The irreducible highest weight Virasoro module $M_{c,h}$ is unitary (i.e. the sesquilinear form  $\langle,\rangle$ is positive definite) if 
\begin{enumerate}[(a)]
\item either $c\geq 1$, $h\geq 0$,
\item \label{l34 Kac unitarity (b)} or $c=1-\frac{6}{m(m+1)}$ with $m=2,3,4,\ldots$ and $h=h_{p,q}$ with $1\leq p \leq m-1$, $1\leq q \leq m$.
\end{enumerate}
\end{thm}

Note that for $c$ as in (\ref{l34 Kac unitarity (b)}) above, one has a symmetry in the table of admissible $h_{p,q}$'s: 
\begin{equation}
h_{p,q}=h_{m-p,m+1-q}.
\end{equation}

\marginpar{Lecture 35,\\
11/18/2022}

Fix $m=2,3,4,\ldots$  
The ``minimal model'' $\MM(m,m+1)$ is defined\footnote{
We say ``defined'' a bit sloppily here. To have a CFT, the definition of the space of states as a $\mr{Vir}\oplus \ol{\mr{Vir}}$-module needs to be supplemented with extra data: OPEs of primary fields, or equivalently, structure coefficients of 3-point correlators of primary fields, which then allows to determine all correlators.
} as a CFT with central charge
\begin{equation}
c=\bar{c}=1-\frac{6}{m(m+1)}
\end{equation}
and space of states (or space of fields)
\begin{equation}\label{l35 H M(m,m+1)}
\HH=\bigoplus_{1\leq p \leq m-1,\; 1\leq q \leq m\;/\ZZ_2}  M_{c,h_{p,q}}\otimes \ol{M}_{c,h_{p,q}}.
\end{equation}
Here the sum is over pairs $(p,q)$ where the pairs $(p,q)$ and $(m-p,m+1-q)$ are understood as equivalent; notation ``$/\ZZ_2$'' above means that we should take one representative for each equivalence class. Each term in the sum in (\ref{l35 H M(m,m+1)}) is a representation of left and right Virasoro algebra, $\mr{Vir}\oplus \ol{\mr{Vir}}$, given as a tensor product of two copies of the same irreducible Virasoro module $M_{c,h_{p,q}}$; bar over the second copy of $M$ indicates that we see it as a module over the right (antiholomorphic) copy of Virasoro algebra.

\begin{example}
For $m=3$, the minimal model $\MM(3,4)$ is the CFT with central charge $\frac12$ and three species of primary fields: 
\begin{itemize}
\item 
$\Phi_{1,1}$ of conformal weight $h=\bar{h}=h_{1,1}=0$, 
\item $\Phi_{2,1}$ of conformal weight $h=\bar{h}=h_{2,1}=\frac12$, 
\item 
$\Phi_{1,2}$ of conformal weight $h=\bar{h}=h_{1,2}=\frac{1}{16}$. 
\end{itemize}
Comparing this to (\ref{l32 H non-chiral fermion}), we see that the space of states for $\MM(3,4)$ is the even part of the space of states of the free Majorana fermion, and we can identify the fields as
\begin{equation}
\Phi_{1,1}=\mathbb{1},\quad \Phi_{2,1}=\epsilon,\quad \Phi_{1,2}=\bbsigma
\end{equation}
-- the identity, ``energy'' and ``spin'' fields.

In particular, the CFT minimal model $\MM(3,4)$  corresponds to the  Ising model at critical temperature (at the point of second-order phase transition), and in particular correlators of $\bbsigma$ reproduce the correlators of spins in critical Ising model.

The selection rules (so-called ``fusion rules'') for OPEs are given by the following table
\begin{equation}
\begin{array}{c|ccc}
\times & [\mathbb{1}] & [\epsilon] & [\bbsigma] 
\\ 
\hline
[\mathbb{1}] &  [\mathbb{1}] & [\epsilon] & [\bbsigma] 
\\ {}
 [\epsilon] &  [\epsilon] & [\mathbb{1}] & [\bbsigma] 
\\ {}
[\bbsigma] & [\bbsigma] & [\bbsigma] & [\mathbb{1}] + [\epsilon]
\end{array}
\end{equation}
Here for $\Phi$  a primary field $[\Phi]$ stands for its conformal family (the span of all descendants, or equivalently, the corresponding term in the sum (\ref{l35 H M(m,m+1)})). For instance, the fusion rule $[\bbsigma]\times [\bbsigma]=[\mathbb{1}]+[\epsilon]$ means that in the r.h.s. of the r.h.s of the OPE of any two descendants of $\bbsigma$ one can find only descendants of $\mathbb{1}$ and of $\mathbb{\epsilon}$. We will comment later on where these selection rules for OPEs come from, see Remark \ref{l35 rem: fusion rules}. 
\end{example}

\begin{example}
Case $m=2$ is the ``trivial CFT'' with $c=0$ and a single conformal family with  $h=\bar{h}=h_{1,1}=0$:
\begin{equation}
\HH=M_{0,0}\otimes \ol{M}_{0,0}.
\end{equation}
In fact the irreducible Virasoro module $M_{0,0}$ consists of just the highest vector $|\vac\rangle$ (or $\mathbb{1}$) and all its descendants are zero. 
\end{example}

\subsection{General minimal models}
Let $c=1-\frac{6}{m(m+1)}$ with $m\in \mathbb{Q}$ (rational but not necessarily integer). Assume that 
\begin{equation}
\frac{m+1}{m}=\frac{Q}{P}
\end{equation}
with $Q,P\geq 1$ coprime integers.

As a consequence of Theorem \ref{l34 thm FF maps of Verma modules}, one has that for such $c$, the \emph{maximal}\footnote{``Maximal'' means that they cannot be embedded as proper submodules into any other Verma module} reducible highest weight Verma modules are $V_{c,h_{p,q}}$ with $0\leq p \leq P$, $0\leq q\leq Q$.

The minimal model $\MM(P,Q)$ is defined as a CFT with the space of states (or space of fields)
\begin{equation}\label{l35 H of M(P,Q)}
\HH=\bigoplus_{1\leq p \leq P-1,\; 1\leq q\leq \; Q-1\;/\mathbb{Z}_2} M_{c,h_{p,q}}\otimes \ol{M}_{c,h_{p,q}},
\end{equation}
where $/\ZZ_2$ again means that from each equivalence class $(p,q)\sim (P-p,Q-q)$ we need to pick one representative. 

The minimal models $\MM(P,Q)$ are not unitary (the sesquilinear product on $\HH$ is not positive-definite), unless one has $(P,Q)=(m,m+1)$ for $m=2,3,\ldots$.

\begin{example}
The minimal model $\MM(2,5)$ corresponds to $c=-22/5$ (and $m=\frac23$) and has two primary fields: 
\begin{itemize}
\item $\Phi_{1,1}=\mathbb{1}$ of conformal weight $h=\bar{h}=h_{1,1}=0$,
\item $\Phi_{1,2}$ of conformal weight $h=\bar{h}=h_{1,2}=-\frac15$.\footnote{Note that the negative conformal weight means that the correlator of two such fields increases as the fields get farther apart: $\langle \Phi_{1,2}(w) \Phi_{1,2}(z) \rangle=|w-z|^{\frac45}$ (cf. Lemma \ref{l26 lemma 2-point}).}
\end{itemize}
In particular, it is clear that the model cannot be unitary, since $c<0$ and there is a field with negative conformal weight (each of these observations separately contradicts unitarity).
\end{example}

\begin{example}
The minimal model $\MM(4,5)$ is unitary. It has $c=7/10$ and the array\footnote{Such arrays for minimal models are called ``Kac tables''} of conformal weights $h_{p,q}$ for admissible $p,q$ is
\begin{equation}
\begin{array}{c|ccc}
 &  p=1& p=2&p=3 \\ \hline
q=1& 0 & 7/16 & 3/2 \\
q=2& 1/10 & 3/80 & 3/5 \\
q=3& 3/5 & 3/80 & 1/10 \\
q=4& 3/2 & 7/16 & 0
\end{array}
\end{equation}
In particular, the model has $6=3\times 4/2$ conformal families/species of primary fields. 
\end{example}

Each minimal model $\MM(P,Q)$ has a collection of primary field $\Phi_{p,q}$ of conformal weight $(h_{p,q},h_{p,q})$, with $p,q$ as in the r.h.s. of (\ref{l35 H of M(P,Q)}); $\Phi_{1,1}=\mathbb{1}$ has conformal weight $(0,0)$ and is identified with the identity field.

Some of the fusion rules are:
\begin{equation}\label{l35 fusion rules}
\begin{gathered}
{[\Phi_{1,1}]}\times [\Phi_{p,q}]=[\Phi_{p,q}],\\
[\Phi_{1,2}]\times [\Phi_{p,q}] = [\Phi_{p,q-1}]+[\Phi_{p,q+1}],\\
[\Phi_{2,1}]\times [\Phi_{p,q}] = [\Phi_{p-1,q}]+[\Phi_{p+1,q}].
\end{gathered}
\end{equation}

Minimal models of CFT describe different 2d  systems of statistical mechanics at the point of second-order phase transition (put another way, they describe universality classes of 2d critical phenomena). For instance, one has the following correspondences were identified between 2d systems at the point of phase-transition and minimal models of CFT:

\begin{tabular}{c|c}
CFT minimal model & phase transition \\ \hline
$\MM(3,4)$ & Ising model at critical temperature \\
$\MM(2,5)$ & Yang-Lee edge singularity \\
$\MM(4,5)$ & tricritical Ising model \\
$\MM(5,6)$ & 3-state Potts model \\
$\MM(6,7)$ & tricritical 3-state Potts model
\end{tabular}

\begin{remark}
All primary fields in a minimal model $\MM(P,Q)$ are highest vectors of reducible Virasoro modules (always corresponding to the case $\mr{III}_-$ in Feigin-Fuchs classification, Theorem \ref{l34 thm FF maps of Verma modules}) and thus have vanishing descendants. Therefore any 4-point correlation function of primary fields in $\MM(P,Q)$ can be reduced to a function $F(\lambda)$ of the cross-ratio $\lambda$ satisfying certain ODE (e.g. a hypergeometric equation in the case of fields $\Phi_{1,2},\Phi_{2,1}$), as in the case of the correlator $\langle \bbsigma\bbsigma\bbsigma\bbsigma \rangle$ in Section \ref{sss free fermion correlators}.
\end{remark}

\begin{definition}
One calls a CFT with finitely many primary fields (or equivalently finitely many conformal families -- irreducible summands in the space of states/ space of fields) a \emph{rational} CFT, or RCFT.
\end{definition}

Thus, minimal models are the prime examples of rational CFT. On the other hand, free boson (with values in $\RR$ or $S^1$) is not rational:
% in the sense of this definition: 
 it contains infinitely many primary fields.

%\marginpar{edit/move: mention axioms: closed under OPEs, modular invariance}
To define a CFT, one needs to present two pieces of data:
\begin{itemize}
\item The space of states $\HH$ or equivalently the space of fields $V$ as a $\mr{Vir}\oplus \ol{\mr{Vir}}$-module (with come central charge $c,\bar{c}$) -- in particular, splitting it into irreducible summands, one has conformal families generated by highest weight vectors/primary fields.
\item The coefficients  in 3-point correlation functions  of primary fields (\ref{l27 3-point fun}).
\end{itemize}
This data allows one to recover all correlation functions of all fields but there are two constraints that the data above must satisfy:
\begin{enumerate}[(i)]
\item ``Crossing symmetry'' -- a certain quadratic constraint on the coefficients of 3-point functions of primary fields, see Section \ref{sss 4-pt corr of primary fields}.
\item Modular invariance of genus one partition function.
\end{enumerate}

%It is modular invariance that prohibits 
\begin{remark}
If one defines the space of states to be just the single conformal family generated by the identity field $\mathbb{1}$, then the corresponding ``CFT'' will have correlators and OPEs on the plane but will fail the modular invariance property (unless $c=0$ which is the case of the trivial CFT $\MM(2,3)$).

%$\HH=M_{c,0}\otimes \ol{M}_{c,0}$
More explicitly, one computes the torus partition function in $\MM(P,Q)$ using  (\ref{l34 torus Z via chi}) and evaluating the characters as in Example \ref{l34 ex character of M_1/2,0}, resulting in the formula
\begin{multline}\label{l35 torus Z for M(P,Q)}
Z(\tau)=\\
=\frac{|\qq|^\frac{1-c}{12}}{|\eta(\tau)|^2}\sum_{1\leq p \leq P-1,\; 1\leq q\leq \; Q-1\;/\mathbb{Z}_2} 
%|\qq|^{2h_{p,q}}
\left| \qq^{h_{p,q}}\sum_{k\in\ZZ}\Big( \qq^{pq+(-p+Pk)(q+Qk)}-\qq^{(p+Pk)(q+Qk)} \Big)\right|^2.
\end{multline}
This expression is modular invariant (which can be proved by Poisson summation). However, restricting to only the $(p,q)=(1,1)$ term in the sum one obtains a non-modular invariant expression.
\end{remark}
%\marginpar{To add: torus partition fun for $\MM(p,q)$}

\section{Correlators and OPEs of primary fields in a general RCFT}
\label{sec 7.5}

Consider a general CFT. Fix $\{\Phi_p\}_{p\in I}$ an orthonormal basis of primary fields, with $I$ an indexing set.
% (since we assumed rationality).
\begin{lemma}\label{l35 lemma OPE}
For $\Phi_1,\Phi_2$ primary fields, the OPE has the form
\begin{equation}\label{l35 OPE ansatz}
\Phi_1(w) \Phi_2(z)\sim \sum_{p\in I} \sum_{\vec{k},\vec{\bar{k}}} C_{12p}^{\vec{k},\vec{\bar{k}}} (w-z)^{-h_1-h_2+h_p+|\vec{k}|}
(\bar{w}-\bar{z})^{-\bar{h}_1-\bar{h}_2+\bar{h}_p+|\vec{\bar{k}}|} \Phi_p^{\vec{k},\vec{\bar{k}}}(z),
\end{equation}
where:
\begin{itemize}
\item The first sum is over species primary fields.
\item The second sum over pairs of nondecreasing sequences $1\leq k_1\leq\cdots \leq k_r$ (which we denote $\vec{k}$) and $1\leq \bar{k}_1\leq\cdots \leq \bar{k}_s$ (denoted $\vec{\bar{k}}$), with $r,s\geq 0$; we also denoted $|\vec{k}|=k_1+\cdots + k_r$ and similarly for $\vec{\bar{k}}$; $\Phi_p^{\vec{k},\vec{\bar{k}}}$ is the descendant
\begin{equation}
\Phi_p^{\vec{k},\vec{\bar{k}}}=L_{-k_r}\cdots L_{-k_1}\ol{L}_{-\bar{k}_s}\cdots \ol{L}_{-\bar{k}_1}\Phi_p
\end{equation}
\item The coefficients on the right are
\begin{equation}\label{l35 C_12p^k}
C_{12p}^{\vec{k},\vec{\bar{k}}}=C_{12p}\beta_{12p}^{\vec{k}} \bar\beta_{12p}^{\vec{\bar{k}}}
\end{equation}
where $C_{12p}$ %is the coefficient of the 3-point function 
are certain coefficients depending on the triple of primary fields $\Phi_1,\Phi_2,\Phi_p$ and $\beta_{12p}^{\vec{k}}$ a certain family of universal\footnote{I.e. not depending on any details of the CFT.} rational functions of $c,h_1,h_2,h_p$ parametrized by the sequence $\vec{k}$; $\bar\beta$ is the same family where $\bar{c},\bar{h}_1,\bar{h}_2,\bar{h}_p$ are used instead. 
\end{itemize}
\end{lemma}
One can always assume the normalization $\beta_{12p}^\varnothing=1$.
\begin{remark}
\begin{enumerate}[(a)]
\item
``Structure constants'' $C_{12p}$ in the r.h.s. of (\ref{l35 C_12p^k}) are the same as the constants appearing in the r.h.s. of the 3-point function (\ref{l27 3-point fun}) of primary fields 
\begin{equation}\label{l35 3-pt fun of primary}
\langle \Phi_1(w)\Phi_2(z)\Phi_p(x) \rangle.
\end{equation} 
Expressed another way, it is the matrix element
\begin{equation}
\langle \Phi_p| \wh\Phi_1(1) |\Phi_2 \rangle.
\end{equation}
It is symmetric under permutations of species $1,2,p$ (as obvious from the previous interpretation).
\item
%\marginpar{not sure that's right.. }
Remark that as a consequence of Lemma \ref{l35 lemma OPE}, a descendant field $\Phi_p^{\vec{k},\vec{\bar{k}}}$ can appear in the OPE $\Phi_1(w)\Phi_2(z)$ only if the primary field $\Phi_p$ itself appears in that OPE.
\item From the ansatz (\ref{l35 lemma OPE}) it is clear that only finitely many descendants of each primary field $\Phi_p$ contribute to the \emph{singular part} of the OPE.
\end{enumerate}
\end{remark}

\begin{proof}[Sketch of proof of Lemma \ref{l35 lemma OPE}]
The exponents in the ansatz (\ref{l35 OPE ansatz}) follow immediately from Lemma \ref{l26 lemma OPE exponents}. The only thing to check is (\ref{l35 C_12p^k}). 
%which is proven by induction in $|\vec{k}|$ and $|\vec{\bar{k}}|$.
%Acting by the vector field $v=-(x-z)^2\dd_x$ on both sides of (\ref{l35 lemma OPE}), one obtains the relation
%\begin{multline}
%((w-z)^2 \dd \Phi_1(w)+2(w-z)h_1\Phi_1(w)) \Phi_2(z) \sim
%\\ \sim
%\sum_p \sum_{\vec{k},\vec{\bar{k}}} C_{12p}^{\vec{k},\vec{\bar{k}}} (w-z)^{-h_1-h_2+h_p+|\vec{k}|}
%(\bar{w}-\bar{z})^{-\bar{h}_1-\bar{h}_2+\bar{h}_p+|\vec{\bar{k}}|} L_1\Phi_p^{\vec{k},\vec{\bar{k}}}(z).
%\end{multline}
%Expressing the l.h.s. here using (\ref{l35 OPE ansatz}) and rewriting the r.h.s. using Virasoro relations as a sum of descendants at level $(|\vec{k}|-1,|\vec{\bar{k}}|)$, one obtains the expression 

%\marginpar{word it better, explain factorization..}
The idea is to consider 3-point correlation functions
\begin{equation}\label{l35 3-point}
\langle \Phi_1(w) \Phi_2(z) \Phi_p^{\vec{l},\vec{\bar{l}}}(x) \rangle.
\end{equation}
for various nondecreasing sequences $\vec{l},\vec{\bar{l}}$ with $|\vec{l}|=|\vec{k}|$, $|\vec{\bar{l}}|=|\vec{\bar{k}}|$.
On one hand one can find these correlators explicitly by reducing them to a differential operator acting on $\langle \Phi_1(w)\Phi_2(z) \Phi_p(x) \rangle$ (cf. Example \ref{l26 ex: corr of a descendant}), resulting in expressions of the form
%\marginpar{edit/correct}
 \begin{equation}
 \langle \Phi_p^{\vec{l},\vec{\bar{l}}} | \wh\Phi_1(1) |\Phi_2 \rangle= C_{12p}\gamma_{12p}^{\vec{l}} \bar{\gamma}_{12p}^{\vec{\bar{l}}}
 \end{equation}
with $\gamma_{12p}^{\vec{l}}$ some universal rational functions of $c,h_1,h_2,h_p$ depending on the sequence $\vec{l}$, and similarly for $\bar\gamma$; for convenience we set $z=0,w=1, x\ra\infty$ in the correlator (\ref{l35 3-point}).
%$A(w,z,x)$ is the standard product of powers of differences of coordinates as in the r.h.s. of (\ref{l27 3-point fun}).
%
%
On the other hand one can replace $\Phi_1(w)\Phi_2(z)$ in (\ref{l35 3-point}) with the r.h.s. of (\ref{l35 OPE ansatz}) and evaluate the remaining 2-point functions of descendants in terms of elements of the Gram matrix (\ref{l34 Gram matrix}): 
\begin{equation}
\langle \Phi_p^{\vec{l},\vec{\bar{l}}}| \wh\Phi_1(1) |\Phi_2  \rangle =  \sum_{\vec{k},\vec{\bar{k}}} C_{12p}^{\vec{k},\vec{\bar{k}}} G^p_{\vec{k},\vec{l}}\, \ol{G}^p_{\vec{\bar{k}},\vec{\bar{l}}}\, ,
\end{equation}
where $G^p_{\vec{k},\vec{l}}$ are the matrix elements of the Gram matrix. Here we again set $z=0,w=1,x\ra\infty$.
Comparing the two sides, we obtain the claimed ansatz (\ref{l35 C_12p^k}) with 
\begin{equation}
\beta_{12p}^{\vec{k}}=\sum_{\vec{l}}((G^p)^{-1})_{\vec{k},\vec{l}}\gamma_{12p}^{\vec{l}}.
\end{equation}
%a system of $P(|\vec{k}|)\cdot P(|\vec{\bar{k}}|)$ linear equations on the same amount of variables.
\end{proof}

\begin{example}
The first coefficients $\beta_{12p}^{\vec{k}}$ appearing in (\ref{l35 C_12p^k}) are:
\begin{equation*}
\begin{aligned}
&\beta_{12p}^{\varnothing}=1,\\
&\beta_{12p}^{\{1\}}= \frac{h_1-h_2+h_p}{2h_p},\\
& \hspace{-1cm} \left(
\begin{array}{c}
\beta_{12p}^{\{2\}}\\
\beta_{12p}^{\{1,1\}}
\end{array}
\right) =
\left(
\begin{array}{cc} 
 4 h_p+\frac{c}{2} & 6 h_p\\
 6h_p &  2h_p(4h_p+2)
\end{array}
\right)^{-1}
\left(
\begin{array}{c}
2h_1-h_2+h_p\\
(-h_1-h_2+h_p)(3h_1-h_2+h_p+1)+6h_1^2 
\end{array}
\right).
\end{aligned}
\end{equation*}
\end{example}

\begin{remark}\label{l35 rem: fusion rules}
Assume that the primary field $\Phi_1$ has a vanishing descendant al level $N$ (corresponding to a null vector in the corresponding Verma module). Then by the argument of Example \ref{l26 ex: corr of a descendant} there is a degree $\leq N$ differential operator annihilating the 3-point function of primary fields (\ref{l35 3-pt fun of primary}). Combining with the expression (\ref{l27 3-point fun}) for the 3-point function this implies an algebraic equation of degree $\leq N$. Thus, there is an algebraic equation of degree $\leq N$ on the conformal weight $h_p$ of a primary field which (and whose descendant) can appear in the r.h.s. of the OPE (\ref{l35 OPE ansatz}).

This is exactly the case in minimal models $\MM(P,Q)$ and this is how one obtains ``fusions rules'' (\ref{l35 fusion rules}) and, more generally, obtains the result that fields $\Phi_{p,q}$ of the minimal model form a closed algebra under OPEs: no fields with other conformal weights can appear.
\end{remark}

\subsection{4-point correlator of primary fields}\label{sss 4-pt corr of primary fields}
%\label{sss 4-point corr of primary fields}
The correlator of four primary fields in a general CFT is bound by global conformal symmetry to be of the form (\ref{l27 n-point fun}):
\begin{equation}\label{l35 4-point fun}
\langle \Phi_1(z_1)\Phi_2 (z_2) \Phi_3(z_3) \Phi_4(z_4) \rangle=
\Big(\prod_{1\leq i<j\leq 4} z_{ij}^{\frac13 \sum_{k=1}^4 h_k-h_i-h_j} 
\bar{z}_{ij}^{\frac13 \sum_{k=1}^4 \bar{h}_k-\bar{h}_i-\bar{h}_j}\Big)
f(\lambda)
\end{equation}
with $f$ a smooth function of the cross-ratio $\lambda=\frac{z_{13}z_{24}}{z_{14} z_{23}}\in \CP^1\backslash\{0,1,\infty\}$. We can use M\"obius symmetry to fix points $z_2,z_3,z_4$ at $1,0,\infty$, then $z_1$ becomes $\lambda$. Thus, we have
\begin{equation}
\langle \Phi_4|\wh\Phi_2(1)  \wh\Phi_1(\lambda) |\Phi_3 \rangle =f(\lambda)
\end{equation}
Applying the OPE (\ref{l35 OPE ansatz}) to the expression $\wh\Phi_1(\lambda) |\Phi_3 \rangle$ above, we obtain
\begin{multline}\label{l35 f}
f(\lambda)=\sum_{p\in I}\sum_{\vec{k},\vec{\bar{k}}} C_{13p}\beta_{13p}^{\vec{k}}\bar\beta_{13p}^{\vec{\bar{k}}} \lambda^{-h_1-h_3+h_p+|\vec{k}|} \bar\lambda^{\bar{h}_1-\bar{h}_3+\bar{h}_p+|\vec{\bar{k}}|} \langle \Phi_4|\wh\Phi_2(1) | \Phi_p^{\vec{k},\vec{\bar{k}}} \rangle\\
=\sum_{p\in I} C_{42p}C_{13p} \mc{F}_{13}^{24}(p|\lambda) \ol{\mc{F}}_{13}^{24}(p|\bar\lambda)
\end{multline}
where 
\begin{equation}\label{l35 conf block}
\mc{F}_{13}^{24}(p|\lambda)\colon= \lambda^{-h_1-h_3+h_p}\sum_{K=0}^\infty \lambda^{K}\sum_{\vec{k},\vec{l}\;\mr{with}\; |\vec{k}|=|\vec{l}|=K}  \beta_{13p}^{\vec{k}} G^p_{\vec{k},\vec{l}}\, \beta_{24p}^{\vec{l}},
\end{equation}
and similarly for $\ol{\mc{F}}$. Here $G^p_{\vec{k},\vec{l}}$ is a matrix element of the Gram matrix (\ref{l34 Gram matrix}).

The r.h.s. of (\ref{l35 conf block}) is a holomorphic function of $\lambda$ (possibly with monodromy at $\lambda=0$), the sum over $K$ is absolutely convergent in the unit disk $|\lambda|<1$. Thus the function $f(\lambda)$ determining the 4-point correlation function is a sum over $I$ (i.e. a finite sum for a rational CFT) of products of certain universal holomorphic and antiholomorphic functions, with coefficients given in terms of coefficients of 3-point functions. This begins to justify the claim that coefficients of 3-point functions determine all correlators in a CFT.

Function (\ref{l35 conf block}) is called the \emph{conformal block} of the 4-point function, cf. (\ref{l30 conf blocks}).

Computation (\ref{l35 f}) can be thought of in terms of Segal's axioms, as cutting a 4-punctured sphere $\CP^1$ by a circle $S^1_r$ of radius $|\lambda|<r<1$ centered at the origin and evaluating the corresponding composition as a sum over the basis in the space of states for the circle $S^1_r$:
\begin{equation}
\langle \Phi_4|\wh\Phi_2(1)  \wh\Phi_1(\lambda) |\Phi_3 \rangle =
\sum_{p\in I}\sum_{\vec{k},\vec{\bar{k}}} \langle \Phi_4|\wh\Phi_2(1) | \Phi_p^{\vec{k},\vec{\bar{k}}} \rangle \langle \Phi_p^{\vec{k},\vec{\bar{k}}}| \wh\Phi_1(\lambda) |\Phi_3 \rangle.
\end{equation}

\begin{figure}[H]
\begin{center}
\includegraphics[scale=0.7]{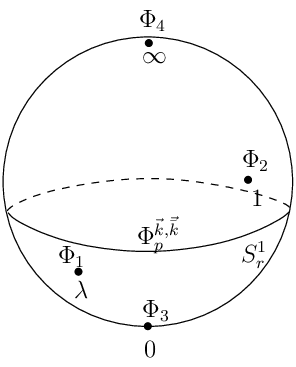}
\end{center}
\caption{Cutting the 4-point correlator on $\CP^1$.}
\end{figure}

\subsubsection{Crossing symmetry.}
Starting from the 4-point function (\ref{l35 4-point fun})
and switching the roles of fields $\Phi_2(z_2)$ and $\Phi_3(z_3)$,
%, one can elect to fix $z_2\ra 0$, $z_3\ra 1$, $z_4\ra \infty$. then the point $z_1$ goes to $\mu=1-\lambda$.
one obtains another expression for the 4-point function:
\begin{equation}\label{l35 f(lambda) 2}
f(\lambda)=\sum_{p\in I} C_{43p}C_{12p}\mc{F}_{12}^{34}(p|1-\lambda)\ol{\mc{F}}_{12}^{34}(p|1-\bar\lambda).
\end{equation}
Expressions (\ref{l35 f}) and (\ref{l35 f(lambda) 2}) must agree in the region where r.h.s. in both cases is defined, i.e., in the region $\{\lambda\in \CC\;|\; |\lambda|<1,\; |1-\lambda|<1\}$. This is the so-called ``crossing symmetry.'' In particular it implies nontrivial quadratic relations (a version of associativity constraint) on the coefficients of 3-point functions of primary fields.

In terms of Segal's axioms, crossing symmetry is just the statement that cutting a 4-punctured sphere in two ways yields the same partition function.

\begin{figure}[H]
\begin{center}
\includegraphics[scale=0.7]{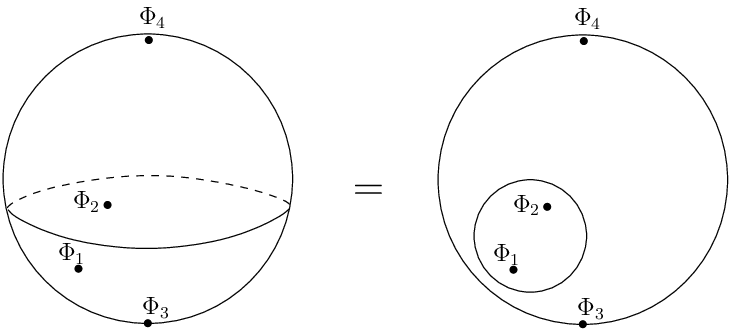}
\end{center}
\caption{Crossing symmetry = cutting the 4-point correlator on $\CP^1$ in two ways.}
\end{figure}

Replacing punctures by finite circles, the same picture can be regarded as cutting a sphere with four holes into two pairs of pants in two different ways.

\begin{figure}[H]
\begin{center}
\includegraphics[scale=0.7]{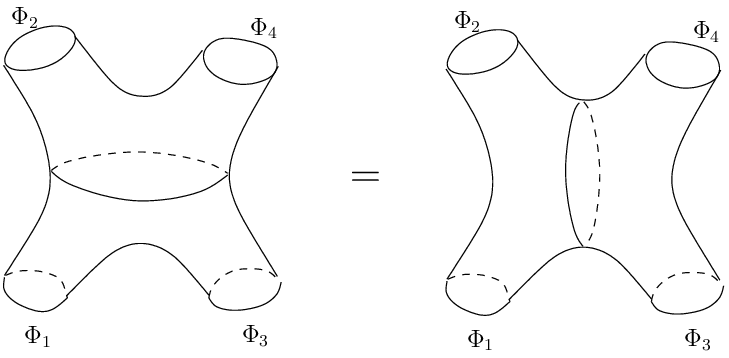}
\end{center}
\caption{Another visualization of crossing symmetry. 
%= cutting the sphere with four holes into pairs of pants in two ways.
}
\end{figure}

\subsection{$n$-point correlator of primary fields}\label{ss n-point correlator of primary fields}
The strategy above can be applied to $n$-point correlators of primary fields on $\CC$ (or $\CP^1$), with $n\geq 4$.

First note that Lemma \ref{l35 lemma OPE} has a straightforward generalization to the case of an OPE between descendants:
\begin{multline}\label{l35 OPE descendants}
\Phi_1^{\vec{k}_1,\vec{\bar{k}}_1}(w) \Phi_2^{\vec{k}_2,\vec{\bar{k}}_2}(z) \sim \\
\sim 
\sum_{p\in I}\sum_{\vec{k},\vec{\bar{k}}} C_{12p}^{\vec{k}_1,\vec{k}_2,\vec{k};\vec{\bar{k}}_1,\vec{\bar{k}}_2,\vec{\bar{k}}} (w-z)^{-h_1-h_2+h_p+|\vec{k}|-|\vec{k}_1|-|\vec{k}_2|} (\bar{w}-\bar{z})^{-\bar{h}_1-\bar{h}_2+\bar{h}_p+|\vec{\bar{k}}|-|\vec{\bar{k}}_1|-|\vec{\bar{k}}_2|} \Phi_p^{\vec{k},\vec{\bar{k}}}(z)
\end{multline}
with 
\begin{equation}
C_{12p}^{\vec{k}_1,\vec{k}_2,\vec{k};\vec{\bar{k}}_1,\vec{\bar{k}}_2,\vec{\bar{k}}}=
C_{12p}\beta_{12p}^{\vec{k}_1,\vec{k}_2,\vec{k}} \bar\beta_{12p}^{\vec{\bar{k}}_1,\vec{\bar{k}}_2\vec{\bar{k}}}.
\end{equation}
Here all conventions are as in Lemma \ref{l35 lemma OPE}; $\beta_{12p}^{\vec{k}_1,\vec{k}_2,\vec{k}}$ are again some universal rational functions of conformal weights and the central charge. To obtain (\ref{l35 OPE descendants}), one repeatedly applies Virasoro generators (centered at $z$ and at $w$) to the OPE (\ref{l35 OPE ansatz}).

Given a correlator of primary fields
\begin{equation} \label{l35 corr of n primary fields}
\langle \Phi_1(z_1)\cdots \Phi_n(z_n) \rangle,
\end{equation}
choose a tree $T$ with trivalent vertices and with $n$ leaves decorated by $\Phi_i(z_i)$. The tree $T$ %can be thought of as 
determines
an asymptotic region in the configuration space $C_n(\CP^1)$, prescribing in which order the points are approaching one another (for instance, $z_2$ approaches $z_1$ and $z_4$ approaches $z_3$; then $z_3$ approaches $z_1$).\footnote{This asymptotic region corresponds to a compactification stratum of complex codimension $n-3$ in the Fulton-MacPherson compactification of $C_n(\CP^1)$.}
In this asymptotic region, one can compute the correlator (\ref{l35 corr of n primary fields}) by iteratively using the OPE (\ref{l35 OPE descendants}). The result is a sum over ``intermediate states/fields'' -- a sum over species of primary fields and partitions $\vec{k},\vec{\bar{k}}$ decorating each inner edge of $T$. Such a sum coverges in a finite region of the configuration space.\footnote{One can also think of $T$ as prescribing a way to cut $\CP^1$ by $n-3$ circles (corresponding to the inner edges of the tree) into pairs of pants. The ``asymptotic region of the configuration space'' picture corresponds to the circles being infinitesimal, while a finite region of the configuration space corresponds to having finite circles.}

The resulting formula for the correlator has the general form
\begin{equation}\label{l35 corr of n fields via conf blocks}
\langle \Phi_1(z_1)\cdots \Phi_n(z_n) \rangle=
\sum_{p_1,\ldots,p_{n-3}\in I} \prod_{v\in V_T} C_{p_{v_1}p_{v_2}p_{v_3}} \cdot \mc{F}_T(z_1,\ldots,z_n) \ol{\mc{F}}_T(\bar{z}_1,\ldots,\bar{z}_n)
\end{equation}
Here the sum is over species of primary fields decorating the inner edges of $T$. The first term in the r.h.s. is the product over inner vertices of the structure constant of 3-point functions, corresponding to the decorations of the three incident edges of the vertex. The two subsequent factors $\mc{F}, \ol{\mc{F}}$ are a holomorphic and an antiholomorphic function on the configuration space, depending on the tree $T$ and on the conformal weights ($h$ for $\mc{F}$ and $\bar{h}$ for $\ol{\mc{F}}$) of the primary fields decorating the leaves and the inner edges. The functions $\mc{F}_T$, $\ol{\mc{F}}_T$ arising from summation over intermediate descendants are the conformal blocks of the $n$-point correlation function (\ref{l35 corr of n primary fields}).

Regions of convergence corresponding to different trees may overlap. The resulting formulae (\ref{l35 corr of n fields via conf blocks}) give compatible answers on the overlap due to the crossing symmetry.

\chapter{Wess-Zumino-Witten model}

\marginpar{Lecture 36,\\ 11/21/2022}

\section{Affine Lie algebras}\label{ss affine Lie algebras}
For details on affine Lie algebras we refer to \cite{Kac}, \cite{Kohno}, \cite{DMS}.

Fix a compact simple Lie group $G$, denote its Lie algebra $\g$ and the complexification of the latter $\g_\CC=\CC\otimes \g$. 

\begin{definition}
The \emph{loop group} $LG=\mr{Map}(S^1,G)$ is the group of $G$-valued smooth functions on a circle with pointwise multiplication. 
Its complexified Lie algebra $L\g=\mr{Map}(S^1,\g_\CC)$ -- the Lie algebra of $\g_\CC$-valued functions on $S^1$ with pointwise Lie bracket is called the \emph{loop Lie algebra}.
\end{definition}

 One can identify loop Lie algebra with the algebra of $\g_\CC$-valued Laurent polynomials
\begin{equation}\label{l36 Lg}
L\g=\g_\CC\otimes \CC[t,t^{-1}]
\end{equation}
where $t$ is the complex coordinate on the unit circle $S^1=\{t\in\CC\;|\; |t|=1\}$.\footnote{One can choose different completions of the algebra of Laurent polynomials in (\ref{l36 Lg}) corresponding to different regularity assumptions on the allowed maps from $S^1$ to $\g_\CC$, cf. the discussion of models of Witt algebra in Section \ref{sss: Witt algebra}. We will not dwell on this point.} The Lie bracket in $L\g$ is
\begin{equation}
[X\otimes f,Y\otimes g]= [X,Y]\otimes fg
\end{equation}
for $X,Y\in \g_\CC$, $f,g\in \CC[t,t^{-1}]$.

%Loop Lie algebra admits 
\begin{definition}
The affine Lie algebra $\wh\g$ associated with $\g$ is defined as the unique (up to normalization) central extension $\wh\g=L\g\oplus \CC\cdot \mathbb{K}$ of the loop Lie algebra, equipped with Lie bracket
\begin{equation}\label{l36 ghat bracket}
[X\otimes f,Y\otimes g]_{\wh{\g}}=[X,Y]\otimes fg+\mathbb{K} \langle X,Y \rangle_\g \mr{res}_{t=0} (df\cdot g).
\end{equation}
Here $\mathbb{K}$ is the central element, $\langle,\rangle_\g$ is the Killing form\footnote{
We assume the normalization of the Killing form $\langle X,Y \rangle_\g=\tr (XY)$ -- the trace of the product in the \emph{fundamental} representation of $\g$ (e.g. in the 2-dimensional representation for $\g=\mathfrak{su}(2)$). 
} on $\g$ and the residue $\mr{res}_{t=0}(\cdots)$ returns the coefficient of  $t^{-1}dt$ in the 1-form $(\cdots)$. %The Lie algebra $\wh{\g}$ is called the \emph{affine Lie algebra} (associated with $\g$).
\end{definition}

One can write the Lie bracket (\ref{l36 ghat bracket}) more explicitly:
\begin{equation}\label{l36 ghat bracket 2}
[X\otimes t^n,Y\otimes t^{m}]=[X,Y]\otimes t^{n+m}+\mathbb{K} \langle X,Y\rangle_\g n\delta_{n,-m}.
%[X_n,Y_m]=[X,Y]_{n+m}+\mathbb{K} (X,Y) n\delta_{n,-m}.
\end{equation}
We will be using a shorthand notation $X_n\colon= X\otimes t^n$.

\begin{remark} The statement that (\ref{l36 ghat bracket}) is the \emph{unique} up to normalization central extension of the loop Lie algebra is tantamount to a statement about Lie algebra cohomology:
\begin{equation}\label{l36 H^2(Lg)}
H^2_\mr{Lie}(L\g,\CC) = \CC,
\end{equation} 
where the nontrivial 2-cocycle is given by the rightmost term in (\ref{l36 ghat bracket}). 

The result (\ref{l36 H^2(Lg)}) uses the fact that $\g$ is simple. For $\g$ semisimple with $n$ simple summands $\g=\g_1\oplus\cdots \oplus \g_n$, the r.h.s. of (\ref{l36 H^2(Lg)}) is $\CC^n$ -- there are $n$ independent 2-cocycles corresponding to Killing forms on $\g_i$. 
%in the simple summands of $\g$.
\end{remark}

\begin{remark} If we set $\g=\RR$ and $\langle X,Y \rangle_\g=XY$ for $X,Y\in \RR$,\footnote{
This example is somewhat outside the setup of this section: $\RR$ is not the Lie algebra of a compact simple group and this choice of $\langle,\rangle_\g$ is not the Killing form (the Killing form for $\g=\RR$ is zero).
} then (\ref{l36 ghat bracket}) becomes the Lie bracket of the Heisenberg Lie algebra (\ref{l17 Heisenberg Lie algebra}), (\ref{l17 Heis comm rel}), so in this case one has $\wh{g}=\mr{Heis}$.
\end{remark}

Similarly to the loop Lie algebra $L\g$, the loop group $LG$ also has a family of central extensions $\wh{LG}^k$,
\begin{equation}\label{l36 LG central extension}
1\ra \CC^*\ra  {\wh{LG}}^k \ra LG \ra 1,
\end{equation}
with the ``level'' parameter $k=1,2,3,\ldots$; here $\wh{LG}^k$ is a principal $\CC^*$-bundle over $LG$ with first Chern class $c_1=k\in H^2(LG,\ZZ)\simeq \ZZ$. 

At the level of Lie algebra, the central extension $\wh{LG}^k$ corresponds to the affine Lie algebra $\wh{\g}$ where $\mathbb{K}$ is identified with $k\cdot\mr{Id}$ -- an integer multiple of identity (in particular, an $\wh{LG}^k$-module is automatically a $\wh{\g}$-module, with $\mathbb{K}$ acting by $k\cdot \mr{Id}$).

\textbf{Notation.} The affine Lie algebra $\wh\g$ with the central element identified with $k\cdot\mr{Id}$, with $k$ an integer, is customarily denoted $\wh\g_k$.

\subsection{Highest weight modules over $\wh\g$}
Fix a decomposition
\begin{equation}\label{l36 g Borel decomp}
\g_\CC=\g_+\oplus \mathfrak{h}\oplus \g_-
\end{equation}
with $\mathfrak{h}$ the Cartan subalgebra, $\g_+$ the span of positive roots $\{e_\alpha\}_{\alpha>0}$ of $\g$ and $\g_-$ the span of negative roots $\{e_{\alpha}\}_{\alpha<0}$.

Consider the following decomposition of the affine Lie algebra $\wh{\g}$:
\begin{equation}\label{l36 ghat Borel decomp}
\wh{\g}=\underbrace{(\g\otimes t\, \CC[t]\oplus \g_+)}_{N_+}\oplus \underbrace{(\CC\cdot \mathbb{K}\oplus \frh)}_{N_0} \oplus \underbrace{(\g\otimes t^{-1}\CC[t^{-1}] \oplus \g_- )}_{N_-}.
\end{equation}

A Verma module over $\wh{\g}$ is defined (cf. footnote \ref{l34 footnote: Verma}) as 
\begin{equation}\label{l36 Verma}
V^{\wh{\g}}_{k,\lambda}=U(\wh{\g})\otimes_{U(N_0\oplus N_+)} \CC_{k,\lambda}.
\end{equation}
Here:
\begin{itemize}
\item $k\in \CC %=1,2,3,\ldots
$ is the level\footnote{
In the context of Verma modules over $\wh\g$, the level does not have to be an integer and $\lambda\in \CC^r$ can be any vector. However, more detailed structure of the Verma module (e.g. null vectors) is sensitive to integrality of $k$ and to $\lambda$ belonging to the weight lattice of $\g$.
%for integer $k$ and special values of $\lambda$ the Verma module has a particularly nice quotient module
} and $\lambda=(\lambda^1,\ldots,\lambda^r)$ is a highest weight of $\g$, with $r=\dim \frh$ the rank of $\g$. We assume that a basis $\tau^1,\ldots,\tau^r$ in $\frh$ is fixed.
\item $\CC_{k,\lambda}$ is a 1-dimensional module over $N_0\oplus N_+$ where $N_+$ acts by zero, $\mathbb{K}$ acts by multiplication by the level $k$ and elements  of the Cartan $\tau^i\in\frh$ act by multiplication by $\lambda^i$.
\item $U(\cdots)$ is the universal enveloping algebra. 
%of the Lie algebra $A$.
\end{itemize}

Let us denote by $v$  the highest weight vector in (\ref{l36 Verma}) -- the generator of $\CC_{k,\lambda}$.

As in Virasoro case, in $V^{\wh{\g}}_{k,\lambda}$ one can have null vectors -- vectors (distinct from the highest weight vector $v$) annihilated by $N_+$.

The \emph{irreducible} highest weight module (of level $k$, with highest weight $\lambda$) is
\begin{equation}
M^{\wh{\g}}_{k,\lambda}=V^{\wh{\g}}_{k,\lambda}/\nu
\end{equation}
-- the quotient of the Verma module by the maximal proper submodule. As in the Virasoro case, $\nu$ can also be described as
\begin{itemize}
\item the submodule generated by the null-vectors,
\item or equivalently as the kernel of the sesquilinear form on $V^{\wh{\g}}_{k,\lambda}$ characterized by the properties $\langle v,v\rangle=1$, $(X\otimes t^n)^+=X^+\otimes t^{-n}$.
\end{itemize}

\begin{remark}
It is convenient to adjoin to $\wh\g$ an extra generator (``grading operator'' or ``Euler vector field'') $\delta=-t\frac{d}{dt}$ satisfying the commutation relations
\begin{equation}\label{l36 [delta,X]}
[\delta, X\otimes t^j]=-j X\otimes t^j,\quad [\delta,\mathbb{K}]=0.
\end{equation}
The algebra $\wh{\g}\oplus \CC\cdot \delta$ is called the \emph{affine Kac-Moody algebra}.
\end{remark}

In a highest weight module $W$, if we set $\delta(v)=0$, the module becomes $\ZZ_{\geq 0}$-graded by eigenvalues $n_\delta$ of $\delta$: 
\begin{equation}\label{l36 W= sum of W(n_delta)}
W=\bigoplus_{n_\delta=0}^\infty W(n_\delta).
\end{equation}
We will call $n_\delta$ ``depth.''\footnote{It is not a standard term; we use it because the word ``level'' already has another meaning in the context of affine Lie algebras.} 

Note that each term $W(n_\delta)$ in the r.h.s. of (\ref{l36 W= sum of W(n_delta)}) carries a representation of $\g$ (without the hat).
In particular, for $W$ the Verma module and $n_\delta=0$ one has that $V^{\wh{\g}}_{k,\lambda}(0)$ is the Verma module $V^\g_\lambda$ of $\g$ with highest weight $\lambda$  obtained by acting on $v$ by elements of $\g_-$.
Similarly, for the irreducible $\wh\g$-module one has that $M^{\wh{\g}}_{k,\lambda}(0)=M^\g_\lambda$ is the irreducible representation of $\g$ with highest weight $\lambda$.
%Note that both $V^{\wh{\g}}_{k,\lambda}$ and $M^{\wh{\g}}_{k,\lambda}$ contain at depth $n_\delta=0$ a representation of $\g$ with highest weight $\lambda$, obtained by acting on $v$ by elements of $\g_-$. In the case of  $V^{\wh{\g}}_{k,\lambda}$, it is a Verma module $V^\g_\lambda$ of $\g$ and in the case of $M^{\wh{\g}}_{k,\lambda}$, it is an irreducible representation $M^\g_\lambda$ of $\g$.

\subsubsection{Integrable highest weight modules.}
There is a distinguished set of irreducible highest weight modules over $\wh{\g}$ -- ``integrable highest weight modules'' for positive integer level $k=1,2,3,\ldots$ Their equivalent characterizations are:
\begin{enumerate}[(i)]
\item The module $M_{k,\lambda}^{\wh{\g}}$ is integrable if the action of $\wh{\g}$ on it integrates to the action of the group $\wh{LG}^k$.
\item (Purely Lie algebraic definition.) The module $M_{k,\lambda}^{\wh{\g}}$ is integrable if it satisfies the ``local nilpotency condition'': for any $u\in M_{k,\lambda}^{\wh{\g}}$, any $j\in \ZZ$ and any root $e_\alpha$ of $\g$ there exists $N$ such that
\begin{equation}
(e_\alpha\otimes t^j)^N u=0.
\end{equation}
\end{enumerate}
If the irreducible module $M^{\wh\g}_{k,\lambda}$ is integrable, we will also denote it $H_{k,\lambda}$.

\begin{thm}[see Kac \cite{Kac}]
There are finitely many integrable highest weight $\wh{\g}$-modules for any given positive integer level $k=1,2,3,\ldots$.
\end{thm}

\begin{example}
Consider the case $G=SU(2)$. In the complexified Lie  algebra $\g_\CC=\CC\otimes\mathfrak{su}(2)=\mathfrak{sl}(2,\CC)$ one can consider the standard basis
\begin{equation}
E= \left(  \begin{array}{cc}
0&1\\ 0&0
\end{array} \right) ,\quad
F= \left(  \begin{array}{cc}
0&0\\ 1&0
\end{array} \right) ,\quad
H= \left(  \begin{array}{cc}
1&0\\ 0&-1
\end{array} \right) 
\end{equation}
satisfying the commutation relations 
\begin{equation}
[H,E]=2E,\quad [H,F]=-2F,\quad [E,F]=H.
\end{equation}
We consider $H$ as the basis vector for the Cartan subalgebra $\frh$, $E$ as the positive root and $F$ the negative root, i.e., the decomposition (\ref{l36 g Borel decomp}) is
\begin{equation}
\mathfrak{sl}(2,\CC)=\underbrace{\CC\cdot E}_{\g_+} \oplus \underbrace{\CC\cdot H}_{\frh} \oplus \underbrace{\CC\cdot F}_{\g_-}
\end{equation}

\begin{figure}[H]
\begin{center}
\includegraphics[scale=0.7]{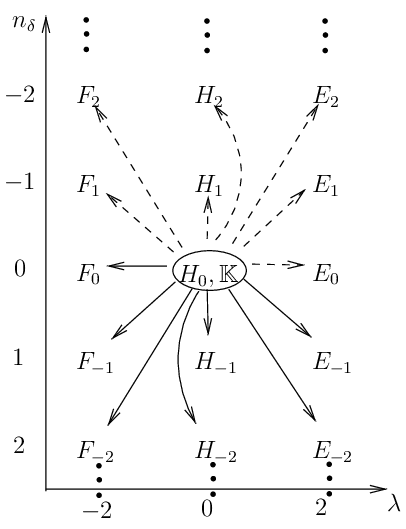}
\end{center}
\caption{Root diagram of $\wh{\mathfrak{su}(2)}$. Positive roots (basis of $N_+$, cf. (\ref{l36 ghat Borel decomp})) are indicated by dashed arrows and negative roots (basis of $N_-$) -- by solid arrows. The encircled part corresponds to the Cartan subalgebra $N_0$. The diagram extends infinitely vertically.
%= cutting the sphere with four holes into pairs of pants in two ways.
}
\end{figure}

Fix the level $k=1,2,3,\ldots$. Then the irreducible highest weight module $M^{\wh{\mathfrak{su}(2)}}_{k,\lambda}$ is integrable if and only if the highest weight $\lambda$ is an integer in the range $0\leq \lambda \leq k$. We denote this integrable module $H_{k,\lambda}$; it can be realized as the quotient of the Verma module
$V^{\wh{\mathfrak{su}(2)}}_{k,\lambda}$ by the submodule $\nu$ generated by two null vectors\footnote{In fact, $V^{\wh{\mathfrak{su}(2)}}_{k,\lambda}$ contains other null vectors, but they are contained in $\nu$.} 
\begin{equation}
\chi=(E_{-1})^{k-\lambda+1}v,\quad \psi =(F_0)^{\lambda+1}v.
\end{equation}
At depth $n_\delta=0$, $H_{k,\lambda}$ is the standard irreducible representation of $\mathfrak{sl}(2,\CC)$ of dimension $\lambda+1$ (or the ``representation of spin $\frac{\lambda}{2}$'').

As an illustration, consider the case $k=1$, $\lambda=0$. Here are the dimensions of first weight spaces (joint eigenspaces of $H$ and $\delta$), a.k.a. multiplicities of weights, in the Verma module $V^{\wh{\mathfrak{su}(2)}}_{1,0}$:\footnote{The generating function for the numbers in this table is 
$$(1-\alpha^2\tau)^{-1}\prod_{n=0}^\infty((1-\alpha^2 \tau^{2+n})(1-\tau^{1+n})(1-\alpha^{-2}\tau^n))^{-1}.$$
The coefficient of $\alpha^{2k}\tau^l$ in this function is the dimension of the weight space with $H$ eigenvalue $2k$ and $n_\delta=l$. This generating function counts the ``nondecreasing'' words made out of the ordered alphabet $E_{-1};F_{0},H_{-1},E_{-2};F_{-1},H_{-2},E_{-3};F_{-2},H_{-3},E_{-4},\ldots$ (ordered by $n_\delta %(\delta\mr{-eigenvalue})
-\frac12\,(H\mr{-eigenvalue})$) -- such words give a Poincar\'e-Birkhoff-Witt basis in $U(N_-)$ and hence in the Verma module.
}
\begin{equation}
\begin{array}{c|ccccccccccc}
n_\delta \;\backslash H\mr{-e.v.} &-10& -8&-6&-4&-2&0&2&4&6&8&10 \\ \hline
0& 1&1& 1&1&1&\boxed{1}&&&&&  \\
1&3& 3&3&3&3 &2&1 &&&& \\
2&9&9&9&9&8&6&3&1&&&  \\
3& 22&22 & 22&21&19&14&8&3&1&& \\
4&51&51&50&48&42&32&19&9&3&1&
\end{array}
\end{equation}
Here an empty cell means that the corresponding weight space is zero; we are indicating $H$-eigenvalue horizontally and $\delta$-eigenvalue vertically. The boxed entry corresponds to the highest vector $v$.  The cell at position $(2i,i)$ corresponds to the weight space $\CC\cdot (E_{-1})^i v$; the cell at position $(-2i,0)$ corresponds to the weight space $\CC\cdot (F_0)^iv$. 

The similar table of multiplicities for the integrable module $H_{1,0}$ is the following:\footnote{See Figure 14.4 and Table 15.1 in \cite{DMS}. For (\ref{l36 H11 multiplicities}) see Table 15.2 in \cite{DMS}.}
\begin{equation}
\begin{array}{c|ccccccc}
n_\delta \;\backslash H\mr{-e.v.}&-6&-4&-2&0&2&4&6 \\ \hline
0 & &&&\boxed{1}&&& \\
1& &&1&1&1&& \\
2&&&1&2&1&&\\
3&&&2&3&2&&\\
4&&1&3&5&3&1&
\end{array}
\end{equation}
This table illustrates e.g. that at the representation of $\mathfrak{sl}(2,\CC)$ arising at a fixed depth $n_\delta>0$ is finite-dimensional but generally not irreducible.

For the second integrable module arising at level $k=1$, $H_{1,1}$, the table of multiplicities is
\begin{equation}\label{l36 H11 multiplicities}
\begin{array}{c|cccccc}
n_\delta \;\backslash H\mr{-e.v.}& -5 & -3 & -1 & 1& 3& 5\\ \hline
0& &&1&\boxed{1}&& \\
1& &&1&1&&\\
2& &1&2&2&1&\\
3& &1&3&3&1&
\end{array}
\end{equation}
\end{example}

\subsection{Sugawara construction}

Sugawara construction is a realization of Virasoro algebra (with some particular value of of central charge) in terms of quadratic expressions in generators of the affine Lie algebra $\wh{\g}$. Put another way, it is an embedding $\mr{Vir}\hookrightarrow U(\wh{\g})$ of Virasoro into (the degree two part of) the enveloping algebra of $\wh\g$.

Let $\{T^a\}$ be an orthonormal basis in $\g$ with respect to the Killing form. The quadratic Casimir element 
\begin{equation}
\mr{Cas}\colon= \sum_a T^a T^a \in U(\g)
\end{equation}
acts on the irreducible $\g$-module with highest weight $\lambda$ %$M^\g_\lambda$
by multiplication by a constant $C_\lambda$,
\begin{equation}\label{l36 C_lambda}
\mr{Cas}=C_\lambda\cdot \mr{Id}\quad \mr{on}\; M^\g_\lambda.
\end{equation}

We also denote the normalized trace of the Casimir element in the adjoint representation of $\g$ by
\begin{equation}
h^\vee\colon=\frac{\tr_\g \mr{ad}(\mr{Cas})}{2\dim\g}
\end{equation}
-- it is the so-called dual Coxeter number of $\g$.

\begin{thm}[Sugawara, \cite{Sugawara}]
Let $W$ be a highest weight
 $\wh{\g}$-module  on which $\mathbb{K}$ acts by multiplication by a number $k\in \CC$, $k\neq -h^\vee$. Consider the elements 
\begin{equation}\label{l36 L_n Sugawara}
L_n=\frac{1/2}{k+h^\vee} \sum_{j\in\ZZ} \sum_{a=1}^{\dim\g} :T^a_j T^a_{n-j}:\qquad \in \mr{End}(W),
\end{equation}
where $T^a_i=T^a\otimes t^i$ and the normal ordering symbol $:\cdots:$ puts $T^a_{>0}$ to the right of $T^0_{<0}$.\footnote{Note that the normal ordering only affects the expression for $L_0$, as $T^a_j$ and $T^a_{n-j}$ commute for $n\neq 0$.}
Then:
\begin{enumerate}[(a)]
\item The operators $L_n$ satisfy Virasoro commutation relations with central charge
\begin{equation}
c=\frac{k\cdot \dim\g}{k+h^\vee}.
\end{equation}
\item The commutation relation between operators (\ref{l36 L_n Sugawara}) and the generators of $\wh\g$ is
\begin{equation}\label{l36 [L_n,X]}
[L_n,X_j]=-j X_{n+j}
\end{equation}
for any $X\in\g$.
\item  If $W=H_{k,\lambda}$ is an integrable $\wh\g$-module and $v$ is the highest weight vector, then one has
\begin{equation}\label{l36 L_0 v}
L_0 v= \frac{\frac12 C_\lambda}{k+h^\vee} v,
\end{equation}
with $C_\lambda$ the value of the quadratic Casimir in the representation $M^\g_\lambda$, as in (\ref{l36 C_lambda}).
\end{enumerate}
\end{thm}
For the proof see e.g. Theorem 10.1 and Proposition 10.1 in \cite{Kac_Raina}.

Comparing (\ref{l36 [L_n,X]}), (\ref{l36 L_0 v}) and (\ref{l36 [delta,X]}) we note that in the decomposition of the integrable module by depth
\begin{equation}
H_{k,\lambda}=\bigoplus_{n_\delta\geq 0} H_{k,\lambda}(n_\delta),
\end{equation}
the term $H_{k,\lambda}(n_\delta)$ in the r.h.s. is the eigenspace of $L_0$ with eigenvalue
\begin{equation}
%\frac{\frac12 C_\lambda}{k+h^\vee} 
\Delta+ n_\delta,
\end{equation}
where
\begin{equation}
\Delta=\frac{\frac12 C_\lambda}{k+h^\vee} 
\end{equation}
is the constant in (\ref{l36 L_0 v}).
Put another way, one has 
\begin{equation}
L_0=\Delta\cdot\mr{Id}+\delta
\end{equation}
as an equality of operators on
$H_{k,\lambda}$.

Also note that all elements of $H_{k,\lambda}(0)$ are annihilated by $L_{>0}$, i.e. they are all Virasoro-highest weight (or ``Virasoro-primary'') vectors with $L_0$-eigenvalue $\Delta$. %\marginpar{They are also $\wh\g$-primary} 
There may also be other Virasoro-primary vectors in $H_{k,\lambda}$ emerging at depths $n_\delta>0$.

\begin{example}
For $\g=\mathfrak{su}(2)$, one has $h^\vee=2$ (more generally, for $\g=\mathfrak{su}(N)$, one has $h^\vee=N$), thus (\ref{l36 L_n Sugawara}) becomes
\begin{equation}\label{l36 L_n su(2)}
L_n=\frac{1/2}{k+2} \sum_{j\in \ZZ}\sum_{a=1}^3 :T^a_j T^a_{n-j}:.
\end{equation}
For the orthonormal basis $\{T^a\}$ in $\mathfrak{su}(2)$, one can choose the appropriately normalized Pauli matrices,
\begin{equation}\label{l36 T^a basis}
T^1= \frac{1}{\sqrt{2}}\left(  \begin{array}{cc}
0&1\\ 1&0
\end{array} \right) ,\quad
T^2= \frac{1}{\sqrt{2}}\left(  \begin{array}{cc}
0&-i\\ i&0
\end{array} \right) ,\quad
T^3= \frac{1}{\sqrt{2}} \left(  \begin{array}{cc}
1&0\\ 0&-1
\end{array} \right) .
\end{equation}
The operators (\ref{l36 L_n su(2)}) satisfy Virasoro commutation relations with central charge
\begin{equation}
c=\frac{3k}{k+2}.
\end{equation}
For $W=H_{k,\lambda}$ an integrable $\wh{\mathfrak{su}(2)}$-module, the highest vector satisfies 
\begin{equation}
L_0 v=\frac{\frac14\lambda (\lambda+2)}{k+2} v,
\end{equation}
since for $\g=\mathfrak{su}(2)$ the value of the quadratic Casimir in an irreducible representation is 
\begin{equation}
C_\lambda=\frac12 \lambda (\lambda+2).
\end{equation}
\end{example}

\marginpar{Lecture 37,\\ 11/28/2022}

\section{Wess-Zumino-Witten model as a classical field theory}
Let $G$ be a compact simple, simply connected matrix group (keeping in mind $G=SU(2)$ in the fundamental representation as the main example).

Consider the following 3-form on $G$:
\begin{equation}\label{l37 sigma}
\sigma=\frac{1}{24\pi^2} \tr \left((X^{-1}dX)\wedge (X^{-1}dX)\wedge (X^{-1}dX)\right)\quad \in \Omega^3(G)
\end{equation}
It is known as the Cartan 3-form on $G$; it is left- and right-invariant under $G$-action and represents the image of the generator of $H^3(G,\ZZ)\simeq \ZZ$ in de Rham cohomology $H^3(G,\RR)$. In particular, the form $\sigma$ has integer periods. 

\begin{example}
For $G=SU(2)$ the group manifold is the 3-sphere and $\sigma$ is a volume form of unit total volume, $\int_G \sigma=1$. The funny normalization factor in (\ref{l37 sigma}) is tuned so as to have this property.
\end{example}

\begin{remark}
The form $\sigma$ is constructed out of the Maurer-Cartan 1-form 
\begin{equation}\label{l37 MC}
\mu=X^{-1}dX\quad \in \Omega^1(G,\g)
\end{equation}
-- the unique left-invariant $\g$-valued 1-form form on the group $G$ such that its value at the group unit $\mu|_{e}\colon \underbrace{T_e G}_{\g} \ra \g$ is identity. In terms of $\mu$, the Cartan 3-form is
\begin{equation}
\sigma=\frac{1}{48\pi^2}\langle \mu \stackrel{\wedge}{,}[\mu\stackrel{\wedge}{,}\mu] \rangle_\g.
\end{equation}
\end{remark}

\subsubsection{The action functional.}
Let $\Sigma$ be a closed Riemannian surface. Fields of the model are smooth maps 
\begin{equation}
g\colon \Sigma \ra G
\end{equation}
 and the action functional is\footnote{Recall that $\bdd =dz \frac{\dd}{\dd z}$, $\bar\bdd=d\bar{z}\frac{\dd}{\dd \bar{z}}$ are the holomorphic and antiholomorphic Dolbeault differentials.}
\begin{equation}\label{l37 S_WZW}
S_\Sigma(g)\colon=-\frac{i}{4\pi}\int_\Sigma \tr\left(g^{-1}\bdd g\wedge g^{-1}\bar\bdd g\right) + \mr{WZ}(g)
%-\frac{i}{12\pi}\int_B \tr\left(\til{g}^{-1}d\til{g}\right)^{\wedge 3} \\
%= -\frac{1}{8\pi} \int_\sigma \langle g^*\mu \stackrel{\wedge}{,} * g^* \mu\rangle_\g-2\pi i \int_B \til{g}^* \sigma
\end{equation}
where the last term is the so-called ``Wess-Zumino term.'' It is defined as
\begin{equation}\label{l37 WZ}
\mr{WZ}(g)\colon=-\frac{i}{12\pi}\int_B \tr\left(\til{g}^{-1}d\til{g}\right)^{\wedge 3}= -2\pi i \int_B \til{g}^* \sigma
\end{equation}
where $B$ is any compact oriented 3-manifold with boundary $\dd B=\Sigma$  (e.g. one can choose $B$ to be a handlebody)\footnote{One says ``the 3-manifold $B$ \emph{cobounds} the surface $\Sigma$.''} and $\til{g}\colon B\ra G$ any smooth extension of the map $g\colon \Sigma \ra G$ into $B$ (``extension'' means that $\til{g}$ must satisfy $\til{g}|_{\dd B}=g$).

\begin{lemma}
For a fixed map $g\colon \Sigma\ra G$, the Wess-Zumino term $\mr{WZ}(g)$ modulo $2\pi i \ZZ$ does not depend on the choice of 3-manifold $B$ cobounding $\Sigma$  and on the choice of extension $\til{g}$.
\end{lemma}

\begin{proof}[Sketch of proof]
Denote by  $\mr{WZ}^{B,\til{g}}(g)$ the r.h.s. of (\ref{l37 WZ}). Let $B,B'$ be two 3-manifolds cobounding $\Sigma$ and $\til{g},\til{g}'$ some extensions of $g$ from $\Sigma$ into $B$ and into $B'$, respectively. One has
\begin{multline}
\mr{WZ}^{B,\til{g}}(g)-\mr{WZ}^{B',\til{g}'}(g)=-2\pi i \left(\int_B \til{g}^* \sigma - \int_{B'} (\til{g}')^*\sigma\right)=\\
=-2\pi i \left(\int_B \til{g}^* \sigma + \int_{\ol{B}'} (\til{g}')^*\sigma\right)
=-2\pi i \int_{\check{B}}\check{g}^* \sigma,
\end{multline}
where $\ol{B}'$ is $B'$ with reversed orientation. Here in the last step we defined the closed 3-manifold $\check{B}$ as $B$ glued to $\ol{B}'$ along $\Sigma$, and we defined the ``glued'' map $\check{g}\colon \check{B}\ra G$ as the map whose restrictions to $B$, $\ol{B}'$ are $\til{g}$ and $\til{g}'$, respectively.
%\textcolor{red}{PICTURE}
\begin{figure}[H]
\begin{center}
\includegraphics[scale=0.7]{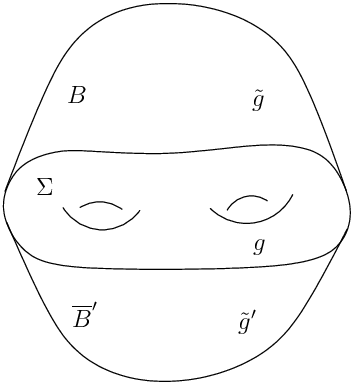}
\end{center}
\caption{Closed 3-manifold $\check{B}$ glued out of $B$ and $\ol{B}'$ along $\Sigma$ and the corresponding glued map to $G$.
}
\end{figure}
 Thus, one has 
\begin{equation}\label{l37 WZ jump}
\mr{WZ}^{B,\til{g}}(g)-\mr{WZ}^{B',\til{g}'}(g)= -2\pi i \langle [\check{B}],\check{g}^* [\sigma] \rangle \in 2\pi i\ZZ
\end{equation}
-- the pairing (up to normalization) of the fundamental class of the closed 3-manifold $\check{B}$ with the pullback along $\check{g}$ of the integral cohomology class $[\sigma]\in H^3(G,\ZZ)$.\footnote{By abuse of notations, here
$[\sigma]$ stands for the class in $H^3(G,\ZZ)$ whose image in $H^3(G,\RR)$ is the class of the Cartan 3-form in de Rham cohomology. We also remark that in the special case $G=SU(2)$ the r.h.s. of (\ref{l37 WZ jump}) admits the interpretation as $-2\pi i$ times the degree of the map  $\check{g}\colon\check{B}\ra SU(2)\simeq S^3$ between oriented closed 3-manifolds.
}
\end{proof}

In particular, this lemma implies that Wess-Zumino-Witten (WZW) action (\ref{l37 S_WZW}) modulo $2\pi i \ZZ$ is well-defined (independent of choices of cobounding 3-manifold $B$ and the extension $\til{g}$). Thus, 
for $k=1,2,3,\ldots$ an integer (the ``level'' of Wess-Zumino-Witten model), the expression
\begin{equation}
e^{-k S_{\Sigma}(g)}
\end{equation}
is well-defined. This expression is the integrand in the path integral for the Wess-Zumino-Witten model,
\begin{equation}
Z_k(\Sigma) = ``\int_{\mr{Map}(\Sigma,G)} \mc{D}g\; e^{-k S_{\Sigma}(g)}."
\end{equation}
Here the level $k=1,2,3,\ldots$ is a parameter of the theory playing the role of inverse Planck constant, $k=\hbar^{-1}$, see Remark \ref{l3 rem hbar}.

\begin{remark}
\begin{enumerate}[(a)]
\item In the action (\ref{l37 S_WZW}) the first term is real and the second term is imaginary.
\item One can write the action (\ref{l37 S_WZW}) in terms of the Maurer-Cartan 1-form on $G$:
\begin{equation}
S_\Sigma(g)=-\frac{1}{8\pi}\int_\Sigma \langle g^*\mu\stackrel{\wedge}{,}*_\mr{Hodge}g^* \mu\rangle_\g \underbrace{-\frac{i}{24\pi} \int_B \til{g}^* \langle \mu\stackrel{\wedge}{,}[\mu\stackrel{\wedge}{,}\mu] \rangle_\g}_{\mr{WZ}(g)}
\end{equation}
The benefit of this rewriting is that it one can use it to define WZW action for non-matrix Lie groups.
\item Although the Wess-Zumino term is non-local (not an integral over $\Sigma$), its variation is local: 
\begin{equation}
\delta \mr{WZ}=\frac{i}{4\pi}\int_\Sigma \tr\, g^{-1}\delta g (\bdd(g^{-1}\bar\bdd g)+\bar\bdd(g^{-1}\bdd g))
\end{equation} 
(note that the integral is over $\Sigma$, not over $B$). 
Putting this together with the variation of the first term of (\ref{l37 S_WZW}) (let us denote it $E(g)$),
\begin{equation}
\delta E= \frac{i}{4\pi}\int_\Sigma \tr\, g^{-1} \delta g (\bdd(g^{-1}\bar\bdd g)-\bar\bdd (g^{-1}\bdd g)),
\end{equation}
one obtains the variation of the full action (\ref{l37 S_WZW}) is%\marginpar{fix the normalization constant}
\begin{equation}\label{l37 delta S}
\delta S_\Sigma = \frac{i}{2\pi}\int_\Sigma \tr \,(g^{-1}\delta g)\bdd (g^{-1}\bar\bdd g).
\end{equation}
An equivalent expression is:
\begin{equation}
\delta S_\Sigma =-\frac{i}{2\pi}\int_\Sigma \tr (\delta g\,g^{-1}) \bar\bdd(\bdd g\, g^{-1}).
\end{equation}
\end{enumerate}
\end{remark}

\subsubsection{Euler-Lagrange equation.}
For the discussion of the Euler-Lagrange equation (especially the holomorphic factorization of solutions (\ref{l37 holom factorization})) and symmetries it is convenient to complexify the space of classical fields, i.e., to allow fields $g$ to be maps from $\Sigma$ to the \emph{complexified} group $G_\CC$ rather than the compact group $G$.

The Euler-Lagrange equation corresponding to the action (\ref{l37 S_WZW}) is read off from the fromula for the variation (\ref{l37 delta S}):
\begin{equation}\label{l37 EL}
\bdd (g^{-1}\bar\bdd g)=0.
\end{equation}
Equivalently, the same equation can be written as $\bar\bdd(\bdd g\, g^{-1})=0$.

The general solution of the Euler-Lagrange equation (\ref{l37 EL}) is:
\begin{equation}\label{l37 holom factorization}
g(z)=h_1(z) \ol{h_2(z)},
\end{equation}
where $h_1,h_2\colon \Sigma\ra G_\CC $ are two \emph{holomorphic} maps into the complexified group. 
%(here it is more convenient to consider solutions of (\ref{l37 EL}) in maps from $\Sigma$ to the complexified group $G_\CC$ rather than the compact group $G$ itself).

\begin{remark} One can  consider Wess-Zumino-Witten theory for $G=U(1)$. (This group fails our assumptions: it is neither simple nor simply connected, but nevertheless one can play with it.) Then the field $g\colon \Sigma \ra G$ can be parametrized as $g=e^{i\phi}$. The action (\ref{l37 S_WZW}) is then simply the action of a free boson (with values in $S^1$); the Wess-Zumino term vanishes. Euler-Lagrange equation (\ref{l37 EL}) becomes the equation of a harmonic function $\Delta \phi=0$. The factorization (\ref{l37 holom factorization}) simply becomes the statement that any harmonic function is a sum of a holomorphic and an antiholomorphic function, $\phi(z)=\chi_1(z)+\ol{\chi_2(z)}$.
\end{remark}

\subsubsection{Symmetry and conserved currents.}
The action (\ref{l37 S_WZW}) is invariant under the following transformations of the field:
\begin{equation}\label{l37 symmetry}
g(z)\mapsto g'(z)=\Omega_1(z) g(z) \ol{\Omega_2(z)}
\end{equation}
where $\Omega_1,\Omega_2\colon \Sigma \ra G_\CC$ are two arbitrary holomorphic maps.\footnote{The transformations (\ref{l37 symmetry}) are sometimes called ``gauge symmetry'' in the literature. We would argue that it is not a very good term here, since the generators of the symmetry are not local: they are holomorphic (rather than, say, smooth) maps from $\Sigma$ to the target, and for holomorphic maps one doesn't have partitions of unity, so one cannot have a bump function as a generator.}

The invariance under transformations (\ref{l37 symmetry}) corresponds by Noether theorem to having two conserved currents
\begin{equation}\label{l37 J Jbar}
\begin{gathered}
\mathbf{J}=\bdd g\cdot g^{-1}\in \Omega^{1,0}(\Sigma,\g), \\
\ol{\mathbf{J}}=g^{-1}\bar\bdd g\in \Omega^{0,1}(\Sigma,\g),
\end{gathered}
\end{equation}
satisfying the conservation properties
\begin{equation}\label{l37 J Jbar conservation}
\bar\bdd \bJ\simEL 0,\quad \bdd\ol{\bJ} \simEL 0.
\end{equation}

\begin{remark} The action (\ref{l37 S_WZW}) is the sum of the action of a sigma model with target a group (the natural quadratic ``energy of a map'') and a seemingly  complicated nonlocal cubic term $\mr{WZ}(g)$.
One might reasonably ask: why add this extra term to the sigma model? The answer is that adding this term actually makes the model much simpler: it creates two separately conserved holomorphic and antiholomorphic Noether currents $\bJ,\ol{\bJ}$, leads to simpler Euler-Lagrange equation which allows an explicit solution (\ref{l37 holom factorization}).
Ultimately, the addition of the Wess-Zumino term to the model results in the factorization of the model into a holomorphic and an antiholomorphic sector (this statement makes sense both at the classical and at the quantum level).
\end{remark}

\begin{remark}
One can also write down the currents (\ref{l37 J Jbar}) without referring to the matrix structure of the group $G$:
\begin{equation}
\bJ=\frac12 (\mr{id}+i *_\mr{Hodge})g^* \mu_{R},\quad
\ol{\bJ}=\frac12 (\mr{id}-i *_\mr{Hodge})g^* \mu_{L},
\end{equation}
where $\mu_L$ is the left-invariant Maurer-Cartan form (\ref{l37 MC}) and $\mu_R$ is its right-invariant counterpart ($\mu_R= dX\, X^{-1}$ for a matrix group).
\end{remark}

\subsubsection{Polyakov-Wiegmann formula.}
For the next discussion it is important to know how the action (\ref{l37 S_WZW}) interacts with pointwise products of fields (as maps to the group). 
\begin{thm}[Polyakov-Wiegmann]
For $\Sigma$ a closed Riemannian surface and $f,g\colon \Sigma \ra G$ two maps to the group, one has
\begin{equation}\label{l37 Polyakov-Wiegmann}
S_\Sigma(f\cdot g) = S_\Sigma(f)+S_\Sigma(g)+\underbrace{\frac{i}{2\pi}\int_\Sigma \tr \left(f^{-1}\bar\bdd f\wedge \bdd g\cdot g^{-1}\right)}_{\Gamma_\Sigma(f,g)}.
\end{equation}
Here $\cdot$ in the l.h.s. stands for the pointwise product of maps to $G$.
\end{thm}
Thus, the action is ``almost'' additive w.r.t. pointwise product of fields, with the defect given by the rightmost term in (\ref{l37 Polyakov-Wiegmann}) which we denoted $\Gamma_\Sigma(f,g)$. 

We note that the ``defect'' $\Gamma_\Sigma$ in (\ref{l37 Polyakov-Wiegmann}) is a 2-cocycle for the group $\mr{Map}(\Sigma,G)$ (with trivial coefficients), i.e., for any triple of maps  $f,g,h\colon \Sigma \ra G$ it satisfies\footnote{
Indeed, $0=S_\Sigma((fg)h)-S_\Sigma(f(gh))=S_\Sigma(fg)+S_\Sigma(h)+\Gamma_\Sigma(fg,h)-S_\Sigma(f)-S_\Sigma(gh)-\Gamma_\Sigma(f,gh)=\cancel{S_\Sigma(f)}+\cancel{S_\Sigma(g)}+ \Gamma_\Sigma(f,g)+\cancel{S_\Sigma(h)}+\Gamma_\Sigma(fg,h)-\cancel{S_\Sigma(f)}-\cancel{S_\Sigma(g)}-\cancel{S_\Sigma(h)}-\Gamma_\Sigma(g,h)-\Gamma_\Sigma(f,gh) =-\mr{l.h.s.\;of\; (\ref{l37 2-cocycle})}$.
}
\begin{equation}\label{l37 2-cocycle}
\Gamma_\Sigma(g,h)-\Gamma_\Sigma(fg,h)+\Gamma_\Sigma(f,gh)-\Gamma_\Sigma(f,g)=0.
\end{equation}

\subsection{%Wess-Zumino-Witten action on 
Case of surfaces with boundary}
Here we briefly sketch a geometric construction from \cite{Kohno}.

It is not straightforward to generalize the action (\ref{l37 S_WZW}) to surfaces with boundary, due to the presence of a nonlocal term in the action. It turns out one can still do it, with two caveats:
\begin{itemize}
\item one should consider the exponential of the action $e^{-k S_\Sigma}$ instead of the action itself (we assume that the level $k=1,2,3,\ldots$ is fixed),
\item instead of obtaining $e^{-k S_\Sigma}$ as a function on the space of fields on a surface with boundary, it will be a section of a certain line bundle over $\F_\Sigma$.
\end{itemize}

Let $\Sigma$ be a compact Riemannian surface with $n$ boundary circles. Construct a closed surface $\Sigma'$ by attaching $n$ disks $D_1,\ldots,D_n$ to the boundary of $\Sigma$ (i.e. attach a disk to each boundary circle): $\Sigma'= \Sigma\cup \bigsqcup_{i=1}^n D_i$.

%\textcolor{red}{PICTURE}
\begin{figure}[H]
\begin{center}
\includegraphics[scale=0.9]{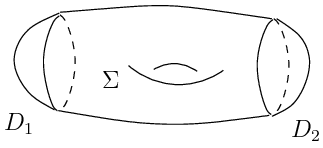}
\end{center}
\caption{Closed surface obtained from $\Sigma$ by attaching disks along boundary circles.
}
\end{figure}

The basic idea is to define the WZW action on  a surface with boundary via
\begin{equation}\label{l37 S for Sigma with boundary}
e^{-k S_\Sigma(g)}\colon= e^{-k S_{\Sigma'}(g')}
\end{equation}
where $g$ is a map $\Sigma\ra G$ and $g'$ is some extension of $g$ as a map $\Sigma'\ra G$ (i.e. an extension of the map $g$ into each disk $D_i$ is to be chosen).

The ambiguity in the choice of the extension $g'$ leads to the idea that the expression $e^{-k S_\Sigma(g)}$ should be understood as taking values in the fiber of the complex line bundle
\begin{equation}\label{l37 line bun over LG^n}
\begin{CD}
\LL^k\boxtimes\cdots \boxtimes \LL^k \\ @VVV \\ LG\times \cdots \times LG
\end{CD}
\end{equation}
over the point $g|_{\dd \Sigma}\in \mr{Map}(\dd \Sigma,G)\simeq LG^{\times n}$, i.e., over the boundary value of the map $g$ seen as a collection of loops in $G$. 

The complex line bundle  over the loop group
\begin{equation}\label{l37 L^k}
%\begin{CD}
%\LL^k\\ @VVV \\ LG
%\end{CD}
\LL^k\ra LG,
\end{equation}
%$\LL^k$ over the loop group 
several copies of which appear in (\ref{l37 line bun over LG^n}), is constructed as follows (see \cite{Kohno} for details). Consider the trivial line bundle 
\begin{equation}\label{l37 triv line bun}
\mr{Map}(D,G)\times \CC \ra \mr{Map}(D,G)
\end{equation} 
with $D$ the unit disk, and consider the following equivalence relation: two 
pairs 
\begin{equation}\label{l37 equivalence}
{(f_D\colon D\ra G,u\in \CC)}\;\;\sim\;\; (g_D\colon D\ra G,v\in \CC)
\end{equation} 
are considered equivalent if
\begin{itemize}
\item $f_D$ and $g_D$ agree on the boundary circle: $f_D|_{\dd D}=g_D|_{\dd D}$,
\item one has 
$$v=u\cdot e^{-k S_{\CP^1}(h)-k \Gamma_D(f_D,h_D)},$$ 
where
$h\colon \CP^1\ra G$ is defined on $D$ as $f_D^{-1}g_D=:h_D$ and extended by $1$ to $\CP^1\backslash D$;  $\Gamma_D$ is given by the same formula as in (\ref{l37 Polyakov-Wiegmann}) (but the integral is over $D$).
\end{itemize}
Quotienting the line bundle (\ref{l37 triv line bun}) by this equivalence relation produces a line bundle over $LG$ (loops in $G$ seen as boundary values of functions $f_D$) which we call $\LL^k$.

By construction and as a consequence of Polyakov-Wiegmann formula, one indeed has that $e^{-k S_\Sigma(g)}$ defined by (\ref{l37 S for Sigma with boundary}) seen as an element in the fiber of (\ref{l37 line bun over LG^n}) over the boundary value of $g$ is independent of the extension $g'$.\footnote{
In a bit more detail, one chooses an extension $g'$ of $g\colon \Sigma\ra G$ into the disks $D_i$ and thinks of $e^{-k S_\Sigma(g)}$ as a tuple $\left(\{g'|_{D_i}\}\in \mr{Map}(D,G)^{\times n},e^{-k S_{\Sigma'}(g')}\in \CC\right)$ up to an equivalence as the one on (\ref{l37 triv line bun}), extended in an obvious way to $n$ disks.

For instance, if $\Sigma$ has a single boundary circle (i.e. $n=1$), if $g'$ and $g''$ are two extensions of the map $g\colon \Sigma\ra G$ into the single attached disk $D$, one has $g''=g'h$ with the map $h\colon \Sigma'\ra G$ (``discrepancy'' of the two extensions) equal to $1$ on $\Sigma$ and nontrivial in $D$, the pairs $(g'h|_D,e^{-k S_\Sigma(g'h)})$ and $(g'|_D, e^{-k S_\Sigma(g')})$ are equivalent precisely because by Polyakov-Wiegmann formula one has $e^{-k S_\Sigma(g'h)}=e^{-k S_\Sigma(g')}\cdot e^{-k S_{\Sigma}(h)-k \Gamma_\Sigma(g',h)}$. We note that in the r.h.s. here $\Gamma_\Sigma$ can be replaced with $\Gamma_D$ and $S_{\Sigma}(h)$ can be replaced with $S_{\CP^1}(\til{h})$ where $\til{h}$ is the extension of $h|_D$ into $\CP^1\backslash D$ by $1$ (the intuition here is that since $h$ is trivial outside $D$, the surface $\Sigma$ can be replaced by anything, including a complementary disk).
}

Put another way, the exponentiated action for $\Sigma$ a surface with boundary is not a function on $\F_\Sigma=\mr{Map}(\Sigma,G)$ but rather is a section of a line bundle,
\begin{equation}
e^{-kS_\Sigma} \in \Gamma(\F_\Sigma, \pi^* (\LL^k)^{\boxtimes n})
\end{equation}
where 
\begin{equation}
\pi\colon \F_\Sigma \ra \underbrace{\mr{Map}(\dd \Sigma,G)}_{\F_{\dd\Sigma}}\simeq LG^{\times n}
\end{equation}
is the restriction of the map to the boundary. We will denote  $\LL_{\dd \Sigma}\colon=(\LL^k)^{\boxtimes n}$ seen as a line bundle over $\F_{\dd \Sigma}$.

\begin{remark} 
\begin{enumerate}[(a)]
\item Denoting $\LL^1=: \LL$ one has 
\begin{equation}
\LL^k= \LL^{\otimes k}.
\end{equation}
Thus, the superscript in $\LL^k$ can be interpreted as the tensor power of a special line bundle corresponding to $k=1$. The first Chern class of the bundle $\LL$ 
%in  
%de Rham 
%cohomology $H^2(LG)$ is %represented by the 2-form\marginpar{Maybe write it more cleanly as pull-push in cohomology, not in de Rham representatives?}
is given by
\begin{equation}
[\omega] = p_*( \mr{ev}^* [\sigma])\quad \in H^2(LG),
\end{equation}
where $[\sigma]$ is the cohomology class of Cartan 3-form (\ref{l37 sigma}) in $H^3(G)$;  $p$ and $\mr{ev}$ are the projection and evaluation maps in the diagram
\begin{equation}
\begin{CD}
LG\times S^1 @>{\mr{ev}}>> G \\
@VpVV \\
LG
\end{CD}
\end{equation}
The map $\mr{ev}$ evaluates the loop in $G$ at a given point of $S^1$; $p_*$ stands for the pushforward in cohomology  (fiber integral over $S^1$).
\item One has a product on the total space of the line bundle $\LL^k$ given by
\begin{equation}
(g_1, u_1e^{-k S_D(g_1)}) * (g_2, u_2 e^{-k S_D(g_2)}) = (g_1\cdot g_2\ ,\  u_1u_2 e^{-k S_D(g_1 g_2)-k \Gamma_D(g_1,g_2)})
\end{equation}
with $u_{1,2}\in \CC$ and $g_{1,2}\colon D\ra G$.
Here we understand that on both sides we pass to equivalence classes under (\ref{l37 equivalence}).
Removing the zero-section from $\LL^k$ one obtains a group which is none other than the central extension of the loop group we mentioned in Section \ref{ss affine Lie algebras} (see (\ref{l36 LG central extension})):
\begin{equation}
\LL^k\backslash \{\mbox{zero-section}\}= \wh{LG}^k.
\end{equation}
\end{enumerate}
\end{remark}

\subsubsection{Symmetry of the model on a surface with boundary.} 
Fix $ \Sigma$ a surface with boundary. 
One has a left and a right action of the group $\mr{Map}(\Sigma,G)$ on $\F_\Sigma=\mr{Map}(\Sigma,G)$ coming from multiplication in the target $G$ from the left or from the right. One also has left and right actions of the group $\mr{Map}(M,G)$ on the space of sections of the line bundle $\LL_{\dd\Sigma}\ra \F_{\dd\Sigma}$. %\marginpar{write formulas for the action} 

The symmetry (\ref{l37 symmetry}) for $\Sigma$ with boundary becomes the following statement. 
\begin{lemma}\label{l37 lemma: invariance with bdry}
The exponentiated action
\begin{equation}\label{l37 symmetry with bdry}
e^{-k S_{\Sigma}}\in \Gamma(\F_\Sigma,\pi^* \LL_{\dd\Sigma})
\end{equation}
is
\begin{itemize} 
\item left-invariant under \emph{holomorphic} maps $\Omega\colon \Sigma \ra G_\CC$ and 
\item right-invariant under \emph{antiholomorphic} maps $\Omega^* \colon \Sigma \ra G_\CC$,
\end{itemize} 
where maps act both on the fields and on the bundle $\LL_{\dd\Sigma}$ in (\ref{l37 symmetry with bdry}).
\end{lemma}
For the proof see \cite[Proposition 1.11]{Kohno}.

\subsubsection{Path integral heuristics.} The path integral on a surface with boundary
\begin{equation}\label{l37 PI Ward identity}
Z(\Sigma)=\int_{g|_{\dd\Sigma}=g_\dd } \mc{D} g\; e^{-k S_\Sigma(g)} \quad \in \Gamma(\F_{\dd\Sigma},\LL_{\dd\Sigma})^{\mr{Hol}(\Sigma,G_\CC)\times \mr{Antihol}(\Sigma,G_\CC)}
\end{equation}
is to be thought of as averaging the exponentiated action over fields with fixed boundary value $g_\dd\in \F_\dd$, and the value of the path integral is not a number but an element in the line
$\LL_{\dd \Sigma}|_{g_\dd}$. Thus, considering the path integral with all possible boundary conditions one has a section of $\LL_{\dd\Sigma}$. By the invariance property of the exponentiated action (Lemma \ref{l37 lemma: invariance with bdry}), this section should be invariant under holomorphic maps $\Sigma\ra G_\CC$ acting from the left and antiholomorphic maps $\Sigma\ra G_\CC$ acting from the right. 
This invariance property of the path integral is a variant of Ward identity.

\marginpar{Lecture 38,\\ 11/30/2022}
\section{Quantum Wess-Zumino-Witten model}
We fix as before a compact simple simply-connected group $G$ and a level $k=1,2,3,\ldots$. It is possible to quantize the classical WZW theory -- either by canonical/geometric quantization in the hamiltonian formalism, or by path integral. Here we will just outline the resulting (quantum) CFT.

\subsubsection{Space of states/space of fields.}
The space of states of the model associated to the circle -- or equivalently the space of fields $V$ --
is
\begin{equation}\label{l38 WZW space of states}
\HH=\bigoplus_{\lambda\in I_k} H_{k,\lambda}\otimes H^*_{k,\lambda}
\end{equation}
where the sum is over integrable highest weight modules of $\wh{\g}$ at level $k$ (we denote the set of corresponding highest weights $I_k$); the summand is a tensor product of the integrable module and the dual one, seen as a module over $\wh{\g}\oplus \wh{\bar{\g}}$ -- two copies of the affine Lie algebra.

The space $\HH_k$ can be identified with the space of sections of the line bundle $\LL^k$ over the loop group (\ref{l37 L^k}) via the inclusion
\begin{equation}\label{l38 inclusion}
\begin{array}{ccc}
\bigoplus_{\lambda\in I_k} H_{k,\lambda}\otimes H^*_{k,\lambda} & \ra & \Gamma(LG,\LL^k) \\
\phi_0\in \mr{End}(H_{k,\lambda}) &\mapsto& \left( \phi(g_+g_0 g_-)\colon=\tr_{M^\g_\lambda}(\phi_0\cdot \rho_\lambda(g_0))\cdot e^{-k S_D(g_+ g_0 g_-)} \right)
\end{array} 
\end{equation}
Here $g_\pm$ are a holomorphic and an antiholomorphic map $D\ra G_\CC$, taking value $1\in G$ at a base point on the boundary circle $1\in \dd D$; $g_0\colon D\ra G_\CC$ is a constant map; $\rho_\lambda(g)$ is the linear operator on $M^\g_\lambda$ representing the action of the group element $g$. In (\ref{l38 inclusion}) both sides carry a natural action of $\wh{\g}\oplus \wh{\bar\g}$ and these actions are intertwined by the inclusion.

By Sugawara construction, $\HH_k$ carries an action of two copies of Virasoro algebra $\mr{Vir}\oplus \ol{\mr{Vir}}$, with central charges 
\begin{equation}\label{l38 c cbar}
c=\bar{c}=\frac{k\dim G}{k+2}.
\end{equation}

\subsubsection{Quantum currents.}
Let $\{T^a\}$ be a fixed orthonormal basis in $\g$ and let $f^{abc}$ be the structure constants of $\g$ in this basis defined by $[T^a,T^b]=\sum_c f^{abc}T^c$. Components of Noether currents (\ref{l37 J Jbar}) become in the quantum setting certain local quantum fields -- elements in the space of fields $V$:
\begin{equation}
J^a, \ol{J}{}^a \in V,\qquad a=1,\ldots, \dim G,
\end{equation}
which are holomorphic/antiholomorphic,\footnote{Under a correlator with any collection of test fields, or as local field operators.}
\begin{equation}
\bar\dd J^a=0,\quad \dd \ol{J}{}^a=0
\end{equation}
(as a reflection of the classical conservation laws (\ref{l37 J Jbar conservation}))
and satisfy the OPEs
\begin{eqnarray}
J^a(w) J^b(z) &\sim &\frac{k \delta^{ab}\mathbb{1}}{(w-z)^2}+\frac{\sum_c f^{abc}J^c(z)}{w-z}+\mr{reg.}  \label{l38 JJ OPE},\\
\ol{J}{}^a(w) \ol{J}{}^b(z) &\sim &\frac{k \delta^{ab}\mathbb{1}}{(\bar{w}-\bar{z})^2}+\frac{\sum_c f^{abc}\ol{J}{}^c(z)}{\bar{w}-\bar{z}}+\mr{reg.} ,\\
J^a(w) \ol{J}{}^b(z) &\sim & \mr{reg.} \label{l38 J Jbar OPE}
\end{eqnarray}

The field $J^a$ acts on the space of states by a local field operator $\wh{J}^a(z)$; one can introduce the corresponding mode operators $\wh{J}^a_n\in \mr{End}(\HH)$ as
\begin{equation}
\wh{J}^a_n \colon = \frac{1}{2\pi i}\oint dz \,z^n\, \wh{J}(z)
\end{equation}
where the integral is over a contour going about the origin. 
Equivalently, we have
\begin{equation}\label{l38 Jhat via Jhat_n}
\wh{J}(z)=\sum_{n\in\ZZ} z^{-n-1}\wh{J}^a_n.
\end{equation}
Repeating the computation of Section \ref{sss: Virasoro from TT OPE}, we obtain from the OPE (\ref{l38 JJ OPE}) the commutation relations between the mode operators
\begin{equation}
[\wh{J}^a_n,\wh{J}^b_m]=\sum_c f^{abc} \wh{J}^c_{n+m} +kn \delta_{n,-m}\wh{\mathbb{1}}
\end{equation}
Note that these are exactly the commutation relations of the affine Lie algebra $\wh\g$. 
Comparing with the notations in (\ref{l36 ghat bracket 2}), we have the identification $\wh{J}^a_n=T^a_n=T^a\otimes t^n$.
Likewise one introduces the mode operators $\wh{\ol{J}}{}^a_n$ for the antiholomorphic current $\ol{J}{}^a$ which again satisfy the commutation relations of $\wh\g$ and commute with the mode operators $\wh{J}^a_n$ (due to (\ref{l38 J Jbar OPE})).
 Therefore, the action of $\wh\g\oplus\wh{\bar\g}$ on the space of states is realized by the mode operators generated by the currents $J$, $\ol{J}$.
 
Similarly to the action on the space of states, we have a local action of $\wh{\g}$ on fields at a point $z$ given by local mode operators $J^a_n\in \mr{End}(V_z)$ defined by
\begin{equation}
J^a_n \Phi(z)\colon = \frac{1}{2\pi i}\oint_{\gamma_z} dw\, (w-z)^n J(w) \Phi(z)
\end{equation}
for any field $\Phi(z)\in V_z$; $\gamma_z$ is a contour going around $z$. Equivalently, the mode operators yield the coefficients in the OPE of a field at $z$ with the current:
\begin{equation}\label{l38 J Phi OPE via modes}
J^a(w) \Phi(z)\sim\sum_{n\in \ZZ} (w-z)^{-n-1}J^a_n\Phi(z).
\end{equation}
One has a similar local action of $\wh{\bar\g}$ on $V_z$ generated by local mode operators of $\ol{J}$.

\subsubsection{The $\wh{\g}$-primary multiplet.}
Fix $\lambda$ a weight of an integrable $\wh\g$-module $H_{k,\lambda}$. Let $e^p$ be a basis in the irreducible $\g$-module $M^\g_\lambda$ (which is also the depth-zero component $H_{k,\lambda}(0)$ of the corresponding integrable $\wh\g$-module). We have a collection (``multiplet'') of $\wh\g\oplus\wh{\bar\g}$-primary fields $\phi^{p\bar{p}}_{\lambda}$  (primary here means ``annihilated by $J^a_{>0},\ol{J}^a_{>0}$'') corresponding to coordinates of a vector in
\begin{equation}
M^\g_\lambda \otimes (M^\g_\lambda)^* = H_{k,\lambda}(0)\otimes H_{k,\lambda}(0)^* \subset V
\end{equation}
By (\ref{l38 J Phi OPE via modes}) and the primary property, we have
\begin{equation}\label{l38 J phi OPE}
J^a(w) \phi_{\lambda}^{p\bar{p}}(z) \sim \frac{\sum_q (T^a_\lambda)^p_q \phi_\lambda^{q\bar{p}}(z)}{w-z}+\mr{reg.}
\end{equation}
where $T^a_\lambda$ is the matrix representing $T^a\in \g$ as an operator on $M^\g_\lambda$.

\subsubsection{Stress-energy tensor.}
The quantum stress-energy tensor of the model is the field
\begin{equation}\label{l38 T}
T(z)=\frac{1/2}{k+h^\vee} \sum_{a} :J^a(z) J^a(z):
\end{equation}
it satisfies the standard TT OPE (\ref{l23 TT}) with central charge (\ref{l38 c cbar}); the expression for $\ol{T}$ is similar (replacing $J$ with $\ol{J}$).

Normal ordering in (\ref{l38 T}) refers to the following definition: for local fields $\Phi_1,\Phi_2$ their normally ordered product $:\Phi_1(z)\Phi_2(z):$ is defined as the constant term in the OPE $\Phi_1(w)\Phi_2(z)$ or equivalently
\begin{equation}
:\Phi_1(z)\Phi_2(z): = \lim_{w\ra z} \left(\Phi_1(w)\Phi_2(z)- [\Phi_1(w)\Phi_2(z)]_\mr{sing}\right) = \frac{1}{2\pi i} \oint_{\gamma_z} \Phi_1(w)\Phi_2(z),
\end{equation}
where $[\cdots]_\mr{sing}$ is the singular part of the OPE and $\gamma_z$ is the contour around $z$.

\begin{remark} The classical Hilbert stress-energy tensor in Wess-Zumino-Witten theory, obtained as a variation w.r.t. the metric, is given by the formula (\ref{l38 T}) without the normal ordering and without the $h^\vee$ shift in the denominator. In this regard, the shift by $h^\vee$ should be understood as a quantum correction: it must be incroporated in the quantum picture, otherwise $T$ would not satisfy the OPE of the standard form (\ref{l23 TT}).
\end{remark}

\begin{remark} Note that substituting the mode expansion of the current (\ref{l38 Jhat via Jhat_n}) into the stress-energy tensor (\ref{l38 T}) we obtain the Sugawara formula (\ref{l36 L_n Sugawara}) expressing Virasoro generators in terms of generators of $\wh\g$:
\begin{equation}\label{l38 Sugawara}
\wh{L}_n=\frac{1/2}{k+h^\vee}\sum_a :\wh{J}^a_m \wh{J}^a_{n-m}:.
\end{equation}
 In this sense, the construction of the stress-energy tensor (\ref{l38 T}) is a restatement of Sugawara construction.
\end{remark}

\begin{remark} The counterpart of the formula (\ref{l38 T}) in the abelian case $\g=\RR$ is the formula (\ref{l23 T scalar}) for the free boson. Note that in that case there is no $h^\vee$ shift.
\end{remark}

Fields $J^a$ are Virasoro-primary, of conformal weight $(1,0)$. Similarly,fields  $\ol{J}^a$ are Virasoro-primary, of conformal weight $(0,1)$.

\subsubsection{Example: WZW model for $G=SU(2)$ at level $k=1$ and the $r=\sqrt{2}$ free boson}
Consider the WZW model in the case $G=SU(2)$, $k=1$. 
In this case there are only two integrable modules of $\wh{\mathfrak{su}(2)}_1$ and space of states (\ref{l38 WZW space of states}) is 
\begin{equation}
\HH= H_{1,0}\otimes H_{1,0}^*\oplus H_{1,1}\otimes H_{1,1}^*.
\end{equation}
By (\ref{l38 c cbar}), the central charge of the model is 
\begin{equation}
c=\bar{c}=1.
\end{equation}
%One might suspect that the model is equivalent to the free boson CFT. 

It turns out that in this very special case the WZW model is equivalent to the free boson with values in a circle of radius $r=\sqrt{2}$ (the ``self-dual radius,'' w.r.t. T-duality).\footnote{
``Equivalence'' of CFTs means that there is an isomorphism of spaces of states as $\mr{Vir}\oplus \ol{\mr{Vir}}$-modules, and all correlators in the two CFTs agree via this isomorphism. In the case at hand the equivalence also implies a ``hidden'' $\wh{\mathfrak{su}(2)}_1$-symmetry of $r=\sqrt{2}$ free boson CFT.
} More specifically, the components of the WZW current corresponds to special fields in the $r=\sqrt{2}$ free boson theory:
\begin{equation}\label{l38 WZW-boson correspondence}
\begin{array}{c|c}
\wh{\mathfrak{su}(2)}_1\;\mr{WZW} & r=\sqrt{2}\;\mr{free \; boson} \\ \hline
J^3 & i\dd\phi \\
J^+ & V_{1,1} \\
J^- & V_{-1,-1} \\ \hline
\ol{J}{}^3 & i\bar\dd\phi \\
\ol{J}{}^+ & V_{1,-1} \\
\ol{J}{}^- & V_{-1,1}
\end{array}
\end{equation}
Here $V_{\e,\m}$ are the vertex operators (\ref{l28 vertex operator}). We are writing the components of the WZW current in terms of the basis $T^3, T^{\pm}=\frac{1}{\sqrt{2}}(T^1\pm i T^2)$, with $T^{1,2,3}$ as in (\ref{l36 T^a basis}). For instance, it is easy to check that the OPE algebra (\ref{l38 JJ OPE})--(\ref{l38 J Jbar OPE}) of the components of the WZW current is reproduced in the free boson theory by the fields in the right column of (\ref{l38 WZW-boson correspondence}).
%\marginpar{check normalization}

The $\wh\g$-primary multiplet for $\lambda=0$ corresponds on the $r=\sqrt{2}$ free boson side to the identity field $\mathbb{1}$. The $\lambda=1$ multiplet corresponds to the quadruple of vertex operators $V_{\pm 1,0},V_{0,\pm 1}$, all of conformal weight $(\frac14,\frac14)$.

\subsection{Ward identity for $\wh\g$-symmetry. Knizhnik-Zamolodchikov equations.}
%\ul{Ward identity for $\wh\g$-symmetry.}
As a consequence of Lemma \ref{l27 lemma: Ward identity for holom field Xi}, one has the Ward identity generated by the holomorphic field $J^a$ as in (\ref{l27 Ward identity for holom field Xi}): for  a collection of points $z_1,\ldots,z_n\in \CC$,  $\alpha$ a $\g$-valued meromorphic function with poles at $z_1,\ldots,z_n$ allowed , $\Phi_1,\ldots,\Phi_n\in V$ a collection of fields, one has
\begin{equation}\label{l38 Ward identity}
\alpha\circ \langle \Phi_1(z_1)\cdots \Phi_n(z_n) \rangle\colon= \sum_{j=1}^n \langle \Phi_1(z_1)\cdots \rho_J^{(z_j)}(\alpha)\circ\Phi_j(z_j) \cdots \Phi_n(z_n) \rangle
=0,
\end{equation}
where 
\begin{equation}
\rho_J^{(z)}(\alpha)\circ\Phi(z)\colon= \frac{1}{2\pi i} \oint_{\gamma_z} dw\,\alpha^a(w) J^a(w) \Phi(z).
\end{equation}
One has a similar Ward identity corresponding to the action on a correlator by an antimeromorphic function using the second current $\ol{J}$.

\begin{remark} %The Ward identity (\ref{l38 Ward identity}) is similar to the expected invariance  property of the WZW path  integral (\ref{l37 PI Ward identity}) -- invariance under the group of holomorphic maps $\Sigma\ra G_\CC$. 
One can think of the Ward identity (\ref{l38 Ward identity}) as corresponding to the expected invariance  property of the WZW path  integral (\ref{l37 PI Ward identity}) where:
\begin{itemize}
\item Boundary circles are shrunk to punctures $z_i$.
\item We consider the infinitesimal action of the Lie algebra of $\g$-valued functions, holomorphic in the complement of the punctures, instead of the group of holomorphic maps to the group $G_\CC$,
\end{itemize}
\end{remark}

Specializing (\ref{l38 Ward identity}) to the case $\alpha(w)=\frac{1}{w-z}$ and the collection of fields being the identity field at $z$ and $\wh\g$-primary fields at $z_1,\ldots,z_n$, we have the identity
\begin{equation}\label{l38 Ward identity with J}
\langle J^a(z) \phi_{\lambda_1}^{p_1\bar{p}_1}(z_1)\cdots \phi_{\lambda_n}^{p_n\bar{p}_n}(z_n) \rangle =\sum_{j=1}^n \sum_{q_j}\frac{(T^a_{\lambda_j})^{p_j}_{q_j}}{z-z_j} \langle \phi_{\lambda_1}^{p_1\bar{p}_1} (z_1)\cdots \phi_{\lambda_j}^{q_j\bar{p}_j}(z_j)\cdots
 \phi_{\lambda_n}^{p_n\bar{p}_n}(z_n)\rangle
\end{equation}
Here we have fixed some weights $\lambda_1,\ldots,\lambda_n$ of $\g$ corresponding to integrable $\wh\g$-modules.

One can also obtain this identity by realizing that due to (\ref{l38 J phi OPE}), the l.h.s. has to be a meromorphic function in $z$ with first-order poles at $z=z_1,\ldots,z_n$, with residues controlled by the r.h.s. of (\ref{l38 J phi OPE}). Such a function (decaying as $z\ra \infty$) is unique and given by the r.h.s. of (\ref{l38 Ward identity with J}).

One can also write the identity (\ref{l38 Ward identity with J}) in slightly more pleasing notations:
\begin{equation}\label{l38 Ward identity with J 2}
\langle J^a(z) \phi_{\lambda_1}(z_1)\cdots \phi_{\lambda_n}(z_n) \rangle =\sum_{j=1}^n \frac{T^a_{\lambda_j}}{z-z_j} \langle \phi_{\lambda_1}(z_1)\cdots
 \phi_{\lambda_n}(z_n)\rangle
\end{equation}
where 
\begin{itemize}
\item We denote 
\begin{equation}\label{l38 full primary multiplet}
\phi_\lambda\colon= \sum_{p,\bar{p}}\phi_\lambda^{p\bar{p}} e_p \otimes \bar{e}_{\bar{p}} \in V\otimes (M^\g_\lambda)^* \otimes M^\g_\lambda.
\end{equation} 
where $\{e_p\}$ is the basis in $(M^\g_\lambda)^*$ dual to the basis $\{e^p\}$ in $M^\g_\lambda$.
(\ref{l38 full primary multiplet}) is a vector-valued  field  -- the ``full'' $\wh\g$-primary multiplet with weight $\lambda$.
\item Both sides of (\ref{l38 Ward identity with J 2}) are valued in tensors
\begin{equation}\label{l38 tensor product}
\bigotimes_{i=1}^n (M^\g_{\lambda_i})^* \otimes M^\g_{\lambda_i}.
\end{equation}
\item %the rightmost subscript $j$ in $\rho_{\lambda_j}(T^a)_j$ means 
We understand
that the operator $T^a_{\lambda_j}$ is acting in the $j$-th factor in the product (\ref{l38 tensor product}).
\end{itemize}

\subsubsection{Knizhnik-Zamolodchikov equations.}

As a special case $n=-1$ of the Sugawara construction (\ref{l38 Sugawara}) one has
\begin{equation}
L_{-1}=\frac{1/2}{k+h^\vee} \sum_{m\in \ZZ} :J^a_m J^a_{-1-m}:
\end{equation}
where we think of both sides as operators acting on the space of fields $V_z$. In particular, for the $\wh\g$-primary multiplet $\phi_\lambda$, we have
\begin{equation}
L_{-1}\phi_{\lambda}(z)=\frac{1}{k+h^\vee} \sum_a J^a_{-1}J^a_0 \phi_\lambda(z)=
\frac{1}{k+h^\vee}J^a_{-1}T^a_\lambda \phi_\lambda (z).
\end{equation}
Using this, we have the following:
\begin{multline}
0=\langle \phi_{\lambda_1}(z_1)\cdots \underbrace{\Big(L_{-1}-\frac{1}{k+h^\vee} \sum_a J^a_{-1} T^a_{\lambda_j}\Big)\phi_{\lambda_j}(z_j)}_0 \cdots \phi_{\lambda_n}(z_n) \rangle=\\
=\frac{\dd}{\dd z_j} \langle \phi_{\lambda_1}(z_1)\cdots \phi_{\lambda_n}(z_n) \rangle - \sum_a \frac{1}{k+h^\vee} T^a_{\lambda_j} \frac{1}{2\pi i}\oint_{\gamma_{z_j}} \frac{dw}{w-z_j} \langle J^a(w) \phi_{\lambda_1}(z_1)\cdots \phi_{\lambda_n}(z_n) \rangle.
\end{multline}
Here we used that $L_{-1}\Phi(z)=\dd \Phi(z)$. Next we deform the integration contour $\gamma_{z_j}$ going around $z_j$ to a collection of contours 
%$\sqcup_{i\neq j}\gamma_{z_i}$ 
going around the punctures $z_i$ in negative direction,
\begin{equation}
\gamma_{z_j}\sim \sqcup_{i\neq j}(-\gamma_{z_i})
\end{equation}
Then, using the Ward identity (\ref{l38 Ward identity with J 2}), we obtain the following.
\begin{thm}[Knizhnik-Zamolodchikov \cite{KZ}]
Given the weights $\lambda_1,\ldots,\lambda_n$ of $\g$ corresponding to integrable $\wh\g$-modules, the correlator of primary multiplets satisfies the following the system of ODEs%\marginpar{check the sign}
\begin{equation}\label{l38 KZ}
\underbrace{\left( \frac{\dd}{\dd z_j}+ \frac{1}{k+h^\vee}\sum_{i\neq j}\sum_a \frac{T^a_{\lambda_i} T^a_{\lambda_j}}{z_i-z_j}\right)}_{\nabla_j^{\KZ}}
\langle \phi_{\lambda_1}(z_1)\cdots \phi_{\lambda_n}(z_n)  \rangle =0,
\end{equation}
for any $j=1,\ldots,n$.
\end{thm}

One can interpret the result as follows: one has a flat connection 
\begin{equation}\label{l38 nabla_KZ}
\nabla_{\KZ}\colon= 
%\sum_j dz_j\left( \frac{\dd}{\dd z_j}+ \frac{1}{k+h^\vee}\sum_{i\neq j}\sum_a \frac{T^a_{\lambda_i} T^a_{\lambda_j}}{z_i-z_j}\right) + \mr{c.c.}
\sum_j dz_j \nabla^{KZ}_j + d\bar{z}_j \ol\nabla^{KZ}_j
\end{equation}
on a vector bundle over the open configuration space 
%\marginpar{$\CC$ or $\CP^1$?} 
$C_n(\CP^1)$ with fiber (\ref{l38 tensor product});\footnote{
The structure of a vector bundle is determined by conformal weights of vectors in the fiber. I.e., a vector with conformal weight $(h,\bar{h})$ in (\ref{l38 tensor product}) contributes a summand $K^{\otimes h}\otimes \ol{K}^{\otimes\bar{h}}$ to the vector bundle.
}
 here $\nabla^{\KZ}_j$ are the differential operators appearing in the equation (\ref{l38 KZ}). The correlator of $\wh\g$-primary multiplets $\langle \phi_{\lambda_1}(z_1)\cdots \phi_{\lambda_n}(z_n)\rangle$
is a section of this bundle that is \emph{horizontal} w.r.t. $\nabla_{\KZ}$.

The flat connection (\ref{l38 nabla_KZ}) is known as the Knizhnik-Zamolodchikov (KZ) connection.

For future reference we will introduce a notation for the holomorphic part of the KZ connection
\begin{equation}\label{l38 nabla KZ hol}
\nabla_\KZ^\mr{hol}= \sum_j dz_j\nabla^\KZ_j + d\bar{z}_j \frac{\dd}{\dd \bar{z}_j}
\end{equation}
as a connection on the vector bundle on $C_n(\CP^1)$ with the fiber
$\bigotimes_{i=1}^n (M^\g_{\lambda_i})^*$ (i.e. taking only the first factor in each term in (\ref{l38 tensor product})).
%\marginpar{remark on KZB connection}

\subsection{Space of conformal blocks. Chiral WZW model.}\label{ss WZW conf blocks}

For $z_1,\ldots,z_n$ distinct points in $\CP^1$,
let us denote  by
$\g(z_1,\ldots,z_n)$ the Lie algebra of $\g$-valued meromorphic  functions on $\CP^1$ with poles allowed only at $z_1,\ldots,z_n$. 

Fix weights $\lambda_1,\ldots,\lambda_n$ of $\g$ corresponding to integrable modules of $\wh\g$ at level $k$. Then the Lie algebra $\g(z_1,\ldots,z_n)$ acts on the tensor product of integrable modules
\begin{equation}
H_{k,\lambda_1}\otimes \cdots \otimes H_{k,\lambda_n}
\end{equation}
by
\begin{equation}\label{l38 alpha action on psi}
\alpha\circ (\psi_1\otimes\cdots \otimes \psi_n)\colon=
\sum_{j=1}^n \psi_1\otimes\cdots \otimes\rho(\mr{Laurent}_{z_j}(\alpha))\circ \psi_j \otimes \cdots\otimes \psi_n
\end{equation}
where $\mr{Laurent}_{z_j}(\alpha)=\sum_{m=-N}^\infty \sum_a \alpha^a_m (T^a\otimes t_j^m)$ is the Laurent expansion of $\alpha$ at $z_j$, in powers of $t_j=z-z_j$; this Laurent expansion acts on $H_{k,\lambda_j}$ (this action is denoted by $\rho$ above) via the tautological embedding 
\begin{equation}
\g\otimes \CC[t_j^{-1},t_j]] \hookrightarrow \wh\g.
\end{equation}

\begin{definition}
For $\lambda_1,\ldots,\lambda_n$ a collection of weights of $\g$ corresponding to integrable modules of $\wh\g$ and a collection of distinct points $z_1,\ldots,z_n\in \CP^1$,
the \emph{space of Wess-Zumino-Witten conformal blocks} is  defined as the complex vector space
\begin{equation}\label{l38 conf blocks}
\mc{B}(z_1,\ldots,z_n; \lambda_1,\ldots,\lambda_n)\colon=
\mr{Hom}_{\g(z_1,\ldots,z_n)}(H_{k,\lambda_1}\otimes\cdots \otimes H_{k,\lambda_n},\CC)
\end{equation}
--  the space of $\g(z_1,\ldots,z_n)$-equivariant maps between two $\g(z_1,\ldots,z_n)$-modules, $H_{k,\lambda_1}\otimes\cdots \otimes H_{k,\lambda_n}$ with module structure (\ref{l38 alpha action on psi}) and $\CC$ as the trivial module.
\end{definition}

One can think of elements of (\ref{l38 conf blocks}) as correlators
\begin{equation}
\langle \psi_1(z_1)\cdots \psi_n(z_n) \rangle^\mr{chiral}
\end{equation}
in the \emph{chiral} WZW model, where the correlators are (possibly multivalued) holomorphic functions on the open configuration space $C_n(\CP^1)$ and only a single copy of $\wh\g$ (and a single copy of Virasoro) acts on the space of states/space of fields. Thus, in the chiral theory one has
\begin{equation}
V^\mr{chiral}\simeq \HH^\mr{chiral}= \bigoplus_\lambda H_{k,\lambda}.
\end{equation}
One can say that the chiral WZW is obtained from usual WZW by setting the antiholomorphic current to zero, $\ol{J}=0$ (and consequently $\ol{T}=0$).

The fact that in (\ref{l38 conf blocks}) the maps are required to be $\g(z_1,\ldots,z_n)$-equivariant is exactly the statement of Ward identity (\ref{l38 Ward identity}) for chiral correlators.

Somewhat surprisingly, the space of conformal blocks is finite-dimensional (with dimension depending on the level and the weights). In fact,
the inclusion 
\begin{equation}
\iota\colon M^\g_{\lambda_1}\otimes \cdots \otimes M^\g_{\lambda_n}  \hookrightarrow H_{k,\lambda_1}\otimes \cdots \otimes H_{k,\lambda_n}
\end{equation}
of depth-zero subspaces in each integrable module
% restricting to depth zero in each $H_{k,\lambda}$, one obtains 
 induces an \emph{injective} map
\begin{equation}\label{l38 B inclusion}
i=\iota^*\colon \mc{B}(z_1,\ldots,z_n;\lambda_1,\ldots, \lambda_n) \hookrightarrow \mr{Hom}_\g (M^\g_{\lambda_1}\otimes\cdots\otimes M^\g_{\lambda_n},\CC).
\end{equation}
This map corresponds to %restricting to depth-zero part in each integrable module, or equivalently to 
considering 
only correlators of $\wh\g$-primary chiral fields. The fact that the map $i$ is injective reflects the fact that using the Ward identity one can reduce a correlator of $\wh\g$-descendants to the correlator of $\wh\g$-primary fields (similarly to  Virasoro case, cf. Example \ref{l26 ex: corr of a descendant}). From (\ref{l38 B inclusion}) is is obvious that the space of conformal blocks must be finite-dimensional.

\begin{example} Consider the case $G=SU(2)$ and fix the level $k=1,2,3,\ldots$. The admissible weights corresponding to integrable modules are $\lambda=0,1,\ldots,k$. 

\begin{itemize}
\item For $n=3$, the space of conformal blocks can be either 0- or 1-dimensional:
\begin{itemize}
\item One has $\mc{B}(z_1,z_2,z_3;\lambda_1,\lambda_2,\lambda_3)=\CC$ if the ``fusion rules'' (or ``quantum Klebsch-Gordan condition'') hold:
\begin{equation}\label{l38 fusion rules}
\lambda_1+\lambda_2+\lambda_3 \in 2\ZZ,\quad |\lambda_1-\lambda_2|\leq \lambda_3\leq \lambda_1+\lambda_2,\quad
\lambda_1+\lambda_2+\lambda_3\leq 2k.
\end{equation}
\item Otherwise one has $\mc{B}(z_1,z_2,z_3;\lambda_1,\lambda_2,\lambda_3)=0$.
\end{itemize}
\item For a general $n$ one can associate a basis in the space of conformal blocks $\mc{B}(z_1,\ldots,z_n;\lambda_1,\ldots,\lambda_n)$  to any trivalent tree %$\mathbb{T}$
 with $n$ leaves decorated with $\lambda_1,\ldots,\lambda_n$. Basis vectors in $\mc{B}$ correspond to ways to decorate the internal edges $e$ of the tree %$\mathbb{T}$ by 
 by labels $\lambda_e \in \{0,1,\ldots,k\}$ so that fusion rules (\ref{l38 fusion rules}) hold at each vertex. The idea behind constructing such a basis is similar to that of Sections \ref{sss 4-pt corr of primary fields}, \ref{ss n-point correlator of primary fields} and comes from a pair-of-pants decomposition of the surface; edges of the graph correspond to circles we cut along and their decorations correspond to intermediate states we sum over. %\marginpar{write better}
 
 In the case when $\CP^1$ is replaced by a Riemannian surface $\Sigma$ of genus $h$, instead of a trivalent tree one should consider decorations of a connected trivalent graph with $h$ 
loops.
\item One has a fascinating explicit formula due to Verlinde \cite{Verlinde} for the dimension of the space of $n$-point conformal blocks for $G=SU(2)$, on %$\CP^1$
a surface $\Sigma$ of genus $h$:%\marginpar{write Verlinde f-la for general genus}
\begin{equation}\label{l38 Verlinde}
\dim \mc{B}(z_1,\ldots,z_n;\lambda_1,\ldots,\lambda_n)=
\sum_{0\leq \lambda \leq k}
%\frac{S_{\lambda_1\lambda}\cdots S_{\lambda_n\lambda}}{(S_{0\lambda})^{n-2}},
(S_{0\lambda})^{2-2h-n} S_{\lambda_1\lambda}\cdots S_{\lambda_n\lambda},
\end{equation}
where 
\begin{equation}\label{l38 S}
S_{\lambda\mu}=\sqrt{\frac{2}{k+2}}\; \sin \pi \frac{(\lambda+1)(\mu+1)}{k+2}.
\end{equation}
The result comes from a ``diagonalization'' of the dimension of the space of 3-point conformal blocks:
\begin{multline}
\dim \mc{B}(z_1,z_2,z_3;\lambda_1,\lambda_2,\lambda_3)=\\
=
\sum_{0\leq \lambda \leq k} \frac{S_{\lambda_1 \lambda} S_{\lambda_2 \lambda} S_{\lambda_3 \lambda}}{S_{0\lambda}}=
\left\{ 
\begin{array}{cl}
1& \mbox{if the fusion rules (\ref{l38 fusion rules}) hold},\\
0& \mbox{otherwise}
\end{array}
\right.
\end{multline}
%\marginpar{Say that $S$ represents the modular $S$ transformation}
The matrix $S$ (\ref{l38 S}) appearing here can be interpreted as representing the action of the modular $S$-transformation $\tau\ra -\frac{1}{\tau}$ on the space of conformal blocks with genus one and no punctures.\footnote{This space is $(k+1)$-dimensional, with a natural basis given by characters of modules $H_{k,\lambda}$ with 
 $0\leq \lambda \leq k$, cf. Section \ref{sss Virasoro characters}.}
\end{itemize}
\end{example}

\subsubsection{The bundle of conformal blocks.} Spaces of conformal blocks (\ref{l38 conf blocks}) with fixed weights $\lambda_1,\ldots,\lambda_n$ and variable  points $z_1,\ldots,z_n$ arrange into a complex vector bundle over the open configuration space of $n$ points,
\begin{equation}\label{l38 bundle of conf blocks}
\begin{CD}
\mc{E}_{\lambda_1\cdots \lambda_n} @<<< \mc{B}(z_1,\ldots,z_n;\lambda_1,\ldots,\lambda_n) \\@VVV \\
C_n(\CP^1)
\end{CD}
\end{equation}
This vector bundle comes equipped with a %holomorphic 
flat connection
\begin{equation}\label{l38 nabla_E}
\nabla_\mc{E}= \sum_{j=1}^ndz_j \left(\frac{\dd}{\dd z_j}-L_{-1}^{(j)}\right)+d\bar{z}_j\frac{\dd}{\dd \bar{z}_j},
\end{equation}
where $L_{-1}^{(j)}$ is  (the dual of) the Sugawara operator acting on $H_{k,\lambda_j}$. Correlators of chiral WZW model yield a horizontal multivalued section %$\Psi$ 
of $\mc{E}$. %(i.e. they are holomorphic -- annihilated by $\frac{\dd}{\dd \bar{z}_j}$ and annihilated by $\nabla_\mc{E}$). 
%\marginpar{A bit confused about duals here. Clarify.}
Restricted to depth zero in each integrable module (i.e. restricted to chiral correlators of $\wh\g$-primary fields), the
holomorphic part of the connection $\nabla_\mc{E}$ becomes  
the holomorphic part of the Knizhnik-Zamolodchikov connection $\nabla_\KZ^\mr{hol}$ (\ref{l38 nabla KZ hol}). %(\ref{l38 nabla_KZ}). 
%(or rather the holomorphic part thereof).

%Accordingly, the restriction of $\Psi$ to depth zero in each integrable module, we obtain a section $\Psi_0$ of the bundle with fiber $\mr{Hom}(M^\g_{\lambda_1}\otimes\cdots \otimes M^\g_{\lambda_n},\CC)$ horizontal with respect to Knizhnik-Zamolodchikov connection.

\marginpar{Lecture 39,\\ 12/02/2022}

\subsection{The ``holographic'' correspondence between 3d Chern-Simons and 2d Wess-Zumino-Witten theories}
Here we quickly mention the remarkable relation between a 3d topological field theory (Chern-Simons theory) on a 3-manifold $M$ and a 2d CFT (Wess-Zumino-Witten model) on the boundary surface $\Sigma=\dd M$. There is a lot of literature on the subject, starting with the seminal work of Witten \cite{Witten89}. 
%\marginpar{edit}
The correspondence between WZW and Chern-Simons is an example in the class of so-called ``holographic correspondences'' between $(d+1)$-dimensional gravity and a $d$-dimensional conformal theory on the boundary.

Fix $G$ a compact, simple, simply connected Lie group with Lie algebra $\g$ and fix $M$ an oriented compact 3-manifold with the boundary surface $\Sigma$ (possibly disconnected); we assume that $\Sigma$ is equipped with complex structure.

Consider Chern-Simons theory on $M$ with space of fields $\F_M^\CS=\Omega^1(M,\g)=\mr{Conn}(M,G)$ -- the space of connections in the trivial principal $G$-bundle over $M$; we identify connections with their $1$-forms on the base. The action functional is
\begin{equation}
S_{\CS}^b(A)\colon=\frac{1}{2\pi}\int_M \tr \left(\frac12 A\wedge dA +\frac16 A\wedge [A\stackrel{\wedge}{,}A] \right) + \underbrace{\frac{1}{4\pi}\int_\Sigma \tr\, A^{1,0}\wedge A^{0,1}}_{b(A|_\Sigma)}.
\end{equation}
The last term here is a boundary term, depending only on the restriction of $A$ to $\Sigma$ (and the decomposition of that restriction into a $(1,0)$-form and a $(0,1)$-form using the complex structure). The superscript $b$ in the action is to emphasize the presence of the boundary term $b$. The term $b$ is designed to tweak the Noether 1-form induced by the action to a convenient form: for the variation of the action one has
\begin{equation}
\delta S^b_{\CS} = -\frac{1}{2\pi} \int_M \tr\, \delta A\wedge F_A + \underbrace{\frac{1}{2\pi}\int_\Sigma \tr\, A^{0,1}\wedge \delta A^{1,0}}_{\alpha}.
\end{equation}
Here the last term is the Noether 1-form on the phase space $\Phi^{\CS}_\Sigma=\Omega^1(\Sigma,\g)$ and the fact that it vanishes on the (Lagrangian) fibers of the fibration 
\begin{equation}
p\colon \Omega^1(\Sigma,\g_\CC)\ra \Omega^{1,0}(\Sigma,\g)
\end{equation}
implies that one flat connections $A$ are actual critical points of $S_{\CS}^b$ on the subspace of fields with prescribed boundary condition $(A|_\Sigma)^{1,0}$. In particular, one can study the path integral for Chern-Simons theory
\begin{equation}\label{l39 CS PI}
Z^{\CS}(A^{1,0})\colon= \int_{\mr{Conn}(M,G)\ni A\;\mr{s.t.}\; (A|_\Sigma)^{1,0}=A^{1,0}} \mc{D}A\, e^{i k S_{\CS}^b(A)}
\end{equation}
with $k=1,2,3,\ldots$ the ``level'' of Chern-Simons theory.  

Consider gauge transformations of the connection
\begin{equation}
A \mapsto A^g = g^{-1}Ag + g^{-1} dg
\end{equation}
with generator $g\colon M\ra G$. If the generator is trivial on the boundary $g|_\Sigma=1$, 
%(i.e. we consider gauge transformations on $M$ \emph{relative to the boundary}), 
one has
\begin{equation}
S^b_{\CS}(A^g)= S^b_{\CS}(A)\quad \mr{mod}\; 2\pi\ZZ,
\end{equation}
i.e., Chern-Simons action is invariant modulo $2\pi\ZZ$ under gauge transformations \emph{relative to the boundary}. The $2\pi\ZZ$-ambiguity is the reason why one wants the normalization factor $k$ in the exponential in the path integral (\ref{l39 CS PI}) to be an integer -- so that the integrand in the path integral is gauge-invariant.

\subsubsection{Classical CS-WZW correspondence.}
If the generator of the gauge transformation is nontrivial on the boundary, one has
\begin{equation}\label{l39 S_CS gauge defect}
S^b_{\CS}(A^g)-S^b_{\CS}(A)= i S_{\WZW}(g|_\Sigma)+\frac{1}{2\pi} \int_\Sigma \tr\, A^{1,0} g^{-1}\dd g.
\end{equation}
The first term on the r.h.s. is the Wess-Zumino-Witten action evaluated on the boundary restriction of the generator of the gauge transformation. Thus, the defect of gauge-invariance of Chern-Simons theory due to the presence of boundary is given by WZW action on the boundary. 

The full r.h.s. of (\ref{l39 S_CS gauge defect}) is sometimes called the \emph{gauged} WZW model. It can also be thought of as the action of the \emph{chiral} WZW model: we can regard the field $A^{1,0}$ as a Lagrange multiplier, integrating it out imposes the vanishing of the antiholomorphic WZW current $\ol{J}=0$.

Formula (\ref{l39 S_CS gauge defect}) is a manifestation of the Chern-Simons/Wess-Zumino-Witten correspondence at the classical level. A consequence of it is the following: if $M$ is a 3-ball, with $\Sigma=\dd M =\CP^1$, any flat connection on $M$ can be written as $A=g^{-1}dg$ (gauge-equivalent to zero connection) for some $g\colon M\ra G$. In this case (\ref{l39 S_CS gauge defect}) implies
\begin{equation}
S^b_{\CS}(A)=iS_{\WZW}(g).
\end{equation}

\subsubsection{Quantum CS-WZW correspondence.}
The relation (\ref{l39 S_CS gauge defect}) has a very nontrivial quantum counterpart:
\begin{equation}\label{l39 B_WZW=H_CS}
\mc{B}^{\WZW}_\Sigma = \HH_\Sigma^{\CS}
\end{equation}
-- the space of states that quantum Chern-Simons theory (as an Atiyah's TQFT) assigns to a surface $\Sigma$ is isomorphic to the space of WZW conformal blocks on the surface.

One  has a version of this statement with punctures on $\Sigma$. For that one should consider a Wilson graph observable $O_\Gamma$ in the Chern-Simons theory on $M$. Let $\Gamma\subset M$ be an embedded oriented graph in $M$, which is allowed to meet the boundary surface transversally; we treat these boundary points of $\Gamma$ as univalent vertices. Bulk vertices are assumed to be trivalent. Assume that the edges of $\Gamma$ are decorated by weights $\lambda$ of integrable representations of $\g$\footnote{
We understand that one can switch the orientation of any edge, switching simultaneously the representation $M^\g_\lambda$ to its dual.
} and the trivalent vertices are decorated by intertwiners -- elements of $\left(M^\g_\lambda\otimes M^\g_{\lambda'}\otimes M^\g_{\lambda''}\right)^\g$, where $\lambda,\lambda',\lambda''$ are the weights decorating the incident edges. At the level of classical field theory, the observable 
\begin{equation}\label{l39 Wilson graph}
O_\Gamma\colon \F_M\ra (M^\g_{\lambda_1})^*\otimes \cdots \otimes (M^\g_{\lambda_n})^*
\end{equation} 
is a function on the space of connections $A$, given by the contraction of holonomies of $A$ along the edges of $\Gamma$, taken in corresponding representations, with the intertwiners at the vertices. In (\ref{l39 Wilson graph}) we are assuming that $\Gamma$ has $n$ boundary vertices at points $z_1,\ldots,z_n\in \Sigma$ and the incident edges are decorated by weights $\lambda_1,\ldots,\lambda_n$.
%\textcolor{red}{PICTURE}
\begin{figure}[H]
\begin{center}
\includegraphics[scale=0.75]{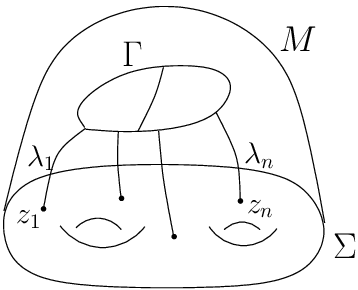}
\end{center}
\caption{Wilson graph observable.
}
\end{figure}

The path integral of Chern-Simons theory with the Wilson graph observable is
\begin{multline}
Z_{M,\Gamma}^{\CS}=\int_{\mr{Conn}(G,\Sigma)\ni A\;\mr{s.t.}\; (A|_\Sigma)^{1,0}=A^{1,0}} \mc{D} A\; e^{ik S^b_{CS}(A)} O_\Gamma(A)
\quad \in 
\\ \in \underbrace{\left(C^\infty(\Omega^{1,0}(\Sigma,\g))\otimes \mr{Hom}(M^\g_{\lambda_1}\otimes \cdots \otimes M^\g_{\lambda_n},\CC)\right)^{\mr{Map}(\Sigma,G)}}_{\HH^{\CS,\Gamma}_{\Sigma
%,\dd\Gamma
}}
\end{multline}
where understand that the path integral is a function of the boundary $(1,0)$-form $A^{1,0}$ (the boundary condition of the path integral) and also takes values in a product of representations due to the presence of $O_\Gamma$-observable. The whole expression is expected to be equivariant under gauge transformations, where only the boundary value of the gauge generator matters after averaging over the fields in the bulk, since the integrand is equivariant (and invariant under gauge transformations relative to the boundary). 
%\marginpar{write the change of $Z$ under a gauge transf.}
The expected equivariance property following from (\ref{l39 S_CS gauge defect}) is:
%\marginpar{double check conventions}
\begin{equation}
Z^\CS_{M,\Gamma}((A^{1,0})^g)=e^{-k S_{\WZW}(g)+\frac{i}{2\pi}\int_\Sigma A^{1,0} g^{-1}\dd g} \bigotimes_{j=1}^n \rho^*_{\lambda_j}(g(z_j))\circ Z^\CS_{M,\Gamma}(A^{1,0}),
\end{equation}
where $A^{1,0}\in \Omega^{1,0}(\Sigma,\g)$ is the boundary condition and $g\colon \Sigma \ra G$ is the gauge transformation on the boundary; $(A^{1,0})^g=g^{-1}A^{1,0}g+g^{-1}\dd g$ is the chiral gauge transformation on the boundary; $\rho_{\lambda}^*(g)$ is the operator representing the group element on the module $(M^\g_\lambda)^*$.

The vector space where the path integral takes values is the space of states assigned to the boundary $\Sigma$ by Chern-Simons theory deformed by the observable $\Gamma$. It depends on the positions $z_1,\ldots,z_n\in \Sigma$ of boundary vertices of $\Gamma$ and the corresponding weights $\lambda_1,\ldots,\lambda_n$.
The 
statement of CS-WZW correspondence generalizing (\ref{l39 B_WZW=H_CS}) in this setting is:
\begin{equation}\label{l39 holography with punctures}
\mc{B}_\Sigma^{\WZW}(z_1,\ldots,z_n;\lambda_1,\ldots,\lambda_n)=\HH_\Sigma^{\CS,\Gamma}
\end{equation}
-- the Chern-Simons space of states on $\Sigma$ deformed by $\Gamma$ is the space of WZW $n$-point conformal blocks on $\Sigma$. We refer the reader to \cite{GK} for details on the correspondence (\ref{l39 holography with punctures}).

\begin{remark}
The space of states $\HH_\Sigma^\CS$ of Chern-Simons theory can also be obtained as as a geometric quantization of the moduli space of flat connections on $\Sigma$ (as a symplectic manifold with singularities, with polarization inferred from complex structure), see \cite{ADPW}. The choice of complex structure serves as a parameter of quantization, and hence one obtains a vector bundle of spaces of states over the moduli space of complex structures
\begin{equation}
\begin{CD}
\mathbb{H} @<<< \mr{GeomQ}(\MM_\mr{flat}(\Sigma),\mr{cx.\, str.\,on\,}\Sigma)\\ @VVV\\ \MM_\Sigma
\end{CD}
\end{equation}
This vector bundle comes with a natural projectively flat connection -- the so-called Hitchin connection --  allowing one to compare quantizations with different choices of complex structure. This bundle in the case of $\Sigma=\CP^1$ with punctures (and up to reduction by the M\"obius group) is the bundle of conformal blocks (\ref{l38 bundle of conf blocks}) with Knizhnik-Zamolodchikov connection.
%its connection $\nabla_\mc{E}$.
\end{remark}

\begin{remark}\label{l39 rem: Z^CS_cylinder from Verlinde} The correspondence (\ref{l39 holography with punctures}) allows one to use things known in WZW to make statements about Chern-Simons theory. For instance, from Atiyah's axioms one has that for a closed 3-manifold of the form $M=\Sigma\times S^1$ with $\Sigma$ a closed surface of genus $h$, the Chern-Simons partition function is the dimension of the space of states on $\Sigma$:
\begin{multline}
Z^\CS_{\Sigma\times S^1} \underset{\mr{Atiyah}}{=} \dim \HH^\CS_\Sigma \underset{\mr{holography}\;(\ref{l39 B_WZW=H_CS})}{=}  \dim \mc{B}_\Sigma^\WZW  =\\
\underset{\mr{Verlinde}\;(\ref{l38 Verlinde})}{=} \left(\frac{k+2}{2}\right)^{h-1}\sum_{\lambda=0}^k \left(\sin\pi \frac{\lambda+1}{k+2}\right)^{2-2h}
\end{multline}
Here we assumed $G=SU(2)$; $k$ is the level.

Likewise, consider the Chern-Simons partition function for the 3-manifold $\Sigma\times S^1$ with observable $\Gamma$ consisting of $n$  circles of the form $\{z_i\}\times S^1$, with $z_1,\ldots,z_n$ an $n$-tuple of distinct points on $\Sigma$, assuming that the circles are decorated with weights $\lambda_1,\ldots,\lambda_n$. By the same logic, this partition function is again given by the Verlinde formula,
\begin{equation}
Z^\CS_{\Sigma\times S^1,\Gamma} = \dim \mc{B}_\Sigma^\WZW(z_1,\ldots,z_n;\lambda_1,\ldots,\lambda_n)=\mbox{r.h.s. of (\ref{l38 Verlinde})}.
\end{equation}
\end{remark}

\subsection{Parallel transport of the KZ connection, $R$-matrix  and representation of the braid group}
%\marginpar{Straighten the dualization conventions for KZ bundle}

Consider the Knizhnik-Zamolodchikov connection $\nabla_{\KZ}^\mr{hol}$ on the depth-zero part of the bundle of $n$-conformal blocks where all weights are the same $\lambda_1=\cdots=\lambda_n=\lambda.$
\begin{equation}
\mc{E}^0_{\lambda\cdots \lambda} \ra C_n(\CC).
\end{equation}
We also restricted the base from $\CP^1$ to $\CC$ for the sake of present discussion. Recall that the base $C_n$ here is the space of \emph{ordered} configurations of points. However, since we chose all weights to be the same, we can quotient the bundle by the symmetric group permuting the $n$ points, obtaining a vector bundle
\begin{equation}
\mc{E}^{'0}_{\lambda\cdots \lambda} \ra C_n^{\mr{unordered}}(\CC)
\end{equation}
over the unordered configuration space. The connection $\nabla_{\KZ}^\mr{hol}$ descends to this quotient.

Consider %the point $P=(z_1=1,z_2=2,\ldots,z_n=n)\in C_n(\CC)$ and
 a path 
\begin{equation}
\gamma_j(t)=\left(1,\ldots,j-1,j+\frac{1-e^{it}}{2} ,j+\frac{1+e^{it}}{2},j+2,\ldots, n\right),\quad t\in [0,\pi]
\end{equation} 
in $C_n(\CC)$ for $j\in \{1,\ldots,n-1\}$ -- it interchanges the points $z_j$ and $z_{j+1}$ by a smooth move,
i.e., it starts at $P=(1,\ldots,j,j+1,\ldots,n)$ and finishes at $Q=(1,\ldots,j+1,j,\ldots,n)$. 
%\textcolor{red}{PICTURE: the path}
\begin{figure}[H]
\begin{center}
\includegraphics[scale=0.75]{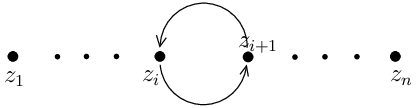}
\end{center}
\caption{Path in the configuration space.
}
\end{figure}
This path descends to a closed loop $\gamma'_j$ in $C_n^\mr{unordered}(\CC)$ starting and ending at the point $\{1,\ldots,n\}$. The parallel transport  of $\nabla_{\KZ}^\mr{hol}$ along this loop is an endomorphism of the fiber of $\mc{E}^{'0}_{\lambda\cdots\lambda}$ of the form
\begin{equation}\label{l39 R}
\underbrace{\mr{id}\otimes\cdots \otimes \mr{id}}_{j-1}\otimes R \otimes \underbrace{\mr{id}\otimes\cdots \otimes \mr{id}}_{n-j-1}\quad \in \mr{End}(((M^\g_\lambda)^*)^{\otimes n})
\end{equation}
with $R$ a certain element
\begin{equation}
R \in \mr{End}\left(((M^\g_\lambda)^*)^{\otimes 2}\right)
\end{equation}
-- it is an example of the so-called ``$R$-matrix.'' (This particular one is the $R$ matrix given by the holonomy of Knizhnik-Zamolodchikov connection.) It satisfies the Yang-Baxter equation
\begin{equation}\label{l39 YB}
(R\otimes \mr{id})(\mr{id}\otimes R)(R\otimes \mr{id}) = (\mr{id}\otimes R)(R\otimes \mr{id})(\mr{id}\otimes R)
\end{equation}
by construction -- because both sides give the parallel transport along loops in $C_n^\mr{unordered}(\CC)$ and the two sides correspond to two \emph{homotopic} loops (recall that $\nabla_{\KZ}^\mr{hol}$ is a \emph{flat} connection, so the parallel transport does not change under homotopy of the loop).

The fundamental group of $C^\mr{unordered}(\CC)$ is also known as the braid group on $n$ strands. Its standard presentation is with $n-1$ generators $c_1,\ldots, c_{n-1}$ subject to relations
\begin{equation}\label{l39 Br relations}
c_j c_{j+1} c_j= c_{j+1} c_j c_{j+1},\qquad c_i c_j =c_j c_i \;\;\mr{if}\; |i-j|\geq 2.
\end{equation}
%\textcolor{red}{PICTURE: Br relations}
\begin{figure}[H]
\begin{center}
\includegraphics[scale=0.75]{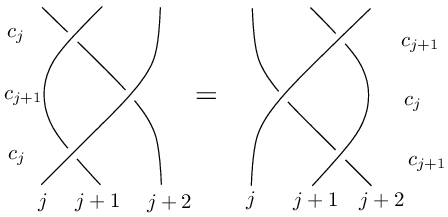}
\end{center}
\caption{Relation in the braid group. One can understand this picture as being in $\RR\times\CC$. The l.h.s. is the graph of the concatenation of paths $\gamma_j*\gamma_{j+1}*\gamma_j$, %concatenated with the path $\gamma_{j+1}$, concatenated with $\gamma_j$, 
and similarly for the r.h.s.
}
\end{figure}

The construction above gives a representation of the braid group on the space 
\begin{equation}\label{l39 Br representation space}
((M^\g_\lambda)^*)^{\otimes n},
\end{equation} with the generator $c_j$ represented by the element (\ref{l39 R}) -- by the $R$-matrix acting in the $j$-th and $(j+1)$-st factors of the representation space (\ref{l39 Br representation space}). The first relation in (\ref{l39 Br relations}) holds due to the Yang-Baxter equation (\ref{l39 YB}) and the second relation is obvious by construction (\ref{l39 R}).

\begin{remark} Let $\gamma$ be a loop in $C_n^\mr{unordered}(\CP^1)$ or equivalently a braid. Gluing the top and the bottom of the braid, we obtain a link $L$ in the 3-manifold $M=\CP^1\times S^1$. Let $\Xi\in \mr{End}(\mc{B}(1,2,\ldots,n;\lambda,\ldots,\lambda))$ be the parallel transport of the connection (\ref{l38 nabla_E}) along $\gamma$. Then by the argument analogous to Remark \ref{l39 rem: Z^CS_cylinder from Verlinde} one has
\begin{equation}\label{l39 Z_CS with a link}
Z^{\CS}_{\CP^1\times S^1,L}= \tr_{\mc{B}(1,2,\ldots,n;\lambda,\ldots,\lambda)} \Xi.
\end{equation}
Here we think of $L$ as a special type of Wilson graph (a disjoint union of circles), with all components of the link decorated by the weight $\lambda$. Given a presentation of $\gamma$ seen as a braid in terms of generators $c_j$,
$
\gamma=c_{j_1}\cdots c_{j_r},
$
 the endomorphism $\Xi$ can be written as a product of $R$-matrices, 
 \begin{equation}\label{l39 Xi}
 \Xi=R_{j_1}\cdots R_{j_r},
 \end{equation} where the subscript $j$ means that the $R$-matrix acts on the $j$-th and $(j+1)$-st factors. Here a remark is that although the r.h.s. of (\ref{l39 Xi}) is an endomorphism of (\ref{l39 Br representation space}), it actually stabilizes the image of the inclusion (\ref{l38 B inclusion}) and hence determines an endomorphism of the space of conformal blocks.
 %$\wh{R}$ -- analogs of the $R$-matrix above but where we consider the parallel transport in the full bundle of conformal blocks $\mc{E}$ instead of the depth-zero truncation $\mc{E}^0$.
 
 We refer to the seminal paper \cite{FK} for details on the invariant of knots and links arising from the construction (\ref{l39 Z_CS with a link}).
\end{remark}

\chapter{A-model}

The A-model introduced by Witten in \cite{Witten88} is an example of a 2d topological conformal field theory which contains a special class $Q$-closed observables (so-called evaluation observables) whose correlators yield closed forms on the moduli space $\MM_{g,n}$. Integrated over $\MM_{g,n}$, these correlators  yield interesting %integer 
numbers -- Gromov-Witten invariants -- solutions of a certain class of enumerative geometric problems. Moreover, field-theoretic origin of these numbers (ultimately, Segal's axioms) result in an equation on Gromov-Witten invariants  -- the Witten-Dijkgraaf-Verlinde-Verlinde or WDVV equation -- which allows in some cases to fully compute the Gromov-Witten invariants, see \cite{KM}.

\section{Closed forms on the moduli space from TCFT correlators.}
\subsection{Genus zero case}
First, recall from Section \ref{l33 rem: closed forms from a TCFT} that in any TCFT, given a collection of $Q$-cocycles $\Phi_1,\ldots,\Phi_n$, the correlator of their total descendants on $\CP^1$,
\begin{equation}\label{l39 corr of descent towers}
\langle \til\Phi_1(z_1)\cdots  \til\Phi_n(z_n)\rangle
\end{equation} 
yields a closed form  (under de Rham differential) on the moduli space $\MM_{0,n}$, which can subsequently be integrated over relevant cycles to yield interesting periods. 

\subsection{Higher genus}

For a surface $\Sigma$ of general genus $g$, given $Q$-closed fields $\Phi_1,\ldots,\Phi_n\in V$,  the construction (\ref{l39 corr of descent towers}) yields a closed form on the configuration space, i.e., on the fiber of the bundle
\begin{equation}
\begin{CD}
\MM_{\Sigma,n} @<<< C_n(\Sigma) \\
@VVV \\
\MM_\Sigma
\end{CD}
\end{equation}
but not on the total space.

As a generalization of the construction (\ref{l39 corr of descent towers}), for fixed $Q$-closed fields $\Phi_1,\ldots,\Phi_n\in V$
and any $p\geq 0$ one can consider the correlator\footnote{
Note that the $(1,1)$-form $d^2x \mu(x)\, G(x)$ appearing in (\ref{l39 alpha_p}) is a coordinate-independent object -- it is the contraction of the Beltrami differential $d\bar{x}\, \mu(x) \frac{\dd}{\dd x}$ and the quadratic differential $G(x) (dx)^2$.
}  
%\marginpar{normalizations?}
\begin{equation}\label{l39 alpha_p}
\alpha_p(\mu_1+\bar\mu_1,\ldots, \mu_p+\bar\mu_p) \colon = 
\left\langle
\prod_{i=1}^p \left(
\int_\Sigma d^2 x_i \big(\mu_i(x_i) G(x_i)+\bar\mu_i(x_i)\ol{G}(x_i)\big)
\right)\; \Phi_1(z_1)\cdots \Phi_n(z_n)
 \right\rangle.
\end{equation}
Here $\{\mu_i+\bar\mu_i\}_{i=1,\ldots,p}$ are a $p$-tuple of Beltrami differentials on $\Sigma$, i.e., tangent vectors to the space of complex structures on $\Sigma$ (see Section \ref{sss Beltrami differentials}). We assume that supports of $\mu_i+\bar\mu_i$ for different $i$ are disjoint and also disjoint from points $z_k$. 

Note that if one changes $\mu_i\mapsto \mu_i+\bar\dd v^{1,0}$ 
for $v^{1,0}$ a section of $T^{1,0}\Sigma$ vanishing at points $z_k$ (cf. (\ref{l10 mu -> mu + dbar v})), the expression (\ref{l39 alpha_p}) does not change (as follows from integration by parts and the property $\bar\dd G=0$). Similarly, $\alpha_p$
 does not change under the transformation $\bar\mu_i\mapsto \bar\mu_i+\dd v^{0,1}$. This shows that $\alpha_p$ descends to a differential form on the moduli space of complex structures with marked points $z_1,\ldots,z_n$:
\begin{equation}
 \alpha_p\in \Omega^p(\MM_{g,n}).
\end{equation}
 
\begin{lemma}\label{l39 lemma alpha_p closed}
The form $\alpha_p$ is closed. If at least one of $\Phi_k$ is $Q$-exact, then $\alpha_p$ is an exact form.
\end{lemma} 
\begin{proof}
We have
\begin{multline}\label{l39 lemma computation}
d\alpha_p(\mu_0+\bar\mu_0,\ldots, \mu_p+\bar\mu_p)=
\sum_{r=0}^p (-1)^r\mc{L}_{\mu_r+\bar{\mu}_r} \alpha_p(\mu_0+\bar\mu_0,\ldots, \wh{\mu_r+\bar\mu_r} ,\ldots, \mu_p+\bar\mu_p)=\\
=\sum_{r=0}^p (-1)^r \Big\langle 
\left(\int_\Sigma d^2x_0 (\mu_0(x_0) G(x_0)+\bar\mu_0(x_0)\ol{G}(x_0) )  \right) 
\cdots
\left(\frac{1}{2\pi}\int_\Sigma d^2x_r (\mu_r(x_r) T(x_r)+\bar\mu_r(x_r)\ol{T}(x_r) )\right)\\ 
\cdots
\left(\int_\Sigma d^2x_p (\mu_p(x_p) G(x_p)+\bar\mu_p(x_p)\ol{G}(x_p) )\right)
\Phi_1(z_1)\cdots  \Phi_n(z_n)
 \Big\rangle\\
 =\int_{\Sigma^p} d^2x_1\cdots d^2x_p \Big\langle 
i \int_\gamma \mathbb{J}(w) \prod_{i=0}^p  
\left( \mu_i(x_i) G(x_i)+\bar\mu_i(x_i)\ol{G}(x_i) \right) \Phi_1(z_1)\cdots  \Phi_n(z_n)
 \Big\rangle =0. 
\end{multline}
Here $\gamma$ -- the  integration contour for the BRST current $\mathbb{J}$ -- is a union of circles going around points $x_i$. We deform it to a homologous contour going around points $z_i$. Since $Q\Phi_k=0$, the latter contributions vanish. The stress-energy appears in the correlator by the mechanism (\ref{l23 delta_g <...>}). Another remark is in order here: we understand the Beltrami differentials $\mu_i+\bar\mu_i$ as a collection of vector fields on the moduli space $\MM_{\Sigma,n}$ defined in the neighborhood of the reference complex structure (we can do it, since Beltrami differentials define finite deformations of complex structures); these vector fields commute, due to the disjoint support condition. This is why there is no second term in the computation of de Rham differential in the computation (\ref{l39 lemma computation}).\footnote{Recall that the general formula for the de Rham differential of a $p$-form is
$d\omega(X_1,\ldots,X_p)=\sum_{r=0}^p (-1)^r\LL_{X_r}\omega(X_0,\ldots,\wh{X_r},\ldots,X_p)+\sum_{0\leq r<s \leq p} (-1)^{r+s}\omega([X_r,X_s],X_0,\ldots, \wh{X_r},\ldots,\wh{X_s},\ldots, X_p)$.}

If $\Phi_k=Q(\Psi)$ for some $k$ and some field $\Psi$, one has a similar contour-deformation argument transforming the integration contour for $\mathbb{J}$ around $\Psi(z_k)$ to a contour around points $x_i$ where the action of $\mathbb{J}$ transforms the field $G$ into the stress-energy tensor $T$, and the whole expression becomes $d \alpha_{p-1}(\Phi_1,\ldots,\Psi,\ldots,\Phi_n)$. Here the arguments indicate the fields to which the construction (\ref{l39 alpha_p}) is applied. 
\end{proof}
 
\begin{remark} In a chirally split TCFT (cf. Remark \ref{l33 rem: chirally split TCFTs}) one has a refined version of the construction above: assuming that fields $\Phi_1,\ldots,\Phi_n\in V$ are both $Q_L$- and $Q_R$-closed, one can construct a family of closed $(p,q)$-forms on the moduli space,
\begin{equation}
\alpha_{p,q}\colon=
\Big\langle \prod_{i=1}^p \left(\int_\Sigma d^2x_i  \mu_i(x_i)G(x_i)\right)
\prod_{j=1}^q \left(\int_\Sigma d^2y_j  \bar\mu_j(y_j)\ol{G}(y_j)\right)
\Phi_1(z_1)\cdots \Phi_n(z_n)
 \Big\rangle\;\; \in \Omega^{p,q}_\mr{cl}(\MM_{g,n}).
\end{equation}
\end{remark}

\begin{remark} Construction (\ref{l39 corr of descent towers}) (correlators of descent towers), specialized to the case of canonical descent (\ref{l33 canonical total descendant}), can be understood as a special case of construction (\ref{l39 alpha_p}) (correlators with fields $G,\ol{G}$ adjoined), using the following observation. For $v\in T_z\Sigma$ a tangent vector and $\Phi(z)\in V_z$ a field, one can express the descent operator $\Gamma$ (\ref{l33 Gamma}) acting on $\Phi$ as the action by a particular Beltrami differential:
\begin{equation}\label{l39 descent via Beltrami}
\iota_v \Gamma \Phi(z) = \int_\Sigma d^2 x (\mu(x) G(x)+\bar\mu(x) \ol{G}(x)) \Phi(z),
\end{equation}
where 
\begin{equation}
\mu= \bar\bdd (\theta_D v^{1,0}),\;\; \bar\mu = \bdd (\theta_D v^{0,1}).
\end{equation}
Here $D$ is a small disk containing $z$, $\theta_D$ is the function equal to $1$ in the disk and zero outside it,  $v^{1,0}$ is a holomorphic vector field whose value at $z$ is the $(1,0)$-component of the vector $v$, and likewise for $v^{0,1}$. The idea is that the r.h.s. of (\ref{l39 descent via Beltrami}) is a contour integral of $G,\ol{G}$ over the boundary of the disk, which is exactly the descent operator $\Gamma$.

From this standpoint, the differential form (\ref{l39 corr of descent towers}) on the configuration space (defined via repeated action of $\Gamma$ at punctures) contracted with some tangent vectors to $C_n(\Sigma)$ can be written in terms of Beltrami differentials, i.e., as a special case of (\ref{l39 alpha_p}).
\end{remark}

\section{2d cohomological field theories}
Given a TCFT, restricting to correlators of $Q$-cocycles (extended to total descent towers as in (\ref{l39 corr of descent towers})), one obtains a simpler structure called a \emph{cohomological field theory}.\footnote{Here we are making an implicit assumption that the correlators extend to the Deligne-Mumford compactification of the moduli spaces $\MM_{g,n}$.}

The following definition is from Kontsevich-Manin \cite[section 6.1]{KM}.

\begin{definition}\label{l39 def: CohFT}
A 2d cohomological field theory is the following data:
\begin{itemize}
\item A $\ZZ$-graded complex vector space $W$ with an inner product $\langle,\rangle$.\footnote{In the cohomological field theory associated with a TCFT, one should think of $W$ as the $Q$-cohomology of the space of fields of the TCFT, $W_\mr{CohFT}
=H_Q(V_\mr{TCFT})$.}
\item A collection of linear maps (correlators)
\begin{equation}
I_{g,n} \colon W^{\otimes n}\ra H^\bt(\ol\MM_{g,n})
%\otimes \mr{Hom}(W^\otimes n,\CC)
\end{equation}
with $g,n\geq 0$ satisfying 
\begin{equation}\label{l39 stability}
2-2g-n<0
\end{equation}
(``stability'' condition). I.e., $I_{g,n}$ maps an $n$-tuple of  elements of $W$ to a de Rham
cohomology class of the Deligne-Mumford compactification $\ol\MM_{g,n}$ of the moduli spaces of complex structures. 
\end{itemize}

The collection of maps $I_{g,n}$ is assumed to satisfy the following factorization axioms. 
%on the Deligne-Mumford compactification strata of the moduli spaces.
\begin{enumerate}[(i)]
\item Let $S=\{i_1,\ldots,i_{n_1}\}\subset \{1,\ldots,n\}$ be a subset with $n_1$ elements 
and $S^c=\{j_1,\ldots,j_{n_2}\}$ its complement, with $n_2=n-n_1$ elements. For $g_1+g_2=g$, let 
\begin{equation}\label{l39 DM stratum I}
\dd^\mr{I}_{g_1;S}\ol\MM_{g,n}\simeq \ol\MM_{g_1,n_1+1}\times \ol\MM_{g_2,n_2+1}
\end{equation}
be the Deligne-Mumford compactification stratum of complex codimension 1 (a.k.a. ``compactification divisor''), corresponding to nodal curves where one component has genus $g_1$ and contains punctures from the subset $S$, plus the ``node'' or ``neck'' puncture and the second component similarly has genus $g_2$ and contains punctures from $S^c$, plus the ``neck'' puncture.\footnote{See Remark \ref{l11 rem: DM for higher genus}.} 
Then the factorization axiom is:
\begin{multline}\label{l39 CohFT factorization}
I_{g,n}(\Phi_1,\ldots,\Phi_n)\Big|_{\dd^\mr{I}_{g_1;S}\ol\MM_{g,n}}=\\
=
\sum_{k,l} I_{g_1,n_1+1}(\Phi_{i_1},\ldots,\Phi_{i_{n_1}},e_k) h^{kl} I_{g_2,n_2+1}(\Phi_{j_1},\ldots,\Phi_{j_{n_2}},e_l)
\end{multline}
Here $\Phi_1,\ldots,\Phi_n\in W$ any elements. We also introduced a basis $\{e_k\}$ in $W$ and $h^{kl}$ is the inverse matrix of the inner product in this basis $h_{kl}=\langle e_k,e_l \rangle$.
\item Consider the second type of Deligne-Mumford compactification stratum, corresponding to introducing a neck on a handle,
\begin{equation}\label{l39 DM stratum II}
\dd^\mr{II}\ol\MM_{g,n}\simeq \ol\MM_{g-1,n+2}.
\end{equation}
The corresponding factorization axiom is:
\begin{equation}\label{l39 CohFT factorization 2}
I_{g,n}(\Phi_1,\ldots,\Phi_n)\Big|_{\dd^\mr{II}\ol\MM_{g,n}}=\sum_{k,l} h^{kl}
I_{g-1,n+2}(\Phi_1,\ldots,\Phi_n,e_k,e_l).
\end{equation}
\end{enumerate}
\end{definition}

%\textcolor{red}{PICTURES: factorization axioms}
\begin{figure}[H]
$$\vcenter{\hbox{ \includegraphics[scale=0.75]{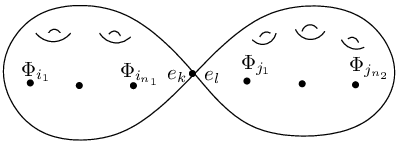}
}}
\quad 
\vcenter{\hbox{ \includegraphics[scale=0.75]{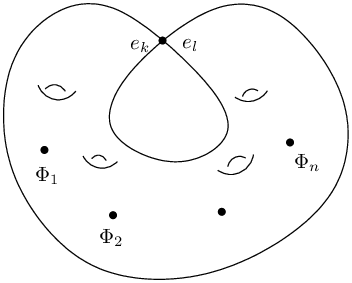}
}}
$$
\caption{Factorization on nodal curves.
}
\end{figure}

\begin{remark}
Thinking of a cohomological field theory as a reduction of a ``parent'' TCFT by passing to $Q$-cohomology, the factorization axioms above are a consequence of Segal's sewing axiom for the parent TCFT.
\end{remark}

\section{Gromov-Witten cohomological field theory %invariants
}
Fix a compact K\"ahler manifold $X$ (the target). We will assume that the K\"ahler symplectic form $\omega$ on $X$ 
%has the form 
%\begin{equation}\label{l39 omega=A omega_0}
%\omega=A \omega_0
%\end{equation} 
%where $A\in\RR_{>0}$ is a scaling factor and $\omega_0$ 
has integer periods.\footnote{In fact, the story of this section goes through under much milder assumptions: one just needs to require $X$ to be a symplectic manifold with compatible almost complex structure, such that the symplectic form has integer periods. 
%equipped with an integer cohomology class $[\omega]\in H^2(X,\ZZ)$ (needed to define the notion of the degree of a map $\Sigma\ra X$); the generating parameter $q$ in the Gromov-Witten potential can be treated as a formal parameter.  
The stronger assumption that $X$ is K\"ahler 
%and that the degree comes from the K\"ahler form 
comes from the field theory side, where one wants to start with a sigma-model, cf. Section \ref{ss A-model} (in the original approach \cite{Witten88}, with $\mc{N}=(2,2)$ supersymmetric sigma-model).
}

We will be constructing a cohomological field theory in the sense of  Definition \ref{l39 def: CohFT} with the space of fields $W=H^\bt_\mr{de\,Rham}(X)$. This cohomological field theory, called Gromov-Witten theory, arises as a reduction by passing to $Q$-cohomology from a certain TCFT -- the A-model, which is a sigma-model with target $X$ (coupled to certain extra fields).

Let $\Sigma$ be a closed Riemannian surface. 
For any smooth map $\phi\colon \Sigma\ra X$, we define the \emph{degree} of $\phi$ as
\begin{equation}\label{l39 d}
d=\int_\Sigma \phi^* \omega\quad \in \ZZ.
\end{equation}
Let us denote by $\mr{Hol}_d(\Sigma,X)$ the space of holomorphic maps $\phi\colon \Sigma\ra X$ of a fixed  degree $d$.

The space $\mr{Hol}_d(\Sigma,X)$ is finite-dimensional for any $d\in \ZZ$;\footnote{For this statement, compactness of $X$ is crucial.} it vanishes for $d<0$ and consists of constant maps for $d=0$: 
\begin{equation}\label{l39 Hol_0 = X}
\mr{Hol}_0(\Sigma,X)=X.
\end{equation}

\begin{example}
Let the surface be $\Sigma=\CP^1$ with homogeneous coordinates $(z_0:z_1)$ and let the target be $X=\CP^N=(\CC^{N+1}\backslash\{0\})/\CC^*$ with homogeneous coordinates $(u_0:\cdots:u_N)$. We assume that the target $\CP^N$ is equipped with the symplectic structure $\omega_0=\omega_\mr{FS}$ the Fubini-Study 2-form normalized to have unit integral over $\CP^1\subset \CP^N$. 
We describe degree $d$ holomorphic maps $\CP^1\ra \CP^N$ as degree $d$ polynomial maps 
\begin{equation}
\CC^2\backslash\{0\} \ra \CC^{N+1}\backslash \{0\}
\end{equation}
where we subsequently quotient both sides by $\CC^*$.

Thus, a degree $d$ holomorphic map $\CP^1\ra \CP^N$ is given as
\begin{equation}
u_p=A_p(z_0,z_1), \quad 0\leq p\leq N,
\end{equation}
where $A_0,\ldots,A_p$ are homogeneous polynomials of degree $d$ in $z_0,z_1$. Tuples of polynomials $\{A_p\}$ and $\{A_p'\}$ determine the same map $\CP^1\ra \CP^n$ if and only if $A'_0=c A_0,\ldots, A'_n=cA_n$ for some $c\in \CC^*$. %for all $p=0,\ldots,N$.
Also, a tuple $\{A_p\}$ determines a map $\CP^1\ra \CP^N$ if and only if the polynomials $\{A_p\}$ do not have a common nontrivial root $(z_0,z_1)$ -- if they do, then there is a point of $\CC^2\backslash\{0\}$ which is mapped to $\{0\}\in \CC^{N+1}$, which does not correspond to any point in $\CP^N$. Such tuples $\{A_p\}$ correspond to so-called \emph{Drinfeld's quasimaps} $\CP^1\ra \CP^N$; they are not however holomorphic maps in the usual sense (in particular they cannot be evaluated at all points of the source), so we will discard them. In summary, the space of holomorphic maps of degree $d$ is
\begin{multline}\label{l39 Hol from CP^1 to CP^n}
\mr{Hol}_d(\CP^1,\CP^N)=\\=\{ (A_p(z_0,z_1)=\sum_{j=0}^d a_{pj }z_0^j z_1^{d-j})_{p=0,\ldots,n}\; |\; \{A_p\} \;\mbox{do not have common roots} \}\;/\;\CC^* \\
= \CP^{(d+1)(N+1)-1}\backslash \mathbb{D}
\end{multline}
where $a_{pj}\in\CC$ are the coefficients of the polynomials -- thus in order to specify a holomorphic map $\CP^1\ra \CP^N$ we need to specify the $(N+1)\times (d+1) $ array of coefficients $a_{pj}$, modulo scaling them all by a number $c\in\CC^*$, which yields the projective space $\CP^{(d+1)(N+1)-1}$.  We denoted the set of ``prohibited'' configurations corresponding to quasimaps by $\mathbb{D}$ -- it is a subvariety in $\CP^{(d+1)(N+1)-1}$ of positive codimension and can be described as $\mathbb{D}\simeq \CP^1\times \CP^{d(N+1)-1}$ (the first factor in the r.h.s. gives the point on the source where the common root occurs).

As a further simplicifation, consider the case $N=1$. Then degree $d$ holomorphic maps $\CP^1\ra \CP^1$ are described by
\begin{equation}
(1:z) \mapsto (1:\frac{A_1(z)}{A_0(z)})
\end{equation}
where $A_0$ and $A_1$ are two degree $d$ polynomials in the variable $z$ without common roots. For instance, for $d=1$ the holomorphic maps are
\begin{equation}
(1:z)\mapsto (1: a \frac{z-b}{z-c})
\end{equation}
with parameters $a,b,c\in \CC$ such that $a\neq 0$ and $b\neq c$ (otherwise it is a quasimap).
 \end{example}

\subsection{Genus zero case.} 
Let $\Sigma=\CP^1$.  
We have a diagram of maps
\begin{equation}\label{l39 ev p diagram}
\begin{CD}
C_n(\Sigma)\times \mr{Hol}_d(\Sigma,X) @>{\mr{ev}}>> \underbrace{X\times\cdots \times X}_n \\
@VpVV \\
C_{n}(\Sigma)
\end{CD}
\end{equation}
Here $\mr{ev}$ is the evaluation map, evaluating the holomorphic map on an $n$-tuple of points in $\Sigma$,
\begin{equation}
\mr{ev}\colon ((z_1,\ldots,z_n),\phi)\mapsto (\phi(z_1),\ldots,\phi(z_n)).
\end{equation}
The vertical map $p$ in (\ref{l39 ev p diagram}) is the projection onto the first factor.

\begin{remark}\label{l39 rem: compactification}
The two objects in the left column in (\ref{l39 ev p diagram}) admit a certain compactification (we will leave it as a black box and denote it by an overline) such that the maps $\mr{ev},p$ extend to it.\footnote{The compactification of the configuration space $C_n(\Sigma)$ is due to Fulton-MacPherson. The compactification of $C_n(\Sigma)\times \mr{Hol}_d(\Sigma,X)$ is a special case of Kontsevich's compactification of the moduli space of stable maps.}
\end{remark}

Fix a collection of closed forms on the target, $\alpha_1,\ldots,\alpha_n\in \Omega^\bt_\mr{cl}(X)$. Then one can define
\begin{equation}\label{l39 I_0,d}
I_{0,n,d}(\alpha_1,\ldots,\alpha_n)=\int_{\mr{Hol}_d(\Sigma,X)} \mr{ev}^*(\pi_1^*(\alpha_1)\wedge \cdots \wedge \pi_n^*(\alpha_n)) \quad \in \Omega^\bt_\mr{cl}({C}_n(\Sigma)),
\end{equation}
where $\pi_i\colon X^n \ra X$ is the projection onto the $i$-th factor. The form (\ref{l39 I_0,d}) has the following properties:
\begin{enumerate}[(i)]
\item \label{l39 I_0,d property (i)} it is closed and its cohomology class is depends only on the cohomology classes of forms $\alpha_i$ -- this fact follows from Stokes' theorem for fiber integrals and relies on the existence of compactifications, cf. Remark \ref{l39 rem: compactification}. 
\item The form (\ref{l39 I_0,d}) extends to a closed form on the Fulton-MacPherson compactified configuration space $\ol{C}_n(\Sigma)$.
\item The form is also basic w.r.t. M\"obius transformations (which act diagonally in the top left corner in (\ref{l39 ev p diagram}) and in the obvious way on the configuration space). Therefore, the $I_{0,n,d}$ descends to a closed from on the moduli space $\ol\MM_{0,n}$:
\begin{equation}
I_{0,n,d}(\alpha_1,\ldots,\alpha_n)\in \Omega^\bt_\mr{cl}(\ol\MM_{0,n}),
\end{equation}
and by (\ref{l39 I_0,d property (i)}) above, the construction descends to de Rham cohomology:
\begin{equation}
I_{0,n,d}([\alpha_1],\ldots,[\alpha_n])\in H^\bt(\ol\MM_{0,n}).
\end{equation}
This is the so-called Gromov-Witten cohomology 
class.
\end{enumerate}

The genus zero part of the Gromov-Witten cohomological field theory is then defined as
\begin{equation}\label{l39 GW class genus 0}
I_{0,n}([\alpha_1],\ldots,[\alpha_n])\colon= \sum_{d\geq 0} q^d I_{0,n,d}([\alpha_1],\ldots,[\alpha_n]),
\end{equation}
where $q$ is a formal (infinitesimal) generating parameter.
%$q\colon= e^{-A}$ is the generating parameter, with $A$ as in (\ref{l39 omega=A omega_0}). 

\begin{remark} The A-model, the ``parent'' TCFT for the Gromov-Witten hohomological field theory, contains a class of $Q$-closed observables: for each closed form $\alpha$ on the target one has an ``evaluation observable'' $O_{\alpha}\in V$, see Section \ref{sss evaluation observables}. The cohomology class (\ref{l39 GW class genus 0})  is the cohomology class of the $n$-point correlator on $\CP^1$,
\begin{equation}
\langle \til{O}_{\alpha_1}\cdots  \til{O}_{\alpha_n}\rangle,
\end{equation}
where tilde means the full descendant, cf. (\ref{l39 corr of descent towers}).
\end{remark}

\begin{comment}
\begin{thm}
The cohomology classes (\ref{l39 GW class genus 0}) satisfy the genus-zero factorization property (\ref{l39 CohFT factorization}).
\end{thm}

The heuristic argument is that a holomorphic map from a nodal curve $\Sigma=\Sigma_1\cup_p \Sigma_2$ to $X$ (with $p$ the node) factorizes %as a fiber product, over the maps from the node to $X$, of 
into two maps from the two components of the curve, so one has \marginpar{write more.. move to after general genus?}
\begin{equation} 
\mr{Hol}_d(\Sigma,X)= \bigsqcup_{d_1+d_2=d}\mr{Hol}_{d_1}(\Sigma_1,X)\times_{\mr{Map}(p,X)} \mr{Hol}_{d_2}(\Sigma_2,X) 
\end{equation}
....
\end{comment}

\marginpar{Lecture 40,\\ 12/5/2022}
\subsection{General genus}
Let $\Sigma$ be a closed oriented smooth surface of any genus $g$ %(we are not fixing a complex structure yet) 
and fix $d\geq 0$. One has a fiber bundle over the moduli space of complex structures on $\Sigma$ with fiber over $J\in \MM_\Sigma$ the space of holomorphic maps (w.r.t. to the complex structure $J$ on $\Sigma$) to $X$ of degree $d$:
\begin{equation}\label{l40 Hol bundle}
\begin{CD}
%\til{\mr{Hol}}_d(\Sigma,X) 
\MM_\Sigma(X,d)
@<<< \mr{Hol}_d(\Sigma,X) \\
@VVV \\
\MM_\Sigma
\end{CD}
\end{equation}
We have the ``forgetful'' map 
\begin{equation}
r\colon \MM_{\Sigma,n}\ra \MM_\Sigma
\end{equation}
from the moduli space with $n$ marked points to the moduli space without marked points, given by forgetting the marked points. The pullback of the bundle (\ref{l40 Hol bundle}) along the forgetful map $
\MM_{\Sigma,n}(X,d)\colon= r^*\MM_\Sigma(X,d) %\til{\mr{Hol}}_d(\Sigma,X)
\ra \MM_{\Sigma,n}$ fits into the diagram similar to (\ref{l39 ev p diagram}):
\begin{equation}
\begin{CD}
\MM_{\Sigma,n}(X,d)
%r^*\til{\mr{Hol}}_d(\Sigma,X)
 @>\mr{ev}>>  X^n 
%\underbrace{X\times \cdots \times X}_n 
\\
@VpVV \\ \MM_{\Sigma,n} %=\MM_{g,n}
\end{CD}
\end{equation}
where $\mr{ev}$ evaluates the holomorphic map at the $n$ marked points. Again, there exists a compactification of the objects in the right column of the diagram -- Kontsevich's moduli space of stable maps at the top and Deligne-Mumford compactification of $\MM_{g,n}$ at the bottom -- such that the maps $\mr{ev},p$ extend to the compactifications:\footnote{
Very roughly, the idea is that in addition to adjoining Deligne-Mumford compactification strata (nodal curves) coming from the compactification of the base, in the total space one needs to blow up the configurations where a quasimap point in the space of holomorphic maps coincides with a marked point on the surface (i.e. exactly the situations where the evaluation of a map at a marked point becomes problematic).
}%\marginpar{check the footnote}
\begin{equation}\label{l40 ev p diagram compactified}
\begin{CD}
\ol\MM_{g,n}(X,d)
%\ol{r^*\til{\mr{Hol}}}_d(\Sigma,X) 
@>\mr{ev}>>  X^n 
%\underbrace{X\times \cdots \times X}_n 
\\
@VpVV \\ \ol\MM_{g,n}
\end{CD}
\end{equation}
Here we put the genus of $\Sigma$ instead of $\Sigma$ as index.

Given closed forms $\alpha_1,\ldots,\alpha_n\in \Omega^\bt_\mr{cl}(X)$, we construct a form
\begin{equation}\label{l40 I form}
I_{g,n,d}(\alpha_1,\ldots,\alpha_n)= \int_{\mr{Hol}_d(\Sigma,X)} \mr{ev}^* (\pi_1^*(\alpha_1)\wedge \cdots \wedge \pi_n^*(\alpha_n))\quad \in H^\bt(\MM_{g,n})
\end{equation}
As in genus zero case, by Stokes' theorem and due to the existence of compactifications, this form is closed and its cohomology class depends only on the cohomology classes of $\alpha_i$; thus the construction descends to cohomology. Also, the form (\ref{l40 I form}) extends to the compactification of the moduli space of complex structures $\ol\MM_{g,n}$.
%changes by an exact form if any of $\alpha_i$ change by an exact form.

This leads to the following definition at the level of cohomology.
%The construction (\ref{l40 I form}) phrased at the level of cohomology is the following:
\begin{definition}
Genus $g$ Gromov-Witten classes are defined
via the diagram (\ref{l40 ev p diagram compactified}) as
\begin{equation}\label{l40 GW class}
I_{g,n,d}([\alpha_1],\ldots,[\alpha_n])=p_*\mr{ev}^* (\pi_1^*[\alpha_1]\wedge \cdots \pi_n^*[\alpha_n])\quad \in H^\bt(\ol\MM_{g,n}),
\end{equation}
where $[\alpha_1],\ldots, [\alpha_n]\in H^\bt(X)$ are any de Rham cohomology classes of the target $X$.
\end{definition}

As a generalization of (\ref{l39 GW class genus 0}) to any genus, Gromov-Witten cohomological field theory is defined by summing the classes (\ref{l40 GW class}) over the degree $d$, weighed with $q^d$, %with $q=e^{-A}$ the generating parameter: 
\begin{equation}\label{l40 GW class summed over d}
I_{g,n}([\alpha_1],\ldots,[\alpha_n])\colon= \sum_{d\geq 0} q^d I_{g,n,d}([\alpha_1],\ldots,[\alpha_n]).
\end{equation}

\begin{thm}\label{l40 thm factorization}
The cohomology classes $I_{g,n}$ %(\ref{l40 GW class summed over d}) 
satisfy the factorization properties (\ref{l39 CohFT factorization}), (\ref{l39 CohFT factorization 2}).
\end{thm}

\begin{proof}[Idea of proof]
 Fix the numbers $g,n,d\geq 0$, fix a splitting of genus $g=g_1+g_2$ and a splitting of the set of marked points into complementary subsets $\{1,\ldots,n\}=S\sqcup S^c$. Consider a compactification stratum $\dd_{g_1,S}\ol\MM_{g,n}$ of the moduli space of complex structures. %with $S\subset \{1,\ldots,n\}$ a subset of marked points with complement $S^c$. 
 The restriction of the bundle (\ref{l40 ev p diagram compactified}) to it is\footnote{The intuition is that a holomorphic map $\phi$ from a nodal curve $\Sigma=\Sigma'\cup_q \Sigma''$ to $X$ is given by a pair of holomorphic maps, $\phi'$ on $\Sigma'$ and $\phi''$ on $\Sigma''$ agreeing at the node $q$. The degree of $\phi$ splits as the degree of $\phi'$ plus the degree of $\phi''$.}
\begin{equation}
p^{-1}(\dd_{g_1,S}\ol\MM_{g,n}) \simeq \bigsqcup_{d_1+d_2=d} \ol\MM_{g_1,S\cup q}(X,d_1)\times_X \ol\MM_{g_2,S^c\cup q^*}(X,d_2)
\end{equation}
Here  $q,q^*$ are the names of the nodal point as point seen as a marked point on either component of the nodal curve; the fiber product in the r.h.s. is w.r.t. evaluations at $q$ and at $q^*$, respectively. The evaluation map on the r.h.s. lands in $X^S\times \Delta \times X^{S^c}$ where $\Delta\subset X\times X$ is the diagonal.

Fix the cohomology classes $[\alpha_1],\ldots,[\alpha_n]\in H^\bt(X)$.
%On the one hand 
Then we have
\begin{equation}\label{l40 factorization proof eq1}
\mr{ev}^*(\prod_{i=1}^n \pi_i^* [\alpha_i])\Big|_{p^{-1}(\dd_{g_1,S}\ol\MM_{g,n})}=\sum_{k,l}\mr{ev}^*_{S\cup q} (\prod_{i\in S} \pi_i^*[\alpha_i]\cdot \pi_q^* e_k )h^{kl} \mr{ev}^*_{S^c\cup q^*} (\prod_{i\in S^c} \pi_i^*[\alpha_i]\cdot \pi^*_{q^*} e_l)
\end{equation}
where $e_{k}$ is a basis in $H^\bt(X)$ and $h^{kl}$ is the inverse matrix of Poincar\'e pairing; $\mr{ev}$ in the l.h.s. is for 
%$\ol\MM_{g,n}(X,d)$ 
holomorphic maps out of the whole nodal curve $\Sigma$
and in the r.h.s. we have maps $\mr{ev}$ for the two components of $\Sigma$. Here we used the fact that the cohomology class of $X\times X$ Poincar\'e dual to the homology class of the diagonal $\Delta\subset X\times X$ is $\sum_{k,l}h^{kl}e_k\otimes e_l$. The appearance of this class in the r.h.s. of (\ref{l40 factorization proof eq1}) effectively forces $q$ and $q^*$ to map to the same point in $X$.

Pushing forward (i.e. performing the fiber integral) the l.h.s. of (\ref{l40 factorization proof eq1}) to the Deligne-Mumford stratum $\dd_{g_1,S}\ol\MM_{g,n}$ and pushing forward the r.h.s. 
%of (\ref{l40 factorization proof eq1})
to the product $\ol\MM_{g_1,S\cup q}\times \ol\MM_{g_2,S^c\cup q^*}$, and summing over the degree $d$ (and the splittings $d=d_1+d_2$) with weight $q^d$, we obtain the desired factorization property (\ref{l39 CohFT factorization}):
\begin{equation}\label{l40 factorization}
I_{g,n}([\alpha_1],\ldots,[\alpha_n])|_{\dd_{g_1,S}\ol\MM_{g,n}}=
\sum_{k,l} I_{g_1,n_1+1}(\{[\alpha_i]\}_{i\in S},e_k) h^{kl} I_{g_2,n_2+1}(\{[\alpha_i]\}_{i\in S^c},e_l)
\end{equation}
The factorization property on the second type of Deligne-Mumford strata (\ref{l39 CohFT factorization 2}) is proved similarly.

\end{proof}

%
%The heuristic argument is that a holomorphic map from a nodal curve $\Sigma=\Sigma_1\cup_p \Sigma_2$ to $X$ (with $p$ the node) factorizes %as a fiber product, over the maps from the node to $X$, of 
%into two maps from the two components of the curve, so one has \marginpar{finish! }
%\begin{equation} 
%\mr{Hol}_d(\Sigma,X)= \bigsqcup_{d_1+d_2=d}\mr{Hol}_{d_1}(\Sigma_1,X)\times_{\mr{Map}(p,X)} \mr{Hol}_{d_2}(\Sigma_2,X) 
%\end{equation}
%....

\begin{definition}
For a collection of cohomology classes $[\alpha_1],\ldots,[\alpha_n]\in H^\bt(X)$. The genus $g$, $n$-point Gromov-Witten invariant of degree $d$ is defined as the pairing of the Gromov-Witten class (\ref{l40 GW class}) with the fundamental class of the moduli space $\ol\MM_{g,n}$:
\begin{equation}
\GW_{g,n,d}([\alpha_1],\ldots,[\alpha_n])\colon = \int_{\ol\MM_{g,n}} I_{g,n,d}([\alpha_1],\ldots,[\alpha_n]) \quad \in \CC
\end{equation}
\end{definition}

\subsection{Enumerative meaning of Gromov-Witten classes}
Fix $c_1,\ldots, c_n\in C_\bt(X,\ZZ)$ -- a collection of cycles in $X$ and let $[\delta_{c_i}]\in H^\bt(X)$ be the Poincar\'e dual classes to the homology classes of $c_i$; $[\delta_{c_i}]$ can be represented in de Rham cohomology by a (cohomologically smeared) Dirac delta-form on $c_i$, hence the notation.

Recall that
for $c\subset X$ a $k$-cycle in a smooth $N$-manifold $X$, the delta-form $\delta_c$ is the distributional $(N-k)$-form characterized by the property 
\begin{equation}
\int_X \delta_c \wedge \alpha =\int_c \alpha|_c
\end{equation} 
for any $\alpha\in \Omega^k(X)$.
A cohomologically smeared $\delta$-form on $c$ is a smooth form with the same property which is only required to hold for $\alpha$ a \emph{closed} $k$-form.

%\marginpar{edit: could be a rational number if a map/surface have automorphisms}
The Gromov-Witten invariant
\begin{equation}\label{l39 GW}
\GW_{g,n,d}([\delta_{c_1}],\ldots,[\delta_{c_n}])
%\colon= \int_{\ol\MM_{g,n}} I_{g,n,d}([\delta_{c_1}],\ldots, [\delta_{c_n}]) 
\quad \in \mathbb{Q}
\end{equation}
is %an integer -- 
the ``virtual'' count of holomorphic curves in $X$ of genus $g$ and degree $d$ passing through the cycles $c_1,\ldots,c_n$. This number is an integer for zero genus. Generally, for higher genus, it is a rational number: holomorphic maps $\phi$ in this virtual count should be weighed with $\frac{1}{|\mr{stab}(\phi)|}$ -- the inverse of the number of holomorphic automorphisms $\Sigma$ commuting with $\phi$.\footnote{A related point: compactified moduli spaces $\ol\MM_{g,n}(X)$ have orbifold singularities which lead to having the ``virtual'' fundamental class defined over $\mathbb{Q}$ rather than $\ZZ$.}%\marginpar{check the footnote}

%This number is called the Gromov-Witten invariant. 
%The integral in (\ref{l39 GW}) stand for the pairing with the fundamental class of $\ol\MM_{g,n}$. 
%\textcolor{red}{PICTURE}
\begin{figure}[H]
$$
\vcenter{\hbox{ \includegraphics[scale=0.75]{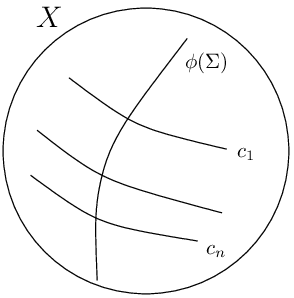}
}}
$$
\caption{The enumerative problem: counting holomorphic maps $\phi\colon \Sigma \ra X$ with the image passing through a given set of cycles $c_1,\ldots,c_n$
}
\end{figure}

\subsection{Quantum cohomology ring}
Consider de Rham cohomology of $X$ as a $\ZZ_2$-graded\footnote{
The reason for why we only consider the mod $2$ projection of the natural $\ZZ$-grading on cohomology is elucidated in Example \ref{l40 ex: CP^1} below: $\ZZ$-grading is not preserved by the deformation of the cup product we are describing.
} vector space $H^\bt(X)$ equipped with an inner product $\langle,\rangle$ given by Poincar\'e pairing 
$\langle [\alpha_1],[\alpha_2] \rangle=\int_X \alpha_1\wedge \alpha_2$ and equipped with a bilinear map
\begin{equation}
m\colon H(X)\otimes H(X) \ra H(X)
\end{equation}
characterized by
\begin{equation}\label{l40 quantum product}
\langle m([\alpha_1],[\alpha_2]),[\alpha_3] \rangle \colon= \sum_{d\geq 0} q^d \GW_{0,3,d}([\alpha_1],[\alpha_2],[\alpha_3])
\end{equation}
with $q%=e^{-A}
$ 
the generating parameter as in (\ref{l39 GW class genus 0}). Note that Gromov-Witten classes in the r.h.s. here are elements of $H^\bt(\ol\MM_{0,3})$, i.e., numbers (since $\ol\MM_{0,3}$ is a point). If $\alpha_i$ are integer classes then the Gromov-Witten invariant $GW_{0,3,d}$ is an integer.

\begin{definition} The bilinear operation $m\colon H(X)\otimes H(X) \ra H(X)$ defined by (\ref{l40 quantum product}) is called the ``quantum product'' on the cohomology $H(X)$. The quantum product endows the cohomology $H(X)$ %equipped with the operation $m$ is called the 
with the structure of a $\ZZ_2$-graded ring called the 
``quantum cohomology ring.''
\end{definition}

Note that due to (\ref{l39 Hol_0 = X}) one has
\begin{equation}
\GW_{0,3,0}([\alpha_1],[\alpha_2],[\alpha_3]) = \int_X \alpha_1\wedge \alpha_2\wedge\alpha_3.
\end{equation}
Thus, the $q^0$ term in $m$ is the usual cup product while $q^{>0}$ terms comprise a deformation of the cup product by the data of (genus-zero, three-point) Gromov-Witten classes.

Implicitly present in the definition above (in the words ``ring'' and ``product'') is the following.
\begin{lemma}
The operation $m$ defined by (\ref{l40 quantum product}) is supercommutative and associative.
\end{lemma}
Supercommutativity is obvious from the definition of Gromov-Witten classes. Associativity is not obvious and is a consequence of the WDVV equation (\ref{l40 WDVV}).%\marginpar{REF to WDVV}

\begin{example}\label{l40 ex: CP^1}
Let $X=\CP^1$. The space of holomorphic maps of degree $d$ is given by (\ref{l39 Hol from CP^1 to CP^n}):
\begin{equation}
\mr{Hol}_d(\CP^1,\CP^1)=\CP^{2d+1}\backslash\mathbb{D}
\end{equation}
-- it is a manifold of real dimension $2(2d+1)=4d+2$.
Thus the Gromov-Witten invariants 
\begin{equation}\label{l40 CP^1 GW_0,3,d}
\GW_{0,3,d}(\alpha_1,\alpha_2,\alpha_3)=\int_{\mr{Hol}_d(\CP^1,\CP^1)} \mr{ev}^*(\pi_1^*\alpha_1\wedge \pi_2^* \alpha_2\wedge \pi_3^*\alpha_3).
\end{equation}
Note that for this number to be nonzero it is necessary that the dimension of the space over which we integrate is equal to the degree of the form we are integrating:
\begin{equation}\label{l40 CP^1 quantum product balancing}
4d+2 = |\alpha_1|+|\alpha_2|+|\alpha_3|,
\end{equation}
where $|\alpha|$ is the de Rham degree of the form $\alpha$.

The cohomology of $\CP^1$ is spanned by two classes, $[1]\in H^0(\CP^1)$ and $[\omega]\in H^2(\CP^1)$ -- the class of the Fubini-Study 2-form normalized to have unit volume. Choosing $\alpha_{1,2,3}$ in (\ref{l40 CP^1 GW_0,3,d}) to be the basis classes in $H^\bt(\CP^1)$ we observe that there are only two possibilities (up to permutations) to satisfy (\ref{l40 CP^1 quantum product balancing}):
\begin{eqnarray}
\GW_{0,3,0}([1],[1],[\omega])&=&1 \label{l40 CP^1 GW(1,1,1)}, \\
\label{l40 CP^1 GW(o,o,o)}
\GW_{0,3,1}([\omega],[\omega],[\omega])&=& 1.
\end{eqnarray}
Note that (\ref{l40 CP^1 GW(1,1,1)}) corresponds to the usual cup product in cohomology ($[1]\cup[1]=[1]$, or $[1]\cup[\omega]=[\omega]$). On the other hand, (\ref{l40 CP^1 GW(o,o,o)}) is the number of degree $1$ holomorphic maps $\CP^1\ra \CP^1$ (i.e., M\"obius transformations) mapping three marked points in the source $\CP^1$ into three fixed points $c_1,c_2,c_3$ in the target $\CP^1$ in general position (we then think of $c_i$ as zero-cycles with $[\omega]$ the Poincar\'e dual cohomology class for each $c_i$). There is exactly one such map.

To summarize the result, the quantum product in the cohomology of $\CP^1$ is given by the following multiplication table.
\begin{equation}
m([1],[1])=[1],\quad m([1],[\omega])=[\omega],\quad m([\omega],[\omega])=q\cdot [1].
\end{equation}
Note that due to the last relation the quantum product does not preserve the de Rham degree. In this particular example, $X=\CP^1$, one can prescribe degree $4$ to $q$ and then $m$ preserves the $\ZZ$-degree.

\end{example}

\subsection{Gromov-Witten potential}
Fix a basis $e_1,\ldots,e_s$ for $H^\bt(X)$. The function 
\begin{equation}
\Phi(t_1,\ldots,t_s)\colon = \sum_{n_1,\ldots,n_s\geq 0} \sum_{d\geq 0} \frac{t_1^{n_1}\cdots t_s^{n_s}}{n_1!\cdots n_s!} q^d\; \GW_{0,\sum n_i,d}(\underbrace{e_1,\ldots,e_1}_{n_1},\ldots,\underbrace{e_s,\ldots,e_s}_{n_s})
\end{equation}
of the generating parameters $t_1,\ldots,t_s$ is called the Gromov-Witten potential. Here we understand that the variable $t_a$ is even (commuting) if $e_a\in H^\mr{even}(X)$ and $t_a$ is odd if $e_a\in H^\mr{odd}(X)$. Thus, $\Phi$ is a generating function for Gromov-Witten invariants. 

One can think of $t_1,\ldots,t_n$ as coordinates on $H^\bt(X)$, i.e., coordinates of the vector $\beta=\sum_a t_a e_a\in H^\bt(X)$. Then one can also write $\Phi$ as
\begin{equation}
\Phi(t_1,\ldots,t_s)=\sum_{n\geq 0}\sum_{d\geq 0} \frac{q^d}{n!}\GW_{0,n,d}(\underbrace{\beta,\ldots,\beta}_n)
\end{equation}
%and understand $\Phi$ as a function on a neighborhood $U$ of zero in $H^\bt(X)\times \RR_{>0}$
One can treat parameters $t_a$ as formal (i.e. treat $\Phi$ as a formal power series in $t_a$'s), however the sum over $n$ %and $d$
 is actually convergent for $\beta$ %$(\beta,q)$ 
 in some open set $U$ in $H^\bt(X)%\times \RR_{>0} 
 $. 
%(the second factor corresponds to $q$).

\subsubsection{``Big'' quantum product.}
One defines the ``big quantum product'' as a family parametrized by $\beta=\sum_a t_a e_a\in H^\bt(X)$ of supercommutative associative products on cohomology
\begin{equation}
m_\beta\colon H(X)\otimes H(X)\ra H(X)
\end{equation}
defined by
\begin{equation}
\langle m_\beta(\alpha_1,\alpha_2),\alpha_3 \rangle \colon=
\sum_{n\geq 0}\sum_{d\geq 0} \frac{q^d}{n!} \GW_{0,n+3,d}(\alpha_1,\alpha_2,\alpha_3,\underbrace{\beta,\ldots,\beta}_{n}),
\end{equation}
for any $\alpha_{1,2,3}\in H(X)$.
Thus, it is the construction of the quantum product (\ref{l40 quantum product}) deformed by the class $\beta\in H(X)$.

Note that the big quantum product can be written in terms of the third derivative of the potential $\Phi$:
\begin{equation}
\langle m_\beta(e_a,e_b),e_c \rangle = \frac{\dd^3 \Phi}{\dd t_a \dd t_b \dd t_c},
\end{equation}
for any $a,b,c=1,\ldots,s$;  both sides are understood as functions of  $\beta \in U\subset H(X)$.
%$(\beta,q)\in U\subset H(X)\times \RR_{>0}$.

The big quantum product endows an open subset of cohomology $U\subset H(X)$ with the structure of a \emph{Frobenius manifold}.

The following definition is due to Dubrovin \cite{Dubrovin92,Dubrovin96}.
\begin{definition}
A Frobenius manifold is a manifold $Y$ equipped with the following data:
\begin{itemize}
\item Affine flat structure on $Y$ and a compatible (flat) Riemannian metric $h$.
\item For each $\beta\in Y$, the tangent space $T_\beta Y$ is equipped with a commutative associative product 
\begin{equation}
m_\beta\colon T_\beta Y\otimes T_\beta Y \ra T_\beta Y
\end{equation}
compatible with $h$, in the sense that $h(m_\beta(x,y),z)=h(x,m_\beta(y,z))$.
\item A potential $\Phi \in C^\infty(Y)$ such that
\begin{equation}
h(m(u,v),w)=u\circ v\circ w\circ \Phi
\end{equation}
for any triple of flat vector fields $u,v,w$ on $Y$.
\end{itemize}
\end{definition} 

This definition has a straightforward $\ZZ_2$-graded generalization. To see the big quantum product as equipping an open set in $H(X)$ with the structure of a Frobenius manifold, one should consider the ring of scalars to be formal power series in $q$. %$\CC[[q]]$

\subsection{WDVV equation}
Let $h^{ab}$ be the inverse 
matrix of the Poincar\'e pairing in the basis $\{e_a\}$ in $H(X)$.  The following theorem is due to Witten-Dijkgraaf-Verlinde-Verlinde \cite{Witten_2dgrav}.
\begin{thm}
Gromov-Witten potential $\Phi$ satisfies the following differential equation:%\marginpar{index $d$ is unfortunate here, can be confused with the degree}
\begin{equation}\label{l40 WDVV}
\sum_{c,d}\frac{\dd^3 \Phi}{\dd t_a \dd t_b \dd t_c} h^{cd} \frac{\dd^3 \Phi}{\dd t_d \dd t_e \dd t_f}=\sum_{c,d}
\frac{\dd^3 \Phi}{\dd t_e \dd t_b \dd t_c} h^{cd} \frac{\dd^3 \Phi}{\dd t_d \dd t_a \dd t_f}
\end{equation}
(the r.h.s. is the l.h.s. with indices $a,e$ switched).
\end{thm}
The equation (\ref{l40 WDVV}) is known as Witten-Dijkgraaf-Verlinde-Verlinde (or WDVV) equation. It is a consequence of the factorization properties of Gromov-Witten classes on compactification divisors in $\ol\MM_{g,n}$ (Theorem \ref{l40 thm factorization}) and certain relations between homology classes of these divisors -- so-called Keel's relations, see Section \ref{sss Keel} below.

\begin{remark} WDVV equation is not specific to Gromov-Witten theory: it holds in any  2d cohomological field theory: %integrating the forms $I_{0,n}$ over $\ol\MM_{0,n}$ 
one can define the potential in a general CohFT as
\begin{equation}
\Phi=\sum_{n\geq 3}\frac{1}{n!}\int_{\ol\MM_{0,n}} I_{0,n}(\beta,\ldots,\beta),
\end{equation}
seen as a function of $\beta=t_1 e_1+\cdots t_s e_s\in W$, and then $\Phi$ satisfies (\ref{l40 WDVV}) (the proof we sketch in Section \ref{sss Keel} carries over to this general case).
\end{remark}

\subsection{Example of Gromov-Witten potential: $X=\CP^1$}
%\ul{Example: $X=\CP^1$.}
Consider the example $X=\CP^1$. In this case the cohomology $H(X)$ has a basis $[1],[\omega]$ (with $\omega$ the Fubini-Study 2-form normalized to have unit volume); let us denote the corresponding generating parameters $t_0,t_1$. We already know the numbers $\GW_{0,3,d}$ from (\ref{l40 CP^1 GW(1,1,1)}), (\ref{l40 CP^1 GW(o,o,o)}). %The remaining nonvanishing 
\begin{lemma}
Gromov-Witten invariants for $n\geq 4$ points are
\begin{equation}\label{l40 CP^1 all GW}
\GW_{0,n,d}(\underbrace{[\omega],\ldots,[\omega]}_k,\underbrace{[1],\ldots,[1]}_l)= 
\left\{
\begin{array}{cc}
1& \mr{if}\; l=0,d=1,\\
0& \mr{otherwise}
\end{array}
 \right.
\end{equation}
here $k+l=n$.  
\end{lemma}
\begin{proof}
If $l>0$, $p_*\mr{ev^*}$ is a class on $\ol\MM_{0,n}$ coming as a pullback of a class from $\ol\MM_{0,k}$ via the map forgetting the $l$ points mapping to $[1]$. Being a pullback, it integrates to zero on $\MM_{0,n}$. 

For the case $l=0$ (and hence $k=n$), we have a balancing condition (degree of the form) = (dimension of $\mr{Hol}_d$)+(dimension of $\MM_{0,n}$):
\begin{equation}
2n=2(2d+1)+2(n-3) \quad \Leftrightarrow \quad d=1.
\end{equation}
In the case $k=n$, $d=1$ -- the only case when we might get a nontrivial Gromov-Witten invariant, we are counting the number of M\"obius transfromations $\CP^1\ra \CP^1$ that take points $(0,1,\infty,z_4,\ldots,z_n)$ to points $(u_1,\ldots,u_n)$ where $u_i$ are fixed distinct points on the target and $z_4,\ldots,z_n$ are arbitrary (integrated over when we integrate over $\MM_{0,n}$ in (\ref{l39 GW})). There is exactly one such map.
\end{proof}

As a corollary, the Gromov-Witten potential for $X=\CP^1$ is
\begin{equation}
\Phi(t_0,t_1)= \frac{t_0^2 t_1}{2}+\sum_{n\geq 3}q\frac{t_1^n}{n!} =
\frac{t_0^2 t_1}{2}+q\left(e^{t_1}-1-t_1-\frac{t_1^2}{2}\right).
\end{equation}

The big quantum product is given on basis elements by
\begin{equation}
m_\beta([1],[1])=[1],\quad m_\beta([1],[\omega])=[\omega],\quad
m_\beta([\omega],[\omega])=qe^{t_1}\cdot[1],
\end{equation}
where the reference point is $\beta=t_0[1]+t_1[\omega]\in H(\CP^1)$.

\subsection{Example of Gromov-Witten potential: $X=\CP^2$}
%Next, consider the example 
We proceed to the case $X=\CP^2$. This example due to Kontsevich-Manin is a spectacular application of WDVV equation to enumerative geometry. We refer to original paper \cite{KM} for details. 

One has three basis cohomology classes: $[1],[\omega],[\omega^2]$ where again $\omega$ is the Fubini-Study 2-form normalized to have unit period on $\CP^1\subset \CP^2$. Let us denote the corresponding generating parameters $t_0,t_1,t_2$.
\begin{thm}[Kontsevich-Manin \cite{KM}]
\begin{enumerate}[(i)]
\item The Gromov-Witten potential for $X=\CP^2$ has the form
\begin{equation}\label{l40 KM thm eq1}
\Phi(t_0,t_1,t_2)=\frac{t_0^2t_2}{2}+\frac{t_0 t_1^2}{2}-q\frac{t_2^2}{2}+\sum_{d\geq 1}\frac{\mc{N}(d)}{(3d-1)!} q^d t_2^{3d-1} e^{dt_1},
\end{equation}
where $\mc{N}(d)$ is the number of rational (i.e. genus zero) holomorphic curves of degree $d$ in $\CP^2$ passing through $3d-1$ points in general position.
\item The numbers $\mc{N}(d)$ satisfy $\mc{N}(1)=1$ and the recurrence relation
\begin{equation}\label{l40 KM thm eq2}
\mc{N}(d)=\sum_{k+l=d}\mc{N}(k)\mc{N}(l) k^2 l \left( 
l \left(  
\begin{array}{c}
3d-4 \\ 3k-2
\end{array}
\right)
-k \left(  
\begin{array}{c}
3d-4 \\ 3k-1
\end{array}
\right)
\right)
\end{equation}
for $d\geq 2$. These two properties define the numbers $\mc{N}(d)$ completely. In particular, the first numbers are:
\begin{equation}
\begin{array}{c|cccccc} 
d& 1&2&3&4&5 &\cdots\\ \hline
\mc{N}(d) &1&1&12&620& 87304&\cdots
\end{array}
\end{equation}
\end{enumerate}
\end{thm}

In particular $\mc{N}(1)=1$ is the number of degree $1$ curves (lines) in $\CP^2$ through $2$ (generic) points, $\mc{N}(2)=1$ is the number of conics through $5$ points, $\mc{N}(3)=12$ is the number of \emph{rational} cubics through $8$ points,\footnote{One can find a cubic through $9$ points in general position,
%but in general position
but it will (in general position) have genus one, not zero.} etc.

The term $-q\frac{t_2^2}{2}$ in (\ref{l40 KM thm eq1}) is inconsequential, it cancels a similar term with the opposite sign present in the sum over $d$; it is put there so that $\Phi$ does not have terms of degree $<3$ in $t$'s (cf. the stability condition (\ref{l39 stability}): we only consider GW invariants in genus zero for $n\geq 3$).

\begin{proof}[Sketch of proof]
(i) Consider the Gromov-Witten invariant
\begin{equation}\label{l40 KM proof eq1}
\GW_{0,n,d}(\underbrace{[1],\ldots,[1]}_{n_0},\underbrace{[\omega],\ldots, [\omega]}_{n_1}. \underbrace{[\omega^2],\ldots, [\omega^2]}_{n_2})
\end{equation}
for $n\geq 4$; we understand that $n=n_0+n_1+n_2$. The number (\ref{l40 KM proof eq1}) vanishes for $n_0> 0$ by the same argument as in (\ref{l40 CP^1 all GW}) for $l>0$. The balancing condition between the form degree and the dimension of the space over which it is integrated is
\begin{equation}
\underbrace{2n_1+4n_2}_{\mr{form\;degree}}= \underbrace{2(3(d+1)-1)}_{\dim_\RR \mr{Hol}}+\underbrace{2(n-3)}_{\dim_\RR\MM_{0,n}} \quad \Leftrightarrow \quad n_2=3d-1
\end{equation}
Denote 
\begin{equation}\label{l40 KM proof eq2}
\mc{N}(d)\colon= \GW_{0,3d-1,d}(\underbrace{[\omega^2],\ldots,[\omega^2]}_{3d-1})
\end{equation}
If we insert $n_1$ additional copies of the class $[\omega]$ (Poincar\'e dual to the class of a hyperplane $H\subset \CP^2$ of complex codimension $1$) into the Gromov-Witten invariant (\ref{l40 KM proof eq2}), the number (\ref{l40 KM proof eq2}) gets multiplied by $d^{n_1}$, since a curve of degree $d$ intersects the hyperplane $H$ $d$ times.

This analysis, together with the straightforward case $n=3$ results in the ansatz (\ref{l40 KM thm eq1}).

(ii) The recurrence relation (\ref{l40 KM thm eq2}) is an immediate consequence of the WDVV equation (\ref{l40 WDVV}), from substituting the ansatz (\ref{l40 KM thm eq1}) into it.

\end{proof}

\marginpar{Lecture 41,\\ 12/7/2022}

\subsection{Keel's theorem}
For a subset $S\subset \{1,\ldots,n\}$, let us denote by $D_S\in H_\bt(\ol\MM_{0,n})$ the homology class of Deligne-Mumford compactification stratum $\dd_{0,S}$ (\ref{l39 DM stratum I}) of the compactified moduli space $\ol\MM_{0,n}$. We will denote $S^c$ the complement of $S$ in $\{1,\ldots,n\}$.

\begin{thm}[Keel \cite{Keel}]
Homology of the moduli space $\ol\MM_{0,n}$ is generated by classes $D_S$ with $S$ subsets of $\{1,\ldots,n\}$ such that $|S|,|S^c|\geq 2$, modulo the following relations:
\begin{itemize}
\item  $D_S=D_{S^c}$.
\item For $i,j,k,l$ distinct,
\begin{equation}\label{l41 Keel's rel}
\sum_{i,j\in S,\; k,l\in S^c}  D_S = 
\sum_{i,k\in S,\; j,l\in S^c}  D_S = 
\sum_{i,l\in S,\; j,k\in S^c}  D_S 
\end{equation}
\item  $D_S\cap D_T=0$ unless  $S\subset T$ or $T\subset S$.
\end{itemize}
%\begin{align}
%D_S=D_{S^c} \\
%\mbox{for $i,j,k,l$ distinct, }\sum_{i,j\in S,\; k,l\in S^c}  D_S = 
%\sum_{i,k\in S,\; j,l\in S^c}  D_S = 
%\sum_{i,l\in S,\; j,k\in S^c}  D_S 
%\end{align}
In the relation (\ref{l41 Keel's rel}) the summation in the left term is over partitions of $\{1,\ldots,n\}$ into two subsets $S,S^c$ such that $S$ contains $i,j$ and $S^c$ contains $S^c$, and similarly for the middle and the right terms.
\end{thm}

\begin{example}
Consider the case $n=4$. Non-compactified moduli space $\MM_{0,4}=C_4(\CP^1)/PSL(2,\CC)$  can be identified with sphere with three punctures, $\CP^1\backslash\{0,1,\infty\}$ (fixing three of the marked points to $0,1,\infty$, the modulus is the position of the fourth point), cf. (\ref{l11 M_0,4}). Deligne-Mumford compactification fills the the three punctures with the compactification strata $\dd_{0,\{1,4\}}$, $\dd_{0,\{2,4\}}$, $\dd_{0,\{3,4\}}$ (configurations where $z_4$ approaches 
$z_1=\infty$, $z_2=0$ or $z_3=1$), see Figure \ref{l11 fig Mbar_0,4}. %\marginpar{conventions in Fig. 15 are a bit different.} 
The compactified moduli space $\ol\MM_{0,4}$ is just a sphere $\CP^1$ and all three Deligne-Mumford strata are in the same homology class -- the class of a point in $\CP^1$. Thus, one indeed has
\begin{equation}
D_{\{1,4\}}=D_{\{2,4\}}=D_{\{3,4\}}
\end{equation}
which is the Keel's relation (\ref{l41 Keel's rel}) for $n=4$. %$i=4,j=1,k=2,l=3$.
\end{example}

\subsection{Explanation of WDVV equation from Keel's theorem and factorization of GW classes}
\label{sss Keel}

Consider the moduli space $\ol\MM_{0,n+4}$ %Deligne-Mumford compacifiaction stratum
with marked points labeled $\{A,B,E,F,1,\ldots,n\}$. Fix $a,b,e,f\in \{1,\ldots, s\}$ a quadruple of basis elements in $H(X)$. Restricting the Gromov-Witten class to Deligne-Mumford compatification strata of $\ol\MM_{0,n+4}$, we obtain
\begin{multline}\label{l41 WDVV proof eq}
\sum_{ S\subset\{1,\ldots,n\}  %A,B\in S,\;E,F\in S^c 
} \int_{D_{S\cup\{A,B\}}} I_{0,n+4}(e_a, e_b,e_e,e_f,\underbrace{\beta,\ldots,\beta}_n) 
\underset{\mr{factorization\;}(\ref{l40 factorization})}{=} \\   =
\sum_{ S\subset\{1,\ldots,n\}  %A,B\in S,\;E,F\in S^c 
} \sum_{c,d} \int_{\ol\MM_{0,%|S|+3   
S\cup \{A,B,C\}
}}I_{0,|S|+3}(e_a,e_b,\underbrace{\beta,\ldots,\beta}_{|S|},e_c)\; h^{cd} 
\int_{\ol\MM_{0, %|S^c|+3
 S^c\cup \{D,E,F\}
}}I_{0,|S^c|+3}(e_e,e_f,\underbrace{\beta,\ldots,\beta}_{|S^c|},e_d)\\=
%\left(
%\begin{array}{c}
%n \\ n_1
%\end{array}
%\right)
\sum_{n_1+n_2=n} \frac{n!}{n_1! n_2!} \sum_{d_1,d_2\geq 0}\sum_{c,d} q^{d_1}\GW_{0,n_1+3}(e_a,e_b,\underbrace{\beta,\ldots,\beta}_{n_1},e_c) h^{cd} q^{d_2}\GW_{0,n_2+3}(e_e,e_f,\underbrace{\beta,\ldots,\beta}_{n_2},e_d)
\end{multline}
In this computation we called $C,D$ the nodal point seen as a marked point on the two components of the curve.
Note that by Keel's relation (\ref{l41 Keel's rel}), expression (\ref{l41 WDVV proof eq}) doesn't change if we switch $A\leftrightarrow E$ and $a\leftrightarrow e$: under this switch, both the cohomology class in the l.h.s. and the homology class $\sum_{S}D_{S\cup\{A,B\}}$ it is paired with are invariant -- the former trivially and the latter by Keel's theorem. %(then we are pairing the same cohomology class with ).

Summing (\ref{l41 WDVV proof eq}) over $n\geq 0$ with weight $\frac{1}{n!}$, we obtain the l.h.s. of the WDVV equation (\ref{l40 WDVV}). Switching $a\leftrightarrow e$ (which doesn't change the expression by the argument above), we obtain the r.h.s. of WDVV.

%\textcolor{red}{PICTURE: nodal curve with two groups of points corresponding to (\ref{l41 WDVV proof eq}).}
\begin{figure}[H]
$$
\vcenter{\hbox{ \includegraphics[scale=1.25]{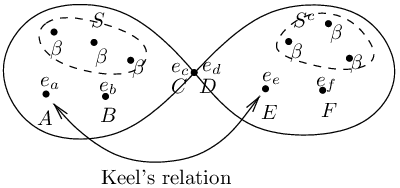}
}}
$$
\caption{A nodal curve with two groups of marked points corresponding to (\ref{l41 WDVV proof eq}).
}
\end{figure}

\section{A-model}\label{ss A-model}\label{s A-model}
Roughly speaking, the A-model is the nonlinear sigma model, with fields being maps into a target manifold $X$, with some extra fields (fermions) adjoined, so that the path integral of the model ``localizes'' -- in the sense that will be explained below -- to just the \emph{holomorphic} maps into the target.

For details on the A-model we refer to Witten's original papers \cite{Witten88,Witten_mirror}. For the viewpoint on the A-model as calculating the Euler class of a vector bundle over the mapping space whose section is the holomorphicity equation, see \cite{Atiyah-Jeffrey}.

Fix a Riemannian surface $\Sigma$ and a target K\"ahler manifold $X$. We will assume that the K\"ahler symplectic form $\omega$ on $X$ has integral periods. 
%For a map $\phi\colon \Sigma \ra X$, we call $d=\int_\Sigma \phi^*\omega\in \ZZ$ the degree of $\phi$.

We will use local complex coordinates %$x^i,x^{\bar{i}}$ 
on the target: holomorphic coordinates  $x^i$ and antiholomorphic coordinates $x^{\bar{i}}$;
 we will denote the real coordinates on the target $x^I$ (equivalently, one may think of $x^I$ as holomorphic and antiholomorphic coordinates jointly).
The action functional of the $A$-model is\footnote{
We put the normalization factor $\frac{1}{2\pi}$ in the action, so that  its bosonic part yields (for a flat target) the standard free boson propagator $\langle  \phi^I(w)\phi^J(z)\rangle= -g^{IJ}2\log|w-z|+\mr{const}$.
}
\begin{equation}\label{l41 S}
S
%\int_\Sigma \frac12 (d\phi,*d\phi)_g + \langle  \psi^{(1,0)}, \ol{\mathbf{D}} \chi\rangle  + \langle  \psi^{(0,1)}, \mathbf{D} \chi \rangle + \frac12 \langle R, \psi^{(1,0)}\otimes \psi^{(0,1)}\otimes \chi \otimes \chi \rangle
%\\
=\frac{1}{2\pi}\int_\Sigma \frac{i}{2} 
 g_{IJ}\bdd\phi^I  \bar\bdd \phi^J+ \psi^{(1,0)}_i \ol{\mathbf{D}} \chi^i - \psi^{(0,1)}_{\bar{i}} \mathbf{D} \chi^{\bar{i}} +i R^{i\bar{i}}_{\;\;j\bar{j}} \psi^{(1,0)}_i \psi^{(0,1)}_{\bar{i}} \chi^j \chi^{\bar{j}}
\end{equation}
The fields are
\begin{itemize}
\item A smooth map $\phi\colon \Sigma \ra X$.
\item An odd (anticommuting) field 
\begin{equation}
\chi \in \Gamma(\Sigma,\phi^* TX).
\end{equation}
\item Odd $(1,0)$- and $(0,1)$-form fields
\begin{equation}
\psi^{(1,0)}\in \Omega^{1,0}(\Sigma, \phi^* (T^{1,0})^*X),\quad
\psi^{(0,1)}\in \Omega^{0,1}(\Sigma, \phi^* (T^{0,1})^*X).
\end{equation}
\end{itemize}
One can assign $\ZZ$-grading to fields (ghost number): 
\begin{equation}
\mr{gh}(\phi)=0,\;\; \mr{gh}(\chi)=1,\;\; \mr{gh}(\psi^{1,0})=\mr{gh}(\psi^{0,1})=-1.
\end{equation}
In the action (\ref{l41 S}), $g=g(\phi)$ is the %K\"ahler
 Riemannian metric on the target $X$ pulled back to $\Sigma$ by the map $\phi$;
\begin{equation}
\ol{\mathbf{D}}\chi^i= \bar\bdd \chi^i+\Gamma^i_{jk}(\phi) \bar\bdd \phi^j \chi^k,\quad 
\mathbf{D}\chi^{\bar{i}}= \bdd \chi^{\bar{i}}+\Gamma^{\bar{i}}_{\bar{j}\bar{k}}(\phi) \bdd \phi^{\bar{j}} \chi^{\bar{k}}
\end{equation}
are the Dolbeault operators on $\Sigma$ twisted by the pullback of the Levi-Civita connection $\nabla_{LC}$ on $X$; $R=R(\phi)$
%$R^{i\bar{i}}_{\;\;j\bar{j}}$ 
is the pullback of the Riemann curvature tensor on $X$ to $\Sigma$. 

Using local complex coordinates on $\Sigma$, we can write the 1-form fields as
\begin{equation}
\psi^{(1,0)}_i=dz\,\psi_i,\quad \psi^{(0,1)}_{\bar{i}}=d\bar{z}\,\psi_{\bar{i}}.
\end{equation}
Then the action (\ref{l41 S}) can be written as
\begin{equation}\label{l41 S in components}
S=\frac{1}{\pi}\int_\Sigma d^2z \frac12 g_{IJ} \dd\phi^I \bar\dd\phi^J + i\psi_i \ol{D} \chi^i+i\psi_{\bar{i}}D \chi^{\bar{i}} - 
R^{i\bar{i}}_{\;\;j\bar{j}} \psi_i \psi_{\bar{i}}\chi^j \chi^{\bar{j}},
\end{equation}
with $D,\ol{D}$ the covariant derivatives on $\Sigma$ in the directions $\dd_z,\dd_{\bar{z}}$, w.r.t. the pullback of the target Levi-Civita connection.

The first term of the action (\ref{l41 S}) is the action of a sigma-model with target $X$ seen as a Riemannian manifold; one can rewrite it as
\begin{equation}\label{l41 rewriting S_sigma}
\frac{1}{2\pi}\int_\Sigma \frac{i}{2} g_{IJ} \bdd \phi^I \bar\bdd \phi^J =
\frac{1}{2\pi}\int_\Sigma i g_{\bar{i}j} \bdd \phi^{\bar{i}} \bar\bdd \phi^{j} +
\underbrace{\frac{1}{4\pi}
\int_{\Sigma} \phi^* \omega}_{S_\mr{top}}
 %\int_\Sigma  g_{\bar{i}j} d \phi^{\bar{i}} d \phi^{j}
\end{equation}
Here $\omega=i g_{i\bar{j}}dx^i\wedge dx^{\bar{j}}$ is the K\"ahler symplectic form on $X$. The last term in the r.h.s. of (\ref{l41 rewriting S_sigma}) is ``topological'': it depends only on the homotopy class of the map $\phi$ (and the cohomology class of $\omega$). In particular, $S_\mr{top}$ is a locally constant function on the space of fields.

The space of fields is equipped with a degree $-1$ odd  derivation $Q$ acting by
\begin{equation}
\begin{gathered}
Q\phi^I=\chi^I,\;\; Q\chi^I=0,\\ Q\psi^{(1,0)}_i=-ig_{i\bar{j}} \bdd\phi^{\bar{j}} +\Gamma^k_{ij} \chi^j \psi^{(1,0)}_k,\\
Q\psi^{(0,1)}_{\bar{i}}=-i g_{\bar{i}j} \bar\bdd\phi^{j} +\Gamma^{\bar{k}}_{\bar{i}\bar{j}} \chi^{\bar{j}} \psi^{(0,1)}_{\bar{k}}.
\end{gathered}
\end{equation}

The operator $Q$ squares to zero modulo equations of motion,\footnote{More precisely, here and in (\ref{l41 S simEL Q(...)}), we only need the part of the Euler-Lagrange equations arising as the variation of $S$ w.r.t. fields $\psi^{(1,0)},\psi^{(0,1)}$. These equations read
$\ol{\mathbf{D}}\chi^i + i R^{i\bar{i}}_{\;\; j\bar{j}} \psi^{(0,1)}_{\bar{i}}\chi^j\chi^{\bar{j}}=0$ and
$\mathbf{D}\chi^{\bar{i}} - i R^{i\bar{i}}_{\;\; j\bar{j}} \psi^{(1,0)}_{i}\chi^j\chi^{\bar{j}}=0$.
}
\begin{equation}
Q^2\simEL 0.
\end{equation}
 One can in fact massage the model (construct a ``first-order'' action) to make $Q$ square to zero on the nose, see Section \ref{sss A-model 1st order}.

The crucial property of the action (\ref{l41 S}) is that it is $Q$-exact, up to the topological term:
\begin{equation}\label{l41 S simEL Q(...)}
S\simEL S_\mr{top}+Q (R)
%\left(\int_\Sigma -\frac12 \psi^{(1,0)}_i\bar\bdd \phi^i + \frac12 \psi^{(0,1)}_{\bar{i}} \bdd \phi^{\bar{i}} %\mr{c.c.}
%\right).
\end{equation}
where
\begin{equation}
R=\frac{1}{4\pi}\int_\Sigma -\psi^{(1,0)}_i\bar\bdd \phi^i +  \psi^{(0,1)}_{\bar{i}} \bdd \phi^{\bar{i}}
\end{equation}
%where $\cdots=\int_\Sigma \psi^{(1,0)}_i\bar\bdd \phi^i + \mr{c.c.}$.
Again, the equality (\ref{l41 S simEL Q(...)}) is true only modulo equation of motion but becomes true everywhere on the space of fields in the version of Section \ref{sss A-model 1st order}.

\begin{remark} In the language of TCFT, the operator $Q$ is given by integrating around a field the conserved current $\mathbb{J}=\bJ+\ol\bJ$, cf. (\ref{l33 Q = action of J}), where
%\marginpar{normalizations for $J,G$?}
\begin{equation}\label{l41 J, Jbar}
\bJ=g_{i\bar{j}}\chi^i \bdd \phi^{\bar{j}},\quad
\ol\bJ=g_{\bar{i}j} \chi^{\bar{i}} \bar\bdd \phi^j.
\end{equation}
%In the A-model 
The currents $\bJ,\ol\bJ$ are conserved separately: $\bar\bdd \bJ\simEL 0$, $\bdd \ol\bJ\simEL 0$.

The fields $G$, $\ol{G}$  -- the $Q$-primitives of the components of the stress-energy tensor (\ref{l33 T=Q(G)}) are given by
\begin{equation}\label{l41 G, Gbar}
G (dz)^2=-i\psi^{(1,0)}_i\bdd \phi^i,\quad
\ol{G} (d\bar{z})^2 = -i\psi^{(0,1)}_{\bar{i}} \bar\bdd \phi^{\bar{i}}.
\end{equation}

The stress-energy tensor itself is
\begin{equation}
\begin{gathered}
T(dz)^2=Q(G)(dz)^2=-g_{i\bar{j}}\bdd \phi^i \bdd \phi^{\bar{j}}+
i\psi^{(1,0)}_i\mathbf{D}\chi^i,\\
%-i \Gamma^k_{ij}\bdd\phi^i \chi^j \psi^{(1,0)}_k
\ol{T}(d\bar{z})^2=Q(\ol{G})(d\bar{z})^2=-g_{i\bar{j}}\bar\bdd \phi^i \bar\bdd \phi^{\bar{j}}+
i\psi^{(0,1)}_{\bar{i}}\ol{\mathbf{D}}\chi^{\bar{i}}.
\end{gathered}
\end{equation}
\end{remark}

\begin{remark} The A-model is described by somewhat lengthy formulae due to the involvement of target geometry. For a flat target all formulae simplify drastically. E.g., the action (\ref{l41 S}) becomes simply a free (quadratic) action
\begin{equation}\label{l41 A flat}
S=S_\mr{top}+\frac{1}{2\pi}\int_\Sigma i g_{\bar{i}j} \bdd \phi^{\bar{i}} \bar\bdd \phi^j+\psi^{(1,0)}_i \bar\bdd \chi^i- \psi^{(0,1)}_{\bar{i}}\bdd \chi^{\bar{i}},
\end{equation}
with $g_{\bar{i}j}$. In fact there is a very interesting class of cases where the target is compact and admits a flat metric  everywhere except for finitely many points -- toric manifolds $X$. In this case one can study the A-model as a free theory with special observables corresponding to the preimages of the special points in $X$ where the metric is singular. This approach is due to Frenkel-Losev  \cite{Frenkel-Losev}.
\end{remark}

\begin{remark}\label{l41 rem: A with flat target = bosons+ghosts}
For the target $X=\CC^n$ with standard 
K\"ahler structure, the A-model (\ref{l41 A flat}) becomes a system of $n$ free complex bosons and $n$ simple ghost systems  (cf. Remark \ref{l32 rem: bc system with j-parameter}). Note that in this case it is clear that the central charge of the system is $2n+(-2)n=0$. 

In fact the central charge of the A-model is zero for any target, by the general argument (\ref{l33 c=0}) holding in any TCFT.
\end{remark}

\begin{remark} The space of states of the A-model on a circle (or equivalently the space of quantum fields $V$) can be identified with the de Rham complex of the free loop space of $X$,
\begin{equation}
V=\HH_{S^1}=\Omega^\bt(LX),
\end{equation}
with differential $Q$ being the de Rham differential $d_{LX}$.
Put another way, states on $S^1$ can be realized as functions of fields $\phi^I,\chi^I$ on $S^1$.%\marginpar{Is it always true that $H^\bt(LX)\simeq H^\bt(X)$ for $X$ K\"ahler?}
The $Q$-cohomology of $V$ is then 
\begin{equation}
H_Q(V)=H^\bt(LX)=H^\bt(X)\oplus (\cdots) %H^\bt(\Omega X)
\end{equation}
where the first term on the right corresponds to constant loops. 
%$\Omega X$ is the space of based loops. 
One can identify $H^\bt(X)\subset H_Q(V)$ %the first term in the 
as evaluation observables (Section \ref{sss evaluation observables}). 
\end{remark}

\subsection{Path integral heuristics: independence on the target geometric data.}
The fact (\ref{l41 S simEL Q(...)}) leads to the following expectation about the A-model path integral: the correlator of any collection of $Q$-closed observables $\Phi_1,\ldots,\Phi_n$ should be invariant under deformations of the geometric data on the target, except for the possible change of the topological term. More precisely, one 
%expects that for $(g_t,J_t,\omega_t)$ a family of K\"ahler data on $X$ with parameter $t$, one 
can split the correlator into contributions of different homotopy classes of the map $\phi\colon \Sigma \ra X$:
\begin{multline}\label{l41 corr}
\langle \Phi_1(z_1)\cdots \Phi_n(z_n) \rangle = \int_{\F} e^{-S} \Phi_1(z_1)\cdots \Phi_n(z_n) =\\
= \sum_{[\phi]\in [\Sigma,X]} e^{-S_\mr{top}([\phi])}  \int_{\F_{[\phi]}} e^{-Q(R)} \Phi_1(z_1)\cdots \Phi_n(z_n)=\\
=
\sum_{[\phi]\in [\Sigma,X]} e^{-S_\mr{top}([\phi])} \langle \Phi_1(z_1)\cdots \Phi_n(z_n) \rangle_{[\phi]}
\end{multline}
where $[\Sigma,X]$ is the set of homotopy classes of maps. Then the expectation is that for $Q$-closed observables $\Phi_i$ and for a path $(g_t,J_t,\omega_t)$ of K\"ahler data on $X$ with parameter $t$, 
the contribution 
\begin{equation}\label{l41 corr in a homotopy class}
\langle \Phi_1(z_1)\cdots \Phi_n(z_n) \rangle_{[\phi]}
\end{equation}
of a homotopy class into the correlator (\ref{l41 corr})
does not depend on $t$. 

The logic is that one differentiates the path integral over a given homotopy class in (\ref{l41 corr}) in the parameter $t$ of the family which results in a $Q$-exact expression (modulo Euler-Legrange equations) under the path integral; such expressions are expected to have zero averages over the space of fields.

%\marginpar{EDIT}
\begin{remark}
Later, in Section \ref{sss evaluation observables}, we will be discussing evaluation observables which are not $Q$-closed, but rather $Q$-closed up to a $d$-exact ``error term.'' The argument above goes through for them, with the caveat that the correlators  (\ref{l41 corr in a homotopy class}) for them change by a $d$-exact term under the deformation of target geometric data.
\end{remark}

\subsection{A-model as an integral representation for the delta-form on holomorphic maps}

If we rescale the target metric $g\ra \frac{1}{\epsilon}g$ with $\epsilon$ a constant, the action becomes
\begin{equation}
S^\epsilon= \underbrace{\frac{1}{4\pi\epsilon}\int_\Sigma \phi^* \omega}_{S_\mr{top}^\epsilon} + 
\underbrace{\frac{1}{2\pi}\int_\Sigma \underbrace{\frac{i}{\epsilon} g_{\bar{i}j} \bdd \phi^{\bar{i}} \bar\bdd \phi^{j}}_{(I)}+
\psi^{(1,0)}_i \ol{\mathbf{D}} \chi^i - \psi^{(0,1)}_{\bar{i}} \mathbf{D} \chi^{\bar{i}} +i\epsilon R^{i\bar{i}}_{\;\;j\bar{j}} \psi^{(1,0)}_i \psi^{(0,1)}_{\bar{i}} \chi^j \chi^{\bar{j}}}_{S'_\epsilon}
\end{equation}
In the limit $\epsilon\ra 0$ the dominating term (I) in the action essentially enforces the constraint $\bar\bdd\phi^i=0$, i.e., enforces the holomorphicity property of the map $\phi\colon \Sigma \ra X$.

More precisely, integrating out fields $\psi^{(1,0)},\psi^{(0,1)}$, we obtain a cohomologically smeared delta-form on the space of smooth maps $\Sigma\ra X$ supported on holomorphic maps:
\begin{equation}\label{l41 smeared delta-fun on hol maps}
\int \mc{D} \psi^{(1,0)}\mc{D}\psi^{(0,1)}\; e^{-S'_\epsilon} = 
\delta^\epsilon_{\mr{Hol}(\Sigma,X)} \quad \in \Omega(\mr{Map}(\Sigma,X)).
\end{equation}
In this identification, one identifies the field $\chi^I$ as
\begin{equation}
\chi^I=d_\mr{Map} \phi^I \quad \in T^*_\phi \mr{Map}(\Sigma,X)
\end{equation}
-- a 1-form/covector %de Rham differential 
on the mapping space; $d_\mr{Map}$ stands for de Rham operator on the mapping space. The parameter $\epsilon$ in (\ref{l41 smeared delta-fun on hol maps}) serves a ``smearing'' parameter, with $\epsilon\ra 0$ limit being the ``true'' (non-smeared) distributional delta-form.

\subsubsection{Prototype 
of a Mathai-Quillen representative.}
Given a function $f\colon M\ra \RR$ (assume that it is smooth, with nonvanishing differential on its zero-locus), one has the following cohomologically smeared delta-form on the hypersurface $f^{-1}(0)\subset M$:
\begin{equation}\label{l41 MQ prototype}
\delta^\epsilon_{f^{-1}(0)}=(2\pi \epsilon)^{-\frac12} e^{-\frac{f(x)^2}{2\epsilon}} df  \in \Omega_\mr{cl}^1(M)
\end{equation}
with $\epsilon>0$ a smearing parameter. In the limit $\epsilon\ra 0$ this form distributionally converges to true delta-form
 $\delta_{f^{-1}(0)}$.
%For any finite value of $\epsilon$ one has 
The form (\ref{l41 MQ prototype}) can be written as a Berezin integral over an auxiliary odd (anticommuting) variable $\psi$:
\begin{equation}
\delta^\epsilon_{f^{-1}(0)}= (2\pi \epsilon)^{-\frac12} \int D\psi \; e^{-\frac{f(x)^2}{2\epsilon}+\psi df
%\chi^i \dd_i f(x)
}.
\end{equation} 
%Here  $x^i$ are local coordinates on $M$ and $\chi^i\colon = dx^i$. 

More generally, for  $f\colon M\ra \RR^k$ a smooth function with surjective differential on $f^{-1}(0)$, the zero-locus is a submanifold of codimension $k$ and one has the following smeared delta-form on it:
\begin{equation}
\delta^\epsilon_{f^{-1}(0)}=(2\pi \epsilon)^{-\frac{k}{2}} \int \prod_{a=1}^k D\psi_a \; e^{-\frac{||f(x)||^2}{2\epsilon}+\psi_a 
df^a
%\chi^i \dd_i f^a(x)
}    \in \Omega^k_\mr{cl}(M),
\end{equation}
where we introduced $k$ auxiliary odd variables $\psi^a$.

\subsubsection{Mathai-Quillen representative of the Euler class of a vector bundle.}
Let $E\ra M$ be a real oriented vector bundle of even rank $k$ over a manifold $M$. Assume that $E$ is equipped with fiberwise metric $g$, a connection $\nabla$ compatible with the metric and a section $s\colon M\ra E$. 
Consider the following differential form:
\begin{multline}\label{l41 S_MQ}
S_{MQ}= \frac{1}{2\epsilon} g(s,s) + i\langle \psi, \nabla s \rangle -\frac{\epsilon}{2} \langle \psi, F_\nabla(g^{-1}(\psi)) \rangle=\\
=\frac{1}{2\epsilon} g_{ab}s^a s^b +i\psi_a (ds^a + A^a_{\;\;b} s^b)-\frac{\epsilon}{4}  g^{bc}F^a_{\;\;c} \psi_a \psi_b \quad \in \Omega^\bt(M,\wedge^\bt E).
\end{multline}
Here we think of odd variables $\psi_a$ as generators of the exterior algebra of the fiber, $\wedge^\bt E_x$ (put another way $\psi_a$ are coordinates on the parity-reversed dual fiber $\Pi E_x^*$). In the second line we rewrote $S_{MQ}$ explicitly in a local trivialization of $E$; $A^a_{\;\; b}$ are the components of the local connection 1-form, $F_\nabla\in \Omega^2(M,\mr{End}(E))$ is the curvature 2-form of the connection and $F^a_{\;\;c}\in \Omega^2(M)$  are its components. The smearing parameter $\epsilon$ in (\ref{l41 S_MQ}) corresponds to scaling the fiber metric $g\mapsto \frac{1}{\epsilon} g$. 

Even more epxlicitly, using local coordinates $x^i$ on $M$, (\ref{l41 S_MQ}) can be written as
\begin{equation}
S_{MQ}=\frac{1}{2\epsilon} g_{ab}s^a s^b +i\psi_a (\dd_i s^a + A^{\;a}_{i\;b} s^b)\chi^i-\frac{\epsilon}{4}  g^{bc}F^{\;\;a}_{ij\;\;c} \psi_a \psi_b  \chi^i \chi^j,
\end{equation}
where we denoted $\chi^i\colon= dx^i$.

Consider the fiber Berezin integral
\begin{equation}\label{l41 Xi}
\Xi= \left(\frac{i}{\sqrt{2\pi \epsilon}}\right)^k\int_{\mr{fiber\;of\;}\Pi E^*\ra M} D\psi\; e^{-S_{MQ}}\quad \in \Omega^k(M).
\end{equation}
Here $D\psi\in \Gamma(M,\wedge^k E^*)$ is the fiber Berezinian (fermionic integration measure) induced from the fiber metric and the orientation of the fiber.

\begin{thm}[Mathai-Quillen \cite{Mathai-Quillen}]\label{l41 thm MQ}
\begin{itemize}
\item Form $\Xi$ is closed.
\item Changing the data $s,g,\nabla,\epsilon$ changes $\Xi$ by an exact form, $\Xi\mapsto \Xi+d(\cdots)$.
\item The class of $\Xi$ in de Rham cohomology $H^k(M)$ is the Euler class of the bundle $E\ra M$.\footnote{
Recall that for rank $k$ oriented real vector bundle $E$ over a closed  manifold $M$,  the Euler class $e$ is the cohomology class of $M$ Poincar\'e dual to the homology class of the zero-locus of a generic section $s\colon M\ra E$ (``generic'' here means ``transversal to the zero-section''). More precisely (to take signs into account), ``zero-locus'' should be understood as the intersection of the graph of $s$ with the graph of the zero-section. An equivalent definition: consider the Thom class of $E$ -- the cohomology class of the total space $\tau\in H^k(E)$ with the property that its pushforward to $M$ by the bundle projection is the constant function $1$. Then the Euler class  is the pullback $e=s^*\tau$ of the Thom class by an (arbitrary) section $s:M\ra E$ (here one doesn't need a transversality condition).
}
\item If the section $s$ intersects the zero-section of $E$ transversally, then one has
\begin{equation}
\lim_{\epsilon\ra 0} \Xi = \delta_{s^{-1}(0)}
\end{equation}
where the limit is understood in distributional sense.
\end{itemize}
\end{thm}

In particular, the form (\ref{l41 Xi}) is a cohomologically smeared delta-form on the zero-locus of the section $s$; $\Xi$ is known as the Mathai-Quillen representative of the Euler class of the bundle $E\ra M$.

Mathai-Quillen construction has a straightforward modification to complex vector bundles equipped with hermitian fiber metric.

\begin{remark} In the limit $\epsilon\ra \infty$, the last term in (\ref{l41 S_MQ}) is dominating and the formula (\ref{l41 Xi}) becomes the Gaussian integral over the odd variable $\psi$. The latter yields the representative for the Euler class of the bundle as a Pfaffian of the curvature 2-form,
\begin{equation}
\Xi=\mr{Pf}\left(\frac{1}{2\pi}F_\nabla\right)\quad \in \Omega^k(M).
\end{equation}
This is the Chern-Weil representative of the Euler class. In the special case when $E=TM$ is the tangent bundle and $\nabla$ is the Levi-Civita connection, integrating $\Xi$ over $M$ one obtains the Chern-Gauss-Bonnet theorem,
\begin{equation}
\chi(M)=\int_M \mr{Pf}\left(\frac{1}{2\pi} R %F_{\nabla_{LC}}
\right),
\end{equation}
where the l.h.s. is the Euler characteristic and $R=F_{\nabla_{LC}}\in \Omega^2(M,\mr{End}(TM))$ is the Riemann curvature tensor.
\end{remark}

\subsubsection{A-model as a Mathai-Quillen representative.}
Consider the vector bundle $\mc{E}$ over the space of smooth maps $M=\mr{Map}(\Sigma,X)$ where the fiber over the map $\phi$ is 
\begin{equation}\label{l41 E_phi}
E_\phi= \Omega^{0,1}(\Sigma,\phi^*T^{1,0}X)
\end{equation}
The bundle $E$ is equipped with:
\begin{itemize}
\item A natural section $s=\bar\bdd\colon M\ra E$. Note that the zero-locus of this section is the submanifold of holomorphic maps inside smooth maps, $\mr{Hol}(\Sigma,X)\subset \mr{Map}(\Sigma,X)$.
\item A natural fiber hermitian metric given by $\langle \xi,\rho\rangle = \int_\Sigma g( \xi \stackrel{\wedge}{,} \bar\rho)$ for $\xi,\rho\in E_\phi$, with $g$ the metric on the target.
\item A connection compatible with fiber metric, induced from Levi-Civita connection on $X$.
\end{itemize}

Comparing (\ref{l41 Xi}) and the l.h.s. of (\ref{l41 smeared delta-fun on hol maps}), we observe that the integral over the field $\psi$ in the A-model can be formally identified with the Mathai-Quillen representative of the Euler class of the vector bundle (\ref{l41 E_phi}) over the space of smooth maps, or, put differently, with the cohomologically smeared delta-form on the cycle of holomorphic maps inside smooth maps.

%Comparing the integral over the field $\psi$ in the A-model (\ref{l41 smeared delta-fun on hol maps}) with (\ref{l41 Xi}), we can iden

\subsection{Evaluation observables}\label{sss evaluation observables}
Consider the evaluation map
\begin{equation}\label{l41 ev}
\mr{ev}\colon \Sigma \times \mr{Map}(\Sigma,X) \ra X
\end{equation}

Given a differential form $\alpha$ on $X$
%, written in local coordinates as
\begin{equation}
\alpha=\alpha_{I_1\cdots I_p}(x)dx^{I_1}\cdots dx^{I_p} \in \Omega^p(X),
\end{equation}
%consider its pullback by $\mr{ev}$:
%\begin{equation}
%\mr{ev}^* \alpha \in \Omega^\bt(\mr{Map}(\Sigma,X))\otimes \Omega^\bt(\Sigma)\subset C^\infty(\F_\Sigma)\otimes \Omega^\bt(\Sigma)
%\end{equation}
one defines the corresponding \emph{evalutation observable}\footnote{Here we think of $\tilde{O}_\alpha$ as an observable in the sense of classical field theory, which can then be put into the path integral. Tilde in the notation refers to the fact that it is a nonhomogeneous form on $\Sigma$ which we will in a moment identify as a total descendant (\ref{l33 total descendant}), for $\alpha$ closed.
} $\til{O}_\alpha(z)$ at a point $z\in \Sigma$ as
\begin{multline}
\til{O}_\alpha(z)\colon =\mr{ev}^* \alpha|_z = \alpha_{I_1\cdots I_p}(\phi)(\chi^{I_1}+d\phi^{I_1})\cdots (\chi^{I_p}+d\phi^{I_p})\big|_z \\    \in \Omega^\bt(\mr{Map}(\Sigma,X))\otimes \wedge^\bt T^*_z\Sigma \;\; \subset C^\infty(\F_\Sigma)\otimes \wedge^\bt T^*_z\Sigma
\end{multline}
Thus, $\til{O}_\alpha$ is a  form on $\Sigma$ depending on field configuration, or more specifically on the fields $\phi$, $\chi=d_\mr{Map}\phi$ and first derivatives of $\phi$ at the point $z$. Evaluation observable can be split according to the de Rham degree on $\Sigma$,
\begin{equation}
\til{O}_\alpha = O_\alpha^{(0)}+O_\alpha^{(1)}+O_\alpha^{(2)}.
\end{equation}

The following is checked by a direct computation.
\begin{lemma}
Evaluation observables satisfy the following properties:%\marginpar{I have trouble with the sign convention for $Q$ here..}
\begin{eqnarray}
(d+Q)\til{O}_\alpha &=& \til{O}_{d_X\alpha},\\
Q O_\alpha^{(0)} &=& O^{(0)}_{d_X \alpha}, 
%\\ \til{O}_\alpha &=& e^\Gamma O_\alpha^{(0)},
\end{eqnarray}
where  $d,d_X$ are the de Rham differentials on the source and the target, respectively. 
%$\Gamma$ is the descent operator (\ref{l33 Gamma}).
\end{lemma}
In particular, if $\alpha$ is a \emph{closed} form on $X$, then $O_\alpha^{(0)}$ is $Q$-closed and $\til{O}_\alpha$ is $(d+Q)$-closed and is the total descendant of $O_\alpha^{(0)}$, cf. (\ref{l33 total descendant}).

\begin{remark}
It is natural to identify the total de Rham differential on 
%the l.h.s. of
$\Sigma\times \mr{Map}(\Sigma,X)$ with $d+Q$, rather than $d-Q$. Thus, in this section we are using a different sign convention than in Section \ref{ss TCFT} for  the descent equations (\ref{l33 descent eq}), (\ref{l33 (d-Q) til Phi =0}): $(d+Q)\til{O}=0$, or $dO^{(k-1)}=-QO^{(k)}$.
\end{remark}

\begin{remark} One \emph{does not} have the equality $\til{O}_\alpha=e^\Gamma O_\alpha^{(0)}$ with $\Gamma$ the descent operator associated with the $G,\ol{G}$ field (\ref{l41 G, Gbar}). I.e., the evaluation observable is not the \emph{canonical} total descendant of its $0$-form component, in the sense of (\ref{l33 canonical total descendant}). However, one can consider adjusting the $Q$-primitive of the total stress-energy tensor by a $Q$-exact term
\begin{multline}
G^\mr{tot}=G(dz)^2 + \ol{G} (d\bar{z})^2  \mapsto \\
\mapsto
G^{'\mr{tot}} = G(dz)^2 + \ol{G} (d\bar{z})^2 + Q(g^{i\bar{i}}\psi^{(1,0)}_i \psi^{(0,1)}_{\bar{i}}) =
-i \psi_i^{(1,0)} d\phi^i - i  \psi^{(0,1)}_{\bar{i}}d\phi^{\bar{i}}.
\end{multline}
Then, denoting by $\Gamma'$ the associated modified descent operator, given by integrating $G^{'\mr{tot}}$ around a field, one has for the evaluation observable the equality
\begin{equation}
\til{O}_\alpha= e^{\Gamma'} O_\alpha^{(0)}.
\end{equation}
\end{remark}

\subsubsection{Gromov-Witten classes as correlators of evaluation observables.}
Consider for simplicity the case $\Sigma=\CP^1$.
Given a collection of closed forms on the target, $\alpha_1,\ldots, \alpha_n \in \Omega_\mr{cl}(X)$, the correlator of the corresponding evaluation observables in the path integral formalism is
\begin{multline}\label{l41 corr of evaluation observables}
\langle \til{O}_{\alpha_1}\cdots \til{O}_{\alpha_n} \rangle =
\int \mc{D}\phi \mc{D} \chi \int \mc{D} \psi\; e^{-S} \til{O}_{\alpha_1}\cdots \til{O}_{\alpha_n} =\\
 \underset{(\ref{l41 smeared delta-fun on hol maps})}{=}
\sum_{d\geq 0} e^{-\frac{d}{4\pi\epsilon}}
\int_{\mr{Map}_d(\Sigma,X)} \delta^\epsilon_{\mr{Hol}(\Sigma,X)} \til{O}_{\alpha_1}\cdots \til{O}_{\alpha_n} =\\
=
\sum_{d\geq 0} q^d \left(\int_{\mr{Hol}_d(\Sigma,X)} \pi_1^* \mr{ev}^*\alpha_1\wedge \cdots \wedge \pi_n^* \mr{ev}^*\alpha_n+d\Big(\cdots\Big)\right)\quad \in \Omega_\mr{cl}(\ol{C}_n(\Sigma)).
\end{multline}
Here in the second line, the prefactor is the exponential of the topological term in the action, $e^{-S_\mr{top}}$, evaluated on maps of degree $d$ (defined by (\ref{l39 d})); we also identify this prefactor as $q^d$ with \begin{equation}\label{l41 q via epsilon}
q\colon= e^{-\frac{1}{4\pi\epsilon}}.
\end{equation} 
In the second step in (\ref{l41 corr of evaluation observables}) we consider the limit $\epsilon\ra 0$ in the path integral over $\mr{Map}_d(\Sigma,X)$, which localizes the integral to holomorphic maps; however the change of $\epsilon$ induces a shift of the value of the integral by an exact form on the configuration space (since we are looking at a \emph{fiber} integral over $C_n(\Sigma)\times \mr{Map}(\Sigma,X)\ra C_n(\Sigma)$ of a closed form changed by an exact form -- such a change induces an exact change of the fiber integral). The cohomology class of the correlator (\ref{l41 corr of evaluation observables}) is the genus zero Gromov-Witten class  (\ref{l39 GW class genus 0}).

\subsection{A-model in the first-order formalism}\label{sss A-model 1st order}

The first-order action for the A-model is
\begin{multline}\label{l41 S 1st order}
S^{\tiny\mbox{first-order}}=S_{\mr{top}}(\phi)+ \\
+
\frac{1}{2\pi}\int_\Sigma -p^{(1,0)}_i \bar\bdd \phi^i + p^{(0,1)}_{\bar{i}}\bdd \phi^{\bar{i}} +i g^{i\bar{j}} p^{(1,0)}_i p^{(0,1)}_{\bar{j}} +\psi^{(1,0)}_i \ol{\mathbf{D}} \chi^i -\psi^{(0,1)}_{\bar{i}} \mathbf{D} \chi^{\bar{i}}+i R^{i\bar{i}}_{\;\; j\bar{j}} \psi^{(1,0)}_i \psi^{(0,1)}_{\bar{i}}\chi^j \chi^{\bar{j}} 
\end{multline}
with $S_\mr{top}(\phi)$ the topological term as in (\ref{l41 rewriting S_sigma}).
Here the fields are as in (\ref{l41 S}), plus two new ``momentum'' fields (even, of ghost number $0$):
\begin{equation}
p^{(0,1)}\in \Omega^{0,1}(\Sigma,\phi^* (T^{1,0})^*X),\quad p^{(1,0)}\in \Omega^{1,0}(\Sigma,\phi^* (T^{0,1})^*X).
\end{equation}
Integrating out the fields, $p^{(1,0)},p^{(0,1)}$, one obtains back the action (\ref{l41 S}):
\begin{equation}
\int \mc{D} p^{(1,0)}\, \mc{D} p^{(0,1)} e^{-S^{\tiny\mbox{first-order}}}= e^{-S} .
\end{equation}

The odd derivation $Q$ acts on fields of the first-order theory as
\begin{equation}\label{l41 Q 1st order}
\begin{gathered}
Q\phi^I=\chi^I,\;\; Q\chi^I=0,\\ Q\psi^{(1,0)}_i=p^{(1,0)}_i+\Gamma^k_{ij} \chi^j \psi_k^{(1,0)},\;\;
Q\psi^{(0,1)}_{\bar{i}}=p^{(0,1)}_{\bar{i}}+\Gamma^{\bar{k}}_{\bar{i}\bar{j}} \chi^{\bar{j}} \psi_{\bar{k}}^{(0,1)},\\
Q p^{(1,0)}_i= \Gamma^k_{ij}\chi^j p^{(1,0)}_k - R^j_{\;\;i k\bar{k}} \psi^{(1,0)}_j \chi^k\chi^{\bar{k}},\;\;
Q p^{(0,1)}_{\bar{i}}= \Gamma^{\bar{k}}_{\bar{i}\bar{j}}\chi^{\bar{j}} p^{(0,1)}_{\bar{k}} - R^{\bar{j}}_{\;\;\bar{i} k\bar{k}} \psi^{(0,1)}_{\bar{j}} \chi^k\chi^{\bar{k}}
\end{gathered}
\end{equation}
The operator $Q$ squares to zero
\begin{equation}\label{l41 1st order Q^2=0}
Q^2=0
\end{equation}
%(strictly, not modulo Euler-Lagrange equations). 
and one has
\begin{equation}\label{l41 1st order S=Q(...)}
S^{\tiny\mbox{first-order}}=S_\mr{top}+Q\left(\frac{1}{2\pi}\int_\Sigma -\psi^{(1,0)}_i \bar\bdd \phi^i+\psi^{(0,1)}_{\bar{i}}\bdd \phi^{\bar{i}} +\frac{i}{2} g^{i\bar{j}}\psi^{(1,0)}_i p^{(0,1)}_{\bar{j}} -\frac{i}{2} g^{\bar{i}j}\psi^{(0,1)}_{\bar{i}}p^{(1,0)}_j \right).
\end{equation}
Both equalities (\ref{l41 1st order Q^2=0}), (\ref{l41 1st order S=Q(...)})
hold strictly, not just modulo Euler-Lagrange equations.

The counterpart of currents (\ref{l41 J, Jbar}) in the first-order formalism is
\begin{equation}
\bJ=\chi^i p^{(1,0)}_i,\quad \ol\bJ=\chi^{\bar{i}} p^{(0,1)}_{\bar{i}},
\end{equation}
whereas formulae (\ref{l41 G, Gbar}) for $G,\ol{G}$ do not change.

\begin{remark} Scaling the metric  as $g\mapsto \frac{1}{\epsilon} g$ in the first-order action, one obtains
\begin{multline}
S^{\tiny\mbox{first-order}}_\epsilon=S^\epsilon_{\mr{top}}(\phi)+ \\
+
\underbrace{\frac{1}{2\pi}\int_\Sigma -p^{(1,0)}_i \bar\bdd \phi^i + p^{(0,1)}_{\bar{i}}\bdd \phi^{\bar{i}} + i\epsilon  g^{i\bar{j}} p^{(1,0)}_i p^{(0,1)}_{\bar{j}} +\psi^{(1,0)}_i \ol{\mathbf{D}} \chi^i -\psi^{(0,1)}_{\bar{i}} \mathbf{D} \chi^{\bar{i}}+i\epsilon R^{i\bar{i}}_{\;\; j\bar{j}} \psi^{(1,0)}_i \psi^{(0,1)}_{\bar{i}}\chi^j \chi^{\bar{j}} }_{S_\epsilon^{'\tiny\mbox{first-order}}}.
\end{multline}
Note that only the \emph{inverse} metric is involved in the first-order action (barring the topological term), so one can
take the limit $\epsilon\ra 0$:
\begin{equation}
\lim_{\epsilon\ra 0} S_\epsilon^{'\tiny\mbox{first-order}} =
\frac{1}{2\pi}\int_\Sigma -p^{(1,0)}_i \bar\bdd \phi^i + p^{(0,1)}_{\bar{i}}\bdd \phi^{\bar{i}} +\psi^{(1,0)}_i \ol{\mathbf{D}} \chi^i -\psi^{(0,1)}_{\bar{i}} \mathbf{D} \chi^{\bar{i}}.
\end{equation}
Here it is clear that fields $p$ play the role of Lagrange multipliers imposing the holomorphicity constraint $\bar\bdd\phi=0$; fields $\psi$ are the corresponding odd Lagrange multipliers imposing the associated constraint on the differential $\chi$ of a holomorpic map $\phi$.
\end{remark}

%\begin{remark} 
\subsubsection{``First-order'' Mathai-Quillen construction.}
The first-order A-model (\ref{l41 S 1st order}) can be seen as an example of the
%One has also a
 ``first-order'' variant of the Mathai-Quillen construction (\ref{l41 S_MQ}), (\ref{l41 Xi}):
\begin{multline}\label{l41 S_MQ 1st order}
S_{MQ}^{\tiny\mbox{first-order}}=\frac{\epsilon}{2}g^{-1}(p,p)-i\langle p,s \rangle+i\langle \psi,\nabla s \rangle -\frac{\epsilon}{2}\langle \psi, F_\nabla(g^{-1}(\psi)) \rangle=\\
=\frac{\epsilon}{2} g^{ab}p_a p_b-ip_a s^a  +i\psi_a (ds^a + A^a_{\;\;b} s^b)-\frac{\epsilon}{4}  g^{bc}F^a_{\;\;c} \psi_a \psi_b \quad \in \Omega(M,\wedge E\otimes \mr{Sym} E).
%\quad \in \Omega^\bt(M,\wedge^\bt E).
\end{multline}
where the new even momentum field $p_a$ is a coordinate on the dual of the fiber $E_x^*$. Denote $j\colon \Pi TM \ra M$ the bundle projection from the parity-reversed tangent bundle of $M$ to $M$ and denote $\mathbb{E}\colon=j^* (E\oplus \Pi E)$ -- a supervector bundle over $\Pi TM$. Then the action (\ref{l41 S_MQ 1st order}) is a function on the total space of $\mathbb{E}$.

One has the following:
\begin{itemize}
\item Integrating out the variable $p$, one gets the ``second order'' Mathai-Quillen action (\ref{l41 S_MQ}):
%\marginpar{fix normalization}
\begin{equation}
e^{-S_{MQ}}=\left(\frac{\epsilon}{2\pi}\right)^{k/2}\int dp\; e^{-S_{MQ}^{\tiny\mbox{first-order}}}.
\end{equation}
Thus, integrating out both $p$ and $\psi$, one obtains the Mathai-Quillen representative of the Euler class (\ref{l41 Xi}):
\begin{equation}
\Xi=(2\pi)^{-k}\int dp\;D\psi\; e^{-S_{MQ}^{\tiny\mbox{first-order}}}\in \Omega^k(M).
\end{equation}
\item One can introduce an odd derivation $Q$ on 
$C^\infty(\mathbb{E})=\Omega(M,\wedge E\otimes \mr{Sym} E)$ 
(i.e. functions of the variables $x,\chi=dx,\psi,p$) defined by 
\begin{equation}
\begin{gathered}
Q(x^i)=\chi^i,\;\; Q(\chi^i)=0,\;\; Q(\psi_a)=p_a+A_{i\;a}^{\;b}\chi^i\psi_b,\\ Q(p_a)=-F_{ij\;a}^{\;\;b}\chi^i \chi^j \psi_b + A_{i\;a}^{\;b}\chi^i p_b,
\end{gathered}
\end{equation}
cf. (\ref{l41 Q 1st order}).
Then one has 
\begin{equation}
Q^2=0,
\end{equation}
i.e., $Q$ is a cohomological vector field on $\mathbb{E}$.
The first-order Mathai-Quillen action is $Q$-exact:
\begin{equation}
S_{MQ}^{\tiny\mbox{first-order}}=Q(-i \langle \psi,s \rangle +\frac{\epsilon}{2} g^{-1}(\psi,p)).
\end{equation}
\item One has a ``$Q$-bundle'' $(\mathbb{E},Q) \ra (\Pi TM,d)$  -- a supervector bundle where both the total space and the base are equipped with cohomological vector fields intertwined by the bundle projection. This perspective leads to a natural proof of Theorem \ref{l41 thm MQ}. E.g., one has Stokes' theorem for fiber integrals for $Q$-bundles, which implies that $\Xi$ is a closed form on $M$ (or equivalently a closed function on $\Pi TM$), being a pushforward of a closed function $e^{-S_{MQ}^{\tiny\mbox{first-order}}}$ on the total space. 
\item Substituting the data of the bundle (\ref{l41 E_phi}) into the construction (\ref{l41 S_MQ 1st order}) one gets back the first-order action of the A-model (\ref{l41 S 1st order}); here one needs to make an appropriate change to account for the fact that the bundle (\ref{l41 E_phi}) carries a  hermitian rather than a Euclidean fiber metric.
\end{itemize}

%\end{remark}

\subsection{A-model from supersymmetric sigma model}\label{ss SUSY sigma model}
We sketch briefly the original approach \cite{Witten88}, \cite{Witten_mirror} to the A-model as a ``twist'' of another (non-topological) CFT -- the $\mc{N}=(2,2)$ supersymmetric sigma model.

Fix a source Riemann surface $\Sigma$ and a target K\"ahler manifold $X$ with metric $g$. 
The $\mc{N}=(2,2)$ supersymmetric sigma model is defined by
the action functional
%\footnote{
%Our general notational convention is that boldface objects $\bdd$, $\ppsi$, $\bJ$ etc. are coordinate-independent expressions that contain an appropriate power of $dz$ or $d\bar{z}$, whereas their non-boldface counterparts do not. E.g., $\ppsi_+=(dz)^{\frac12} \psi_+$, $\bar\bdd=d\bar{z} \bar\dd$, $\bJ=dz\,J$ etc.
%}
%\begin{equation}
%S^\mr{SUSY}=\int_\Sigma \frac{i}{2} g_{IJ} \bdd \phi^I \bar\bdd \phi^J + g_{\bar{i}j} \ppsi_+^{\bar{i}} \ol{\mathbf{D}} \ppsi_+^j-
%g_{\bar{i}j} \ppsi_-^{\bar{i}} \mathbf{D} \ppsi_-^j 
%+i R_{i\bar{i}j\bar{j}} \ppsi_+^i \ppsi_+^{\bar{i}}\ppsi_-^j \ppsi_-^{\bar{j}},
%\end{equation}
\begin{equation}\label{l41 S^SUSY}
S^\mr{SUSY}=\frac{1}{\pi}\int_\Sigma d^2 z\left( \frac12 g_{IJ} \dd \phi^I \bar\dd \phi^J + i g_{\bar{i}j} \psi_+^{\bar{i}} \ol{D} \psi_+^j+i
g_{\bar{i}j} \psi_-^{\bar{i}} D \psi_-^j 
+ R_{i\bar{i}j\bar{j}} \psi_+^i \psi_+^{\bar{i}}\psi_-^j \psi_-^{\bar{j}} \right).
\end{equation}
As in (\ref{l41 S}), we are using complex coordinates on the target; here we are additionally using a complex coordinate $z$ on the surface. 
%Symbols $D,\ol{D}$ are the covariant derivatives w.r.t. the pullback of the target Levi-Civita connection.
The fields of the supersymmetric model are:
\begin{itemize}
\item A smooth map $\phi\colon \Sigma \ra X$.
\item Odd spinors (fermions)
\begin{equation}\label{l41 SUSY sigma model fields}
\begin{gathered}
\psi_+^i (dz)^{\frac12} \in \Gamma(\Sigma, K^{\frac12}\otimes \phi^* T^{1,0}\Sigma),\quad 
\psi_+^{\bar{i}} (dz)^{\frac12}\in \Gamma(\Sigma, K^{\frac12}\otimes \phi^* T^{0,1}\Sigma),\\
\psi_-^i (d\bar{z})^{\frac12} \in \Gamma(\Sigma, \ol{K}^{\frac12}\otimes \phi^* T^{1,0}\Sigma),\quad 
\psi_-^{\bar{i}} (d\bar{z})^{\frac12} \in \Gamma(\Sigma, \ol{K}^{\frac12}\otimes \phi^* T^{0,1}\Sigma).
\end{gathered}
\end{equation}
\end{itemize}

This model has two distinguished odd holomorphic fields of conformal weight $(\frac32,0)$ -- the supercurrents
\begin{equation}
J_1 = -g_{i\bar{j}}\psi_+^i \dd \phi^{\bar{j}},\;\;
J_2 = i g_{\bar{i}j}\psi_+^{\bar{i}} \dd \phi^{j},
%\quad\in \Gamma(\Sigma,K^{\frac32}),
\end{equation}
and their antiholomorphic counterparts
\begin{equation}
\ol{J}_1 = -g_{i\bar{j}}\psi_-^i \bar\dd \phi^{\bar{j}},\;\;
\ol{J}_2 = i g_{\bar{i}j}\psi_-^{\bar{i}} \bar\dd \phi^{j}.
%\quad\in \Gamma(\Sigma,\ol{K}^{\frac32}).
\end{equation}
The supercurrents $J_{1,2}(dz)^{\frac32}$ can be contracted with a meromorphic section of $K^{-\frac12}$ and integrated around any field, and similarly for $\ol{J}_{1,2}(d\bar{z})^{\frac32}$. This gives rise to the action of the $\mc{N}=(2,2)$ superconformal algebra on the (quantum) space of fields $V$.

The stress-energy tensor of the supersymmetric model is:
\begin{equation}
\begin{gathered}
T^\mr{SUSY}=-g_{i\bar{j}}\dd\phi^i \dd\phi^{\bar{j}} - \frac{i}{2} g_{\bar{i}j}\psi^{\bar{i}}_+ D \psi_+^j - 
\frac{i}2 g_{i\bar{j}} \psi^{i}_+ D \psi_+^{\bar{j}},\\
\ol{T}^\mr{SUSY}=-g_{i\bar{j}}\bar\dd\phi^i \bar\dd\phi^{\bar{j}} - \frac{i}{2} g_{\bar{i}j}\psi^{\bar{i}}_- \ol{D} \psi_-^j - 
\frac{i}{2} g_{i\bar{j}} \psi^{i}_- \ol{D} \psi_-^{\bar{j}}
\end{gathered}
\end{equation}

Finally, the model contains an even holomorphic field $\sj$ of conformal weight $(1,0)$  -- the ``$R$-symmetry current,'' or the ``$U(1)$-current''\footnote{\label{l41 footnote: j}
$\sj$ is the Noether current for the symmetry of the action $\psi_+^i\mapsto e^{i\alpha}\psi_+^i$, $\psi_+^{\bar{i}}\mapsto e^{-i\alpha}\psi_+^{\bar{i}}$, for $\alpha$ any holomorphic function on $\Sigma$.
} --  and its antiholomorphic counterpart:
\begin{equation}
\sj =ig_{i\bar{j}}\psi_+^i \psi_+^{\bar{j}},\qquad
\bar{\sj}=ig_{i\bar{j}}\psi_-^i \psi_-^{\bar{j}}.
\end{equation}

\begin{remark} In the case of the target $X=\CC^n$ with standard K\"ahler structure, the action (\ref{l41 S^SUSY}) describes a system of $n$ complex free bosons and $n$ free Dirac fermions. In particular, the central charge of the system is $c=2\cdot n+1\cdot n=3n$ (cf. Remark \ref{l41 rem: A with flat target = bosons+ghosts}). In fact, this result remain true for a nontrivial target geometry, see (\ref{l41 c SUSY}) below. 
%In the case $X=\CC^n$, the $U(1)$-current 
\end{remark}

\subsubsection{%$\mc{N}=(2,2)$ superconformal algebra
OPE algebra of distinguished fields and commutation relations of their mode operators
}
Denote $\mathsf{n}\colon= \dim_\CC X$ -- the complex dimension of the target.

Distinguished holomorphic fields $T^\mr{SUSY},J_{1,2},\sj$ satisfy the following OPEs:
%\begin{equation}
\begin{gather}
T^\mr{SUSY}(w) T^\mr{SUSY}(z) \sim \frac{\frac32 \mathsf{n}}{(w-z)^4}+ \frac{2T(z)}{(w-z)^2} +\frac{\dd T(z)}{w-z}, \label{l41 T^SUSY OPE}\\
T^\mr{SUSY}(w) J_{1,2}(z)\sim\frac{\frac32 J_{1,2}(z)}{(w-z)^2}+\frac{\dd J_{1,2}(z)}{w-z},\quad T^\mr{SUSY}(w)\sj(z)=\frac{\sj(z)}{(w-z)^2}+\frac{\dd \sj(z)}{w-z}, \label{l41 OPE T^SUSY with J,j}\\
J_1(w) J_1(z)\sim \mr{reg.},\;\; J_2(w)J_2(z)\sim \mr{reg.},\label{l41 OPEs J1 J1, J2 J2} \\
J_1(w) J_2(z) \sim \frac{\mathsf{n}}{(w-z)^3}+\frac{\sj(z)}{(w-z)^2}+\frac{T^\mr{SUSY}(z)+\frac12 \dd \sj(z)}{w-z},\label{l41 J1 J2 OPE}\\
\sj(w) \sj(z)\sim \frac{\mathsf{n}}{(w-z)^2},\label{l41 jj OPE}\\
\sj(w) J_1(z)\sim \frac{J_1(z)}{w-z},\quad \sj(w) J_2(z)\sim \frac{-J_2(z)}{w-z} \label{l41 j OPE}
\end{gather}
%\end{equation}
and similar OPEs for the antiholomorphic counterparts. OPEs between holomorphic and antiholomorphic  fields are regular.

In particular, (\ref{l41 T^SUSY OPE}) implies that the central charge of the supersymmetric model is 
\begin{equation}\label{l41 c SUSY}
c=3\mathsf{n}=3\dim_\CC X.
\end{equation}
OPEs (\ref{l41 OPE T^SUSY with J,j}) say that fields $J_{1,2}$ are primary, with $h=\frac32$ and that $\sj$ is primary with $h=1$. (\ref{l41 j OPE}) means that fields $J_{1,2}$ have $U(1)$-charge $\pm 1$.

%\textcolor{red}{COMM RELATIONS OF N=2 SUPER-VIRASORO}
As a consequence of the OPEs above, using Lemma \ref{l25 lemma: algebra of mode operators}, the mode operators of the holomorphic fields $T^\mr{SUSY},J_{1,2},\sj$ form a Lie superalgebra and satisfy the following (super)commutation relations:
\begin{equation}
\begin{gathered}
{}[L_n,L_m]=(n-m)L_{n+m}+\frac{\sn}{4}(n^3-n)\delta_{n,-m},\\
[L_n,J^{1,2}_r]=\left(\frac{n}{2}-r\right)J^{1,2}_{n+r},\quad
[L_n,\sj_m]=n \sj_{n+m} ,\\
[J^1_r,J^1_s]_+=0,\quad [J^2_r,J^2_s]_+=0,\\
[J^1_r,J^2_s]_+=L_{r+s}+\frac12 (r-s)\sj_{r+s}+\frac{\sn}{2}\left(r^2-\frac14\right) \delta_{r,-s},\\
[\sj_n,\sj_m]=\sn\, n \delta_{n,-m},\\
[\sj_n,J^1_r]=J^1_{n+r},\quad [\sj_n,J^2_r]=-J^2_{n+r}.
\end{gathered}
\end{equation}
Here $L_n,\sj_n$ are the even mode operators of $T^\mr{SUSY}$, with $n\in \ZZ$;  $J^{1,2}_r$ are the odd mode operators of $J_{1,2}$, with $r$ ranging either over integers or over half-integers, depending on the choice of spin-structure\footnote{
If the mode operators are understood as acting on fields at $z$, then here we are talking about the choice of spin-structure (or periodicity/antiperiodicity condition for fermions) on the punctured disk around $z$. Periodic condition on the punctured disk (Neveu-Schwarz spin-structure) corresponds to $r$ ranging in half-integers; antiperiodic condition (Ramond spin structure) corresponds to integer $r$, cf. Section \ref{ss fermion can quantization}.
} %on $\CC^*$ 
in (\ref{l41 SUSY sigma model fields}). This Lie superalgebra is known as $\mc{N}=2$ super-Virasoro algebra, or equivalently as $\mc{N}=2$ superconformal algebra.

\subsubsection{The ``A-twist''}
The ``A-twist'' of the supersymmetric sigma-model consists in changing the stress-energy tensor as\footnote{
One also has a ``B-twist'' where the sign of the shift for $\ol{T}$ is $+$, leading to the ``B-model,'' \cite{Witten88}, \cite{Witten_mirror}.
}
\begin{equation}
T^\mr{SUSY} \mapsto T^\mr{A-model}=T^\mr{SUSY}+\frac12 \dd \sj,
\qquad  \ol{T}^\mr{SUSY} \mapsto \ol{T}^\mr{A-model}=\ol{T}^\mr{SUSY}-\frac12 \bar\dd\, \bar\sj.
\end{equation}
The action, fields (locally) and equations of motion are unchanged, see Remark \ref{l41 rem: twist map} below.

The change of the stress-energy tensor affects the conformal weights of fields (recall that they are determined by the quadratic pole in the OPE of the stress-energy tensor with the field). In particular: 
\begin{itemize}
\item Conformal weight of $\psi_+^{i}$ changes from $(h,\bar{h})=(\frac12,0)$ (a left Weyl spinor) to $(h,\bar{h})=(0,0)$ (scalar).
\item $\psi_+^{\bar{i}}$ changes from $(\frac12,0)$ (a left Weyl spinor) to $(1,0)$ (thus, $dz\, \psi_+^{\bar{i}}$ is a $(1,0)$-form field).
\item $\psi_-^i$ changes from $(0,\frac12)$ (a right Weyl spinor) to $(0,1)$ (i.e., $d\bar{z}\, \psi_-^i$ is a $(0,1)$-form field),
\item $\psi_-^{\bar{i}}$ changes from $(0,\frac12)$ (a right Weyl spinor) to $(0,0)$ (a scalar).
\item Conformal weights of the supercurrent shift as
\begin{equation}
\begin{gathered}
J_1\colon (3/2,0)\mapsto (1,0),\quad
J_2\colon (3/2,0)\mapsto (2,0),\\
\ol{J}_1 \colon (0,3/2) \mapsto (0,2),\quad
\ol{J}_2\colon (0,3/2) \mapsto (0,1).
\end{gathered}
\end{equation}
\end{itemize}

Thus, the twist transforms the spinor fields of the supersymmetric sigma model into differential form fields of the A-model.

The correspondence of notations for fields of the supersymmetric sigma-model and the A-model is given by the following dictionary:
\begin{equation}\label{l41 SUSY-A dictionary}
\begin{array}{c|c}
\mr{SUSY\;sigma\;model} & \mr{A-model}\\ \hline
\phi^I \;\;(0,0) & \phi^I \;\; (0,0)\\
\psi_+^i\;\; (1/2,0) & \chi^i \;\; (0,0)\\
\psi_+^{\bar{i}}\;\; (1/2,0) & g^{\bar{i}j} \psi_j\;\; (1,0) \\
\psi_-^i \;\; (0,1/2) & g^{i\bar{j}} \psi_{\bar{j}}\;\; (0,1) \\
\psi_-^{\bar{i}}\;\; (0,1/2) & \chi^{\bar{i}} \;\; (0,0) \\ \hline
J_1\;\; (3/2,0) & J\;\; (1,0) \\
J_2 \;\; (3/2,0) & G \;\; (2,0) \\
\ol{J}_1 \;\; (0,3/2) & i\ol{G} \;\; (0,2) \\
\ol{J}_2 \;\; (0,3/2) & -i\ol{J} \;\; (0,1)
\end{array}
\end{equation}
Here we are indicating the conformal weight $(h,\bar{h})$ of each field. 
In particular, the supercurrents after the twist become the natural objects of a TCFT -- the holomorphic/antiholomorphic BRST currents $J,\ol{J}$ and the primitive fields for the stress-energy tensor, $G,\ol{G}$.

All fields in the table (\ref{l41 SUSY-A dictionary}) are primary, both on the supersymmetric and on the A-model side.

\begin{remark}\label{l41 rem: twist map}
In a coordinate-independent language on the source surface, the twist yields a mapping of fields of the supersymmetric model on a contratible open set $U\subset \Sigma$ to fields of the A-model:
\begin{equation}\label{l41 twist map}
\begin{array}{ccc}
\F_U^\mr{SUSY} & \ra & \F_U^\mr{A-model} \\
(\phi^I,\ppsi_+^i,\ppsi_+^{\bar{i}},\ppsi_-^i,\ppsi_i^{\bar{i}}) & \mapsto &
(\phi^I, \underbrace{\lambda^{-1} \ppsi_+^i}_{\chi^i}, \underbrace{\lambda \ppsi_+^{\bar{i}}}_{g^{\bar{i}j}\psi_j^{(1,0)}},\underbrace{\bar\lambda \ppsi_-^i}_{g^{i\bar{j}}\psi^{(0,1)}_{\bar{j}}}, \underbrace{\bar\lambda^{-1} \ppsi_-^{\bar{i}}}_{\chi^{\bar{i}}})
\end{array}
\end{equation}
Here $\lambda\in \Gamma(U,K^{\frac12})$ is some reference nonvanishing holomorphic spinor on $U$ (e.g. if $U$ is equipped with a complex coordinate $z$, one can choose $\lambda=(dz)^{\frac12}$).

Since the Lagrangian density of the action (\ref{l41 S^SUSY}) is invariant under the R-symmetry (footnote \ref{l41 footnote: j}), its pushforward under the map (\ref{l41 twist map}) is independent under the choice of the reference spinor $\lambda$ and yields the Lagrangian density of the A-model (\ref{l41 S}).
\end{remark}

%\chapter{Appendix}
%\subfile{chapter10_App}

\begin{comment}

\end{comment}


\begin{thebibliography}{99}
\bibitem{Anderson} I. M. Anderson, ``Introduction to the variational bicomplex.'' (1992).
\bibitem{ADPW} S. Axelrod, S. Della Pietra, E. Witten, ``Geometric quantization of Chern-Simons gauge theory,'' J. Diff. Geom. 33.3 (1991) 787--902.

\bibitem{Atiyah} M. Atiyah, ``Topological quantum field theory.'' Publications Mathématiques de l'IH\'ES 68 (1988) 175--186.

\bibitem{Atiyah-Jeffrey} M. Atiyah, L. Jeffrey, ``Topological Lagrangians and cohomology.'' Journal of Geometry and Physics 7.1 (1990) 119--136.


\bibitem{Barannikov-Kontsevich} S. Barannikov, M. Kontsevich, ``Frobenius manifolds and formality of Lie algebras of polyvector fields,'' Internat. Math. Res. Notes 14(1998), 201--215.
\bibitem{BPZ} 	
A. A. Belavin, A. M. Polyakov, A. B. Zamolodchikov, ``Infinite conformal symmetry in two-dimensional quantum field theory,'' Nuclear Physics B 241.2 (1984) 333--380.

\bibitem{CdS}  A. Cannas Da Silva, ``Lectures on symplectic geometry,'' Vol. 3575. Berlin: Springer (2008).

\bibitem{CMRpert} A. S. Cattaneo, P. Mnev, N. Reshetikhin, ``Perturbative quantum gauge theories on manifolds with boundary.'' Communications in Mathematical Physics 357, no. 2 (2018) 631--730.



 
\bibitem{DMS} P. Di Francesco, P. Mathieu, D. S\'en\'echal, ``Conformal field theory.'' Springer Science \& Business Media, 2012.

\bibitem{Dubrovin92} B. Dubrovin, ``Integrable systems in topological field theory.'' Nucl. Phys. B379 (1992) 627--689.

\bibitem{Dubrovin96}  B. Dubrovin, ``Geometry of 2D topological field theories.'' In: Springer LNM, 1620 (1996) 120--348.

\bibitem{Farb_Margalit} B. Farb, D. Margalit. ``A primer on mapping class groups,'' %(pms-49).'' 
Princeton university press (2011).

\bibitem{Feigin-Fuchs} B. L. Feigin, D. B.  Fuchs, ``Verma modules over the Virasoro algebra.'' Topology. Springer, Berlin, Heidelberg (1984) 230--245.

\bibitem{Frenkel-Losev} E. Frenkel, A. Losev, ``Mirror symmetry in two steps: A–I–B,'' Communications in mathematical physics 269.1 (2007) 39--86.

\bibitem{FK}
J. Fr\"ohlich, C. King, ``The Chern-Simons theory and knot polynomials,'' Communications in mathematical physics 126.1 (1989) 167--199.

\bibitem{Gawedzki_IAS}
K. Gawedzki, ``Lectures on conformal field theory.'' No. IHES-P-97-02. SCAN-9703129, 1997.


\bibitem{GK} K. Gawedzki, A. Kupiainen, ``SU(2) Chern-Simons theory at genus zero,''  Communications in mathematical physics 135.3 (1991) 531--546.

\bibitem{Getzler_BV} E. Getzler, "Batalin-Vilkovisky algebras and two-dimensional topological field theories," Communications in mathematical physics 159.2 (1994) 265--285


\bibitem{Ginsparg} P. Ginsparg,  ``Applied conformal field theory,'' arXiv preprint hep-th/9108028 (1988).

\bibitem{GJ}
J. Glimm, A. Jaffe, ``Quantum physics. A functional integral point of view,'' Springer-Verlag,
New York, second edition (1987).


\bibitem{IM} R. Iraso, P. Mnev, ``Two-dimensional Yang-Mills theory on surfaces with corners in Batalin–Vilkovisky formalism.'' Communications in Mathematical Physics 370, no. 2 (2019) 637--702.

\bibitem{TJF} Th. Johnson-Freyd, ``Feynman-diagrammatic description of the asymptotics of the time evolution operator in quantum mechanics.'' Letters in Mathematical Physics 94, no. 2 (2010) 123--149. 

\bibitem{Kac} V. G. Kac, "Highest weight representations of infinite dimensional Lie algebras", Proc. Internat. Congress Mathematicians (Helsinki, 1978).

\bibitem{Kac_Raina} %V. G. Kac, A. K. Raina, N. Rozhkovskaya, ``Bombay lectures on highest weight representations of infinite dimensional Lie algebras.'' Vol. 29. World scientific, 2013.
V. G. Kac, A. K. Raina,  ``Bombay lectures on highest weight representations of infinite dimensional Lie algebras.'' World scientific, 1987.

\bibitem{2dscalar} S. Kandel, P. Mnev, K. Wernli, ``Two-dimensional perturbative scalar QFT and Atiyah-Segal gluing.'' arXiv preprint arXiv:1912.11202 (2019).

\bibitem{Keel} S. Keel, ``Intersection theory of moduli space of stable N-pointed curves of genus zero,'' Trans. AMS 330.2 (1992) 545--574.

\bibitem{KZ} V. G. Knizhnik, A. B. Zamolodchikov, ``Current Algebra and Wess–Zumino Model in Two-Dimensions,'' Nucl. Phys. B, 247.1 (1984) 83--103.

\bibitem{Kohno} T. Kohno, ``Conformal field theory and topology.'' American Mathematical Soc., 2002.

\bibitem{KM} M. Kontsevich, Yu. Manin, ``Gromov-Witten classes, quantum cohomology, and enumerative geometry,'' Commun. Math. Phys. 164 (1994) 525--562.

\bibitem{Losev2008} A. S. Losev, Lectures on topological quantum field theory, 2008 (lectures given online in Russian).

\bibitem{Losev_TQM} A. S. Losev, ``TQFT, homological algebra and elements of K. Saito’s theory of Primitive form: an attempt of mathematical text written by mathematical physicist.'' In: Primitive Forms and Related Subjects--Kavli IPMU 2014, vol. 83, pp. 269--294. Mathematical Society of Japan, 2019.

\bibitem{LMY1} A. S. Losev, P. Mnev, D. R. Youmans,
"Two-dimensional abelian BF theory in Lorenz gauge as a twisted N=(2, 2) superconformal field theory," Journal of Geometry and Physics 131 (2018) 122--137.

\bibitem{LMY2} A. S. Losev, P. Mnev, D. R. Youmans, "Two-dimensional non-abelian BF theory in Lorenz gauge as a solvable logarithmic TCFT," Communications in Mathematical Physics 376.2 (2020) 993--1052.

\bibitem{Mathai-Quillen} V. Mathai, D. Quillen,  ``Superconnections, Thom classes, and equivariant differential forms,'' Topology, 25.1 (1986) 85--110.

\bibitem{Onsager} L. Onsager, ``Crystal statistics. I. A two-dimensional model with an order-disorder transition,'' Physical Review, Series II, 65 (3–4) (1944) 117--149.
\bibitem{Penner} R. C. Penner, ``Decorated Teichm\"uller theory.'' Vol. 1. European Mathematical Society, 2012.
\bibitem{Reshetikhin} N. Reshetikhin, ``Lectures on quantization of gauge systems,'' In: \textit{New Paths Towards Quantum Gravity}. Springer, Berlin, Heidelberg (2010) 125--190.
\bibitem{Schottenloher} M. Schottenloher,  ``A mathematical introduction to conformal field theory.'' Vol. 759. Springer  (2008).
\bibitem{Segal88} G. Segal, ``The definition of conformal field theory,'' Differential geometrical methods in theoretical physics. Springer, Dordrecht (1988) 165--171.


\bibitem{Simon}
 B. Simon, ``The $P(\phi)_2$ Euclidean (quantum) field theory,''  Princeton University Press, Princeton, N.J. (1974). 
 %Princeton Series in Physics.

\bibitem{Sugawara} H. Sugawara, ``A field theory of currents,'' Phys. Rev. 170, 1659 (1968).

\bibitem{Verlinde}
E. Verlinde, ``Fusion rules and modular transformations in 2D conformal field theory,'' Nuclear Physics B, 300.3 (1988) 360--376.

\bibitem{Donald_PhD}
D. R. Youmans, ``Topological conformal field theories from gauge-fixed topological gauge theories: a case study,'' Ph.D. dissertation, Universit\'e de Gen\`eve (2020).

\bibitem{Witten_Morse} E. Witten, ``Supersymmetry and Morse theory,'' Journal of differential geometry 17.4 (1982) 661--692.

\bibitem{Witten88} E. Witten, ``Topological sigma models,'' Commun. Math. Phys. 118 (1988) 411--449.

\bibitem{Witten_Donaldson} E. Witten, ``Topological quantum field theory,'' Comm. Math. Phys. Volume 117, Number 3 (1988) 353--386 

\bibitem{Witten89} E. Witten, ``Quantum field theory and the Jones polynomial.'' Communications in Mathematical Physics 121.3 (1989) 351--399.

\bibitem{Witten_2dgrav} E. Witten, ``Two-dimensional gravity and intersection theory on moduli
space.'' Surveys in Diff. Geom. 1 (1991) 243--310.

\bibitem{Witten_mirror} E. Witten, ``Mirror manifolds and topological field theory.'' arXiv:hep-th/9112056 (1991).

\bibitem{Witten_superstring} E. Witten, ``Superstring perturbation theory revisited,''  arXiv:1209.5461 (2012).
\end{thebibliography}
\end{document}